%% file: paper.tex
\documentclass[nolinenum]{jfp}

 \DeclareTextCommandDefault{\nobreakspace}{\leavevmode\nobreak\ }

\input{macros.tex}

\makeatletter
\@ifundefined{lhs2tex.lhs2tex.sty.read}%
  {\@namedef{lhs2tex.lhs2tex.sty.read}{}%
   \newcommand\SkipToFmtEnd{}%
   \newcommand\EndFmtInput{}%
   \long\def\SkipToFmtEnd#1\EndFmtInput{}%
  }\SkipToFmtEnd

\newcommand\ReadOnlyOnce[1]{\@ifundefined{#1}{\@namedef{#1}{}}\SkipToFmtEnd}
\usepackage{amstext}
\usepackage{amssymb}
\usepackage{stmaryrd}
\DeclareFontFamily{OT1}{cmtex}{}
\DeclareFontShape{OT1}{cmtex}{m}{n}
  {<5><6><7><8>cmtex8
   <9>cmtex9
   <10><10.95><12><14.4><17.28><20.74><24.88>cmtex10}{}
\DeclareFontShape{OT1}{cmtex}{m}{it}
  {<-> ssub * cmtt/m/it}{}

\DeclareFontShape{OT1}{cmtt}{bx}{n}
  {<5><6><7><8>cmtt8
   <9>cmbtt9
   <10><10.95><12><14.4><17.28><20.74><24.88>cmbtt10}{}
\DeclareFontShape{OT1}{cmtex}{bx}{n}
  {<-> ssub * cmtt/bx/n}{}

\newcommand{\Conid}[1]{\mathit{#1}}
\newcommand{\Varid}[1]{\mathit{#1}}
\newcommand{\anonymous}{\kern0.06em \vbox{\hrule\@width.5em}}

\renewcommand{\geq}{\geqslant}
\usepackage{polytable}

\@ifundefined{mathindent}%
  {\newdimen\mathindent\mathindent\leftmargini}%
  {}%

\def\resethooks{%
  \global\let\SaveRestoreHook\empty
  \global\let\ColumnHook\empty}
\newcommand*{\savecolumns}[1][default]%
  {\g@addto@macro\SaveRestoreHook{\savecolumns[#1]}}
\newcommand*{\restorecolumns}[1][default]%
  {\g@addto@macro\SaveRestoreHook{\restorecolumns[#1]}}
\newcommand*{\aligncolumn}[2]%
  {\g@addto@macro\ColumnHook{\column{#1}{#2}}}

\resethooks

\newcommand{\onelinecommentchars}{\quad-{}- }
\newcommand{\commentbeginchars}{\enskip\{-}
\newcommand{\commentendchars}{-\}\enskip}

\newcommand{\visiblecomments}{%
  \let\onelinecomment=\onelinecommentchars
  \let\commentbegin=\commentbeginchars
  \let\commentend=\commentendchars}

\newcommand{\invisiblecomments}{%
  \let\onelinecomment=\empty
  \let\commentbegin=\empty
  \let\commentend=\empty}

\visiblecomments

\newlength{\blanklineskip}
\newcommand{\hsindent}[1]{\quad}%
\let\hspre\empty
\let\hspost\empty

\EndFmtInput
\makeatother
\ReadOnlyOnce{polycode.fmt}%
\makeatletter

\newcommand{\hsnewpar}[1]%
  {{\parskip=0pt\parindent=0pt\par\vskip #1\noindent}}

\newcommand{\hscodestyle}{}

\newcommand{\sethscode}[1]%
  {\expandafter\let\expandafter\hscode\csname #1\endcsname
   \expandafter\let\expandafter\endhscode\csname end#1\endcsname}

  {\par\noindent
   \advance\leftskip\mathindent
   \hscodestyle
   \let\\=\@normalcr
   \let\hspre\(\let\hspost\)%
   \pboxed}%
  {\endpboxed\)%
   \par\noindent
   \ignorespacesafterend}

  {\hsnewpar\abovedisplayskip
   \advance\leftskip\mathindent
   \hscodestyle
   \let\hspre\(\let\hspost\)%
   \pboxed}%
  {\endpboxed%
   \hsnewpar\belowdisplayskip
   \ignorespacesafterend}

  {\hsnewpar\abovedisplayskip
   \advance\leftskip\mathindent
   \hscodestyle
   \let\\=\@normalcr
   \(\pboxed}%
  {\endpboxed\)%
   \hsnewpar\belowdisplayskip
   \ignorespacesafterend}

\newcommand{\plainhs}{\sethscode{plainhscode}}

\plainhs

  {\hsnewpar\abovedisplayskip
   \advance\leftskip\mathindent
   \hscodestyle
   \let\\=\@normalcr
   \(\parray}%
  {\endparray\)%
   \hsnewpar\belowdisplayskip
   \ignorespacesafterend}

  {\parray}{\endparray}

  {\(\parray}{\endparray\)}

\def\codeframewidth{\arrayrulewidth}
\RequirePackage{calc}

  {\parskip=\abovedisplayskip\par\noindent
   \hscodestyle
   \arrayrulewidth=\codeframewidth
   \tabular{@{}|p{\linewidth-2\arraycolsep-2\arrayrulewidth-2pt}|@{}}%
   \hline\framedhslinecorrect\\{-1.5ex}%
   \let\endoflinesave=\\
   \let\\=\@normalcr
   \(\pboxed}%
  {\endpboxed\)%
   \framedhslinecorrect\endoflinesave{.5ex}\hline
   \endtabular
   \parskip=\belowdisplayskip\par\noindent
   \ignorespacesafterend}

\newcommand{\framedhslinecorrect}[2]%
  {#1[#2]}

  {\(\def\column##1##2{}%
   \let\>\undefined\let\<\undefined\let\\\undefined
   \newcommand\>[1][]{}\newcommand\<[1][]{}\newcommand\\[1][]{}%
   \def\fromto##1##2##3{##3}%
   }{\) }%

  {\let\orighscode=\hscode
   \let\origendhscode=\endhscode
   \def\endhscode{\def\hscode{\endgroup\def\@currenvir{hscode}\\}\begingroup}
   \orighscode\def\hscode{\endgroup\def\@currenvir{hscode}}}%
  {\origendhscode
   \global\let\hscode=\orighscode
   \global\let\endhscode=\origendhscode}%

\makeatother
\EndFmtInput
\ReadOnlyOnce{forall.fmt}%
\makeatletter

\let\HaskellResetHook\empty
\newcommand*{\AtHaskellReset}[1]{%
  \g@addto@macro\HaskellResetHook{#1}}
\newcommand*{\HaskellReset}{\HaskellResetHook}

\newcommand\hsforall{\global\let\hsdot=\hsperiodonce}
\newcommand*\hsperiodonce[2]{#2\global\let\hsdot=\hscompose}
\newcommand*\hscompose[2]{#1}

\AtHaskellReset{\global\let\hsdot=\hscompose}

\HaskellReset

\makeatother
\EndFmtInput

\usepackage{thm-restate}

\usepackage{enumitem}
\newlist{qwq}{itemize}{1}
\setlist[qwq]{label={}, nosep, leftmargin=1em}
\newcommand{\indentbegin}{\begin{qwq} \item}
\newcommand{\indentend}{\end{qwq}}

\newtheorem{theorem}{Theorem}
\newtheorem{lemma}{Lemma}

\begin{document}

\journaltitle{Preprint submitted to JFP}
\cpr{Cambridge University Press}
\doival{}
\jnlDoiYr{2023}

\totalpg{\pageref{lastpage01}}

\title{From High to Low: Simulating Nondeterminism and State with State}

\begin{authgrp}
\author{Wenhao Tang}
\affiliation{The University of Edinburgh \\
        (\email{wenhao.tang@ed.ac.uk})}
\author{Tom Schrijvers}
\affiliation{KU Leuven Department of Computer Science \\
        (\email{tom.schrijvers@kuleuven.be})}
\end{authgrp}

\begin{abstract}

  Some effects are considered to be higher-level than others.
  High-level effects provide expressive and succinct abstraction of
  programming concepts, while low-level effects allow more fine-grained
  control over program execution and resources.
  Yet, often it is desirable to write programs using the convenient abstraction
  offered by high-level effects, and meanwhile still benefit
  from the optimisations enabled by low-level effects.
  One solution is to translate high-level effects to low-level ones.

  This paper studies how algebraic effects and handlers allow us to simulate
  high-level effects in terms of low-level effects.
  In particular, we focus on the interaction between state and
  nondeterminism known as the local state, as provided by Prolog. We map
  this high-level semantics in successive steps
  onto a low-level composite state effect, similar to that managed by Prolog's Warren Abstract
  Machine.
  We first give a translation from the high-level local-state semantics to
  the low-level global-state semantics, by explicitly restoring state updates on
  backtracking. Next, we eliminate nondeterminsm altogether in favor of a
  lower-level state containing a choicepoint stack.
  Then we avoid copying the state by restricting ourselves to incremental, reversible state updates.
  We show how these updates can be stored on a trail stack with another state effect.
  We prove the correctness of all our steps using program calculation where
  the fusion laws of effect handlers play a central role.

\end{abstract}

\maketitle

\section{Introduction}
\label{sec:introduction}

The trade-off between ``high-level'' and ``low-level'' styles of programming is almost as old as the field of computer sciences itself.
In a high-level style of programming, we lean on abstractions to make our
programs easier to read and write, and less error-prone.  We pay for this
comfort by giving up precise control over the underlying machinery; we forego
optimisation opportunities or have to trust a (usually opaque) compiler to
perform low-level optimisations for us. For performance-sensitive applications, compiler optimisations
are not reliable enough; instead we often resort to lower-level programming
techniques ourselves.  Although they allow a
fine-grained control over program execution and the implementation of
optimisation techniques, they tend to be harder to write and not compose very well.  This is an
important trade-off to take into account when choosing an appropriate programming language
for implementing an application.

Maybe surprisingly, as they are rarely described in this way, there is a
similar pattern for side-effects within programming languages: some effects can
be described as ``lower-level'' than others.  We say that an effect is
lower-level than another effect when the lower-level effect can simulate the
higher-level effect. In other words, it is possible to write a
program using lower-level effects that has identical semantics to the same program with
higher-level effects.
Yet, due to the lack of abstraction of low-level effects, writing a faithful
simulation requires careful discipline and is quite error-prone.

This article investigates how we can construct programs that are most
naturally expressed with a high-level effect, but where we still want access to
the optimisation opportunities of a lower-level effect. In particular,
inspired by Prolog and Constraint Programming systems, we investigate programs
that rely on the high-level interaction between the nondeterminism and state
effects which we call \emph{local state}. Following low-level implementation
techniques for these systems, like the Warren Abstract Machine (WAM)
\citep{AICPub641:1983,AitKaci91}, we show how these can be simulated in terms of the low-level
\emph{global state} interaction of state and nondeterminism, and finally by state alone. This
allows us to incorporate typical optimisations like exploiting mutable state
for efficient backtracking based on \emph{trailing} as opposed to copying or recomputing
the state from scratch ~\citep{Schulte:ICLP:1999}.

Our approach is based on algebraic effects and handlers
~\citep{PlotkinP03, Plotkin09, Plotkin13} to cleanly separate the
syntax and semantics of effects.
For programs written with high-level effects and interpreted by their
handlers, we can define a general translation handler to transform
these high-level effects to low-level effects, and then interpret the
translated programs with the handlers of low-level effects.

Of particular interest is the way we reason about the correctness of our
approach. There has been much debate in the literature on different equational
reasoning approaches for effectful computations. \citet{Hutton08} break the
abstraction boundaries and use the actual implementation in their equational
reasoning approach. \citet{Gibbons11} promote an alternative, law-based
approach to preserve abstraction boundaries and combine axiomatic with
equational reasoning. In an earlier version of this work \citep{Pauwels19}, we
have followed the latter, law-based approach for reasoning about the
correctness of simulating local state with global state. However, we have found
that approach to be unsatisfactory because it incorporates elements
that are usually found in the syntactic approach for reasoning about
programming languages \citep{Felleisen94}, leading to more boilerplate and
complication in the proofs: notions of contextual equivalence and explicit
manipulation of program contexts. Hence, for that reason we return to the
implementation-based reasoning approach, which we believe works well with
algebraic effects and handlers. Indeed, we prove all of our simulations
correct using equational reasoning techniques, exploiting in particular the
fusion property of handlers ~\citep{Wu15,Gibbons00}.

After introducing the reader to the appropriate background material and motivating the problem (\Cref{sec:background}
and \Cref{sec:overview}),
this paper makes the following contributions:

\begin{itemize}
	\item We distinguish between local-state and global-state semantics,
              and simulate the former in terms of the latter (\Cref{sec:local-global}).
	\item We simulate nondeterminism using a state that consists of a choicepoint stack (\Cref{sec:nondeterminism-state}).
	\item We combine the previous two simulations and merge the two states into
		  a single state effect (\Cref{sec:combination}).
	\item By only allowing incremental, reversible updates to the state we can avoid holding on to multiple copies of the state (\Cref{sec:undo}).
	\item By storing the incremental updates in a trail stack state, we can restore them in batch when backtracking (\Cref{sec:trail-stack}).
	\item We prove all simulations correct using equational reasoning techniques and the fusion law for handlers in particular
        (\Cref{app:local-global}, \Cref{app:nondet-state}, \Cref{app:states-state}, \Cref{app:final-simulate}, \Cref{app:modify-local-global} and \Cref{app:immutable-trail-stack}).
\end{itemize}
Finally, we discuss related work (\Cref{sec:related-work}) and
conclude (\Cref{sec:conclusion}).
\Cref{tab:overview} gives an overview of the simulations of high-level
effects with low-level effects we implemented and proved in the paper.
Throughout the paper, we use Haskell as a means to illustrate
our findings with code.

\begin{table}[h]
\begin{center}
\begin{tabular}{ c|c|c }
 Translations & Descriptions & Correctness \\
 \hline
 \ensuremath{\Varid{local2global}} & local state to global state (\S\ref{sec:local2global}) &\Cref{thm:local-global}\\
 \ensuremath{\Varid{nondet2state}} & nondeterminism to state (\S\ref{sec:nondet2state}) &\Cref{thm:nondet-state}\\
 \ensuremath{\Varid{states2state}} & multiple states to a single state (\S\ref{sec:multiple-states}) &\Cref{thm:states-state}\\
 \ensuremath{\Varid{local2global_M}} & local state to global state with reversible updates (\S\ref{sec:local2globalM}) &\Cref{thm:modify-local-global}\\
 \ensuremath{\Varid{local2trail}} & local state to global state with trail stacks (\S\ref{sec:local2trail}) &\Cref{thm:trail-local-global}\\
\end{tabular}
\end{center}
\caption{Overview of translations from high-level effects to low-level effects in the paper.}
\label{tab:overview}
\end{table}

\section{Background and Motivation}
\label{sec:background}

This section summarises the main prerequisites for equational
reasoning with effects and motivates our translations from high-level
effects to low-level effects.
We discuss the two central effects of this paper: state and
nondeterminism.

\subsection{Functors and Monads}
\label{sec:functors-and-monads}

\paragraph*{Functors}\
In Haskell, a functor \ensuremath{\Varid{f}\mathbin{::}\mathbin{\ast}\to \mathbin{\ast}} instantiates the functor type class, which has a single
functor mapping operation.\indentbegin \begin{hscode}\SaveRestoreHook
\column{B}{@{}>{\hspre}l<{\hspost}@{}}%
\column{3}{@{}>{\hspre}l<{\hspost}@{}}%
\column{5}{@{}>{\hspre}l<{\hspost}@{}}%
\column{E}{@{}>{\hspre}l<{\hspost}@{}}%
\>[3]{}\mathbf{class}\;\Conid{Functor}\;\Varid{f}\;\mathbf{where}{}\<[E]%
\\
\>[3]{}\hsindent{2}{}\<[5]%
\>[5]{}\Varid{fmap}\mathbin{::}(\Varid{a}\to \Varid{b})\to \Varid{f}\;\Varid{a}\to \Varid{f}\;\Varid{b}{}\<[E]%
\ColumnHook
\end{hscode}\resethooks
\indentend Furthermore, a functor should satisfy the following two functor laws:
\begin{alignat}{2}
    &\mbox{\bf identity}:\quad &
    \ensuremath{\Varid{fmap}\;\Varid{id}} &= \ensuremath{\Varid{id}}\mbox{~~,} \label{eq:functor-identity}\\
    &\mbox{\bf composition}:~ &
    \ensuremath{\Varid{fmap}\;(\Varid{f}\hsdot{\circ }{.}\Varid{g})} &= \ensuremath{\Varid{fmap}\;\Varid{f}\hsdot{\circ }{.}\Varid{fmap}\;\Varid{g}} \mbox{~~.} \label{eq:functor-composition}
\end{alignat}
We sometimes use the operator \ensuremath{\mathbin{\langle\hspace{1.6pt}\mathclap{\raisebox{0.1pt}{\scalebox{1}{\$}}}\hspace{1.6pt}\rangle}} as an alias for \ensuremath{\Varid{fmap}}.
\indentbegin \begin{hscode}\SaveRestoreHook
\column{B}{@{}>{\hspre}l<{\hspost}@{}}%
\column{3}{@{}>{\hspre}l<{\hspost}@{}}%
\column{E}{@{}>{\hspre}l<{\hspost}@{}}%
\>[3]{}(\mathbin{\langle\hspace{1.6pt}\mathclap{\raisebox{0.1pt}{\scalebox{1}{\$}}}\hspace{1.6pt}\rangle})\mathbin{::}\Conid{Functor}\;\Varid{f}\Rightarrow (\Varid{a}\to \Varid{b})\to \Varid{f}\;\Varid{a}\to \Varid{f}\;\Varid{b}{}\<[E]%
\\
\>[3]{}(\mathbin{\langle\hspace{1.6pt}\mathclap{\raisebox{0.1pt}{\scalebox{1}{\$}}}\hspace{1.6pt}\rangle})\mathrel{=}\Varid{fmap}{}\<[E]%
\ColumnHook
\end{hscode}\resethooks
\indentend 
\paragraph*{Monads}\
Monadic side-effects~\citep{Moggi91}, the main focus of this paper,
are those that can dynamically determine what happens next.
A monad \ensuremath{\Varid{m}\mathbin{::}\mathbin{\ast}\to \mathbin{\ast}} is a functor instantiates the monad type class,
which has two operations return (\ensuremath{\Varid{\eta}}) and bind (\ensuremath{>\!\!>\!\!=}).
\indentbegin \begin{hscode}\SaveRestoreHook
\column{B}{@{}>{\hspre}l<{\hspost}@{}}%
\column{3}{@{}>{\hspre}l<{\hspost}@{}}%
\column{5}{@{}>{\hspre}l<{\hspost}@{}}%
\column{13}{@{}>{\hspre}l<{\hspost}@{}}%
\column{E}{@{}>{\hspre}l<{\hspost}@{}}%
\>[3]{}\mathbf{class}\;\Conid{Functor}\;\Varid{m}\Rightarrow \Conid{Monad}\;\Varid{m}\;\mathbf{where}{}\<[E]%
\\
\>[3]{}\hsindent{2}{}\<[5]%
\>[5]{}\Varid{\eta}{}\<[13]%
\>[13]{}\mathbin{::}\Varid{a}\to \Varid{m}\;\Varid{a}{}\<[E]%
\\
\>[3]{}\hsindent{2}{}\<[5]%
\>[5]{}(>\!\!>\!\!=){}\<[13]%
\>[13]{}\mathbin{::}\Varid{m}\;\Varid{a}\to (\Varid{a}\to \Varid{m}\;\Varid{b})\to \Varid{m}\;\Varid{b}{}\<[E]%
\ColumnHook
\end{hscode}\resethooks
\indentend Furthermore, a monad should satisfy the following three monad laws:
\begin{alignat}{2}
    &\mbox{\bf return-bind}:\quad &
    \ensuremath{\Varid{\eta}\;\Varid{x}>\!\!>\!\!=\Varid{f}} &= \ensuremath{\Varid{f}\;\Varid{x}}\mbox{~~,} \label{eq:monad-ret-bind}\\
    &\mbox{\bf bind-return}:~ &
    \ensuremath{\Varid{m}>\!\!>\!\!=\Varid{\eta}} &= \ensuremath{\Varid{m}} \mbox{~~,} \label{eq:monad-bind-ret}\\
    &\mbox{\bf associativity}:~ &
    \ensuremath{(\Varid{m}>\!\!>\!\!=\Varid{f})>\!\!>\!\!=\Varid{g}} &= \ensuremath{\Varid{m}>\!\!>\!\!=(\lambda \Varid{x}\to \Varid{f}\;\Varid{x}>\!\!>\!\!=\Varid{g})} \mbox{~~.}
    \label{eq:monad-assoc}
\end{alignat}
Haskell supports \ensuremath{\mathbf{do}} blocks as syntactic sugar for monadic computations.
For example, \ensuremath{\mathbf{do}\;\Varid{x}\leftarrow \Varid{m};\Varid{f}\;\Varid{x}} is translated to \ensuremath{\Varid{m}>\!\!>\!\!=\Varid{f}}.
Two convenient derived operators are \ensuremath{>\!\!>} and \ensuremath{\mathbin{\langle\hspace{1.6pt}\mathclap{\raisebox{0.1pt}{\scalebox{1}{$\ast$}}}\hspace{1.6pt}\rangle}}.\footnote{We
deviate from the type class hierarchy of \ensuremath{\Conid{Functor}}, \ensuremath{\Conid{Applicative}} and \ensuremath{\Conid{Monad}}
that can be found in Haskell's standard library because its additional complexity
is not needed in this article.}
\indentbegin \begin{hscode}\SaveRestoreHook
\column{B}{@{}>{\hspre}l<{\hspost}@{}}%
\column{3}{@{}>{\hspre}l<{\hspost}@{}}%
\column{E}{@{}>{\hspre}l<{\hspost}@{}}%
\>[3]{}(>\!\!>)\mathbin{::}\Conid{Monad}\;\Varid{m}\Rightarrow \Varid{m}\;\Varid{a}\to \Varid{m}\;\Varid{b}\to \Varid{m}\;\Varid{b}{}\<[E]%
\\
\>[3]{}\Varid{m}_{1}>\!\!>\Varid{m}_{2}\mathrel{=}\Varid{m}_{1}>\!\!>\!\!=\lambda \anonymous \to \Varid{m}_{2}{}\<[E]%
\\[\blanklineskip]%
\>[3]{}(\mathbin{\langle\hspace{1.6pt}\mathclap{\raisebox{0.1pt}{\scalebox{1}{$\ast$}}}\hspace{1.6pt}\rangle})\mathbin{::}\Conid{Monad}\;\Varid{m}\Rightarrow \Varid{m}\;(\Varid{a}\to \Varid{b})\to \Varid{m}\;\Varid{a}\to \Varid{m}\;\Varid{b}{}\<[E]%
\\
\>[3]{}\Varid{mf}\mathbin{\langle\hspace{1.6pt}\mathclap{\raisebox{0.1pt}{\scalebox{1}{$\ast$}}}\hspace{1.6pt}\rangle}\Varid{mx}\mathrel{=}\Varid{mf}>\!\!>\!\!=\lambda \Varid{f}\to \Varid{mx}>\!\!>\!\!=\lambda \Varid{x}\to \Varid{\eta}\;(\Varid{f}\;\Varid{x}){}\<[E]%
\ColumnHook
\end{hscode}\resethooks
\indentend %
\subsection{Nondeterminism and State}
\label{sec:nondeterminism}
\label{sec:state}

Following both the approaches of \citet{Hutton08} and of \citet{Gibbons11}, we introduce
effects as subclasses of the \ensuremath{\Conid{Monad}} type class.

\paragraph*{Nondeterminism}\
The first monadic effect we introduce is nondeterminism. We define a
subclass \ensuremath{\Conid{MNondet}} of \ensuremath{\Conid{Monad}} to capture the nondeterministic
interfaces as follows:
\indentbegin \begin{hscode}\SaveRestoreHook
\column{B}{@{}>{\hspre}l<{\hspost}@{}}%
\column{3}{@{}>{\hspre}l<{\hspost}@{}}%
\column{10}{@{}>{\hspre}l<{\hspost}@{}}%
\column{E}{@{}>{\hspre}l<{\hspost}@{}}%
\>[B]{}\mathbf{class}\;\Conid{Monad}\;\Varid{m}\Rightarrow \Conid{MNondet}\;\Varid{m}\;\mathbf{where}{}\<[E]%
\\
\>[B]{}\hsindent{3}{}\<[3]%
\>[3]{}\Varid{\varnothing}{}\<[10]%
\>[10]{}\mathbin{::}\Varid{m}\;\Varid{a}{}\<[E]%
\\
\>[B]{}\hsindent{3}{}\<[3]%
\>[3]{}(\talloblong){}\<[10]%
\>[10]{}\mathbin{::}\Varid{m}\;\Varid{a}\to \Varid{m}\;\Varid{a}\to \Varid{m}\;\Varid{a}{}\<[E]%
\ColumnHook
\end{hscode}\resethooks
\indentend Here, \ensuremath{\Varid{\varnothing}} denotes failures and \ensuremath{(\talloblong)} denotes nondeterministic
choices.  Instances of the \ensuremath{\Conid{MNondet}} interface should satisfy the
following four laws:
\footnote{
One might expect additional laws such as idempotence or commutativity.
As argued by \cite{Kiselyov:15:Laws}, these laws differ depending on
how the monad is used and how it should interact with other effects.
We choose to present a minimal setting for nondeterminism here.}

\begin{alignat}{2}
    &\mbox{\bf identity}:\quad &
      \ensuremath{\Varid{\varnothing}\mathbin{\talloblong}\Varid{m}} ~=~ & \ensuremath{\Varid{m}} ~=~ \ensuremath{\Varid{m}\mathbin{\talloblong}\Varid{\varnothing}}\mbox{~~,}
      \label{eq:mzero}\\
    &\mbox{\bf associativity}:~ &
      \ensuremath{(\Varid{m}\mathbin{\talloblong}\Varid{n})\mathbin{\talloblong}\Varid{k}}~ &=~ \ensuremath{\Varid{m}\mathbin{\talloblong}(\Varid{n}\mathbin{\talloblong}\Varid{k})} \mbox{~~,}
      \label{eq:mplus-assoc}\\
    &\mbox{\bf right-distributivity}:~ &
      \ensuremath{(\Varid{m}_{1}\mathbin{\talloblong}\Varid{m}_{2})>\!\!>\!\!=\Varid{f}} ~&=~ \ensuremath{(\Varid{m}_{1}>\!\!>\!\!=\Varid{f})\mathbin{\talloblong}(\Varid{m}_{2}>\!\!>\!\!=\Varid{f})} \mbox{~~,}
      \label{eq:mplus-dist}\\
    &\mbox{\bf left-identity}:\quad &
      \ensuremath{\Varid{\varnothing}>\!\!>\!\!=\Varid{f}} ~&=~ \ensuremath{\Varid{\varnothing}} \label{eq:mzero-zero} \mbox{~~.}
\end{alignat}
The first two laws state that \ensuremath{(\talloblong)} and \ensuremath{\Varid{\varnothing}} should form a monoid,
i.e., \ensuremath{(\talloblong)} should be associative with \ensuremath{\Varid{\varnothing}} as its neutral element.
The last two laws show that \ensuremath{(>\!\!>\!\!=)} is right-distributive
over \ensuremath{(\talloblong)} and that \ensuremath{\Varid{\varnothing}} cancels bind on the left.

The approach of \citet{Gibbons11} is to reason about effectful
programs using an axiomatic characterisation given by these laws. It
does not rely on the specific implementation of any particular
instance of \ensuremath{\Conid{MNondet}}.
In contrast, \citet{Hutton08} reason directly in terms of a particular
instance.  In the case of \ensuremath{\Conid{MNondet}}, the quintessential instance is
the list monad, which extends the conventional \ensuremath{\Conid{Monad}} instance for lists.

\begin{minipage}{0.5\textwidth}
\indentbegin \begin{hscode}\SaveRestoreHook
\column{B}{@{}>{\hspre}l<{\hspost}@{}}%
\column{3}{@{}>{\hspre}l<{\hspost}@{}}%
\column{10}{@{}>{\hspre}c<{\hspost}@{}}%
\column{10E}{@{}l@{}}%
\column{13}{@{}>{\hspre}l<{\hspost}@{}}%
\column{E}{@{}>{\hspre}l<{\hspost}@{}}%
\>[B]{}\mathbf{instance}\;\Conid{MNondet}\;[\mskip1.5mu \mskip1.5mu]\;\mathbf{where}{}\<[E]%
\\
\>[B]{}\hsindent{3}{}\<[3]%
\>[3]{}\Varid{\varnothing}{}\<[10]%
\>[10]{}\mathrel{=}{}\<[10E]%
\>[13]{}[\mskip1.5mu \mskip1.5mu]{}\<[E]%
\\
\>[B]{}\hsindent{3}{}\<[3]%
\>[3]{}(\talloblong){}\<[10]%
\>[10]{}\mathrel{=}{}\<[10E]%
\>[13]{}(+\!\!+){}\<[E]%
\ColumnHook
\end{hscode}\resethooks
\indentend \end{minipage}
\begin{minipage}{0.5\textwidth}\indentbegin \begin{hscode}\SaveRestoreHook
\column{B}{@{}>{\hspre}l<{\hspost}@{}}%
\column{3}{@{}>{\hspre}l<{\hspost}@{}}%
\column{5}{@{}>{\hspre}l<{\hspost}@{}}%
\column{16}{@{}>{\hspre}l<{\hspost}@{}}%
\column{17}{@{}>{\hspre}l<{\hspost}@{}}%
\column{E}{@{}>{\hspre}l<{\hspost}@{}}%
\>[3]{}\mathbf{instance}\;\Conid{Monad}\;[\mskip1.5mu \mskip1.5mu]\;\mathbf{where}{}\<[E]%
\\
\>[3]{}\hsindent{2}{}\<[5]%
\>[5]{}\Varid{\eta}\;\Varid{x}{}\<[16]%
\>[16]{}\mathrel{=}[\mskip1.5mu \Varid{x}\mskip1.5mu]{}\<[E]%
\\
\>[3]{}\hsindent{2}{}\<[5]%
\>[5]{}\Varid{xs}>\!\!>\!\!=\Varid{f}{}\<[17]%
\>[17]{}\mathrel{=}\Varid{concatMap}\;\Varid{f}\;\Varid{xs}{}\<[E]%
\ColumnHook
\end{hscode}\resethooks
\indentend \end{minipage}

\paragraph*{State}\
The signature for the state effect has two operations:
a \ensuremath{\Varid{get}} operation that reads and returns the state,
and a \ensuremath{\Varid{put}} operation that modifies the state, overwriting it with the given
value, and returns nothing.
Again, we define a subclass \ensuremath{\Conid{MState}} of \ensuremath{\Conid{Monad}} to capture its
interfaces.

\indentbegin \begin{hscode}\SaveRestoreHook
\column{B}{@{}>{\hspre}l<{\hspost}@{}}%
\column{5}{@{}>{\hspre}l<{\hspost}@{}}%
\column{E}{@{}>{\hspre}l<{\hspost}@{}}%
\>[B]{}\mathbf{class}\;\Conid{Monad}\;\Varid{m}\Rightarrow \Conid{MState}\;\Varid{s}\;\Varid{m}\mid \Varid{m}\to \Varid{s}\;\mathbf{where}{}\<[E]%
\\
\>[B]{}\hsindent{5}{}\<[5]%
\>[5]{}\Varid{get}\mathbin{::}\Varid{m}\;\Varid{s}{}\<[E]%
\\
\>[B]{}\hsindent{5}{}\<[5]%
\>[5]{}\Varid{put}\mathbin{::}\Varid{s}\to \Varid{m}\;(){}\<[E]%
\ColumnHook
\end{hscode}\resethooks
\indentend These operations are regulated by the following four laws:
\begin{alignat}{2}
    &\mbox{\bf put-put}:\quad &
    \ensuremath{\Varid{put}\;\Varid{s}>\!\!>\Varid{put}\;\Varid{s'}} &= \ensuremath{\Varid{put}\;\Varid{s'}}~~\mbox{,} \label{eq:put-put}\\
    &\mbox{\bf put-get}:~ &
    \ensuremath{\Varid{put}\;\Varid{s}>\!\!>\Varid{get}} &= \ensuremath{\Varid{put}\;\Varid{s}>\!\!>\Varid{\eta}\;\Varid{s}} ~~\mbox{,} \label{eq:put-get}\\
    &\mbox{\bf get-put}:~ &
    \ensuremath{\Varid{get}>\!\!>\!\!=\Varid{put}} &= \ensuremath{\Varid{\eta}\;()} ~~\mbox{,} \label{eq:get-put}\\
    &\mbox{\bf get-get}:\quad &
    \ensuremath{\Varid{get}>\!\!>\!\!=(\lambda \Varid{s}\to \Varid{get}>\!\!>\!\!=\Varid{k}\;\Varid{s})} &= \ensuremath{\Varid{get}>\!\!>\!\!=(\lambda \Varid{s}\to \Varid{k}\;\Varid{s}\;\Varid{s})}
    ~~\mbox{.} \label{eq:get-get}
\end{alignat}
The standard instance of \ensuremath{\Conid{MState}} is the state monad \ensuremath{\Conid{State}\;\Varid{s}}.

\begin{minipage}[t]{0.4\textwidth}
\indentbegin \begin{hscode}\SaveRestoreHook
\column{B}{@{}>{\hspre}l<{\hspost}@{}}%
\column{3}{@{}>{\hspre}l<{\hspost}@{}}%
\column{10}{@{}>{\hspre}l<{\hspost}@{}}%
\column{E}{@{}>{\hspre}l<{\hspost}@{}}%
\>[B]{}\mathbf{newtype}\;\Conid{State}\;\Varid{s}\;\Varid{a}\mathrel{=}{}\<[E]%
\\
\>[B]{}\hsindent{3}{}\<[3]%
\>[3]{}\Conid{State}\;\{\mskip1.5mu \Varid{run_{State}}\mathbin{::}\Varid{s}\to (\Varid{a},\Varid{s})\mskip1.5mu\}{}\<[E]%
\\
\>[B]{}\mathbf{instance}\;\Conid{MState}\;\Varid{s}\;(\Conid{State}\;\Varid{s})\;\mathbf{where}{}\<[E]%
\\
\>[B]{}\hsindent{3}{}\<[3]%
\>[3]{}\Varid{get}{}\<[10]%
\>[10]{}\mathrel{=}\Conid{State}\;(\lambda \Varid{s}\to (\Varid{s},\Varid{s})){}\<[E]%
\\
\>[B]{}\hsindent{3}{}\<[3]%
\>[3]{}\Varid{put}\;\Varid{s}{}\<[10]%
\>[10]{}\mathrel{=}\Conid{State}\;(\lambda \anonymous \to ((),\Varid{s})){}\<[E]%
\ColumnHook
\end{hscode}\resethooks
\indentend \end{minipage}
\begin{minipage}[t]{0.6\textwidth}
\indentbegin \begin{hscode}\SaveRestoreHook
\column{B}{@{}>{\hspre}l<{\hspost}@{}}%
\column{3}{@{}>{\hspre}l<{\hspost}@{}}%
\column{12}{@{}>{\hspre}l<{\hspost}@{}}%
\column{28}{@{}>{\hspre}l<{\hspost}@{}}%
\column{E}{@{}>{\hspre}l<{\hspost}@{}}%
\>[B]{}\mathbf{instance}\;\Conid{Monad}\;(\Conid{State}\;\Varid{s})\;\mathbf{where}{}\<[E]%
\\
\>[B]{}\hsindent{3}{}\<[3]%
\>[3]{}\Varid{\eta}\;\Varid{x}\mathrel{=}\Conid{State}\;(\lambda \Varid{s}\to (\Varid{x},\Varid{s})){}\<[E]%
\\
\>[B]{}\hsindent{3}{}\<[3]%
\>[3]{}\Varid{m}>\!\!>\!\!=\Varid{f}{}\<[12]%
\>[12]{}\mathrel{=}\Conid{State}\;(\lambda \Varid{s}\to {}\<[28]%
\>[28]{}\mathbf{let}\;(\Varid{x},\Varid{s'})\mathrel{=}\Varid{run_{State}}\;\Varid{m}\;\Varid{s}{}\<[E]%
\\
\>[28]{}\mathbf{in}\;\Varid{run_{State}}\;(\Varid{f}\;\Varid{x})\;\Varid{s'}){}\<[E]%
\ColumnHook
\end{hscode}\resethooks
\indentend \end{minipage}

\subsection{The N-queens Puzzle}
\label{sec:motivation-and-challenges}

The n-queens problem used here is an adapted and simplified version from that of
\cite{Gibbons11}.
The aim of the puzzle is to place $n$ queens on a $n \times n$ chess board such
that no two queens can attack each other.
Given $n$, we number the rows and columns by \ensuremath{[\mskip1.5mu \mathrm{1}\mathinner{\ldotp\ldotp}\Varid{n}\mskip1.5mu]}.
Since all queens should be placed on distinct rows and distinct columns, a
potential solution can be represented by a permutation \ensuremath{\Varid{xs}} of the list \ensuremath{[\mskip1.5mu \mathrm{1}\mathinner{\ldotp\ldotp}\Varid{n}\mskip1.5mu]},
such that \ensuremath{\Varid{xs}\mathbin{!!}\Varid{i}\mathrel{=}\Varid{j}} denotes that the queen on the \ensuremath{\Varid{i}}th column is placed on
the \ensuremath{\Varid{j}}th row.
Using this representation, queens cannot be put on the same row or column.

\paragraph*{A Naive Algorithm}\
We have the following naive nondeterministic algorithm for n-queens.
\indentbegin \begin{hscode}\SaveRestoreHook
\column{B}{@{}>{\hspre}l<{\hspost}@{}}%
\column{E}{@{}>{\hspre}l<{\hspost}@{}}%
\>[B]{}\Varid{queens_{naive}}\mathbin{::}\Conid{MNondet}\;\Varid{m}\Rightarrow \Conid{Int}\to \Varid{m}\;[\mskip1.5mu \Conid{Int}\mskip1.5mu]{}\<[E]%
\\
\>[B]{}\Varid{queens_{naive}}\;\Varid{n}\mathrel{=}\Varid{choose}\;(\Varid{permutations}\;[\mskip1.5mu \mathrm{1}\mathinner{\ldotp\ldotp}\Varid{n}\mskip1.5mu])>\!\!>\!\!=\Varid{filtr}\;\Varid{valid}{}\<[E]%
\ColumnHook
\end{hscode}\resethooks
\indentend The program \ensuremath{\Varid{queens_{naive}}\;\mathrm{4}\mathbin{::}[\mskip1.5mu [\mskip1.5mu \Conid{Int}\mskip1.5mu]\mskip1.5mu]} gives as result \ensuremath{[\mskip1.5mu [\mskip1.5mu \mathrm{2},\mathrm{4},\mathrm{1},\mathrm{3}\mskip1.5mu],[\mskip1.5mu \mathrm{3},\mathrm{1},\mathrm{4},\mathrm{2}\mskip1.5mu]\mskip1.5mu]}.  The program uses a generate-and-test strategy: it
generates all permutations of queens as candiate solutions, and then
tests which ones are valid.

The function \ensuremath{\Varid{permutations}\mathbin{::}[\mskip1.5mu \Varid{a}\mskip1.5mu]\to [\mskip1.5mu [\mskip1.5mu \Varid{a}\mskip1.5mu]\mskip1.5mu]} from \ensuremath{\Conid{\Conid{Data}.List}} computes
all the permutations of its input.
The function \ensuremath{\Varid{choose}} implemented as follows nondeterministically
picks an element from a list.

\indentbegin \begin{hscode}\SaveRestoreHook
\column{B}{@{}>{\hspre}l<{\hspost}@{}}%
\column{E}{@{}>{\hspre}l<{\hspost}@{}}%
\>[B]{}\Varid{choose}\mathbin{::}\Conid{MNondet}\;\Varid{m}\Rightarrow [\mskip1.5mu \Varid{a}\mskip1.5mu]\to \Varid{m}\;\Varid{a}{}\<[E]%
\\
\>[B]{}\Varid{choose}\mathrel{=}\Varid{foldr}\;((\talloblong)\hsdot{\circ }{.}\Varid{\eta})\;\Varid{\varnothing}{}\<[E]%
\ColumnHook
\end{hscode}\resethooks
\indentend 
The function \ensuremath{\Varid{filtr}\;\Varid{p}\;\Varid{x}} returns \ensuremath{\Varid{x}} if \ensuremath{\Varid{p}\;\Varid{x}} holds, and fails otherwise.
\indentbegin \begin{hscode}\SaveRestoreHook
\column{B}{@{}>{\hspre}l<{\hspost}@{}}%
\column{E}{@{}>{\hspre}l<{\hspost}@{}}%
\>[B]{}\Varid{filtr}\mathbin{::}\Conid{MNondet}\;\Varid{m}\Rightarrow (\Varid{a}\to \Conid{Bool})\to \Varid{a}\to \Varid{m}\;\Varid{a}{}\<[E]%
\\
\>[B]{}\Varid{filtr}\;\Varid{p}\;\Varid{x}\mathrel{=}\mathbf{if}\;\Varid{p}\;\Varid{x}\;\mathbf{then}\;\Varid{\eta}\;\Varid{x}\;\mathbf{else}\;\Varid{\varnothing}{}\<[E]%
\ColumnHook
\end{hscode}\resethooks
\indentend 
The pure function \ensuremath{\Varid{valid}\mathbin{::}[\mskip1.5mu \Conid{Int}\mskip1.5mu]\to \Conid{Bool}} determines whether the
input is a valid solution. 
\indentbegin \begin{hscode}\SaveRestoreHook
\column{B}{@{}>{\hspre}l<{\hspost}@{}}%
\column{15}{@{}>{\hspre}l<{\hspost}@{}}%
\column{E}{@{}>{\hspre}l<{\hspost}@{}}%
\>[B]{}\Varid{valid}\mathbin{::}[\mskip1.5mu \Conid{Int}\mskip1.5mu]\to \Conid{Bool}{}\<[E]%
\\
\>[B]{}\Varid{valid}\;[\mskip1.5mu \mskip1.5mu]{}\<[15]%
\>[15]{}\mathrel{=}\Conid{True}{}\<[E]%
\\
\>[B]{}\Varid{valid}\;(\Varid{q}\mathbin{:}\Varid{qs}){}\<[15]%
\>[15]{}\mathrel{=}\Varid{valid}\;\Varid{qs}\mathrel{\wedge}\Varid{safe}\;\Varid{q}\;\mathrm{1}\;\Varid{qs}{}\<[E]%
\ColumnHook
\end{hscode}\resethooks
\indentend A solution is valid when each queen is \ensuremath{\Varid{safe}} with respect
to the subsequent queens:
\indentbegin \begin{hscode}\SaveRestoreHook
\column{B}{@{}>{\hspre}l<{\hspost}@{}}%
\column{19}{@{}>{\hspre}l<{\hspost}@{}}%
\column{E}{@{}>{\hspre}l<{\hspost}@{}}%
\>[B]{}\Varid{safe}\mathbin{::}\Conid{Int}\to \Conid{Int}\to [\mskip1.5mu \Conid{Int}\mskip1.5mu]\to \Conid{Bool}{}\<[E]%
\\
\>[B]{}\Varid{safe}\;\anonymous \;\anonymous \;[\mskip1.5mu \mskip1.5mu]{}\<[19]%
\>[19]{}\mathrel{=}\Conid{True}{}\<[E]%
\\
\>[B]{}\Varid{safe}\;\Varid{q}\;\Varid{n}\;(\Varid{q}_{1}\mathbin{:}\Varid{qs}){}\<[19]%
\>[19]{}\mathrel{=}\Varid{and}\;[\mskip1.5mu \Varid{q}\not\equiv \Varid{q}_{1},\Varid{q}\not\equiv \Varid{q}_{1}\mathbin{+}\Varid{n},\Varid{q}\not\equiv \Varid{q}_{1}\mathbin{-}\Varid{n},\Varid{safe}\;\Varid{q}\;(\Varid{n}\mathbin{+}\mathrm{1})\;\Varid{qs}\mskip1.5mu]{}\<[E]%
\ColumnHook
\end{hscode}\resethooks
\indentend The call \ensuremath{\Varid{safe}\;\Varid{q}\;\Varid{n}\;\Varid{qs}} checks whether the current queen \ensuremath{\Varid{q}} is on a different ascending and descending diagonal
than the other queens \ensuremath{\Varid{qs}}, where \ensuremath{\Varid{n}} is the number of columns that \ensuremath{\Varid{q}} is apart from the first queen \ensuremath{\Varid{q}_{1}} in \ensuremath{\Varid{qs}}.

Although this generate-and-test approach works and is quite intuitive, it is not very efficient.
For example, all solutions of the form \ensuremath{(\mathrm{1}\mathbin{:}\mathrm{2}\mathbin{:}\Varid{qs})} are invalid because the first two queens are on the same diagonal.
However, the algorithm still needs to generate and test all $(n-2)!$ candidate solutions of this form.

\paragraph*{A Backtracking Algorithm}\
We can fuse the two phases of the naive algorithm to obtain a more
efficient algorithm, where both generating candidates and checking for
validity happens in a single pass.  The idea is to move to a
state-based backtracking implementation that allows early pruning of
branches that are invalid.
In particular, when placing the new queen in the next column, we make
sure that it is only placed in positions that are valid with respect
to the previously placed queens.

We use a state \ensuremath{(\Conid{Int},[\mskip1.5mu \Conid{Int}\mskip1.5mu])} to contain the current column and the
previously placed queens.  The backtracking algorithm of n-queens is
implemented as follows.
\indentbegin \begin{hscode}\SaveRestoreHook
\column{B}{@{}>{\hspre}l<{\hspost}@{}}%
\column{3}{@{}>{\hspre}l<{\hspost}@{}}%
\column{14}{@{}>{\hspre}l<{\hspost}@{}}%
\column{23}{@{}>{\hspre}l<{\hspost}@{}}%
\column{E}{@{}>{\hspre}l<{\hspost}@{}}%
\>[B]{}\Varid{queens}\mathbin{::}(\Conid{MState}\;(\Conid{Int},[\mskip1.5mu \Conid{Int}\mskip1.5mu])\;\Varid{m},\Conid{MNondet}\;\Varid{m})\Rightarrow \Conid{Int}\to \Varid{m}\;[\mskip1.5mu \Conid{Int}\mskip1.5mu]{}\<[E]%
\\
\>[B]{}\Varid{queens}\;\Varid{n}\mathrel{=}\Varid{loop}\;\mathbf{where}{}\<[E]%
\\
\>[B]{}\hsindent{3}{}\<[3]%
\>[3]{}\Varid{loop}\mathrel{=}\mathbf{do}\;{}\<[14]%
\>[14]{}(\Varid{c},\Varid{sol})\leftarrow \Varid{get}{}\<[E]%
\\
\>[14]{}\mathbf{if}\;\Varid{c}\geq \Varid{n}\;\mathbf{then}\;\Varid{\eta}\;\Varid{sol}{}\<[E]%
\\
\>[14]{}\mathbf{else}\;\mathbf{do}\;{}\<[23]%
\>[23]{}\Varid{r}\leftarrow \Varid{choose}\;[\mskip1.5mu \mathrm{1}\mathinner{\ldotp\ldotp}\Varid{n}\mskip1.5mu]{}\<[E]%
\\
\>[23]{}\Varid{guard}\;(\Varid{safe}\;\Varid{r}\;\mathrm{1}\;\Varid{sol}){}\<[E]%
\\
\>[23]{}\Varid{s}\leftarrow \Varid{get}{}\<[E]%
\\
\>[23]{}\Varid{put}\;(\Varid{s}\mathbin{\oplus}\Varid{r}){}\<[E]%
\\
\>[23]{}\Varid{loop}{}\<[E]%
\ColumnHook
\end{hscode}\resethooks
\indentend 
The function \ensuremath{\Varid{guard}} fails when the input is false.
\indentbegin \begin{hscode}\SaveRestoreHook
\column{B}{@{}>{\hspre}l<{\hspost}@{}}%
\column{13}{@{}>{\hspre}l<{\hspost}@{}}%
\column{E}{@{}>{\hspre}l<{\hspost}@{}}%
\>[B]{}\Varid{guard}\mathbin{::}\Conid{MNondet}\;\Varid{m}\Rightarrow \Conid{Bool}\to \Varid{m}\;(){}\<[E]%
\\
\>[B]{}\Varid{guard}\;\Conid{True}{}\<[13]%
\>[13]{}\mathrel{=}\Varid{\eta}\;(){}\<[E]%
\\
\>[B]{}\Varid{guard}\;\Conid{False}\mathrel{=}\Varid{\varnothing}{}\<[E]%
\ColumnHook
\end{hscode}\resethooks
\indentend %
The function \ensuremath{\Varid{s}\mathbin{\oplus}\Varid{r}} updates the state with a new queen placed on
row \ensuremath{\Varid{r}} in the next column.
\indentbegin \begin{hscode}\SaveRestoreHook
\column{B}{@{}>{\hspre}l<{\hspost}@{}}%
\column{8}{@{}>{\hspre}l<{\hspost}@{}}%
\column{E}{@{}>{\hspre}l<{\hspost}@{}}%
\>[B]{}(\mathbin{\oplus}){}\<[8]%
\>[8]{}\mathbin{::}(\Conid{Int},[\mskip1.5mu \Conid{Int}\mskip1.5mu])\to \Conid{Int}\to (\Conid{Int},[\mskip1.5mu \Conid{Int}\mskip1.5mu]){}\<[E]%
\\
\>[B]{}(\mathbin{\oplus})\;{}\<[8]%
\>[8]{}(\Varid{c},\Varid{sol})\;\Varid{r}\mathrel{=}(\Varid{c}\mathbin{+}\mathrm{1},\Varid{r}\mathbin{:}\Varid{sol}){}\<[E]%
\ColumnHook
\end{hscode}\resethooks
\indentend %

The above monadic version of \ensuremath{\Varid{queens}} essentially assume that each
searching branch has its own state; we do not need to explicitly
restore the state when backtracking.
Though it is a convenient high-level programming assumption for
programmers, it causes obstacles to low-level implementations and
optimisations.
In the following sections, we investigate how low-level implementation
and optimisation techniques, such as those found in Prolog's Warren
Abstract Machine and Constraint Programming systems, can be
incorporated and proved correct.

\section{Algebraic Effects and Handlers}
\label{sec:overview}

This section introduces algebraic effects and handlers, the approach
we use to define syntax, semantics, and simulations for effects.
Comparing to giving concrete monad implementations for effects,
algebraic effects and handlers allow us to easily provide different
interpretations for the same effects due to the clear separation of
syntax and semantics.
As a result, we can smoothly specify translations from high-level
effects to low-level effects as handlers of these high-level effects,
and then compose them with the handlers of low-level effects to
interpret high-level programs.
Algebraic effects and handlers also provide us with a modular way to
combine our translations with other effects, and a useful tool, the
fusion property, to prove the correctness of translations.

\subsection{Free Monads and Their Folds}
\label{sec:free-monads-and-their-folds}

We implement algebraic effects and handlers as free monads and their folds.

\paragraph*{Free Monads}\
Free monads are gaining popularity for their use in algebraic effects \citep{Plotkin02, PlotkinP03}
and handlers \citep{Plotkin09, Plotkin13},
which elegantly separate syntax and semantics of effectful
operations.
A free monad, the syntax of an effectful program,
can be captured generically in Haskell.
\indentbegin \begin{hscode}\SaveRestoreHook
\column{B}{@{}>{\hspre}l<{\hspost}@{}}%
\column{E}{@{}>{\hspre}l<{\hspost}@{}}%
\>[B]{}\mathbf{data}\;\Conid{Free}\;\Varid{f}\;\Varid{a}\mathrel{=}\Conid{Var}\;\Varid{a}\mid \Conid{Op}\;(\Varid{f}\;(\Conid{Free}\;\Varid{f}\;\Varid{a})){}\<[E]%
\ColumnHook
\end{hscode}\resethooks
\indentend This data type is a form of abstract syntax tree (AST) consisting of leaves (\ensuremath{\Conid{Var}\;\Varid{a}})
and internal nodes (\ensuremath{\Conid{Op}\;(\Varid{f}\;(\Conid{Free}\;\Varid{f}\;\Varid{a}))}), whose branching structure is
determined by the functor \ensuremath{\Varid{f}}.
This functor is also known as the \emph{signature} of operations.

\paragraph*{A Fold Recursion Scheme}\
Free monads %
come equipped with a fold recursion scheme.
\indentbegin \begin{hscode}\SaveRestoreHook
\column{B}{@{}>{\hspre}l<{\hspost}@{}}%
\column{23}{@{}>{\hspre}c<{\hspost}@{}}%
\column{23E}{@{}l@{}}%
\column{26}{@{}>{\hspre}l<{\hspost}@{}}%
\column{E}{@{}>{\hspre}l<{\hspost}@{}}%
\>[B]{}\Varid{fold}\mathbin{::}\Conid{Functor}\;\Varid{f}\Rightarrow (\Varid{a}\to \Varid{b})\to (\Varid{f}\;\Varid{b}\to \Varid{b})\to \Conid{Free}\;\Varid{f}\;\Varid{a}\to \Varid{b}{}\<[E]%
\\
\>[B]{}\Varid{fold}\;\Varid{gen}\;\Varid{alg}\;(\Conid{Var}\;\Varid{x}){}\<[23]%
\>[23]{}\mathrel{=}{}\<[23E]%
\>[26]{}\Varid{gen}\;\Varid{x}{}\<[E]%
\\
\>[B]{}\Varid{fold}\;\Varid{gen}\;\Varid{alg}\;(\Conid{Op}\;\Varid{op}){}\<[23]%
\>[23]{}\mathrel{=}{}\<[23E]%
\>[26]{}\Varid{alg}\;(\Varid{fmap}\;(\Varid{fold}\;\Varid{gen}\;\Varid{alg})\;\Varid{op}){}\<[E]%
\ColumnHook
\end{hscode}\resethooks
\indentend This fold interprets an AST structure of type \ensuremath{\Conid{Free}\;\Varid{f}\;\Varid{a}} into some
semantic domain \ensuremath{\Varid{b}}. It does so compositionally using a generator
\ensuremath{\Varid{gen}\mathbin{::}\Varid{a}\to \Varid{b}} for the leaves and an algebra \ensuremath{\Varid{alg}\mathbin{::}\Varid{f}\;\Varid{b}\to \Varid{b}} for the internal
nodes; together these are also known as a \emph{handler}.

The monad instance of \ensuremath{\Conid{Free}} is straightforwardly implemented with fold.
\indentbegin \begin{hscode}\SaveRestoreHook
\column{B}{@{}>{\hspre}l<{\hspost}@{}}%
\column{5}{@{}>{\hspre}l<{\hspost}@{}}%
\column{14}{@{}>{\hspre}l<{\hspost}@{}}%
\column{E}{@{}>{\hspre}l<{\hspost}@{}}%
\>[B]{}\mathbf{instance}\;\Conid{Functor}\;\Varid{f}\Rightarrow \Conid{Monad}\;(\Conid{Free}\;\Varid{f})\;\mathbf{where}{}\<[E]%
\\
\>[B]{}\hsindent{5}{}\<[5]%
\>[5]{}\Varid{\eta}{}\<[14]%
\>[14]{}\mathrel{=}\Conid{Var}{}\<[E]%
\\
\>[B]{}\hsindent{5}{}\<[5]%
\>[5]{}\Varid{m}>\!\!>\!\!=\Varid{f}{}\<[14]%
\>[14]{}\mathrel{=}\Varid{fold}\;\Varid{f}\;\Conid{Op}\;\Varid{m}{}\<[E]%
\ColumnHook
\end{hscode}\resethooks
\indentend 

Under certain conditions folds can be fused with functions that are
composed with them~\citep{Wu15, Gibbons00}.
This gives rise to the following laws:
\begin{alignat}{2}
    &\mbox{\bf fusion-pre}: &
    \ensuremath{\Varid{fold}\;(\Varid{gen}\hsdot{\circ }{.}\Varid{h})\;\Varid{alg}} &= \ensuremath{\Varid{fold}\;\Varid{gen}\;\Varid{alg}\hsdot{\circ }{.}\Varid{fmap}\;\Varid{h}}\mbox{~~,} \label{eq:fusion-pre}\\
    &\mbox{\bf fusion-post}: &
    \ensuremath{\Varid{h}\hsdot{\circ }{.}\Varid{fold}\;\Varid{gen}\;\Varid{alg}} &= \ensuremath{\Varid{fold}\;(\Varid{h}\hsdot{\circ }{.}\Varid{gen})\;\Varid{alg'}} \text{ with } \ensuremath{\Varid{h}\hsdot{\circ }{.}\Varid{alg}\mathrel{=}\Varid{alg'}\hsdot{\circ }{.}\Varid{fmap}\;\Varid{h}} \label{eq:fusion-post}\mbox{~~,}\\
    &\mbox{\bf fusion-post'}: &
    \ensuremath{\Varid{h}\hsdot{\circ }{.}\Varid{fold}\;\Varid{gen}\;\Varid{alg}} &= \ensuremath{\Varid{fold}\;(\Varid{h}\hsdot{\circ }{.}\Varid{gen})\;\Varid{alg'}} \label{eq:fusion-post-strong}\\
    & & \span\qquad \text{ with } \ensuremath{\Varid{h}\hsdot{\circ }{.}\Varid{alg}\hsdot{\circ }{.}\Varid{fmap}\;\Varid{f}\mathrel{=}\Varid{alg'}\hsdot{\circ }{.}\Varid{fmap}\;\Varid{h}\hsdot{\circ }{.}\Varid{fmap}\;\Varid{f}} \text{ and } \ensuremath{\Varid{f}\mathrel{=}\Varid{fold}\;\Varid{gen}\;\Varid{alg}} \mbox{~~.}\nonumber
\end{alignat}
These three fusion laws turn out to be essential in the further proofs of this paper.

\paragraph*{Nondeterminism}\
Instead of using a concrete monad like \ensuremath{\Conid{List}}, we use the free monad
\ensuremath{\Conid{Free}\;\Varid{Nondet_{F}}} over the signature \ensuremath{\Varid{Nondet_{F}}} following algebraic effects.
\indentbegin \begin{hscode}\SaveRestoreHook
\column{B}{@{}>{\hspre}l<{\hspost}@{}}%
\column{18}{@{}>{\hspre}l<{\hspost}@{}}%
\column{E}{@{}>{\hspre}l<{\hspost}@{}}%
\>[B]{}\mathbf{data}\;\Varid{Nondet_{F}}\;\Varid{a}{}\<[18]%
\>[18]{}\mathrel{=}\Conid{Fail}\mid \Conid{Or}\;\Varid{a}\;\Varid{a}{}\<[E]%
\ColumnHook
\end{hscode}\resethooks
\indentend This signatures gives rise to a trivial \ensuremath{\Conid{MNondet}} instance:
\indentbegin \begin{hscode}\SaveRestoreHook
\column{B}{@{}>{\hspre}l<{\hspost}@{}}%
\column{3}{@{}>{\hspre}l<{\hspost}@{}}%
\column{14}{@{}>{\hspre}l<{\hspost}@{}}%
\column{E}{@{}>{\hspre}l<{\hspost}@{}}%
\>[B]{}\mathbf{instance}\;\Conid{MNondet}\;(\Conid{Free}\;\Varid{Nondet_{F}})\;\mathbf{where}{}\<[E]%
\\
\>[B]{}\hsindent{3}{}\<[3]%
\>[3]{}\Varid{\varnothing}{}\<[14]%
\>[14]{}\mathrel{=}\Conid{Op}\;\Conid{Fail}{}\<[E]%
\\
\>[B]{}\hsindent{3}{}\<[3]%
\>[3]{}(\talloblong)\;\Varid{p}\;\Varid{q}{}\<[14]%
\>[14]{}\mathrel{=}\Conid{Op}\;(\Conid{Or}\;\Varid{p}\;\Varid{q}){}\<[E]%
\ColumnHook
\end{hscode}\resethooks
\indentend With this representation the 
\textbf{right-distributivity} law and the \textbf{left-identity} law 
follow trivially from the definition of \ensuremath{(>\!\!>\!\!=)} for the free monad.

In contrast, the \textbf{identity} and \textbf{associativity} laws are not satisfied on the nose. Indeed,
\ensuremath{\Conid{Op}\;(\Conid{Or}\;\Conid{Fail}\;\Varid{p})} is for instance a different abstract syntax tree than \ensuremath{\Varid{p}}.
Yet, these syntactic differences do not matter as long as their interpretation
is the same. This is where the handlers come in; the meaning they assign
to effectful programs should respect the laws.
We have the following \ensuremath{\Varid{h_{ND}}} handler which interprets the free monad in
terms of lists.
\indentbegin \begin{hscode}\SaveRestoreHook
\column{B}{@{}>{\hspre}l<{\hspost}@{}}%
\column{3}{@{}>{\hspre}l<{\hspost}@{}}%
\column{5}{@{}>{\hspre}l<{\hspost}@{}}%
\column{21}{@{}>{\hspre}l<{\hspost}@{}}%
\column{E}{@{}>{\hspre}l<{\hspost}@{}}%
\>[B]{}\Varid{h_{ND}}\mathbin{::}\Conid{Free}\;\Varid{Nondet_{F}}\;\Varid{a}\to [\mskip1.5mu \Varid{a}\mskip1.5mu]{}\<[E]%
\\
\>[B]{}\Varid{h_{ND}}\mathrel{=}\Varid{fold}\;\Varid{gen_{ND}}\;\Varid{alg_{ND}}{}\<[E]%
\\
\>[B]{}\hsindent{3}{}\<[3]%
\>[3]{}\mathbf{where}{}\<[E]%
\\
\>[3]{}\hsindent{2}{}\<[5]%
\>[5]{}\Varid{gen_{ND}}\;\Varid{x}{}\<[21]%
\>[21]{}\mathrel{=}[\mskip1.5mu \Varid{x}\mskip1.5mu]{}\<[E]%
\\
\>[3]{}\hsindent{2}{}\<[5]%
\>[5]{}\Varid{alg_{ND}}\;\Conid{Fail}{}\<[21]%
\>[21]{}\mathrel{=}[\mskip1.5mu \mskip1.5mu]{}\<[E]%
\\
\>[3]{}\hsindent{2}{}\<[5]%
\>[5]{}\Varid{alg_{ND}}\;(\Conid{Or}\;\Varid{p}\;\Varid{q}){}\<[21]%
\>[21]{}\mathrel{=}\Varid{p}+\!\!+\Varid{q}{}\<[E]%
\ColumnHook
\end{hscode}\resethooks
\indentend 
With this handler, the \textbf{identity} and \textbf{associativity} laws are 
satisfied up to handling as follows:
\begin{equation*}
\begin{array}{r@{}c@{}l}
      \ensuremath{\Varid{h_{ND}}\;(\Varid{\varnothing}\mathbin{\talloblong}\Varid{m})} & ~=~ & \ensuremath{\Varid{h_{ND}}\;\Varid{m}} ~=~ \ensuremath{\Varid{h_{ND}}\;(\Varid{m}\mathbin{\talloblong}\Varid{\varnothing})} \\
      \ensuremath{\Varid{h_{ND}}\;((\Varid{m}\mathbin{\talloblong}\Varid{n})\mathbin{\talloblong}\Varid{o})} & ~=~ & \ensuremath{\Varid{h_{ND}}\;(\Varid{m}\mathbin{\talloblong}(\Varid{n}\mathbin{\talloblong}\Varid{o}))}
\end{array}
\end{equation*}
In fact, two stronger \textit{contextual} equalities hold:
\begin{equation*}
\begin{array}{r@{}c@{}l}
      \ensuremath{\Varid{h_{ND}}\;((\Varid{\varnothing}\mathbin{\talloblong}\Varid{m})>\!\!>\!\!=\Varid{k})} & ~=~ & \ensuremath{\Varid{h_{ND}}\;(\Varid{m}>\!\!>\!\!=\Varid{k})} ~=~ \ensuremath{\Varid{h_{ND}}\;((\Varid{m}\mathbin{\talloblong}\Varid{\varnothing})>\!\!>\!\!=\Varid{k})} \\
      \ensuremath{\Varid{h_{ND}}\;(((\Varid{m}\mathbin{\talloblong}\Varid{n})\mathbin{\talloblong}\Varid{o})>\!\!>\!\!=\Varid{k})} & ~=~ & \ensuremath{\Varid{h_{ND}}\;((\Varid{m}\mathbin{\talloblong}(\Varid{n}\mathbin{\talloblong}\Varid{o}))>\!\!>\!\!=\Varid{k})}
\end{array}
\end{equation*}
These equations state that the intepretations of the left- and right-hand sides are
indistinguishable even when put in a larger program context \ensuremath{>\!\!>\!\!=\Varid{k}}. 
They follow from the definitions of \ensuremath{\Varid{h_{ND}}} and \ensuremath{(>\!\!>\!\!=)}, as well as the associativity
and identity proprerties of \ensuremath{(+\!\!+)}.

We obtain the two non-contextual equations as a corollary by choosing \ensuremath{\Varid{k}\mathrel{=}\Varid{\eta}}.

\paragraph*{State}\
Again, instead of using the concrete \ensuremath{\Conid{State}} monad in \Cref{sec:state},
we model states via the free monad \ensuremath{\Conid{Free}\;(\Varid{State_{F}}\;\Varid{s})} over the state
signature.

\indentbegin \begin{hscode}\SaveRestoreHook
\column{B}{@{}>{\hspre}l<{\hspost}@{}}%
\column{18}{@{}>{\hspre}l<{\hspost}@{}}%
\column{E}{@{}>{\hspre}l<{\hspost}@{}}%
\>[B]{}\mathbf{data}\;\Varid{State_{F}}\;\Varid{s}\;\Varid{a}{}\<[18]%
\>[18]{}\mathrel{=}\Conid{Get}\;(\Varid{s}\to \Varid{a})\mid \Conid{Put}\;\Varid{s}\;\Varid{a}{}\<[E]%
\ColumnHook
\end{hscode}\resethooks
\indentend This state signature gives the following \ensuremath{\Conid{MState}\;\Varid{s}} instance:
\indentbegin \begin{hscode}\SaveRestoreHook
\column{B}{@{}>{\hspre}l<{\hspost}@{}}%
\column{3}{@{}>{\hspre}l<{\hspost}@{}}%
\column{10}{@{}>{\hspre}c<{\hspost}@{}}%
\column{10E}{@{}l@{}}%
\column{13}{@{}>{\hspre}l<{\hspost}@{}}%
\column{E}{@{}>{\hspre}l<{\hspost}@{}}%
\>[B]{}\mathbf{instance}\;\Conid{MState}\;\Varid{s}\;(\Conid{Free}\;(\Varid{State_{F}}\;\Varid{s}))\;\mathbf{where}{}\<[E]%
\\
\>[B]{}\hsindent{3}{}\<[3]%
\>[3]{}\Varid{get}{}\<[10]%
\>[10]{}\mathrel{=}{}\<[10E]%
\>[13]{}\Conid{Op}\;(\Conid{Get}\;\Varid{\eta}){}\<[E]%
\\
\>[B]{}\hsindent{3}{}\<[3]%
\>[3]{}\Varid{put}\;\Varid{s}{}\<[10]%
\>[10]{}\mathrel{=}{}\<[10E]%
\>[13]{}\Conid{Op}\;(\Conid{Put}\;\Varid{s}\;(\Varid{\eta}\;())){}\<[E]%
\ColumnHook
\end{hscode}\resethooks
\indentend 
The following handler \ensuremath{\Varid{h_{State}^\prime}} maps this free monad to the \ensuremath{\Conid{State}\;\Varid{s}} monad.
\indentbegin \begin{hscode}\SaveRestoreHook
\column{B}{@{}>{\hspre}l<{\hspost}@{}}%
\column{3}{@{}>{\hspre}l<{\hspost}@{}}%
\column{5}{@{}>{\hspre}l<{\hspost}@{}}%
\column{20}{@{}>{\hspre}l<{\hspost}@{}}%
\column{24}{@{}>{\hspre}l<{\hspost}@{}}%
\column{E}{@{}>{\hspre}l<{\hspost}@{}}%
\>[B]{}\Varid{h_{State}^\prime}\mathbin{::}\Conid{Free}\;(\Varid{State_{F}}\;\Varid{s})\;\Varid{a}\to \Conid{State}\;\Varid{s}\;\Varid{a}{}\<[E]%
\\
\>[B]{}\Varid{h_{State}^\prime}\mathrel{=}\Varid{fold}\;\Varid{gen_{S}^\prime}\;\Varid{alg_{S}^\prime}{}\<[E]%
\\
\>[B]{}\hsindent{3}{}\<[3]%
\>[3]{}\mathbf{where}{}\<[E]%
\\
\>[3]{}\hsindent{2}{}\<[5]%
\>[5]{}\Varid{gen_{S}^\prime}\;\Varid{x}{}\<[24]%
\>[24]{}\mathrel{=}\Conid{State}\mathbin{\$}\lambda \Varid{s}\to (\Varid{x},\Varid{s}){}\<[E]%
\\
\>[3]{}\hsindent{2}{}\<[5]%
\>[5]{}\Varid{alg_{S}^\prime}\;(\Conid{Get}\;{}\<[20]%
\>[20]{}\Varid{k}){}\<[24]%
\>[24]{}\mathrel{=}\Conid{State}\mathbin{\$}\lambda \Varid{s}\to \Varid{run_{State}}\;(\Varid{k}\;\Varid{s})\;\Varid{s}{}\<[E]%
\\
\>[3]{}\hsindent{2}{}\<[5]%
\>[5]{}\Varid{alg_{S}^\prime}\;(\Conid{Put}\;\Varid{s'}\;{}\<[20]%
\>[20]{}\Varid{k}){}\<[24]%
\>[24]{}\mathrel{=}\Conid{State}\mathbin{\$}\lambda \Varid{s}\to \Varid{run_{State}}\;\Varid{k}\;\Varid{s'}{}\<[E]%
\ColumnHook
\end{hscode}\resethooks
\indentend 
It is easy to verify that the four state laws hold contextually up to interpretation with \ensuremath{\Varid{h_{State}^\prime}}.

\subsection{Modularly Combining Effects}
\label{sec:combining-effects}

Combining multiple effects is relatively easy in the axiomatic approach based
on type classes. By imposing multiple constraints on the monad \ensuremath{\Varid{m}}, e.g.,
\ensuremath{(\Conid{MState}\;\Varid{s}\;\Varid{m},\Conid{MNondet}\;\Varid{m})}, we can express that \ensuremath{\Varid{m}} should support both state
and nondeterminism and respect their associated laws. In practice, this is
often insufficient: we usually require additional laws that govern the
interactions between the combined effects. We discuss possible interaction
laws between state and nondeterminism in details in Section
\ref{sec:local-global}.

\paragraph*{The Coproduct Operator for Combining Effects}\
To combine the syntax of effects given by free monads,
we need to define a coproduct operator \ensuremath{\mathrel{{:}{+}{:}}} for signatures.
\indentbegin \begin{hscode}\SaveRestoreHook
\column{B}{@{}>{\hspre}l<{\hspost}@{}}%
\column{E}{@{}>{\hspre}l<{\hspost}@{}}%
\>[B]{}\mathbf{data}\;(\Varid{f}\mathrel{{:}{+}{:}}\Varid{g})\;\Varid{a}\mathrel{=}\Conid{Inl}\;(\Varid{f}\;\Varid{a})\mid \Conid{Inr}\;(\Varid{g}\;\Varid{a}){}\<[E]%
\ColumnHook
\end{hscode}\resethooks
\indentend %
Note that given two functors \ensuremath{\Varid{f}} and \ensuremath{\Varid{g}}, it is obvious that \ensuremath{\Varid{f}\mathrel{{:}{+}{:}}\Varid{g}}
is again a functor.  This coproduct operator allows a modular
definition of the signatures of combined effects.  For instance, we
can encode programs with both state and nondeterminism as effects
using the data type \ensuremath{\Conid{Free}\;(\Varid{State_{F}}\mathrel{{:}{+}{:}}\Varid{Nondet_{F}})\;\Varid{a}}.  The coproduct also
has a neutral element \ensuremath{\Varid{Nil_{F}}}, representing the empty effect set.
\indentbegin \begin{hscode}\SaveRestoreHook
\column{B}{@{}>{\hspre}l<{\hspost}@{}}%
\column{E}{@{}>{\hspre}l<{\hspost}@{}}%
\>[B]{}\mathbf{data}\;\Varid{Nil_{F}}\;\Varid{a}\mbox{\onelinecomment  no constructors}{}\<[E]%
\ColumnHook
\end{hscode}\resethooks
\indentend 
We define the following two instances, which allow us to compose state
effects with any other effect functor \ensuremath{\Varid{f}}, and nondeterminism effects
with any other effect functors \ensuremath{\Varid{f}} and \ensuremath{\Varid{g}}, respectively.  As a
result, it is easy to see that \ensuremath{\Conid{Free}\;(\Varid{State_{F}}\;\Varid{s}\mathrel{{:}{+}{:}}\Varid{Nondet_{F}}\mathrel{{:}{+}{:}}\Varid{f})}
supports both state and nondeterminism for any functor \ensuremath{\Varid{f}}.

\indentbegin \begin{hscode}\SaveRestoreHook
\column{B}{@{}>{\hspre}l<{\hspost}@{}}%
\column{3}{@{}>{\hspre}l<{\hspost}@{}}%
\column{5}{@{}>{\hspre}l<{\hspost}@{}}%
\column{14}{@{}>{\hspre}l<{\hspost}@{}}%
\column{16}{@{}>{\hspre}l<{\hspost}@{}}%
\column{E}{@{}>{\hspre}l<{\hspost}@{}}%
\>[B]{}\mathbf{instance}\;(\Conid{Functor}\;\Varid{f})\Rightarrow \Conid{MState}\;\Varid{s}\;(\Conid{Free}\;(\Varid{State_{F}}\;\Varid{s}\mathrel{{:}{+}{:}}\Varid{f}))\;\mathbf{where}{}\<[E]%
\\
\>[B]{}\hsindent{5}{}\<[5]%
\>[5]{}\Varid{get}{}\<[14]%
\>[14]{}\mathrel{=}\Conid{Op}\mathbin{\$}\Conid{Inl}\mathbin{\$}\Conid{Get}\;\Varid{\eta}{}\<[E]%
\\
\>[B]{}\hsindent{5}{}\<[5]%
\>[5]{}\Varid{put}\;\Varid{x}{}\<[14]%
\>[14]{}\mathrel{=}\Conid{Op}\mathbin{\$}\Conid{Inl}\mathbin{\$}\Conid{Put}\;\Varid{x}\;(\Varid{\eta}\;()){}\<[E]%
\\[\blanklineskip]%
\>[B]{}\mathbf{instance}\;(\Conid{Functor}\;\Varid{f},\Conid{Functor}\;\Varid{g})\Rightarrow \Conid{MNondet}\;(\Conid{Free}\;(\Varid{g}\mathrel{{:}{+}{:}}\Varid{Nondet_{F}}\mathrel{{:}{+}{:}}\Varid{f}))\;\mathbf{where}{}\<[E]%
\\
\>[B]{}\hsindent{3}{}\<[3]%
\>[3]{}\Varid{\varnothing}{}\<[16]%
\>[16]{}\mathrel{=}\Conid{Op}\mathbin{\$}\Conid{Inr}\mathbin{\$}\Conid{Inl}\;\Conid{Fail}{}\<[E]%
\\
\>[B]{}\hsindent{3}{}\<[3]%
\>[3]{}\Varid{x}\mathbin{\talloblong}\Varid{y}{}\<[16]%
\>[16]{}\mathrel{=}\Conid{Op}\mathbin{\$}\Conid{Inr}\mathbin{\$}\Conid{Inl}\;(\Conid{Or}\;\Varid{x}\;\Varid{y}){}\<[E]%
\ColumnHook
\end{hscode}\resethooks
\indentend 

\paragraph*{Modularly Combining Effect Handlers}\
In order to interpret composite signatures, we use the forwarding approach of
\citet{Schrijvers19}. This way the handlers can be modularly composed: they
only need to know about the part of the syntax their effect is handling, and
forward the rest of the syntax to other handlers.

A mediator \ensuremath{(\mathbin{\#})} is used to separate the algebra \ensuremath{\Varid{alg}} for the handled effects and
the forwarding algebra \ensuremath{\Varid{fwd}} for the unhandled effects.
\indentbegin \begin{hscode}\SaveRestoreHook
\column{B}{@{}>{\hspre}l<{\hspost}@{}}%
\column{14}{@{}>{\hspre}l<{\hspost}@{}}%
\column{24}{@{}>{\hspre}c<{\hspost}@{}}%
\column{24E}{@{}l@{}}%
\column{27}{@{}>{\hspre}l<{\hspost}@{}}%
\column{E}{@{}>{\hspre}l<{\hspost}@{}}%
\>[B]{}(\mathbin{\#})\mathbin{::}(\Varid{f}\;\Varid{a}\to \Varid{b})\to (\Varid{g}\;\Varid{a}\to \Varid{b})\to (\Varid{f}\mathrel{{:}{+}{:}}\Varid{g})\;\Varid{a}\to \Varid{b}{}\<[E]%
\\
\>[B]{}(\Varid{alg}\mathbin{\#}\Varid{fwd})\;{}\<[14]%
\>[14]{}(\Conid{Inl}\;\Varid{op}){}\<[24]%
\>[24]{}\mathrel{=}{}\<[24E]%
\>[27]{}\Varid{alg}\;\Varid{op}{}\<[E]%
\\
\>[B]{}(\Varid{alg}\mathbin{\#}\Varid{fwd})\;{}\<[14]%
\>[14]{}(\Conid{Inr}\;\Varid{op}){}\<[24]%
\>[24]{}\mathrel{=}{}\<[24E]%
\>[27]{}\Varid{fwd}\;\Varid{op}{}\<[E]%
\ColumnHook
\end{hscode}\resethooks
\indentend 
The handlers for state and nondeterminism we have given earlier require a bit of
adjustment to be used in the composite setting since they only consider the signature
of their own effects.
We need to interpret the free monads into composite domains,
\ensuremath{\Conid{StateT}\;(\Conid{Free}\;\Varid{f})\;\Varid{a}} and \ensuremath{\Conid{Free}\;\Varid{f}\;[\mskip1.5mu \Varid{a}\mskip1.5mu]}, respectively.
Here \ensuremath{\Conid{StateT}} is the state transformer from the Monad Transformer Library \citep{mtl}.\indentbegin \begin{hscode}\SaveRestoreHook
\column{B}{@{}>{\hspre}l<{\hspost}@{}}%
\column{3}{@{}>{\hspre}l<{\hspost}@{}}%
\column{E}{@{}>{\hspre}l<{\hspost}@{}}%
\>[3]{}\mathbf{newtype}\;\Conid{StateT}\;\Varid{s}\;\Varid{m}\;\Varid{a}\mathrel{=}\Conid{StateT}\;\{\mskip1.5mu \Varid{run_{StateT}}\mathbin{::}\Varid{s}\to \Varid{m}\;(\Varid{a},\Varid{s})\mskip1.5mu\}{}\<[E]%
\ColumnHook
\end{hscode}\resethooks
\indentend The new handlers, into these composite domains, are defined as follows:
\indentbegin \begin{hscode}\SaveRestoreHook
\column{B}{@{}>{\hspre}l<{\hspost}@{}}%
\column{3}{@{}>{\hspre}l<{\hspost}@{}}%
\column{5}{@{}>{\hspre}l<{\hspost}@{}}%
\column{19}{@{}>{\hspre}l<{\hspost}@{}}%
\column{23}{@{}>{\hspre}l<{\hspost}@{}}%
\column{E}{@{}>{\hspre}l<{\hspost}@{}}%
\>[B]{}\Varid{h_{State}}\mathbin{::}\Conid{Functor}\;\Varid{f}\Rightarrow \Conid{Free}\;(\Varid{State_{F}}\;\Varid{s}\mathrel{{:}{+}{:}}\Varid{f})\;\Varid{a}\to \Conid{StateT}\;\Varid{s}\;(\Conid{Free}\;\Varid{f})\;\Varid{a}{}\<[E]%
\\
\>[B]{}\Varid{h_{State}}\mathrel{=}\Varid{fold}\;\Varid{gen_{S}}\;(\Varid{alg_{S}}\mathbin{\#}\Varid{fwd_{S}}){}\<[E]%
\\
\>[B]{}\hsindent{3}{}\<[3]%
\>[3]{}\mathbf{where}{}\<[E]%
\\
\>[3]{}\hsindent{2}{}\<[5]%
\>[5]{}\Varid{gen_{S}}\;\Varid{x}{}\<[23]%
\>[23]{}\mathrel{=}\Conid{StateT}\mathbin{\$}\lambda \Varid{s}\to \Varid{\eta}\;(\Varid{x},\Varid{s}){}\<[E]%
\\
\>[3]{}\hsindent{2}{}\<[5]%
\>[5]{}\Varid{alg_{S}}\;(\Conid{Get}\;{}\<[19]%
\>[19]{}\Varid{k}){}\<[23]%
\>[23]{}\mathrel{=}\Conid{StateT}\mathbin{\$}\lambda \Varid{s}\to \Varid{run_{StateT}}\;(\Varid{k}\;\Varid{s})\;\Varid{s}{}\<[E]%
\\
\>[3]{}\hsindent{2}{}\<[5]%
\>[5]{}\Varid{alg_{S}}\;(\Conid{Put}\;\Varid{s'}\;{}\<[19]%
\>[19]{}\Varid{k}){}\<[23]%
\>[23]{}\mathrel{=}\Conid{StateT}\mathbin{\$}\lambda \Varid{s}\to \Varid{run_{StateT}}\;\Varid{k}\;\Varid{s'}{}\<[E]%
\\
\>[3]{}\hsindent{2}{}\<[5]%
\>[5]{}\Varid{fwd_{S}}\;\Varid{op}{}\<[23]%
\>[23]{}\mathrel{=}\Conid{StateT}\mathbin{\$}\lambda \Varid{s}\to \Conid{Op}\mathbin{\$}\Varid{fmap}\;(\lambda \Varid{k}\to \Varid{run_{StateT}}\;\Varid{k}\;\Varid{s})\;\Varid{op}{}\<[E]%
\ColumnHook
\end{hscode}\resethooks
\indentend \indentbegin \begin{hscode}\SaveRestoreHook
\column{B}{@{}>{\hspre}l<{\hspost}@{}}%
\column{3}{@{}>{\hspre}l<{\hspost}@{}}%
\column{5}{@{}>{\hspre}l<{\hspost}@{}}%
\column{7}{@{}>{\hspre}c<{\hspost}@{}}%
\column{7E}{@{}l@{}}%
\column{10}{@{}>{\hspre}l<{\hspost}@{}}%
\column{22}{@{}>{\hspre}l<{\hspost}@{}}%
\column{E}{@{}>{\hspre}l<{\hspost}@{}}%
\>[B]{}\Varid{h_{ND+f}}\mathbin{::}\Conid{Functor}\;\Varid{f}\Rightarrow \Conid{Free}\;(\Varid{Nondet_{F}}\mathrel{{:}{+}{:}}\Varid{f})\;\Varid{a}\to \Conid{Free}\;\Varid{f}\;[\mskip1.5mu \Varid{a}\mskip1.5mu]{}\<[E]%
\\
\>[B]{}\Varid{h_{ND+f}}{}\<[7]%
\>[7]{}\mathrel{=}{}\<[7E]%
\>[10]{}\Varid{fold}\;\Varid{gen_{ND+f}}\;(\Varid{alg_{ND+f}}\mathbin{\#}\Varid{fwd_{ND+f}}){}\<[E]%
\\
\>[B]{}\hsindent{3}{}\<[3]%
\>[3]{}\mathbf{where}{}\<[E]%
\\
\>[3]{}\hsindent{2}{}\<[5]%
\>[5]{}\Varid{gen_{ND+f}}{}\<[22]%
\>[22]{}\mathrel{=}\Conid{Var}\hsdot{\circ }{.}\Varid{\eta}{}\<[E]%
\\
\>[3]{}\hsindent{2}{}\<[5]%
\>[5]{}\Varid{alg_{ND+f}}\;\Conid{Fail}{}\<[22]%
\>[22]{}\mathrel{=}\Conid{Var}\;[\mskip1.5mu \mskip1.5mu]{}\<[E]%
\\
\>[3]{}\hsindent{2}{}\<[5]%
\>[5]{}\Varid{alg_{ND+f}}\;(\Conid{Or}\;\Varid{p}\;\Varid{q}){}\<[22]%
\>[22]{}\mathrel{=}\Varid{liftM2}\;(+\!\!+)\;\Varid{p}\;\Varid{q}{}\<[E]%
\\
\>[3]{}\hsindent{2}{}\<[5]%
\>[5]{}\Varid{fwd_{ND+f}}\;\Varid{op}{}\<[22]%
\>[22]{}\mathrel{=}\Conid{Op}\;\Varid{op}{}\<[E]%
\ColumnHook
\end{hscode}\resethooks
\indentend 

Also, the empty signature \ensuremath{\Varid{Nil_{F}}} has a trivial associated handler.
\indentbegin \begin{hscode}\SaveRestoreHook
\column{B}{@{}>{\hspre}l<{\hspost}@{}}%
\column{E}{@{}>{\hspre}l<{\hspost}@{}}%
\>[B]{}\Varid{h_{Nil}}\mathbin{::}\Conid{Free}\;\Varid{Nil_{F}}\;\Varid{a}\to \Varid{a}{}\<[E]%
\\
\>[B]{}\Varid{h_{Nil}}\;(\Conid{Var}\;\Varid{x})\mathrel{=}\Varid{x}{}\<[E]%
\ColumnHook
\end{hscode}\resethooks
\indentend 
\section{Modelling Local State with Global State}
\label{sec:local-global}

This section studies two flavors of effect interaction between state and
nondeterminism: local-state and global-state semantics.  Local
state is a higher-level effect than globale state.
In a program with local state, each nondeterministic branch has its own local
copy of the state.  This is a convenient programming abstraction provided
by many systems that solve search problems, e.g., Prolog.  In contrast, global
state linearly threads a single state through the nondeterministic
branches. This can be interesting for performance reasons: we can limit memory
use by avoiding multiple copies, and perform in-place updates to reduce allocation
and garbage collection, and to improve locality.

In this section, we first formally characterise local-state and
global-state semantics, and then define a translation from the former
to the latter which uses the mechanism of nondeterminism to store
previous states and insert backtracking branches.

\subsection{Local-State Semantics}
\label{sec:local-state}

When a branch of a nondeterministic computation runs into a dead end and
the continuation is picked up at the most recent branching point,
any alterations made to the state by the terminated branch are invisible to
the continuation.
We refer to this semantics as \emph{local-state semantics}.
\citet{Gibbons11} also call it backtrackable state.

\paragraph*{The Local-State Laws}\
The following two laws characterise the local-state semantics for a monad
with state and nondeterminism:
\begin{alignat}{2}
    &\mbox{\bf put-right-identity}:\quad &
    \ensuremath{\Varid{put}\;\Varid{s}>\!\!>\Varid{\varnothing}} &= \ensuremath{\Varid{\varnothing}}~~\mbox{,} \label{eq:put-right-identity}\\
    &\mbox{\bf put-left-distributivity}:~ &
    \ensuremath{\Varid{put}\;\Varid{s}>\!\!>(\Varid{m}_{1}\mathbin{\talloblong}\Varid{m}_{2})} &= \ensuremath{(\Varid{put}\;\Varid{s}>\!\!>\Varid{m}_{1})\mathbin{\talloblong}(\Varid{put}\;\Varid{s}>\!\!>\Varid{m}_{2})} ~~\mbox{.} \label{eq:put-left-dist}
\end{alignat}

The equation (\ref{eq:put-right-identity}) expresses that \ensuremath{\Varid{\varnothing}} is the right
identity of \ensuremath{\Varid{put}}; it annihilates state updates.  The other law expresses that
\ensuremath{\Varid{put}} distributes from the left in \ensuremath{(\talloblong)}.

These two laws only focus on the interaction between \ensuremath{\Varid{put}} and
nondeterminism. The following laws for \ensuremath{\Varid{get}} can be derived from other
laws. The proof can be found in \Cref{app:local-law}.
\begin{alignat}{2}
    &\mbox{\bf get-right-identity}:\quad &
    \ensuremath{\Varid{get}>\!\!>\Varid{\varnothing}} &= \ensuremath{\Varid{\varnothing}}~~\mbox{,} \label{eq:get-right-identity}\\
    &\mbox{\bf get-left-distributivity}:~ &
    \ensuremath{\Varid{get}>\!\!>\!\!=(\lambda \Varid{x}\to \Varid{k}_{1}\;\Varid{x}\mathbin{\talloblong}\Varid{k}_{2}\;\Varid{x})} &= \ensuremath{(\Varid{get}>\!\!>\!\!=\Varid{k}_{1})\mathbin{\talloblong}(\Varid{get}>\!\!>\!\!=\Varid{k}_{2})} ~~\mbox{.} \label{eq:get-left-dist}
\end{alignat}

If we take these four equations together with the left-identity and right-distributivity
laws of nondeterminism, we can say that
that nondeterminism and state ``commute'';
if a \ensuremath{\Varid{get}} or \ensuremath{\Varid{put}} precedes a \ensuremath{\Varid{\varnothing}} or \ensuremath{\mathbin{\talloblong}}, we
can exchange their order (and vice versa).

\paragraph*{Implementation}\
Implementation-wise, the laws
imply that each nondeterministic branch has its own copy of the state.
For instance, \Cref{eq:put-left-dist} gives us\indentbegin \begin{hscode}\SaveRestoreHook
\column{B}{@{}>{\hspre}l<{\hspost}@{}}%
\column{3}{@{}>{\hspre}l<{\hspost}@{}}%
\column{E}{@{}>{\hspre}l<{\hspost}@{}}%
\>[3]{}\Varid{put}\;\mathrm{42}\;(\Varid{put}\;\mathrm{21}\mathbin{\talloblong}\Varid{get})\mathrel{=}(\Varid{put}\;\mathrm{42}>\!\!>\Varid{put}\;\mathrm{21})\mathbin{\talloblong}(\Varid{put}\;\mathrm{42}>\!\!>\Varid{get}){}\<[E]%
\ColumnHook
\end{hscode}\resethooks
\indentend %
The state we \ensuremath{\Varid{get}} in the second branch is still \ensuremath{\mathrm{42}}, despite the
\ensuremath{\Varid{put}\;\mathrm{21}} in the first branch.

One implementation satisfying the laws is\indentbegin \begin{hscode}\SaveRestoreHook
\column{B}{@{}>{\hspre}l<{\hspost}@{}}%
\column{3}{@{}>{\hspre}l<{\hspost}@{}}%
\column{E}{@{}>{\hspre}l<{\hspost}@{}}%
\>[3]{}\mathbf{type}\;\Conid{Local}\;\Varid{s}\;\Varid{m}\;\Varid{a}\mathrel{=}\Varid{s}\to \Varid{m}\;(\Varid{a},\Varid{s}){}\<[E]%
\ColumnHook
\end{hscode}\resethooks
\indentend where \ensuremath{\Varid{m}} is a nondeterministic monad, the simplest structure of which is a list.
This implementation is exactly that of \ensuremath{\Conid{StateT}\;\Varid{s}\;\Varid{m}\;\Varid{a}}
in the Monad Transformer Library \citep{mtl} which we have introduced in
\Cref{sec:combining-effects}.

With effect handling \citep{Kiselyov15, Wu14}, we get the local state semantics when
we run the state handler before the nondeterminism handler:
\indentbegin \begin{hscode}\SaveRestoreHook
\column{B}{@{}>{\hspre}l<{\hspost}@{}}%
\column{E}{@{}>{\hspre}l<{\hspost}@{}}%
\>[B]{}\Varid{h_{Local}}\mathbin{::}\Conid{Functor}\;\Varid{f}\Rightarrow \Conid{Free}\;(\Varid{State_{F}}\;\Varid{s}\mathrel{{:}{+}{:}}\Varid{Nondet_{F}}\mathrel{{:}{+}{:}}\Varid{f})\;\Varid{a}\to (\Varid{s}\to \Conid{Free}\;\Varid{f}\;[\mskip1.5mu \Varid{a}\mskip1.5mu]){}\<[E]%
\\
\>[B]{}\Varid{h_{Local}}\mathrel{=}\Varid{fmap}\;(\Varid{fmap}\;(\Varid{fmap}\;\Varid{fst})\hsdot{\circ }{.}\Varid{h_{ND+f}})\hsdot{\circ }{.}\Varid{run_{StateT}}\hsdot{\circ }{.}\Varid{h_{State}}{}\<[E]%
\ColumnHook
\end{hscode}\resethooks
\indentend 
In the case where the remaining signature is empty (\ensuremath{\Varid{f}\mathrel{=}\Varid{Nil_{F}}}), we get:
\indentbegin \begin{hscode}\SaveRestoreHook
\column{B}{@{}>{\hspre}l<{\hspost}@{}}%
\column{3}{@{}>{\hspre}l<{\hspost}@{}}%
\column{E}{@{}>{\hspre}l<{\hspost}@{}}%
\>[3]{}\Varid{fmap}\;\Varid{h_{Nil}}\hsdot{\circ }{.}\Varid{h_{Local}}\mathbin{::}\Conid{Free}\;(\Varid{State_{F}}\;\Varid{s}\mathrel{{:}{+}{:}}\Varid{Nondet_{F}}\mathrel{{:}{+}{:}}\Varid{Nil_{F}})\;\Varid{a}\to (\Varid{s}\to [\mskip1.5mu \Varid{a}\mskip1.5mu]){}\<[E]%
\ColumnHook
\end{hscode}\resethooks
\indentend %
Here, the result type \ensuremath{(\Varid{s}\to [\mskip1.5mu \Varid{a}\mskip1.5mu])} differs from \ensuremath{\Varid{s}\to [\mskip1.5mu (\Varid{a},\Varid{s})\mskip1.5mu]} in that it produces
only a list of results (\ensuremath{[\mskip1.5mu \Varid{a}\mskip1.5mu]}) and not pairs of results and their final state
(\ensuremath{[\mskip1.5mu (\Varid{a},\Varid{s})\mskip1.5mu]}). The latter is needed for \ensuremath{\Conid{Local}\;\Varid{s}\;\Varid{m}} to have the structure of a monad, in particular
to support the modular composition of computations with \ensuremath{(>\!\!>\!\!=)}. Such is not needed for the carriers of handlers, because
the composition of computations is taken care of by the \ensuremath{(>\!\!>\!\!=)} operator of the free monad.

\subsection{Global-State Semantics}
\label{sec:global-state}

Alternatively, one can choose a semantics where state reigns over nondeterminism.
In this case of non-backtrackable state, alterations to the state persist over
backtracks. Because only a single state is shared over all branches of
nondeterministic computation, we call this state the \emph{global-state
semantics}.

\paragraph*{The Global-State Law}\
The global-state semantics sets apart non-backtrackable state from
backtrackable state.
In addition to the general laws for nondeterminism
((\ref{eq:mzero}) -- (\ref{eq:mzero-zero})) and state
((\ref{eq:put-put}) -- (\ref{eq:get-get})), we provide a \emph{global-state law}
to govern the interaction between nondeterminism and state.
\begin{alignat}{2}
    &\mbox{\bf put-or}:\quad &
    \ensuremath{(\Varid{put}\;\Varid{s}>\!\!>\Varid{m})\mathbin{\talloblong}\Varid{n}} &= \ensuremath{\Varid{put}\;\Varid{s}>\!\!>(\Varid{m}\mathbin{\talloblong}\Varid{n})}~~\mbox{.} \label{eq:put-or}
\end{alignat}

This law allows lifting a \ensuremath{\Varid{put}} operation from the left branch of a
nondeterministic choice.
For instance, if \ensuremath{\Varid{m}\mathrel{=}\Varid{\varnothing}} in the left-hand side of the equation,
then under local-state semantics
(laws (\ref{eq:mzero}) and (\ref{eq:put-right-identity}))
the left-hand side becomes equal to \ensuremath{\Varid{n}},
whereas under global-state semantics
(laws (\ref{eq:mzero}) and (\ref{eq:put-or}))
the equation simplifies to \ensuremath{\Varid{put}\;\Varid{s}>\!\!>\Varid{n}}.

\paragraph*{Implementation}\
Figuring out a correct implementation for the global-state monad is tricky.
One might believe that \ensuremath{\Conid{Global}\;\Varid{s}\;\Varid{m}\;\Varid{a}\mathrel{=}\Varid{s}\to (\Varid{m}\;\Varid{a},\Varid{s})}
is a natural implementation of such a monad.
However, the usual naive implementation of \ensuremath{(>\!\!>\!\!=)} for it does not satisfy
right-distributivity (\ref{eq:mplus-dist}) and is therefore not even a monad.
The type \ensuremath{\Conid{ListT}\;(\Conid{State}\;\Varid{s})} from the Monad Transformer Library \citep{mtl}
expands to essentially the same implementation with
monad \ensuremath{\Varid{m}} instantiated by the list monad.
This implementation has the same flaws.
More careful implementations of \ensuremath{\Conid{ListT}} (often referred to as
``\ensuremath{\Conid{ListT}} done right'') satisfying right-distributivity
(\ref{eq:mplus-dist}) and other monad laws have been proposed by
\citet{Volkov14, Gale}.
The following implementation is essentially that of Gale.
\indentbegin \begin{hscode}\SaveRestoreHook
\column{B}{@{}>{\hspre}l<{\hspost}@{}}%
\column{E}{@{}>{\hspre}l<{\hspost}@{}}%
\>[B]{}\mathbf{newtype}\;\Conid{Global}\;\Varid{s}\;\Varid{a}\mathrel{=}\Conid{Gl}\;\{\mskip1.5mu \Varid{runGl}\mathbin{::}\Varid{s}\to (\Conid{Maybe}\;(\Varid{a},\Conid{Global}\;\Varid{s}\;\Varid{a}),\Varid{s})\mskip1.5mu\}{}\<[E]%
\ColumnHook
\end{hscode}\resethooks
\indentend The \ensuremath{\Conid{Maybe}} in this type indicates that a computation may fail to produce a
result. However, since the \ensuremath{\Varid{s}} is outside of the \ensuremath{\Conid{Maybe}}, a modified state
is returned even if the computation fails.
This \ensuremath{\Conid{Global}\;\Varid{s}\;\Varid{a}} type is an instance of the \ensuremath{\Conid{MState}} and \ensuremath{\Conid{MNondet}}
typeclasses.

\begin{minipage}{0.5\textwidth}
\indentbegin \begin{hscode}\SaveRestoreHook
\column{B}{@{}>{\hspre}l<{\hspost}@{}}%
\column{5}{@{}>{\hspre}l<{\hspost}@{}}%
\column{9}{@{}>{\hspre}l<{\hspost}@{}}%
\column{18}{@{}>{\hspre}l<{\hspost}@{}}%
\column{24}{@{}>{\hspre}l<{\hspost}@{}}%
\column{29}{@{}>{\hspre}c<{\hspost}@{}}%
\column{29E}{@{}l@{}}%
\column{31}{@{}>{\hspre}l<{\hspost}@{}}%
\column{34}{@{}>{\hspre}l<{\hspost}@{}}%
\column{E}{@{}>{\hspre}l<{\hspost}@{}}%
\>[B]{}\mathbf{instance}\;\Conid{MNondet}\;(\Conid{Global}\;\Varid{s})\;\mathbf{where}{}\<[E]%
\\
\>[B]{}\hsindent{5}{}\<[5]%
\>[5]{}\Varid{\varnothing}{}\<[18]%
\>[18]{}\mathrel{=}\Conid{Gl}\;(\lambda \Varid{s}\to {}\<[31]%
\>[31]{}(\Conid{Nothing},\Varid{s})){}\<[E]%
\\
\>[B]{}\hsindent{5}{}\<[5]%
\>[5]{}\Varid{p}\mathbin{\talloblong}\Varid{q}{}\<[18]%
\>[18]{}\mathrel{=}\Conid{Gl}\;(\lambda \Varid{s}\to {}\<[31]%
\>[31]{}\mathbf{case}\;\Varid{runGl}\;\Varid{p}\;\Varid{s}\;\mathbf{of}{}\<[E]%
\\
\>[5]{}\hsindent{4}{}\<[9]%
\>[9]{}(\Conid{Nothing},{}\<[24]%
\>[24]{}\Varid{t}){}\<[29]%
\>[29]{}\to {}\<[29E]%
\>[34]{}\Varid{runGl}\;\Varid{q}\;\Varid{t}{}\<[E]%
\\
\>[5]{}\hsindent{4}{}\<[9]%
\>[9]{}(\Conid{Just}\;(\Varid{x},\Varid{r}),{}\<[24]%
\>[24]{}\Varid{t}){}\<[29]%
\>[29]{}\to {}\<[29E]%
\>[34]{}(\Conid{Just}\;(\Varid{x},\Varid{r}\mathbin{\talloblong}\Varid{q}),\Varid{t})){}\<[E]%
\ColumnHook
\end{hscode}\resethooks
\indentend \end{minipage}
\begin{minipage}{0.5\textwidth}
\indentbegin \begin{hscode}\SaveRestoreHook
\column{B}{@{}>{\hspre}l<{\hspost}@{}}%
\column{5}{@{}>{\hspre}l<{\hspost}@{}}%
\column{12}{@{}>{\hspre}c<{\hspost}@{}}%
\column{12E}{@{}l@{}}%
\column{15}{@{}>{\hspre}l<{\hspost}@{}}%
\column{19}{@{}>{\hspre}l<{\hspost}@{}}%
\column{24}{@{}>{\hspre}c<{\hspost}@{}}%
\column{24E}{@{}l@{}}%
\column{25}{@{}>{\hspre}c<{\hspost}@{}}%
\column{25E}{@{}l@{}}%
\column{28}{@{}>{\hspre}l<{\hspost}@{}}%
\column{29}{@{}>{\hspre}l<{\hspost}@{}}%
\column{40}{@{}>{\hspre}l<{\hspost}@{}}%
\column{41}{@{}>{\hspre}l<{\hspost}@{}}%
\column{50}{@{}>{\hspre}l<{\hspost}@{}}%
\column{51}{@{}>{\hspre}l<{\hspost}@{}}%
\column{E}{@{}>{\hspre}l<{\hspost}@{}}%
\>[B]{}\mathbf{instance}\;\Conid{MState}\;\Varid{s}\;(\Conid{Global}\;\Varid{s})\;\mathbf{where}{}\<[E]%
\\
\>[B]{}\hsindent{5}{}\<[5]%
\>[5]{}\Varid{get}{}\<[12]%
\>[12]{}\mathrel{=}{}\<[12E]%
\>[15]{}\Conid{Gl}\;{}\<[19]%
\>[19]{}(\lambda \Varid{s}{}\<[24]%
\>[24]{}\to {}\<[24E]%
\>[28]{}(\Conid{Just}\;(\Varid{s},{}\<[40]%
\>[40]{}\Varid{\varnothing}),{}\<[50]%
\>[50]{}\Varid{s})){}\<[E]%
\\
\>[B]{}\hsindent{5}{}\<[5]%
\>[5]{}\Varid{put}\;\Varid{s}{}\<[12]%
\>[12]{}\mathrel{=}{}\<[12E]%
\>[15]{}\Conid{Gl}\;{}\<[19]%
\>[19]{}(\lambda \anonymous {}\<[25]%
\>[25]{}\to {}\<[25E]%
\>[29]{}(\Conid{Just}\;((),{}\<[41]%
\>[41]{}\Varid{\varnothing}),{}\<[51]%
\>[51]{}\Varid{s})){}\<[E]%
\ColumnHook
\end{hscode}\resethooks
\indentend \end{minipage}

Failure, of course, returns an empty continuation and an unmodified state.
Branching first exhausts the left branch before switching to the right branch.

Effect handlers \citep{Kiselyov15, Wu14} also provide implementations
that match our intuition of non-backtrackable computations.
The global-state semantics can be implemented by simply switching the
order of the two effect handlers compared to the local state handler
\ensuremath{\Varid{h_{Local}}}.
\indentbegin \begin{hscode}\SaveRestoreHook
\column{B}{@{}>{\hspre}l<{\hspost}@{}}%
\column{E}{@{}>{\hspre}l<{\hspost}@{}}%
\>[B]{}\Varid{h_{Global}}\mathbin{::}(\Conid{Functor}\;\Varid{f})\Rightarrow \Conid{Free}\;(\Varid{State_{F}}\;\Varid{s}\mathrel{{:}{+}{:}}\Varid{Nondet_{F}}\mathrel{{:}{+}{:}}\Varid{f})\;\Varid{a}\to (\Varid{s}\to \Conid{Free}\;\Varid{f}\;[\mskip1.5mu \Varid{a}\mskip1.5mu]){}\<[E]%
\\
\>[B]{}\Varid{h_{Global}}\mathrel{=}\Varid{fmap}\;(\Varid{fmap}\;\Varid{fst})\hsdot{\circ }{.}\Varid{run_{StateT}}\hsdot{\circ }{.}\Varid{h_{State}}\hsdot{\circ }{.}\Varid{h_{ND+f}}\hsdot{\circ }{.}\Varid{(\Leftrightarrow)}{}\<[E]%
\ColumnHook
\end{hscode}\resethooks
\indentend This also runs a single state through a nondeterministic computation.
The \ensuremath{\Varid{(\Leftrightarrow)}} isomorphism swaps the order of two functors in the
co-product signature of the free monad in order to let \ensuremath{\Varid{h_{Local}}} and
\ensuremath{\Varid{h_{Global}}} have the same type signature.
\indentbegin \begin{hscode}\SaveRestoreHook
\column{B}{@{}>{\hspre}l<{\hspost}@{}}%
\column{27}{@{}>{\hspre}l<{\hspost}@{}}%
\column{47}{@{}>{\hspre}l<{\hspost}@{}}%
\column{E}{@{}>{\hspre}l<{\hspost}@{}}%
\>[B]{}\Varid{(\Leftrightarrow)}\mathbin{::}(\Conid{Functor}\;\Varid{f}_{1},\Conid{Functor}\;\Varid{f}_{2},\Conid{Functor}\;\Varid{f})\Rightarrow \Conid{Free}\;(\Varid{f}_{1}\mathrel{{:}{+}{:}}\Varid{f}_{2}\mathrel{{:}{+}{:}}\Varid{f})\;\Varid{a}\to \Conid{Free}\;(\Varid{f}_{2}\mathrel{{:}{+}{:}}\Varid{f}_{1}\mathrel{{:}{+}{:}}\Varid{f})\;\Varid{a}{}\<[E]%
\\
\>[B]{}\Varid{(\Leftrightarrow)}\;(\Conid{Var}\;\Varid{x}){}\<[27]%
\>[27]{}\mathrel{=}\Conid{Var}\;\Varid{x}{}\<[E]%
\\
\>[B]{}\Varid{(\Leftrightarrow)}\;(\Conid{Op}\;(\Conid{Inl}\;\Varid{k})){}\<[27]%
\>[27]{}\mathrel{=}(\Conid{Op}\hsdot{\circ }{.}\Conid{Inr}\hsdot{\circ }{.}\Conid{Inl})\;{}\<[47]%
\>[47]{}(\Varid{fmap}\;\Varid{(\Leftrightarrow)}\;\Varid{k}){}\<[E]%
\\
\>[B]{}\Varid{(\Leftrightarrow)}\;(\Conid{Op}\;(\Conid{Inr}\;(\Conid{Inl}\;\Varid{k}))){}\<[27]%
\>[27]{}\mathrel{=}(\Conid{Op}\hsdot{\circ }{.}\Conid{Inl})\;{}\<[47]%
\>[47]{}(\Varid{fmap}\;\Varid{(\Leftrightarrow)}\;\Varid{k}){}\<[E]%
\\
\>[B]{}\Varid{(\Leftrightarrow)}\;(\Conid{Op}\;(\Conid{Inr}\;(\Conid{Inr}\;\Varid{k}))){}\<[27]%
\>[27]{}\mathrel{=}(\Conid{Op}\hsdot{\circ }{.}\Conid{Inr}\hsdot{\circ }{.}\Conid{Inr})\;{}\<[47]%
\>[47]{}(\Varid{fmap}\;\Varid{(\Leftrightarrow)}\;\Varid{k}){}\<[E]%
\ColumnHook
\end{hscode}\resethooks
\indentend %

In the case where the remaining signature is empty (\ensuremath{\Varid{f}\mathrel{=}\Varid{Nil_{F}}}), we get:
\indentbegin \begin{hscode}\SaveRestoreHook
\column{B}{@{}>{\hspre}l<{\hspost}@{}}%
\column{3}{@{}>{\hspre}l<{\hspost}@{}}%
\column{E}{@{}>{\hspre}l<{\hspost}@{}}%
\>[3]{}\Varid{fmap}\;\Varid{h_{Nil}}\hsdot{\circ }{.}\Varid{h_{Global}}\mathbin{::}\Conid{Free}\;(\Varid{State_{F}}\;\Varid{s}\mathrel{{:}{+}{:}}\Varid{Nondet_{F}}\mathrel{{:}{+}{:}}\Varid{Nil_{F}})\;\Varid{a}\to (\Varid{s}\to [\mskip1.5mu \Varid{a}\mskip1.5mu]){}\<[E]%
\ColumnHook
\end{hscode}\resethooks
\indentend %
The carrier type is again simpler than \ensuremath{\Conid{Global}\;\Varid{s}\;\Varid{a}} because it does not have to
support the \ensuremath{(>\!\!>\!\!=)} operator.

\subsection{Simulating Local State with Global State}
\label{sec:transforming-between-local-and-global-state}
\label{sec:local2global}

Both local state and global state have their own laws and semantics.
Also, both interpretations of nondeterminism with state have their own
advantages and disadvantages.

Local-state semantics imply that each nondeterministic branch has its own state.
This may be expensive if the state is represented by data structures, e.g. arrays,
that are costly to duplicate.
For example, when each new state is only slightly different from the previous,
we have a wasteful duplication of information.

Global-state semantics, however, threads a single state through the entire
computation without making any implicit copies.
Consequently, it is easier to control resource usage and apply
optimisation strategies in this setting.
However, doing this to a program that has a backtracking structure, and would
be more naturally expressed in a local-state style,
comes at a great loss of clarity.
Furthermore, it is significantly more challenging for programmers to
reason about global-state semantics than local-state semantics.

To resolve this dilemma, we can write our programs in a local-state style
and then translate them to the global-state style to enable low-level optimisations.
In this subsection, we show one systematic program translation that
alters a program written for local-state semantics to a program that,
when interpreted under global-state semantics, behaves exactly the
same as the original program interpreted under local-state semantics.
This translation makes explicit copying of the whole state and relies
on using the nondeterminism mechanism to insert state-restoring
branches.
We will show other translations from local-state semantics to
global-state semantics which avoid the copying and do not rely on
nondeterminism in \Cref{sec:undo} and \Cref{sec:trail-stack}.

\paragraph*{State-Restoring Put}\
Central to the implementation of backtracking in the global state setting is
the backtracking variant of \ensuremath{\Varid{put}}.
Going forward, such a state-restoring \ensuremath{\Varid{put_{R}}} modifies the state as usual,
but, when backtracked over, it restores the old state.

We implement \ensuremath{\Varid{put_{R}}} using both state and nondeterminism as follows:
\indentbegin \begin{hscode}\SaveRestoreHook
\column{B}{@{}>{\hspre}l<{\hspost}@{}}%
\column{E}{@{}>{\hspre}l<{\hspost}@{}}%
\>[B]{}\Varid{put_{R}}\mathbin{::}(\Conid{MState}\;\Varid{s}\;\Varid{m},\Conid{MNondet}\;\Varid{m})\Rightarrow \Varid{s}\to \Varid{m}\;(){}\<[E]%
\\
\>[B]{}\Varid{put_{R}}\;\Varid{s}\mathrel{=}\Varid{get}>\!\!>\!\!=\lambda \Varid{s'}\to \Varid{put}\;\Varid{s}\mathbin{\talloblong}\Varid{side}\;(\Varid{put}\;\Varid{s'}){}\<[E]%
\ColumnHook
\end{hscode}\resethooks
\indentend \label{eq:state-restoring-put}

Here the \ensuremath{\Varid{side}} branch is executed for its side-effect only; it fails
before yielding a result.
\indentbegin \begin{hscode}\SaveRestoreHook
\column{B}{@{}>{\hspre}l<{\hspost}@{}}%
\column{E}{@{}>{\hspre}l<{\hspost}@{}}%
\>[B]{}\Varid{side}\mathbin{::}\Conid{MNondet}\;\Varid{m}\Rightarrow \Varid{m}\;\Varid{a}\to \Varid{m}\;\Varid{b}{}\<[E]%
\\
\>[B]{}\Varid{side}\;\Varid{m}\mathrel{=}\Varid{m}>\!\!>\Varid{\varnothing}{}\<[E]%
\ColumnHook
\end{hscode}\resethooks
\indentend 
Intuitively, the second branch generated by \ensuremath{\Varid{put_{R}}} can be understood
as a backtracking or state-restoring branch.
The \ensuremath{\Varid{put_{R}}\;\Varid{s}} operation changes the state to \ensuremath{\Varid{s}} in the first branch
\ensuremath{\Varid{put}\;\Varid{s}}, and then restores the it to the original state \ensuremath{\Varid{s'}} in the
second branch after we finish all computations in the first branch.
Then, the second branch immediately fails so that we can keep going to
other branches with the original state \ensuremath{\Varid{s'}}.
For example, assuming an arbitrary computation \ensuremath{\Varid{comp}} is placed after
a state-restoring put, we have the following calculation.\indentbegin \begin{hscode}\SaveRestoreHook
\column{B}{@{}>{\hspre}l<{\hspost}@{}}%
\column{3}{@{}>{\hspre}l<{\hspost}@{}}%
\column{6}{@{}>{\hspre}l<{\hspost}@{}}%
\column{E}{@{}>{\hspre}l<{\hspost}@{}}%
\>[6]{}\Varid{put_{R}}\;\Varid{s}>\!\!>\Varid{comp}{}\<[E]%
\\
\>[3]{}\mathrel{=}\mbox{\commentbegin ~  definition of \ensuremath{\Varid{put_{R}}}   \commentend}{}\<[E]%
\\
\>[3]{}\hsindent{3}{}\<[6]%
\>[6]{}(\Varid{get}>\!\!>\!\!=\lambda \Varid{s'}\to \Varid{put}\;\Varid{s}\mathbin{\talloblong}\Varid{side}\;(\Varid{put}\;\Varid{s'}))>\!\!>\Varid{comp}{}\<[E]%
\\
\>[3]{}\mathrel{=}\mbox{\commentbegin ~  right-distributivity (\ref{eq:mplus-dist})   \commentend}{}\<[E]%
\\
\>[3]{}\hsindent{3}{}\<[6]%
\>[6]{}(\Varid{get}>\!\!>\!\!=\lambda \Varid{s'}\to (\Varid{put}\;\Varid{s}>\!\!>\Varid{comp})\mathbin{\talloblong}(\Varid{side}\;(\Varid{put}\;\Varid{s'})>\!\!>\Varid{comp})){}\<[E]%
\\
\>[3]{}\mathrel{=}\mbox{\commentbegin ~  left identity (\ref{eq:mzero-zero})   \commentend}{}\<[E]%
\\
\>[3]{}\hsindent{3}{}\<[6]%
\>[6]{}(\Varid{get}>\!\!>\!\!=\lambda \Varid{s'}\to (\Varid{put}\;\Varid{s}>\!\!>\Varid{comp})\mathbin{\talloblong}\Varid{side}\;(\Varid{put}\;\Varid{s'})){}\<[E]%
\ColumnHook
\end{hscode}\resethooks
\indentend This program saves the current state \ensuremath{\Varid{s'}}, computes \ensuremath{\Varid{comp}} using state \ensuremath{\Varid{s}},
and then restores the saved state \ensuremath{\Varid{s'}}.
Figure \ref{fig:state-restoring-put} shows how the state-passing works
and the flow of execution for a computation after a state-restoring put.

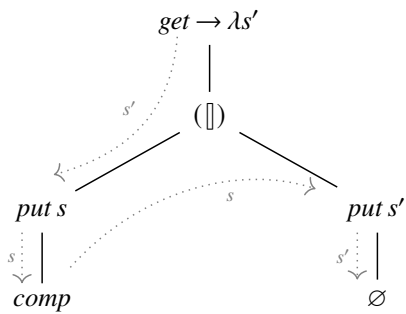
\begin{figure}[ht]
\[\begin{tikzcd}
    & {\ensuremath{\Varid{get}\to \lambda \Varid{s'}}} \\
    & {\ensuremath{(\talloblong)}} \\
    {\ensuremath{\Varid{put}\;\Varid{s}}} && {\ensuremath{\Varid{put}\;\Varid{s'}}} \\
    {\ensuremath{\Varid{comp}}} && {\ensuremath{\Varid{\varnothing}}}
    \arrow[no head, from=1-2, to=2-2]
    \arrow[no head, from=2-2, to=3-1]
    \arrow[no head, from=2-2, to=3-3]
    \arrow[no head, from=3-1, to=4-1]
    \arrow[no head, from=3-3, to=4-3]
    \arrow["{\ensuremath{\Varid{s'}}}"', shift right=5, color={rgb,255:red,128;green,128;blue,128}, curve={height=-18pt}, shorten <=5pt, dotted, from=1-2, to=3-1]
    \arrow["{\ensuremath{\Varid{s}}}"', shift right=3, color={rgb,255:red,128;green,128;blue,128}, dotted, from=3-1, to=4-1]
    \arrow["{\ensuremath{\Varid{s'}}}"', shift right=3, color={rgb,255:red,128;green,128;blue,128}, dotted, from=3-3, to=4-3]
    \arrow["{\ensuremath{\Varid{s}}}"'{pos=0.6}, color={rgb,255:red,128;green,128;blue,128}, curve={height=-30pt}, shorten <=8pt, shorten >=8pt, dotted, from=4-1, to=3-3]
\end{tikzcd}\]
\caption{State-restoring put-operation.}
\label{fig:state-restoring-put}
\end{figure}

Another example of \ensuremath{\Varid{put_{R}}} is shown in Table \ref{tab:state-restoring-put}, where
three programs are run with initial state \ensuremath{\Varid{s}_{0}}.
Note the difference between the final state and the program result for the
state-restoring put.

\begin{table}[h]
\begin{center}
\begin{tabular}{l|ll}
                            & Program result & Final state \\ \hline
\ensuremath{\Varid{\eta}\;\Varid{x}>\!\!>\Varid{get}}           & \ensuremath{\Varid{s}_{0}}           & \ensuremath{\Varid{s}_{0}}        \\
\ensuremath{\Varid{put}\;\Varid{s}>\!\!>\Varid{\eta}\;\Varid{x}>\!\!>\Varid{get}}  & \ensuremath{\Varid{s}}            & \ensuremath{\Varid{s}}         \\
\ensuremath{\Varid{put_{R}}\;\Varid{s}>\!\!>\Varid{\eta}\;\Varid{x}>\!\!>\Varid{get}} & \ensuremath{\Varid{s}}            & \ensuremath{\Varid{s}_{0}}
\end{tabular}
\end{center}
\caption{Comparing \ensuremath{\Varid{put}} and \ensuremath{\Varid{put_{R}}}.}
\label{tab:state-restoring-put}
\end{table}

\paragraph*{Translation with State-Restoring Put}\
The idea is that \ensuremath{\Varid{put_{R}}}, when run with a global state, satisfies laws
(\ref{eq:put-put}) to (\ref{eq:put-left-dist}) --- the state laws and
the local-state laws.
Then, one can take a program written for local-state semantics, replace
all occurrences of \ensuremath{\Varid{put}} by \ensuremath{\Varid{put_{R}}}, and run the program with a global state.
However, to satisfy all of these laws,
care should be taken to replace \emph{all} occurrences of \ensuremath{\Varid{put}}.
Particularly, placing a program in a larger context, where \ensuremath{\Varid{put}} has not been replaced, can change the meaning
of its subprograms.
An example of such a problematic context is \ensuremath{(>\!\!>\Varid{put}\;\Varid{t})}, where the \ensuremath{\Varid{get}}-\ensuremath{\Varid{put}} law
(\ref{eq:get-put}) breaks and programs \ensuremath{\Varid{get}>\!\!>\Varid{put_{R}}} and \ensuremath{\Varid{\eta}\;()} can be
differentiated:

\begin{minipage}{0.7\textwidth}\indentbegin \begin{hscode}\SaveRestoreHook
\column{B}{@{}>{\hspre}l<{\hspost}@{}}%
\column{3}{@{}>{\hspre}l<{\hspost}@{}}%
\column{6}{@{}>{\hspre}l<{\hspost}@{}}%
\column{E}{@{}>{\hspre}l<{\hspost}@{}}%
\>[6]{}(\Varid{get}>\!\!>\Varid{put_{R}})>\!\!>\Varid{put}\;\Varid{t}{}\<[E]%
\\
\>[3]{}\mathrel{=}\mbox{\commentbegin ~  definition of \ensuremath{\Varid{put_{R}}}   \commentend}{}\<[E]%
\\
\>[3]{}\hsindent{3}{}\<[6]%
\>[6]{}(\Varid{get}>\!\!>\!\!=\lambda \Varid{s}\to \Varid{get}>\!\!>\!\!=\lambda \Varid{s}_{0}\to \Varid{put}\;\Varid{s}\mathbin{\talloblong}\Varid{side}\;(\Varid{put}\;\Varid{s}_{0}))>\!\!>\Varid{put}\;\Varid{t}{}\<[E]%
\\
\>[3]{}\mathrel{=}\mbox{\commentbegin ~  get-get (\ref{eq:get-get})   \commentend}{}\<[E]%
\\
\>[3]{}\hsindent{3}{}\<[6]%
\>[6]{}(\Varid{get}>\!\!>\!\!=\lambda \Varid{s}\to \Varid{put}\;\Varid{s}\mathbin{\talloblong}\Varid{side}\;(\Varid{put}\;\Varid{s}))>\!\!>\Varid{put}\;\Varid{t}{}\<[E]%
\\
\>[3]{}\mathrel{=}\mbox{\commentbegin ~  right-distributivity (\ref{eq:mplus-dist})   \commentend}{}\<[E]%
\\
\>[3]{}\hsindent{3}{}\<[6]%
\>[6]{}(\Varid{get}>\!\!>\!\!=\lambda \Varid{s}\to (\Varid{put}\;\Varid{s}>\!\!>\Varid{put}\;\Varid{t})\mathbin{\talloblong}(\Varid{side}\;(\Varid{put}\;\Varid{s}))>\!\!>\Varid{put}\;\Varid{t}){}\<[E]%
\\
\>[3]{}\mathrel{=}\mbox{\commentbegin ~  left-identity (\ref{eq:mzero-zero})   \commentend}{}\<[E]%
\\
\>[3]{}\hsindent{3}{}\<[6]%
\>[6]{}(\Varid{get}>\!\!>\!\!=\lambda \Varid{s}\to (\Varid{put}\;\Varid{s}>\!\!>\Varid{put}\;\Varid{t})\mathbin{\talloblong}\Varid{side}\;(\Varid{put}\;\Varid{s})){}\<[E]%
\\
\>[3]{}\mathrel{=}\mbox{\commentbegin ~  put-put (\ref{eq:put-put})   \commentend}{}\<[E]%
\\
\>[3]{}\hsindent{3}{}\<[6]%
\>[6]{}(\Varid{get}>\!\!>\!\!=\lambda \Varid{s}\to \Varid{put}\;\Varid{t}\mathbin{\talloblong}\Varid{side}\;(\Varid{put}\;\Varid{s})){}\<[E]%
\ColumnHook
\end{hscode}\resethooks
\indentend \end{minipage}%
\begin{minipage}{0.3\textwidth}\indentbegin \begin{hscode}\SaveRestoreHook
\column{B}{@{}>{\hspre}l<{\hspost}@{}}%
\column{3}{@{}>{\hspre}c<{\hspost}@{}}%
\column{3E}{@{}l@{}}%
\column{6}{@{}>{\hspre}l<{\hspost}@{}}%
\column{E}{@{}>{\hspre}l<{\hspost}@{}}%
\>[6]{}\Varid{\eta}\;()>\!\!>\Varid{put}\;\Varid{t}{}\<[E]%
\\
\>[3]{}\mathrel{=}{}\<[3E]%
\>[6]{}\Varid{put}\;\Varid{t}{}\<[E]%
\ColumnHook
\end{hscode}\resethooks
\indentend \end{minipage}

Those two programs do not behave in the same way when \ensuremath{\Varid{s}\not\equiv \Varid{t}}.
Hence, only provided that \emph{all} occurences of \ensuremath{\Varid{put}} in a program are replaced by \ensuremath{\Varid{put_{R}}}
can we simulate local-state semantics. The replacement itself as well as the correctness statement
that incorporates this requirement can be easily epressed with effect handlers. As we will explain
below we need to articulate this global replacement in the
correctness proof. This requires using the {\bf fusion-post'} rule rather than the more widely used {\bf fusion-post} rule.

We realize the global replacement of \ensuremath{\Varid{put}} with
a \ensuremath{\Varid{put_{R}}} with the effect handler \ensuremath{\Varid{local2global}}:
\indentbegin \begin{hscode}\SaveRestoreHook
\column{B}{@{}>{\hspre}l<{\hspost}@{}}%
\column{3}{@{}>{\hspre}l<{\hspost}@{}}%
\column{5}{@{}>{\hspre}l<{\hspost}@{}}%
\column{15}{@{}>{\hspre}l<{\hspost}@{}}%
\column{25}{@{}>{\hspre}l<{\hspost}@{}}%
\column{E}{@{}>{\hspre}l<{\hspost}@{}}%
\>[B]{}\Varid{local2global}{}\<[15]%
\>[15]{}\mathbin{::}\Conid{Functor}\;\Varid{f}{}\<[E]%
\\
\>[15]{}\Rightarrow \Conid{Free}\;(\Varid{State_{F}}\;\Varid{s}\mathrel{{:}{+}{:}}\Varid{Nondet_{F}}\mathrel{{:}{+}{:}}\Varid{f})\;\Varid{a}{}\<[E]%
\\
\>[15]{}\to \Conid{Free}\;(\Varid{State_{F}}\;\Varid{s}\mathrel{{:}{+}{:}}\Varid{Nondet_{F}}\mathrel{{:}{+}{:}}\Varid{f})\;\Varid{a}{}\<[E]%
\\
\>[B]{}\Varid{local2global}\mathrel{=}\Varid{fold}\;\Conid{Var}\;\Varid{alg}{}\<[E]%
\\
\>[B]{}\hsindent{3}{}\<[3]%
\>[3]{}\mathbf{where}{}\<[E]%
\\
\>[3]{}\hsindent{2}{}\<[5]%
\>[5]{}\Varid{alg}\;(\Conid{Inl}\;(\Conid{Put}\;\Varid{t}\;\Varid{k}))\mathrel{=}\Varid{put_{R}}\;\Varid{t}>\!\!>\Varid{k}{}\<[E]%
\\
\>[3]{}\hsindent{2}{}\<[5]%
\>[5]{}\Varid{alg}\;\Varid{p}{}\<[25]%
\>[25]{}\mathrel{=}\Conid{Op}\;\Varid{p}{}\<[E]%
\ColumnHook
\end{hscode}\resethooks
\indentend %

The following theorem shows that the translation \ensuremath{\Varid{local2global}}
preserves the meaning when switching from local-state to global-state
semantics:
\begin{restatable}[]{theorem}{localGlobal}
\label{thm:local-global}\indentbegin \begin{hscode}\SaveRestoreHook
\column{B}{@{}>{\hspre}l<{\hspost}@{}}%
\column{3}{@{}>{\hspre}l<{\hspost}@{}}%
\column{E}{@{}>{\hspre}l<{\hspost}@{}}%
\>[3]{}\Varid{h_{Global}}\hsdot{\circ }{.}\Varid{local2global}\mathrel{=}\Varid{h_{Local}}{}\<[E]%
\ColumnHook
\end{hscode}\resethooks
\indentend \end{restatable}

\begin{proof}
Both the left-hand side and the right-hand side of the equation consist of
function compositions involving one or more folds.
We apply fold fusion separately on both sides to contract each
into a single fold:
\begin{eqnarray*}
\ensuremath{\Varid{h_{Global}}\hsdot{\circ }{.}\Varid{local2global}} & = & \ensuremath{\Varid{fold}\;\Varid{gen}_{\Varid{LHS}}\;(\Varid{alg}_{\Varid{LHS}}^{\Varid{S}}\mathbin{\#}\Varid{alg}_{\Varid{RHS}}^{\Varid{ND}}\mathbin{\#}\Varid{fwd}_{\Varid{LHS}})} \\
\ensuremath{\Varid{h_{Local}}}& = & \ensuremath{\Varid{fold}\;\Varid{gen}_{\Varid{RHS}}\;(\Varid{alg}_{\Varid{RHS}}^{\Varid{S}}\mathbin{\#}\Varid{alg}_{\Varid{RHS}}^{\Varid{ND}}\mathbin{\#}\Varid{fwd}_{\Varid{RHS}})}
\end{eqnarray*}
We approach this calculationally. That is to say, we do not first postulate
definitions of the unknowns above (\ensuremath{\Varid{alg}_{\Varid{LHS}}^{\Varid{S}}} and so on) and then verify whether
the fusion conditions are satisfied. Instead, we discover the definitions of the unknowns.
We start from the known side of
each fusion condition and perform case analysis on the possible shapes of
input. By simplifying the resulting case-specific expression, and pushing the handler
applications inwards, we end up at a point where we can read off the definition
of the unknown that makes the fusion condition hold for that case.

Finally, we show that both folds are equal by showing that their
corresponding parameters are equal:
\begin{eqnarray*}
\ensuremath{\Varid{gen}_{\Varid{LHS}}} & = & \ensuremath{\Varid{gen}_{\Varid{RHS}}} \\
\ensuremath{\Varid{alg}_{\Varid{LHS}}^{\Varid{S}}} & = & \ensuremath{\Varid{alg}_{\Varid{RHS}}^{\Varid{S}}} \\
\ensuremath{\Varid{alg}_{\Varid{LHS}}^{\Varid{ND}}} & = & \ensuremath{\Varid{alg}_{\Varid{RHS}}^{\Varid{ND}}} \\
\ensuremath{\Varid{fwd}_{\Varid{LHS}}} & = & \ensuremath{\Varid{fwd}_{\Varid{RHS}}}
\end{eqnarray*}

A noteworthy observation is that, for fusing the left-hand side of the equation, we do not use the standard
fusion rule {\bf fusion-post}~(\ref{eq:fusion-post}):
\begin{eqnarray*}
    \ensuremath{\Varid{h_{Global}}\hsdot{\circ }{.}\Varid{fold}\;\Conid{Var}\;\Varid{alg}} & = & \ensuremath{\Varid{fold}\;(\Varid{h_{Global}}\hsdot{\circ }{.}\Conid{Var})\;\Varid{alg'}} \\
     \Leftarrow \qquad
   \ensuremath{\Varid{h_{Global}}\hsdot{\circ }{.}\Varid{alg}} & = & \ensuremath{\Varid{alg'}\hsdot{\circ }{.}\Varid{fmap}\;\Varid{h_{Global}}}
\end{eqnarray*}
where \ensuremath{\Varid{local2global}\mathrel{=}\Varid{fold}\;\Conid{Var}\;\Varid{alg}}. The problem is that we will not find an
appropriate \ensuremath{\Varid{alg'}} such that \ensuremath{\Varid{alg'}\;(\Varid{fmap}\;\Varid{h_{Global}}\;\Varid{t})} restores the state for any
\ensuremath{\Varid{t}} of type \ensuremath{(\Varid{State_{F}}\;\Varid{s}\mathrel{{:}{+}{:}}\Varid{Nondet_{F}}\mathrel{{:}{+}{:}}\Varid{f})\;(\Conid{Free}\;(\Varid{State_{F}}\;\Varid{s}\mathrel{{:}{+}{:}}\Varid{Nondet_{F}}\mathrel{{:}{+}{:}}\Varid{f})\;\Varid{a})}.

Fortunately, we do not need such an \ensuremath{\Varid{alg'}}. As we have already pointed out, we
can assume that the subterms of \ensuremath{\Varid{t}} have already been transformed by
\ensuremath{\Varid{local2global}}, and thus all occurrences of \ensuremath{\Conid{Put}} appear in the \ensuremath{\Varid{put_{R}}}
constellation.

We can incorporate this assumption by using the alternative fusion rule
{\bf fusion-post'}~(\ref{eq:fusion-post-strong}):
\begin{eqnarray*}
    \ensuremath{\Varid{h_{Global}}\hsdot{\circ }{.}\Varid{fold}\;\Conid{Var}\;\Varid{alg}} & = & \ensuremath{\Varid{fold}\;(\Varid{h_{Global}}\hsdot{\circ }{.}\Conid{Var})\;\Varid{alg'}} \\
     \Leftarrow \qquad
   \ensuremath{\Varid{h_{Global}}\hsdot{\circ }{.}\Varid{alg}\hsdot{\circ }{.}\Varid{fmap}\;\Varid{local2global}} & = & \ensuremath{\Varid{alg'}\hsdot{\circ }{.}\Varid{fmap}\;\Varid{h_{Global}}\hsdot{\circ }{.}\Varid{fmap}\;\Varid{local2global}}
\end{eqnarray*}
The additional \ensuremath{\Varid{fmap}\;\Varid{local2global}} in the condition captures the property that
all the subterms have been transformed by \ensuremath{\Varid{local2global}}.

In order to not clutter the proofs, we abstract everywhere over this additional \ensuremath{\Varid{fmap}\;\Varid{local2global}} application, except
for the key lemma which
expresses that the syntactic transformation \ensuremath{\Varid{local2global}} makes sure
that, despite any temporary changes, the computation \ensuremath{\Varid{t}} restores the state
back to its initial value.

We elaborate each of these steps in \Cref{app:local-global}.
\end{proof}

\paragraph*{N-queens with Global-State Semantics}\
Recall the backtracking algorithm \ensuremath{\Varid{queens}} for the n-queens example in
\Cref{sec:motivation-and-challenges}.
It is initially designed to run in the local-state semantics because
every branch maintains its own copy of the state and has no influence
on other branches. We can handle it with \ensuremath{\Varid{h_{Local}}} as follows.
\indentbegin \begin{hscode}\SaveRestoreHook
\column{B}{@{}>{\hspre}l<{\hspost}@{}}%
\column{E}{@{}>{\hspre}l<{\hspost}@{}}%
\>[B]{}\Varid{queens_{Local}}\mathbin{::}\Conid{Int}\to [\mskip1.5mu [\mskip1.5mu \Conid{Int}\mskip1.5mu]\mskip1.5mu]{}\<[E]%
\\
\>[B]{}\Varid{queens_{Local}}\mathrel{=}\Varid{h_{Nil}}\hsdot{\circ }{.}\Varid{flip}\;\Varid{h_{Local}}\;(\mathrm{0},[\mskip1.5mu \mskip1.5mu])\hsdot{\circ }{.}\Varid{queens}{}\<[E]%
\ColumnHook
\end{hscode}\resethooks
\indentend %
With the simulation \ensuremath{\Varid{local2global}}, we can also translate \ensuremath{\Varid{queens}} to
an equivalent program in global-state semantics and handle it with
\ensuremath{\Varid{h_{Global}}}.%
\indentbegin \begin{hscode}\SaveRestoreHook
\column{B}{@{}>{\hspre}l<{\hspost}@{}}%
\column{E}{@{}>{\hspre}l<{\hspost}@{}}%
\>[B]{}\Varid{queens_{Global}}\mathbin{::}\Conid{Int}\to [\mskip1.5mu [\mskip1.5mu \Conid{Int}\mskip1.5mu]\mskip1.5mu]{}\<[E]%
\\
\>[B]{}\Varid{queens_{Global}}\mathrel{=}\Varid{h_{Nil}}\hsdot{\circ }{.}\Varid{flip}\;\Varid{h_{Global}}\;(\mathrm{0},[\mskip1.5mu \mskip1.5mu])\hsdot{\circ }{.}\Varid{local2global}\hsdot{\circ }{.}\Varid{queens}{}\<[E]%
\ColumnHook
\end{hscode}\resethooks
\indentend %

\section{Modelling Nondeterminism with State}
\label{sec:nondeterminism-state}

In the previous section, we have translated the local-state semantics,
a high-level combination of the state and nondeterminism effects, to
the global-state semantics, a low-level combination of the state
and nondeterminism effects.
In this section, we further translate the resulting nondeterminism
component, which is itself a relatively high-level effect, to a
lower-level implementation with the state effect.
Our translation coincides with the fact that, while nondeterminism is
typically modelled using the \ensuremath{\Conid{List}} monad, many efficient
nondeterministic systems, such as Prolog, use a low-level
state-based implementation to implement the nondeterminism mechanism.

\subsection{Simulating Nondeterminism with State}
\label{sec:sim-nondet-state}

The main idea of simulating nondeterminism with state is to explictly
manage
\begin{enumerate}
\item
a list of the \ensuremath{\Varid{results}} found so far, and
\item
a list of yet to be explored branches, which we call a \ensuremath{\Varid{stack}}.
\end{enumerate}

This stack corresponds to the choicepoint stack in Prolog.
When entering one branch, we can push other branches to the stack.
When leaving the branch, we collect its result and pop a new branch
from the stack to continue.

We define a new type \ensuremath{\Conid{S}\;\Varid{a}} consisting of the \ensuremath{\Varid{results}} and \ensuremath{\Varid{stack}}.
\indentbegin \begin{hscode}\SaveRestoreHook
\column{B}{@{}>{\hspre}l<{\hspost}@{}}%
\column{E}{@{}>{\hspre}l<{\hspost}@{}}%
\>[B]{}\mathbf{type}\;\Conid{Comp}\;\Varid{s}\;\Varid{a}\mathrel{=}\Conid{Free}\;(\Varid{State_{F}}\;\Varid{s})\;\Varid{a}{}\<[E]%
\\
\>[B]{}\mathbf{data}\;\Conid{S}\;\Varid{a}\mathrel{=}\Conid{S}\;\{\mskip1.5mu \Varid{results}\mathbin{::}[\mskip1.5mu \Varid{a}\mskip1.5mu],\Varid{stack}\mathbin{::}[\mskip1.5mu \Conid{Comp}\;(\Conid{S}\;\Varid{a})\;()\mskip1.5mu]\mskip1.5mu\}{}\<[E]%
\ColumnHook
\end{hscode}\resethooks
\indentend %
The branches in the stack are represented by computations in the form
of free monads over the \ensuremath{\Varid{State_{F}}} signature.
We do not allow branches to use other effects here to show the idea
more clearly.  In \Cref{sec:nondet2state} we will consider the more
general case where branches can have any effects abstracted by a
functor \ensuremath{\Varid{f}}.

\noindent
\begin{figure}[ht]
\small
\begin{subfigure}[t]{0.3\linewidth}
\indentbegin \begin{hscode}\SaveRestoreHook
\column{B}{@{}>{\hspre}l<{\hspost}@{}}%
\column{3}{@{}>{\hspre}l<{\hspost}@{}}%
\column{5}{@{}>{\hspre}l<{\hspost}@{}}%
\column{7}{@{}>{\hspre}l<{\hspost}@{}}%
\column{14}{@{}>{\hspre}l<{\hspost}@{}}%
\column{E}{@{}>{\hspre}l<{\hspost}@{}}%
\>[B]{}\Varid{pop_S}\mathbin{::}\Conid{Comp}\;(\Conid{S}\;\Varid{a})\;(){}\<[E]%
\\
\>[B]{}\Varid{pop_S}\mathrel{=}\mathbf{do}{}\<[E]%
\\
\>[B]{}\hsindent{3}{}\<[3]%
\>[3]{}\Conid{S}\;\Varid{xs}\;\Varid{stack}\leftarrow \Varid{get}{}\<[E]%
\\
\>[B]{}\hsindent{3}{}\<[3]%
\>[3]{}\mathbf{case}\;\Varid{stack}\;\mathbf{of}{}\<[E]%
\\
\>[3]{}\hsindent{2}{}\<[5]%
\>[5]{}[\mskip1.5mu \mskip1.5mu]{}\<[14]%
\>[14]{}\to \Varid{\eta}\;(){}\<[E]%
\\
\>[3]{}\hsindent{2}{}\<[5]%
\>[5]{}\Varid{p}\mathbin{:}\Varid{ps}{}\<[14]%
\>[14]{}\to \mathbf{do}{}\<[E]%
\\
\>[5]{}\hsindent{2}{}\<[7]%
\>[7]{}\Varid{put}\;(\Conid{S}\;\Varid{xs}\;\Varid{ps});\Varid{p}{}\<[E]%
\ColumnHook
\end{hscode}\resethooks
\indentend \caption{Popping from the stack.}
\label{fig:pop}
\end{subfigure}
\begin{subfigure}[t]{0.3\linewidth}
\indentbegin \begin{hscode}\SaveRestoreHook
\column{B}{@{}>{\hspre}l<{\hspost}@{}}%
\column{3}{@{}>{\hspre}l<{\hspost}@{}}%
\column{9}{@{}>{\hspre}c<{\hspost}@{}}%
\column{9E}{@{}l@{}}%
\column{13}{@{}>{\hspre}l<{\hspost}@{}}%
\column{E}{@{}>{\hspre}l<{\hspost}@{}}%
\>[B]{}\Varid{push_S}{}\<[9]%
\>[9]{}\mathbin{::}{}\<[9E]%
\>[13]{}\Conid{Comp}\;(\Conid{S}\;\Varid{a})\;(){}\<[E]%
\\
\>[9]{}\to {}\<[9E]%
\>[13]{}\Conid{Comp}\;(\Conid{S}\;\Varid{a})\;(){}\<[E]%
\\
\>[9]{}\to {}\<[9E]%
\>[13]{}\Conid{Comp}\;(\Conid{S}\;\Varid{a})\;(){}\<[E]%
\\
\>[B]{}\Varid{push_S}\;\Varid{q}\;\Varid{p}\mathrel{=}\mathbf{do}{}\<[E]%
\\
\>[B]{}\hsindent{3}{}\<[3]%
\>[3]{}\Conid{S}\;\Varid{xs}\;\Varid{stack}\leftarrow \Varid{get}{}\<[E]%
\\
\>[B]{}\hsindent{3}{}\<[3]%
\>[3]{}\Varid{put}\;(\Conid{S}\;\Varid{xs}\;(\Varid{q}\mathbin{:}\Varid{stack})){}\<[E]%
\\
\>[B]{}\hsindent{3}{}\<[3]%
\>[3]{}\Varid{p}{}\<[E]%
\ColumnHook
\end{hscode}\resethooks
\indentend \caption{Pushing to the stack.}
\label{fig:push}
\end{subfigure}
\begin{subfigure}[t]{0.3\linewidth}
\indentbegin \begin{hscode}\SaveRestoreHook
\column{B}{@{}>{\hspre}l<{\hspost}@{}}%
\column{3}{@{}>{\hspre}l<{\hspost}@{}}%
\column{10}{@{}>{\hspre}c<{\hspost}@{}}%
\column{10E}{@{}l@{}}%
\column{14}{@{}>{\hspre}l<{\hspost}@{}}%
\column{E}{@{}>{\hspre}l<{\hspost}@{}}%
\>[B]{}\Varid{append_S}{}\<[10]%
\>[10]{}\mathbin{::}{}\<[10E]%
\>[14]{}\Varid{a}{}\<[E]%
\\
\>[10]{}\to {}\<[10E]%
\>[14]{}\Conid{Comp}\;(\Conid{S}\;\Varid{a})\;(){}\<[E]%
\\
\>[10]{}\to {}\<[10E]%
\>[14]{}\Conid{Comp}\;(\Conid{S}\;\Varid{a})\;(){}\<[E]%
\\
\>[B]{}\Varid{append_S}\;\Varid{x}\;\Varid{p}\mathrel{=}\mathbf{do}{}\<[E]%
\\
\>[B]{}\hsindent{3}{}\<[3]%
\>[3]{}\Conid{S}\;\Varid{xs}\;\Varid{stack}\leftarrow \Varid{get}{}\<[E]%
\\
\>[B]{}\hsindent{3}{}\<[3]%
\>[3]{}\Varid{put}\;(\Conid{S}\;(\Varid{xs}+\!\!+[\mskip1.5mu \Varid{x}\mskip1.5mu])\;\Varid{stack}){}\<[E]%
\\
\>[B]{}\hsindent{3}{}\<[3]%
\>[3]{}\Varid{p}{}\<[E]%
\ColumnHook
\end{hscode}\resethooks
\indentend \caption{Appending a result.}
\label{fig:append}
\end{subfigure}%
\caption{Auxiliary functions \ensuremath{\Varid{pop_S}}, \ensuremath{\Varid{push_S}} and \ensuremath{\Varid{append_S}}.}
\label{fig:pop-push-append}
\end{figure}

For brevity, instead of defining a new stack effect capturing the
stack operations like pop and push, we implement stack operations with
the state effect.
We define three auxiliary functions in \Cref{fig:pop-push-append} to
interact with the stack in \ensuremath{\Conid{S}\;\Varid{a}}:
\begin{itemize}
\item
The function \ensuremath{\Varid{pop_S}} removes and executes the top element of the stack.
\item
The function \ensuremath{\Varid{push_S}} pushes a branch into the stack.
\item
The function \ensuremath{\Varid{append_S}} adds a result to the current results.
\end{itemize}

Now, everything is in place to define a simulation function
\ensuremath{\Varid{nondet2state_S}} that interprets nondeterministic programs
represented by the free monad \ensuremath{\Conid{Free}\;\Varid{Nondet_{F}}\;\Varid{a}} as state-wrapped
programs represented by the free monad \ensuremath{\Conid{Free}\;(\Varid{State_{F}}\;(\Conid{S}\;\Varid{a}))\;()}.
\indentbegin \begin{hscode}\SaveRestoreHook
\column{B}{@{}>{\hspre}l<{\hspost}@{}}%
\column{3}{@{}>{\hspre}l<{\hspost}@{}}%
\column{5}{@{}>{\hspre}l<{\hspost}@{}}%
\column{19}{@{}>{\hspre}l<{\hspost}@{}}%
\column{E}{@{}>{\hspre}l<{\hspost}@{}}%
\>[B]{}\Varid{nondet2state_S}\mathbin{::}\Conid{Free}\;\Varid{Nondet_{F}}\;\Varid{a}\to \Conid{Free}\;(\Varid{State_{F}}\;(\Conid{S}\;\Varid{a}))\;(){}\<[E]%
\\
\>[B]{}\Varid{nondet2state_S}\mathrel{=}\Varid{fold}\;\Varid{gen}\;\Varid{alg}{}\<[E]%
\\
\>[B]{}\hsindent{3}{}\<[3]%
\>[3]{}\mathbf{where}{}\<[E]%
\\
\>[3]{}\hsindent{2}{}\<[5]%
\>[5]{}\Varid{gen}\;\Varid{x}{}\<[19]%
\>[19]{}\mathrel{=}\Varid{append_S}\;\Varid{x}\;\Varid{pop_S}{}\<[E]%
\\
\>[3]{}\hsindent{2}{}\<[5]%
\>[5]{}\Varid{alg}\;\Conid{Fail}{}\<[19]%
\>[19]{}\mathrel{=}\Varid{pop_S}{}\<[E]%
\\
\>[3]{}\hsindent{2}{}\<[5]%
\>[5]{}\Varid{alg}\;(\Conid{Or}\;\Varid{p}\;\Varid{q}){}\<[19]%
\>[19]{}\mathrel{=}\Varid{push_S}\;\Varid{q}\;\Varid{p}{}\<[E]%
\ColumnHook
\end{hscode}\resethooks
\indentend The generator of this handler records a new result and then pops the
next branch from the stack and proceeds with it. Likewise, for failure
the handler simply pops and proceeds with the next branch. For
nondeterministic choices, the handler pushes the second branch on the
stack and proceeds with the first branch.

To extract the final result from the \ensuremath{\Conid{S}} wrapper, we define the \ensuremath{\Varid{extract_S}} function.
\indentbegin \begin{hscode}\SaveRestoreHook
\column{B}{@{}>{\hspre}l<{\hspost}@{}}%
\column{E}{@{}>{\hspre}l<{\hspost}@{}}%
\>[B]{}\Varid{extract_S}\mathbin{::}\Conid{State}\;(\Conid{S}\;\Varid{a})\;()\to [\mskip1.5mu \Varid{a}\mskip1.5mu]{}\<[E]%
\\
\>[B]{}\Varid{extract_S}\;\Varid{x}\mathrel{=}\Varid{results}\hsdot{\circ }{.}\Varid{snd}\mathbin{\$}\Varid{run_{State}}\;\Varid{x}\;(\Conid{S}\;[\mskip1.5mu \mskip1.5mu]\;[\mskip1.5mu \mskip1.5mu]){}\<[E]%
\ColumnHook
\end{hscode}\resethooks
\indentend Finally, we define the function \ensuremath{\Varid{run_{ND}}} which wraps everything up to
handle a nondeterministic computation to a list of results.
The state handler \ensuremath{\Varid{h_{State}^\prime}} is defined in
\Cref{sec:free-monads-and-their-folds}.
\indentbegin \begin{hscode}\SaveRestoreHook
\column{B}{@{}>{\hspre}l<{\hspost}@{}}%
\column{E}{@{}>{\hspre}l<{\hspost}@{}}%
\>[B]{}\Varid{run_{ND}}\mathbin{::}\Conid{Free}\;\Varid{Nondet_{F}}\;\Varid{a}\to [\mskip1.5mu \Varid{a}\mskip1.5mu]{}\<[E]%
\\
\>[B]{}\Varid{run_{ND}}\mathrel{=}\Varid{extract_S}\hsdot{\circ }{.}\Varid{h_{State}^\prime}\hsdot{\circ }{.}\Varid{nondet2state_S}{}\<[E]%
\ColumnHook
\end{hscode}\resethooks
\indentend 
We have the following theorem showing the correctness of the
simulation via the equivalence of the \ensuremath{\Varid{run_{ND}}} function
and the nondeterminism handler \ensuremath{\Varid{h_{ND}}} defined in
\Cref{sec:free-monads-and-their-folds}.

\begin{restatable}[]{theorem}{nondetStateS}
\label{thm:nondet-stateS}\indentbegin \begin{hscode}\SaveRestoreHook
\column{B}{@{}>{\hspre}l<{\hspost}@{}}%
\column{3}{@{}>{\hspre}l<{\hspost}@{}}%
\column{E}{@{}>{\hspre}l<{\hspost}@{}}%
\>[3]{}\Varid{run_{ND}}\mathrel{=}\Varid{h_{ND}}{}\<[E]%
\ColumnHook
\end{hscode}\resethooks
\indentend \end{restatable}

The proof can be found in \Cref{app:runnd-hnd}.
The main idea is again to use fold fusion.
Consider the expanded form\indentbegin \begin{hscode}\SaveRestoreHook
\column{B}{@{}>{\hspre}l<{\hspost}@{}}%
\column{3}{@{}>{\hspre}l<{\hspost}@{}}%
\column{E}{@{}>{\hspre}l<{\hspost}@{}}%
\>[3]{}(\Varid{extract_S}\hsdot{\circ }{.}\Varid{h_{State}^\prime})\hsdot{\circ }{.}\Varid{nondet2state_S}\mathrel{=}\Varid{h_{ND}}{}\<[E]%
\ColumnHook
\end{hscode}\resethooks
\indentend %
Both \ensuremath{\Varid{nondet2state_S}} and \ensuremath{\Varid{h_{ND}}} are written as folds.
We use the law {\bf fusion-post'}~(\ref{eq:fusion-post-strong}) to
fuse the left-hand side into a single fold.
Since the right-hand side is already a fold, to prove the equation we
just need to check the components of the fold \ensuremath{\Varid{h_{ND}}} satisfy the
conditions of the fold fusion, i.e., the following two equations:
For the latter we only need to prove the following two equations:
\[\ba{rl}
    &\ensuremath{(\Varid{extract_S}\hsdot{\circ }{.}\Varid{h_{State}^\prime})\hsdot{\circ }{.}\Varid{gen}\mathrel{=}\Varid{gen_{ND}}} \\
    &\ensuremath{(\Varid{extract_S}\hsdot{\circ }{.}\Varid{h_{State}^\prime})\hsdot{\circ }{.}\Varid{alg}\hsdot{\circ }{.}\Varid{fmap}\;\Varid{nondet2state_S}}\\
 \ensuremath{\mathrel{=}}&  \ensuremath{\Varid{alg_{ND}}\hsdot{\circ }{.}\Varid{fmap}\;(\Varid{extract_S}\hsdot{\circ }{.}\Varid{h_{State}^\prime})\hsdot{\circ }{.}\Varid{fmap}\;\Varid{nondet2state_S}}
\ea\]
where \ensuremath{\Varid{gen}} and \ensuremath{\Varid{alg}} are from the definition of \ensuremath{\Varid{nondet2state_S}}, and
\ensuremath{\Varid{gen_{ND}}} and \ensuremath{\Varid{alg_{ND}}} are from the definition of \ensuremath{\Varid{h_{ND}}}.

\subsection{Combining the Simulation with Other Effects}
\label{sec:combining-the-simulation-with-other-effects}
\label{sec:nondet2state}

The \ensuremath{\Varid{nondet2state_S}} function only considers nondeterminism as the only
effect. In this section, we generalise it to work in combination with
other effects. One immediate benefit is that we can use it in
together with our previous simulation \ensuremath{\Varid{local2global}} in
\Cref{sec:local2global}.

Firstly, we need to augment the signature in the computation type for
branches with an additional functor \ensuremath{\Varid{f}} for other effects.
The computation type is essentially changed from \ensuremath{\Conid{Free}\;(\Varid{State_{F}}\;\Varid{s})\;\Varid{a}}
to \ensuremath{\Conid{Free}\;(\Varid{State_{F}}\;\Varid{s}\mathrel{{:}{+}{:}}\Varid{f})\;\Varid{a}}.
We define the state type \ensuremath{\Conid{SS}\;\Varid{f}\;\Varid{a}} as follows:
\indentbegin \begin{hscode}\SaveRestoreHook
\column{B}{@{}>{\hspre}l<{\hspost}@{}}%
\column{E}{@{}>{\hspre}l<{\hspost}@{}}%
\>[B]{}\mathbf{type}\;\Varid{Comp_{SS}}\;\Varid{s}\;\Varid{f}\;\Varid{a}\mathrel{=}\Conid{Free}\;(\Varid{State_{F}}\;\Varid{s}\mathrel{{:}{+}{:}}\Varid{f})\;\Varid{a}{}\<[E]%
\\
\>[B]{}\mathbf{data}\;\Conid{SS}\;\Varid{f}\;\Varid{a}\mathrel{=}\Conid{SS}\;\{\mskip1.5mu \Varid{results_{SS}}\mathbin{::}[\mskip1.5mu \Varid{a}\mskip1.5mu],\Varid{stack_{SS}}\mathbin{::}[\mskip1.5mu \Varid{Comp_{SS}}\;(\Conid{SS}\;\Varid{f}\;\Varid{a})\;\Varid{f}\;()\mskip1.5mu]\mskip1.5mu\}{}\<[E]%
\ColumnHook
\end{hscode}\resethooks
\indentend 

\noindent
\begin{figure}[t]
\noindent
\small
\begin{subfigure}[t]{0.3\linewidth}
\indentbegin \begin{hscode}\SaveRestoreHook
\column{B}{@{}>{\hspre}l<{\hspost}@{}}%
\column{3}{@{}>{\hspre}l<{\hspost}@{}}%
\column{5}{@{}>{\hspre}l<{\hspost}@{}}%
\column{7}{@{}>{\hspre}l<{\hspost}@{}}%
\column{8}{@{}>{\hspre}l<{\hspost}@{}}%
\column{14}{@{}>{\hspre}l<{\hspost}@{}}%
\column{E}{@{}>{\hspre}l<{\hspost}@{}}%
\>[B]{}\Varid{pop_{SS}}{}\<[8]%
\>[8]{}\mathbin{::}\Conid{Functor}\;\Varid{f}{}\<[E]%
\\
\>[B]{}\hsindent{3}{}\<[3]%
\>[3]{}\Rightarrow \Varid{Comp_{SS}}\;(\Conid{SS}\;\Varid{f}\;\Varid{a})\;\Varid{f}\;(){}\<[E]%
\\
\>[B]{}\Varid{pop_{SS}}\mathrel{=}\mathbf{do}{}\<[E]%
\\
\>[B]{}\hsindent{3}{}\<[3]%
\>[3]{}\Conid{SS}\;\Varid{xs}\;\Varid{stack}\leftarrow \Varid{get}{}\<[E]%
\\
\>[B]{}\hsindent{3}{}\<[3]%
\>[3]{}\mathbf{case}\;\Varid{stack}\;\mathbf{of}{}\<[E]%
\\
\>[3]{}\hsindent{2}{}\<[5]%
\>[5]{}[\mskip1.5mu \mskip1.5mu]{}\<[14]%
\>[14]{}\to \Varid{\eta}\;(){}\<[E]%
\\
\>[3]{}\hsindent{2}{}\<[5]%
\>[5]{}\Varid{p}\mathbin{:}\Varid{ps}{}\<[14]%
\>[14]{}\to \mathbf{do}{}\<[E]%
\\
\>[5]{}\hsindent{2}{}\<[7]%
\>[7]{}\Varid{put}\;(\Conid{SS}\;\Varid{xs}\;\Varid{ps});\Varid{p}{}\<[E]%
\ColumnHook
\end{hscode}\resethooks
\indentend \caption{Popping from the stack.}
\label{fig:pop-ss}
\end{subfigure}%
\begin{subfigure}[t]{0.3\linewidth}
\indentbegin \begin{hscode}\SaveRestoreHook
\column{B}{@{}>{\hspre}l<{\hspost}@{}}%
\column{3}{@{}>{\hspre}l<{\hspost}@{}}%
\column{9}{@{}>{\hspre}l<{\hspost}@{}}%
\column{E}{@{}>{\hspre}l<{\hspost}@{}}%
\>[B]{}\Varid{push_{SS}}{}\<[9]%
\>[9]{}\mathbin{::}\Conid{Functor}\;\Varid{f}{}\<[E]%
\\
\>[B]{}\hsindent{3}{}\<[3]%
\>[3]{}\Rightarrow \Varid{Comp_{SS}}\;(\Conid{SS}\;\Varid{f}\;\Varid{a})\;\Varid{f}\;(){}\<[E]%
\\
\>[B]{}\hsindent{3}{}\<[3]%
\>[3]{}\to \Varid{Comp_{SS}}\;(\Conid{SS}\;\Varid{f}\;\Varid{a})\;\Varid{f}\;(){}\<[E]%
\\
\>[B]{}\hsindent{3}{}\<[3]%
\>[3]{}\to \Varid{Comp_{SS}}\;(\Conid{SS}\;\Varid{f}\;\Varid{a})\;\Varid{f}\;(){}\<[E]%
\\
\>[B]{}\Varid{push_{SS}}\;\Varid{q}\;\Varid{p}\mathrel{=}\mathbf{do}{}\<[E]%
\\
\>[B]{}\hsindent{3}{}\<[3]%
\>[3]{}\Conid{SS}\;\Varid{xs}\;\Varid{stack}\leftarrow \Varid{get}{}\<[E]%
\\
\>[B]{}\hsindent{3}{}\<[3]%
\>[3]{}\Varid{put}\;(\Conid{SS}\;\Varid{xs}\;(\Varid{q}\mathbin{:}\Varid{stack})){}\<[E]%
\\
\>[B]{}\hsindent{3}{}\<[3]%
\>[3]{}\Varid{p}{}\<[E]%
\ColumnHook
\end{hscode}\resethooks
\indentend \caption{Pushing to the stack.}
\label{fig:push-ss}
\end{subfigure}
\begin{subfigure}[t]{0.35\linewidth}
\indentbegin \begin{hscode}\SaveRestoreHook
\column{B}{@{}>{\hspre}l<{\hspost}@{}}%
\column{3}{@{}>{\hspre}l<{\hspost}@{}}%
\column{11}{@{}>{\hspre}l<{\hspost}@{}}%
\column{E}{@{}>{\hspre}l<{\hspost}@{}}%
\>[B]{}\Varid{append_{SS}}{}\<[11]%
\>[11]{}\mathbin{::}\Conid{Functor}\;\Varid{f}{}\<[E]%
\\
\>[B]{}\hsindent{3}{}\<[3]%
\>[3]{}\Rightarrow \Varid{a}{}\<[E]%
\\
\>[B]{}\hsindent{3}{}\<[3]%
\>[3]{}\to \Varid{Comp_{SS}}\;(\Conid{SS}\;\Varid{f}\;\Varid{a})\;\Varid{f}\;(){}\<[E]%
\\
\>[B]{}\hsindent{3}{}\<[3]%
\>[3]{}\to \Varid{Comp_{SS}}\;(\Conid{SS}\;\Varid{f}\;\Varid{a})\;\Varid{f}\;(){}\<[E]%
\\
\>[B]{}\Varid{append_{SS}}\;\Varid{x}\;\Varid{p}\mathrel{=}\mathbf{do}{}\<[E]%
\\
\>[B]{}\hsindent{3}{}\<[3]%
\>[3]{}\Conid{SS}\;\Varid{xs}\;\Varid{stack}\leftarrow \Varid{get}{}\<[E]%
\\
\>[B]{}\hsindent{3}{}\<[3]%
\>[3]{}\Varid{put}\;(\Conid{SS}\;(\Varid{xs}+\!\!+[\mskip1.5mu \Varid{x}\mskip1.5mu])\;\Varid{stack}){}\<[E]%
\\
\>[B]{}\hsindent{3}{}\<[3]%
\>[3]{}\Varid{p}{}\<[E]%
\ColumnHook
\end{hscode}\resethooks
\indentend \caption{Appending a result.}
\label{fig:append-ss}
\end{subfigure}%
\caption{Auxiliary functions \ensuremath{\Varid{pop_{SS}}}, \ensuremath{\Varid{push_{SS}}} and \ensuremath{\Varid{append_{SS}}}.}
\label{fig:pop-push-append-SS}
\end{figure}

\vspace{-\baselineskip}

We also modify the three auxiliary functions in \Cref{fig:pop-push-append}
to \ensuremath{\Varid{pop_{SS}}}, \ensuremath{\Varid{push_{SS}}} and \ensuremath{\Varid{append_{SS}}} in \Cref{fig:pop-push-append-SS}.
They are almost the same as the previous versions apart from being
adapted to use the new state-wrapper type \ensuremath{\Conid{SS}\;\Varid{f}\;\Varid{a}}.

The simulation function \ensuremath{\Varid{nondet2state}} is also very similar to
\ensuremath{\Varid{nondet2state_S}} except for requiring a forwarding algebra \ensuremath{\Varid{fwd}} to
deal with the additional effects in \ensuremath{\Varid{f}}.

\indentbegin \begin{hscode}\SaveRestoreHook
\column{B}{@{}>{\hspre}l<{\hspost}@{}}%
\column{3}{@{}>{\hspre}l<{\hspost}@{}}%
\column{5}{@{}>{\hspre}l<{\hspost}@{}}%
\column{15}{@{}>{\hspre}l<{\hspost}@{}}%
\column{19}{@{}>{\hspre}l<{\hspost}@{}}%
\column{E}{@{}>{\hspre}l<{\hspost}@{}}%
\>[B]{}\Varid{nondet2state}{}\<[15]%
\>[15]{}\mathbin{::}\Conid{Functor}\;\Varid{f}\Rightarrow \Conid{Free}\;(\Varid{Nondet_{F}}\mathrel{{:}{+}{:}}\Varid{f})\;\Varid{a}\to \Conid{Free}\;(\Varid{State_{F}}\;(\Conid{SS}\;\Varid{f}\;\Varid{a})\mathrel{{:}{+}{:}}\Varid{f})\;(){}\<[E]%
\\
\>[B]{}\Varid{nondet2state}\mathrel{=}\Varid{fold}\;\Varid{gen}\;(\Varid{alg}\mathbin{\#}\Varid{fwd}){}\<[E]%
\\
\>[B]{}\hsindent{3}{}\<[3]%
\>[3]{}\mathbf{where}{}\<[E]%
\\
\>[3]{}\hsindent{2}{}\<[5]%
\>[5]{}\Varid{gen}\;\Varid{x}{}\<[19]%
\>[19]{}\mathrel{=}\Varid{append_{SS}}\;\Varid{x}\;\Varid{pop_{SS}}{}\<[E]%
\\
\>[3]{}\hsindent{2}{}\<[5]%
\>[5]{}\Varid{alg}\;\Conid{Fail}{}\<[19]%
\>[19]{}\mathrel{=}\Varid{pop_{SS}}{}\<[E]%
\\
\>[3]{}\hsindent{2}{}\<[5]%
\>[5]{}\Varid{alg}\;(\Conid{Or}\;\Varid{p}\;\Varid{q}){}\<[19]%
\>[19]{}\mathrel{=}\Varid{push_{SS}}\;\Varid{q}\;\Varid{p}{}\<[E]%
\\
\>[3]{}\hsindent{2}{}\<[5]%
\>[5]{}\Varid{fwd}\;\Varid{y}{}\<[19]%
\>[19]{}\mathrel{=}\Conid{Op}\;(\Conid{Inr}\;\Varid{y}){}\<[E]%
\ColumnHook
\end{hscode}\resethooks
\indentend 
The function \ensuremath{\Varid{run_{ND+f}}} puts everything together: it translates the
nondeterminism effect into the state effect and forwards other
effects using \ensuremath{\Varid{nondet2state}}, then handles the state effect using
\ensuremath{\Varid{h_{State}}}, and finally extracts the results from the final state
using \ensuremath{\Varid{extract_{SS}}}.
\indentbegin \begin{hscode}\SaveRestoreHook
\column{B}{@{}>{\hspre}l<{\hspost}@{}}%
\column{28}{@{}>{\hspre}l<{\hspost}@{}}%
\column{E}{@{}>{\hspre}l<{\hspost}@{}}%
\>[B]{}\Varid{run_{ND+f}}\mathbin{::}\Conid{Functor}\;\Varid{f}\Rightarrow \Conid{Free}\;(\Varid{Nondet_{F}}\mathrel{{:}{+}{:}}\Varid{f})\;\Varid{a}\to \Conid{Free}\;\Varid{f}\;[\mskip1.5mu \Varid{a}\mskip1.5mu]{}\<[E]%
\\
\>[B]{}\Varid{run_{ND+f}}\mathrel{=}\Varid{extract_{SS}}\hsdot{\circ }{.}\Varid{h_{State}}\hsdot{\circ }{.}\Varid{nondet2state}{}\<[E]%
\\
\>[B]{}\Varid{extract_{SS}}\mathbin{::}\Conid{Functor}\;\Varid{f}\Rightarrow \Conid{StateT}\;(\Conid{SS}\;\Varid{f}\;\Varid{a})\;(\Conid{Free}\;\Varid{f})\;()\to \Conid{Free}\;\Varid{f}\;[\mskip1.5mu \Varid{a}\mskip1.5mu]{}\<[E]%
\\
\>[B]{}\Varid{extract_{SS}}\;\Varid{x}\mathrel{=}\Varid{results_{SS}}\hsdot{\circ }{.}{}\<[28]%
\>[28]{}\Varid{snd}\mathbin{\langle\hspace{1.6pt}\mathclap{\raisebox{0.1pt}{\scalebox{1}{\$}}}\hspace{1.6pt}\rangle}\Varid{run_{StateT}}\;\Varid{x}\;(\Conid{SS}\;[\mskip1.5mu \mskip1.5mu]\;[\mskip1.5mu \mskip1.5mu]){}\<[E]%
\ColumnHook
\end{hscode}\resethooks
\indentend 

We have the following theorem showing that the simulation \ensuremath{\Varid{run_{ND+f}}} is
equivalent to the modular nondeterminism handler \ensuremath{\Varid{h_{ND+f}}} in
\Cref{sec:combining-effects}.

\begin{restatable}[]{theorem}{nondetState}
\label{thm:nondet-state}\indentbegin \begin{hscode}\SaveRestoreHook
\column{B}{@{}>{\hspre}l<{\hspost}@{}}%
\column{3}{@{}>{\hspre}l<{\hspost}@{}}%
\column{E}{@{}>{\hspre}l<{\hspost}@{}}%
\>[3]{}\Varid{run_{ND+f}}\mathrel{=}\Varid{h_{ND+f}}{}\<[E]%
\ColumnHook
\end{hscode}\resethooks
\indentend \end{restatable}

The proof proceeds essentially in the same way as in the non-modular
setting.  The main difference, due to the modularity, is an additional
proof case for the forwarding algebra.
\[\ba{rl}
    &\ensuremath{(\Varid{extract_{SS}}\hsdot{\circ }{.}\Varid{h_{State}})\hsdot{\circ }{.}\Varid{fwd}\hsdot{\circ }{.}\Varid{fmap}\;\Varid{nondet2state_S}}\\
 \ensuremath{\mathrel{=}}&  \ensuremath{\Varid{fwd_{ND+f}}\hsdot{\circ }{.}\Varid{fmap}\;(\Varid{extract_{SS}}\hsdot{\circ }{.}\Varid{h_{State}})\hsdot{\circ }{.}\Varid{fmap}\;\Varid{nondet2state_S}}
\ea\]
The full proof can be found in \Cref{app:in-combination-with-other-effects}.

\section{All in One}
\label{sec:combination}

This section combines the results of the previous two sections to ultimately simulate the combination of
nondeterminism and state with a single state effect.

\subsection{Modelling Two States with One State}
\label{sec:multiple-states}

When we combine the two simulation steps from the two previous sections, we end
up with a computation that features two state effects. The first state effect
is the one present originally, and the second state effect keeps track of the
results and the stack of remaining branches to simulate the nondeterminism.

For a computation of type \ensuremath{\Conid{Free}\;(\Varid{State_{F}}\;\Varid{s}_{1}\mathrel{{:}{+}{:}}\Varid{State_{F}}\;\Varid{s}_{2}\mathrel{{:}{+}{:}}\Varid{f})\;\Varid{a}}
that features two state effects,
we can go to a slightly more primitive representation
\ensuremath{\Conid{Free}\;(\Varid{State_{F}}\;(\Varid{s}_{1},\Varid{s}_{2})\mathrel{{:}{+}{:}}\Varid{f})\;\Varid{a}} featuring only a single state
effect that is a pair of the previous two states.

The handler \ensuremath{\Varid{states2state}} implements the simulation by projecting
different \ensuremath{\Varid{get}} and \ensuremath{\Varid{put}} operations to different components of the
pair of states.

\indentbegin \begin{hscode}\SaveRestoreHook
\column{B}{@{}>{\hspre}l<{\hspost}@{}}%
\column{3}{@{}>{\hspre}l<{\hspost}@{}}%
\column{5}{@{}>{\hspre}l<{\hspost}@{}}%
\column{15}{@{}>{\hspre}l<{\hspost}@{}}%
\column{23}{@{}>{\hspre}l<{\hspost}@{}}%
\column{40}{@{}>{\hspre}l<{\hspost}@{}}%
\column{45}{@{}>{\hspre}l<{\hspost}@{}}%
\column{E}{@{}>{\hspre}l<{\hspost}@{}}%
\>[B]{}\Varid{states2state}{}\<[15]%
\>[15]{}\mathbin{::}\Conid{Functor}\;\Varid{f}{}\<[E]%
\\
\>[15]{}\Rightarrow \Conid{Free}\;(\Varid{State_{F}}\;\Varid{s}_{1}\mathrel{{:}{+}{:}}\Varid{State_{F}}\;\Varid{s}_{2}\mathrel{{:}{+}{:}}\Varid{f})\;\Varid{a}{}\<[E]%
\\
\>[15]{}\to \Conid{Free}\;(\Varid{State_{F}}\;(\Varid{s}_{1},\Varid{s}_{2})\mathrel{{:}{+}{:}}\Varid{f})\;\Varid{a}{}\<[E]%
\\
\>[B]{}\Varid{states2state}{}\<[15]%
\>[15]{}\mathrel{=}\Varid{fold}\;\Conid{Var}\;(\Varid{alg}_{1}\mathbin{\#}\Varid{alg}_{2}\mathbin{\#}\Varid{fwd}){}\<[E]%
\\
\>[B]{}\hsindent{3}{}\<[3]%
\>[3]{}\mathbf{where}{}\<[E]%
\\
\>[3]{}\hsindent{2}{}\<[5]%
\>[5]{}\Varid{alg}_{1}\;(\Conid{Get}\;\Varid{k}){}\<[23]%
\>[23]{}\mathrel{=}\Varid{get}>\!\!>\!\!=\lambda (\Varid{s}_{1},{}\<[40]%
\>[40]{}\anonymous ){}\<[45]%
\>[45]{}\to \Varid{k}\;\Varid{s}_{1}{}\<[E]%
\\
\>[3]{}\hsindent{2}{}\<[5]%
\>[5]{}\Varid{alg}_{1}\;(\Conid{Put}\;\Varid{s}_{1}'\;\Varid{k}){}\<[23]%
\>[23]{}\mathrel{=}\Varid{get}>\!\!>\!\!=\lambda (\anonymous ,{}\<[40]%
\>[40]{}\Varid{s}_{2}){}\<[45]%
\>[45]{}\to \Varid{put}\;(\Varid{s}_{1}',\Varid{s}_{2})>\!\!>\Varid{k}{}\<[E]%
\\
\>[3]{}\hsindent{2}{}\<[5]%
\>[5]{}\Varid{alg}_{2}\;(\Conid{Get}\;\Varid{k}){}\<[23]%
\>[23]{}\mathrel{=}\Varid{get}>\!\!>\!\!=\lambda (\anonymous ,{}\<[40]%
\>[40]{}\Varid{s}_{2}){}\<[45]%
\>[45]{}\to \Varid{k}\;\Varid{s}_{2}{}\<[E]%
\\
\>[3]{}\hsindent{2}{}\<[5]%
\>[5]{}\Varid{alg}_{2}\;(\Conid{Put}\;\Varid{s}_{2}'\;\Varid{k}){}\<[23]%
\>[23]{}\mathrel{=}\Varid{get}>\!\!>\!\!=\lambda (\Varid{s}_{1},{}\<[40]%
\>[40]{}\anonymous ){}\<[45]%
\>[45]{}\to \Varid{put}\;(\Varid{s}_{1},\Varid{s}_{2}')>\!\!>\Varid{k}{}\<[E]%
\\
\>[3]{}\hsindent{2}{}\<[5]%
\>[5]{}\Varid{fwd}\;\Varid{op}{}\<[23]%
\>[23]{}\mathrel{=}\Conid{Op}\;(\Conid{Inr}\;\Varid{op}){}\<[E]%
\ColumnHook
\end{hscode}\resethooks
\indentend %

We have the following theorem showing the correctness of \ensuremath{\Varid{states2state}}:
\begin{restatable}[]{theorem}{statesState}
\label{thm:states-state}\indentbegin \begin{hscode}\SaveRestoreHook
\column{B}{@{}>{\hspre}l<{\hspost}@{}}%
\column{3}{@{}>{\hspre}l<{\hspost}@{}}%
\column{E}{@{}>{\hspre}l<{\hspost}@{}}%
\>[3]{}\Varid{h_{States}}\mathrel{=}\Varid{nest}\hsdot{\circ }{.}\Varid{h_{State}}\hsdot{\circ }{.}\Varid{states2state}{}\<[E]%
\ColumnHook
\end{hscode}\resethooks
\indentend \end{restatable}
\noindent
On the left-hand side, we write \ensuremath{\Varid{h_{States}}} for the composition of two
consecutive state handlers:
\indentbegin \begin{hscode}\SaveRestoreHook
\column{B}{@{}>{\hspre}l<{\hspost}@{}}%
\column{E}{@{}>{\hspre}l<{\hspost}@{}}%
\>[B]{}\Varid{h_{States}}\mathbin{::}\Conid{Functor}\;\Varid{f}\Rightarrow \Conid{Free}\;(\Varid{State_{F}}\;\Varid{s}_{1}\mathrel{{:}{+}{:}}\Varid{State_{F}}\;\Varid{s}_{2}\mathrel{{:}{+}{:}}\Varid{f})\;\Varid{a}\to \Conid{StateT}\;\Varid{s}_{1}\;(\Conid{StateT}\;\Varid{s}_{2}\;(\Conid{Free}\;\Varid{f}))\;\Varid{a}{}\<[E]%
\\
\>[B]{}\Varid{h_{States}}\;\Varid{x}\mathrel{=}\Conid{StateT}\;(\Varid{h_{State}}\hsdot{\circ }{.}\Varid{run_{StateT}}\;(\Varid{h_{State}}\;\Varid{x})){}\<[E]%
\\
\>[B]{}\Varid{h_{States}^\prime}\mathbin{::}\Conid{Functor}\;\Varid{f}\Rightarrow \Conid{Free}\;(\Varid{State_{F}}\;\Varid{s}_{1}\mathrel{{:}{+}{:}}\Varid{State_{F}}\;\Varid{s}_{2}\mathrel{{:}{+}{:}}\Varid{f})\;\Varid{a}\to \Conid{StateT}\;(\Varid{s}_{1},\Varid{s}_{2})\;(\Conid{Free}\;\Varid{f})\;\Varid{a}{}\<[E]%
\\
\>[B]{}\Varid{h_{States}^\prime}\;\Varid{t}\mathrel{=}\Conid{StateT}\mathbin{\$}\lambda (\Varid{s}_{1},\Varid{s}_{2})\to \alpha\mathbin{\langle\hspace{1.6pt}\mathclap{\raisebox{0.1pt}{\scalebox{1}{\$}}}\hspace{1.6pt}\rangle}\Varid{run_{StateT}}\;(\Varid{h_{State}}\;(\Varid{run_{StateT}}\;(\Varid{h_{State}}\;\Varid{t})\;\Varid{s}_{1}))\;\Varid{s}_{2}{}\<[E]%
\ColumnHook
\end{hscode}\resethooks
\indentend On the right-hand side, we use the isomorphism \ensuremath{\Varid{nest}} to mediate
between the two different carrier types. The definition of \ensuremath{\Varid{nest}} and
its inverse \ensuremath{\Varid{flatten}} are defined as follows:
\indentbegin \begin{hscode}\SaveRestoreHook
\column{B}{@{}>{\hspre}l<{\hspost}@{}}%
\column{13}{@{}>{\hspre}l<{\hspost}@{}}%
\column{30}{@{}>{\hspre}l<{\hspost}@{}}%
\column{E}{@{}>{\hspre}l<{\hspost}@{}}%
\>[B]{}\Varid{nest}{}\<[13]%
\>[13]{}\mathbin{::}\Conid{Functor}\;\Varid{f}\Rightarrow {}\<[30]%
\>[30]{}\Conid{StateT}\;(\Varid{s}_{1},\Varid{s}_{2})\;(\Conid{Free}\;\Varid{f})\;\Varid{a}\to \Conid{StateT}\;\Varid{s}_{1}\;(\Conid{StateT}\;\Varid{s}_{2}\;(\Conid{Free}\;\Varid{f}))\;\Varid{a}{}\<[E]%
\\
\>[B]{}\Varid{nest}\;\Varid{t}{}\<[13]%
\>[13]{}\mathrel{=}\Conid{StateT}\mathbin{\$}\lambda \Varid{s}_{1}\to \Conid{StateT}\mathbin{\$}\lambda \Varid{s}_{2}\to \alpha^{-1}\mathbin{\langle\hspace{1.6pt}\mathclap{\raisebox{0.1pt}{\scalebox{1}{\$}}}\hspace{1.6pt}\rangle}\Varid{run_{StateT}}\;\Varid{t}\;(\Varid{s}_{1},\Varid{s}_{2}){}\<[E]%
\\
\>[B]{}\Varid{flatten}{}\<[13]%
\>[13]{}\mathbin{::}\Conid{Functor}\;\Varid{f}\Rightarrow {}\<[30]%
\>[30]{}\Conid{StateT}\;\Varid{s}_{1}\;(\Conid{StateT}\;\Varid{s}_{2}\;(\Conid{Free}\;\Varid{f}))\;\Varid{a}\to \Conid{StateT}\;(\Varid{s}_{1},\Varid{s}_{2})\;(\Conid{Free}\;\Varid{f})\;\Varid{a}{}\<[E]%
\\
\>[B]{}\Varid{flatten}\;\Varid{t}{}\<[13]%
\>[13]{}\mathrel{=}\Conid{StateT}\mathbin{\$}\lambda (\Varid{s}_{1},\Varid{s}_{2})\to \alpha\mathbin{\langle\hspace{1.6pt}\mathclap{\raisebox{0.1pt}{\scalebox{1}{\$}}}\hspace{1.6pt}\rangle}\Varid{run_{StateT}}\;(\Varid{run_{StateT}}\;\Varid{t}\;\Varid{s}_{1})\;\Varid{s}_{2}{}\<[E]%
\ColumnHook
\end{hscode}\resethooks
\indentend where the isomorphism \ensuremath{\alpha^{-1}} and its inverse \ensuremath{\alpha} rearrange a
nested tuple

\begin{minipage}[t]{0.5\textwidth}
\indentbegin \begin{hscode}\SaveRestoreHook
\column{B}{@{}>{\hspre}l<{\hspost}@{}}%
\column{9}{@{}>{\hspre}l<{\hspost}@{}}%
\column{21}{@{}>{\hspre}l<{\hspost}@{}}%
\column{E}{@{}>{\hspre}l<{\hspost}@{}}%
\>[B]{}\alpha{}\<[9]%
\>[9]{}\mathbin{::}((\Varid{a},\Varid{x}),\Varid{y})\to (\Varid{a},(\Varid{x},\Varid{y})){}\<[E]%
\\
\>[B]{}\alpha\;((\Varid{a},\Varid{x}),\Varid{y}){}\<[21]%
\>[21]{}\mathrel{=}(\Varid{a},(\Varid{x},\Varid{y})){}\<[E]%
\ColumnHook
\end{hscode}\resethooks
\indentend \end{minipage}
\begin{minipage}[t]{0.5\textwidth}
\indentbegin \begin{hscode}\SaveRestoreHook
\column{B}{@{}>{\hspre}l<{\hspost}@{}}%
\column{9}{@{}>{\hspre}l<{\hspost}@{}}%
\column{21}{@{}>{\hspre}l<{\hspost}@{}}%
\column{E}{@{}>{\hspre}l<{\hspost}@{}}%
\>[B]{}\alpha^{-1}{}\<[9]%
\>[9]{}\mathbin{::}(\Varid{a},(\Varid{x},\Varid{y}))\to ((\Varid{a},\Varid{x}),\Varid{y}){}\<[E]%
\\
\>[B]{}\alpha^{-1}\;(\Varid{a},(\Varid{x},\Varid{y})){}\<[21]%
\>[21]{}\mathrel{=}((\Varid{a},\Varid{x}),\Varid{y}){}\<[E]%
\ColumnHook
\end{hscode}\resethooks
\indentend \end{minipage}
The proof of \Cref{thm:states-state} can be found in
\Cref{app:states-state}. Instead of proving it directly, we show the
correctness of the isomorphism of \ensuremath{\Varid{nest}} and \ensuremath{\Varid{flatten}}, and prove the
following equation:\indentbegin \begin{hscode}\SaveRestoreHook
\column{B}{@{}>{\hspre}l<{\hspost}@{}}%
\column{3}{@{}>{\hspre}l<{\hspost}@{}}%
\column{E}{@{}>{\hspre}l<{\hspost}@{}}%
\>[3]{}\Varid{flatten}\hsdot{\circ }{.}\Varid{h_{States}}\mathrel{=}\Varid{h_{State}}\hsdot{\circ }{.}\Varid{states2state}{}\<[E]%
\ColumnHook
\end{hscode}\resethooks
\indentend %

The following commuting diagram simmuarises the simulation.

\[\begin{tikzcd}
  {\ensuremath{\Conid{Free}\;(\Varid{State_{F}}\;\Varid{s}_{1}\mathrel{{:}{+}{:}}\Varid{State_{F}}\;\Varid{s}_{2}\mathrel{{:}{+}{:}}\Varid{f})\;\Varid{a}}} && {\ensuremath{\Conid{StateT}\;\Varid{s}_{1}\;(\Conid{StateT}\;\Varid{s}_{2}\;(\Conid{Free}\;\Varid{f}))\;\Varid{a}}} \\
  \\
  {\ensuremath{\Conid{Free}\;(\Varid{State_{F}}\;(\Varid{s}_{1},\Varid{s}_{2})\mathrel{{:}{+}{:}}\Varid{f})\;\Varid{a}}} && {\ensuremath{\Conid{StateT}\;(\Varid{s}_{1},\Varid{s}_{2})\;(\Conid{Free}\;\Varid{f})\;\Varid{a}}}
  \arrow["{\ensuremath{\Varid{flatten}}}"', shift right=5, from=1-3, to=3-3]
  \arrow["{\ensuremath{\Varid{nest}}}"', shift right=5, from=3-3, to=1-3]
  \arrow["{\ensuremath{\Varid{h_{States}}}}", from=1-1, to=1-3]
  \arrow["{\ensuremath{\Varid{h_{State}}}}"', from=3-1, to=3-3]
  \arrow["{\ensuremath{\Varid{states2state}}}"', from=1-1, to=3-1]
\end{tikzcd}\]

\subsection{Putting Everything Together}\
\label{sec:final-simulate}
We have defined three translations for encoding high-level effects as
low-level effects.
\begin{itemize}
  \item The function \ensuremath{\Varid{local2global}} simulates the high-level
  local-state semantics with global-state semantics for the
  nondeterminism and state effects  (\Cref{sec:local-global}).
  \item The function \ensuremath{\Varid{nondet2state}} simulates the high-level
  nondeterminism effect with the state effect
  (\Cref{sec:nondeterminism-state}).
  \item The function \ensuremath{\Varid{states2state}} simulates multiple state effects
  with a single state effect (\Cref{sec:multiple-states}).
\end{itemize}

Combining them, we can encode the local-state semantics
for nondeterminism and state with just one state effect. The ultimate
simulation function \ensuremath{\Varid{simulate}} is defined as follows:
\indentbegin \begin{hscode}\SaveRestoreHook
\column{B}{@{}>{\hspre}l<{\hspost}@{}}%
\column{11}{@{}>{\hspre}l<{\hspost}@{}}%
\column{E}{@{}>{\hspre}l<{\hspost}@{}}%
\>[B]{}\Varid{simulate}{}\<[11]%
\>[11]{}\mathbin{::}\Conid{Functor}\;\Varid{f}{}\<[E]%
\\
\>[11]{}\Rightarrow \Conid{Free}\;(\Varid{State_{F}}\;\Varid{s}\mathrel{{:}{+}{:}}\Varid{Nondet_{F}}\mathrel{{:}{+}{:}}\Varid{f})\;\Varid{a}{}\<[E]%
\\
\>[11]{}\to \Varid{s}\to \Conid{Free}\;\Varid{f}\;[\mskip1.5mu \Varid{a}\mskip1.5mu]{}\<[E]%
\\
\>[B]{}\Varid{simulate}{}\<[11]%
\>[11]{}\mathrel{=}\Varid{extract}\hsdot{\circ }{.}\Varid{h_{State}}\hsdot{\circ }{.}\Varid{states2state}\hsdot{\circ }{.}\Varid{nondet2state}\hsdot{\circ }{.}\Varid{(\Leftrightarrow)}\hsdot{\circ }{.}\Varid{local2global}{}\<[E]%
\ColumnHook
\end{hscode}\resethooks
\indentend Similar to the \ensuremath{\Varid{extract_{SS}}} function in \Cref{sec:nondet2state}, we use
the \ensuremath{\Varid{extract}} function to get the final results from the final state.
\indentbegin \begin{hscode}\SaveRestoreHook
\column{B}{@{}>{\hspre}l<{\hspost}@{}}%
\column{11}{@{}>{\hspre}l<{\hspost}@{}}%
\column{E}{@{}>{\hspre}l<{\hspost}@{}}%
\>[B]{}\Varid{extract}{}\<[11]%
\>[11]{}\mathbin{::}\Conid{Functor}\;\Varid{f}{}\<[E]%
\\
\>[11]{}\Rightarrow \Conid{StateT}\;(\Conid{SS}\;(\Varid{State_{F}}\;\Varid{s}\mathrel{{:}{+}{:}}\Varid{f})\;\Varid{a},\Varid{s})\;(\Conid{Free}\;\Varid{f})\;(){}\<[E]%
\\
\>[11]{}\to \Varid{s}\to \Conid{Free}\;\Varid{f}\;[\mskip1.5mu \Varid{a}\mskip1.5mu]{}\<[E]%
\\
\>[B]{}\Varid{extract}\;\Varid{x}\;\Varid{s}\mathrel{=}\Varid{results_{SS}}\hsdot{\circ }{.}\Varid{fst}\hsdot{\circ }{.}\Varid{snd}\mathbin{\langle\hspace{1.6pt}\mathclap{\raisebox{0.1pt}{\scalebox{1}{\$}}}\hspace{1.6pt}\rangle}\Varid{run_{StateT}}\;\Varid{x}\;(\Conid{SS}\;[\mskip1.5mu \mskip1.5mu]\;[\mskip1.5mu \mskip1.5mu],\Varid{s}){}\<[E]%
\ColumnHook
\end{hscode}\resethooks
\indentend %
\Cref{fig:simulation} illustrates each step of this simulation.
\begin{figure}[h]
\[\begin{tikzcd}
	{\ensuremath{\Conid{Free}\;(\Varid{State_{F}}\;\Varid{s}\mathrel{{:}{+}{:}}\Varid{Nondet_{F}}\mathrel{{:}{+}{:}}\Varid{f})\;\Varid{a}}} \\
	{\ensuremath{\Conid{Free}\;(\Varid{State_{F}}\;\Varid{s}\mathrel{{:}{+}{:}}\Varid{Nondet_{F}}\mathrel{{:}{+}{:}}\Varid{f})\;\Varid{a}}} \\
	{\ensuremath{\Conid{Free}\;(\Varid{Nondet_{F}}\mathrel{{:}{+}{:}}\Varid{State_{F}}\;\Varid{s}\mathrel{{:}{+}{:}}\Varid{f})\;\Varid{a}}} \\
	{\ensuremath{\Conid{Free}\;(\Varid{State_{F}}\;(\Conid{SS}\;(\Varid{State_{F}}\;\Varid{s}\mathrel{{:}{+}{:}}\Varid{f})\;\Varid{a})\mathrel{{:}{+}{:}}\Varid{State_{F}}\;\Varid{s}\mathrel{{:}{+}{:}}\Varid{f})\;()}} \\
	{\ensuremath{\Conid{Free}\;(\Varid{State_{F}}\;(\Conid{SS}\;(\Varid{State_{F}}\;\Varid{s}\mathrel{{:}{+}{:}}\Varid{f})\;\Varid{a},\Varid{s})\mathrel{{:}{+}{:}}\Varid{f})\;()}} \\
	{\ensuremath{\Conid{StateT}\;(\Conid{SS}\;(\Varid{State_{F}}\;\Varid{s}\mathrel{{:}{+}{:}}\Varid{f})\;\Varid{a},\Varid{s})\;(\Conid{Free}\;\Varid{f})\;()}} \\
	{\ensuremath{\Varid{s}\to \Conid{Free}\;\Varid{f}\;[\mskip1.5mu \Varid{a}\mskip1.5mu]}}
	\arrow["{\ensuremath{\Varid{local2global}}}", from=1-1, to=2-1]
	\arrow["{\ensuremath{\Varid{(\Leftrightarrow)}}}", from=2-1, to=3-1]
	\arrow["{\ensuremath{\Varid{states2state}}}", from=4-1, to=5-1]
	\arrow["{\ensuremath{\Varid{h_{State}}}}", from=5-1, to=6-1]
	\arrow["{\ensuremath{\Varid{extract}}}", from=6-1, to=7-1]
	\arrow["{\ensuremath{\Varid{nondet2state}}}", from=3-1, to=4-1]
\end{tikzcd}\]
\caption{An overview of the \ensuremath{\Varid{simulate}} function.}
\label{fig:simulation}
\end{figure}

In the \ensuremath{\Varid{simulate}} function, we first use our three simulations
\ensuremath{\Varid{local2global}}, \ensuremath{\Varid{nondet2state}} and \ensuremath{\Varid{states2state}} to interpret the
local-state semantics for state and nondeterminism in terms of only
one state effect. Then, we use the handler \ensuremath{\Varid{h_{State}}} to interpret the
state effect into a state monad transformer. Finally, we use the
function \ensuremath{\Varid{extract}} to get the final results.

We have the following theorem showing that the \ensuremath{\Varid{simulate}} function exactly
behaves the same as the local-state semantics given by \ensuremath{\Varid{h_{Local}}}.
\begin{restatable}[]{theorem}{finalSimulate}
\label{thm:final-simulate}\indentbegin \begin{hscode}\SaveRestoreHook
\column{B}{@{}>{\hspre}l<{\hspost}@{}}%
\column{3}{@{}>{\hspre}l<{\hspost}@{}}%
\column{E}{@{}>{\hspre}l<{\hspost}@{}}%
\>[3]{}\Varid{simulate}\mathrel{=}\Varid{h_{Local}}{}\<[E]%
\ColumnHook
\end{hscode}\resethooks
\indentend \end{restatable}
The proof can be found in \Cref{app:final-simulate}.

We provide a more compact and direct definition of \ensuremath{\Varid{simulate}} by
fusing all the consecutive steps into a single handler:
\indentbegin \begin{hscode}\SaveRestoreHook
\column{B}{@{}>{\hspre}l<{\hspost}@{}}%
\column{3}{@{}>{\hspre}l<{\hspost}@{}}%
\column{5}{@{}>{\hspre}l<{\hspost}@{}}%
\column{12}{@{}>{\hspre}l<{\hspost}@{}}%
\column{19}{@{}>{\hspre}l<{\hspost}@{}}%
\column{21}{@{}>{\hspre}l<{\hspost}@{}}%
\column{39}{@{}>{\hspre}c<{\hspost}@{}}%
\column{39E}{@{}l@{}}%
\column{42}{@{}>{\hspre}l<{\hspost}@{}}%
\column{44}{@{}>{\hspre}l<{\hspost}@{}}%
\column{52}{@{}>{\hspre}l<{\hspost}@{}}%
\column{E}{@{}>{\hspre}l<{\hspost}@{}}%
\>[B]{}\mathbf{type}\;\Conid{Comp}\;\Varid{s}\;\Varid{f}\;\Varid{a}\mathrel{=}(\Conid{CP}\;\Varid{f}\;\Varid{a}\;\Varid{s},\Varid{s})\to \Conid{Free}\;\Varid{f}\;[\mskip1.5mu \Varid{a}\mskip1.5mu]{}\<[E]%
\\
\>[B]{}\mathbf{data}\;\Conid{CP}\;\Varid{f}\;\Varid{a}\;\Varid{s}\mathrel{=}\Conid{CP}\;\{\mskip1.5mu \Varid{results}\mathbin{::}[\mskip1.5mu \Varid{a}\mskip1.5mu],\Varid{cpStack}\mathbin{::}[\mskip1.5mu \Conid{Comp}\;\Varid{s}\;\Varid{f}\;\Varid{a}\mskip1.5mu]\mskip1.5mu\}{}\<[E]%
\\[\blanklineskip]%
\>[B]{}\Varid{simulate_F}{}\<[12]%
\>[12]{}\mathbin{::}\Conid{Functor}\;\Varid{f}{}\<[E]%
\\
\>[12]{}\Rightarrow \Conid{Free}\;(\Varid{State_{F}}\;\Varid{s}\mathrel{{:}{+}{:}}\Varid{Nondet_{F}}\mathrel{{:}{+}{:}}\Varid{f})\;\Varid{a}{}\<[E]%
\\
\>[12]{}\to \Varid{s}{}\<[E]%
\\
\>[12]{}\to \Conid{Free}\;\Varid{f}\;[\mskip1.5mu \Varid{a}\mskip1.5mu]{}\<[E]%
\\
\>[B]{}\Varid{simulate_F}\;{}\<[12]%
\>[12]{}\Varid{x}\;\Varid{s}\mathrel{=}{}\<[19]%
\>[19]{}\Varid{fold}\;\Varid{gen}\;(\Varid{alg}_{1}\mathbin{\#}\Varid{alg}_{2}\mathbin{\#}\Varid{fwd})\;\Varid{x}\;(\Conid{CP}\;[\mskip1.5mu \mskip1.5mu]\;[\mskip1.5mu \mskip1.5mu],\Varid{s}){}\<[E]%
\\
\>[B]{}\hsindent{3}{}\<[3]%
\>[3]{}\mathbf{where}{}\<[E]%
\\
\>[3]{}\hsindent{2}{}\<[5]%
\>[5]{}\Varid{gen}\;\Varid{x}\;{}\<[21]%
\>[21]{}(\Conid{CP}\;\Varid{xs}\;\Varid{stack},\Varid{s}){}\<[39]%
\>[39]{}\mathrel{=}{}\<[39E]%
\>[42]{}\Varid{continue}\;(\Varid{xs}+\!\!+[\mskip1.5mu \Varid{x}\mskip1.5mu])\;\Varid{stack}\;\Varid{s}{}\<[E]%
\\
\>[3]{}\hsindent{2}{}\<[5]%
\>[5]{}\Varid{alg}_{1}\;(\Conid{Get}\;\Varid{k})\;{}\<[21]%
\>[21]{}(\Conid{CP}\;\Varid{xs}\;\Varid{stack},\Varid{s}){}\<[39]%
\>[39]{}\mathrel{=}{}\<[39E]%
\>[42]{}\Varid{k}\;\Varid{s}\;(\Conid{CP}\;\Varid{xs}\;\Varid{stack},\Varid{s}){}\<[E]%
\\
\>[3]{}\hsindent{2}{}\<[5]%
\>[5]{}\Varid{alg}_{1}\;(\Conid{Put}\;\Varid{t}\;\Varid{k})\;{}\<[21]%
\>[21]{}(\Conid{CP}\;\Varid{xs}\;\Varid{stack},\Varid{s}){}\<[39]%
\>[39]{}\mathrel{=}{}\<[39E]%
\>[42]{}\Varid{k}\;(\Conid{CP}\;\Varid{xs}\;(\Varid{backtracking}\;\Varid{s}\mathbin{:}\Varid{stack}),\Varid{t}){}\<[E]%
\\
\>[3]{}\hsindent{2}{}\<[5]%
\>[5]{}\Varid{alg}_{2}\;\Conid{Fail}\;{}\<[21]%
\>[21]{}(\Conid{CP}\;\Varid{xs}\;\Varid{stack},\Varid{s}){}\<[39]%
\>[39]{}\mathrel{=}{}\<[39E]%
\>[42]{}\Varid{continue}\;\Varid{xs}\;\Varid{stack}\;\Varid{s}{}\<[E]%
\\
\>[3]{}\hsindent{2}{}\<[5]%
\>[5]{}\Varid{alg}_{2}\;(\Conid{Or}\;\Varid{p}\;\Varid{q})\;{}\<[21]%
\>[21]{}(\Conid{CP}\;\Varid{xs}\;\Varid{stack},\Varid{s}){}\<[39]%
\>[39]{}\mathrel{=}{}\<[39E]%
\>[42]{}\Varid{p}\;(\Conid{CP}\;\Varid{xs}\;(\Varid{q}\mathbin{:}\Varid{stack}),\Varid{s}){}\<[E]%
\\
\>[3]{}\hsindent{2}{}\<[5]%
\>[5]{}\Varid{fwd}\;\Varid{op}\;{}\<[21]%
\>[21]{}(\Conid{CP}\;\Varid{xs}\;\Varid{stack},\Varid{s}){}\<[39]%
\>[39]{}\mathrel{=}{}\<[39E]%
\>[42]{}\Conid{Op}\;(\Varid{fmap}\;(\mathbin{\$}(\Conid{CP}\;\Varid{xs}\;\Varid{stack},\Varid{s}))\;\Varid{op}){}\<[E]%
\\
\>[3]{}\hsindent{2}{}\<[5]%
\>[5]{}\Varid{backtracking}\;\Varid{s}\;{}\<[21]%
\>[21]{}(\Conid{CP}\;\Varid{xs}\;\Varid{stack},\anonymous ){}\<[39]%
\>[39]{}\mathrel{=}{}\<[39E]%
\>[42]{}\Varid{continue}\;\Varid{xs}\;\Varid{stack}\;\Varid{s}{}\<[E]%
\\
\>[3]{}\hsindent{2}{}\<[5]%
\>[5]{}\Varid{continue}\;\Varid{xs}\;\Varid{stack}\;\Varid{s}{}\<[39]%
\>[39]{}\mathrel{=}{}\<[39E]%
\>[42]{}\mathbf{case}\;\Varid{stack}\;\mathbf{of}{}\<[E]%
\\
\>[42]{}\hsindent{2}{}\<[44]%
\>[44]{}[\mskip1.5mu \mskip1.5mu]{}\<[52]%
\>[52]{}\to \Varid{\eta}\;\Varid{xs}{}\<[E]%
\\
\>[42]{}\hsindent{2}{}\<[44]%
\>[44]{}(\Varid{p}\mathbin{:}\Varid{ps}){}\<[52]%
\>[52]{}\to \Varid{p}\;(\Conid{CP}\;\Varid{xs}\;\Varid{ps},\Varid{s}){}\<[E]%
\ColumnHook
\end{hscode}\resethooks
\indentend The common carrier of the above algebras \ensuremath{\Varid{alg}_{1}\mathbin{\#}\Varid{alg}_{2}\mathbin{\#}\Varid{fwd}} is \ensuremath{\Conid{Comp}\;\Varid{s}\;\Varid{f}\;\Varid{a}}. This is a computation
that takes the current results, choicepoint stack and application state, and returns the list of all results.  The
first two inputs are bundled in the \ensuremath{\Conid{CP}} type.

\paragraph*{N-queens with Only One State}\
With \ensuremath{\Varid{simulate}}, we can implement the backtracking algorithm of the
n-queens problem in \Cref{sec:motivation-and-challenges} with only
one state effect as follows.

\indentbegin \begin{hscode}\SaveRestoreHook
\column{B}{@{}>{\hspre}l<{\hspost}@{}}%
\column{12}{@{}>{\hspre}l<{\hspost}@{}}%
\column{E}{@{}>{\hspre}l<{\hspost}@{}}%
\>[B]{}\Varid{queens_{Sim}}{}\<[12]%
\>[12]{}\mathbin{::}\Conid{Int}\to [\mskip1.5mu [\mskip1.5mu \Conid{Int}\mskip1.5mu]\mskip1.5mu]{}\<[E]%
\\
\>[B]{}\Varid{queens_{Sim}}{}\<[12]%
\>[12]{}\mathrel{=}\Varid{h_{Nil}}\hsdot{\circ }{.}\Varid{flip}\;\Varid{simulate}\;(\mathrm{0},[\mskip1.5mu \mskip1.5mu])\hsdot{\circ }{.}\Varid{queens}{}\<[E]%
\ColumnHook
\end{hscode}\resethooks
\indentend

\section{Modelling Local State with Undo}
\label{sec:undo}

In \Cref{sec:local2global}, we give a translation \ensuremath{\Varid{local2global}} which
simulates local state with global state by replacing \ensuremath{\Varid{put}} with its
state-restoring version \ensuremath{\Varid{put_{R}}}.  The \ensuremath{\Varid{put_{R}}} operation makes the
implicit copying of the local-state semantics explicit in the
global-state semantics. However, this copying still exists and can be
rather costly if the state is big (e.g., a long array), and especially
wasteful if the modifications made to that state are small (e.g., a
single entry in the array).
Fortunately, low-level features like the global-state semantics
give us more possibility to apply more fine-grained optimisation
strategies.
As a result, instead of copying the whole state to implement the
backtracking behaviour in the global-state semantics, we can just keep
track of the modifications made to the state, and undo them when
necessary.
In this section, we formalise this intuition with a translation from
the local-state semantics to the global-state semantics for
incremental and reversible state updates.

\subsection{State Update and Restoration}

We first need to characterise a specific subset of state effects where
all state update operations can be undone. We call them
modification-based state effects.
For example, the \ensuremath{\Varid{queens}} program in
\Cref{sec:motivation-and-challenges} uses the operation \ensuremath{\Varid{s}\mathbin{\oplus}\Varid{r}}
to update the state.
We can undo it using the following \ensuremath{\mathbin{\ominus}\Varid{r}} operation which is
essentially the left inverse of \ensuremath{\mathbin{\oplus}\Varid{r}}.
\indentbegin \begin{hscode}\SaveRestoreHook
\column{B}{@{}>{\hspre}l<{\hspost}@{}}%
\column{9}{@{}>{\hspre}l<{\hspost}@{}}%
\column{E}{@{}>{\hspre}l<{\hspost}@{}}%
\>[B]{}(\mathbin{\ominus}){}\<[9]%
\>[9]{}\mathbin{::}(\Conid{Int},[\mskip1.5mu \Conid{Int}\mskip1.5mu])\to \Conid{Int}\to (\Conid{Int},[\mskip1.5mu \Conid{Int}\mskip1.5mu]){}\<[E]%
\\
\>[B]{}(\mathbin{\ominus})\;{}\<[9]%
\>[9]{}(\Varid{c},\Varid{sol})\;\Varid{r}\mathrel{=}(\Varid{c}\mathbin{-}\mathrm{1},\Varid{tail}\;\Varid{sol}){}\<[E]%
\ColumnHook
\end{hscode}\resethooks
\indentend %
These two operators satisfy the equation \ensuremath{(\mathbin{\ominus}\Varid{r})\hsdot{\circ }{.}(\mathbin{\oplus}\Varid{r})\mathrel{=}\Varid{id}}
for any \ensuremath{\Varid{r}\mathbin{::}\Conid{Int}}.

In general, we define a typeclass \ensuremath{\Conid{Undo}\;\Varid{s}\;\Varid{r}} with two operations
\ensuremath{(\mathbin{\oplus})} and \ensuremath{(\mathbin{\ominus})} to characterise restorable state updates. Here, \ensuremath{\Varid{s}}
is the type of states and \ensuremath{\Varid{r}} is the type of deltas.  We can
implement the previous state update and restoration operations of
n-queens as an instance \ensuremath{\Conid{Undo}\;(\Conid{Int},[\mskip1.5mu \Conid{Int}\mskip1.5mu])\;\Conid{Int}} of the typeclass.
\indentbegin \begin{hscode}\SaveRestoreHook
\column{B}{@{}>{\hspre}l<{\hspost}@{}}%
\column{3}{@{}>{\hspre}l<{\hspost}@{}}%
\column{10}{@{}>{\hspre}l<{\hspost}@{}}%
\column{21}{@{}>{\hspre}l<{\hspost}@{}}%
\column{E}{@{}>{\hspre}l<{\hspost}@{}}%
\>[B]{}\mathbf{class}\;\Conid{Undo}\;\Varid{s}\;\Varid{r}\;\mathbf{where}{}\<[E]%
\\
\>[B]{}\hsindent{3}{}\<[3]%
\>[3]{}(\mathbin{\oplus}){}\<[10]%
\>[10]{}\mathbin{::}\Varid{s}\to \Varid{r}\to \Varid{s}{}\<[E]%
\\
\>[B]{}\hsindent{3}{}\<[3]%
\>[3]{}(\mathbin{\ominus}){}\<[10]%
\>[10]{}\mathbin{::}\Varid{s}\to \Varid{r}\to \Varid{s}{}\<[E]%
\\
\>[B]{}\mathbf{instance}\;\Conid{Undo}\;(\Conid{Int},[\mskip1.5mu \Conid{Int}\mskip1.5mu])\;\Conid{Int}\;\mathbf{where}{}\<[E]%
\\
\>[B]{}\hsindent{3}{}\<[3]%
\>[3]{}(\mathbin{\oplus})\;(\Varid{c},\Varid{sol})\;\Varid{r}{}\<[21]%
\>[21]{}\mathrel{=}(\Varid{c}\mathbin{+}\mathrm{1},\Varid{r}\mathbin{:}\Varid{sol}){}\<[E]%
\\
\>[B]{}\hsindent{3}{}\<[3]%
\>[3]{}(\mathbin{\ominus})\;(\Varid{c},\Varid{sol})\;\Varid{r}{}\<[21]%
\>[21]{}\mathrel{=}(\Varid{c}\mathbin{-}\mathrm{1},\Varid{tail}\;\Varid{sol}){}\<[E]%
\ColumnHook
\end{hscode}\resethooks
\indentend %
Instances of \ensuremath{\Conid{Undo}} should satisfy the following law which says
\ensuremath{\mathbin{\ominus}\Varid{x}} is a left inverse of \ensuremath{\mathbin{\oplus}\Varid{x}}:
\begin{alignat}{2}
    &\mbox{\bf plus-minus}:\quad &
      \ensuremath{(\mathbin{\ominus}\Varid{x})\hsdot{\circ }{.}(\mathbin{\oplus}\Varid{x})} ~=~ & \ensuremath{\Varid{id}} \label{eq:plus-minus} \mbox{~~.}
\end{alignat}

Modification-based state effects restrict the general \ensuremath{\Varid{put}} operation
of \ensuremath{\Conid{MState}} to modification oprations.  We define a new typeclass
\ensuremath{\Conid{MModify}\;\Varid{s}\;\Varid{r}\;\Varid{m}} which inherits from \ensuremath{\Conid{Monad}\;\Varid{m}} and \ensuremath{\Conid{Undo}\;\Varid{s}\;\Varid{r}} to
capture the interfaces of state updates and restoration.  It has three
operations: a \ensuremath{\Varid{mget}} operation that reads and returns the state
(similar to the \ensuremath{\Varid{get}} operation of \ensuremath{\Conid{MState}}), a \ensuremath{\Varid{update}\;\Varid{r}} operation
that updates the state with the delta \ensuremath{\Varid{r}}, and a \ensuremath{\Varid{restore}\;\Varid{r}}
operations that restores the update introduced by the delta \ensuremath{\Varid{r}}.
Note that only the \ensuremath{\Varid{mget}} and \ensuremath{\Varid{update}} operations are expected to be
used by programmers; \ensuremath{\Varid{restore}} operations are automatically generated
by the translation to the global-state semantics.

\indentbegin \begin{hscode}\SaveRestoreHook
\column{B}{@{}>{\hspre}l<{\hspost}@{}}%
\column{5}{@{}>{\hspre}l<{\hspost}@{}}%
\column{14}{@{}>{\hspre}l<{\hspost}@{}}%
\column{E}{@{}>{\hspre}l<{\hspost}@{}}%
\>[B]{}\mathbf{class}\;(\Conid{Monad}\;\Varid{m},\Conid{Undo}\;\Varid{s}\;\Varid{r})\Rightarrow \Conid{MModify}\;\Varid{s}\;\Varid{r}\;\Varid{m}\mid \Varid{m}\to \Varid{s},\Varid{m}\to \Varid{r}\;\mathbf{where}{}\<[E]%
\\
\>[B]{}\hsindent{5}{}\<[5]%
\>[5]{}\Varid{mget}{}\<[14]%
\>[14]{}\mathbin{::}\Varid{m}\;\Varid{s}{}\<[E]%
\\
\>[B]{}\hsindent{5}{}\<[5]%
\>[5]{}\Varid{update}{}\<[14]%
\>[14]{}\mathbin{::}\Varid{r}\to \Varid{m}\;(){}\<[E]%
\\
\>[B]{}\hsindent{5}{}\<[5]%
\>[5]{}\Varid{restore}{}\<[14]%
\>[14]{}\mathbin{::}\Varid{r}\to \Varid{m}\;(){}\<[E]%
\ColumnHook
\end{hscode}\resethooks
\indentend %
The three operations satisfy the following laws:
\begin{alignat}{2}
    &\mbox{\bf mget-mget}:\quad &
    \ensuremath{\Varid{mget}>\!\!>\!\!=(\lambda \Varid{s}\to \Varid{mget}>\!\!>\!\!=\Varid{k}\;\Varid{s})} &= \ensuremath{\Varid{mget}>\!\!>\!\!=(\lambda \Varid{s}\to \Varid{k}\;\Varid{s}\;\Varid{s})}
    ~~\mbox{,} \label{eq:mget-mget} \\
    &\mbox{\bf update-mget}:~ &
    \ensuremath{\Varid{mget}>\!\!>\!\!=\lambda \Varid{s}\to \Varid{update}\;\Varid{r}>\!\!>\Varid{\eta}\;(\Varid{s}\mathbin{\oplus}\Varid{r})}
    &=
    \ensuremath{\Varid{update}\;\Varid{r}>\!\!>\Varid{mget}}
    ~~\mbox{,} \label{eq:update-mget}\\
    &\mbox{\bf restore-mget}:~ &
    \ensuremath{\Varid{mget}>\!\!>\!\!=\lambda \Varid{s}\to \Varid{restore}\;\Varid{r}>\!\!>\Varid{\eta}\;(\Varid{s}\mathbin{\ominus}\Varid{r})}
    &=
    \ensuremath{\Varid{restore}\;\Varid{r}>\!\!>\Varid{mget}}
    ~~\mbox{,} \label{eq:restore-mget}\\
    &\mbox{\bf update-restore}:\quad &
    \ensuremath{\Varid{update}\;\Varid{r}>\!\!>\Varid{restore}\;\Varid{r}} &= \ensuremath{\Varid{\eta}\;()}
    ~~\mbox{.} \label{eq:update-restore}
\end{alignat}
The first law for \ensuremath{\Varid{mget}} corresponds to that for \ensuremath{\Varid{get}}. The second
and third law respecively capture the impact of \ensuremath{\Varid{update}} and \ensuremath{\Varid{restore}}
on \ensuremath{\Varid{mget}}. Finally, the fourth law expresses that \ensuremath{\Varid{restore}} undoes the effect of
\ensuremath{\Varid{update}}.

As what we did for the nondeterminism and state effects in
\Cref{sec:free-monads-and-their-folds}, we define a new signature
\ensuremath{\Varid{Modify_{F}}} representing the syntax of modification-based state effects,
and implement the free monad \ensuremath{\Conid{Free}\;(\Varid{Modify_{F}}\;\Varid{s}\;\Varid{r}\mathrel{{:}{+}{:}}\Varid{f})} as an instance
of \ensuremath{\Conid{MModify}\;\Varid{s}\;\Varid{r}}.
\indentbegin \begin{hscode}\SaveRestoreHook
\column{B}{@{}>{\hspre}l<{\hspost}@{}}%
\column{E}{@{}>{\hspre}l<{\hspost}@{}}%
\>[B]{}\mathbf{data}\;\Varid{Modify_{F}}\;\Varid{s}\;\Varid{r}\;\Varid{a}\mathrel{=}\Conid{MGet}\;(\Varid{s}\to \Varid{a})\mid \Conid{MUpdate}\;\Varid{r}\;\Varid{a}\mid \Conid{MRestore}\;\Varid{r}\;\Varid{a}{}\<[E]%
\ColumnHook
\end{hscode}\resethooks
\indentend %
\indentbegin \begin{hscode}\SaveRestoreHook
\column{B}{@{}>{\hspre}l<{\hspost}@{}}%
\column{3}{@{}>{\hspre}l<{\hspost}@{}}%
\column{14}{@{}>{\hspre}c<{\hspost}@{}}%
\column{14E}{@{}l@{}}%
\column{17}{@{}>{\hspre}l<{\hspost}@{}}%
\column{E}{@{}>{\hspre}l<{\hspost}@{}}%
\>[B]{}\mathbf{instance}\;(\Conid{Functor}\;\Varid{f},\Conid{Undo}\;\Varid{s}\;\Varid{r})\Rightarrow \Conid{MModify}\;\Varid{s}\;\Varid{r}\;(\Conid{Free}\;(\Varid{Modify_{F}}\;\Varid{s}\;\Varid{r}\mathrel{{:}{+}{:}}\Varid{f}))\;\mathbf{where}{}\<[E]%
\\
\>[B]{}\hsindent{3}{}\<[3]%
\>[3]{}\Varid{mget}{}\<[14]%
\>[14]{}\mathrel{=}{}\<[14E]%
\>[17]{}\Conid{Op}\;(\Conid{Inl}\;(\Conid{MGet}\;\Varid{\eta})){}\<[E]%
\\
\>[B]{}\hsindent{3}{}\<[3]%
\>[3]{}\Varid{update}\;\Varid{r}{}\<[14]%
\>[14]{}\mathrel{=}{}\<[14E]%
\>[17]{}\Conid{Op}\;(\Conid{Inl}\;(\Conid{MUpdate}\;\Varid{r}\;(\Varid{\eta}\;()))){}\<[E]%
\\
\>[B]{}\hsindent{3}{}\<[3]%
\>[3]{}\Varid{restore}\;\Varid{r}{}\<[14]%
\>[14]{}\mathrel{=}{}\<[14E]%
\>[17]{}\Conid{Op}\;(\Conid{Inl}\;(\Conid{MRestore}\;\Varid{r}\;(\Varid{\eta}\;()))){}\<[E]%
\ColumnHook
\end{hscode}\resethooks
\indentend 
The following handler \ensuremath{\Varid{h_{Modify}}} maps this free monad to the \ensuremath{\Conid{StateT}}
monad transformer using the operations \ensuremath{(\mathbin{\oplus})} and \ensuremath{(\mathbin{\ominus})} provided by
\ensuremath{\Conid{Undo}\;\Varid{s}\;\Varid{r}}.
\indentbegin \begin{hscode}\SaveRestoreHook
\column{B}{@{}>{\hspre}l<{\hspost}@{}}%
\column{3}{@{}>{\hspre}l<{\hspost}@{}}%
\column{5}{@{}>{\hspre}l<{\hspost}@{}}%
\column{25}{@{}>{\hspre}l<{\hspost}@{}}%
\column{E}{@{}>{\hspre}l<{\hspost}@{}}%
\>[B]{}\Varid{h_{Modify}}\mathbin{::}(\Conid{Functor}\;\Varid{f},\Conid{Undo}\;\Varid{s}\;\Varid{r})\Rightarrow \Conid{Free}\;(\Varid{Modify_{F}}\;\Varid{s}\;\Varid{r}\mathrel{{:}{+}{:}}\Varid{f})\;\Varid{a}\to \Conid{StateT}\;\Varid{s}\;(\Conid{Free}\;\Varid{f})\;\Varid{a}{}\<[E]%
\\
\>[B]{}\Varid{h_{Modify}}\mathrel{=}\Varid{fold}\;\Varid{gen}\;(\Varid{alg}\mathbin{\#}\Varid{fwd}){}\<[E]%
\\
\>[B]{}\hsindent{3}{}\<[3]%
\>[3]{}\mathbf{where}{}\<[E]%
\\
\>[3]{}\hsindent{2}{}\<[5]%
\>[5]{}\Varid{gen}\;\Varid{x}{}\<[25]%
\>[25]{}\mathrel{=}\Conid{StateT}\mathbin{\$}\lambda \Varid{s}\to \Varid{\eta}\;(\Varid{x},\Varid{s}){}\<[E]%
\\
\>[3]{}\hsindent{2}{}\<[5]%
\>[5]{}\Varid{alg}\;(\Conid{MGet}\;\Varid{k}){}\<[25]%
\>[25]{}\mathrel{=}\Conid{StateT}\mathbin{\$}\lambda \Varid{s}\to \Varid{run_{StateT}}\;(\Varid{k}\;\Varid{s})\;\Varid{s}{}\<[E]%
\\
\>[3]{}\hsindent{2}{}\<[5]%
\>[5]{}\Varid{alg}\;(\Conid{MUpdate}\;\Varid{r}\;\Varid{k}){}\<[25]%
\>[25]{}\mathrel{=}\Conid{StateT}\mathbin{\$}\lambda \Varid{s}\to \Varid{run_{StateT}}\;\Varid{k}\;(\Varid{s}\mathbin{\oplus}\Varid{r}){}\<[E]%
\\
\>[3]{}\hsindent{2}{}\<[5]%
\>[5]{}\Varid{alg}\;(\Conid{MRestore}\;\Varid{r}\;\Varid{k}){}\<[25]%
\>[25]{}\mathrel{=}\Conid{StateT}\mathbin{\$}\lambda \Varid{s}\to \Varid{run_{StateT}}\;\Varid{k}\;(\Varid{s}\mathbin{\ominus}\Varid{r}){}\<[E]%
\\
\>[3]{}\hsindent{2}{}\<[5]%
\>[5]{}\Varid{fwd}\;\Varid{y}{}\<[25]%
\>[25]{}\mathrel{=}\Conid{StateT}\mathbin{\$}\lambda \Varid{s}\to \Conid{Op}\;(\Varid{fmap}\;(\lambda \Varid{k}\to \Varid{run_{StateT}}\;\Varid{k}\;\Varid{s})\;\Varid{y}){}\<[E]%
\ColumnHook
\end{hscode}\resethooks
\indentend %
It is easy to check that the four laws hold contextually up to
interpretation with \ensuremath{\Varid{h_{Modify}}}.

Note that here we still use the \ensuremath{\Conid{StateT}} monad transformer and
immutable states for the clarity of presentation and simplicity of
proofs. The \ensuremath{(\mathbin{\oplus})} and \ensuremath{(\mathbin{\ominus})} operations also take immutable
arguments. To be more efficient, we can use mutable states to
implement in-place updates or use the technique of functional but
in-place update~\citep{LorenzenLS23}. We leave them as future work.

Similar to \Cref{sec:local-state} and \Cref{sec:global-state}, the
local-state and global-state semantics of \ensuremath{\Varid{Modify_{F}}} and \ensuremath{\Varid{Nondet_{F}}} are
given by the following functions \ensuremath{\Varid{h_{LocalM}}} and \ensuremath{\Varid{h_{GlobalM}}}, respectively.
\indentbegin \begin{hscode}\SaveRestoreHook
\column{B}{@{}>{\hspre}l<{\hspost}@{}}%
\column{11}{@{}>{\hspre}l<{\hspost}@{}}%
\column{E}{@{}>{\hspre}l<{\hspost}@{}}%
\>[B]{}\Varid{h_{LocalM}}{}\<[11]%
\>[11]{}\mathbin{::}(\Conid{Functor}\;\Varid{f},\Conid{Undo}\;\Varid{s}\;\Varid{r}){}\<[E]%
\\
\>[11]{}\Rightarrow \Conid{Free}\;(\Varid{Modify_{F}}\;\Varid{s}\;\Varid{r}\mathrel{{:}{+}{:}}\Varid{Nondet_{F}}\mathrel{{:}{+}{:}}\Varid{f})\;\Varid{a}\to (\Varid{s}\to \Conid{Free}\;\Varid{f}\;[\mskip1.5mu \Varid{a}\mskip1.5mu]){}\<[E]%
\\
\>[B]{}\Varid{h_{LocalM}}{}\<[11]%
\>[11]{}\mathrel{=}\Varid{fmap}\;(\Varid{fmap}\;(\Varid{fmap}\;\Varid{fst})\hsdot{\circ }{.}\Varid{h_{ND+f}})\hsdot{\circ }{.}\Varid{run_{StateT}}\hsdot{\circ }{.}\Varid{h_{Modify}}{}\<[E]%
\\[\blanklineskip]%
\>[B]{}\Varid{h_{GlobalM}}{}\<[11]%
\>[11]{}\mathbin{::}(\Conid{Functor}\;\Varid{f},\Conid{Undo}\;\Varid{s}\;\Varid{r}){}\<[E]%
\\
\>[11]{}\Rightarrow \Conid{Free}\;(\Varid{Modify_{F}}\;\Varid{s}\;\Varid{r}\mathrel{{:}{+}{:}}\Varid{Nondet_{F}}\mathrel{{:}{+}{:}}\Varid{f})\;\Varid{a}\to (\Varid{s}\to \Conid{Free}\;\Varid{f}\;[\mskip1.5mu \Varid{a}\mskip1.5mu]){}\<[E]%
\\
\>[B]{}\Varid{h_{GlobalM}}{}\<[11]%
\>[11]{}\mathrel{=}\Varid{fmap}\;(\Varid{fmap}\;\Varid{fst})\hsdot{\circ }{.}\Varid{run_{StateT}}\hsdot{\circ }{.}\Varid{h_{Modify}}\hsdot{\circ }{.}\Varid{h_{ND+f}}\hsdot{\circ }{.}\Varid{(\Leftrightarrow)}{}\<[E]%
\ColumnHook
\end{hscode}\resethooks
\indentend 
\subsection{Simulating Local State with Global State and Undo}
\label{sec:local2globalM}

We can implement the translation from local-state semantics to
global-state semantics for the modification-based state effects
in a similar style to the translation \ensuremath{\Varid{local2global}} in
\Cref{sec:local2global}.
The translation \ensuremath{\Varid{local2global_M}} still uses the mechanism of
nondeterminism to restore previous state updates for backtracking.
In \Cref{sec:trail-stack} we will show a lower-level simulation of
local-state semantics without relying on nondeterminism.
\indentbegin \begin{hscode}\SaveRestoreHook
\column{B}{@{}>{\hspre}l<{\hspost}@{}}%
\column{3}{@{}>{\hspre}l<{\hspost}@{}}%
\column{5}{@{}>{\hspre}l<{\hspost}@{}}%
\column{16}{@{}>{\hspre}l<{\hspost}@{}}%
\column{25}{@{}>{\hspre}l<{\hspost}@{}}%
\column{E}{@{}>{\hspre}l<{\hspost}@{}}%
\>[B]{}\Varid{local2global_M}{}\<[16]%
\>[16]{}\mathbin{::}(\Conid{Functor}\;\Varid{f},\Conid{Undo}\;\Varid{s}\;\Varid{r}){}\<[E]%
\\
\>[16]{}\Rightarrow \Conid{Free}\;(\Varid{Modify_{F}}\;\Varid{s}\;\Varid{r}\mathrel{{:}{+}{:}}\Varid{Nondet_{F}}\mathrel{{:}{+}{:}}\Varid{f})\;\Varid{a}{}\<[E]%
\\
\>[16]{}\to \Conid{Free}\;(\Varid{Modify_{F}}\;\Varid{s}\;\Varid{r}\mathrel{{:}{+}{:}}\Varid{Nondet_{F}}\mathrel{{:}{+}{:}}\Varid{f})\;\Varid{a}{}\<[E]%
\\
\>[B]{}\Varid{local2global_M}{}\<[16]%
\>[16]{}\mathrel{=}\Varid{fold}\;\Conid{Var}\;\Varid{alg}{}\<[E]%
\\
\>[B]{}\hsindent{3}{}\<[3]%
\>[3]{}\mathbf{where}{}\<[E]%
\\
\>[3]{}\hsindent{2}{}\<[5]%
\>[5]{}\Varid{alg}\;(\Conid{Inl}\;(\Conid{MUpdate}\;\Varid{r}\;\Varid{k}))\mathrel{=}(\Varid{update}\;\Varid{r}\mathbin{\talloblong}\Varid{side}\;(\Varid{restore}\;\Varid{r}))>\!\!>\Varid{k}{}\<[E]%
\\
\>[3]{}\hsindent{2}{}\<[5]%
\>[5]{}\Varid{alg}\;\Varid{p}{}\<[25]%
\>[25]{}\mathrel{=}\Conid{Op}\;\Varid{p}{}\<[E]%
\ColumnHook
\end{hscode}\resethooks
\indentend 
Compared to \ensuremath{\Varid{local2global}}, the main difference is that we do not need
to copy and store the whole state. Instead, we store the delta \ensuremath{\Varid{r}} and
undo the state update using \ensuremath{\Varid{restore}\;\Varid{r}} in the second branch.
The following theorem shows the correctness of \ensuremath{\Varid{local2global_M}}.
\begin{restatable}[]{theorem}{modifyLocalGlobal}
\label{thm:modify-local-global}
Given \ensuremath{\Conid{Functor}\;\Varid{f}} and \ensuremath{\Conid{Undo}\;\Varid{s}\;\Varid{r}}, the equation\indentbegin \begin{hscode}\SaveRestoreHook
\column{B}{@{}>{\hspre}l<{\hspost}@{}}%
\column{3}{@{}>{\hspre}l<{\hspost}@{}}%
\column{E}{@{}>{\hspre}l<{\hspost}@{}}%
\>[3]{}\Varid{h_{GlobalM}}\hsdot{\circ }{.}\Varid{local2global_M}\mathrel{=}\Varid{h_{LocalM}}{}\<[E]%
\ColumnHook
\end{hscode}\resethooks
\indentend holds for all programs \ensuremath{\Varid{p}\mathbin{::}\Conid{Free}\;(\Varid{Modify_{F}}\;\Varid{s}\;\Varid{r}\mathrel{{:}{+}{:}}\Varid{Nondet_{F}}\mathrel{{:}{+}{:}}\Varid{f})\;\Varid{a}}
that do not use the operation \ensuremath{\Conid{Op}\;(\Conid{Inl}\;\Conid{MRestore}\;\anonymous \;\anonymous )}.
\end{restatable}
The proof of this theorem can be found in \Cref{app:modify-local-global}.

\paragraph*{N-queens with State Update and Restoration}\
We can rewrite the \ensuremath{\Varid{queens}} program with modification-based state and
nondeterminism. Compared to the \ensuremath{\Varid{queens}} in
\Cref{sec:motivation-and-challenges}, we only need to change the \ensuremath{\Varid{get}}
and \ensuremath{\Varid{put}\;(\Varid{s}\mathbin{\oplus}\Varid{r})} with the update operation \ensuremath{\Varid{update}\;\Varid{r}}.
\indentbegin \begin{hscode}\SaveRestoreHook
\column{B}{@{}>{\hspre}l<{\hspost}@{}}%
\column{3}{@{}>{\hspre}l<{\hspost}@{}}%
\column{14}{@{}>{\hspre}l<{\hspost}@{}}%
\column{23}{@{}>{\hspre}l<{\hspost}@{}}%
\column{E}{@{}>{\hspre}l<{\hspost}@{}}%
\>[B]{}\Varid{queens_{M}}\mathbin{::}(\Conid{MModify}\;(\Conid{Int},[\mskip1.5mu \Conid{Int}\mskip1.5mu])\;\Conid{Int}\;\Varid{m},\Conid{MNondet}\;\Varid{m})\Rightarrow \Conid{Int}\to \Varid{m}\;[\mskip1.5mu \Conid{Int}\mskip1.5mu]{}\<[E]%
\\
\>[B]{}\Varid{queens_{M}}\;\Varid{n}\mathrel{=}\Varid{loop}\;\mathbf{where}{}\<[E]%
\\
\>[B]{}\hsindent{3}{}\<[3]%
\>[3]{}\Varid{loop}\mathrel{=}\mathbf{do}\;{}\<[14]%
\>[14]{}(\Varid{c},\Varid{sol})\leftarrow \Varid{mget}{}\<[E]%
\\
\>[14]{}\mathbf{if}\;\Varid{c}\geq \Varid{n}\;\mathbf{then}\;\Varid{\eta}\;\Varid{sol}{}\<[E]%
\\
\>[14]{}\mathbf{else}\;\mathbf{do}\;{}\<[23]%
\>[23]{}\Varid{r}\leftarrow \Varid{choose}\;[\mskip1.5mu \mathrm{1}\mathinner{\ldotp\ldotp}\Varid{n}\mskip1.5mu]{}\<[E]%
\\
\>[23]{}\Varid{guard}\;(\Varid{safe}\;\Varid{r}\;\mathrm{1}\;\Varid{sol}){}\<[E]%
\\
\>[23]{}\Varid{update}\;\Varid{r}{}\<[E]%
\\
\>[23]{}\Varid{loop}{}\<[E]%
\ColumnHook
\end{hscode}\resethooks
\indentend 
We can interpret it using either \ensuremath{\Varid{h_{LocalM}}} or \ensuremath{\Varid{h_{GlobalM}}} composed with
\ensuremath{\Varid{local2global_M}}.
\indentbegin \begin{hscode}\SaveRestoreHook
\column{B}{@{}>{\hspre}l<{\hspost}@{}}%
\column{E}{@{}>{\hspre}l<{\hspost}@{}}%
\>[B]{}\Varid{queens_{LocalM}}\mathbin{::}\Conid{Int}\to [\mskip1.5mu [\mskip1.5mu \Conid{Int}\mskip1.5mu]\mskip1.5mu]{}\<[E]%
\\
\>[B]{}\Varid{queens_{LocalM}}\mathrel{=}\Varid{h_{Nil}}\hsdot{\circ }{.}\Varid{flip}\;\Varid{h_{LocalM}}\;(\mathrm{0},[\mskip1.5mu \mskip1.5mu])\hsdot{\circ }{.}\Varid{queens_{M}}{}\<[E]%
\\[\blanklineskip]%
\>[B]{}\Varid{queens_{GlobalM}}\mathbin{::}\Conid{Int}\to [\mskip1.5mu [\mskip1.5mu \Conid{Int}\mskip1.5mu]\mskip1.5mu]{}\<[E]%
\\
\>[B]{}\Varid{queens_{GlobalM}}\mathrel{=}\Varid{h_{Nil}}\hsdot{\circ }{.}\Varid{flip}\;\Varid{h_{GlobalM}}\;(\mathrm{0},[\mskip1.5mu \mskip1.5mu])\hsdot{\circ }{.}\Varid{local2global_M}\hsdot{\circ }{.}\Varid{queens_{M}}{}\<[E]%
\ColumnHook
\end{hscode}\resethooks
\indentend

\section{Modelling Local State with Trail Stack}
\label{sec:trail-stack}

In order
to trigger the restoration of the previous state, the simulations \ensuremath{\Varid{local2global}} in \Cref{sec:local2global} and \ensuremath{\Varid{local2global_M}}
in \Cref{sec:undo} introduce a call to \ensuremath{\Varid{or}} and to \ensuremath{\Varid{fail}}
at every modification of the state.
The Warren Abstract Machine (WAM)~\citep{AitKaci91} does this in a more efficient and
lower-level way: it uses a \emph{trail stack} to batch consecutive restorative steps.
In this section, we first make use of the idea of trail stacks to
implement a lower-level translation from local-state semantics to
global-state semantics for the modifcation-based version of state
effects in \Cref{sec:undo} which does not require extra calls to
nondeterminism operations.
Then, we combine this simulation with other simulations
to obtain another ultimate simulation function which uses two stacks,
a choicepoint stack and a trail stack, simultaneously.

\subsection{Simulating Local State with Global State and Trail Stack}
\label{sec:local2trail}

Let us consider modification-based version of simulation \ensuremath{\Varid{local2global_M}}.
We can use a trail stack to contain elements of type \ensuremath{\Conid{Either}\;\Varid{r}\;()}, where \ensuremath{\Varid{r}} is the type of deltas to the states. Each \ensuremath{\Conid{Left}\;\Varid{x}} entry represents
an update to the state with the delta \ensuremath{\Varid{x}}, and each \ensuremath{\Conid{Right}\;()} is a
marker.
When we enter a left branch, we push a marker on the trail stack.
For every update we perform in that branch, we push the corresponding
delta on the trail stack, on top of the marker.
When we backtrack to the right branch, we unwind the trail stack down to the
marker and restore all deltas along the way. This process is known as ``untrailing''.

We can easily model the \ensuremath{\Conid{Stack}} data type with Haskell lists.
\indentbegin \begin{hscode}\SaveRestoreHook
\column{B}{@{}>{\hspre}l<{\hspost}@{}}%
\column{E}{@{}>{\hspre}l<{\hspost}@{}}%
\>[B]{}\mathbf{newtype}\;\Conid{Stack}\;\Varid{s}\mathrel{=}\Conid{Stack}\;[\mskip1.5mu \Varid{s}\mskip1.5mu]{}\<[E]%
\ColumnHook
\end{hscode}\resethooks
\indentend 
We thread the stack through the computation using the state effect and define primitive pop and push operations as follows.

\begin{minipage}[t]{0.5\textwidth}
\indentbegin \begin{hscode}\SaveRestoreHook
\column{B}{@{}>{\hspre}l<{\hspost}@{}}%
\column{3}{@{}>{\hspre}l<{\hspost}@{}}%
\column{5}{@{}>{\hspre}l<{\hspost}@{}}%
\column{11}{@{}>{\hspre}l<{\hspost}@{}}%
\column{14}{@{}>{\hspre}l<{\hspost}@{}}%
\column{21}{@{}>{\hspre}l<{\hspost}@{}}%
\column{E}{@{}>{\hspre}l<{\hspost}@{}}%
\>[B]{}\Varid{popStack}{}\<[11]%
\>[11]{}\mathbin{::}\Conid{MState}\;(\Conid{Stack}\;\Varid{s})\;\Varid{m}{}\<[E]%
\\
\>[11]{}\Rightarrow \Varid{m}\;(\Conid{Maybe}\;\Varid{s}){}\<[E]%
\\
\>[B]{}\Varid{popStack}{}\<[11]%
\>[11]{}\mathrel{=}\mathbf{do}{}\<[E]%
\\
\>[B]{}\hsindent{3}{}\<[3]%
\>[3]{}\Conid{Stack}\;\Varid{xs}\leftarrow \Varid{get}{}\<[E]%
\\
\>[B]{}\hsindent{3}{}\<[3]%
\>[3]{}\mathbf{case}\;\Varid{xs}\;\mathbf{of}{}\<[E]%
\\
\>[3]{}\hsindent{2}{}\<[5]%
\>[5]{}[\mskip1.5mu \mskip1.5mu]{}\<[14]%
\>[14]{}\to \Varid{\eta}\;\Conid{Nothing}{}\<[E]%
\\
\>[3]{}\hsindent{2}{}\<[5]%
\>[5]{}(\Varid{x}\mathbin{:}\Varid{xs'}){}\<[14]%
\>[14]{}\to \mathbf{do}\;{}\<[21]%
\>[21]{}\Varid{put}\;(\Conid{Stack}\;\Varid{xs'});\Varid{\eta}\;(\Conid{Just}\;\Varid{x}){}\<[E]%
\ColumnHook
\end{hscode}\resethooks
\indentend \end{minipage}
\begin{minipage}[t]{0.5\textwidth}
\indentbegin \begin{hscode}\SaveRestoreHook
\column{B}{@{}>{\hspre}l<{\hspost}@{}}%
\column{3}{@{}>{\hspre}l<{\hspost}@{}}%
\column{12}{@{}>{\hspre}l<{\hspost}@{}}%
\column{E}{@{}>{\hspre}l<{\hspost}@{}}%
\>[B]{}\Varid{pushStack}{}\<[12]%
\>[12]{}\mathbin{::}\Conid{MState}\;(\Conid{Stack}\;\Varid{s})\;\Varid{m}{}\<[E]%
\\
\>[12]{}\Rightarrow \Varid{s}\to \Varid{m}\;(){}\<[E]%
\\
\>[B]{}\Varid{pushStack}\;\Varid{x}\mathrel{=}\mathbf{do}{}\<[E]%
\\
\>[B]{}\hsindent{3}{}\<[3]%
\>[3]{}\Conid{Stack}\;\Varid{xs}\leftarrow \Varid{get}{}\<[E]%
\\
\>[B]{}\hsindent{3}{}\<[3]%
\>[3]{}\Varid{put}\;(\Conid{Stack}\;(\Varid{x}\mathbin{:}\Varid{xs})){}\<[E]%
\ColumnHook
\end{hscode}\resethooks
\indentend \end{minipage}

In order to correctly use state operations to interact with the trail
stack in the translation, we also need to define a new instance of
\ensuremath{\Conid{MState}}.
\indentbegin \begin{hscode}\SaveRestoreHook
\column{B}{@{}>{\hspre}l<{\hspost}@{}}%
\column{3}{@{}>{\hspre}l<{\hspost}@{}}%
\column{5}{@{}>{\hspre}l<{\hspost}@{}}%
\column{14}{@{}>{\hspre}l<{\hspost}@{}}%
\column{E}{@{}>{\hspre}l<{\hspost}@{}}%
\>[B]{}\mathbf{instance}\;(\Conid{Functor}\;\Varid{f},\Conid{Functor}\;\Varid{g},\Conid{Functor}\;\Varid{h}){}\<[E]%
\\
\>[B]{}\hsindent{3}{}\<[3]%
\>[3]{}\Rightarrow \Conid{MState}\;\Varid{s}\;(\Conid{Free}\;(\Varid{f}\mathrel{{:}{+}{:}}\Varid{g}\mathrel{{:}{+}{:}}\Varid{State_{F}}\;\Varid{s}\mathrel{{:}{+}{:}}\Varid{h}))\;\mathbf{where}{}\<[E]%
\\
\>[3]{}\hsindent{2}{}\<[5]%
\>[5]{}\Varid{get}{}\<[14]%
\>[14]{}\mathrel{=}\Conid{Op}\hsdot{\circ }{.}\Conid{Inr}\hsdot{\circ }{.}\Conid{Inr}\hsdot{\circ }{.}\Conid{Inl}\mathbin{\$}\Conid{Get}\;\Varid{\eta}{}\<[E]%
\\
\>[3]{}\hsindent{2}{}\<[5]%
\>[5]{}\Varid{put}\;\Varid{x}{}\<[14]%
\>[14]{}\mathrel{=}\Conid{Op}\hsdot{\circ }{.}\Conid{Inr}\hsdot{\circ }{.}\Conid{Inr}\hsdot{\circ }{.}\Conid{Inl}\mathbin{\$}\Conid{Put}\;\Varid{x}\;(\Varid{\eta}\;()){}\<[E]%
\ColumnHook
\end{hscode}\resethooks
\indentend 

With these in place, the following translation function \ensuremath{\Varid{local2trail}} simulates the
local-state semantics with global-state semantics by means of the trail stack.

\indentbegin \begin{hscode}\SaveRestoreHook
\column{B}{@{}>{\hspre}l<{\hspost}@{}}%
\column{3}{@{}>{\hspre}l<{\hspost}@{}}%
\column{5}{@{}>{\hspre}l<{\hspost}@{}}%
\column{13}{@{}>{\hspre}l<{\hspost}@{}}%
\column{21}{@{}>{\hspre}l<{\hspost}@{}}%
\column{23}{@{}>{\hspre}l<{\hspost}@{}}%
\column{25}{@{}>{\hspre}l<{\hspost}@{}}%
\column{40}{@{}>{\hspre}l<{\hspost}@{}}%
\column{E}{@{}>{\hspre}l<{\hspost}@{}}%
\>[B]{}\Varid{local2trail}\mathbin{::}(\Conid{Functor}\;\Varid{f},\Conid{Undo}\;\Varid{s}\;\Varid{r}){}\<[E]%
\\
\>[B]{}\hsindent{13}{}\<[13]%
\>[13]{}\Rightarrow \Conid{Free}\;(\Varid{Modify_{F}}\;\Varid{s}\;\Varid{r}\mathrel{{:}{+}{:}}\Varid{Nondet_{F}}\mathrel{{:}{+}{:}}\Varid{f})\;\Varid{a}{}\<[E]%
\\
\>[B]{}\hsindent{13}{}\<[13]%
\>[13]{}\to \Conid{Free}\;(\Varid{Modify_{F}}\;\Varid{s}\;\Varid{r}\mathrel{{:}{+}{:}}\Varid{Nondet_{F}}\mathrel{{:}{+}{:}}\Varid{State_{F}}\;(\Conid{Stack}\;(\Conid{Either}\;\Varid{r}\;()))\mathrel{{:}{+}{:}}\Varid{f})\;\Varid{a}{}\<[E]%
\\
\>[B]{}\Varid{local2trail}\mathrel{=}\Varid{fold}\;\Conid{Var}\;(\Varid{alg}_{1}\mathbin{\#}\Varid{alg}_{2}\mathbin{\#}\Varid{fwd}){}\<[E]%
\\
\>[B]{}\hsindent{3}{}\<[3]%
\>[3]{}\mathbf{where}{}\<[E]%
\\
\>[3]{}\hsindent{2}{}\<[5]%
\>[5]{}\Varid{alg}_{1}\;(\Conid{MUpdate}\;\Varid{r}\;\Varid{k}){}\<[25]%
\>[25]{}\mathrel{=}\Varid{pushStack}\;(\Conid{Left}\;\Varid{r})>\!\!>\Varid{update}\;\Varid{r}>\!\!>\Varid{k}{}\<[E]%
\\
\>[3]{}\hsindent{2}{}\<[5]%
\>[5]{}\Varid{alg}_{1}\;\Varid{p}{}\<[25]%
\>[25]{}\mathrel{=}\Conid{Op}\hsdot{\circ }{.}\Conid{Inl}\mathbin{\$}\Varid{p}{}\<[E]%
\\
\>[3]{}\hsindent{2}{}\<[5]%
\>[5]{}\Varid{alg}_{2}\;(\Conid{Or}\;\Varid{p}\;\Varid{q}){}\<[25]%
\>[25]{}\mathrel{=}(\Varid{pushStack}\;(\Conid{Right}\;())>\!\!>\Varid{p})\mathbin{\talloblong}(\Varid{untrail}>\!\!>\Varid{q}){}\<[E]%
\\
\>[3]{}\hsindent{2}{}\<[5]%
\>[5]{}\Varid{alg}_{2}\;\Varid{p}{}\<[25]%
\>[25]{}\mathrel{=}\Conid{Op}\hsdot{\circ }{.}\Conid{Inr}\hsdot{\circ }{.}\Conid{Inl}\mathbin{\$}\Varid{p}{}\<[E]%
\\
\>[3]{}\hsindent{2}{}\<[5]%
\>[5]{}\Varid{fwd}\;\Varid{p}{}\<[25]%
\>[25]{}\mathrel{=}\Conid{Op}\hsdot{\circ }{.}\Conid{Inr}\hsdot{\circ }{.}\Conid{Inr}\hsdot{\circ }{.}\Conid{Inr}\mathbin{\$}\Varid{p}{}\<[E]%
\\
\>[3]{}\hsindent{2}{}\<[5]%
\>[5]{}\Varid{untrail}\mathrel{=}\mathbf{do}\;{}\<[21]%
\>[21]{}\Varid{top}\leftarrow \Varid{popStack}{}\<[E]%
\\
\>[21]{}\mathbf{case}\;\Varid{top}\;\mathbf{of}{}\<[E]%
\\
\>[21]{}\hsindent{2}{}\<[23]%
\>[23]{}\Conid{Nothing}{}\<[40]%
\>[40]{}\to \Varid{\eta}\;(){}\<[E]%
\\
\>[21]{}\hsindent{2}{}\<[23]%
\>[23]{}\Conid{Just}\;(\Conid{Right}\;()){}\<[40]%
\>[40]{}\to \Varid{\eta}\;(){}\<[E]%
\\
\>[21]{}\hsindent{2}{}\<[23]%
\>[23]{}\Conid{Just}\;(\Conid{Left}\;\Varid{r}){}\<[40]%
\>[40]{}\to \Varid{restore}\;\Varid{r}>\!\!>\Varid{untrail}{}\<[E]%
\ColumnHook
\end{hscode}\resethooks
\indentend As already informally explained above, this translation function
introduces code to push a marker in the left branch of a choice, and 
untrail in the right branch. Whenever an update happens, it is also
recorded on the trail stack. All other operations remain as is.

Now, we can combine the simulation \ensuremath{\Varid{local2trail}} with the global-state
semantics provided by \ensuremath{\Varid{h_{GlobalM}}}, and handle the trail stack at the
end.
\indentbegin \begin{hscode}\SaveRestoreHook
\column{B}{@{}>{\hspre}l<{\hspost}@{}}%
\column{11}{@{}>{\hspre}l<{\hspost}@{}}%
\column{E}{@{}>{\hspre}l<{\hspost}@{}}%
\>[B]{}\Varid{h_{GlobalT}}{}\<[11]%
\>[11]{}\mathbin{::}(\Conid{Functor}\;\Varid{f},\Conid{Undo}\;\Varid{s}\;\Varid{r}){}\<[E]%
\\
\>[11]{}\Rightarrow \Conid{Free}\;(\Varid{Modify_{F}}\;\Varid{s}\;\Varid{r}\mathrel{{:}{+}{:}}\Varid{Nondet_{F}}\mathrel{{:}{+}{:}}\Varid{f})\;\Varid{a}\to \Varid{s}\to \Conid{Free}\;\Varid{f}\;[\mskip1.5mu \Varid{a}\mskip1.5mu]{}\<[E]%
\\
\>[B]{}\Varid{h_{GlobalT}}{}\<[11]%
\>[11]{}\mathrel{=}\Varid{fmap}\;(\Varid{fmap}\;\Varid{fst}\hsdot{\circ }{.}\Varid{flip}\;\Varid{run_{StateT}}\;(\Conid{Stack}\;[\mskip1.5mu \mskip1.5mu])\hsdot{\circ }{.}\Varid{h_{State}})\hsdot{\circ }{.}\Varid{h_{GlobalM}}\hsdot{\circ }{.}\Varid{local2trail}{}\<[E]%
\ColumnHook
\end{hscode}\resethooks
\indentend 
The following theorem establishes the correctness of \ensuremath{\Varid{h_{GlobalT}}} with respect to
the local-state semantics given by \ensuremath{\Varid{h_{Local}}} defined in \Cref{sec:local-state}.
\begin{restatable}[]{theorem}{localTrail}
\label{thm:trail-local-global}
Given \ensuremath{\Conid{Functor}\;\Varid{f}} and \ensuremath{\Conid{Undo}\;\Varid{s}\;\Varid{r}}, the equation\indentbegin \begin{hscode}\SaveRestoreHook
\column{B}{@{}>{\hspre}l<{\hspost}@{}}%
\column{3}{@{}>{\hspre}l<{\hspost}@{}}%
\column{E}{@{}>{\hspre}l<{\hspost}@{}}%
\>[3]{}\Varid{h_{GlobalT}}\mathrel{=}\Varid{h_{LocalM}}{}\<[E]%
\ColumnHook
\end{hscode}\resethooks
\indentend holds for all programs \ensuremath{\Varid{p}\mathbin{::}\Conid{Free}\;(\Varid{Modify_{F}}\;\Varid{s}\;\Varid{r}\mathrel{{:}{+}{:}}\Varid{Nondet_{F}}\mathrel{{:}{+}{:}}\Varid{f})\;\Varid{a}}
that do not use the operation \ensuremath{\Conid{Op}\;(\Conid{Inl}\;\Conid{MRestore}\;\anonymous \;\anonymous )}.
\end{restatable}
The proof can be found in Appendix~\ref{app:immutable-trail-stack};
it uses the same fold fusion strategy as in the proofs of other theorems.

\subsection{Putting Everything Together, Again}

We can further combine the simulation \ensuremath{\Varid{local2trail}} with the
simulation \ensuremath{\Varid{nondet2state}} of nondeterminism in
\Cref{sec:nondeterminism-state} and the simulation \ensuremath{\Varid{state2state}} of
multiple states in \Cref{sec:multiple-states}.
The result simulation encodes the local-state semantics with one
modification-based state and two stacks, a choicepoint stack generated
by \ensuremath{\Varid{nondet2state}} and a trail stack generated by \ensuremath{\Varid{local2trail}}.
This has a close relationship to the WAM of Prolog.
The modifcation-based state models the state of the program.
The choicepoint stack stores the remaining branches to implement the
nondeterministic searching.
The trail stack stores the privous state updates to implement the
backtracking.

The combined simulation function \ensuremath{\Varid{simulate_T}} is defined
as follows:
\indentbegin \begin{hscode}\SaveRestoreHook
\column{B}{@{}>{\hspre}l<{\hspost}@{}}%
\column{16}{@{}>{\hspre}l<{\hspost}@{}}%
\column{E}{@{}>{\hspre}l<{\hspost}@{}}%
\>[B]{}\Varid{simulate_T}{}\<[16]%
\>[16]{}\mathbin{::}(\Conid{Functor}\;\Varid{f},\Conid{Undo}\;\Varid{s}\;\Varid{r}){}\<[E]%
\\
\>[16]{}\Rightarrow \Conid{Free}\;(\Varid{Modify_{F}}\;\Varid{s}\;\Varid{r}\mathrel{{:}{+}{:}}\Varid{Nondet_{F}}\mathrel{{:}{+}{:}}\Varid{f})\;\Varid{a}{}\<[E]%
\\
\>[16]{}\to \Varid{s}\to \Conid{Free}\;\Varid{f}\;[\mskip1.5mu \Varid{a}\mskip1.5mu]{}\<[E]%
\\
\>[B]{}\Varid{simulate_T}\;\Varid{x}\;\Varid{s}{}\<[16]%
\>[16]{}\mathrel{=}\Varid{extractT}\hsdot{\circ }{.}\Varid{h_{State}}{}\<[E]%
\\
\>[16]{}\hsdot{\circ }{.}\Varid{fmap}\;\Varid{fst}\hsdot{\circ }{.}\Varid{flip}\;\Varid{run_{StateT}}\;\Varid{s}\hsdot{\circ }{.}\Varid{h_{Modify}}{}\<[E]%
\\
\>[16]{}\hsdot{\circ }{.}\Varid{(\Leftrightarrow)}\hsdot{\circ }{.}\Varid{states2state}\hsdot{\circ }{.}\Varid{(\circlearrowleft)}{}\<[E]%
\\
\>[16]{}\hsdot{\circ }{.}\Varid{(\Leftrightarrow)}\hsdot{\circ }{.}\Varid{nondet2state}\hsdot{\circ }{.}\Varid{(\Leftrightarrow)}{}\<[E]%
\\
\>[16]{}\hsdot{\circ }{.}\Varid{local2trail}\mathbin{\$}\Varid{x}{}\<[E]%
\ColumnHook
\end{hscode}\resethooks
\indentend 
It uses the auxiliary function \ensuremath{\Varid{extractT}} to get the final results,
and \ensuremath{\Varid{(\circlearrowleft)}} to reorder the signatures.
Note that the initial state used by \ensuremath{\Varid{extractT}} is \ensuremath{(\Conid{SS}\;[\mskip1.5mu \mskip1.5mu]\;[\mskip1.5mu \mskip1.5mu],\Conid{Stack}\;[\mskip1.5mu \mskip1.5mu])}, which essentially contains an empty results list, an empty
choicepoint stack, and an empty trail stack.
\indentbegin \begin{hscode}\SaveRestoreHook
\column{B}{@{}>{\hspre}l<{\hspost}@{}}%
\column{10}{@{}>{\hspre}l<{\hspost}@{}}%
\column{35}{@{}>{\hspre}l<{\hspost}@{}}%
\column{61}{@{}>{\hspre}l<{\hspost}@{}}%
\column{E}{@{}>{\hspre}l<{\hspost}@{}}%
\>[B]{}\Varid{extractT}\;\Varid{x}\mathrel{=}\Varid{results_{SS}}\hsdot{\circ }{.}\Varid{fst}\hsdot{\circ }{.}\Varid{snd}\mathbin{\langle\hspace{1.6pt}\mathclap{\raisebox{0.1pt}{\scalebox{1}{\$}}}\hspace{1.6pt}\rangle}\Varid{run_{StateT}}\;\Varid{x}\;(\Conid{SS}\;[\mskip1.5mu \mskip1.5mu]\;[\mskip1.5mu \mskip1.5mu],\Conid{Stack}\;[\mskip1.5mu \mskip1.5mu]){}\<[E]%
\\[\blanklineskip]%
\>[B]{}\Varid{(\circlearrowleft)}{}\<[10]%
\>[10]{}\mathbin{::}(\Conid{Functor}\;\Varid{f}_{1},\Conid{Functor}\;\Varid{f}_{2},\Conid{Functor}\;\Varid{f3},\Conid{Functor}\;\Varid{f4}){}\<[E]%
\\
\>[10]{}\Rightarrow \Conid{Free}\;(\Varid{f}_{1}\mathrel{{:}{+}{:}}\Varid{f}_{2}\mathrel{{:}{+}{:}}\Varid{f3}\mathrel{{:}{+}{:}}\Varid{f4})\;\Varid{a}\to \Conid{Free}\;(\Varid{f}_{2}\mathrel{{:}{+}{:}}\Varid{f3}\mathrel{{:}{+}{:}}\Varid{f}_{1}\mathrel{{:}{+}{:}}\Varid{f4})\;\Varid{a}{}\<[E]%
\\
\>[B]{}\Varid{(\circlearrowleft)}\;(\Conid{Var}\;\Varid{x}){}\<[35]%
\>[35]{}\mathrel{=}\Conid{Var}\;\Varid{x}{}\<[E]%
\\
\>[B]{}\Varid{(\circlearrowleft)}\;(\Conid{Op}\;(\Conid{Inl}\;\Varid{k})){}\<[35]%
\>[35]{}\mathrel{=}(\Conid{Op}\hsdot{\circ }{.}\Conid{Inr}\hsdot{\circ }{.}\Conid{Inr}\hsdot{\circ }{.}\Conid{Inl})\;{}\<[61]%
\>[61]{}(\Varid{fmap}\;\Varid{(\circlearrowleft)}\;\Varid{k}){}\<[E]%
\\
\>[B]{}\Varid{(\circlearrowleft)}\;(\Conid{Op}\;(\Conid{Inr}\;(\Conid{Inl}\;\Varid{k}))){}\<[35]%
\>[35]{}\mathrel{=}(\Conid{Op}\hsdot{\circ }{.}\Conid{Inl})\;{}\<[61]%
\>[61]{}(\Varid{fmap}\;\Varid{(\circlearrowleft)}\;\Varid{k}){}\<[E]%
\\
\>[B]{}\Varid{(\circlearrowleft)}\;(\Conid{Op}\;(\Conid{Inr}\;(\Conid{Inr}\;(\Conid{Inl}\;\Varid{k})))){}\<[35]%
\>[35]{}\mathrel{=}(\Conid{Op}\hsdot{\circ }{.}\Conid{Inr}\hsdot{\circ }{.}\Conid{Inl})\;{}\<[61]%
\>[61]{}(\Varid{fmap}\;\Varid{(\circlearrowleft)}\;\Varid{k}){}\<[E]%
\\
\>[B]{}\Varid{(\circlearrowleft)}\;(\Conid{Op}\;(\Conid{Inr}\;(\Conid{Inr}\;(\Conid{Inr}\;\Varid{k})))){}\<[35]%
\>[35]{}\mathrel{=}(\Conid{Op}\hsdot{\circ }{.}\Conid{Inr}\hsdot{\circ }{.}\Conid{Inr}\hsdot{\circ }{.}\Conid{Inr})\;{}\<[61]%
\>[61]{}(\Varid{fmap}\;\Varid{(\circlearrowleft)}\;\Varid{k}){}\<[E]%
\ColumnHook
\end{hscode}\resethooks
\indentend 
\Cref{fig:simulation-trail} illustrates each step of this simulation.
The state type \ensuremath{\Conid{St}\;\Varid{s}\;\Varid{r}\;\Varid{f}\;\Varid{a}} is defined as \ensuremath{\Conid{SS}\;(\Varid{Modify_{F}}\;\Varid{s}\;\Varid{r}\mathrel{{:}{+}{:}}\Varid{State_{F}}\;(\Conid{Stack}\;(\Conid{Either}\;\Varid{r}\;()))\mathrel{{:}{+}{:}}\Varid{f})\;\Varid{a}}.
\begin{figure}[h]
\[\begin{tikzcd}
	{\ensuremath{\Conid{Free}\;(\Varid{Modify_{F}}\;\Varid{s}\;\Varid{r}\mathrel{{:}{+}{:}}\Varid{Nondet_{F}}\mathrel{{:}{+}{:}}\Varid{f})\;\Varid{a}}} \\
	{\ensuremath{\Conid{Free}\;(\Varid{Modify_{F}}\;\Varid{s}\;\Varid{r}\mathrel{{:}{+}{:}}\Varid{Nondet_{F}}\mathrel{{:}{+}{:}}\Varid{State_{F}}\;(\Conid{Stack}\;(\Conid{Either}\;\Varid{r}\;()))\mathrel{{:}{+}{:}}\Varid{f})\;\Varid{a}}} \\
	{\ensuremath{\Conid{Free}\;(\Varid{Modify_{F}}\;\Varid{s}\;\Varid{r}\mathrel{{:}{+}{:}}\Varid{State_{F}}\;(\Conid{St}\;\Varid{s}\;\Varid{r}\;\Varid{f}\;\Varid{a})\mathrel{{:}{+}{:}}\Varid{State_{F}}\;(\Conid{Stack}\;(\Conid{Either}\;\Varid{r}\;()))\mathrel{{:}{+}{:}}\Varid{f})\;()}} \\
	{\ensuremath{\Conid{Free}\;(\Varid{Modify_{F}}\;\Varid{s}\;\Varid{r}\mathrel{{:}{+}{:}}\Varid{State_{F}}\;(\Conid{St}\;\Varid{s}\;\Varid{r}\;\Varid{f}\;\Varid{a},\Conid{Stack}\;(\Conid{Either}\;\Varid{r}\;()))\mathrel{{:}{+}{:}}\Varid{f})\;()}} \\
	{\ensuremath{\Conid{Free}\;(\Varid{State_{F}}\;(\Conid{St}\;\Varid{s}\;\Varid{r}\;\Varid{f}\;\Varid{a},\Conid{Stack}\;(\Conid{Either}\;\Varid{r}\;()))\mathrel{{:}{+}{:}}\Varid{f})\;()}} \\
	{\ensuremath{\Conid{Free}\;\Varid{f}\;[\mskip1.5mu \Varid{a}\mskip1.5mu]}}
	\arrow["{\ensuremath{\Varid{local2trail}}}", from=1-1, to=2-1]
	\arrow["{\ensuremath{\Varid{(\Leftrightarrow)}\hsdot{\circ }{.}\Varid{nondet2state}\hsdot{\circ }{.}\Varid{(\Leftrightarrow)}}}", from=2-1, to=3-1]
	\arrow["{\ensuremath{\Varid{fmap}\;\Varid{fst}\hsdot{\circ }{.}\Varid{flip}\;\Varid{run_{StateT}}\;\Varid{s}\hsdot{\circ }{.}\Varid{h_{Modify}}}}", from=4-1, to=5-1]
	\arrow["{\ensuremath{\Varid{(\Leftrightarrow)}\hsdot{\circ }{.}\Varid{states2state}\hsdot{\circ }{.}\Varid{(\circlearrowleft)}}}", from=3-1, to=4-1]
	\arrow["{\ensuremath{\Varid{extractT}\hsdot{\circ }{.}\Varid{h_{State}}}}", from=5-1, to=6-1]
\end{tikzcd}\]
\caption{An overview of the \ensuremath{\Varid{simulate_T}} function.}
\label{fig:simulation-trail}
\end{figure}

In the \ensuremath{\Varid{simulate_T}} function, we first use our three simulations
\ensuremath{\Varid{local2trail}}, \ensuremath{\Varid{nondet2state}} and \ensuremath{\Varid{states2state}} (together with some
reordering of signatures) to interpret the local-state semantics for
state and nondeterminism in terms of a modification-based state and a
general state containing two stacks. Then, we use the handler
\ensuremath{\Varid{h_{Modify}}} to interpret the modification-based state effect, and use
the handler \ensuremath{\Varid{h_{State}}} to interpret the two stacks. Finally, we use the
function \ensuremath{\Varid{extractT}} to get the final results.

As in \Cref{sec:final-simulate}, we can also fuse \ensuremath{\Varid{simulate_T}} into a
single handler.
\indentbegin \begin{hscode}\SaveRestoreHook
\column{B}{@{}>{\hspre}l<{\hspost}@{}}%
\column{3}{@{}>{\hspre}l<{\hspost}@{}}%
\column{5}{@{}>{\hspre}l<{\hspost}@{}}%
\column{13}{@{}>{\hspre}l<{\hspost}@{}}%
\column{20}{@{}>{\hspre}l<{\hspost}@{}}%
\column{25}{@{}>{\hspre}l<{\hspost}@{}}%
\column{26}{@{}>{\hspre}l<{\hspost}@{}}%
\column{36}{@{}>{\hspre}l<{\hspost}@{}}%
\column{44}{@{}>{\hspre}c<{\hspost}@{}}%
\column{44E}{@{}l@{}}%
\column{47}{@{}>{\hspre}l<{\hspost}@{}}%
\column{49}{@{}>{\hspre}l<{\hspost}@{}}%
\column{56}{@{}>{\hspre}l<{\hspost}@{}}%
\column{65}{@{}>{\hspre}l<{\hspost}@{}}%
\column{E}{@{}>{\hspre}l<{\hspost}@{}}%
\>[B]{}\mathbf{type}\;\Conid{Comp}\;\Varid{f}\;\Varid{a}\;\Varid{s}\;\Varid{r}\mathrel{=}(\Conid{WAM}\;\Varid{f}\;\Varid{a}\;\Varid{s}\;\Varid{r},\Varid{s})\to \Conid{Free}\;\Varid{f}\;[\mskip1.5mu \Varid{a}\mskip1.5mu]{}\<[E]%
\\
\>[B]{}\mathbf{data}\;\Conid{WAM}\;\Varid{f}\;\Varid{a}\;\Varid{s}\;\Varid{r}\mathrel{=}\Conid{WAM}\;{}\<[25]%
\>[25]{}\{\mskip1.5mu \Varid{results}{}\<[36]%
\>[36]{}\mathbin{::}[\mskip1.5mu \Varid{a}\mskip1.5mu]{}\<[E]%
\\
\>[25]{},\Varid{cpStack}{}\<[36]%
\>[36]{}\mathbin{::}[\mskip1.5mu \Conid{Comp}\;\Varid{f}\;\Varid{a}\;\Varid{s}\;\Varid{r}\mskip1.5mu]{}\<[E]%
\\
\>[25]{},\Varid{trStack}{}\<[36]%
\>[36]{}\mathbin{::}[\mskip1.5mu \Conid{Either}\;\Varid{r}\;()\mskip1.5mu]\mskip1.5mu\}{}\<[E]%
\\[\blanklineskip]%
\>[B]{}\Varid{simulate_{TF}}{}\<[13]%
\>[13]{}\mathbin{::}(\Conid{Functor}\;\Varid{f},\Conid{Undo}\;\Varid{s}\;\Varid{r}){}\<[E]%
\\
\>[13]{}\Rightarrow \Conid{Free}\;(\Varid{Modify_{F}}\;\Varid{s}\;\Varid{r}\mathrel{{:}{+}{:}}\Varid{Nondet_{F}}\mathrel{{:}{+}{:}}\Varid{f})\;\Varid{a}{}\<[E]%
\\
\>[13]{}\to \Varid{s}{}\<[E]%
\\
\>[13]{}\to \Conid{Free}\;\Varid{f}\;[\mskip1.5mu \Varid{a}\mskip1.5mu]{}\<[E]%
\\
\>[B]{}\Varid{simulate_{TF}}\;{}\<[13]%
\>[13]{}\Varid{x}\;\Varid{s}\mathrel{=}{}\<[20]%
\>[20]{}\Varid{fold}\;\Varid{gen}\;(\Varid{alg}_{1}\mathbin{\#}\Varid{alg}_{2}\mathbin{\#}\Varid{fwd})\;\Varid{x}\;(\Conid{WAM}\;[\mskip1.5mu \mskip1.5mu]\;[\mskip1.5mu \mskip1.5mu]\;[\mskip1.5mu \mskip1.5mu],\Varid{s}){}\<[E]%
\\
\>[B]{}\hsindent{3}{}\<[3]%
\>[3]{}\mathbf{where}{}\<[E]%
\\
\>[3]{}\hsindent{2}{}\<[5]%
\>[5]{}\Varid{gen}\;\Varid{x}\;{}\<[26]%
\>[26]{}(\Conid{WAM}\;\Varid{xs}\;\Varid{cp}\;\Varid{tr},\Varid{s})\mathrel{=}{}\<[47]%
\>[47]{}\Varid{continue}\;(\Varid{xs}+\!\!+[\mskip1.5mu \Varid{x}\mskip1.5mu])\;\Varid{cp}\;\Varid{tr}\;\Varid{s}{}\<[E]%
\\
\>[3]{}\hsindent{2}{}\<[5]%
\>[5]{}\Varid{alg}_{1}\;(\Conid{MGet}\;\Varid{k})\;{}\<[26]%
\>[26]{}(\Conid{WAM}\;\Varid{xs}\;\Varid{cp}\;\Varid{tr},\Varid{s})\mathrel{=}{}\<[47]%
\>[47]{}\Varid{k}\;\Varid{s}\;(\Conid{WAM}\;\Varid{xs}\;\Varid{cp}\;\Varid{tr},\Varid{s}){}\<[E]%
\\
\>[3]{}\hsindent{2}{}\<[5]%
\>[5]{}\Varid{alg}_{1}\;(\Conid{MUpdate}\;\Varid{r}\;\Varid{k})\;{}\<[26]%
\>[26]{}(\Conid{WAM}\;\Varid{xs}\;\Varid{cp}\;\Varid{tr},\Varid{s})\mathrel{=}{}\<[47]%
\>[47]{}\Varid{k}\;(\Conid{WAM}\;\Varid{xs}\;\Varid{cp}\;(\Conid{Left}\;\Varid{r}\mathbin{:}\Varid{tr}),\Varid{s}\mathbin{\oplus}\Varid{r}){}\<[E]%
\\
\>[3]{}\hsindent{2}{}\<[5]%
\>[5]{}\Varid{alg}_{1}\;(\Conid{MRestore}\;\Varid{r}\;\Varid{k})\;{}\<[26]%
\>[26]{}(\Conid{WAM}\;\Varid{xs}\;\Varid{cp}\;\Varid{tr},\Varid{s})\mathrel{=}{}\<[47]%
\>[47]{}\Varid{k}\;(\Conid{WAM}\;\Varid{xs}\;\Varid{cp}\;\Varid{tr},\Varid{s}\mathbin{\ominus}\Varid{r}){}\<[E]%
\\
\>[3]{}\hsindent{2}{}\<[5]%
\>[5]{}\Varid{alg}_{2}\;\Conid{Fail}\;{}\<[26]%
\>[26]{}(\Conid{WAM}\;\Varid{xs}\;\Varid{cp}\;\Varid{tr},\Varid{s})\mathrel{=}{}\<[47]%
\>[47]{}\Varid{continue}\;\Varid{xs}\;\Varid{cp}\;\Varid{tr}\;\Varid{s}{}\<[E]%
\\
\>[3]{}\hsindent{2}{}\<[5]%
\>[5]{}\Varid{alg}_{2}\;(\Conid{Or}\;\Varid{p}\;\Varid{q})\;{}\<[26]%
\>[26]{}(\Conid{WAM}\;\Varid{xs}\;\Varid{cp}\;\Varid{tr},\Varid{s})\mathrel{=}{}\<[47]%
\>[47]{}\Varid{p}\;(\Conid{WAM}\;\Varid{xs}\;(\Varid{untrail}\;\Varid{q}\mathbin{:}\Varid{cp})\;(\Conid{Right}\;()\mathbin{:}\Varid{tr}),\Varid{s}){}\<[E]%
\\
\>[3]{}\hsindent{2}{}\<[5]%
\>[5]{}\Varid{fwd}\;\Varid{op}\;{}\<[26]%
\>[26]{}(\Conid{WAM}\;\Varid{xs}\;\Varid{cp}\;\Varid{tr},\Varid{s})\mathrel{=}{}\<[47]%
\>[47]{}\Conid{Op}\;(\Varid{fmap}\;(\mathbin{\$}(\Conid{WAM}\;\Varid{xs}\;\Varid{cp}\;\Varid{tr},\Varid{s}))\;\Varid{op}){}\<[E]%
\\
\>[3]{}\hsindent{2}{}\<[5]%
\>[5]{}\Varid{untrail}\;\Varid{q}\;{}\<[26]%
\>[26]{}(\Conid{WAM}\;\Varid{xs}\;\Varid{cp}\;\Varid{tr},\Varid{s})\mathrel{=}{}\<[47]%
\>[47]{}\mathbf{case}\;\Varid{tr}\;\mathbf{of}{}\<[E]%
\\
\>[47]{}\hsindent{2}{}\<[49]%
\>[49]{}[\mskip1.5mu \mskip1.5mu]\to \Varid{q}\;(\Conid{WAM}\;\Varid{xs}\;\Varid{cp}\;\Varid{tr},\Varid{s}){}\<[E]%
\\
\>[47]{}\hsindent{2}{}\<[49]%
\>[49]{}\Conid{Right}\;()\mathbin{:}\Varid{tr'}{}\<[65]%
\>[65]{}\to \Varid{q}\;(\Conid{WAM}\;\Varid{xs}\;\Varid{cp}\;\Varid{tr'},\Varid{s}){}\<[E]%
\\
\>[47]{}\hsindent{2}{}\<[49]%
\>[49]{}\Conid{Left}\;\Varid{r}\mathbin{:}\Varid{tr'}{}\<[65]%
\>[65]{}\to \Varid{untrail}\;\Varid{q}\;(\Conid{WAM}\;\Varid{xs}\;\Varid{cp}\;\Varid{tr'},\Varid{s}\mathbin{\ominus}\Varid{r}){}\<[E]%
\\
\>[3]{}\hsindent{2}{}\<[5]%
\>[5]{}\Varid{continue}\;\Varid{xs}\;\Varid{cp}\;\Varid{tr}\;\Varid{s}{}\<[44]%
\>[44]{}\mathrel{=}{}\<[44E]%
\>[47]{}\mathbf{case}\;\Varid{cp}\;\mathbf{of}{}\<[E]%
\\
\>[47]{}\hsindent{2}{}\<[49]%
\>[49]{}[\mskip1.5mu \mskip1.5mu]{}\<[56]%
\>[56]{}\to \Varid{\eta}\;\Varid{xs}{}\<[E]%
\\
\>[47]{}\hsindent{2}{}\<[49]%
\>[49]{}\Varid{p}\mathbin{:}\Varid{cp'}{}\<[56]%
\>[56]{}\to \Varid{p}\;(\Conid{WAM}\;\Varid{xs}\;\Varid{cp'}\;\Varid{tr},\Varid{s}){}\<[E]%
\ColumnHook
\end{hscode}\resethooks
\indentend Here, the carrier type of the algebras is \ensuremath{\Conid{Comp}\;\Varid{f}\;\Varid{a}\;\Varid{s}\;\Varid{r}}. It differs from that of \ensuremath{\Varid{simulate_F}}
in that it also takes a trail stack as an input.

\paragraph*{N-queens with Two Stacks}\
With \ensuremath{\Varid{simulate_T}}, we can implement the backtracking algorithm of the
n-queens problem with one modification-based state and two stacks.

\indentbegin \begin{hscode}\SaveRestoreHook
\column{B}{@{}>{\hspre}l<{\hspost}@{}}%
\column{E}{@{}>{\hspre}l<{\hspost}@{}}%
\>[B]{}\Varid{queensSimT}\mathbin{::}\Conid{Int}\to [\mskip1.5mu [\mskip1.5mu \Conid{Int}\mskip1.5mu]\mskip1.5mu]{}\<[E]%
\\
\>[B]{}\Varid{queensSimT}\mathrel{=}\Varid{h_{Nil}}\hsdot{\circ }{.}\Varid{flip}\;\Varid{simulate_T}\;(\mathrm{0},[\mskip1.5mu \mskip1.5mu])\hsdot{\circ }{.}\Varid{queens_{M}}{}\<[E]%
\ColumnHook
\end{hscode}\resethooks
\indentend

\section{Related Work}
\label{sec:related-work}

There are various related works.

\subsection{Prolog}
\label{sec:prolog}

Prolog is a prominent example of a system that exposes nondeterminism with local
state to the user, but is itself implemented in terms of a single, global state.

\paragraph*{Warren Abstract Machine}\
The folklore idea of undoing modifications upon backtracking is a key feature
of many Prolog implementations, in particular those based on the Warren 
Abstract Machine (WAM) \cite{AICPub641:1983,AitKaci91}.
The WAM's global state is the program heap and Prolog programs modify this heap
during unification only in a very specific manner: following the union-find
algorithm, they overwrite cells that contain self-references with pointers to
other cells. 
Undoing these modifications only requires knowledge of the modified cell's
address, which can be written back in that cell during backtracking. 
The WAM has a special stack, called the trail stack, for storing these addresses, 
and the process of restoring those cells is called \emph{untrailing}.

\paragraph*{WAM Derivation and Correctness}\
Several authors have studied the derivation of the WAM from a specification
of Prolog, and its correctness.

\cite{DBLP:books/el/beierleP95/BorgerR95} start from an operational semantics
of Prolog based on derivation trees and refine this in successive steps to the WAM.
Their approach was later mechanized in Isabelle/HOL by \cite{10.5555/646523.694570}.
\cite{wam} sketch how the WAM can be derived from a Prolog interpreter
following the functional correspondence between evaluator and abstract
machine~\cite{AGER2005149}. 

Neither of these approaches is based on an abstraction of
effects that separates them from other aspects of Prolog.

\paragraph*{The 4-Port Box Model}\
While trailing happens under the hood, there is a folklore Prolog programming
pattern for observing and intervening at different point in the control flow of a
procedure call, known as the \emph{4-port box model}.
In this model, upon the first entrance of a Prolog procedure 
it is \emph{called};
it may yield a result and \emph{exits}; 
when the subsequent procedure fails and backtracks, it is asked to \emph{redo}
its computation, possibly yielding the next result;
finally it may fail. 
Given a Prolog procedure \ensuremath{\Varid{p}} implemented in Haskell, the following program prints
debugging messages when each of the four ports are used:\indentbegin \begin{hscode}\SaveRestoreHook
\column{B}{@{}>{\hspre}l<{\hspost}@{}}%
\column{3}{@{}>{\hspre}l<{\hspost}@{}}%
\column{E}{@{}>{\hspre}l<{\hspost}@{}}%
\>[3]{}(\Varid{putStr}\;\text{\ttfamily \char34 call\char34}\mathbin{\talloblong}\Varid{side}\;(\Varid{putStr}\;\text{\ttfamily \char34 fail\char34}))>\!\!>{}\<[E]%
\\
\>[3]{}\Varid{p}>\!\!>\!\!=\lambda \Varid{x}\to {}\<[E]%
\\
\>[3]{}(\Varid{putStr}\;\text{\ttfamily \char34 exit\char34}\mathbin{\talloblong}\Varid{side}\;(\Varid{putStr}\;\text{\ttfamily \char34 redo\char34}))>\!\!>{}\<[E]%
\\
\>[3]{}\Varid{\eta}\;\Varid{x}{}\<[E]%
\ColumnHook
\end{hscode}\resethooks
\indentend This technique was applied in the monadic setting by \citet{monadicbacktracking},
and it has been our inspiration for expressing the state restoration with global
state.

\paragraph*{Functional Models of Prolog}\
Various authors have modelled (aspects of) Prolog in functional programming
languages, often using monads to capture nondeterminism and state effects.
Notably, \cite{prologinhaskell} develop an embedding of Prolog in Haskell.

Most attention has gone towards modelling the nondeterminsm or search aspect of
Prolog, with various monads and monad transformers being proposed
\citep{DBLP:conf/icfp/Hinze00,DBLP:conf/icfp/KiselyovSFS05}. Notably, \cite{DBLP:conf/ppdp/SchrijversWDD14} shows how
Prolog's search can be exposed with a free monad and manipulated using handlers.

None of these works consider mapping high-level to low-level representations of the effects.

\subsection{Reasoning About Side Effects}
\label{sec:reasoning-about-side-effects}

There are many works on reasoning and modelling side effects. 
Here, we cover those that have most directly inspired this paper. 

\paragraph*{Axiomatic Reasoning}\
Gibbons and Hinze \cite{Gibbons11} proposed to reason axiomatically about
programs with effects and provided an axiomatic characterization of
local state semantics. Our earlier work in \cite{Pauwels19} was
directly inspired by their work: we introduced an axiomatic
characterization of global state and used axiomatic reasoning to prove
handling local state with global state correct.  We also provided
models that satisfy the axioms, whereas their paper mistakenly claims
that one model satisfies the local state axioms and that another model
is monadic.
This paper is an extension of \cite{Pauwels19}, but notably, we
depart from the axiomatic reasoning approach; instead we use proof
techniques based on algebraic effects and handlers.

\paragraph*{Algebraic Effects}\
Our formulation of implementing local state with global state is directly 
inspired by the effect handlers approach of \citet{Plotkin09}.
By making the free monad explicit our proofs benefit directly from the induction
principle that Bauer and Pretnar established for effect handler programs
\cite{}.
While Lawvere theories were originally Plotkin's inspiration for studying 
algebraic effects, the effect handlers community has for a long time paid little
attention to them. 
Yet, \citet{LuksicP20}
have investigated a framework for encoding axioms or effect theories in the type
system: the type of an effectful function declares the operators used in the 
function, as well as the equalities that handlers for these operators should
comply with.  The type of a handler indicates which operators it handles and
which equations it complies with.  This allows expressing at the a handles a
higher-level effect in terms of a lower-level one.

\citet{Wu15} first presented fusion as a technique for optimizing compositions
of effect handlers. They use a specific form of fusion known as fold--build
fusion or short-cut fusion~\citep{shortcut}. To enable this kind of fusion they
transform the handler algebras to use the codensity monad as their carrier.
Their approach is not directly usable because it does not fuse non-handler functions,
and we derive simpler algebras (not obfuscated by the condisity monad) than those they do.

More recently \citet{YangW21} have used the fusion approach of \citet{Wu15} (but with
the continuation monad rather than the condensity monad) for reasoning;
they remark that, although handlers are composable, the
semantics of these composed handlers are not always obvious and that
determining the correct order of composition to arrive at a desired
semantics is nontrivial.
They propose a technique based on modular handlers \citep{Schrijvers19} which
considers conditions under which the fusion of these modular handlers
respect not only the laws of each of the handler's algebraic theories,
but also additional interaction laws. Using this technique they
provide succinct proofs of the correctness of local state
handlers, constructed from a fusion of state and nondeterminism
handlers.

\paragraph*{Earlier Versions}\
This paper refines and much expands on two earlier works of the last author.

\cite{Pauwels19} has the same goal as \Cref{sec:local-global}: it uses
the state-restoring version of put to simulate local state with global state.
It differs from this work in that it relies on an axiomatic (i.e., law-based),
as opposed to handler-based, semantics for local and global state. This means
that handler fusion cannot be used as a reasoning technique. Moreover, it uses a rather heavy-handed
syntactic approach to contextual equivalence, and it assumes
that no other effects are invoked.

Another precursor is the work of \cite{Seynaeve20}, which establishes similar results as
those in \Cref{sec:sim-nondet-state}. However, instead of generic definitions for the free 
monad and its fold, they use a specialized free monad for nondeterminism and explicitly recursive
handler functions. As a consequence, their proofs use structural induction rather than fold fusion.
Furthermore, they did not consider other effects either.

\section{Conclusion and Future Work}
\label{sec:conclusion}

We studied the simulations of higher-level effects with lower-level
effects for state and nondeterminism.
We started with the translation from the local-state semantics of
state and nondeterminism to the global-state semantics. Then, we
further showed how to translate nondeterminism to state (a choicepoint stack), and translate
multiple state effects into one state effect. Combining these results,
we can simulate the local-state semantics, a high-level programming
abstraction, with only one low-level state effect.
We also demonstrated that we can simulate the local-state semantics
using a trail stack in a similar style to the
Warren Abstract Machine of Prolog.
We define the effects and their translations with algebraic effects and
effect handlers respectively. These are implemented as free monads and folds in
Haskell.
The correctness of all these translations has been proved using the
technique of program calculation, especially using the fusion properties.

In future work, we would like to explore the potential optimisations
enabled by mutable states. Mutable states fit the global-state
semantics naturally. With mutable states, we can implement more
efficient state update and restoration operations for the simulation
\ensuremath{\Varid{local2global_M}} (\Cref{sec:undo}), as well as more efficient
implementations of the choicepoint stacks and trail stacks used by the
simulations \ensuremath{\Varid{nondet2state}} (\Cref{sec:nondet2state}) and \ensuremath{\Varid{local2trail}}
(\Cref{sec:trail-stack}), respectively.
We would also like to consider the low-level simulations of other
control-flow constructs used in logical programming languages such as
Prolog's \ensuremath{\Varid{cut}} operator for trimming the search space.
Since operators like \ensuremath{\Varid{cut}} are usually implemented as scoped
or higher-order effects~\citep{Pirog18,Wu14,YangPWBS22,BergS23}, we would have to 
adapt our approach accordingly.

\subsection*{Conflicts of Interest}

None.

\bibliographystyle{jfplike}
\bibliography{bibliography}

\clearpage

\appendix

\section{Proofs for Get Laws in Local-State Semantics}
\label{app:local-law}

In this section we prove two equations about the interaction of
nondeterminism and state in the local-state semantics.

\noindent
\Cref{eq:get-right-identity}: \ensuremath{\Varid{get}>\!\!>\Varid{\varnothing}\mathrel{=}\Varid{\varnothing}}

\begin{proof}~
\indentbegin \begin{hscode}\SaveRestoreHook
\column{B}{@{}>{\hspre}l<{\hspost}@{}}%
\column{3}{@{}>{\hspre}l<{\hspost}@{}}%
\column{6}{@{}>{\hspre}l<{\hspost}@{}}%
\column{E}{@{}>{\hspre}l<{\hspost}@{}}%
\>[6]{}\Varid{get}>\!\!>\Varid{\varnothing}{}\<[E]%
\\
\>[3]{}\mathrel{=}\mbox{\commentbegin ~  definition of \ensuremath{(>\!\!>)}   \commentend}{}\<[E]%
\\
\>[3]{}\hsindent{3}{}\<[6]%
\>[6]{}\Varid{get}>\!\!>\!\!=(\lambda \Varid{s}\to \Varid{\varnothing}){}\<[E]%
\\
\>[3]{}\mathrel{=}\mbox{\commentbegin ~  Law (\ref{eq:put-right-identity}): put-right-identity   \commentend}{}\<[E]%
\\
\>[3]{}\hsindent{3}{}\<[6]%
\>[6]{}\Varid{get}>\!\!>\!\!=(\lambda \Varid{s}\to \Varid{put}\;\Varid{s}>\!\!>\Varid{\varnothing}){}\<[E]%
\\
\>[3]{}\mathrel{=}\mbox{\commentbegin ~  Law (\ref{eq:monad-assoc}): associativity of \ensuremath{(>\!\!>)}   \commentend}{}\<[E]%
\\
\>[3]{}\hsindent{3}{}\<[6]%
\>[6]{}(\Varid{get}>\!\!>\!\!=\Varid{put})>\!\!>\Varid{\varnothing}{}\<[E]%
\\
\>[3]{}\mathrel{=}\mbox{\commentbegin ~  Law (\ref{eq:get-put}): get-put   \commentend}{}\<[E]%
\\
\>[3]{}\hsindent{3}{}\<[6]%
\>[6]{}\Varid{\eta}\;()>\!\!>\Varid{\varnothing}{}\<[E]%
\\
\>[3]{}\mathrel{=}\mbox{\commentbegin ~  Law (\ref{eq:monad-ret-bind}): return-bind and definition of \ensuremath{(>\!\!>)}   \commentend}{}\<[E]%
\\
\>[3]{}\hsindent{3}{}\<[6]%
\>[6]{}\Varid{\varnothing}{}\<[E]%
\ColumnHook
\end{hscode}\resethooks
\indentend \end{proof}

\noindent
\Cref{eq:get-left-dist}:
\ensuremath{\Varid{get}>\!\!>\!\!=(\lambda \Varid{x}\to \Varid{k}_{1}\;\Varid{x}\mathbin{\talloblong}\Varid{k}_{2}\;\Varid{x})\mathrel{=}(\Varid{get}>\!\!>\!\!=\Varid{k}_{1})\mathbin{\talloblong}(\Varid{get}>\!\!>\!\!=\Varid{k}_{2})}

\begin{proof}~
\indentbegin \begin{hscode}\SaveRestoreHook
\column{B}{@{}>{\hspre}l<{\hspost}@{}}%
\column{3}{@{}>{\hspre}l<{\hspost}@{}}%
\column{6}{@{}>{\hspre}l<{\hspost}@{}}%
\column{E}{@{}>{\hspre}l<{\hspost}@{}}%
\>[6]{}\Varid{get}>\!\!>\!\!=(\lambda \Varid{x}\to \Varid{k}_{1}\;\Varid{x}\mathbin{\talloblong}\Varid{k}_{2}\;\Varid{x}){}\<[E]%
\\
\>[3]{}\mathrel{=}\mbox{\commentbegin ~  Law (\ref{eq:monad-ret-bind}): return-bind and definition of \ensuremath{(>\!\!>)}   \commentend}{}\<[E]%
\\
\>[3]{}\hsindent{3}{}\<[6]%
\>[6]{}\Varid{\eta}\;()>\!\!>(\Varid{get}>\!\!>\!\!=(\lambda \Varid{x}\to \Varid{k}_{1}\;\Varid{x}\mathbin{\talloblong}\Varid{k}_{2}\;\Varid{x})){}\<[E]%
\\
\>[3]{}\mathrel{=}\mbox{\commentbegin ~  Law (\ref{eq:monad-assoc}): associativity of \ensuremath{(>\!\!>\!\!=)}   \commentend}{}\<[E]%
\\
\>[3]{}\hsindent{3}{}\<[6]%
\>[6]{}(\Varid{\eta}\;()>\!\!>\Varid{get})>\!\!>\!\!=(\lambda \Varid{x}\to \Varid{k}_{1}\;\Varid{x}\mathbin{\talloblong}\Varid{k}_{2}\;\Varid{x}){}\<[E]%
\\
\>[3]{}\mathrel{=}\mbox{\commentbegin ~  Law (\ref{eq:get-put}): get-put   \commentend}{}\<[E]%
\\
\>[3]{}\hsindent{3}{}\<[6]%
\>[6]{}((\Varid{get}>\!\!>\!\!=\Varid{put})>\!\!>\Varid{get})>\!\!>\!\!=(\lambda \Varid{x}\to \Varid{k}_{1}\;\Varid{x}\mathbin{\talloblong}\Varid{k}_{2}\;\Varid{x}){}\<[E]%
\\
\>[3]{}\mathrel{=}\mbox{\commentbegin ~  Law (\ref{eq:monad-assoc}): associativity of \ensuremath{(>\!\!>)}   \commentend}{}\<[E]%
\\
\>[3]{}\hsindent{3}{}\<[6]%
\>[6]{}(\Varid{get}>\!\!>\!\!=(\lambda \Varid{s}\to \Varid{put}\;\Varid{s}>\!\!>\Varid{get}))>\!\!>\!\!=(\lambda \Varid{x}\to \Varid{k}_{1}\;\Varid{x}\mathbin{\talloblong}\Varid{k}_{2}\;\Varid{x}){}\<[E]%
\\
\>[3]{}\mathrel{=}\mbox{\commentbegin ~  Law (\ref{eq:monad-assoc}): associativity of \ensuremath{(>\!\!>\!\!=)}   \commentend}{}\<[E]%
\\
\>[3]{}\hsindent{3}{}\<[6]%
\>[6]{}\Varid{get}>\!\!>\!\!=(\lambda \Varid{s}\to (\lambda \Varid{s}\to \Varid{put}\;\Varid{s}>\!\!>\Varid{get})\;\Varid{s}>\!\!>\!\!=(\lambda \Varid{x}\to \Varid{k}_{1}\;\Varid{x}\mathbin{\talloblong}\Varid{k}_{2}\;\Varid{x})){}\<[E]%
\\
\>[3]{}\mathrel{=}\mbox{\commentbegin ~  function application   \commentend}{}\<[E]%
\\
\>[3]{}\hsindent{3}{}\<[6]%
\>[6]{}\Varid{get}>\!\!>\!\!=(\lambda \Varid{s}\to (\Varid{put}\;\Varid{s}>\!\!>\Varid{get})>\!\!>\!\!=(\lambda \Varid{x}\to \Varid{k}_{1}\;\Varid{x}\mathbin{\talloblong}\Varid{k}_{2}\;\Varid{x})){}\<[E]%
\\
\>[3]{}\mathrel{=}\mbox{\commentbegin ~  Law (\ref{eq:put-get}): put-get   \commentend}{}\<[E]%
\\
\>[3]{}\hsindent{3}{}\<[6]%
\>[6]{}\Varid{get}>\!\!>\!\!=(\lambda \Varid{s}\to (\Varid{put}\;\Varid{s}>\!\!>\Varid{\eta}\;\Varid{s})>\!\!>\!\!=(\lambda \Varid{x}\to \Varid{k}_{1}\;\Varid{x}\mathbin{\talloblong}\Varid{k}_{2}\;\Varid{x})){}\<[E]%
\\
\>[3]{}\mathrel{=}\mbox{\commentbegin ~  Law (\ref{eq:monad-assoc}): associativity of \ensuremath{(>\!\!>)}   \commentend}{}\<[E]%
\\
\>[3]{}\hsindent{3}{}\<[6]%
\>[6]{}\Varid{get}>\!\!>\!\!=(\lambda \Varid{s}\to \Varid{put}\;\Varid{s}>\!\!>(\Varid{\eta}\;\Varid{s}>\!\!>\!\!=(\lambda \Varid{x}\to \Varid{k}_{1}\;\Varid{x}\mathbin{\talloblong}\Varid{k}_{2}\;\Varid{x}))){}\<[E]%
\\
\>[3]{}\mathrel{=}\mbox{\commentbegin ~  Law (\ref{eq:monad-ret-bind}): return-bind and function application   \commentend}{}\<[E]%
\\
\>[3]{}\hsindent{3}{}\<[6]%
\>[6]{}\Varid{get}>\!\!>\!\!=(\lambda \Varid{s}\to \Varid{put}\;\Varid{s}>\!\!>(\Varid{k}_{1}\;\Varid{s}\mathbin{\talloblong}\Varid{k}_{2}\;\Varid{s})){}\<[E]%
\\
\>[3]{}\mathrel{=}\mbox{\commentbegin ~  Law (\ref{eq:put-left-dist}): put-left-distributivity   \commentend}{}\<[E]%
\\
\>[3]{}\hsindent{3}{}\<[6]%
\>[6]{}\Varid{get}>\!\!>\!\!=(\lambda \Varid{s}\to (\Varid{put}\;\Varid{s}>\!\!>\Varid{k}_{1}\;\Varid{s})\mathbin{\talloblong}(\Varid{put}\;\Varid{s}>\!\!>\Varid{k}_{2}\;\Varid{s})){}\<[E]%
\\
\>[3]{}\mathrel{=}\mbox{\commentbegin ~  Law (\ref{eq:monad-ret-bind}): return-bind (twice)  \commentend}{}\<[E]%
\\
\>[3]{}\hsindent{3}{}\<[6]%
\>[6]{}\Varid{get}>\!\!>\!\!=(\lambda \Varid{s}\to (\Varid{put}\;\Varid{s}>\!\!>(\Varid{\eta}\;\Varid{s}>\!\!>\!\!=\Varid{k}_{1}))\mathbin{\talloblong}(\Varid{put}\;\Varid{s}>\!\!>(\Varid{\eta}\;\Varid{s}>\!\!>\!\!=\Varid{k}_{2}))){}\<[E]%
\\
\>[3]{}\mathrel{=}\mbox{\commentbegin ~  Law (\ref{eq:monad-assoc}): associativity of \ensuremath{(>\!\!>)}   \commentend}{}\<[E]%
\\
\>[3]{}\hsindent{3}{}\<[6]%
\>[6]{}\Varid{get}>\!\!>\!\!=(\lambda \Varid{s}\to ((\Varid{put}\;\Varid{s}>\!\!>\Varid{\eta}\;\Varid{s})>\!\!>\!\!=\Varid{k}_{1})\mathbin{\talloblong}((\Varid{put}\;\Varid{s}>\!\!>\Varid{\eta}\;\Varid{s})>\!\!>\!\!=\Varid{k}_{2})){}\<[E]%
\\
\>[3]{}\mathrel{=}\mbox{\commentbegin ~  Law (\ref{eq:put-get}): put-get   \commentend}{}\<[E]%
\\
\>[3]{}\hsindent{3}{}\<[6]%
\>[6]{}\Varid{get}>\!\!>\!\!=(\lambda \Varid{s}\to ((\Varid{put}\;\Varid{s}>\!\!>\Varid{get})>\!\!>\!\!=\Varid{k}_{1})\mathbin{\talloblong}((\Varid{put}\;\Varid{s}>\!\!>\Varid{get})>\!\!>\!\!=\Varid{k}_{2})){}\<[E]%
\\
\>[3]{}\mathrel{=}\mbox{\commentbegin ~  Law (\ref{eq:monad-assoc}): associativity of \ensuremath{(>\!\!>)}   \commentend}{}\<[E]%
\\
\>[3]{}\hsindent{3}{}\<[6]%
\>[6]{}\Varid{get}>\!\!>\!\!=(\lambda \Varid{s}\to (\Varid{put}\;\Varid{s}>\!\!>(\Varid{get}>\!\!>\!\!=\Varid{k}_{1}))\mathbin{\talloblong}(\Varid{put}\;\Varid{s}>\!\!>(\Varid{get}>\!\!>\!\!=\Varid{k}_{2}))){}\<[E]%
\\
\>[3]{}\mathrel{=}\mbox{\commentbegin ~  Law (\ref{eq:put-left-dist}): put-left-distributivity   \commentend}{}\<[E]%
\\
\>[3]{}\hsindent{3}{}\<[6]%
\>[6]{}\Varid{get}>\!\!>\!\!=(\lambda \Varid{s}\to \Varid{put}\;\Varid{s}>\!\!>((\Varid{get}>\!\!>\!\!=\Varid{k}_{1})\mathbin{\talloblong}(\Varid{get}>\!\!>\!\!=\Varid{k}_{2}))){}\<[E]%
\\
\>[3]{}\mathrel{=}\mbox{\commentbegin ~  Law (\ref{eq:monad-assoc}): associativity of \ensuremath{(>\!\!>\!\!=)}   \commentend}{}\<[E]%
\\
\>[3]{}\hsindent{3}{}\<[6]%
\>[6]{}(\Varid{get}>\!\!>\!\!=\Varid{put})>\!\!>((\Varid{get}>\!\!>\!\!=\Varid{k}_{1})\mathbin{\talloblong}(\Varid{get}>\!\!>\!\!=\Varid{k}_{2})){}\<[E]%
\\
\>[3]{}\mathrel{=}\mbox{\commentbegin ~  Law (\ref{eq:get-put}): get-put  \commentend}{}\<[E]%
\\
\>[3]{}\hsindent{3}{}\<[6]%
\>[6]{}\Varid{\eta}\;()>\!\!>((\Varid{get}>\!\!>\!\!=\Varid{k}_{1})\mathbin{\talloblong}(\Varid{get}>\!\!>\!\!=\Varid{k}_{2})){}\<[E]%
\\
\>[3]{}\mathrel{=}\mbox{\commentbegin ~  Law (\ref{eq:monad-ret-bind}): return-bind and definition of \ensuremath{(>\!\!>)}   \commentend}{}\<[E]%
\\
\>[3]{}\hsindent{3}{}\<[6]%
\>[6]{}(\Varid{get}>\!\!>\!\!=\Varid{k}_{1})\mathbin{\talloblong}(\Varid{get}>\!\!>\!\!=\Varid{k}_{2}){}\<[E]%
\ColumnHook
\end{hscode}\resethooks
\indentend \end{proof}

\section{Proofs for Modelling Local State with Global State}
\label{app:local-global}

This section proves the following theorem in \Cref{sec:local2global}.

\localGlobal*

\paragraph*{Preliminary}
It is easy to see that \ensuremath{\Varid{run_{StateT}}\hsdot{\circ }{.}\Varid{h_{State}}} can be fused into a single fold defined as follows:
\indentbegin \begin{hscode}\SaveRestoreHook
\column{B}{@{}>{\hspre}l<{\hspost}@{}}%
\column{3}{@{}>{\hspre}l<{\hspost}@{}}%
\column{5}{@{}>{\hspre}l<{\hspost}@{}}%
\column{10}{@{}>{\hspre}c<{\hspost}@{}}%
\column{10E}{@{}l@{}}%
\column{13}{@{}>{\hspre}l<{\hspost}@{}}%
\column{21}{@{}>{\hspre}l<{\hspost}@{}}%
\column{24}{@{}>{\hspre}l<{\hspost}@{}}%
\column{E}{@{}>{\hspre}l<{\hspost}@{}}%
\>[B]{}\Varid{h_{State1}}\mathbin{::}\Conid{Functor}\;\Varid{f}\Rightarrow \Conid{Free}\;(\Varid{State_{F}}\;\Varid{s}\mathrel{{:}{+}{:}}\Varid{f})\;\Varid{a}\to (\Varid{s}\to \Conid{Free}\;\Varid{f}\;(\Varid{a},\Varid{s})){}\<[E]%
\\
\>[B]{}\Varid{h_{State1}}{}\<[10]%
\>[10]{}\mathrel{=}{}\<[10E]%
\>[13]{}\Varid{fold}\;\Varid{gen_{S}}\;(\Varid{alg_{S}}\mathbin{\#}\Varid{fwd_{S}}){}\<[E]%
\\
\>[B]{}\hsindent{3}{}\<[3]%
\>[3]{}\mathbf{where}{}\<[E]%
\\
\>[3]{}\hsindent{2}{}\<[5]%
\>[5]{}\Varid{gen_{S}}\;\Varid{x}\;{}\<[21]%
\>[21]{}\Varid{s}{}\<[24]%
\>[24]{}\mathrel{=}\Conid{Var}\;(\Varid{x},\Varid{s}){}\<[E]%
\\
\>[3]{}\hsindent{2}{}\<[5]%
\>[5]{}\Varid{alg_{S}}\;(\Conid{Get}\;\Varid{k})\;{}\<[21]%
\>[21]{}\Varid{s}{}\<[24]%
\>[24]{}\mathrel{=}\Varid{k}\;\Varid{s}\;\Varid{s}{}\<[E]%
\\
\>[3]{}\hsindent{2}{}\<[5]%
\>[5]{}\Varid{alg_{S}}\;(\Conid{Put}\;\Varid{s}\;\Varid{k})\;{}\<[21]%
\>[21]{}\anonymous {}\<[24]%
\>[24]{}\mathrel{=}\Varid{k}\;\Varid{s}{}\<[E]%
\\
\>[3]{}\hsindent{2}{}\<[5]%
\>[5]{}\Varid{fwd_{S}}\;\Varid{y}\;{}\<[21]%
\>[21]{}\Varid{s}{}\<[24]%
\>[24]{}\mathrel{=}\Conid{Op}\;(\Varid{fmap}\;(\mathbin{\$}\Varid{s})\;\Varid{y}){}\<[E]%
\ColumnHook
\end{hscode}\resethooks
\indentend For brevity, we use \ensuremath{\Varid{h_{State1}}} to replace \ensuremath{\Varid{run_{StateT}}\hsdot{\circ }{.}\Varid{h_{State}}} in the following proofs.

\subsection{Main Proof Structure}
The main theorem we prove in this section is:
\begin{theorem}\label{eq:local-global}
\ensuremath{\Varid{h_{Global}}\hsdot{\circ }{.}\Varid{local2global}\mathrel{=}\Varid{h_{Local}}}
\end{theorem}
\begin{proof}
Both the left-hand side and the right-hand side of the equation consist of 
function compositions involving one or more folds.
We apply fold fusion separately on both sides to contract each
into a single fold:
\begin{eqnarray*}
\ensuremath{\Varid{h_{Global}}\hsdot{\circ }{.}\Varid{local2global}} & = & \ensuremath{\Varid{fold}\;\Varid{gen}_{\Varid{LHS}}\;(\Varid{alg}_{\Varid{LHS}}^{\Varid{S}}\mathbin{\#}\Varid{alg}_{\Varid{RHS}}^{\Varid{ND}}\mathbin{\#}\Varid{fwd}_{\Varid{LHS}})} \\
\ensuremath{\Varid{h_{Local}}}& = & \ensuremath{\Varid{fold}\;\Varid{gen}_{\Varid{RHS}}\;(\Varid{alg}_{\Varid{RHS}}^{\Varid{S}}\mathbin{\#}\Varid{alg}_{\Varid{RHS}}^{\Varid{ND}}\mathbin{\#}\Varid{fwd}_{\Varid{RHS}})}
\end{eqnarray*}
We approach this calculationally. That is to say, we do not first postulate
definitions of the unknowns above (\ensuremath{\Varid{alg}_{\Varid{LHS}}^{\Varid{S}}} and so on) and then verify whether
the fusion conditions are satisfied. Instead, we discover the definitions of the unknowns.
We start from the known side of
each fusion condition and perform case analysis on the possible shapes of
input. By simplifying the resulting case-specific expression, and pushing the handler
applications inwards, we end up at a point where we can read off the definition
of the unknown that makes the fusion condition hold for that case.

Finally, we show that both folds are equal by showing that their
corresponding parameters are equal:
\begin{eqnarray*}
\ensuremath{\Varid{gen}_{\Varid{LHS}}} & = & \ensuremath{\Varid{gen}_{\Varid{RHS}}} \\
\ensuremath{\Varid{alg}_{\Varid{LHS}}^{\Varid{S}}} & = & \ensuremath{\Varid{alg}_{\Varid{RHS}}^{\Varid{S}}} \\
\ensuremath{\Varid{alg}_{\Varid{LHS}}^{\Varid{ND}}} & = & \ensuremath{\Varid{alg}_{\Varid{RHS}}^{\Varid{ND}}} \\
\ensuremath{\Varid{fwd}_{\Varid{LHS}}} & = & \ensuremath{\Varid{fwd}_{\Varid{RHS}}}
\end{eqnarray*}

A noteworthy observation is that, for fusing the left-hand side of the equation, we do not use the standard
fusion rule:
\begin{eqnarray*}
    \ensuremath{\Varid{h_{Global}}\hsdot{\circ }{.}\Varid{fold}\;\Conid{Var}\;\Varid{alg}} & = & \ensuremath{\Varid{fold}\;(\Varid{h_{Global}}\hsdot{\circ }{.}\Conid{Var})\;\Varid{alg'}} \\
     \Leftarrow \qquad
   \ensuremath{\Varid{h_{Global}}\hsdot{\circ }{.}\Varid{alg}} & = & \ensuremath{\Varid{alg'}\hsdot{\circ }{.}\Varid{fmap}\;\Varid{h_{Global}}}
\end{eqnarray*}
where \ensuremath{\Varid{local2global}\mathrel{=}\Varid{fold}\;\Conid{Var}\;\Varid{alg}}. The problem is that we will not find an
appropriate \ensuremath{\Varid{alg'}} such that \ensuremath{\Varid{alg'}\;(\Varid{fmap}\;\Varid{h_{Global}}\;\Varid{t})} restores the state for any
\ensuremath{\Varid{t}} of type \ensuremath{(\Varid{State_{F}}\;\Varid{s}\mathrel{{:}{+}{:}}\Varid{Nondet_{F}}\mathrel{{:}{+}{:}}\Varid{f})\;(\Conid{Free}\;(\Varid{State_{F}}\;\Varid{s}\mathrel{{:}{+}{:}}\Conid{NonDetF}\mathrel{{:}{+}{:}}\Varid{f})\;\Varid{a})}. 

Fortunately, we do not need such an \ensuremath{\Varid{alg'}}. As we have already pointed out, we
can assume that the subterms of \ensuremath{\Varid{t}} have already been transformed by
\ensuremath{\Varid{local2global}}, and thus all occurrences of \ensuremath{\Conid{Put}} appear in the \ensuremath{\Varid{put_{R}}}
constellation.

We can incorporate this assumption by using the alternativee fusion rule:
\begin{eqnarray*}
    \ensuremath{\Varid{h_{Global}}\hsdot{\circ }{.}\Varid{fold}\;\Conid{Var}\;\Varid{alg}} & = & \ensuremath{\Varid{fold}\;(\Varid{h_{Global}}\hsdot{\circ }{.}\Conid{Var})\;\Varid{alg'}} \\
     \Leftarrow \qquad
   \ensuremath{\Varid{h_{Global}}\hsdot{\circ }{.}\Varid{alg}\hsdot{\circ }{.}\Varid{fmap}\;\Varid{local2global}} & = & \ensuremath{\Varid{alg'}\hsdot{\circ }{.}\Varid{fmap}\;\Varid{h_{Global}}\hsdot{\circ }{.}\Varid{fmap}\;\Varid{local2global}}
\end{eqnarray*}
The additional \ensuremath{\Varid{fmap}\;\Varid{local2global}} in the condition captures the property that
all the subterms have been transformed by \ensuremath{\Varid{local2global}}.

In order to not clutter the proofs, we abstract everywhere over this additional \ensuremath{\Varid{fmap}\;\Varid{local2global}} application, except
in the one place where we need it. That is the appeal to the key lemma:
\begin{eqnarray*}
& \ensuremath{\Varid{h_{State1}}\;(\Varid{h_{ND+f}}\;(\Varid{(\Leftrightarrow)}\;(\Varid{local2global}\;\Varid{t})))\;\Varid{s}} & \\
& = & \\
& \ensuremath{\mathbf{do}\;(\Varid{x},\anonymous )\leftarrow \Varid{h_{State1}}\;(\Varid{h_{ND+f}}\;(\Varid{(\Leftrightarrow)}\;(\Varid{local2global}\;\Varid{t})))\;\Varid{s};\Varid{\eta}\;(\Varid{x},\Varid{s})} &
\end{eqnarray*}
This expresses that the syntactic transformation \ensuremath{\Varid{local2global}} makes sure
that, despite any temporary changes, the computation \ensuremath{\Varid{t}} restores the state
back to its initial value.

We elaborate each of these steps below.
\end{proof}

\subsection{Fusing the Right-Hand Side}
We calculate as follows:
\indentbegin \begin{hscode}\SaveRestoreHook
\column{B}{@{}>{\hspre}l<{\hspost}@{}}%
\column{5}{@{}>{\hspre}l<{\hspost}@{}}%
\column{7}{@{}>{\hspre}l<{\hspost}@{}}%
\column{10}{@{}>{\hspre}l<{\hspost}@{}}%
\column{E}{@{}>{\hspre}l<{\hspost}@{}}%
\>[5]{}\Varid{h_{Local}}{}\<[E]%
\\
\>[B]{}\mathrel{=}\mbox{\commentbegin ~  definition  \commentend}{}\<[E]%
\\
\>[B]{}\hsindent{5}{}\<[5]%
\>[5]{}\Varid{h_{L}}\hsdot{\circ }{.}\Varid{h_{State1}}{}\<[E]%
\\
\>[5]{}\hsindent{2}{}\<[7]%
\>[7]{} \text{with } {}\<[E]%
\\
\>[7]{}\hsindent{3}{}\<[10]%
\>[10]{}\Varid{h_{L}}\mathbin{::}(\Conid{Functor}\;\Varid{f})\Rightarrow (\Varid{s}\to \Conid{Free}\;(\Varid{Nondet_{F}}\mathrel{{:}{+}{:}}\Varid{f})\;(\Varid{a},\Varid{s}))\to \Varid{s}\to \Conid{Free}\;\Varid{f}\;[\mskip1.5mu \Varid{a}\mskip1.5mu]{}\<[E]%
\\
\>[7]{}\hsindent{3}{}\<[10]%
\>[10]{}\Varid{h_{L}}\mathrel{=}\Varid{fmap}\;(\Varid{fmap}\;(\Varid{fmap}\;\Varid{fst})\hsdot{\circ }{.}\Varid{h_{ND+f}}){}\<[E]%
\\
\>[B]{}\mathrel{=}\mbox{\commentbegin ~  definition of \ensuremath{\Varid{h_{State1}}}   \commentend}{}\<[E]%
\\
\>[B]{}\hsindent{5}{}\<[5]%
\>[5]{}\Varid{h_{L}}\hsdot{\circ }{.}\Varid{fold}\;\Varid{gen_{S}}\;(\Varid{alg_{S}}\mathbin{\#}\Varid{fwd_{S}}){}\<[E]%
\\
\>[B]{}\mathrel{=}\mbox{\commentbegin ~  fold fusion-post (Equation \ref{eq:fusion-post})   \commentend}{}\<[E]%
\\
\>[B]{}\hsindent{5}{}\<[5]%
\>[5]{}\Varid{fold}\;\Varid{gen}_{\Varid{RHS}}\;(\Varid{alg}_{\Varid{RHS}}^{\Varid{S}}\mathbin{\#}\Varid{alg}_{\Varid{RHS}}^{\Varid{ND}}\mathbin{\#}\Varid{fwd}_{\Varid{RHS}}){}\<[E]%
\ColumnHook
\end{hscode}\resethooks
\indentend This last step is valid provided that the fusion conditions are satisfied:
\begin{eqnarray*}
\ensuremath{\Varid{h_{L}}\hsdot{\circ }{.}\Varid{gen_{S}}} & = & \ensuremath{\Varid{gen}_{\Varid{RHS}}} \\
\ensuremath{\Varid{h_{L}}\hsdot{\circ }{.}(\Varid{alg_{S}}\mathbin{\#}\Varid{fwd_{S}})} & = & \ensuremath{(\Varid{alg}_{\Varid{RHS}}^{\Varid{S}}\mathbin{\#}\Varid{alg}_{\Varid{RHS}}^{\Varid{ND}}\mathbin{\#}\Varid{fwd}_{\Varid{RHS}})\hsdot{\circ }{.}\Varid{fmap}\;\Varid{h_{L}}}
\end{eqnarray*}

We calculate for the first fusion condition:\indentbegin \begin{hscode}\SaveRestoreHook
\column{B}{@{}>{\hspre}l<{\hspost}@{}}%
\column{3}{@{}>{\hspre}l<{\hspost}@{}}%
\column{5}{@{}>{\hspre}l<{\hspost}@{}}%
\column{E}{@{}>{\hspre}l<{\hspost}@{}}%
\>[5]{}\Varid{h_{L}}\;(\Varid{gen_{S}}\;\Varid{x}){}\<[E]%
\\
\>[3]{}\mathrel{=}\mbox{\commentbegin ~ definition of \ensuremath{\Varid{gen_{S}}}  \commentend}{}\<[E]%
\\
\>[3]{}\hsindent{2}{}\<[5]%
\>[5]{}\Varid{h_{L}}\;(\lambda \Varid{s}\to \Conid{Var}\;(\Varid{x},\Varid{s})){}\<[E]%
\\
\>[3]{}\mathrel{=}\mbox{\commentbegin ~ definition of \ensuremath{\Varid{h_{L}}}  \commentend}{}\<[E]%
\\
\>[3]{}\hsindent{2}{}\<[5]%
\>[5]{}\Varid{fmap}\;(\Varid{fmap}\;(\Varid{fmap}\;\Varid{fst})\hsdot{\circ }{.}\Varid{h_{ND+f}})\;(\lambda \Varid{s}\to \Conid{Var}\;(\Varid{x},\Varid{s})){}\<[E]%
\\
\>[3]{}\mathrel{=}\mbox{\commentbegin ~ definition of \ensuremath{\Varid{fmap}}  \commentend}{}\<[E]%
\\
\>[3]{}\hsindent{2}{}\<[5]%
\>[5]{}\lambda \Varid{s}\to \Varid{fmap}\;(\Varid{fmap}\;\Varid{fst})\;(\Varid{h_{ND+f}}\;(\Conid{Var}\;(\Varid{x},\Varid{s}))){}\<[E]%
\\
\>[3]{}\mathrel{=}\mbox{\commentbegin ~ definition of \ensuremath{\Varid{h_{ND+f}}}  \commentend}{}\<[E]%
\\
\>[3]{}\hsindent{2}{}\<[5]%
\>[5]{}\lambda \Varid{s}\to \Varid{fmap}\;(\Varid{fmap}\;\Varid{fst})\;(\Conid{Var}\;[\mskip1.5mu (\Varid{x},\Varid{s})\mskip1.5mu]){}\<[E]%
\\
\>[3]{}\mathrel{=}\mbox{\commentbegin ~ definition of \ensuremath{\Varid{fmap}} (twice)  \commentend}{}\<[E]%
\\
\>[3]{}\hsindent{2}{}\<[5]%
\>[5]{}\lambda \Varid{s}\to \Conid{Var}\;[\mskip1.5mu \Varid{x}\mskip1.5mu]{}\<[E]%
\\
\>[3]{}\mathrel{=}\mbox{\commentbegin ~ define \ensuremath{\Varid{gen}_{\Varid{RHS}}\;\Varid{x}\mathrel{=}\lambda \Varid{s}\to \Conid{Var}\;[\mskip1.5mu \Varid{x}\mskip1.5mu]}  \commentend}{}\<[E]%
\\
\>[3]{}\mathrel{=}\Varid{gen}_{\Varid{RHS}}\;\Varid{x}{}\<[E]%
\ColumnHook
\end{hscode}\resethooks
\indentend We conclude that the first fusion condition is satisfied by:
\indentbegin \begin{hscode}\SaveRestoreHook
\column{B}{@{}>{\hspre}l<{\hspost}@{}}%
\column{3}{@{}>{\hspre}l<{\hspost}@{}}%
\column{E}{@{}>{\hspre}l<{\hspost}@{}}%
\>[3]{}\Varid{gen}_{\Varid{RHS}}\mathbin{::}\Conid{Functor}\;\Varid{f}\Rightarrow \Varid{a}\to (\Varid{s}\to \Conid{Free}\;\Varid{f}\;[\mskip1.5mu \Varid{a}\mskip1.5mu]){}\<[E]%
\\
\>[3]{}\Varid{gen}_{\Varid{RHS}}\;\Varid{x}\mathrel{=}\lambda \Varid{s}\to \Conid{Var}\;[\mskip1.5mu \Varid{x}\mskip1.5mu]{}\<[E]%
\ColumnHook
\end{hscode}\resethooks
\indentend The second fusion condition decomposes into two separate conditions:
\begin{eqnarray*}
\ensuremath{\Varid{h_{L}}\hsdot{\circ }{.}\Varid{alg_{S}}} & = & \ensuremath{\Varid{alg}_{\Varid{RHS}}^{\Varid{S}}\hsdot{\circ }{.}\Varid{fmap}\;\Varid{h_{L}}} \\
\ensuremath{\Varid{h_{L}}\hsdot{\circ }{.}\Varid{fwd_{S}}} & = & \ensuremath{(\Varid{alg}_{\Varid{RHS}}^{\Varid{ND}}\mathbin{\#}\Varid{fwd}_{\Varid{RHS}})\hsdot{\circ }{.}\Varid{fmap}\;\Varid{h_{L}}}
\end{eqnarray*}

We calculate for the first subcondition:

\noindent \mbox{\underline{case \ensuremath{\Varid{t}\mathrel{=}\Conid{Get}\;\Varid{k}}}}\indentbegin \begin{hscode}\SaveRestoreHook
\column{B}{@{}>{\hspre}l<{\hspost}@{}}%
\column{3}{@{}>{\hspre}l<{\hspost}@{}}%
\column{5}{@{}>{\hspre}l<{\hspost}@{}}%
\column{E}{@{}>{\hspre}l<{\hspost}@{}}%
\>[5]{}\Varid{h_{L}}\;(\Varid{alg_{S}}\;(\Conid{Get}\;\Varid{k})){}\<[E]%
\\
\>[3]{}\mathrel{=}\mbox{\commentbegin ~  definition of \ensuremath{\Varid{alg_{S}}}  \commentend}{}\<[E]%
\\
\>[3]{}\hsindent{2}{}\<[5]%
\>[5]{}\Varid{h_{L}}\;(\lambda \Varid{s}\to \Varid{k}\;\Varid{s}\;\Varid{s}){}\<[E]%
\\
\>[3]{}\mathrel{=}\mbox{\commentbegin ~  definition of \ensuremath{\Varid{h_{L}}}  \commentend}{}\<[E]%
\\
\>[3]{}\hsindent{2}{}\<[5]%
\>[5]{}\Varid{fmap}\;(\Varid{fmap}\;(\Varid{fmap}\;\Varid{fst})\hsdot{\circ }{.}\Varid{h_{ND+f}})\;(\lambda \Varid{s}\to \Varid{k}\;\Varid{s}\;\Varid{s}){}\<[E]%
\\
\>[3]{}\mathrel{=}\mbox{\commentbegin ~  definition of \ensuremath{\Varid{fmap}}  \commentend}{}\<[E]%
\\
\>[3]{}\hsindent{2}{}\<[5]%
\>[5]{}\lambda \Varid{s}\to \Varid{fmap}\;(\Varid{fmap}\;\Varid{fst})\;(\Varid{h_{ND+f}}\;(\Varid{k}\;\Varid{s}\;\Varid{s})){}\<[E]%
\\
\>[3]{}\mathrel{=}\mbox{\commentbegin ~  beta-expansion (twice)  \commentend}{}\<[E]%
\\
\>[3]{}\mathrel{=}\lambda \Varid{s}\to (\lambda \Varid{s}_{1}\;\Varid{s}_{2}\to \Varid{fmap}\;(\Varid{fmap}\;\Varid{fst})\;(\Varid{h_{ND+f}}\;(\Varid{k}\;\Varid{s}_{2}\;\Varid{s}_{1})))\;\Varid{s}\;\Varid{s}{}\<[E]%
\\
\>[3]{}\mathrel{=}\mbox{\commentbegin ~  definition of \ensuremath{\Varid{fmap}} (twice)  \commentend}{}\<[E]%
\\
\>[3]{}\mathrel{=}\lambda \Varid{s}\to (\Varid{fmap}\;(\Varid{fmap}\;(\Varid{fmap}\;(\Varid{fmap}\;\Varid{fst})\hsdot{\circ }{.}\Varid{h_{ND+f}}))\;(\lambda \Varid{s}_{1}\;\Varid{s}_{2}\to \Varid{k}\;\Varid{s}_{2}\;\Varid{s}_{1}))\;\Varid{s}\;\Varid{s}{}\<[E]%
\\
\>[3]{}\mathrel{=}\mbox{\commentbegin ~  eta-expansion of \ensuremath{\Varid{k}}  \commentend}{}\<[E]%
\\
\>[3]{}\mathrel{=}\lambda \Varid{s}\to (\Varid{fmap}\;(\Varid{fmap}\;(\Varid{fmap}\;(\Varid{fmap}\;\Varid{fst})\hsdot{\circ }{.}\Varid{h_{ND+f}}))\;\Varid{k})\;\Varid{s}\;\Varid{s}{}\<[E]%
\\
\>[3]{}\mathrel{=}\mbox{\commentbegin ~  define \ensuremath{\Varid{alg}_{\Varid{RHS}}^{\Varid{S}}\;(\Conid{Get}\;\Varid{k})\mathrel{=}\lambda \Varid{s}\to \Varid{k}\;\Varid{s}\;\Varid{s}}  \commentend}{}\<[E]%
\\
\>[3]{}\mathrel{=}\Varid{alg}_{\Varid{RHS}}^{\Varid{S}}\;(\Conid{Get}\;(\Varid{fmap}\;(\Varid{fmap}\;(\Varid{fmap}\;(\Varid{fmap}\;\Varid{fst})\hsdot{\circ }{.}\Varid{h_{ND+f}}))\;\Varid{k})){}\<[E]%
\\
\>[3]{}\mathrel{=}\mbox{\commentbegin ~  definition of \ensuremath{\Varid{fmap}}  \commentend}{}\<[E]%
\\
\>[3]{}\mathrel{=}\Varid{alg}_{\Varid{RHS}}^{\Varid{S}}\;(\Varid{fmap}\;(\Varid{fmap}\;(\Varid{fmap}\;(\Varid{fmap}\;\Varid{fst})\hsdot{\circ }{.}\Varid{h_{ND+f}}))\;(\Conid{Get}\;\Varid{k})){}\<[E]%
\\
\>[3]{}\mathrel{=}\mbox{\commentbegin ~  definition of \ensuremath{\Varid{h_{L}}}  \commentend}{}\<[E]%
\\
\>[3]{}\mathrel{=}\Varid{alg}_{\Varid{RHS}}^{\Varid{S}}\;(\Varid{fmap}\;\Varid{h_{L}}\;(\Conid{Get}\;\Varid{k})){}\<[E]%
\ColumnHook
\end{hscode}\resethooks
\indentend \noindent \mbox{\underline{case \ensuremath{\Varid{t}\mathrel{=}\Conid{Put}\;\Varid{s}\;\Varid{k}}}}\indentbegin \begin{hscode}\SaveRestoreHook
\column{B}{@{}>{\hspre}l<{\hspost}@{}}%
\column{3}{@{}>{\hspre}l<{\hspost}@{}}%
\column{5}{@{}>{\hspre}l<{\hspost}@{}}%
\column{E}{@{}>{\hspre}l<{\hspost}@{}}%
\>[5]{}\Varid{h_{L}}\;(\Varid{alg_{S}}\;(\Conid{Put}\;\Varid{s}\;\Varid{k})){}\<[E]%
\\
\>[3]{}\mathrel{=}\mbox{\commentbegin ~  definition of \ensuremath{\Varid{alg_{S}}}  \commentend}{}\<[E]%
\\
\>[3]{}\hsindent{2}{}\<[5]%
\>[5]{}\Varid{h_{L}}\;(\lambda \anonymous \to \Varid{k}\;\Varid{s}){}\<[E]%
\\
\>[3]{}\mathrel{=}\mbox{\commentbegin ~  definition of \ensuremath{\Varid{h_{L}}}  \commentend}{}\<[E]%
\\
\>[3]{}\hsindent{2}{}\<[5]%
\>[5]{}\Varid{fmap}\;(\Varid{fmap}\;(\Varid{fmap}\;\Varid{fst})\hsdot{\circ }{.}\Varid{h_{ND+f}})\;(\lambda \anonymous \to \Varid{k}\;\Varid{s}){}\<[E]%
\\
\>[3]{}\mathrel{=}\mbox{\commentbegin ~  definition of \ensuremath{\Varid{fmap}}  \commentend}{}\<[E]%
\\
\>[3]{}\hsindent{2}{}\<[5]%
\>[5]{}\lambda \anonymous \to \Varid{fmap}\;(\Varid{fmap}\;\Varid{fst})\;(\Varid{h_{ND+f}}\;(\Varid{k}\;\Varid{s})){}\<[E]%
\\
\>[3]{}\mathrel{=}\mbox{\commentbegin ~  beta-expansion  \commentend}{}\<[E]%
\\
\>[3]{}\mathrel{=}\lambda \anonymous \to (\lambda \Varid{s}_{1}\to \Varid{fmap}\;(\Varid{fmap}\;\Varid{fst})\;(\Varid{h_{ND+f}}\;(\Varid{k}\;\Varid{s}_{1})))\;\Varid{s}{}\<[E]%
\\
\>[3]{}\mathrel{=}\mbox{\commentbegin ~  definition of \ensuremath{\Varid{fmap}}  \commentend}{}\<[E]%
\\
\>[3]{}\mathrel{=}\lambda \anonymous \to (\Varid{fmap}\;(\Varid{fmap}\;(\Varid{fmap}\;\Varid{fst})\hsdot{\circ }{.}\Varid{h_{ND+f}})\;(\lambda \Varid{s}_{1}\to \Varid{k}\;\Varid{s}_{1}))\;\Varid{s}{}\<[E]%
\\
\>[3]{}\mathrel{=}\mbox{\commentbegin ~  eta-expansion of \ensuremath{\Varid{k}}  \commentend}{}\<[E]%
\\
\>[3]{}\mathrel{=}\lambda \anonymous \to (\Varid{fmap}\;(\Varid{fmap}\;(\Varid{fmap}\;\Varid{fst})\hsdot{\circ }{.}\Varid{h_{ND+f}})\;\Varid{k})\;\Varid{s}{}\<[E]%
\\
\>[3]{}\mathrel{=}\mbox{\commentbegin ~  define \ensuremath{\Varid{alg}_{\Varid{RHS}}^{\Varid{S}}\;(\Conid{Pus}\;\Varid{s}\;\Varid{k})\mathrel{=}\mathbin{\char92 \char95 }\to \Varid{k}\;\Varid{s}}  \commentend}{}\<[E]%
\\
\>[3]{}\mathrel{=}\Varid{alg}_{\Varid{RHS}}^{\Varid{S}}\;(\Conid{Put}\;\Varid{s}\;(\Varid{fmap}\;(\Varid{fmap}\;(\Varid{fmap}\;\Varid{fst})\hsdot{\circ }{.}\Varid{h_{ND+f}})\;\Varid{k})){}\<[E]%
\\
\>[3]{}\mathrel{=}\mbox{\commentbegin ~  definition of \ensuremath{\Varid{fmap}}  \commentend}{}\<[E]%
\\
\>[3]{}\mathrel{=}\Varid{alg}_{\Varid{RHS}}^{\Varid{S}}\;(\Varid{fmap}\;(\Varid{fmap}\;(\Varid{fmap}\;\Varid{fst})\hsdot{\circ }{.}\Varid{h_{ND+f}}))\;(\Conid{Put}\;\Varid{s}\;\Varid{k})){}\<[E]%
\\
\>[3]{}\mathrel{=}\mbox{\commentbegin ~  definition of \ensuremath{\Varid{h_{L}}}  \commentend}{}\<[E]%
\\
\>[3]{}\mathrel{=}\Varid{alg}_{\Varid{RHS}}^{\Varid{S}}\;(\Varid{fmap}\;\Varid{h_{L}}\;(\Conid{Put}\;\Varid{s}\;\Varid{k})){}\<[E]%
\ColumnHook
\end{hscode}\resethooks
\indentend We conclude that the first subcondition is met by taking:
\indentbegin \begin{hscode}\SaveRestoreHook
\column{B}{@{}>{\hspre}l<{\hspost}@{}}%
\column{3}{@{}>{\hspre}l<{\hspost}@{}}%
\column{22}{@{}>{\hspre}l<{\hspost}@{}}%
\column{E}{@{}>{\hspre}l<{\hspost}@{}}%
\>[3]{}\Varid{alg}_{\Varid{RHS}}^{\Varid{S}}\mathbin{::}\Conid{Functor}\;\Varid{f}\Rightarrow \Varid{State_{F}}\;\Varid{s}\;(\Varid{s}\to \Conid{Free}\;\Varid{f}\;[\mskip1.5mu \Varid{a}\mskip1.5mu])\to (\Varid{s}\to \Conid{Free}\;\Varid{f}\;[\mskip1.5mu \Varid{a}\mskip1.5mu]){}\<[E]%
\\
\>[3]{}\Varid{alg}_{\Varid{RHS}}^{\Varid{S}}\;(\Conid{Get}\;\Varid{k}){}\<[22]%
\>[22]{}\mathrel{=}\lambda \Varid{s}\to \Varid{k}\;\Varid{s}\;\Varid{s}{}\<[E]%
\\
\>[3]{}\Varid{alg}_{\Varid{RHS}}^{\Varid{S}}\;(\Conid{Put}\;\Varid{s}\;\Varid{k}){}\<[22]%
\>[22]{}\mathrel{=}\lambda \anonymous \to \Varid{k}\;\Varid{s}{}\<[E]%
\ColumnHook
\end{hscode}\resethooks
\indentend The second subcondition can be split up in two further subconditions:
\begin{eqnarray*}
\ensuremath{\Varid{h_{L}}\hsdot{\circ }{.}\Varid{fwd_{S}}\hsdot{\circ }{.}\Conid{Inl}}& = & \ensuremath{\Varid{alg}_{\Varid{RHS}}^{\Varid{ND}}\hsdot{\circ }{.}\Varid{fmap}\;\Varid{h_{L}}} \\
\ensuremath{\Varid{h_{L}}\hsdot{\circ }{.}\Varid{fwd_{S}}\hsdot{\circ }{.}\Conid{Inr}}& = & \ensuremath{\Varid{fwd}_{\Varid{RHS}}\hsdot{\circ }{.}\Varid{fmap}\;\Varid{h_{L}}}
\end{eqnarray*}

For the first of these, we calculate:
\indentbegin \begin{hscode}\SaveRestoreHook
\column{B}{@{}>{\hspre}l<{\hspost}@{}}%
\column{3}{@{}>{\hspre}l<{\hspost}@{}}%
\column{5}{@{}>{\hspre}l<{\hspost}@{}}%
\column{E}{@{}>{\hspre}l<{\hspost}@{}}%
\>[5]{}\Varid{h_{L}}\;(\Varid{fwd_{S}}\;(\Conid{Inl}\;\Varid{op})){}\<[E]%
\\
\>[3]{}\mathrel{=}\mbox{\commentbegin ~ definition of \ensuremath{\Varid{fwd_{S}}}  \commentend}{}\<[E]%
\\
\>[3]{}\hsindent{2}{}\<[5]%
\>[5]{}\Varid{h_{L}}\;(\lambda \Varid{s}\to \Conid{Op}\;(\Varid{fmap}\;(\mathbin{\$}\Varid{s})\;(\Conid{Inl}\;\Varid{op}))){}\<[E]%
\\
\>[3]{}\mathrel{=}\mbox{\commentbegin ~ definition of \ensuremath{\Varid{fmap}}  \commentend}{}\<[E]%
\\
\>[3]{}\hsindent{2}{}\<[5]%
\>[5]{}\Varid{h_{L}}\;(\lambda \Varid{s}\to \Conid{Op}\;(\Conid{Inl}\;(\Varid{fmap}\;(\mathbin{\$}\Varid{s})\;\Varid{op}))){}\<[E]%
\\
\>[3]{}\mathrel{=}\mbox{\commentbegin ~ definition of \ensuremath{\Varid{h_{L}}}  \commentend}{}\<[E]%
\\
\>[3]{}\hsindent{2}{}\<[5]%
\>[5]{}\Varid{fmap}\;(\Varid{fmap}\;(\Varid{fmap}\;\Varid{fst})\hsdot{\circ }{.}\Varid{h_{ND+f}})\;(\lambda \Varid{s}\to \Conid{Op}\;(\Conid{Inl}\;(\Varid{fmap}\;(\mathbin{\$}\Varid{s})\;\Varid{op}))){}\<[E]%
\\
\>[3]{}\mathrel{=}\mbox{\commentbegin ~ definition of \ensuremath{\Varid{fmap}}  \commentend}{}\<[E]%
\\
\>[3]{}\hsindent{2}{}\<[5]%
\>[5]{}\lambda \Varid{s}\to \Varid{fmap}\;(\Varid{fmap}\;\Varid{fst})\;(\Varid{h_{ND+f}}\;(\Conid{Op}\;(\Conid{Inl}\;(\Varid{fmap}\;(\mathbin{\$}\Varid{s})\;\Varid{op})))){}\<[E]%
\\
\>[3]{}\mathrel{=}\mbox{\commentbegin ~ definition of \ensuremath{\Varid{h_{ND+f}}}  \commentend}{}\<[E]%
\\
\>[3]{}\hsindent{2}{}\<[5]%
\>[5]{}\lambda \Varid{s}\to \Varid{fmap}\;(\Varid{fmap}\;\Varid{fst})\;(\Varid{alg_{ND+f}}\;(\Varid{fmap}\;\Varid{h_{ND+f}}\;(\Varid{fmap}\;(\mathbin{\$}\Varid{s})\;\Varid{op}))){}\<[E]%
\ColumnHook
\end{hscode}\resethooks
\indentend We split on \ensuremath{\Varid{op}}:

\noindent \mbox{\underline{case \ensuremath{\Varid{op}\mathrel{=}\Conid{Fail}}}}\indentbegin \begin{hscode}\SaveRestoreHook
\column{B}{@{}>{\hspre}l<{\hspost}@{}}%
\column{3}{@{}>{\hspre}l<{\hspost}@{}}%
\column{5}{@{}>{\hspre}l<{\hspost}@{}}%
\column{E}{@{}>{\hspre}l<{\hspost}@{}}%
\>[5]{}\lambda \Varid{s}\to \Varid{fmap}\;(\Varid{fmap}\;\Varid{fst})\;(\Varid{alg_{ND+f}}\;(\Varid{fmap}\;\Varid{h_{ND+f}}\;(\Varid{fmap}\;(\mathbin{\$}\Varid{s})\;\Conid{Fail}))){}\<[E]%
\\
\>[3]{}\mathrel{=}\mbox{\commentbegin ~ defintion of \ensuremath{\Varid{fmap}} (twice)  \commentend}{}\<[E]%
\\
\>[3]{}\hsindent{2}{}\<[5]%
\>[5]{}\lambda \Varid{s}\to \Varid{fmap}\;(\Varid{fmap}\;\Varid{fst})\;(\Varid{alg_{ND+f}}\;\Conid{Fail}){}\<[E]%
\\
\>[3]{}\mathrel{=}\mbox{\commentbegin ~ definition of \ensuremath{\Varid{alg_{ND+f}}}  \commentend}{}\<[E]%
\\
\>[3]{}\hsindent{2}{}\<[5]%
\>[5]{}\lambda \Varid{s}\to \Varid{fmap}\;(\Varid{fmap}\;\Varid{fst})\;(\Conid{Var}\;[\mskip1.5mu \mskip1.5mu]){}\<[E]%
\\
\>[3]{}\mathrel{=}\mbox{\commentbegin ~ definition of \ensuremath{\Varid{fmap}} (twice)  \commentend}{}\<[E]%
\\
\>[3]{}\hsindent{2}{}\<[5]%
\>[5]{}\lambda \Varid{s}\to \Conid{Var}\;[\mskip1.5mu \mskip1.5mu]{}\<[E]%
\\
\>[3]{}\mathrel{=}\mbox{\commentbegin ~ define \ensuremath{\Varid{alg}_{\Varid{RHS}}^{\Varid{ND}}\;\Conid{Fail}\mathrel{=}\lambda \Varid{s}\to \Conid{Var}\;[\mskip1.5mu \mskip1.5mu]}  \commentend}{}\<[E]%
\\
\>[3]{}\hsindent{2}{}\<[5]%
\>[5]{}\Varid{alg}_{\Varid{RHS}}^{\Varid{ND}}\;\Conid{Fail}{}\<[E]%
\\
\>[3]{}\mathrel{=}\mbox{\commentbegin  definition fo \ensuremath{\Varid{fmap}}  \commentend}{}\<[E]%
\\
\>[3]{}\hsindent{2}{}\<[5]%
\>[5]{}\Varid{alg}_{\Varid{RHS}}^{\Varid{ND}}\;(\Varid{fmap}\;\Varid{h_{L}}\;\Varid{fail}){}\<[E]%
\ColumnHook
\end{hscode}\resethooks
\indentend \noindent \mbox{\underline{case \ensuremath{\Varid{op}\mathrel{=}\Conid{Or}\;\Varid{p}\;\Varid{q}}}}\indentbegin \begin{hscode}\SaveRestoreHook
\column{B}{@{}>{\hspre}l<{\hspost}@{}}%
\column{3}{@{}>{\hspre}l<{\hspost}@{}}%
\column{5}{@{}>{\hspre}l<{\hspost}@{}}%
\column{E}{@{}>{\hspre}l<{\hspost}@{}}%
\>[5]{}\lambda \Varid{s}\to \Varid{fmap}\;(\Varid{fmap}\;\Varid{fst})\;(\Varid{alg_{ND+f}}\;(\Varid{fmap}\;\Varid{h_{ND+f}}\;(\Varid{fmap}\;(\mathbin{\$}\Varid{s})\;(\Conid{Or}\;\Varid{p}\;\Varid{q})))){}\<[E]%
\\
\>[3]{}\mathrel{=}\mbox{\commentbegin ~ defintion of \ensuremath{\Varid{fmap}} (twice)  \commentend}{}\<[E]%
\\
\>[3]{}\hsindent{2}{}\<[5]%
\>[5]{}\lambda \Varid{s}\to \Varid{fmap}\;(\Varid{fmap}\;\Varid{fst})\;(\Varid{alg_{ND+f}}\;(\Conid{Or}\;(\Varid{h_{ND+f}}\;(\Varid{p}\;\Varid{s}))\;(\Varid{h_{ND+f}}\;(\Varid{q}\;\Varid{s})))){}\<[E]%
\\
\>[3]{}\mathrel{=}\mbox{\commentbegin ~ definition of \ensuremath{\Varid{alg_{ND+f}}}  \commentend}{}\<[E]%
\\
\>[3]{}\hsindent{2}{}\<[5]%
\>[5]{}\lambda \Varid{s}\to \Varid{fmap}\;(\Varid{fmap}\;\Varid{fst})\;(\Varid{liftM2}\;(+\!\!+)\;(\Varid{h_{ND+f}}\;(\Varid{p}\;\Varid{s}))\;(\Varid{h_{ND+f}}\;(\Varid{q}\;\Varid{s}))){}\<[E]%
\\
\>[3]{}\mathrel{=}\mbox{\commentbegin ~ Lemma~\ref{eq:liftM2-fst-comm}  \commentend}{}\<[E]%
\\
\>[3]{}\hsindent{2}{}\<[5]%
\>[5]{}\lambda \Varid{s}\to \Varid{liftM2}\;(+\!\!+)\;(\Varid{fmap}\;(\Varid{fmap}\;\Varid{fst})\;(\Varid{h_{ND+f}}\;(\Varid{p}\;\Varid{s})))\;(\Varid{fmap}\;(\Varid{fmap}\;\Varid{fst})\;(\Varid{h_{ND+f}}\;(\Varid{q}\;\Varid{s}))){}\<[E]%
\\
\>[3]{}\mathrel{=}\mbox{\commentbegin ~ define \ensuremath{\Varid{alg}_{\Varid{RHS}}^{\Varid{ND}}\;(\Conid{Or}\;\Varid{p}\;\Varid{q})\mathrel{=}\lambda \Varid{s}\to \Varid{liftM2}\;(+\!\!+)\;(\Varid{p}\;\Varid{s})\;(\Varid{q}\;\Varid{s})}  \commentend}{}\<[E]%
\\
\>[3]{}\hsindent{2}{}\<[5]%
\>[5]{}\Varid{alg}_{\Varid{RHS}}^{\Varid{ND}}\;(\Conid{Or}\;(\Varid{fmap}\;(\Varid{fmap}\;\Varid{fst})\hsdot{\circ }{.}\Varid{h_{ND+f}}\hsdot{\circ }{.}\Varid{p})\;(\Varid{fmap}\;(\Varid{fmap}\;\Varid{fst})\hsdot{\circ }{.}\Varid{h_{ND+f}}\hsdot{\circ }{.}\Varid{q})){}\<[E]%
\\
\>[3]{}\mathrel{=}\mbox{\commentbegin ~ defintion of \ensuremath{\Varid{fmap}} (twice)  \commentend}{}\<[E]%
\\
\>[3]{}\hsindent{2}{}\<[5]%
\>[5]{}\Varid{alg}_{\Varid{RHS}}^{\Varid{ND}}\;(\Varid{fmap}\;(\Varid{fmap}\;(\Varid{fmap}\;(\Varid{fmap}\;\Varid{fst})\hsdot{\circ }{.}\Varid{h_{ND+f}}))\;(\Conid{Or}\;\Varid{p}\;\Varid{q})){}\<[E]%
\\
\>[3]{}\mathrel{=}\mbox{\commentbegin ~ defintion of \ensuremath{\Varid{h_{L}}}  \commentend}{}\<[E]%
\\
\>[3]{}\hsindent{2}{}\<[5]%
\>[5]{}\Varid{alg}_{\Varid{RHS}}^{\Varid{ND}}\;(\Varid{fmap}\;\Varid{h_{L}}\;(\Conid{Or}\;\Varid{p}\;\Varid{q})){}\<[E]%
\ColumnHook
\end{hscode}\resethooks
\indentend From this we conclude that the definition of \ensuremath{\Varid{alg}_{\Varid{RHS}}^{\Varid{ND}}} should be:
\indentbegin \begin{hscode}\SaveRestoreHook
\column{B}{@{}>{\hspre}l<{\hspost}@{}}%
\column{3}{@{}>{\hspre}l<{\hspost}@{}}%
\column{22}{@{}>{\hspre}l<{\hspost}@{}}%
\column{E}{@{}>{\hspre}l<{\hspost}@{}}%
\>[3]{}\Varid{alg}_{\Varid{RHS}}^{\Varid{ND}}\mathbin{::}\Conid{Functor}\;\Varid{f}\Rightarrow \Varid{Nondet_{F}}\;(\Varid{s}\to \Conid{Free}\;\Varid{f}\;[\mskip1.5mu \Varid{a}\mskip1.5mu])\to (\Varid{s}\to \Conid{Free}\;\Varid{f}\;[\mskip1.5mu \Varid{a}\mskip1.5mu]){}\<[E]%
\\
\>[3]{}\Varid{alg}_{\Varid{RHS}}^{\Varid{ND}}\;\Conid{Fail}{}\<[22]%
\>[22]{}\mathrel{=}\lambda \Varid{s}\to \Conid{Var}\;[\mskip1.5mu \mskip1.5mu]{}\<[E]%
\\
\>[3]{}\Varid{alg}_{\Varid{RHS}}^{\Varid{ND}}\;(\Conid{Or}\;\Varid{p}\;\Varid{q}){}\<[22]%
\>[22]{}\mathrel{=}\lambda \Varid{s}\to \Varid{liftM2}\;(+\!\!+)\;(\Varid{p}\;\Varid{s})\;(\Varid{q}\;\Varid{s}){}\<[E]%
\ColumnHook
\end{hscode}\resethooks
\indentend For the last subcondition, we calculate:
\indentbegin \begin{hscode}\SaveRestoreHook
\column{B}{@{}>{\hspre}l<{\hspost}@{}}%
\column{3}{@{}>{\hspre}l<{\hspost}@{}}%
\column{5}{@{}>{\hspre}l<{\hspost}@{}}%
\column{E}{@{}>{\hspre}l<{\hspost}@{}}%
\>[5]{}\Varid{h_{L}}\;(\Varid{fwd_{S}}\;(\Conid{Inr}\;\Varid{op})){}\<[E]%
\\
\>[3]{}\mathrel{=}\mbox{\commentbegin ~ definition of \ensuremath{\Varid{fwd_{S}}}  \commentend}{}\<[E]%
\\
\>[3]{}\hsindent{2}{}\<[5]%
\>[5]{}\Varid{h_{L}}\;(\lambda \Varid{s}\to \Conid{Op}\;(\Varid{fmap}\;(\mathbin{\$}\Varid{s})\;(\Conid{Inr}\;\Varid{op}))){}\<[E]%
\\
\>[3]{}\mathrel{=}\mbox{\commentbegin ~ definition of \ensuremath{\Varid{fmap}}  \commentend}{}\<[E]%
\\
\>[3]{}\hsindent{2}{}\<[5]%
\>[5]{}\Varid{h_{L}}\;(\lambda \Varid{s}\to \Conid{Op}\;(\Conid{Inr}\;(\Varid{fmap}\;(\mathbin{\$}\Varid{s})\;\Varid{op}))){}\<[E]%
\\
\>[3]{}\mathrel{=}\mbox{\commentbegin ~ definition of \ensuremath{\Varid{h_{L}}}  \commentend}{}\<[E]%
\\
\>[3]{}\hsindent{2}{}\<[5]%
\>[5]{}\Varid{fmap}\;(\Varid{fmap}\;(\Varid{fmap}\;\Varid{fst})\hsdot{\circ }{.}\Varid{h_{ND+f}})\;(\lambda \Varid{s}\to \Conid{Op}\;(\Conid{Inr}\;(\Varid{fmap}\;(\mathbin{\$}\Varid{s})\;\Varid{op}))){}\<[E]%
\\
\>[3]{}\mathrel{=}\mbox{\commentbegin ~ definition of \ensuremath{\Varid{fmap}}  \commentend}{}\<[E]%
\\
\>[3]{}\hsindent{2}{}\<[5]%
\>[5]{}\lambda \Varid{s}\to \Varid{fmap}\;(\Varid{fmap}\;\Varid{fst})\;(\Varid{h_{ND+f}}\;(\Conid{Op}\;(\Conid{Inr}\;(\Varid{fmap}\;(\mathbin{\$}\Varid{s})\;\Varid{op})))){}\<[E]%
\\
\>[3]{}\mathrel{=}\mbox{\commentbegin ~ definition of \ensuremath{\Varid{h_{ND+f}}}  \commentend}{}\<[E]%
\\
\>[3]{}\hsindent{2}{}\<[5]%
\>[5]{}\lambda \Varid{s}\to \Varid{fmap}\;(\Varid{fmap}\;\Varid{fst})\;(\Varid{fwd_{ND+f}}\;(\Varid{fmap}\;\Varid{h_{ND+f}}\;(\Varid{fmap}\;(\mathbin{\$}\Varid{s})\;\Varid{op}))){}\<[E]%
\\
\>[3]{}\mathrel{=}\mbox{\commentbegin ~ definition of \ensuremath{\Varid{fwd_{ND+f}}}  \commentend}{}\<[E]%
\\
\>[3]{}\hsindent{2}{}\<[5]%
\>[5]{}\lambda \Varid{s}\to \Varid{fmap}\;(\Varid{fmap}\;\Varid{fst})\;(\Conid{Op}\;(\Varid{fmap}\;\Varid{h_{ND+f}}\;(\Varid{fmap}\;(\mathbin{\$}\Varid{s})\;\Varid{op}))){}\<[E]%
\\
\>[3]{}\mathrel{=}\mbox{\commentbegin ~ definition of \ensuremath{\Varid{fmap}}  \commentend}{}\<[E]%
\\
\>[3]{}\hsindent{2}{}\<[5]%
\>[5]{}\lambda \Varid{s}\to \Conid{Op}\;(\Varid{fmap}\;(\Varid{fmap}\;(\Varid{fmap}\;\Varid{fst}))\;(\Varid{fmap}\;\Varid{h_{ND+f}}\;(\Varid{fmap}\;(\mathbin{\$}\Varid{s})\;\Varid{op}))){}\<[E]%
\\
\>[3]{}\mathrel{=}\mbox{\commentbegin ~ definition of \ensuremath{\Varid{h_{L}}}  \commentend}\mid {}\<[E]%
\\
\>[3]{}\hsindent{2}{}\<[5]%
\>[5]{}\lambda \Varid{s}\to \Conid{Op}\;(\Varid{h_{L}}\;(\Varid{fmap}\;(\mathbin{\$}\Varid{s})\;\Varid{op})){}\<[E]%
\\
\>[3]{}\mathrel{=}\mbox{\commentbegin ~ \Cref{eq:comm-app-fmap}  \commentend}{}\<[E]%
\\
\>[3]{}\hsindent{2}{}\<[5]%
\>[5]{}\lambda \Varid{s}\to \Conid{Op}\;(\Varid{fmap}\;(\mathbin{\$}\Varid{s})\;(\Varid{fmap}\;\Varid{h_{L}}\;\Varid{op})){}\<[E]%
\\
\>[3]{}\mathrel{=}\mbox{\commentbegin ~ define \ensuremath{\Varid{fwd}_{\Varid{RHS}}\;\Varid{op}\mathrel{=}\lambda \Varid{s}\to \Conid{Op}\;(\Varid{fmap}\;(\mathbin{\$}\Varid{s})\;\Varid{op})}  \commentend}{}\<[E]%
\\
\>[3]{}\hsindent{2}{}\<[5]%
\>[5]{}\Varid{fwd}_{\Varid{RHS}}\;(\Varid{fmap}\;\Varid{h_{L}}\;\Varid{op}){}\<[E]%
\ColumnHook
\end{hscode}\resethooks
\indentend From this we conclude that the definition of \ensuremath{\Varid{fwd}_{\Varid{RHS}}} should be:
\indentbegin \begin{hscode}\SaveRestoreHook
\column{B}{@{}>{\hspre}l<{\hspost}@{}}%
\column{3}{@{}>{\hspre}l<{\hspost}@{}}%
\column{E}{@{}>{\hspre}l<{\hspost}@{}}%
\>[3]{}\Varid{fwd}_{\Varid{RHS}}\mathbin{::}\Conid{Functor}\;\Varid{f}\Rightarrow \Varid{f}\;(\Varid{s}\to \Conid{Free}\;\Varid{f}\;[\mskip1.5mu \Varid{a}\mskip1.5mu])\to (\Varid{s}\to \Conid{Free}\;\Varid{f}\;[\mskip1.5mu \Varid{a}\mskip1.5mu]){}\<[E]%
\\
\>[3]{}\Varid{fwd}_{\Varid{RHS}}\;\Varid{op}\mathrel{=}\lambda \Varid{s}\to \Conid{Op}\;(\Varid{fmap}\;(\mathbin{\$}\Varid{s})\;\Varid{op}){}\<[E]%
\ColumnHook
\end{hscode}\resethooks
\indentend \subsection{Fusing the Left-Hand Side}

We proceed in the same fashion with fusing left-hand side,
discovering the definitions that we need to satisfy the fusion
condition.

We calculate as follows:
\indentbegin \begin{hscode}\SaveRestoreHook
\column{B}{@{}>{\hspre}l<{\hspost}@{}}%
\column{5}{@{}>{\hspre}l<{\hspost}@{}}%
\column{7}{@{}>{\hspre}l<{\hspost}@{}}%
\column{9}{@{}>{\hspre}l<{\hspost}@{}}%
\column{29}{@{}>{\hspre}l<{\hspost}@{}}%
\column{E}{@{}>{\hspre}l<{\hspost}@{}}%
\>[5]{}\Varid{h_{Global}}\hsdot{\circ }{.}\Varid{local2global}{}\<[E]%
\\
\>[B]{}\mathrel{=}\mbox{\commentbegin ~  definition of \ensuremath{\Varid{local2global}}  \commentend}{}\<[E]%
\\
\>[B]{}\hsindent{5}{}\<[5]%
\>[5]{}\Varid{h_{Global}}\hsdot{\circ }{.}\Varid{fold}\;\Conid{Var}\;\Varid{alg}{}\<[E]%
\\
\>[5]{}\hsindent{2}{}\<[7]%
\>[7]{}\mathbf{where}{}\<[E]%
\\
\>[7]{}\hsindent{2}{}\<[9]%
\>[9]{}\Varid{alg}\;(\Conid{Inl}\;(\Conid{Put}\;\Varid{t}\;\Varid{k}))\mathrel{=}\Varid{put_{R}}\;\Varid{t}>\!\!>\Varid{k}{}\<[E]%
\\
\>[7]{}\hsindent{2}{}\<[9]%
\>[9]{}\Varid{alg}\;\Varid{p}{}\<[29]%
\>[29]{}\mathrel{=}\Conid{Op}\;\Varid{p}{}\<[E]%
\\
\>[B]{}\mathrel{=}\mbox{\commentbegin ~  fold fusion-post (Equation \ref{eq:fusion-post})   \commentend}{}\<[E]%
\\
\>[B]{}\hsindent{5}{}\<[5]%
\>[5]{}\Varid{fold}\;\Varid{gen}_{\Varid{LHS}}\;(\Varid{alg}_{\Varid{LHS}}^{\Varid{S}}\mathbin{\#}\Varid{alg}_{\Varid{LHS}}^{\Varid{ND}}\mathbin{\#}\Varid{fwd}_{\Varid{LHS}}){}\<[E]%
\ColumnHook
\end{hscode}\resethooks
\indentend 
This last step is valid provided that the fusion conditions are satisfied:
\begin{eqnarray*}
\ensuremath{\Varid{h_{Global}}\hsdot{\circ }{.}\Conid{Var}} & = & \ensuremath{\Varid{gen}_{\Varid{LHS}}} \\
\ensuremath{\Varid{h_{Global}}\hsdot{\circ }{.}\Varid{alg}} & = & \ensuremath{(\Varid{alg}_{\Varid{LHS}}^{\Varid{S}}\mathbin{\#}\Varid{alg}_{\Varid{LHS}}^{\Varid{ND}}\mathbin{\#}\Varid{fwd}_{\Varid{LHS}})\hsdot{\circ }{.}\Varid{fmap}\;\Varid{h_{Global}}}
\end{eqnarray*}

We calculate for the first fusion condition:\indentbegin \begin{hscode}\SaveRestoreHook
\column{B}{@{}>{\hspre}l<{\hspost}@{}}%
\column{3}{@{}>{\hspre}l<{\hspost}@{}}%
\column{5}{@{}>{\hspre}l<{\hspost}@{}}%
\column{E}{@{}>{\hspre}l<{\hspost}@{}}%
\>[5]{}\Varid{h_{Global}}\;(\Conid{Var}\;\Varid{x}){}\<[E]%
\\
\>[3]{}\mathrel{=}\mbox{\commentbegin ~ definition of \ensuremath{\Varid{h_{Global}}}  \commentend}{}\<[E]%
\\
\>[3]{}\hsindent{2}{}\<[5]%
\>[5]{}\Varid{fmap}\;(\Varid{fmap}\;\Varid{fst})\;(\Varid{h_{State1}}\;(\Varid{h_{ND+f}}\;(\Varid{(\Leftrightarrow)}\;(\Conid{Var}\;\Varid{x})))){}\<[E]%
\\
\>[3]{}\mathrel{=}\mbox{\commentbegin ~ definition of \ensuremath{\Varid{(\Leftrightarrow)}}  \commentend}{}\<[E]%
\\
\>[3]{}\hsindent{2}{}\<[5]%
\>[5]{}\Varid{fmap}\;(\Varid{fmap}\;\Varid{fst})\;(\Varid{h_{State1}}\;(\Varid{h_{ND+f}}\;(\Conid{Var}\;\Varid{x}))){}\<[E]%
\\
\>[3]{}\mathrel{=}\mbox{\commentbegin ~ definition of \ensuremath{\Varid{h_{ND+f}}}  \commentend}{}\<[E]%
\\
\>[3]{}\hsindent{2}{}\<[5]%
\>[5]{}\Varid{fmap}\;(\Varid{fmap}\;\Varid{fst})\;(\Varid{h_{State1}}\;(\Conid{Var}\;[\mskip1.5mu \Varid{x}\mskip1.5mu])){}\<[E]%
\\
\>[3]{}\mathrel{=}\mbox{\commentbegin ~ definition of \ensuremath{\Varid{h_{State1}}}  \commentend}{}\<[E]%
\\
\>[3]{}\hsindent{2}{}\<[5]%
\>[5]{}\Varid{fmap}\;(\Varid{fmap}\;\Varid{fst})\;(\lambda \Varid{s}\to \Conid{Var}\;([\mskip1.5mu \Varid{x}\mskip1.5mu],\Varid{s})){}\<[E]%
\\
\>[3]{}\mathrel{=}\mbox{\commentbegin ~ definition of \ensuremath{\Varid{fmap}} (twice)  \commentend}{}\<[E]%
\\
\>[3]{}\hsindent{2}{}\<[5]%
\>[5]{}\lambda \Varid{s}\to \Conid{Var}\;[\mskip1.5mu \Varid{x}\mskip1.5mu]{}\<[E]%
\\
\>[3]{}\mathrel{=}\mbox{\commentbegin ~ define \ensuremath{\Varid{gen}_{\Varid{LHS}}\;\Varid{x}\mathrel{=}\lambda \Varid{s}\to \Conid{Var}\;[\mskip1.5mu \Varid{x}\mskip1.5mu]}  \commentend}{}\<[E]%
\\
\>[3]{}\hsindent{2}{}\<[5]%
\>[5]{}\Varid{gen}_{\Varid{LHS}}\;\Varid{x}{}\<[E]%
\ColumnHook
\end{hscode}\resethooks
\indentend We conclude that the first fusion condition is satisfied by:
\indentbegin \begin{hscode}\SaveRestoreHook
\column{B}{@{}>{\hspre}l<{\hspost}@{}}%
\column{3}{@{}>{\hspre}l<{\hspost}@{}}%
\column{E}{@{}>{\hspre}l<{\hspost}@{}}%
\>[3]{}\Varid{gen}_{\Varid{LHS}}\mathbin{::}\Conid{Functor}\;\Varid{f}\Rightarrow \Varid{a}\to (\Varid{s}\to \Conid{Free}\;\Varid{f}\;[\mskip1.5mu \Varid{a}\mskip1.5mu]){}\<[E]%
\\
\>[3]{}\Varid{gen}_{\Varid{LHS}}\;\Varid{x}\mathrel{=}\lambda \Varid{s}\to \Conid{Var}\;[\mskip1.5mu \Varid{x}\mskip1.5mu]{}\<[E]%
\ColumnHook
\end{hscode}\resethooks
\indentend We can split the second fusion condition in three subconditions:
\begin{eqnarray*}
\ensuremath{\Varid{h_{Global}}\hsdot{\circ }{.}\Varid{alg}\hsdot{\circ }{.}\Conid{Inl}} & = & \ensuremath{\Varid{alg}_{\Varid{LHS}}^{\Varid{S}}\hsdot{\circ }{.}\Varid{fmap}\;\Varid{h_{Global}}} \\
\ensuremath{\Varid{h_{Global}}\hsdot{\circ }{.}\Varid{alg}\hsdot{\circ }{.}\Conid{Inr}\hsdot{\circ }{.}\Conid{Inl}} & = & \ensuremath{\Varid{alg}_{\Varid{LHS}}^{\Varid{ND}}\hsdot{\circ }{.}\Varid{fmap}\;\Varid{h_{Global}}} \\
\ensuremath{\Varid{h_{Global}}\hsdot{\circ }{.}\Varid{alg}\hsdot{\circ }{.}\Conid{Inr}\hsdot{\circ }{.}\Conid{Inr}} & = & \ensuremath{\Varid{fwd}_{\Varid{LHS}}\hsdot{\circ }{.}\Varid{fmap}\;\Varid{h_{Global}}}
\end{eqnarray*}

Let's consider the first subconditions. It has two cases:

\noindent \mbox{\underline{case \ensuremath{\Varid{op}\mathrel{=}\Conid{Get}\;\Varid{k}}}}\indentbegin \begin{hscode}\SaveRestoreHook
\column{B}{@{}>{\hspre}l<{\hspost}@{}}%
\column{3}{@{}>{\hspre}l<{\hspost}@{}}%
\column{5}{@{}>{\hspre}l<{\hspost}@{}}%
\column{E}{@{}>{\hspre}l<{\hspost}@{}}%
\>[5]{}\Varid{h_{Global}}\;(\Varid{alg}\;(\Conid{Inl}\;(\Conid{Get}\;\Varid{k}))){}\<[E]%
\\
\>[3]{}\mathrel{=}\mbox{\commentbegin ~ definition of \ensuremath{\Varid{alg}}  \commentend}{}\<[E]%
\\
\>[3]{}\hsindent{2}{}\<[5]%
\>[5]{}\Varid{h_{Global}}\;(\Conid{Op}\;(\Conid{Inl}\;(\Conid{Get}\;\Varid{k}))){}\<[E]%
\\
\>[3]{}\mathrel{=}\mbox{\commentbegin ~ definition of \ensuremath{\Varid{h_{Global}}}  \commentend}{}\<[E]%
\\
\>[3]{}\hsindent{2}{}\<[5]%
\>[5]{}\Varid{fmap}\;(\Varid{fmap}\;\Varid{fst})\;(\Varid{h_{State1}}\;(\Varid{h_{ND+f}}\;(\Varid{(\Leftrightarrow)}\;(\Conid{Op}\;(\Conid{Inl}\;(\Conid{Get}\;\Varid{k})))))){}\<[E]%
\\
\>[3]{}\mathrel{=}\mbox{\commentbegin ~ definition of \ensuremath{\Varid{(\Leftrightarrow)}}  \commentend}{}\<[E]%
\\
\>[3]{}\hsindent{2}{}\<[5]%
\>[5]{}\Varid{fmap}\;(\Varid{fmap}\;\Varid{fst})\;(\Varid{h_{State1}}\;(\Varid{h_{ND+f}}\;(\Conid{Op}\;(\Conid{Inr}\;(\Conid{Inl}\;(\Varid{fmap}\;\Varid{(\Leftrightarrow)}\;(\Conid{Get}\;\Varid{k}))))))){}\<[E]%
\\
\>[3]{}\mathrel{=}\mbox{\commentbegin ~ definition of \ensuremath{\Varid{fmap}}  \commentend}{}\<[E]%
\\
\>[3]{}\hsindent{2}{}\<[5]%
\>[5]{}\Varid{fmap}\;(\Varid{fmap}\;\Varid{fst})\;(\Varid{h_{State1}}\;(\Varid{h_{ND+f}}\;(\Conid{Op}\;(\Conid{Inr}\;(\Conid{Inl}\;(\Conid{Get}\;(\Varid{(\Leftrightarrow)}\hsdot{\circ }{.}\Varid{k}))))))){}\<[E]%
\\
\>[3]{}\mathrel{=}\mbox{\commentbegin ~ definition of \ensuremath{\Varid{h_{ND+f}}}  \commentend}{}\<[E]%
\\
\>[3]{}\hsindent{2}{}\<[5]%
\>[5]{}\Varid{fmap}\;(\Varid{fmap}\;\Varid{fst})\;(\Varid{h_{State1}}\;(\Conid{Op}\;(\Varid{fmap}\;\Varid{h_{ND+f}}\;(\Conid{Inl}\;(\Conid{Get}\;(\Varid{(\Leftrightarrow)}\hsdot{\circ }{.}\Varid{k})))))){}\<[E]%
\\
\>[3]{}\mathrel{=}\mbox{\commentbegin ~ definition of \ensuremath{\Varid{fmap}}  \commentend}{}\<[E]%
\\
\>[3]{}\hsindent{2}{}\<[5]%
\>[5]{}\Varid{fmap}\;(\Varid{fmap}\;\Varid{fst})\;(\Varid{h_{State1}}\;(\Conid{Op}\;(\Conid{Inl}\;(\Conid{Get}\;(\Varid{h_{ND+f}}\hsdot{\circ }{.}\Varid{(\Leftrightarrow)}\hsdot{\circ }{.}\Varid{k}))))){}\<[E]%
\\
\>[3]{}\mathrel{=}\mbox{\commentbegin ~ definition of \ensuremath{\Varid{h_{State1}}}  \commentend}{}\<[E]%
\\
\>[3]{}\hsindent{2}{}\<[5]%
\>[5]{}\Varid{fmap}\;(\Varid{fmap}\;\Varid{fst})\;(\lambda \Varid{s}\to (\Varid{h_{State1}}\hsdot{\circ }{.}\Varid{h_{ND+f}}\hsdot{\circ }{.}\Varid{(\Leftrightarrow)}\hsdot{\circ }{.}\Varid{k})\;\Varid{s}\;\Varid{s}){}\<[E]%
\\
\>[3]{}\mathrel{=}\mbox{\commentbegin ~ definition of \ensuremath{\Varid{fmap}}  \commentend}{}\<[E]%
\\
\>[3]{}\hsindent{2}{}\<[5]%
\>[5]{}(\lambda \Varid{s}\to \Varid{fmap}\;\Varid{fst}\;((\Varid{h_{State1}}\hsdot{\circ }{.}\Varid{h_{ND+f}}\hsdot{\circ }{.}\Varid{(\Leftrightarrow)}\hsdot{\circ }{.}\Varid{k})\;\Varid{s}\;\Varid{s})){}\<[E]%
\\
\>[3]{}\mathrel{=}\mbox{\commentbegin ~ definition of \ensuremath{\Varid{fmap}}  \commentend}{}\<[E]%
\\
\>[3]{}\hsindent{2}{}\<[5]%
\>[5]{}(\lambda \Varid{s}\to ((\Varid{fmap}\;(\Varid{fmap}\;\Varid{fst})\hsdot{\circ }{.}\Varid{h_{State1}}\hsdot{\circ }{.}\Varid{h_{ND+f}}\hsdot{\circ }{.}\Varid{(\Leftrightarrow)}\hsdot{\circ }{.}\Varid{k})\;\Varid{s}\;\Varid{s})){}\<[E]%
\\
\>[3]{}\mathrel{=}\mbox{\commentbegin ~ define \ensuremath{\Varid{alg}_{\Varid{LHS}}^{\Varid{S}}\;(\Conid{Get}\;\Varid{k})\mathrel{=}\lambda \Varid{s}\to \Varid{k}\;\Varid{s}\;\Varid{s}}  \commentend}{}\<[E]%
\\
\>[3]{}\hsindent{2}{}\<[5]%
\>[5]{}\Varid{alg}_{\Varid{LHS}}^{\Varid{S}}\;(\Conid{Get}\;(\Varid{h_{Global}}\hsdot{\circ }{.}\Varid{k})){}\<[E]%
\\
\>[3]{}\mathrel{=}\mbox{\commentbegin ~ definition of \ensuremath{\Varid{fmap}}  \commentend}{}\<[E]%
\\
\>[3]{}\hsindent{2}{}\<[5]%
\>[5]{}\Varid{alg}_{\Varid{LHS}}^{\Varid{S}}\;(\Varid{fmap}\;\Varid{h_{Global}}\;(\Conid{Get}\;\Varid{k})){}\<[E]%
\ColumnHook
\end{hscode}\resethooks
\indentend \noindent \mbox{\underline{case \ensuremath{\Varid{op}\mathrel{=}\Conid{Put}\;\Varid{s}\;\Varid{k}}}}
\indentbegin \begin{hscode}\SaveRestoreHook
\column{B}{@{}>{\hspre}l<{\hspost}@{}}%
\column{3}{@{}>{\hspre}l<{\hspost}@{}}%
\column{5}{@{}>{\hspre}l<{\hspost}@{}}%
\column{7}{@{}>{\hspre}l<{\hspost}@{}}%
\column{14}{@{}>{\hspre}l<{\hspost}@{}}%
\column{17}{@{}>{\hspre}l<{\hspost}@{}}%
\column{22}{@{}>{\hspre}l<{\hspost}@{}}%
\column{26}{@{}>{\hspre}l<{\hspost}@{}}%
\column{36}{@{}>{\hspre}l<{\hspost}@{}}%
\column{41}{@{}>{\hspre}l<{\hspost}@{}}%
\column{46}{@{}>{\hspre}l<{\hspost}@{}}%
\column{52}{@{}>{\hspre}l<{\hspost}@{}}%
\column{E}{@{}>{\hspre}l<{\hspost}@{}}%
\>[5]{}\Varid{h_{Global}}\;(\Varid{alg}\;(\Conid{Inl}\;(\Conid{Put}\;\Varid{s}\;\Varid{k}))){}\<[E]%
\\
\>[3]{}\mathrel{=}\mbox{\commentbegin ~ definition of \ensuremath{\Varid{alg}}  \commentend}{}\<[E]%
\\
\>[3]{}\hsindent{2}{}\<[5]%
\>[5]{}\Varid{h_{Global}}\;(\Varid{put_{R}}\;\Varid{s}>\!\!>\Varid{k}){}\<[E]%
\\
\>[3]{}\mathrel{=}\mbox{\commentbegin ~ definition of \ensuremath{\Varid{put_{R}}}  \commentend}{}\<[E]%
\\
\>[3]{}\hsindent{2}{}\<[5]%
\>[5]{}\Varid{h_{Global}}\;((\Varid{get}>\!\!>\!\!=\lambda \Varid{t}\to \Varid{put}\;\Varid{s}\mathbin{\talloblong}\Varid{side}\;(\Varid{put}\;\Varid{t}))>\!\!>\Varid{k}){}\<[E]%
\\
\>[3]{}\mathrel{=}\mbox{\commentbegin ~ definitions of \ensuremath{\Varid{side}}, \ensuremath{\Varid{get}}, \ensuremath{\Varid{put}}, \ensuremath{(\talloblong)}, \ensuremath{(>\!\!>\!\!=)}  \commentend}{}\<[E]%
\\
\>[3]{}\hsindent{2}{}\<[5]%
\>[5]{}\Varid{h_{Global}}\;(\Conid{Op}\;(\Conid{Inl}\;(\Conid{Get}\;(\lambda \Varid{t}\to \Conid{Op}\;(\Conid{Inr}\;(\Conid{Inl}\;(\Conid{Or}\;{}\<[52]%
\>[52]{}(\Conid{Op}\;(\Conid{Inl}\;(\Conid{Put}\;\Varid{s}\;\Varid{k})))\;{}\<[E]%
\\
\>[52]{}(\Conid{Op}\;(\Conid{Inl}\;(\Conid{Put}\;\Varid{t}\;(\Conid{Op}\;(\Conid{Inr}\;(\Conid{Inl}\;\Conid{Fail}))))))))))))){}\<[E]%
\\
\>[3]{}\mathrel{=}\mbox{\commentbegin ~ definition of \ensuremath{\Varid{h_{Global}}}  \commentend}{}\<[E]%
\\
\>[3]{}\hsindent{2}{}\<[5]%
\>[5]{}\Varid{fmap}\;(\Varid{fmap}\;\Varid{fst})\;(\Varid{h_{State1}}\;(\Varid{h_{ND+f}}\;(\Varid{(\Leftrightarrow)}{}\<[E]%
\\
\>[5]{}\hsindent{2}{}\<[7]%
\>[7]{}(\Conid{Op}\;(\Conid{Inl}\;(\Conid{Get}\;(\lambda \Varid{t}\to \Conid{Op}\;(\Conid{Inr}\;(\Conid{Inl}\;(\Conid{Or}\;{}\<[46]%
\>[46]{}(\Conid{Op}\;(\Conid{Inl}\;(\Conid{Put}\;\Varid{s}\;\Varid{k})))\;{}\<[E]%
\\
\>[46]{}(\Conid{Op}\;(\Conid{Inl}\;(\Conid{Put}\;\Varid{t}\;(\Conid{Op}\;\Conid{Inr}\;((\Conid{Inl}\;\Conid{Fail})))))))))))))))){}\<[E]%
\\
\>[3]{}\mathrel{=}\mbox{\commentbegin ~ definition of \ensuremath{\Varid{(\Leftrightarrow)}}  \commentend}{}\<[E]%
\\
\>[3]{}\hsindent{2}{}\<[5]%
\>[5]{}\Varid{fmap}\;(\Varid{fmap}\;\Varid{fst})\;(\Varid{h_{State1}}\;(\Varid{h_{ND+f}}\;({}\<[E]%
\\
\>[5]{}\hsindent{2}{}\<[7]%
\>[7]{}(\Conid{Op}\;(\Conid{Inr}\;(\Conid{Inl}\;(\Conid{Get}\;(\lambda \Varid{t}\to \Conid{Op}\;(\Conid{Inl}\;(\Conid{Or}\;{}\<[46]%
\>[46]{}(\Conid{Op}\;(\Conid{Inr}\;(\Conid{Inl}\;(\Conid{Put}\;\Varid{s}\;(\Varid{(\Leftrightarrow)}\;\Varid{k})))))\;{}\<[E]%
\\
\>[46]{}(\Conid{Op}\;(\Conid{Inr}\;(\Conid{Inl}\;(\Conid{Put}\;\Varid{t}\;(\Conid{Op}\;(\Conid{Inl}\;\Conid{Fail})))))))))))))))){}\<[E]%
\\
\>[3]{}\mathrel{=}\mbox{\commentbegin ~ definition of \ensuremath{\Varid{h_{ND+f}}}  \commentend}{}\<[E]%
\\
\>[3]{}\hsindent{2}{}\<[5]%
\>[5]{}\Varid{fmap}\;(\Varid{fmap}\;\Varid{fst})\;(\Varid{h_{State1}}\;({}\<[E]%
\\
\>[5]{}\hsindent{2}{}\<[7]%
\>[7]{}(\Conid{Op}\;(\Conid{Inl}\;(\Conid{Get}\;(\lambda \Varid{t}\to \Varid{liftM2}\;(+\!\!+)\;{}\<[41]%
\>[41]{}(\Conid{Op}\;(\Conid{Inl}\;(\Conid{Put}\;\Varid{s}\;(\Varid{h_{ND+f}}\;(\Varid{(\Leftrightarrow)}\;\Varid{k})))))\;{}\<[E]%
\\
\>[41]{}(\Conid{Op}\;(\Conid{Inl}\;(\Conid{Put}\;\Varid{t}\;(\Conid{Var}\;[\mskip1.5mu \mskip1.5mu])))))))))){}\<[E]%
\\
\>[3]{}\mathrel{=}\mbox{\commentbegin ~ definition of \ensuremath{\Varid{h_{State1}}}  \commentend}{}\<[E]%
\\
\>[3]{}\hsindent{2}{}\<[5]%
\>[5]{}\Varid{fmap}\;(\Varid{fmap}\;\Varid{fst})\;{}\<[E]%
\\
\>[5]{}\hsindent{2}{}\<[7]%
\>[7]{}(\lambda \Varid{t}\to \Varid{h_{State1}}\;(\Varid{liftM2}\;(+\!\!+)\;{}\<[36]%
\>[36]{}(\Conid{Op}\;(\Conid{Inl}\;(\Conid{Put}\;\Varid{s}\;(\Varid{h_{ND+f}}\;(\Varid{(\Leftrightarrow)}\;\Varid{k})))))\;{}\<[E]%
\\
\>[36]{}(\Conid{Op}\;(\Conid{Inl}\;(\Conid{Put}\;\Varid{t}\;(\Conid{Var}\;[\mskip1.5mu \mskip1.5mu])))))\;\Varid{t}){}\<[E]%
\\
\>[3]{}\mathrel{=}\mbox{\commentbegin ~ definition of \ensuremath{\Varid{liftM2}}  \commentend}{}\<[E]%
\\
\>[3]{}\hsindent{2}{}\<[5]%
\>[5]{}\Varid{fmap}\;(\Varid{fmap}\;\Varid{fst})\;{}\<[E]%
\\
\>[5]{}\hsindent{2}{}\<[7]%
\>[7]{}(\lambda \Varid{t}\to \Varid{h_{State1}}\;(\mathbf{do}\;\Varid{x}\leftarrow \Conid{Op}\;(\Conid{Inl}\;(\Conid{Put}\;\Varid{s}\;(\Varid{h_{ND+f}}\;(\Varid{(\Leftrightarrow)}\;\Varid{k})))){}\<[E]%
\\
\>[7]{}\hsindent{19}{}\<[26]%
\>[26]{}\Varid{y}\leftarrow \Conid{Op}\;(\Conid{Inl}\;(\Conid{Put}\;\Varid{t}\;(\Conid{Var}\;[\mskip1.5mu \mskip1.5mu]))){}\<[E]%
\\
\>[7]{}\hsindent{19}{}\<[26]%
\>[26]{}\Conid{Var}\;(\Varid{x}+\!\!+\Varid{y}){}\<[E]%
\\
\>[7]{}\hsindent{15}{}\<[22]%
\>[22]{})\;\Varid{t}){}\<[E]%
\\
\>[3]{}\mathrel{=}\mbox{\commentbegin ~ \Cref{lemma:dist-hState1}  \commentend}{}\<[E]%
\\
\>[3]{}\hsindent{2}{}\<[5]%
\>[5]{}\Varid{fmap}\;(\Varid{fmap}\;\Varid{fst})\;{}\<[E]%
\\
\>[5]{}\hsindent{2}{}\<[7]%
\>[7]{}(\lambda \Varid{t}\to \mathbf{do}\;(\Varid{x},\Varid{t}_{1})\leftarrow \Varid{h_{State1}}\;(\Conid{Op}\;(\Conid{Inl}\;(\Conid{Put}\;\Varid{s}\;(\Varid{h_{ND+f}}\;(\Varid{(\Leftrightarrow)}\;\Varid{k})))))\;\Varid{t}{}\<[E]%
\\
\>[7]{}\hsindent{10}{}\<[17]%
\>[17]{}(\Varid{y},\Varid{t}_{2})\leftarrow \Varid{h_{State1}}\;(\Conid{Op}\;(\Conid{Inl}\;(\Conid{Put}\;\Varid{t}\;(\Conid{Var}\;[\mskip1.5mu \mskip1.5mu]))))\;\Varid{t}_{1}{}\<[E]%
\\
\>[7]{}\hsindent{10}{}\<[17]%
\>[17]{}\Varid{h_{State1}}\;(\Conid{Var}\;(\Varid{x}+\!\!+\Varid{y}))\;\Varid{t}_{2}{}\<[E]%
\\
\>[5]{}\hsindent{2}{}\<[7]%
\>[7]{}){}\<[E]%
\\
\>[3]{}\mathrel{=}\mbox{\commentbegin ~ definition of \ensuremath{\Varid{h_{State1}}}  \commentend}{}\<[E]%
\\
\>[3]{}\hsindent{2}{}\<[5]%
\>[5]{}\Varid{fmap}\;(\Varid{fmap}\;\Varid{fst})\;{}\<[E]%
\\
\>[5]{}\hsindent{2}{}\<[7]%
\>[7]{}(\lambda \Varid{t}\to \mathbf{do}\;(\Varid{x},\anonymous )\leftarrow \Varid{h_{State1}}\;(\Varid{h_{ND+f}}\;(\Varid{(\Leftrightarrow)}\;\Varid{k}))\;\Varid{s}{}\<[E]%
\\
\>[7]{}\hsindent{10}{}\<[17]%
\>[17]{}(\Varid{y},\Varid{t}_{2})\leftarrow \Conid{Var}\;([\mskip1.5mu \mskip1.5mu],\Varid{t}){}\<[E]%
\\
\>[7]{}\hsindent{10}{}\<[17]%
\>[17]{}\Conid{Var}\;(\Varid{x}+\!\!+\Varid{y},\Varid{t}_{2}){}\<[E]%
\\
\>[5]{}\hsindent{2}{}\<[7]%
\>[7]{}){}\<[E]%
\\
\>[3]{}\mathrel{=}\mbox{\commentbegin ~ monad law  \commentend}{}\<[E]%
\\
\>[3]{}\hsindent{2}{}\<[5]%
\>[5]{}\Varid{fmap}\;(\Varid{fmap}\;\Varid{fst})\;{}\<[E]%
\\
\>[5]{}\hsindent{2}{}\<[7]%
\>[7]{}(\lambda \Varid{t}\to \mathbf{do}\;(\Varid{x},\anonymous )\leftarrow \Varid{h_{State1}}\;(\Varid{h_{ND+f}}\;(\Varid{(\Leftrightarrow)}\;\Varid{k}))\;\Varid{s}{}\<[E]%
\\
\>[7]{}\hsindent{10}{}\<[17]%
\>[17]{}\Conid{Var}\;(\Varid{x}+\!\!+[\mskip1.5mu \mskip1.5mu],\Varid{t}){}\<[E]%
\\
\>[5]{}\hsindent{2}{}\<[7]%
\>[7]{}){}\<[E]%
\\
\>[3]{}\mathrel{=}\mbox{\commentbegin ~ right unit of \ensuremath{(+\!\!+)}  \commentend}{}\<[E]%
\\
\>[3]{}\hsindent{2}{}\<[5]%
\>[5]{}\Varid{fmap}\;(\Varid{fmap}\;\Varid{fst})\;{}\<[E]%
\\
\>[5]{}\hsindent{2}{}\<[7]%
\>[7]{}(\lambda \Varid{t}\to \mathbf{do}\;(\Varid{x},\anonymous )\leftarrow \Varid{h_{State1}}\;(\Varid{h_{ND+f}}\;(\Varid{(\Leftrightarrow)}\;\Varid{k}))\;\Varid{s}{}\<[E]%
\\
\>[7]{}\hsindent{10}{}\<[17]%
\>[17]{}\Conid{Var}\;(\Varid{x},\Varid{t}){}\<[E]%
\\
\>[5]{}\hsindent{2}{}\<[7]%
\>[7]{}){}\<[E]%
\\
\>[3]{}\mathrel{=}\mbox{\commentbegin ~ definition of \ensuremath{\Varid{fmap}\;\Varid{fst}}  \commentend}{}\<[E]%
\\
\>[3]{}\hsindent{2}{}\<[5]%
\>[5]{}\Varid{fmap}\;(\Varid{fmap}\;\Varid{fst})\;{}\<[E]%
\\
\>[5]{}\hsindent{2}{}\<[7]%
\>[7]{}(\lambda \Varid{t}\to \mathbf{do}\;\Varid{x}\leftarrow \Varid{fmap}\;\Varid{fst}\;(\Varid{h_{State1}}\;(\Varid{h_{ND+f}}\;(\Varid{(\Leftrightarrow)}\;\Varid{k}))\;\Varid{s}){}\<[E]%
\\
\>[7]{}\hsindent{10}{}\<[17]%
\>[17]{}\Conid{Var}\;(\Varid{x},\Varid{t}){}\<[E]%
\\
\>[5]{}\hsindent{2}{}\<[7]%
\>[7]{}){}\<[E]%
\\
\>[3]{}\mathrel{=}\mbox{\commentbegin ~ definition of \ensuremath{\Varid{fmap}}  \commentend}{}\<[E]%
\\
\>[3]{}\hsindent{2}{}\<[5]%
\>[5]{}\Varid{fmap}\;(\Varid{fmap}\;\Varid{fst})\;{}\<[E]%
\\
\>[5]{}\hsindent{2}{}\<[7]%
\>[7]{}(\lambda \Varid{t}\to \mathbf{do}\;\Varid{x}\leftarrow (\Varid{fmap}\;(\Varid{fmap}\;\Varid{fst})\;(\Varid{h_{State1}}\;(\Varid{h_{ND+f}}\;(\Varid{(\Leftrightarrow)}\;\Varid{k}))))\;\Varid{s}{}\<[E]%
\\
\>[7]{}\hsindent{10}{}\<[17]%
\>[17]{}\Conid{Var}\;(\Varid{x},\Varid{t}){}\<[E]%
\\
\>[5]{}\hsindent{2}{}\<[7]%
\>[7]{}){}\<[E]%
\\
\>[3]{}\mathrel{=}\mbox{\commentbegin ~ definition of \ensuremath{\Varid{fmap}\;(\Varid{fmap}\;\Varid{fst})}  \commentend}{}\<[E]%
\\
\>[3]{}\hsindent{2}{}\<[5]%
\>[5]{}\mathbin{\char92 \char95 }\to \mathbf{do}\;\Varid{x}\leftarrow (\Varid{fmap}\;(\Varid{fmap}\;\Varid{fst})\;(\Varid{h_{State1}}\;(\Varid{h_{ND+f}}\;(\Varid{(\Leftrightarrow)}\;\Varid{k}))))\;\Varid{s}{}\<[E]%
\\
\>[5]{}\hsindent{9}{}\<[14]%
\>[14]{}\Conid{Var}\;\Varid{x}{}\<[E]%
\\
\>[3]{}\mathrel{=}\mbox{\commentbegin ~ monad law  \commentend}{}\<[E]%
\\
\>[3]{}\hsindent{2}{}\<[5]%
\>[5]{}\mathbin{\char92 \char95 }\to (\Varid{fmap}\;(\Varid{fmap}\;\Varid{fst})\;(\Varid{h_{State1}}\;(\Varid{h_{ND+f}}\;(\Varid{(\Leftrightarrow)}\;\Varid{k}))))\;\Varid{s}{}\<[E]%
\\
\>[3]{}\mathrel{=}\mbox{\commentbegin ~ definition of \ensuremath{\Varid{h_{Global}}}  \commentend}{}\<[E]%
\\
\>[3]{}\hsindent{2}{}\<[5]%
\>[5]{}\mathbin{\char92 \char95 }\to (\Varid{h_{Global}}\;\Varid{k})\;\Varid{s}{}\<[E]%
\\
\>[3]{}\mathrel{=}\mbox{\commentbegin ~ define \ensuremath{\Varid{alg}_{\Varid{LHS}}^{\Varid{S}}\;(\Conid{Put}\;\Varid{s}\;\Varid{k})\mathrel{=}\mathbin{\char92 \char95 }\to \Varid{k}\;\Varid{s}}  \commentend}{}\<[E]%
\\
\>[3]{}\hsindent{2}{}\<[5]%
\>[5]{}\Varid{alg}_{\Varid{LHS}}^{\Varid{S}}\;(\Conid{Put}\;\Varid{s}\;(\Varid{h_{Global}}\;\Varid{k})){}\<[E]%
\\
\>[3]{}\mathrel{=}\mbox{\commentbegin ~ definition of \ensuremath{\Varid{fmap}}  \commentend}{}\<[E]%
\\
\>[3]{}\hsindent{2}{}\<[5]%
\>[5]{}\Varid{alg}_{\Varid{LHS}}^{\Varid{S}}\;(\Varid{fmap}\;\Varid{h_{Global}}\;(\Conid{Put}\;\Varid{s})){}\<[E]%
\ColumnHook
\end{hscode}\resethooks
\indentend We conclude that this fusion subcondition holds provided that:
\indentbegin \begin{hscode}\SaveRestoreHook
\column{B}{@{}>{\hspre}l<{\hspost}@{}}%
\column{3}{@{}>{\hspre}l<{\hspost}@{}}%
\column{22}{@{}>{\hspre}c<{\hspost}@{}}%
\column{22E}{@{}l@{}}%
\column{25}{@{}>{\hspre}l<{\hspost}@{}}%
\column{E}{@{}>{\hspre}l<{\hspost}@{}}%
\>[3]{}\Varid{alg}_{\Varid{LHS}}^{\Varid{S}}\mathbin{::}\Conid{Functor}\;\Varid{f}\Rightarrow \Varid{State_{F}}\;\Varid{s}\;(\Varid{s}\to \Conid{Free}\;\Varid{f}\;[\mskip1.5mu \Varid{a}\mskip1.5mu])\to (\Varid{s}\to \Conid{Free}\;\Varid{f}\;[\mskip1.5mu \Varid{a}\mskip1.5mu]){}\<[E]%
\\
\>[3]{}\Varid{alg}_{\Varid{LHS}}^{\Varid{S}}\;(\Conid{Get}\;\Varid{k}){}\<[22]%
\>[22]{}\mathrel{=}{}\<[22E]%
\>[25]{}\lambda \Varid{s}\to \Varid{k}\;\Varid{s}\;\Varid{s}{}\<[E]%
\\
\>[3]{}\Varid{alg}_{\Varid{LHS}}^{\Varid{S}}\;(\Conid{Put}\;\Varid{s}\;\Varid{k}){}\<[22]%
\>[22]{}\mathrel{=}{}\<[22E]%
\>[25]{}\mathbin{\char92 \char95 }\to \Varid{k}\;\Varid{s}{}\<[E]%
\ColumnHook
\end{hscode}\resethooks
\indentend Let's consider the second subcondition. It has also two cases:

\noindent \mbox{\underline{case \ensuremath{\Varid{op}\mathrel{=}\Conid{Fail}}}}\indentbegin \begin{hscode}\SaveRestoreHook
\column{B}{@{}>{\hspre}l<{\hspost}@{}}%
\column{3}{@{}>{\hspre}l<{\hspost}@{}}%
\column{4}{@{}>{\hspre}l<{\hspost}@{}}%
\column{5}{@{}>{\hspre}l<{\hspost}@{}}%
\column{E}{@{}>{\hspre}l<{\hspost}@{}}%
\>[5]{}\Varid{h_{Global}}\;(\Varid{alg}\;(\Conid{Inr}\;(\Conid{Inl}\;\Conid{Fail}))){}\<[E]%
\\
\>[3]{}\mathrel{=}\mbox{\commentbegin ~ definition of \ensuremath{\Varid{alg}}  \commentend}{}\<[E]%
\\
\>[3]{}\hsindent{2}{}\<[5]%
\>[5]{}\Varid{h_{Global}}\;(\Conid{Op}\;(\Conid{Inr}\;(\Conid{Inl}\;\Conid{Fail}))){}\<[E]%
\\
\>[3]{}\mathrel{=}\mbox{\commentbegin ~ definition of \ensuremath{\Varid{h_{Global}}}  \commentend}{}\<[E]%
\\
\>[3]{}\hsindent{2}{}\<[5]%
\>[5]{}\Varid{fmap}\;(\Varid{fmap}\;\Varid{fst})\;(\Varid{h_{State1}}\;(\Varid{h_{ND+f}}\;(\Varid{(\Leftrightarrow)}\;(\Conid{Op}\;(\Conid{Inr}\;(\Conid{Inl}\;\Conid{Fail})))))){}\<[E]%
\\
\>[3]{}\mathrel{=}\mbox{\commentbegin ~ definition of \ensuremath{\Varid{(\Leftrightarrow)}}  \commentend}{}\<[E]%
\\
\>[3]{}\hsindent{2}{}\<[5]%
\>[5]{}\Varid{fmap}\;(\Varid{fmap}\;\Varid{fst})\;(\Varid{h_{State1}}\;(\Varid{h_{ND+f}}\;(\Conid{Op}\;(\Conid{Inl}\;(\Varid{fmap}\;\Varid{(\Leftrightarrow)}\;\Conid{Fail}))))){}\<[E]%
\\
\>[3]{}\mathrel{=}\mbox{\commentbegin ~ definition of \ensuremath{\Varid{fmap}}  \commentend}{}\<[E]%
\\
\>[3]{}\hsindent{2}{}\<[5]%
\>[5]{}\Varid{fmap}\;(\Varid{fmap}\;\Varid{fst})\;(\Varid{h_{State1}}\;(\Varid{h_{ND+f}}\;(\Conid{Op}\;(\Conid{Inl}\;\Conid{Fail})))){}\<[E]%
\\
\>[3]{}\mathrel{=}\mbox{\commentbegin ~ definition of \ensuremath{\Varid{h_{ND+f}}}  \commentend}{}\<[E]%
\\
\>[3]{}\hsindent{2}{}\<[5]%
\>[5]{}\Varid{fmap}\;(\Varid{fmap}\;\Varid{fst})\;(\Varid{h_{State1}}\;(\Conid{Var}\;[\mskip1.5mu \mskip1.5mu])){}\<[E]%
\\
\>[3]{}\mathrel{=}\mbox{\commentbegin ~ definition of \ensuremath{\Varid{h_{State1}}}  \commentend}{}\<[E]%
\\
\>[3]{}\hsindent{2}{}\<[5]%
\>[5]{}\Varid{fmap}\;(\Varid{fmap}\;\Varid{fst})\;(\lambda \Varid{s}\to \Conid{Var}\;([\mskip1.5mu \mskip1.5mu],\Varid{s})){}\<[E]%
\\
\>[3]{}\mathrel{=}\mbox{\commentbegin ~ definition of \ensuremath{\Varid{fmap}} twice and \ensuremath{\Varid{fst}}  \commentend}{}\<[E]%
\\
\>[3]{}\hsindent{2}{}\<[5]%
\>[5]{}\lambda \Varid{s}\to \Conid{Var}\;[\mskip1.5mu \mskip1.5mu]{}\<[E]%
\\
\>[3]{}\mathrel{=}\mbox{\commentbegin ~ define \ensuremath{\Varid{alg}_{\Varid{RHS}}^{\Varid{ND}}\;\Conid{Fail}\mathrel{=}\lambda \Varid{s}\to \Conid{Var}\;[\mskip1.5mu \mskip1.5mu]}  \commentend}{}\<[E]%
\\
\>[3]{}\hsindent{1}{}\<[4]%
\>[4]{}\Varid{alg}_{\Varid{RHS}}^{\Varid{ND}}\;\Conid{Fail}{}\<[E]%
\\
\>[3]{}\mathrel{=}\mbox{\commentbegin ~ definition of \ensuremath{\Varid{fmap}}  \commentend}{}\<[E]%
\\
\>[3]{}\hsindent{1}{}\<[4]%
\>[4]{}\Varid{alg}_{\Varid{RHS}}^{\Varid{ND}}\;(\Varid{fmap}\;\Varid{h_{Global}}\;\Conid{Fail}){}\<[E]%
\ColumnHook
\end{hscode}\resethooks
\indentend \noindent \mbox{\underline{case \ensuremath{\Varid{op}\mathrel{=}\Conid{Or}\;\Varid{p}\;\Varid{q}}}}\indentbegin \begin{hscode}\SaveRestoreHook
\column{B}{@{}>{\hspre}l<{\hspost}@{}}%
\column{3}{@{}>{\hspre}l<{\hspost}@{}}%
\column{5}{@{}>{\hspre}l<{\hspost}@{}}%
\column{16}{@{}>{\hspre}l<{\hspost}@{}}%
\column{33}{@{}>{\hspre}l<{\hspost}@{}}%
\column{35}{@{}>{\hspre}l<{\hspost}@{}}%
\column{E}{@{}>{\hspre}l<{\hspost}@{}}%
\>[5]{}\Varid{h_{Global}}\;(\Varid{alg}\;(\Conid{Inr}\;(\Conid{Inl}\;(\Conid{Or}\;\Varid{p}\;\Varid{q})))){}\<[E]%
\\
\>[3]{}\mathrel{=}\mbox{\commentbegin ~ definition of \ensuremath{\Varid{alg}}  \commentend}{}\<[E]%
\\
\>[3]{}\hsindent{2}{}\<[5]%
\>[5]{}\Varid{h_{Global}}\;(\Conid{Op}\;(\Conid{Inr}\;(\Conid{Inl}\;(\Conid{Or}\;\Varid{p}\;\Varid{q})))){}\<[E]%
\\
\>[3]{}\mathrel{=}\mbox{\commentbegin ~ definition of \ensuremath{\Varid{h_{Global}}}  \commentend}{}\<[E]%
\\
\>[3]{}\hsindent{2}{}\<[5]%
\>[5]{}\Varid{fmap}\;(\Varid{fmap}\;\Varid{fst})\;(\Varid{h_{State1}}\;(\Varid{h_{ND+f}}\;(\Varid{(\Leftrightarrow)}\;(\Conid{Op}\;(\Conid{Inr}\;(\Conid{Inl}\;(\Conid{Or}\;\Varid{p}\;\Varid{q}))))))){}\<[E]%
\\
\>[3]{}\mathrel{=}\mbox{\commentbegin ~ definition of \ensuremath{\Varid{(\Leftrightarrow)}}  \commentend}{}\<[E]%
\\
\>[3]{}\hsindent{2}{}\<[5]%
\>[5]{}\Varid{fmap}\;(\Varid{fmap}\;\Varid{fst})\;(\Varid{h_{State1}}\;(\Varid{h_{ND+f}}\;(\Conid{Op}\;(\Conid{Inl}\;(\Varid{fmap}\;\Varid{(\Leftrightarrow)}\;(\Conid{Or}\;\Varid{p}\;\Varid{q})))))){}\<[E]%
\\
\>[3]{}\mathrel{=}\mbox{\commentbegin ~ definition of \ensuremath{\Varid{fmap}}  \commentend}{}\<[E]%
\\
\>[3]{}\hsindent{2}{}\<[5]%
\>[5]{}\Varid{fmap}\;(\Varid{fmap}\;\Varid{fst})\;(\Varid{h_{State1}}\;(\Varid{h_{ND+f}}\;(\Conid{Op}\;(\Conid{Inl}\;(\Conid{Or}\;(\Varid{(\Leftrightarrow)}\;\Varid{p})\;(\Varid{(\Leftrightarrow)}\;\Varid{q})))))){}\<[E]%
\\
\>[3]{}\mathrel{=}\mbox{\commentbegin ~ definition of \ensuremath{\Varid{h_{ND+f}}}  \commentend}{}\<[E]%
\\
\>[3]{}\hsindent{2}{}\<[5]%
\>[5]{}\Varid{fmap}\;(\Varid{fmap}\;\Varid{fst})\;(\Varid{h_{State1}}\;(\Varid{liftM2}\;(+\!\!+)\;(\Varid{h_{ND+f}}\;(\Varid{(\Leftrightarrow)}\;\Varid{p}))\;(\Varid{h_{ND+f}}\;(\Varid{(\Leftrightarrow)}\;\Varid{q})))){}\<[E]%
\\
\>[3]{}\mathrel{=}\mbox{\commentbegin ~ definition of \ensuremath{\Varid{liftM2}}  \commentend}{}\<[E]%
\\
\>[3]{}\hsindent{2}{}\<[5]%
\>[5]{}\Varid{fmap}\;(\Varid{fmap}\;\Varid{fst})\;(\Varid{h_{State1}}\;(\mathbf{do}\;{}\<[35]%
\>[35]{}\Varid{x}\leftarrow \Varid{h_{ND+f}}\;(\Varid{(\Leftrightarrow)}\;\Varid{p}){}\<[E]%
\\
\>[35]{}\Varid{y}\leftarrow \Varid{h_{ND+f}}\;(\Varid{(\Leftrightarrow)}\;\Varid{q}){}\<[E]%
\\
\>[35]{}\Varid{\eta}\;(\Varid{x}+\!\!+\Varid{y}))){}\<[E]%
\\
\>[3]{}\mathrel{=}\mbox{\commentbegin ~ \Cref{lemma:dist-hState1}  \commentend}{}\<[E]%
\\
\>[3]{}\hsindent{2}{}\<[5]%
\>[5]{}\Varid{fmap}\;(\Varid{fmap}\;\Varid{fst})\;(\lambda \Varid{s}_{0}\to (\mathbf{do}\;{}\<[33]%
\>[33]{}(\Varid{x},\Varid{s}_{1})\leftarrow \Varid{h_{State1}}\;(\Varid{h_{ND+f}}\;(\Varid{(\Leftrightarrow)}\;\Varid{p}))\;\Varid{s}_{0}{}\<[E]%
\\
\>[33]{}(\Varid{y},\Varid{s}_{2})\leftarrow \Varid{h_{State1}}\;(\Varid{h_{ND+f}}\;(\Varid{(\Leftrightarrow)}\;\Varid{q}))\;\Varid{s}_{1}{}\<[E]%
\\
\>[33]{}\Varid{h_{State1}}\;(\Varid{\eta}\;(\Varid{x}+\!\!+\Varid{y}))\;\Varid{s}_{2})){}\<[E]%
\\
\>[3]{}\mathrel{=}\mbox{\commentbegin ~ definition of \ensuremath{\Varid{h_{State1}}}  \commentend}{}\<[E]%
\\
\>[3]{}\hsindent{2}{}\<[5]%
\>[5]{}\Varid{fmap}\;(\Varid{fmap}\;\Varid{fst})\;(\lambda \Varid{s}_{0}\to (\mathbf{do}\;{}\<[33]%
\>[33]{}(\Varid{x},\Varid{s}_{1})\leftarrow \Varid{h_{State1}}\;(\Varid{h_{ND+f}}\;(\Varid{(\Leftrightarrow)}\;\Varid{p}))\;\Varid{s}_{0}{}\<[E]%
\\
\>[33]{}(\Varid{y},\Varid{s}_{2})\leftarrow \Varid{h_{State1}}\;(\Varid{h_{ND+f}}\;(\Varid{(\Leftrightarrow)}\;\Varid{q}))\;\Varid{s}_{1}{}\<[E]%
\\
\>[33]{}\Conid{Var}\;(\Varid{x}+\!\!+\Varid{y},\Varid{s}_{2}))){}\<[E]%
\\
\>[3]{}\mathrel{=}\mbox{\commentbegin ~ Lemma \ref{lemma:state-restore} (\ensuremath{\Varid{p}} and \ensuremath{\Varid{q}} are in the codomain of \ensuremath{\Varid{local2global}})  \commentend}{}\<[E]%
\\
\>[3]{}\hsindent{2}{}\<[5]%
\>[5]{}\Varid{fmap}\;(\Varid{fmap}\;\Varid{fst})\;(\lambda \Varid{s}_{0}\to (\mathbf{do}\;{}\<[33]%
\>[33]{}(\Varid{x},\Varid{s}_{1})\leftarrow \mathbf{do}\;\{\mskip1.5mu (\Varid{x},\anonymous )\leftarrow \Varid{h_{State1}}\;(\Varid{h_{ND+f}}\;(\Varid{(\Leftrightarrow)}\;\Varid{p}))\;\Varid{s}_{0};\Varid{\eta}\;(\Varid{x},\Varid{s}_{0})\mskip1.5mu\}{}\<[E]%
\\
\>[33]{}(\Varid{y},\Varid{s}_{2})\leftarrow \mathbf{do}\;\{\mskip1.5mu (\Varid{y},\anonymous )\leftarrow \Varid{h_{State1}}\;(\Varid{h_{ND+f}}\;(\Varid{(\Leftrightarrow)}\;\Varid{q}))\;\Varid{s}_{1};\Varid{\eta}\;(\Varid{x},\Varid{s}_{1})\mskip1.5mu\}{}\<[E]%
\\
\>[33]{}\Conid{Var}\;(\Varid{x}+\!\!+\Varid{y},\Varid{s}_{2}))){}\<[E]%
\\
\>[3]{}\mathrel{=}\mbox{\commentbegin ~ monad laws  \commentend}{}\<[E]%
\\
\>[3]{}\hsindent{2}{}\<[5]%
\>[5]{}\Varid{fmap}\;(\Varid{fmap}\;\Varid{fst})\;(\lambda \Varid{s}_{0}\to (\mathbf{do}\;{}\<[33]%
\>[33]{}(\Varid{x},\anonymous )\leftarrow \Varid{h_{State1}}\;(\Varid{h_{ND+f}}\;(\Varid{(\Leftrightarrow)}\;\Varid{p}))\;\Varid{s}_{0}{}\<[E]%
\\
\>[33]{}(\Varid{y},\anonymous )\leftarrow \Varid{h_{State1}}\;(\Varid{h_{ND+f}}\;(\Varid{(\Leftrightarrow)}\;\Varid{q}))\;\Varid{s}_{0}{}\<[E]%
\\
\>[33]{}\Conid{Var}\;(\Varid{x}+\!\!+\Varid{y},\Varid{s}_{0}))){}\<[E]%
\\
\>[3]{}\mathrel{=}\mbox{\commentbegin ~ definition of \ensuremath{\Varid{fmap}} (twice) and \ensuremath{\Varid{fst}}  \commentend}{}\<[E]%
\\
\>[3]{}\hsindent{2}{}\<[5]%
\>[5]{}\lambda \Varid{s}_{0}\to (\mathbf{do}\;{}\<[16]%
\>[16]{}(\Varid{x},\anonymous )\leftarrow \Varid{h_{State1}}\;(\Varid{h_{ND+f}}\;(\Varid{(\Leftrightarrow)}\;\Varid{p}))\;\Varid{s}_{0}{}\<[E]%
\\
\>[16]{}(\Varid{y},\anonymous )\leftarrow \Varid{h_{State1}}\;(\Varid{h_{ND+f}}\;(\Varid{(\Leftrightarrow)}\;\Varid{q}))\;\Varid{s}_{0}{}\<[E]%
\\
\>[16]{}\Conid{Var}\;(\Varid{x}+\!\!+\Varid{y})){}\<[E]%
\\
\>[3]{}\mathrel{=}\mbox{\commentbegin ~ definition of \ensuremath{\Varid{fmap}}, \ensuremath{\Varid{fst}} and monad laws  \commentend}{}\<[E]%
\\
\>[3]{}\hsindent{2}{}\<[5]%
\>[5]{}\lambda \Varid{s}_{0}\to (\mathbf{do}\;{}\<[16]%
\>[16]{}\Varid{x}\leftarrow \Varid{fmap}\;\Varid{fst}\;(\Varid{h_{State1}}\;(\Varid{h_{ND+f}}\;(\Varid{(\Leftrightarrow)}\;\Varid{p}))\;\Varid{s}_{0}){}\<[E]%
\\
\>[16]{}\Varid{y}\leftarrow \Varid{fmap}\;\Varid{fst}\;(\Varid{h_{State1}}\;(\Varid{h_{ND+f}}\;(\Varid{(\Leftrightarrow)}\;\Varid{q}))\;\Varid{s}_{0}){}\<[E]%
\\
\>[16]{}\Conid{Var}\;(\Varid{x}+\!\!+\Varid{y})){}\<[E]%
\\
\>[3]{}\mathrel{=}\mbox{\commentbegin ~ definition of \ensuremath{\Varid{fmap}}  \commentend}{}\<[E]%
\\
\>[3]{}\hsindent{2}{}\<[5]%
\>[5]{}\lambda \Varid{s}_{0}\to (\mathbf{do}\;{}\<[16]%
\>[16]{}\Varid{x}\leftarrow \Varid{fmap}\;(\Varid{fmap}\;\Varid{fst})\;(\Varid{h_{State1}}\;(\Varid{h_{ND+f}}\;(\Varid{(\Leftrightarrow)}\;\Varid{p})))\;\Varid{s}_{0}{}\<[E]%
\\
\>[16]{}\Varid{y}\leftarrow \Varid{fmap}\;(\Varid{fmap}\;\Varid{fst})\;(\Varid{h_{State1}}\;(\Varid{h_{ND+f}}\;(\Varid{(\Leftrightarrow)}\;\Varid{q})))\;\Varid{s}_{0}{}\<[E]%
\\
\>[16]{}\Conid{Var}\;(\Varid{x}+\!\!+\Varid{y})){}\<[E]%
\\
\>[3]{}\mathrel{=}\mbox{\commentbegin ~ definition of \ensuremath{\Varid{h_{Global}}}  \commentend}{}\<[E]%
\\
\>[3]{}\hsindent{2}{}\<[5]%
\>[5]{}\lambda \Varid{s}_{0}\to (\mathbf{do}\;{}\<[16]%
\>[16]{}\Varid{x}\leftarrow \Varid{h_{Global}}\;\Varid{p}\;\Varid{s}_{0}{}\<[E]%
\\
\>[16]{}\Varid{y}\leftarrow \Varid{h_{Global}}\;\Varid{q}\;\Varid{s}_{0}{}\<[E]%
\\
\>[16]{}\Conid{Var}\;(\Varid{x}+\!\!+\Varid{y})){}\<[E]%
\\
\>[3]{}\mathrel{=}\mbox{\commentbegin ~ definition of \ensuremath{\Varid{liftM2}}  \commentend}{}\<[E]%
\\
\>[3]{}\hsindent{2}{}\<[5]%
\>[5]{}\lambda \Varid{s}_{0}\to \Varid{liftM2}\;(+\!\!+)\;(\Varid{h_{Global}}\;\Varid{p}\;\Varid{s}_{0})\;(\Varid{h_{Global}}\;\Varid{q}\;\Varid{s}_{0}){}\<[E]%
\\
\>[3]{}\mathrel{=}\mbox{\commentbegin ~ define \ensuremath{\Varid{alg}_{\Varid{LHS}}^{\Varid{ND}}\;(\Conid{Or}\;\Varid{p}\;\Varid{q})\mathrel{=}\lambda \Varid{s}\to \Varid{liftM2}\;(+\!\!+)\;(\Varid{p}\;\Varid{s})\;(\Varid{q}\;\Varid{s})}  \commentend}{}\<[E]%
\\
\>[3]{}\hsindent{2}{}\<[5]%
\>[5]{}\Varid{alg}_{\Varid{LHS}}^{\Varid{ND}}\;(\Conid{Or}\;(\Varid{h_{Global}}\;\Varid{p})\;(\Varid{h_{Global}}\;\Varid{q})){}\<[E]%
\\
\>[3]{}\mathrel{=}\mbox{\commentbegin ~ definition of \ensuremath{\Varid{fmap}}  \commentend}{}\<[E]%
\\
\>[3]{}\hsindent{2}{}\<[5]%
\>[5]{}\Varid{alg}_{\Varid{LHS}}^{\Varid{ND}}\;(\Varid{fmap}\;\Varid{hGobal}\;(\Conid{Or}\;\Varid{p}\;\Varid{q})){}\<[E]%
\ColumnHook
\end{hscode}\resethooks
\indentend We conclude that this fusion subcondition holds provided that:
\indentbegin \begin{hscode}\SaveRestoreHook
\column{B}{@{}>{\hspre}l<{\hspost}@{}}%
\column{3}{@{}>{\hspre}l<{\hspost}@{}}%
\column{22}{@{}>{\hspre}l<{\hspost}@{}}%
\column{E}{@{}>{\hspre}l<{\hspost}@{}}%
\>[3]{}\Varid{alg}_{\Varid{LHS}}^{\Varid{ND}}\mathbin{::}\Conid{Functor}\;\Varid{f}\Rightarrow \Varid{Nondet_{F}}\;(\Varid{s}\to \Conid{Free}\;\Varid{f}\;[\mskip1.5mu \Varid{a}\mskip1.5mu])\to (\Varid{s}\to \Conid{Free}\;\Varid{f}\;[\mskip1.5mu \Varid{a}\mskip1.5mu]){}\<[E]%
\\
\>[3]{}\Varid{alg}_{\Varid{LHS}}^{\Varid{ND}}\;\Conid{Fail}{}\<[22]%
\>[22]{}\mathrel{=}\lambda \Varid{s}\to \Conid{Var}\;[\mskip1.5mu \mskip1.5mu]{}\<[E]%
\\
\>[3]{}\Varid{alg}_{\Varid{LHS}}^{\Varid{ND}}\;(\Conid{Or}\;\Varid{p}\;\Varid{q}){}\<[22]%
\>[22]{}\mathrel{=}\lambda \Varid{s}\to \Varid{liftM2}\;(+\!\!+)\;(\Varid{p}\;\Varid{s})\;(\Varid{q}\;\Varid{s}){}\<[E]%
\ColumnHook
\end{hscode}\resethooks
\indentend Finally, the last subcondition:\indentbegin \begin{hscode}\SaveRestoreHook
\column{B}{@{}>{\hspre}l<{\hspost}@{}}%
\column{3}{@{}>{\hspre}l<{\hspost}@{}}%
\column{5}{@{}>{\hspre}l<{\hspost}@{}}%
\column{E}{@{}>{\hspre}l<{\hspost}@{}}%
\>[5]{}\Varid{h_{Global}}\;(\Varid{alg}\;(\Conid{Inr}\;(\Conid{Inr}\;\Varid{op}))){}\<[E]%
\\
\>[3]{}\mathrel{=}\mbox{\commentbegin ~ definition of \ensuremath{\Varid{alg}}  \commentend}{}\<[E]%
\\
\>[3]{}\hsindent{2}{}\<[5]%
\>[5]{}\Varid{h_{Global}}\;(\Conid{Op}\;(\Conid{Inr}\;(\Conid{Inr}\;\Varid{op}))){}\<[E]%
\\
\>[3]{}\mathrel{=}\mbox{\commentbegin ~ definition of \ensuremath{\Varid{h_{Global}}}  \commentend}{}\<[E]%
\\
\>[3]{}\hsindent{2}{}\<[5]%
\>[5]{}\Varid{fmap}\;(\Varid{fmap}\;\Varid{fst})\;(\Varid{h_{State1}}\;(\Varid{h_{ND+f}}\;(\Varid{(\Leftrightarrow)}\;(\Conid{Op}\;(\Conid{Inr}\;(\Conid{Inr}\;\Varid{op})))))){}\<[E]%
\\
\>[3]{}\mathrel{=}\mbox{\commentbegin ~ definition of \ensuremath{\Varid{(\Leftrightarrow)}}  \commentend}{}\<[E]%
\\
\>[3]{}\hsindent{2}{}\<[5]%
\>[5]{}\Varid{fmap}\;(\Varid{fmap}\;\Varid{fst})\;(\Varid{h_{State1}}\;(\Varid{h_{ND+f}}\;(\Conid{Op}\;(\Conid{Inr}\;(\Conid{Inr}\;(\Varid{fmap}\;\Varid{(\Leftrightarrow)}\;\Varid{op})))))){}\<[E]%
\\
\>[3]{}\mathrel{=}\mbox{\commentbegin ~ definition of \ensuremath{\Varid{h_{ND+f}}}  \commentend}{}\<[E]%
\\
\>[3]{}\hsindent{2}{}\<[5]%
\>[5]{}\Varid{fmap}\;(\Varid{fmap}\;\Varid{fst})\;(\Varid{h_{State1}}\;(\Conid{Op}\;(\Varid{fmap}\;\Varid{h_{ND+f}}\;(\Conid{Inr}\;(\Varid{fmap}\;\Varid{(\Leftrightarrow)}\;\Varid{op}))))){}\<[E]%
\\
\>[3]{}\mathrel{=}\mbox{\commentbegin ~ definition of \ensuremath{\Varid{fmap}}  \commentend}{}\<[E]%
\\
\>[3]{}\hsindent{2}{}\<[5]%
\>[5]{}\Varid{fmap}\;(\Varid{fmap}\;\Varid{fst})\;(\Varid{h_{State1}}\;(\Conid{Op}\;(\Conid{Inr}\;(\Varid{fmap}\;\Varid{h_{ND+f}}\;(\Varid{fmap}\;\Varid{(\Leftrightarrow)}\;\Varid{op}))))){}\<[E]%
\\
\>[3]{}\mathrel{=}\mbox{\commentbegin ~ \ensuremath{\Varid{fmap}} fusion  \commentend}{}\<[E]%
\\
\>[3]{}\hsindent{2}{}\<[5]%
\>[5]{}\Varid{fmap}\;(\Varid{fmap}\;\Varid{fst})\;(\Varid{h_{State1}}\;(\Conid{Op}\;(\Conid{Inr}\;(\Varid{fmap}\;(\Varid{h_{ND+f}}\hsdot{\circ }{.}\Varid{(\Leftrightarrow)})\;\Varid{op})))){}\<[E]%
\\
\>[3]{}\mathrel{=}\mbox{\commentbegin ~ definition of \ensuremath{\Varid{h_{State1}}}  \commentend}{}\<[E]%
\\
\>[3]{}\hsindent{2}{}\<[5]%
\>[5]{}\Varid{fmap}\;(\Varid{fmap}\;\Varid{fst})\;(\lambda \Varid{s}\to \Conid{Op}\;(\Varid{fmap}\;(\mathbin{\$}\Varid{s})\;(\Varid{fmap}\;\Varid{h_{State1}}\;(\Varid{fmap}\;(\Varid{h_{ND+f}}\hsdot{\circ }{.}\Varid{(\Leftrightarrow)})\;\Varid{op})))){}\<[E]%
\\
\>[3]{}\mathrel{=}\mbox{\commentbegin ~ \ensuremath{\Varid{fmap}} fusion  \commentend}{}\<[E]%
\\
\>[3]{}\hsindent{2}{}\<[5]%
\>[5]{}\Varid{fmap}\;(\Varid{fmap}\;\Varid{fst})\;(\lambda \Varid{s}\to \Conid{Op}\;(\Varid{fmap}\;(\mathbin{\$}\Varid{s})\;(\Varid{fmap}\;(\Varid{h_{State1}}\hsdot{\circ }{.}\Varid{h_{ND+f}}\hsdot{\circ }{.}\Varid{(\Leftrightarrow)})\;\Varid{op}))){}\<[E]%
\\
\>[3]{}\mathrel{=}\mbox{\commentbegin ~ definition of \ensuremath{\Varid{fmap}}  \commentend}{}\<[E]%
\\
\>[3]{}\hsindent{2}{}\<[5]%
\>[5]{}\lambda \Varid{s}\to \Varid{fmap}\;\Varid{fst}\;(\Conid{Op}\;(\Varid{fmap}\;(\mathbin{\$}\Varid{s})\;(\Varid{fmap}\;(\Varid{h_{State1}}\hsdot{\circ }{.}\Varid{h_{ND+f}}\hsdot{\circ }{.}\Varid{(\Leftrightarrow)})\;\Varid{op}))){}\<[E]%
\\
\>[3]{}\mathrel{=}\mbox{\commentbegin ~ definition of \ensuremath{\Varid{fmap}}  \commentend}{}\<[E]%
\\
\>[3]{}\hsindent{2}{}\<[5]%
\>[5]{}\lambda \Varid{s}\to \Conid{Op}\;(\Varid{fmap}\;(\Varid{fmap}\;\Varid{fst})\;(\Varid{fmap}\;(\mathbin{\$}\Varid{s})\;(\Varid{fmap}\;(\Varid{h_{State1}}\hsdot{\circ }{.}\Varid{h_{ND+f}}\hsdot{\circ }{.}\Varid{(\Leftrightarrow)})\;\Varid{op}))){}\<[E]%
\\
\>[3]{}\mathrel{=}\mbox{\commentbegin ~ \ensuremath{\Varid{fmap}} fusion  \commentend}{}\<[E]%
\\
\>[3]{}\hsindent{2}{}\<[5]%
\>[5]{}\lambda \Varid{s}\to \Conid{Op}\;(\Varid{fmap}\;(\Varid{fmap}\;\Varid{fst}\hsdot{\circ }{.}(\mathbin{\$}\Varid{s}))\;(\Varid{fmap}\;(\Varid{h_{State1}}\hsdot{\circ }{.}\Varid{h_{ND+f}}\hsdot{\circ }{.}\Varid{(\Leftrightarrow)})\;\Varid{op}))){}\<[E]%
\\
\>[3]{}\mathrel{=}\mbox{\commentbegin ~ \Cref{eq:comm-app-fmap}  \commentend}{}\<[E]%
\\
\>[3]{}\hsindent{2}{}\<[5]%
\>[5]{}\lambda \Varid{s}\to \Conid{Op}\;(\Varid{fmap}\;((\mathbin{\$}\Varid{s})\hsdot{\circ }{.}\Varid{fmap}\;(\Varid{fmap}\;\Varid{fst}))\;(\Varid{fmap}\;(\Varid{h_{State1}}\hsdot{\circ }{.}\Varid{h_{ND+f}}\hsdot{\circ }{.}\Varid{(\Leftrightarrow)})\;\Varid{op}))){}\<[E]%
\\
\>[3]{}\mathrel{=}\mbox{\commentbegin ~ \ensuremath{\Varid{fmap}} fission  \commentend}{}\<[E]%
\\
\>[3]{}\hsindent{2}{}\<[5]%
\>[5]{}\lambda \Varid{s}\to \Conid{Op}\;((\Varid{fmap}\;(\mathbin{\$}\Varid{s})\hsdot{\circ }{.}\Varid{fmap}\;(\Varid{fmap}\;(\Varid{fmap}\;\Varid{fst})))\;(\Varid{fmap}\;(\Varid{h_{State1}}\hsdot{\circ }{.}\Varid{h_{ND+f}}\hsdot{\circ }{.}\Varid{(\Leftrightarrow)})\;\Varid{op})){}\<[E]%
\\
\>[3]{}\mathrel{=}\mbox{\commentbegin ~ \ensuremath{\Varid{fmap}} fusion  \commentend}{}\<[E]%
\\
\>[3]{}\hsindent{2}{}\<[5]%
\>[5]{}\lambda \Varid{s}\to \Conid{Op}\;(\Varid{fmap}\;(\mathbin{\$}\Varid{s})\;(\Varid{fmap}\;(\Varid{fmap}\;(\Varid{fmap}\;\Varid{fst})\hsdot{\circ }{.}\Varid{h_{State1}}\hsdot{\circ }{.}\Varid{h_{ND+f}}\hsdot{\circ }{.}\Varid{(\Leftrightarrow)})\;\Varid{op})){}\<[E]%
\\
\>[3]{}\mathrel{=}\mbox{\commentbegin ~ definition of \ensuremath{\Varid{h_{Global}}}  \commentend}{}\<[E]%
\\
\>[3]{}\hsindent{2}{}\<[5]%
\>[5]{}\lambda \Varid{s}\to \Conid{Op}\;(\Varid{fmap}\;(\mathbin{\$}\Varid{s})\;(\Varid{fmap}\;\Varid{h_{Global}}\;\Varid{op})){}\<[E]%
\\
\>[3]{}\mathrel{=}\mbox{\commentbegin ~ define \ensuremath{\Varid{fwd}_{\Varid{LHS}}\;\Varid{op}\mathrel{=}\lambda \Varid{s}\to \Conid{Op}\;(\Varid{fmap}\;(\mathbin{\$}\Varid{s})\;\Varid{op}}  \commentend}{}\<[E]%
\\
\>[3]{}\hsindent{2}{}\<[5]%
\>[5]{}\Varid{fwd}_{\Varid{LHS}}\;(\Varid{fmap}\;\Varid{h_{Global}}\;\Varid{op}){}\<[E]%
\ColumnHook
\end{hscode}\resethooks
\indentend We conclude that this fusion subcondition holds provided that:
\indentbegin \begin{hscode}\SaveRestoreHook
\column{B}{@{}>{\hspre}l<{\hspost}@{}}%
\column{3}{@{}>{\hspre}l<{\hspost}@{}}%
\column{E}{@{}>{\hspre}l<{\hspost}@{}}%
\>[3]{}\Varid{fwd}_{\Varid{LHS}}\mathbin{::}\Conid{Functor}\;\Varid{f}\Rightarrow \Varid{f}\;(\Varid{s}\to \Conid{Free}\;\Varid{f}\;[\mskip1.5mu \Varid{a}\mskip1.5mu])\to (\Varid{s}\to \Conid{Free}\;\Varid{f}\;[\mskip1.5mu \Varid{a}\mskip1.5mu]){}\<[E]%
\\
\>[3]{}\Varid{fwd}_{\Varid{LHS}}\;\Varid{op}\mathrel{=}\lambda \Varid{s}\to \Conid{Op}\;(\Varid{fmap}\;(\mathbin{\$}\Varid{s})\;\Varid{op}){}\<[E]%
\ColumnHook
\end{hscode}\resethooks
\indentend \subsection{Equating the Fused Sides}

We observe that the following equations hold trivially. 
\begin{eqnarray*}
\ensuremath{\Varid{gen}_{\Varid{LHS}}} & = & \ensuremath{\Varid{gen}_{\Varid{RHS}}} \\
\ensuremath{\Varid{alg}_{\Varid{LHS}}^{\Varid{S}}} & = & \ensuremath{\Varid{alg}_{\Varid{RHS}}^{\Varid{S}}} \\
\ensuremath{\Varid{alg}_{\Varid{LHS}}^{\Varid{ND}}} & = & \ensuremath{\Varid{alg}_{\Varid{RHS}}^{\Varid{ND}}} \\
\ensuremath{\Varid{fwd}_{\Varid{LHS}}} & = & \ensuremath{\Varid{fwd}_{\Varid{RHS}}}
\end{eqnarray*}

Therefore, the main theorem holds.

\subsection{Key Lemma: State Restoration}

The key lemma is the following, which guarantees that
\ensuremath{\Varid{local2global}} restores the initial state after a computation.

\begin{lemma}[State is Restored] \label{lemma:state-restore} \ \\
\begin{eqnarray*}
& \ensuremath{\Varid{h_{State1}}\;(\Varid{h_{ND+f}}\;(\Varid{(\Leftrightarrow)}\;(\Varid{local2global}\;\Varid{t})))\;\Varid{s}} & \\
& = & \\
& \ensuremath{\mathbf{do}\;(\Varid{x},\anonymous )\leftarrow \Varid{h_{State1}}\;(\Varid{h_{ND+f}}\;(\Varid{(\Leftrightarrow)}\;(\Varid{local2global}\;\Varid{t})))\;\Varid{s};\Varid{\eta}\;(\Varid{x},\Varid{s})} &
\end{eqnarray*}
\end{lemma}

\begin{proof}
The proof proceeds by structural induction on \ensuremath{\Varid{t}}.

\noindent \mbox{\underline{case \ensuremath{\Varid{t}\mathrel{=}\Conid{Var}\;\Varid{y}}}}\indentbegin \begin{hscode}\SaveRestoreHook
\column{B}{@{}>{\hspre}l<{\hspost}@{}}%
\column{3}{@{}>{\hspre}l<{\hspost}@{}}%
\column{6}{@{}>{\hspre}l<{\hspost}@{}}%
\column{E}{@{}>{\hspre}l<{\hspost}@{}}%
\>[6]{}\Varid{h_{State1}}\;(\Varid{h_{ND+f}}\;(\Varid{(\Leftrightarrow)}\;(\Varid{local2global}\;(\Conid{Var}\;\Varid{y}))))\;\Varid{s}{}\<[E]%
\\
\>[3]{}\mathrel{=}\mbox{\commentbegin ~  definition of \ensuremath{\Varid{local2global}}   \commentend}{}\<[E]%
\\
\>[3]{}\hsindent{3}{}\<[6]%
\>[6]{}\Varid{h_{State1}}\;(\Varid{h_{ND+f}}\;(\Varid{(\Leftrightarrow)}\;(\Conid{Var}\;\Varid{y})))\;\Varid{s}{}\<[E]%
\\
\>[3]{}\mathrel{=}\mbox{\commentbegin ~  definition of \ensuremath{\Varid{(\Leftrightarrow)}}   \commentend}{}\<[E]%
\\
\>[3]{}\hsindent{3}{}\<[6]%
\>[6]{}\Varid{h_{State1}}\;(\Varid{h_{ND+f}}\;(\Conid{Var}\;\Varid{y}))\;\Varid{s}{}\<[E]%
\\
\>[3]{}\mathrel{=}\mbox{\commentbegin ~  definition of \ensuremath{\Varid{h_{ND+f}}}   \commentend}{}\<[E]%
\\
\>[3]{}\hsindent{3}{}\<[6]%
\>[6]{}\Varid{h_{State1}}\;(\Conid{Var}\;[\mskip1.5mu \Varid{y}\mskip1.5mu])\;\Varid{s}{}\<[E]%
\\
\>[3]{}\mathrel{=}\mbox{\commentbegin ~  definition of \ensuremath{\Varid{h_{State1}}}   \commentend}{}\<[E]%
\\
\>[3]{}\hsindent{3}{}\<[6]%
\>[6]{}\Conid{Var}\;([\mskip1.5mu \Varid{y}\mskip1.5mu],\Varid{s}){}\<[E]%
\\
\>[3]{}\mathrel{=}\mbox{\commentbegin ~  monad law  \commentend}{}\<[E]%
\\
\>[3]{}\hsindent{3}{}\<[6]%
\>[6]{}\mathbf{do}\;(\Varid{x},\anonymous )\leftarrow \Conid{Var}\;([\mskip1.5mu \Varid{y}\mskip1.5mu],\Varid{s});\Conid{Var}\;(\Varid{x},\Varid{s}){}\<[E]%
\\
\>[3]{}\mathrel{=}\mbox{\commentbegin ~  definition of \ensuremath{\Varid{local2global},\Varid{h_{ND+f}},\Varid{(\Leftrightarrow)},\Varid{h_{State1}}} and \ensuremath{\Varid{\eta}}   \commentend}{}\<[E]%
\\
\>[3]{}\hsindent{3}{}\<[6]%
\>[6]{}\mathbf{do}\;(\Varid{x},\anonymous )\leftarrow \Varid{h_{State1}}\;(\Varid{h_{ND+f}}\;(\Varid{(\Leftrightarrow)}\;(\Varid{local2global}\;(\Conid{Var}\;\Varid{y}))))\;\Varid{s};\Varid{\eta}\;(\Varid{x},\Varid{s}){}\<[E]%
\ColumnHook
\end{hscode}\resethooks
\indentend \noindent \mbox{\underline{case \ensuremath{\Varid{t}\mathrel{=}\Conid{Op}\;(\Conid{Inl}\;(\Conid{Get}\;\Varid{k}))}}}\indentbegin \begin{hscode}\SaveRestoreHook
\column{B}{@{}>{\hspre}l<{\hspost}@{}}%
\column{3}{@{}>{\hspre}l<{\hspost}@{}}%
\column{6}{@{}>{\hspre}l<{\hspost}@{}}%
\column{E}{@{}>{\hspre}l<{\hspost}@{}}%
\>[6]{}\Varid{h_{State1}}\;(\Varid{h_{ND+f}}\;(\Varid{(\Leftrightarrow)}\;(\Varid{local2global}\;(\Conid{Op}\;(\Conid{Inl}\;(\Conid{Get}\;\Varid{k}))))))\;\Varid{s}{}\<[E]%
\\
\>[3]{}\mathrel{=}\mbox{\commentbegin ~  definition of \ensuremath{\Varid{local2global}}   \commentend}{}\<[E]%
\\
\>[3]{}\hsindent{3}{}\<[6]%
\>[6]{}\Varid{h_{State1}}\;(\Varid{h_{ND+f}}\;(\Varid{(\Leftrightarrow)}\;(\Conid{Op}\;(\Conid{Inl}\;(\Conid{Get}\;(\Varid{local2global}\hsdot{\circ }{.}\Varid{k}))))))\;\Varid{s}{}\<[E]%
\\
\>[3]{}\mathrel{=}\mbox{\commentbegin ~  definition of \ensuremath{\Varid{(\Leftrightarrow)}}   \commentend}{}\<[E]%
\\
\>[3]{}\hsindent{3}{}\<[6]%
\>[6]{}\Varid{h_{State1}}\;(\Varid{h_{ND+f}}\;(\Conid{Op}\;(\Conid{Inr}\;(\Conid{Inl}\;(\Conid{Get}\;(\Varid{(\Leftrightarrow)}\hsdot{\circ }{.}\Varid{local2global}\hsdot{\circ }{.}\Varid{k}))))))\;\Varid{s}{}\<[E]%
\\
\>[3]{}\mathrel{=}\mbox{\commentbegin ~  definition of \ensuremath{\Varid{h_{ND+f}}}   \commentend}{}\<[E]%
\\
\>[3]{}\hsindent{3}{}\<[6]%
\>[6]{}\Varid{h_{State1}}\;(\Conid{Op}\;(\Conid{Inl}\;(\Conid{Get}\;(\Varid{h_{ND+f}}\hsdot{\circ }{.}\Varid{(\Leftrightarrow)}\hsdot{\circ }{.}\Varid{local2global}\hsdot{\circ }{.}\Varid{k}))))\;\Varid{s}{}\<[E]%
\\
\>[3]{}\mathrel{=}\mbox{\commentbegin ~  definition of \ensuremath{\Varid{h_{State1}}}   \commentend}{}\<[E]%
\\
\>[3]{}\hsindent{3}{}\<[6]%
\>[6]{}(\Varid{h_{State1}}\hsdot{\circ }{.}\Varid{h_{ND+f}}\hsdot{\circ }{.}\Varid{(\Leftrightarrow)}\hsdot{\circ }{.}\Varid{local2global}\hsdot{\circ }{.}\Varid{k})\;\Varid{s}\;\Varid{s}{}\<[E]%
\\
\>[3]{}\mathrel{=}\mbox{\commentbegin ~  definition of \ensuremath{(\hsdot{\circ }{.})}   \commentend}{}\<[E]%
\\
\>[3]{}\hsindent{3}{}\<[6]%
\>[6]{}(\Varid{h_{State1}}\;(\Varid{h_{ND+f}}\;(\Varid{(\Leftrightarrow)}\;(\Varid{local2global}\;(\Varid{k}\;\Varid{s})))))\;\Varid{s}{}\<[E]%
\\
\>[3]{}\mathrel{=}\mbox{\commentbegin ~  induction hypothesis   \commentend}{}\<[E]%
\\
\>[3]{}\hsindent{3}{}\<[6]%
\>[6]{}\mathbf{do}\;(\Varid{x},\anonymous )\leftarrow \Varid{h_{State1}}\;(\Varid{(\Leftrightarrow)}\;(\Varid{h_{ND+f}}\;(\Varid{local2global}\;(\Varid{k}\;\Varid{s}))))\;\Varid{s};\Varid{\eta}\;(\Varid{x},\Varid{s}){}\<[E]%
\\
\>[3]{}\mathrel{=}\mbox{\commentbegin ~  definition of \ensuremath{\Varid{local2global},\Varid{(\Leftrightarrow)},\Varid{h_{ND+f}},\Varid{h_{State1}}}   \commentend}{}\<[E]%
\\
\>[3]{}\hsindent{3}{}\<[6]%
\>[6]{}\mathbf{do}\;(\Varid{x},\anonymous )\leftarrow \Varid{h_{State1}}\;(\Varid{h_{ND+f}}\;(\Varid{local2global}\;(\Conid{Op}\;(\Conid{Inl}\;(\Conid{Get}\;\Varid{k})))))\;\Varid{s};\Varid{\eta}\;(\Varid{x},\Varid{s}){}\<[E]%
\ColumnHook
\end{hscode}\resethooks
\indentend \noindent \mbox{\underline{case \ensuremath{\Varid{t}\mathrel{=}\Conid{Op}\;(\Conid{Inr}\;(\Conid{Inl}\;\Conid{Fail}))}}}\indentbegin \begin{hscode}\SaveRestoreHook
\column{B}{@{}>{\hspre}l<{\hspost}@{}}%
\column{3}{@{}>{\hspre}l<{\hspost}@{}}%
\column{6}{@{}>{\hspre}l<{\hspost}@{}}%
\column{E}{@{}>{\hspre}l<{\hspost}@{}}%
\>[6]{}\Varid{h_{State1}}\;(\Varid{h_{ND+f}}\;(\Varid{(\Leftrightarrow)}\;(\Varid{local2global}\;(\Conid{Op}\;(\Conid{Inr}\;(\Conid{Inl}\;\Conid{Fail}))))))\;\Varid{s}{}\<[E]%
\\
\>[3]{}\mathrel{=}\mbox{\commentbegin ~  definition of \ensuremath{\Varid{local2global}}   \commentend}{}\<[E]%
\\
\>[3]{}\hsindent{3}{}\<[6]%
\>[6]{}\Varid{h_{State1}}\;(\Varid{h_{ND+f}}\;(\Varid{(\Leftrightarrow)}\;(\Conid{Op}\;(\Conid{Inr}\;(\Conid{Inl}\;\Conid{Fail})))))\;\Varid{s}{}\<[E]%
\\
\>[3]{}\mathrel{=}\mbox{\commentbegin ~  definition of \ensuremath{\Varid{(\Leftrightarrow)}}   \commentend}{}\<[E]%
\\
\>[3]{}\hsindent{3}{}\<[6]%
\>[6]{}\Varid{h_{State1}}\;(\Varid{h_{ND+f}}\;(\Conid{Op}\;(\Conid{Inl}\;\Conid{Fail})))\;\Varid{s}{}\<[E]%
\\
\>[3]{}\mathrel{=}\mbox{\commentbegin ~  definition of \ensuremath{\Varid{h_{ND+f}}}   \commentend}{}\<[E]%
\\
\>[3]{}\hsindent{3}{}\<[6]%
\>[6]{}\Varid{h_{State1}}\;(\Conid{Var}\;[\mskip1.5mu \mskip1.5mu])\;\Varid{s}{}\<[E]%
\\
\>[3]{}\mathrel{=}\mbox{\commentbegin ~  definition of \ensuremath{\Varid{h_{State1}}}   \commentend}{}\<[E]%
\\
\>[3]{}\hsindent{3}{}\<[6]%
\>[6]{}\Conid{Var}\;([\mskip1.5mu \mskip1.5mu],\Varid{s}){}\<[E]%
\\
\>[3]{}\mathrel{=}\mbox{\commentbegin ~  monad law  \commentend}{}\<[E]%
\\
\>[3]{}\hsindent{3}{}\<[6]%
\>[6]{}\mathbf{do}\;(\Varid{x},\anonymous )\leftarrow \Conid{Var}\;([\mskip1.5mu \mskip1.5mu],\Varid{s});\Conid{Var}\;(\Varid{x},\Varid{s}){}\<[E]%
\\
\>[3]{}\mathrel{=}\mbox{\commentbegin ~  definition of \ensuremath{\Varid{local2global},\Varid{(\Leftrightarrow)},\Varid{h_{ND+f}},\Varid{h_{State1}}}   \commentend}{}\<[E]%
\\
\>[3]{}\hsindent{3}{}\<[6]%
\>[6]{}\mathbf{do}\;(\Varid{x},\anonymous )\leftarrow \Varid{h_{State1}}\;(\Varid{h_{ND+f}}\;(\Varid{(\Leftrightarrow)}\;(\Varid{local2global}\;(\Conid{Op}\;(\Conid{Inr}\;(\Conid{Inl}\;\Conid{Fail}))))))\;\Varid{s};\Varid{\eta}\;(\Varid{x},\Varid{s}){}\<[E]%
\ColumnHook
\end{hscode}\resethooks
\indentend \noindent \mbox{\underline{case \ensuremath{\Varid{t}\mathrel{=}\Conid{Op}\;(\Conid{Inl}\;(\Conid{Put}\;\Varid{t}\;\Varid{k}))}}}\indentbegin \begin{hscode}\SaveRestoreHook
\column{B}{@{}>{\hspre}l<{\hspost}@{}}%
\column{3}{@{}>{\hspre}l<{\hspost}@{}}%
\column{6}{@{}>{\hspre}l<{\hspost}@{}}%
\column{10}{@{}>{\hspre}l<{\hspost}@{}}%
\column{15}{@{}>{\hspre}l<{\hspost}@{}}%
\column{20}{@{}>{\hspre}l<{\hspost}@{}}%
\column{26}{@{}>{\hspre}l<{\hspost}@{}}%
\column{37}{@{}>{\hspre}l<{\hspost}@{}}%
\column{39}{@{}>{\hspre}l<{\hspost}@{}}%
\column{50}{@{}>{\hspre}l<{\hspost}@{}}%
\column{57}{@{}>{\hspre}l<{\hspost}@{}}%
\column{E}{@{}>{\hspre}l<{\hspost}@{}}%
\>[6]{}\Varid{h_{State1}}\;(\Varid{h_{ND+f}}\;(\Varid{(\Leftrightarrow)}\;(\Varid{local2global}\;(\Conid{Op}\;(\Conid{Inl}\;(\Conid{Put}\;\Varid{t}\;\Varid{k}))))))\;\Varid{s}{}\<[E]%
\\
\>[3]{}\mathrel{=}\mbox{\commentbegin ~  definition of \ensuremath{\Varid{local2global}}   \commentend}{}\<[E]%
\\
\>[3]{}\hsindent{3}{}\<[6]%
\>[6]{}\Varid{h_{State1}}\;(\Varid{h_{ND+f}}\;(\Varid{(\Leftrightarrow)}\;(\Varid{put_{R}}\;\Varid{t}>\!\!>\Varid{local2global}\;\Varid{k})))\;\Varid{s}{}\<[E]%
\\
\>[3]{}\mathrel{=}\mbox{\commentbegin ~  definition of \ensuremath{\Varid{put_{R}}}   \commentend}{}\<[E]%
\\
\>[3]{}\hsindent{3}{}\<[6]%
\>[6]{}\Varid{h_{State1}}\;(\Varid{h_{ND+f}}\;(\Varid{(\Leftrightarrow)}\;((\Varid{get}>\!\!>\!\!=\lambda \Varid{t'}\to \Varid{put}\;\Varid{t}\mathbin{\talloblong}\Varid{side}\;(\Varid{put}\;\Varid{t'}))>\!\!>\Varid{local2global}\;\Varid{k})))\;\Varid{s}{}\<[E]%
\\
\>[3]{}\mathrel{=}\mbox{\commentbegin ~  definition of \ensuremath{(\talloblong)}, \ensuremath{\Varid{get}}, \ensuremath{\Varid{put}}, \ensuremath{\Varid{side}} and \ensuremath{(>\!\!>\!\!=)}   \commentend}{}\<[E]%
\\
\>[3]{}\hsindent{3}{}\<[6]%
\>[6]{}\Varid{h_{State1}}\;(\Varid{h_{ND+f}}\;(\Varid{(\Leftrightarrow)}\;(\Conid{Op}\;(\Conid{Inl}\;(\Conid{Get}\;(\lambda \Varid{t'}\to {}\<[E]%
\\
\>[6]{}\hsindent{33}{}\<[39]%
\>[39]{}\Conid{Op}\;(\Conid{Inr}\;(\Conid{Inl}\;(\Conid{Or}\;{}\<[57]%
\>[57]{}(\Conid{Op}\;(\Conid{Inl}\;(\Conid{Put}\;\Varid{t}\;(\Varid{local2global}\;\Varid{k}))))\;{}\<[E]%
\\
\>[57]{}(\Conid{Op}\;(\Conid{Inl}\;(\Conid{Put}\;\Varid{t'}\;(\Conid{Op}\;(\Conid{Inr}\;(\Conid{Inl}\;\Conid{Fail})))))))))))))))\;\Varid{s}{}\<[E]%
\\
\>[3]{}\mathrel{=}\mbox{\commentbegin ~  definition of \ensuremath{\Varid{(\Leftrightarrow)}}   \commentend}{}\<[E]%
\\
\>[3]{}\hsindent{3}{}\<[6]%
\>[6]{}\Varid{h_{State1}}\;(\Varid{h_{ND+f}}\;(\Conid{Op}\;(\Conid{Inr}\;(\Conid{Inl}\;(\Conid{Get}\;(\lambda \Varid{t'}\to {}\<[E]%
\\
\>[6]{}\hsindent{31}{}\<[37]%
\>[37]{}\Conid{Op}\;(\Conid{Inl}\;(\Conid{Or}\;{}\<[50]%
\>[50]{}(\Conid{Op}\;(\Conid{Inr}\;(\Conid{Inl}\;(\Conid{Put}\;\Varid{t}\;(\Varid{(\Leftrightarrow)}\;(\Varid{local2global}\;\Varid{k}))))))\;{}\<[E]%
\\
\>[50]{}(\Conid{Op}\;(\Conid{Inr}\;(\Conid{Inl}\;(\Conid{Put}\;\Varid{t'}\;(\Conid{Op}\;(\Conid{Inl}\;\Conid{Fail}))))))))))))))\;\Varid{s}{}\<[E]%
\\
\>[3]{}\mathrel{=}\mbox{\commentbegin ~  definition of \ensuremath{\Varid{h_{ND+f}}}   \commentend}{}\<[E]%
\\
\>[3]{}\hsindent{3}{}\<[6]%
\>[6]{}\Varid{h_{State1}}\;(\Conid{Op}\;(\Conid{Inl}\;(\Conid{Get}\;(\lambda \Varid{t'}\to {}\<[E]%
\\
\>[6]{}\hsindent{20}{}\<[26]%
\>[26]{}\Varid{liftM2}\;(+\!\!+)\;{}\<[39]%
\>[39]{}(\Conid{Op}\;(\Conid{Inl}\;(\Conid{Put}\;\Varid{t}\;(\Varid{h_{ND+f}}\;(\Varid{(\Leftrightarrow)}\;(\Varid{local2global}\;\Varid{k}))))))\;{}\<[E]%
\\
\>[39]{}(\Conid{Op}\;(\Conid{Inl}\;(\Conid{Put}\;\Varid{t'}\;(\Conid{Var}\;[\mskip1.5mu \mskip1.5mu]))))))))\;\Varid{s}{}\<[E]%
\\
\>[3]{}\mathrel{=}\mbox{\commentbegin ~  definition of \ensuremath{\Varid{h_{State1}}}   \commentend}{}\<[E]%
\\
\>[3]{}\hsindent{3}{}\<[6]%
\>[6]{}\Varid{h_{State1}}\;(\Varid{liftM2}\;(+\!\!+)\;(\Conid{Op}\;(\Conid{Inl}\;(\Conid{Put}\;\Varid{t}\;(\Varid{h_{ND+f}}\;(\Varid{(\Leftrightarrow)}\;(\Varid{local2global}\;\Varid{k}))))))\;(\Conid{Op}\;(\Conid{Inl}\;(\Conid{Put}\;\Varid{s}\;(\Conid{Var}\;[\mskip1.5mu \mskip1.5mu])))))\;\Varid{s}{}\<[E]%
\\
\>[3]{}\mathrel{=}\mbox{\commentbegin ~  definition of \ensuremath{\Varid{liftM2}}  \commentend}{}\<[E]%
\\
\>[3]{}\hsindent{3}{}\<[6]%
\>[6]{}\Varid{h_{State1}}\;{}\<[15]%
\>[15]{}(\mathbf{do}\;{}\<[20]%
\>[20]{}\Varid{x}\leftarrow \Conid{Op}\;(\Conid{Inl}\;(\Conid{Put}\;\Varid{t}\;(\Varid{h_{ND+f}}\;(\Varid{(\Leftrightarrow)}\;(\Varid{local2global}\;\Varid{k}))))){}\<[E]%
\\
\>[20]{}\Varid{y}\leftarrow \Conid{Op}\;(\Conid{Inl}\;(\Conid{Put}\;\Varid{s}\;(\Conid{Var}\;[\mskip1.5mu \mskip1.5mu]))){}\<[E]%
\\
\>[20]{}\Conid{Var}\;(\Varid{x}+\!\!+\Varid{y}){}\<[E]%
\\
\>[15]{})\;\Varid{s}{}\<[E]%
\\
\>[3]{}\mathrel{=}\mbox{\commentbegin ~  Lemma~\ref{lemma:dist-hState1}  \commentend}{}\<[E]%
\\
\>[3]{}\hsindent{3}{}\<[6]%
\>[6]{}\mathbf{do}\;{}\<[10]%
\>[10]{}(\Varid{x},\Varid{s}_{1})\leftarrow \Varid{h_{State1}}\;(\Conid{Op}\;(\Conid{Inl}\;(\Conid{Put}\;\Varid{t}\;(\Varid{h_{ND+f}}\;(\Varid{(\Leftrightarrow)}\;(\Varid{local2global}\;\Varid{k}))))))\;\Varid{s}{}\<[E]%
\\
\>[10]{}(\Varid{y},\Varid{s}_{2})\leftarrow \Varid{h_{State1}}\;(\Conid{Op}\;(\Conid{Inl}\;(\Conid{Put}\;\Varid{s}\;(\Conid{Var}\;[\mskip1.5mu \mskip1.5mu]))))\;\Varid{s}_{1}{}\<[E]%
\\
\>[10]{}\Conid{Var}\;(\Varid{x}+\!\!+\Varid{y},\Varid{s}_{2}){}\<[E]%
\\
\>[3]{}\mathrel{=}\mbox{\commentbegin ~  definition of \ensuremath{\Varid{h_{State1}}}  \commentend}{}\<[E]%
\\
\>[3]{}\hsindent{3}{}\<[6]%
\>[6]{}\mathbf{do}\;{}\<[10]%
\>[10]{}(\Varid{x},\Varid{s}_{1})\leftarrow \Varid{h_{State1}}\;(\Varid{h_{ND+f}}\;(\Varid{(\Leftrightarrow)}\;(\Varid{local2global}\;\Varid{k})))\;\Varid{t}{}\<[E]%
\\
\>[10]{}(\Varid{y},\Varid{s}_{2})\leftarrow \Conid{Var}\;([\mskip1.5mu \mskip1.5mu],\Varid{s}){}\<[E]%
\\
\>[10]{}\Conid{Var}\;(\Varid{x}+\!\!+\Varid{y},\Varid{s}_{2}){}\<[E]%
\\
\>[3]{}\mathrel{=}\mbox{\commentbegin ~  monad laws  \commentend}{}\<[E]%
\\
\>[3]{}\hsindent{3}{}\<[6]%
\>[6]{}\mathbf{do}\;{}\<[10]%
\>[10]{}(\Varid{x},\anonymous )\leftarrow \Varid{h_{State1}}\;(\Varid{h_{ND+f}}\;(\Varid{(\Leftrightarrow)}\;(\Varid{local2global}\;\Varid{k})))\;\Varid{t}{}\<[E]%
\\
\>[10]{}\Conid{Var}\;(\Varid{x}+\!\!+[\mskip1.5mu \mskip1.5mu],\Varid{s}){}\<[E]%
\\
\>[3]{}\mathrel{=}\mbox{\commentbegin ~  right unit of \ensuremath{(+\!\!+)}  \commentend}{}\<[E]%
\\
\>[3]{}\hsindent{3}{}\<[6]%
\>[6]{}\mathbf{do}\;{}\<[10]%
\>[10]{}(\Varid{x},\anonymous )\leftarrow \Varid{h_{State1}}\;(\Varid{h_{ND+f}}\;(\Varid{(\Leftrightarrow)}\;(\Varid{local2global}\;\Varid{k})))\;\Varid{t}{}\<[E]%
\\
\>[10]{}\Conid{Var}\;(\Varid{x},\Varid{s}){}\<[E]%
\\
\>[3]{}\mathrel{=}\mbox{\commentbegin ~  monad laws   \commentend}{}\<[E]%
\\
\>[3]{}\hsindent{3}{}\<[6]%
\>[6]{}\mathbf{do}\;{}\<[10]%
\>[10]{}(\Varid{x},\anonymous )\leftarrow \mathbf{do}\;\{\mskip1.5mu (\Varid{x},\anonymous )\leftarrow \Varid{h_{State1}}\;(\Varid{h_{ND+f}}\;(\Varid{(\Leftrightarrow)}\;(\Varid{local2global}\;\Varid{k})))\;\Varid{t};\Varid{\eta}\;(\Varid{x},\Varid{s})\mskip1.5mu\}{}\<[E]%
\\
\>[10]{}\Conid{Var}\;(\Varid{x},\Varid{s}){}\<[E]%
\\
\>[3]{}\mathrel{=}\mbox{\commentbegin ~  deriviation in reverse   \commentend}{}\<[E]%
\\
\>[3]{}\hsindent{3}{}\<[6]%
\>[6]{}\mathbf{do}\;{}\<[10]%
\>[10]{}(\Varid{x},\anonymous )\leftarrow \Varid{h_{State1}}\;(\Varid{h_{ND+f}}\;(\Varid{(\Leftrightarrow)}\;(\Varid{local2global}\;(\Conid{Op}\;(\Conid{Inl}\;(\Conid{Put}\;\Varid{t}\;\Varid{k}))))))\;\Varid{s}{}\<[E]%
\\
\>[10]{}\Conid{Var}\;(\Varid{x},\Varid{s}){}\<[E]%
\ColumnHook
\end{hscode}\resethooks
\indentend %

\noindent \mbox{\underline{case \ensuremath{\Varid{t}\mathrel{=}\Conid{Op}\;(\Conid{Inr}\;(\Conid{Inl}\;(\Conid{Or}\;\Varid{p}\;\Varid{q})))}}}\indentbegin \begin{hscode}\SaveRestoreHook
\column{B}{@{}>{\hspre}l<{\hspost}@{}}%
\column{3}{@{}>{\hspre}l<{\hspost}@{}}%
\column{6}{@{}>{\hspre}l<{\hspost}@{}}%
\column{10}{@{}>{\hspre}l<{\hspost}@{}}%
\column{12}{@{}>{\hspre}l<{\hspost}@{}}%
\column{14}{@{}>{\hspre}l<{\hspost}@{}}%
\column{16}{@{}>{\hspre}l<{\hspost}@{}}%
\column{19}{@{}>{\hspre}l<{\hspost}@{}}%
\column{E}{@{}>{\hspre}l<{\hspost}@{}}%
\>[6]{}\Varid{h_{State1}}\;(\Varid{h_{ND+f}}\;(\Varid{(\Leftrightarrow)}\;(\Varid{local2global}\;(\Conid{Op}\;(\Conid{Inr}\;(\Conid{Inl}\;(\Conid{Or}\;\Varid{p}\;\Varid{q})))))))\;\Varid{s}{}\<[E]%
\\
\>[3]{}\mathrel{=}\mbox{\commentbegin ~  definition of \ensuremath{\Varid{local2global}}   \commentend}{}\<[E]%
\\
\>[3]{}\hsindent{3}{}\<[6]%
\>[6]{}\Varid{h_{State1}}\;(\Varid{h_{ND+f}}\;(\Varid{(\Leftrightarrow)}\;(\Conid{Op}\;(\Conid{Inr}\;(\Conid{Inl}\;(\Conid{Or}\;(\Varid{local2global}\;\Varid{p})\;(\Varid{local2global}\;\Varid{q})))))))\;\Varid{s}{}\<[E]%
\\
\>[3]{}\mathrel{=}\mbox{\commentbegin ~  definition of \ensuremath{\Varid{(\Leftrightarrow)}}   \commentend}{}\<[E]%
\\
\>[3]{}\hsindent{3}{}\<[6]%
\>[6]{}\Varid{h_{State1}}\;(\Varid{h_{ND+f}}\;(\Conid{Op}\;(\Conid{Inl}\;(\Conid{Or}\;(\Varid{(\Leftrightarrow)}\;(\Varid{local2global}\;\Varid{p}))\;(\Varid{(\Leftrightarrow)}\;(\Varid{local2global}\;\Varid{q}))))))\;\Varid{s}{}\<[E]%
\\
\>[3]{}\mathrel{=}\mbox{\commentbegin ~  definition of \ensuremath{\Varid{h_{ND+f}}}   \commentend}{}\<[E]%
\\
\>[3]{}\hsindent{3}{}\<[6]%
\>[6]{}\Varid{h_{State1}}\;(\Varid{liftM2}\;(+\!\!+)\;(\Varid{h_{ND+f}}\;(\Varid{(\Leftrightarrow)}\;(\Varid{local2global}\;\Varid{p})))\;(\Varid{h_{ND+f}}\;(\Varid{(\Leftrightarrow)}\;(\Varid{local2global}\;\Varid{q}))))\;\Varid{s}{}\<[E]%
\\
\>[3]{}\mathrel{=}\mbox{\commentbegin ~  definition of \ensuremath{\Varid{liftM2}}   \commentend}{}\<[E]%
\\
\>[3]{}\hsindent{3}{}\<[6]%
\>[6]{}\Varid{h_{State1}}\;(\mathbf{do}\;{}\<[19]%
\>[19]{}\Varid{x}\leftarrow \Varid{h_{ND+f}}\;(\Varid{(\Leftrightarrow)}\;(\Varid{local2global}\;\Varid{p})){}\<[E]%
\\
\>[19]{}\Varid{y}\leftarrow \Varid{h_{ND+f}}\;(\Varid{(\Leftrightarrow)}\;(\Varid{local2global}\;\Varid{q})){}\<[E]%
\\
\>[19]{}\Conid{Var}\;(\Varid{x}+\!\!+\Varid{y}){}\<[E]%
\\
\>[6]{}\hsindent{8}{}\<[14]%
\>[14]{})\;\Varid{s}{}\<[E]%
\\
\>[3]{}\mathrel{=}\mbox{\commentbegin ~  Lemma~\ref{lemma:dist-hState1}   \commentend}{}\<[E]%
\\
\>[3]{}\hsindent{3}{}\<[6]%
\>[6]{}\mathbf{do}\;{}\<[10]%
\>[10]{}(\Varid{x},\Varid{s}_{1})\leftarrow \Varid{h_{State1}}\;(\Varid{h_{ND+f}}\;(\Varid{(\Leftrightarrow)}\;(\Varid{local2global}\;\Varid{p})))\;\Varid{s}{}\<[E]%
\\
\>[10]{}(\Varid{y},\Varid{s}_{2})\leftarrow \Varid{h_{State1}}\;(\Varid{h_{ND+f}}\;(\Varid{(\Leftrightarrow)}\;(\Varid{local2global}\;\Varid{q})))\;\Varid{s}_{1}{}\<[E]%
\\
\>[10]{}\Varid{h_{State1}}\;(\Conid{Var}\;(\Varid{x}+\!\!+\Varid{y}))\;\Varid{s}_{2}{}\<[E]%
\\
\>[3]{}\mathrel{=}\mbox{\commentbegin ~  induction hypothesis   \commentend}{}\<[E]%
\\
\>[3]{}\hsindent{3}{}\<[6]%
\>[6]{}\mathbf{do}\;{}\<[10]%
\>[10]{}(\Varid{x},\Varid{s}_{1})\leftarrow \mathbf{do}\;\{\mskip1.5mu (\Varid{x},\anonymous )\leftarrow \Varid{h_{State1}}\;(\Varid{h_{ND+f}}\;(\Varid{(\Leftrightarrow)}\;(\Varid{local2global}\;\Varid{p})))\;\Varid{s};\Varid{\eta}\;(\Varid{x},\Varid{s})\mskip1.5mu\}{}\<[E]%
\\
\>[10]{}(\Varid{y},\Varid{s}_{2})\leftarrow \mathbf{do}\;\{\mskip1.5mu (\Varid{y},\anonymous )\leftarrow \Varid{h_{State1}}\;(\Varid{h_{ND+f}}\;(\Varid{(\Leftrightarrow)}\;(\Varid{local2global}\;\Varid{q})))\;\Varid{s}_{1};\Varid{\eta}\;(\Varid{y},\Varid{s}_{1})\mskip1.5mu\}{}\<[E]%
\\
\>[10]{}\Varid{h_{State1}}\;(\Conid{Var}\;(\Varid{x}+\!\!+\Varid{y}))\;\Varid{s}_{2}{}\<[E]%
\\
\>[3]{}\mathrel{=}\mbox{\commentbegin ~  monad laws   \commentend}{}\<[E]%
\\
\>[3]{}\hsindent{3}{}\<[6]%
\>[6]{}\mathbf{do}\;{}\<[10]%
\>[10]{}(\Varid{x},\anonymous )\leftarrow \Varid{h_{State1}}\;(\Varid{h_{ND+f}}\;(\Varid{(\Leftrightarrow)}\;(\Varid{local2global}\;\Varid{p})))\;\Varid{s}{}\<[E]%
\\
\>[10]{}(\Varid{y},\anonymous )\leftarrow \Varid{h_{State1}}\;(\Varid{h_{ND+f}}\;(\Varid{(\Leftrightarrow)}\;(\Varid{local2global}\;\Varid{q})))\;\Varid{s}_{1}{}\<[E]%
\\
\>[10]{}\Varid{h_{State1}}\;(\Conid{Var}\;(\Varid{x}+\!\!+\Varid{y}))\;\Varid{s}{}\<[E]%
\\
\>[3]{}\mathrel{=}\mbox{\commentbegin ~  definition of \ensuremath{\Varid{h_{State1}}}  \commentend}{}\<[E]%
\\
\>[3]{}\hsindent{3}{}\<[6]%
\>[6]{}\mathbf{do}\;{}\<[10]%
\>[10]{}(\Varid{x},\anonymous )\leftarrow \Varid{h_{State1}}\;(\Varid{h_{ND+f}}\;(\Varid{(\Leftrightarrow)}\;(\Varid{local2global}\;\Varid{p})))\;\Varid{s}{}\<[E]%
\\
\>[10]{}(\Varid{y},\anonymous )\leftarrow \Varid{h_{State1}}\;(\Varid{h_{ND+f}}\;(\Varid{(\Leftrightarrow)}\;(\Varid{local2global}\;\Varid{q})))\;\Varid{s}{}\<[E]%
\\
\>[10]{}\Varid{\eta}\;(\Varid{x}+\!\!+\Varid{y},\Varid{s}){}\<[E]%
\\
\>[3]{}\mathrel{=}\mbox{\commentbegin ~  monad laws   \commentend}{}\<[E]%
\\
\>[3]{}\hsindent{3}{}\<[6]%
\>[6]{}\mathbf{do}\;{}\<[10]%
\>[10]{}(\Varid{x},\anonymous )\leftarrow ({}\<[E]%
\\
\>[10]{}\hsindent{2}{}\<[12]%
\>[12]{}\mathbf{do}\;{}\<[16]%
\>[16]{}(\Varid{x},\anonymous )\leftarrow \Varid{h_{State1}}\;(\Varid{h_{ND+f}}\;(\Varid{(\Leftrightarrow)}\;(\Varid{local2global}\;\Varid{p})))\;\Varid{s}{}\<[E]%
\\
\>[16]{}(\Varid{y},\anonymous )\leftarrow \Varid{h_{State1}}\;(\Varid{h_{ND+f}}\;(\Varid{(\Leftrightarrow)}\;(\Varid{local2global}\;\Varid{q})))\;\Varid{s}_{1}{}\<[E]%
\\
\>[16]{}\Varid{\eta}\;(\Varid{x}+\!\!+\Varid{y},\Varid{s}){}\<[E]%
\\
\>[10]{}\hsindent{2}{}\<[12]%
\>[12]{}){}\<[E]%
\\
\>[10]{}\Varid{\eta}\;(\Varid{x},\Varid{s}){}\<[E]%
\\
\>[3]{}\mathrel{=}\mbox{\commentbegin ~  derivation in reverse (similar to before)   \commentend}{}\<[E]%
\\
\>[3]{}\hsindent{3}{}\<[6]%
\>[6]{}\mathbf{do}\;{}\<[10]%
\>[10]{}(\Varid{x},\anonymous )\leftarrow \Varid{h_{State1}}\;(\Varid{h_{ND+f}}\;(\Varid{(\Leftrightarrow)}\;(\Varid{local2global}\;(\Conid{Op}\;(\Conid{Inr}\;(\Conid{Inl}\;(\Conid{Or}\;\Varid{p}\;\Varid{q})))))))\;\Varid{s}{}\<[E]%
\\
\>[10]{}\Varid{\eta}\;(\Varid{x},\Varid{s}){}\<[E]%
\ColumnHook
\end{hscode}\resethooks
\indentend \noindent \mbox{\underline{case \ensuremath{\Varid{t}\mathrel{=}\Conid{Op}\;(\Conid{Inr}\;(\Conid{Inr}\;\Varid{y}))}}}\indentbegin \begin{hscode}\SaveRestoreHook
\column{B}{@{}>{\hspre}l<{\hspost}@{}}%
\column{3}{@{}>{\hspre}l<{\hspost}@{}}%
\column{6}{@{}>{\hspre}l<{\hspost}@{}}%
\column{10}{@{}>{\hspre}l<{\hspost}@{}}%
\column{E}{@{}>{\hspre}l<{\hspost}@{}}%
\>[6]{}\Varid{h_{State1}}\;(\Varid{h_{ND+f}}\;(\Varid{(\Leftrightarrow)}\;(\Varid{local2global}\;(\Conid{Op}\;(\Conid{Inr}\;(\Conid{Inr}\;\Varid{y}))))))\;\Varid{s}{}\<[E]%
\\
\>[3]{}\mathrel{=}\mbox{\commentbegin ~  definition of \ensuremath{\Varid{local2global}}   \commentend}{}\<[E]%
\\
\>[3]{}\hsindent{3}{}\<[6]%
\>[6]{}\Varid{h_{State1}}\;(\Varid{h_{ND+f}}\;(\Varid{(\Leftrightarrow)}\;(\Conid{Op}\;(\Conid{Inr}\;(\Conid{Inr}\;(\Varid{fmap}\;\Varid{local2global}\;\Varid{y}))))))\;\Varid{s}{}\<[E]%
\\
\>[3]{}\mathrel{=}\mbox{\commentbegin ~  definition of \ensuremath{\Varid{(\Leftrightarrow)}}; \ensuremath{\Varid{fmap}} fusion   \commentend}{}\<[E]%
\\
\>[3]{}\hsindent{3}{}\<[6]%
\>[6]{}\Varid{h_{State1}}\;(\Varid{h_{ND+f}}\;(\Conid{Op}\;(\Conid{Inr}\;(\Conid{Inr}\;(\Varid{fmap}\;(\Varid{(\Leftrightarrow)}\hsdot{\circ }{.}\Varid{local2global})\;\Varid{y})))))\;\Varid{s}{}\<[E]%
\\
\>[3]{}\mathrel{=}\mbox{\commentbegin ~  definition of \ensuremath{\Varid{h_{ND+f}}}; \ensuremath{\Varid{fmap}} fusion   \commentend}{}\<[E]%
\\
\>[3]{}\hsindent{3}{}\<[6]%
\>[6]{}\Varid{h_{State1}}\;(\Conid{Op}\;(\Conid{Inr}\;(\Varid{fmap}\;(\Varid{h_{ND+f}}\hsdot{\circ }{.}\Varid{(\Leftrightarrow)}\hsdot{\circ }{.}\Varid{local2global})\;\Varid{y})))\;\Varid{s}{}\<[E]%
\\
\>[3]{}\mathrel{=}\mbox{\commentbegin ~  definition of \ensuremath{\Varid{h_{State1}}}; \ensuremath{\Varid{fmap}} fusion   \commentend}{}\<[E]%
\\
\>[3]{}\hsindent{3}{}\<[6]%
\>[6]{}\Conid{Op}\;(\Varid{fmap}\;((\mathbin{\$}\Varid{s})\hsdot{\circ }{.}\Varid{h_{State1}}\hsdot{\circ }{.}\Varid{h_{ND+f}}\hsdot{\circ }{.}\Varid{(\Leftrightarrow)}\hsdot{\circ }{.}\Varid{local2global})\;\Varid{y}){}\<[E]%
\\
\>[3]{}\mathrel{=}\mbox{\commentbegin ~  induction hypothesis   \commentend}{}\<[E]%
\\
\>[3]{}\hsindent{3}{}\<[6]%
\>[6]{}\Conid{Op}\;(\Varid{fmap}\;((>\!\!>\!\!=\lambda (\Varid{x},\anonymous )\to \Varid{\eta}\;(\Varid{x},\Varid{s}))\hsdot{\circ }{.}(\mathbin{\$}\Varid{s})\hsdot{\circ }{.}\Varid{h_{State1}}\hsdot{\circ }{.}\Varid{h_{ND+f}}\hsdot{\circ }{.}\Varid{(\Leftrightarrow)}\hsdot{\circ }{.}\Varid{local2global})\;\Varid{y}){}\<[E]%
\\
\>[3]{}\mathrel{=}\mbox{\commentbegin ~  \ensuremath{\Varid{fmap}} fission; definition of \ensuremath{(>\!\!>\!\!=)}   \commentend}{}\<[E]%
\\
\>[3]{}\hsindent{3}{}\<[6]%
\>[6]{}\mathbf{do}\;{}\<[10]%
\>[10]{}(\Varid{x},\anonymous )\leftarrow \Conid{Op}\;(\Varid{fmap}\;((\mathbin{\$}\Varid{s})\hsdot{\circ }{.}\Varid{h_{State1}}\hsdot{\circ }{.}\Varid{h_{ND+f}}\hsdot{\circ }{.}\Varid{(\Leftrightarrow)}\hsdot{\circ }{.}\Varid{local2global})\;\Varid{y}){}\<[E]%
\\
\>[10]{}\Varid{\eta}\;(\Varid{x},\Varid{s}){}\<[E]%
\\
\>[3]{}\mathrel{=}\mbox{\commentbegin ~  deriviation in reverse (similar to before)   \commentend}{}\<[E]%
\\
\>[3]{}\hsindent{3}{}\<[6]%
\>[6]{}\mathbf{do}\;{}\<[10]%
\>[10]{}(\Varid{x},\anonymous )\leftarrow \Varid{h_{State1}}\;(\Varid{h_{ND+f}}\;(\Varid{(\Leftrightarrow)}\;(\Varid{local2global}\;(\Conid{Op}\;(\Conid{Inr}\;(\Conid{Inr}\;\Varid{y}))))))\;\Varid{s}{}\<[E]%
\\
\>[10]{}\Varid{\eta}\;(\Varid{x},\Varid{s}){}\<[E]%
\ColumnHook
\end{hscode}\resethooks
\indentend \end{proof}

\subsection{Auxiliary Lemmas}

The derivations above make use of two auxliary lemmas.
We prove them here.

\begin{lemma}[Naturality of \ensuremath{(\mathbin{\$}\Varid{s})}] \label{eq:comm-app-fmap}\indentbegin \begin{hscode}\SaveRestoreHook
\column{B}{@{}>{\hspre}l<{\hspost}@{}}%
\column{3}{@{}>{\hspre}l<{\hspost}@{}}%
\column{E}{@{}>{\hspre}l<{\hspost}@{}}%
\>[3]{}(\mathbin{\$}\Varid{x})\hsdot{\circ }{.}\Varid{fmap}\;\Varid{f}\mathrel{=}\Varid{f}\hsdot{\circ }{.}(\mathbin{\$}\Varid{x}){}\<[E]%
\ColumnHook
\end{hscode}\resethooks
\indentend \end{lemma}
\begin{proof}~\indentbegin \begin{hscode}\SaveRestoreHook
\column{B}{@{}>{\hspre}l<{\hspost}@{}}%
\column{3}{@{}>{\hspre}l<{\hspost}@{}}%
\column{5}{@{}>{\hspre}l<{\hspost}@{}}%
\column{E}{@{}>{\hspre}l<{\hspost}@{}}%
\>[5]{}((\mathbin{\$}\Varid{x})\hsdot{\circ }{.}\Varid{fmap}\;\Varid{f})\;\Varid{m}{}\<[E]%
\\
\>[3]{}\mathrel{=}\mbox{\commentbegin ~ function application  \commentend}{}\<[E]%
\\
\>[3]{}\hsindent{2}{}\<[5]%
\>[5]{}(\Varid{fmap}\;\Varid{f}\;\Varid{m})\;\Varid{x}{}\<[E]%
\\
\>[3]{}\mathrel{=}\mbox{\commentbegin ~ eta-expansion  \commentend}{}\<[E]%
\\
\>[3]{}\hsindent{2}{}\<[5]%
\>[5]{}(\Varid{fmap}\;\Varid{f}\;(\lambda \Varid{y}\hsdot{\circ }{.}\Varid{m}\;\Varid{y}))\;\Varid{x}{}\<[E]%
\\
\>[3]{}\mathrel{=}\mbox{\commentbegin ~ definition of \ensuremath{\Varid{fmap}}  \commentend}{}\<[E]%
\\
\>[3]{}\hsindent{2}{}\<[5]%
\>[5]{}(\lambda \Varid{y}\hsdot{\circ }{.}\Varid{f}\;(\Varid{m}\;\Varid{y}))\;\Varid{x}{}\<[E]%
\\
\>[3]{}\mathrel{=}\mbox{\commentbegin ~ function application  \commentend}{}\<[E]%
\\
\>[3]{}\hsindent{2}{}\<[5]%
\>[5]{}\Varid{f}\;(\Varid{m}\;\Varid{x}){}\<[E]%
\\
\>[3]{}\mathrel{=}\mbox{\commentbegin ~ definition of \ensuremath{\hsdot{\circ }{.}} and \ensuremath{\mathbin{\$}}  \commentend}{}\<[E]%
\\
\>[3]{}\hsindent{2}{}\<[5]%
\>[5]{}(\Varid{f}\hsdot{\circ }{.}(\mathbin{\$}\Varid{x}))\;\Varid{m}{}\<[E]%
\ColumnHook
\end{hscode}\resethooks
\indentend \end{proof}

\begin{lemma} \label{eq:liftM2-fst-comm}~\indentbegin \begin{hscode}\SaveRestoreHook
\column{B}{@{}>{\hspre}l<{\hspost}@{}}%
\column{3}{@{}>{\hspre}l<{\hspost}@{}}%
\column{E}{@{}>{\hspre}l<{\hspost}@{}}%
\>[3]{}\Varid{fmap}\;(\Varid{fmap}\;\Varid{fst})\;(\Varid{liftM2}\;(+\!\!+)\;\Varid{p}\;\Varid{q})\mathrel{=}\Varid{liftM2}\;(+\!\!+)\;(\Varid{fmap}\;(\Varid{fmap}\;\Varid{fst})\;\Varid{p})\;(\Varid{fmap}\;(\Varid{fmap}\;\Varid{fst})\;\Varid{q}){}\<[E]%
\ColumnHook
\end{hscode}\resethooks
\indentend \end{lemma}
\begin{proof}~
\indentbegin \begin{hscode}\SaveRestoreHook
\column{B}{@{}>{\hspre}l<{\hspost}@{}}%
\column{3}{@{}>{\hspre}l<{\hspost}@{}}%
\column{6}{@{}>{\hspre}l<{\hspost}@{}}%
\column{E}{@{}>{\hspre}l<{\hspost}@{}}%
\>[6]{}\Varid{fmap}\;(\Varid{fmap}\;\Varid{fst})\;(\Varid{liftM2}\;(+\!\!+)\;\Varid{p}\;\Varid{q}){}\<[E]%
\\
\>[3]{}\mathrel{=}\mbox{\commentbegin ~  definition of \ensuremath{\Varid{liftM2}}  \commentend}{}\<[E]%
\\
\>[3]{}\hsindent{3}{}\<[6]%
\>[6]{}\Varid{fmap}\;(\Varid{fmap}\;\Varid{fst})\;(\mathbf{do}\;\{\mskip1.5mu \Varid{x}\leftarrow \Varid{p};\Varid{y}\leftarrow \Varid{q};\Varid{\eta}\;(\Varid{x}+\!\!+\Varid{y})\mskip1.5mu\}){}\<[E]%
\\
\>[3]{}\mathrel{=}\mbox{\commentbegin ~  derived property for monad: \ensuremath{\Varid{fmap}\;\Varid{f}\;(\Varid{m}>\!\!>\!\!=\Varid{k})\mathrel{=}\Varid{m}>\!\!>\!\!=\Varid{fmap}\;\Varid{f}\hsdot{\circ }{.}\Varid{k}}  \commentend}{}\<[E]%
\\
\>[3]{}\hsindent{3}{}\<[6]%
\>[6]{}\mathbf{do}\;\{\mskip1.5mu \Varid{x}\leftarrow \Varid{p};\Varid{y}\leftarrow \Varid{q};\Varid{fmap}\;(\Varid{fmap}\;\Varid{fst})\;(\Varid{\eta}\;(\Varid{x}+\!\!+\Varid{y}))\mskip1.5mu\}{}\<[E]%
\\
\>[3]{}\mathrel{=}\mbox{\commentbegin ~  definition of \ensuremath{\Varid{fmap}}  \commentend}{}\<[E]%
\\
\>[3]{}\hsindent{3}{}\<[6]%
\>[6]{}\mathbf{do}\;\{\mskip1.5mu \Varid{x}\leftarrow \Varid{p};\Varid{y}\leftarrow \Varid{q};\Varid{\eta}\;(\Varid{fmap}\;\Varid{fst}\;(\Varid{x}+\!\!+\Varid{y}))\mskip1.5mu\}{}\<[E]%
\\
\>[3]{}\mathrel{=}\mbox{\commentbegin ~  naturality of \ensuremath{(+\!\!+)}  \commentend}{}\<[E]%
\\
\>[3]{}\hsindent{3}{}\<[6]%
\>[6]{}\mathbf{do}\;\{\mskip1.5mu \Varid{x}\leftarrow \Varid{p};\Varid{y}\leftarrow \Varid{q};\Varid{\eta}\;((\Varid{fmap}\;\Varid{fst}\;\Varid{x})+\!\!+(\Varid{fmap}\;\Varid{fst}\;\Varid{y}))\mskip1.5mu\}{}\<[E]%
\\
\>[3]{}\mathrel{=}\mbox{\commentbegin ~  monad left unit law (twice)  \commentend}{}\<[E]%
\\
\>[3]{}\hsindent{3}{}\<[6]%
\>[6]{}\mathbf{do}\;\{\mskip1.5mu \Varid{x}\leftarrow \Varid{p};\Varid{x'}\leftarrow \Varid{\eta}\;(\Varid{fmap}\;\Varid{fst}\;\Varid{x});\Varid{y}\leftarrow \Varid{q};\Varid{y'}\leftarrow \Varid{\eta}\;(\Varid{fmap}\;\Varid{fst}\;\Varid{y})\;\Varid{\eta}\;(\Varid{x'}+\!\!+\Varid{y'})\mskip1.5mu\}{}\<[E]%
\\
\>[3]{}\mathrel{=}\mbox{\commentbegin ~  definition of \ensuremath{\Varid{fmap}}  \commentend}{}\<[E]%
\\
\>[3]{}\hsindent{3}{}\<[6]%
\>[6]{}\mathbf{do}\;\{\mskip1.5mu \Varid{x}\leftarrow \Varid{fmap}\;(\Varid{fmap}\;\Varid{fst})\;\Varid{p};\Varid{y}\leftarrow \Varid{fmap}\;(\Varid{fmap}\;\Varid{fst})\;\Varid{q};\Varid{\eta}\;(\Varid{x}+\!\!+\Varid{y})\mskip1.5mu\}{}\<[E]%
\\
\>[3]{}\mathrel{=}\mbox{\commentbegin ~  definition of \ensuremath{\Varid{liftM2}}  \commentend}{}\<[E]%
\\
\>[3]{}\hsindent{3}{}\<[6]%
\>[6]{}\Varid{liftM2}\;(+\!\!+)\;(\Varid{fmap}\;(\Varid{fmap}\;\Varid{fst})\;\Varid{p})\;(\Varid{fmap}\;(\Varid{fmap}\;\Varid{fst})\;\Varid{q}){}\<[E]%
\ColumnHook
\end{hscode}\resethooks
\indentend \end{proof}

\begin{lemma}[Distributivity of \ensuremath{\Varid{h_{State1}}}] \label{lemma:dist-hState1} \ \\\indentbegin \begin{hscode}\SaveRestoreHook
\column{B}{@{}>{\hspre}l<{\hspost}@{}}%
\column{3}{@{}>{\hspre}l<{\hspost}@{}}%
\column{E}{@{}>{\hspre}l<{\hspost}@{}}%
\>[3]{}\Varid{h_{State1}}\;(\Varid{p}>\!\!>\!\!=\Varid{k})\;\Varid{s}\mathrel{=}\Varid{h_{State1}}\;\Varid{p}\;\Varid{s}>\!\!>\!\!=\lambda (\Varid{x},\Varid{s'})\to \Varid{h_{State1}}\;(\Varid{k}\;\Varid{x})\;\Varid{s'}{}\<[E]%
\ColumnHook
\end{hscode}\resethooks
\indentend %
\end{lemma}

\begin{proof}
The proof proceeds by induction on \ensuremath{\Varid{p}}.

\noindent \mbox{\underline{case \ensuremath{\Varid{p}\mathrel{=}\Conid{Var}\;\Varid{x}}}}
\indentbegin \begin{hscode}\SaveRestoreHook
\column{B}{@{}>{\hspre}l<{\hspost}@{}}%
\column{3}{@{}>{\hspre}l<{\hspost}@{}}%
\column{6}{@{}>{\hspre}l<{\hspost}@{}}%
\column{E}{@{}>{\hspre}l<{\hspost}@{}}%
\>[6]{}\Varid{h_{State1}}\;(\Conid{Var}\;\Varid{x}>\!\!>\!\!=\Varid{k})\;\Varid{s}{}\<[E]%
\\
\>[3]{}\mathrel{=}\mbox{\commentbegin ~  monad law   \commentend}{}\<[E]%
\\
\>[3]{}\hsindent{3}{}\<[6]%
\>[6]{}\Varid{h_{State1}}\;(\Varid{k}\;\Varid{x})\;\Varid{s}{}\<[E]%
\\
\>[3]{}\mathrel{=}\mbox{\commentbegin ~  monad law   \commentend}{}\<[E]%
\\
\>[3]{}\hsindent{3}{}\<[6]%
\>[6]{}\Varid{\eta}\;(\Varid{x},\Varid{s})>\!\!>\!\!=\lambda (\Varid{x},\Varid{s'})\to \Varid{h_{State1}}\;(\Varid{k}\;\Varid{x})\;\Varid{s'}{}\<[E]%
\\
\>[3]{}\mathrel{=}\mbox{\commentbegin ~  definition of \ensuremath{\Varid{h_{State1}}}   \commentend}{}\<[E]%
\\
\>[3]{}\hsindent{3}{}\<[6]%
\>[6]{}\Varid{h_{State1}}\;(\Conid{Var}\;\Varid{x})\;\Varid{s}>\!\!>\!\!=\lambda (\Varid{x},\Varid{s'})\to \Varid{h_{State1}}\;(\Varid{k}\;\Varid{x})\;\Varid{s'}{}\<[E]%
\ColumnHook
\end{hscode}\resethooks
\indentend \noindent \mbox{\underline{case \ensuremath{\Varid{p}\mathrel{=}\Conid{Op}\;(\Conid{Inl}\;(\Conid{Get}\;\Varid{p}))}}}
\indentbegin \begin{hscode}\SaveRestoreHook
\column{B}{@{}>{\hspre}l<{\hspost}@{}}%
\column{3}{@{}>{\hspre}l<{\hspost}@{}}%
\column{6}{@{}>{\hspre}l<{\hspost}@{}}%
\column{E}{@{}>{\hspre}l<{\hspost}@{}}%
\>[6]{}\Varid{h_{State1}}\;(\Conid{Op}\;(\Conid{Inl}\;(\Conid{Get}\;\Varid{p}))>\!\!>\!\!=\Varid{k})\;\Varid{s}{}\<[E]%
\\
\>[3]{}\mathrel{=}\mbox{\commentbegin ~  definition of \ensuremath{(>\!\!>\!\!=)} for free monad   \commentend}{}\<[E]%
\\
\>[3]{}\hsindent{3}{}\<[6]%
\>[6]{}\Varid{h_{State1}}\;(\Conid{Op}\;(\Varid{fmap}\;(>\!\!>\!\!=\Varid{k})\;(\Conid{Inl}\;(\Conid{Get}\;\Varid{p}))))\;\Varid{s}{}\<[E]%
\\
\>[3]{}\mathrel{=}\mbox{\commentbegin ~  definition of \ensuremath{\Varid{fmap}} for coproduct \ensuremath{(\mathrel{{:}{+}{:}})}   \commentend}{}\<[E]%
\\
\>[3]{}\hsindent{3}{}\<[6]%
\>[6]{}\Varid{h_{State1}}\;(\Conid{Op}\;(\Conid{Inl}\;(\Varid{fmap}\;(>\!\!>\!\!=\Varid{k})\;(\Conid{Get}\;\Varid{p}))))\;\Varid{s}{}\<[E]%
\\
\>[3]{}\mathrel{=}\mbox{\commentbegin ~  definition of \ensuremath{\Varid{fmap}} for \ensuremath{\Conid{Get}}   \commentend}{}\<[E]%
\\
\>[3]{}\hsindent{3}{}\<[6]%
\>[6]{}\Varid{h_{State1}}\;(\Conid{Op}\;(\Conid{Inl}\;(\Conid{Get}\;(\lambda \Varid{x}\to \Varid{p}\;\Varid{s}>\!\!>\!\!=\Varid{k}))))\;\Varid{s}{}\<[E]%
\\
\>[3]{}\mathrel{=}\mbox{\commentbegin ~  definition of \ensuremath{\Varid{h_{State1}}}   \commentend}{}\<[E]%
\\
\>[3]{}\hsindent{3}{}\<[6]%
\>[6]{}\Varid{h_{State1}}\;(\Varid{p}\;\Varid{s}>\!\!>\!\!=\Varid{k})\;\Varid{s}{}\<[E]%
\\
\>[3]{}\mathrel{=}\mbox{\commentbegin ~  induction hypothesis   \commentend}{}\<[E]%
\\
\>[3]{}\hsindent{3}{}\<[6]%
\>[6]{}\Varid{h_{State1}}\;(\Varid{p}\;\Varid{s})\;\Varid{s}>\!\!>\!\!=\lambda (\Varid{x},\Varid{s'})\to \Varid{h_{State1}}\;(\Varid{k}\;\Varid{x})\;\Varid{s'}{}\<[E]%
\\
\>[3]{}\mathrel{=}\mbox{\commentbegin ~  definition of \ensuremath{\Varid{h_{State1}}}   \commentend}{}\<[E]%
\\
\>[3]{}\hsindent{3}{}\<[6]%
\>[6]{}\Varid{h_{State1}}\;(\Conid{Op}\;(\Conid{Inl}\;(\Conid{Get}\;\Varid{p})))\;\Varid{s}>\!\!>\!\!=\lambda (\Varid{x},\Varid{s'})\to \Varid{h_{State1}}\;(\Varid{k}\;\Varid{x})\;\Varid{s'}{}\<[E]%
\ColumnHook
\end{hscode}\resethooks
\indentend \noindent \mbox{\underline{case \ensuremath{\Varid{p}\mathrel{=}\Conid{Op}\;(\Conid{Inl}\;(\Conid{Put}\;\Varid{t}\;\Varid{p}))}}}
\indentbegin \begin{hscode}\SaveRestoreHook
\column{B}{@{}>{\hspre}l<{\hspost}@{}}%
\column{3}{@{}>{\hspre}l<{\hspost}@{}}%
\column{6}{@{}>{\hspre}l<{\hspost}@{}}%
\column{E}{@{}>{\hspre}l<{\hspost}@{}}%
\>[6]{}\Varid{h_{State1}}\;(\Conid{Op}\;(\Conid{Inl}\;(\Conid{Put}\;\Varid{t}\;\Varid{p}))>\!\!>\!\!=\Varid{k})\;\Varid{s}{}\<[E]%
\\
\>[3]{}\mathrel{=}\mbox{\commentbegin ~  definition of \ensuremath{(>\!\!>\!\!=)} for free monad   \commentend}{}\<[E]%
\\
\>[3]{}\hsindent{3}{}\<[6]%
\>[6]{}\Varid{h_{State1}}\;(\Conid{Op}\;(\Varid{fmap}\;(>\!\!>\!\!=\Varid{k})\;(\Conid{Inl}\;(\Conid{Put}\;\Varid{t}\;\Varid{p}))))\;\Varid{s}{}\<[E]%
\\
\>[3]{}\mathrel{=}\mbox{\commentbegin ~  definition of \ensuremath{\Varid{fmap}} for coproduct \ensuremath{(\mathrel{{:}{+}{:}})}   \commentend}{}\<[E]%
\\
\>[3]{}\hsindent{3}{}\<[6]%
\>[6]{}\Varid{h_{State1}}\;(\Conid{Op}\;(\Conid{Inl}\;(\Varid{fmap}\;(>\!\!>\!\!=\Varid{k})\;(\Conid{Put}\;\Varid{t}\;\Varid{p}))))\;\Varid{s}{}\<[E]%
\\
\>[3]{}\mathrel{=}\mbox{\commentbegin ~  definition of \ensuremath{\Varid{fmap}} for \ensuremath{\Conid{Put}}   \commentend}{}\<[E]%
\\
\>[3]{}\hsindent{3}{}\<[6]%
\>[6]{}\Varid{h_{State1}}\;(\Conid{Op}\;(\Conid{Inl}\;(\Conid{Put}\;\Varid{t}\;(\Varid{p}>\!\!>\!\!=\Varid{k}))))\;\Varid{s}{}\<[E]%
\\
\>[3]{}\mathrel{=}\mbox{\commentbegin ~  definition of \ensuremath{\Varid{h_{State1}}}   \commentend}{}\<[E]%
\\
\>[3]{}\hsindent{3}{}\<[6]%
\>[6]{}\Varid{h_{State1}}\;(\Varid{p}>\!\!>\!\!=\Varid{k})\;\Varid{t}{}\<[E]%
\\
\>[3]{}\mathrel{=}\mbox{\commentbegin ~  induction hypothesis   \commentend}{}\<[E]%
\\
\>[3]{}\hsindent{3}{}\<[6]%
\>[6]{}\Varid{h_{State1}}\;\Varid{p}\;\Varid{t}>\!\!>\!\!=\lambda (\Varid{x},\Varid{s'})\to \Varid{h_{State1}}\;(\Varid{k}\;\Varid{x})\;\Varid{s'}{}\<[E]%
\\
\>[3]{}\mathrel{=}\mbox{\commentbegin ~  definition of \ensuremath{\Varid{h_{State1}}}   \commentend}{}\<[E]%
\\
\>[3]{}\hsindent{3}{}\<[6]%
\>[6]{}\Varid{h_{State1}}\;(\Conid{Op}\;(\Conid{Inl}\;(\Conid{Put}\;\Varid{t}\;\Varid{p})))\;\Varid{s}>\!\!>\!\!=\lambda (\Varid{x},\Varid{s'})\to \Varid{h_{State1}}\;(\Varid{k}\;\Varid{x})\;\Varid{s'}{}\<[E]%
\ColumnHook
\end{hscode}\resethooks
\indentend \noindent \mbox{\underline{case \ensuremath{\Varid{p}\mathrel{=}\Conid{Op}\;(\Conid{Inr}\;\Varid{y})}}}
\indentbegin \begin{hscode}\SaveRestoreHook
\column{B}{@{}>{\hspre}l<{\hspost}@{}}%
\column{3}{@{}>{\hspre}l<{\hspost}@{}}%
\column{6}{@{}>{\hspre}l<{\hspost}@{}}%
\column{E}{@{}>{\hspre}l<{\hspost}@{}}%
\>[6]{}\Varid{h_{State1}}\;(\Conid{Op}\;(\Conid{Inr}\;\Varid{y})>\!\!>\!\!=\Varid{k})\;\Varid{s}{}\<[E]%
\\
\>[3]{}\mathrel{=}\mbox{\commentbegin ~  definition of \ensuremath{(>\!\!>\!\!=)} for free monad   \commentend}{}\<[E]%
\\
\>[3]{}\hsindent{3}{}\<[6]%
\>[6]{}\Varid{h_{State1}}\;(\Conid{Op}\;(\Varid{fmap}\;(>\!\!>\!\!=\Varid{k})\;(\Conid{Inr}\;\Varid{y})))\;\Varid{s}{}\<[E]%
\\
\>[3]{}\mathrel{=}\mbox{\commentbegin ~  definition of \ensuremath{\Varid{fmap}} for coproduct \ensuremath{(\mathrel{{:}{+}{:}})}   \commentend}{}\<[E]%
\\
\>[3]{}\hsindent{3}{}\<[6]%
\>[6]{}\Varid{h_{State1}}\;(\Conid{Op}\;(\Conid{Inr}\;(\Varid{fmap}\;(>\!\!>\!\!=\Varid{k})\;\Varid{y})))\;\Varid{s}{}\<[E]%
\\
\>[3]{}\mathrel{=}\mbox{\commentbegin ~  definition of \ensuremath{\Varid{h_{State1}}}   \commentend}{}\<[E]%
\\
\>[3]{}\hsindent{3}{}\<[6]%
\>[6]{}\Conid{Op}\;(\Varid{fmap}\;(\lambda \Varid{x}\to \Varid{h_{State1}}\;\Varid{x}\;\Varid{s})\;(\Varid{fmap}\;(>\!\!>\!\!=\Varid{k})\;\Varid{y})){}\<[E]%
\\
\>[3]{}\mathrel{=}\mbox{\commentbegin ~  \ensuremath{\Varid{fmap}} fusion   \commentend}{}\<[E]%
\\
\>[3]{}\hsindent{3}{}\<[6]%
\>[6]{}\Conid{Op}\;(\Varid{fmap}\;((\lambda \Varid{x}\to \Varid{h_{State1}}\;(\Varid{x}>\!\!>\!\!=\Varid{k})\;\Varid{s}))\;\Varid{y}){}\<[E]%
\\
\>[3]{}\mathrel{=}\mbox{\commentbegin ~  induction hypothesis   \commentend}{}\<[E]%
\\
\>[3]{}\hsindent{3}{}\<[6]%
\>[6]{}\Conid{Op}\;(\Varid{fmap}\;(\lambda \Varid{x}\to \Varid{h_{State1}}\;\Varid{x}\;\Varid{s}>\!\!>\!\!=\lambda (\Varid{x'},\Varid{s'})\to \Varid{h_{State1}}\;(\Varid{k}\;\Varid{x'})\;\Varid{s'})\;\Varid{y}){}\<[E]%
\\
\>[3]{}\mathrel{=}\mbox{\commentbegin ~  \ensuremath{\Varid{fmap}} fission  \commentend}{}\<[E]%
\\
\>[3]{}\hsindent{3}{}\<[6]%
\>[6]{}\Conid{Op}\;(\Varid{fmap}\;(\lambda \Varid{x}\to \Varid{x}>\!\!>\!\!=\lambda (\Varid{x'},\Varid{s'})\to \Varid{h_{State1}}\;(\Varid{k}\;\Varid{x'})\;\Varid{s'})\;(\Varid{fmap}\;(\lambda \Varid{x}\to \Varid{h_{State1}}\;\Varid{x}\;\Varid{s})\;\Varid{y})){}\<[E]%
\\
\>[3]{}\mathrel{=}\mbox{\commentbegin ~  definition of \ensuremath{(>\!\!>\!\!=)}  \commentend}{}\<[E]%
\\
\>[3]{}\hsindent{3}{}\<[6]%
\>[6]{}\Conid{Op}\;((\Varid{fmap}\;(\lambda \Varid{x}\to \Varid{h_{State1}}\;\Varid{x}\;\Varid{s})\;\Varid{y}))>\!\!>\!\!=\lambda (\Varid{x'},\Varid{s'})\to \Varid{h_{State1}}\;(\Varid{k}\;\Varid{x'})\;\Varid{s'}{}\<[E]%
\\
\>[3]{}\mathrel{=}\mbox{\commentbegin ~  definition of \ensuremath{\Varid{h_{State1}}}   \commentend}{}\<[E]%
\\
\>[3]{}\hsindent{3}{}\<[6]%
\>[6]{}\Conid{Op}\;(\Conid{Inr}\;\Varid{y})\;\Varid{s}>\!\!>\!\!=\lambda (\Varid{x'},\Varid{s'})\to \Varid{h_{State1}}\;(\Varid{k}\;\Varid{x'})\;\Varid{s'}{}\<[E]%
\ColumnHook
\end{hscode}\resethooks
\indentend \end{proof}

import Data.Bitraversable (Bitraversable)
\section{Proofs for Modelling Nondeterminism with State}
\label{app:nondet-state}

In this secton, we prove the theorems in
\Cref{sec:nondeterminism-state}.

\subsection{Only Nondeterminism}
\label{app:runnd-hnd}

This section proves the following theorem in
\Cref{sec:sim-nondet-state}.

\nondetStateS*

\begin{proof}
We start with expanding the definition of \ensuremath{\Varid{run_{ND}}}:\indentbegin \begin{hscode}\SaveRestoreHook
\column{B}{@{}>{\hspre}l<{\hspost}@{}}%
\column{3}{@{}>{\hspre}l<{\hspost}@{}}%
\column{E}{@{}>{\hspre}l<{\hspost}@{}}%
\>[3]{}\Varid{extract_S}\hsdot{\circ }{.}\Varid{h_{State}^\prime}\hsdot{\circ }{.}\Varid{nondet2state_S}\mathrel{=}\Varid{h_{ND}}{}\<[E]%
\ColumnHook
\end{hscode}\resethooks
\indentend Both \ensuremath{\Varid{nondet2state_S}} and \ensuremath{\Varid{h_{ND}}} are written as folds.
We use the fold fusion law {\bf fusion-post'}~(\ref{eq:fusion-post-strong}) to
fuse the left-hand side.
Since the right-hand side is already a fold, to prove the equation we
just need to check the components of the fold \ensuremath{\Varid{h_{ND}}} satisfy the
conditions of the fold fusion, i.e., the following two equations:

\[\ba{rl}
    &\ensuremath{(\Varid{extract_S}\hsdot{\circ }{.}\Varid{h_{State}^\prime})\hsdot{\circ }{.}\Varid{gen}\mathrel{=}\Varid{gen_{ND}}} \\
    &\ensuremath{(\Varid{extract_S}\hsdot{\circ }{.}\Varid{h_{State}^\prime})\hsdot{\circ }{.}\Varid{alg}\hsdot{\circ }{.}\Varid{fmap}\;\Varid{nondet2state_S}}\\
 \ensuremath{\mathrel{=}}&  \ensuremath{\Varid{alg_{ND}}\hsdot{\circ }{.}\Varid{fmap}\;(\Varid{extract_S}\hsdot{\circ }{.}\Varid{h_{State}^\prime})\hsdot{\circ }{.}\Varid{fmap}\;\Varid{nondet2state_S}}
\ea\]

For brevity, we omit the last common part \ensuremath{\Varid{fmap}\;\Varid{nondet2state_S}} of
the second equation in the following proof. Instead, we assume that
the input is in the codomain of \ensuremath{\Varid{fmap}\;\Varid{nondet2state_S}}.

For the first equation, we calculate as follows:
\indentbegin \begin{hscode}\SaveRestoreHook
\column{B}{@{}>{\hspre}l<{\hspost}@{}}%
\column{3}{@{}>{\hspre}l<{\hspost}@{}}%
\column{6}{@{}>{\hspre}l<{\hspost}@{}}%
\column{E}{@{}>{\hspre}l<{\hspost}@{}}%
\>[6]{}\Varid{extract_S}\;(\Varid{h_{State}^\prime}\;(\Varid{gen}\;\Varid{x})){}\<[E]%
\\
\>[3]{}\mathrel{=}\mbox{\commentbegin ~  definition of \ensuremath{\Varid{gen}}   \commentend}{}\<[E]%
\\
\>[3]{}\hsindent{3}{}\<[6]%
\>[6]{}\Varid{extract_S}\;(\Varid{h_{State}^\prime}\;(\Varid{append_S}\;\Varid{x}\;\Varid{pop_S})){}\<[E]%
\\
\>[3]{}\mathrel{=}\mbox{\commentbegin ~  definition of \ensuremath{\Varid{extract_S}}   \commentend}{}\<[E]%
\\
\>[3]{}\hsindent{3}{}\<[6]%
\>[6]{}\Varid{results}\hsdot{\circ }{.}\Varid{snd}\mathbin{\$}\Varid{run_{State}}\;(\Varid{h_{State}^\prime}\;(\Varid{append_S}\;\Varid{x}\;\Varid{pop_S}))\;(\Conid{S}\;[\mskip1.5mu \mskip1.5mu]\;[\mskip1.5mu \mskip1.5mu]){}\<[E]%
\\
\>[3]{}\mathrel{=}\mbox{\commentbegin ~  \Cref{eq:eval-append}   \commentend}{}\<[E]%
\\
\>[3]{}\hsindent{3}{}\<[6]%
\>[6]{}\Varid{results}\hsdot{\circ }{.}\Varid{snd}\mathbin{\$}\Varid{run_{State}}\;(\Varid{h_{State}^\prime}\;\Varid{pop_S})\;(\Conid{S}\;([\mskip1.5mu \mskip1.5mu]+\!\!+[\mskip1.5mu \Varid{x}\mskip1.5mu])\;[\mskip1.5mu \mskip1.5mu]){}\<[E]%
\\
\>[3]{}\mathrel{=}\mbox{\commentbegin ~  definition of \ensuremath{(+\!\!+)}   \commentend}{}\<[E]%
\\
\>[3]{}\hsindent{3}{}\<[6]%
\>[6]{}\Varid{results}\hsdot{\circ }{.}\Varid{snd}\mathbin{\$}\Varid{run_{State}}\;(\Varid{h_{State}^\prime}\;\Varid{pop_S})\;(\Conid{S}\;[\mskip1.5mu \Varid{x}\mskip1.5mu]\;[\mskip1.5mu \mskip1.5mu]){}\<[E]%
\\
\>[3]{}\mathrel{=}\mbox{\commentbegin ~  \Cref{eq:eval-pop1}   \commentend}{}\<[E]%
\\
\>[3]{}\hsindent{3}{}\<[6]%
\>[6]{}\Varid{results}\hsdot{\circ }{.}\Varid{snd}\mathbin{\$}((),\Conid{S}\;[\mskip1.5mu \Varid{x}\mskip1.5mu]\;[\mskip1.5mu \mskip1.5mu]){}\<[E]%
\\
\>[3]{}\mathrel{=}\mbox{\commentbegin ~  definition of \ensuremath{\Varid{snd}}   \commentend}{}\<[E]%
\\
\>[3]{}\hsindent{3}{}\<[6]%
\>[6]{}\Varid{results}\;(\Conid{S}\;[\mskip1.5mu \Varid{x}\mskip1.5mu]\;[\mskip1.5mu \mskip1.5mu]){}\<[E]%
\\
\>[3]{}\mathrel{=}\mbox{\commentbegin ~  definition of \ensuremath{\Varid{results}}   \commentend}{}\<[E]%
\\
\>[3]{}\hsindent{3}{}\<[6]%
\>[6]{}[\mskip1.5mu \Varid{x}\mskip1.5mu]{}\<[E]%
\\
\>[3]{}\mathrel{=}\mbox{\commentbegin ~  definition of \ensuremath{\Varid{gen_{ND}}}   \commentend}{}\<[E]%
\\
\>[3]{}\hsindent{3}{}\<[6]%
\>[6]{}\Varid{gen_{ND}}\;\Varid{x}{}\<[E]%
\ColumnHook
\end{hscode}\resethooks
\indentend 
For the second equation, we proceed with a case analysis on the input.

\noindent \mbox{\underline{case \ensuremath{\Conid{Fail}}}}
\indentbegin \begin{hscode}\SaveRestoreHook
\column{B}{@{}>{\hspre}l<{\hspost}@{}}%
\column{3}{@{}>{\hspre}l<{\hspost}@{}}%
\column{6}{@{}>{\hspre}l<{\hspost}@{}}%
\column{E}{@{}>{\hspre}l<{\hspost}@{}}%
\>[6]{}\Varid{extract_S}\;(\Varid{h_{State}^\prime}\;(\Varid{alg}\;\Conid{Fail})){}\<[E]%
\\
\>[3]{}\mathrel{=}\mbox{\commentbegin ~  definition of \ensuremath{\Varid{alg}}   \commentend}{}\<[E]%
\\
\>[3]{}\hsindent{3}{}\<[6]%
\>[6]{}\Varid{extract_S}\;(\Varid{h_{State}^\prime}\;\Varid{pop_S}){}\<[E]%
\\
\>[3]{}\mathrel{=}\mbox{\commentbegin ~  definition of \ensuremath{\Varid{extract_S}}   \commentend}{}\<[E]%
\\
\>[3]{}\hsindent{3}{}\<[6]%
\>[6]{}\Varid{results}\hsdot{\circ }{.}\Varid{snd}\mathbin{\$}\Varid{run_{State}}\;(\Varid{h_{State}^\prime}\;\Varid{pop_S})\;(\Conid{S}\;[\mskip1.5mu \mskip1.5mu]\;[\mskip1.5mu \mskip1.5mu]){}\<[E]%
\\
\>[3]{}\mathrel{=}\mbox{\commentbegin ~  \Cref{eq:eval-pop1}   \commentend}{}\<[E]%
\\
\>[3]{}\hsindent{3}{}\<[6]%
\>[6]{}\Varid{results}\hsdot{\circ }{.}\Varid{snd}\mathbin{\$}((),\Conid{S}\;[\mskip1.5mu \mskip1.5mu]\;[\mskip1.5mu \mskip1.5mu]){}\<[E]%
\\
\>[3]{}\mathrel{=}\mbox{\commentbegin ~  evaluation of \ensuremath{\Varid{results}}, \ensuremath{\Varid{snd}}   \commentend}{}\<[E]%
\\
\>[3]{}\hsindent{3}{}\<[6]%
\>[6]{}[\mskip1.5mu \mskip1.5mu]{}\<[E]%
\\
\>[3]{}\mathrel{=}\mbox{\commentbegin ~  definition of \ensuremath{\Varid{alg_{ND}}}   \commentend}{}\<[E]%
\\
\>[3]{}\hsindent{3}{}\<[6]%
\>[6]{}\Varid{alg_{ND}}\;\Conid{Fail}{}\<[E]%
\\
\>[3]{}\mathrel{=}\mbox{\commentbegin ~  definition of \ensuremath{\Varid{fmap}}   \commentend}{}\<[E]%
\\
\>[3]{}\hsindent{3}{}\<[6]%
\>[6]{}(\Varid{alg_{ND}}\hsdot{\circ }{.}\Varid{fmap}\;(\Varid{extract_S}\hsdot{\circ }{.}\Varid{h_{State}^\prime}))\;\Conid{Fail}{}\<[E]%
\ColumnHook
\end{hscode}\resethooks
\indentend 
\noindent \mbox{\underline{case \ensuremath{\Conid{Or}\;\Varid{p}\;\Varid{q}}}}
\indentbegin \begin{hscode}\SaveRestoreHook
\column{B}{@{}>{\hspre}l<{\hspost}@{}}%
\column{3}{@{}>{\hspre}l<{\hspost}@{}}%
\column{6}{@{}>{\hspre}l<{\hspost}@{}}%
\column{E}{@{}>{\hspre}l<{\hspost}@{}}%
\>[6]{}\Varid{extract_S}\;(\Varid{h_{State}^\prime}\;(\Varid{alg}\;(\Conid{Or}\;\Varid{p}\;\Varid{q}))){}\<[E]%
\\
\>[3]{}\mathrel{=}\mbox{\commentbegin ~  definition of \ensuremath{\Varid{alg}}   \commentend}{}\<[E]%
\\
\>[3]{}\hsindent{3}{}\<[6]%
\>[6]{}\Varid{extract_S}\;(\Varid{h_{State}^\prime}\;(\Varid{push_S}\;\Varid{q}\;\Varid{p})){}\<[E]%
\\
\>[3]{}\mathrel{=}\mbox{\commentbegin ~  definition of \ensuremath{\Varid{extract}}   \commentend}{}\<[E]%
\\
\>[3]{}\hsindent{3}{}\<[6]%
\>[6]{}\Varid{results}\hsdot{\circ }{.}\Varid{snd}\mathbin{\$}\Varid{run_{State}}\;(\Varid{h_{State}^\prime}\;(\Varid{push_S}\;\Varid{q}\;\Varid{p}))\;(\Conid{S}\;[\mskip1.5mu \mskip1.5mu]\;[\mskip1.5mu \mskip1.5mu]){}\<[E]%
\\
\>[3]{}\mathrel{=}\mbox{\commentbegin ~  \Cref{eq:eval-push}   \commentend}{}\<[E]%
\\
\>[3]{}\hsindent{3}{}\<[6]%
\>[6]{}\Varid{results}\hsdot{\circ }{.}\Varid{snd}\mathbin{\$}\Varid{run_{State}}\;(\Varid{h_{State}^\prime}\;\Varid{p})\;(\Conid{S}\;[\mskip1.5mu \mskip1.5mu]\;[\mskip1.5mu \Varid{q}\mskip1.5mu]){}\<[E]%
\\
\>[3]{}\mathrel{=}\mbox{\commentbegin ~  \Cref{eq:pop-extract}   \commentend}{}\<[E]%
\\
\>[3]{}\hsindent{3}{}\<[6]%
\>[6]{}\Varid{results}\hsdot{\circ }{.}\Varid{snd}\mathbin{\$}\Varid{run_{State}}\;(\Varid{h_{State}^\prime}\;\Varid{pop_S})\;(\Conid{S}\;([\mskip1.5mu \mskip1.5mu]+\!\!+\Varid{extract_S}\;(\Varid{h_{State}^\prime}\;\Varid{p}))\;[\mskip1.5mu \Varid{q}\mskip1.5mu]){}\<[E]%
\\
\>[3]{}\mathrel{=}\mbox{\commentbegin ~  definition of \ensuremath{(+\!\!+)}   \commentend}{}\<[E]%
\\
\>[3]{}\hsindent{3}{}\<[6]%
\>[6]{}\Varid{results}\hsdot{\circ }{.}\Varid{snd}\mathbin{\$}\Varid{run_{State}}\;(\Varid{h_{State}^\prime}\;\Varid{pop_S})\;(\Conid{S}\;(\Varid{extract_S}\;(\Varid{h_{State}^\prime}\;\Varid{p}))\;[\mskip1.5mu \Varid{q}\mskip1.5mu]){}\<[E]%
\\
\>[3]{}\mathrel{=}\mbox{\commentbegin ~  \Cref{eq:eval-pop2}   \commentend}{}\<[E]%
\\
\>[3]{}\hsindent{3}{}\<[6]%
\>[6]{}\Varid{results}\hsdot{\circ }{.}\Varid{snd}\mathbin{\$}\Varid{run_{State}}\;(\Varid{h_{State}^\prime}\;\Varid{q})\;(\Conid{S}\;(\Varid{extract_S}\;(\Varid{h_{State}^\prime}\;\Varid{p}))\;[\mskip1.5mu \mskip1.5mu]){}\<[E]%
\\
\>[3]{}\mathrel{=}\mbox{\commentbegin ~  \Cref{eq:pop-extract}   \commentend}{}\<[E]%
\\
\>[3]{}\hsindent{3}{}\<[6]%
\>[6]{}\Varid{results}\hsdot{\circ }{.}\Varid{snd}\mathbin{\$}\Varid{run_{State}}\;(\Varid{h_{State}^\prime}\;\Varid{pop_S})\;(\Conid{S}\;(\Varid{extract_S}\;(\Varid{h_{State}^\prime}\;\Varid{p})+\!\!+\Varid{extract_S}\;(\Varid{h_{State}^\prime}\;\Varid{q}))\;[\mskip1.5mu \mskip1.5mu]){}\<[E]%
\\
\>[3]{}\mathrel{=}\mbox{\commentbegin ~  \Cref{eq:eval-pop1}   \commentend}{}\<[E]%
\\
\>[3]{}\hsindent{3}{}\<[6]%
\>[6]{}\Varid{results}\hsdot{\circ }{.}\Varid{snd}\mathbin{\$}((),\Conid{S}\;(\Varid{extract_S}\;(\Varid{h_{State}^\prime}\;\Varid{p})+\!\!+\Varid{extract_S}\;(\Varid{h_{State}^\prime}\;\Varid{q}))\;[\mskip1.5mu \mskip1.5mu]){}\<[E]%
\\
\>[3]{}\mathrel{=}\mbox{\commentbegin ~  evaluation of \ensuremath{\Varid{results},\Varid{snd}}   \commentend}{}\<[E]%
\\
\>[3]{}\hsindent{3}{}\<[6]%
\>[6]{}\Varid{extract_S}\;(\Varid{h_{State}^\prime}\;\Varid{p})+\!\!+\Varid{extract_S}\;(\Varid{h_{State}^\prime}\;\Varid{q}){}\<[E]%
\\
\>[3]{}\mathrel{=}\mbox{\commentbegin ~  definition of \ensuremath{\Varid{alg_{ND}}}   \commentend}{}\<[E]%
\\
\>[3]{}\hsindent{3}{}\<[6]%
\>[6]{}\Varid{alg_{ND}}\;(\Conid{Or}\;((\Varid{extract_S}\hsdot{\circ }{.}\Varid{h_{State}^\prime})\;\Varid{p})\;((\Varid{extract_S}\hsdot{\circ }{.}\Varid{h_{State}^\prime})\;\Varid{q})){}\<[E]%
\\
\>[3]{}\mathrel{=}\mbox{\commentbegin ~  definition of \ensuremath{\Varid{fmap}}   \commentend}{}\<[E]%
\\
\>[3]{}\hsindent{3}{}\<[6]%
\>[6]{}(\Varid{alg_{ND}}\hsdot{\circ }{.}\Varid{fmap}\;(\Varid{extract_S}\hsdot{\circ }{.}\Varid{h_{State}^\prime}))\;(\Conid{Or}\;\Varid{p}\;\Varid{q}){}\<[E]%
\ColumnHook
\end{hscode}\resethooks
\indentend 
\end{proof}

In the above proof we have used several lemmas. Now we prove them.

\begin{lemma}[pop-extract]\label{eq:pop-extract}
~
\indentbegin \begin{hscode}\SaveRestoreHook
\column{B}{@{}>{\hspre}l<{\hspost}@{}}%
\column{6}{@{}>{\hspre}l<{\hspost}@{}}%
\column{E}{@{}>{\hspre}l<{\hspost}@{}}%
\>[6]{}\Varid{run_{State}}\;(\Varid{h_{State}^\prime}\;\Varid{p})\;(\Conid{S}\;\Varid{xs}\;\Varid{stack})\mathrel{=}\Varid{run_{State}}\;(\Varid{h_{State}^\prime}\;\Varid{pop_S})\;(\Conid{S}\;(\Varid{xs}+\!\!+\Varid{extract_S}\;(\Varid{h_{State}^\prime}\;\Varid{p}))\;\Varid{stack}){}\<[E]%
\ColumnHook
\end{hscode}\resethooks
\indentend holds for all \ensuremath{\Varid{p}} in the codomain of the function \ensuremath{\Varid{nondet2state_S}}.
\end{lemma}

\begin{proof} ~
We prove this lemma by structural induction on \ensuremath{\Varid{p}\mathbin{::}\Conid{Free}\;(\Varid{State_{F}}\;(\Conid{S}\;\Varid{a}))\;()}.
For each inductive case of \ensuremath{\Varid{p}}, we not only assume this lemma holds
for its sub-terms (this is the standard induction hypothesis), but
also assume \Cref{thm:nondet-stateS} holds for \ensuremath{\Varid{p}} and its sub-terms.
This is sound because in the proof of \Cref{thm:nondet-stateS}, for
\ensuremath{(\Varid{extract_S}\hsdot{\circ }{.}\Varid{h_{State}^\prime}\hsdot{\circ }{.}\Varid{nondet2state_S})\;\Varid{p}\mathrel{=}\Varid{h_{ND}}\;\Varid{p}}, we only apply
\Cref{eq:pop-extract} to the sub-terms of \ensuremath{\Varid{p}}, which is already
included in the induction hypothesis so there is no circular argument.

Since we assume \Cref{thm:nondet-stateS} holds for \ensuremath{\Varid{p}} and its
sub-terms, we can use several useful properties proved in the
sub-cases of the proof of \Cref{thm:nondet-stateS}. We list them here
for easy reference:
\begin{itemize}
\item {extract-gen}:
\ensuremath{\Varid{extract_S}\hsdot{\circ }{.}\Varid{h_{State}^\prime}\hsdot{\circ }{.}\Varid{gen}\mathrel{=}\Varid{\eta}}
\item {extract-alg1}:
\ensuremath{\Varid{extract_S}\;(\Varid{h_{State}^\prime}\;(\Varid{alg}\;\Conid{Fail}))\mathrel{=}[\mskip1.5mu \mskip1.5mu]}
\item {extract-alg2}:
\ensuremath{\Varid{extract_S}\;(\Varid{h_{State}^\prime}\;(\Varid{alg}\;(\Conid{Or}\;\Varid{p}\;\Varid{q})))\mathrel{=}\Varid{extract_S}\;(\Varid{h_{State}^\prime}\;\Varid{p})+\!\!+\Varid{extract_S}\;(\Varid{h_{State}^\prime}\;\Varid{q})}
\end{itemize}

We proceed by structural induction on \ensuremath{\Varid{p}}.
Note that for all \ensuremath{\Varid{p}} in the codomain of \ensuremath{\Varid{nondet2state_S}}, it is either
generated by the \ensuremath{\Varid{gen}} or the \ensuremath{\Varid{alg}} of \ensuremath{\Varid{nondet2state_S}}.  Thus, we only
need to prove the following two equations where \ensuremath{\Varid{p}\mathrel{=}\Varid{gen}\;\Varid{x}} or \ensuremath{\Varid{p}\mathrel{=}\Varid{alg}\;\Varid{x}} and \ensuremath{\Varid{x}} is in the codomain of \ensuremath{\Varid{fmap}\;\Varid{nondet2state_S}}.
\begin{enumerate}
    \item \ensuremath{\Varid{run_{State}}\;(\Varid{h_{State}^\prime}\;(\Varid{gen}\;\Varid{x}))\;(\Conid{S}\;\Varid{xs}\;\Varid{stack})\mathrel{=}\Varid{run_{State}}\;(\Varid{h_{State}^\prime}\;\Varid{pop_S})\;(\Conid{S}\;(\Varid{xs}+\!\!+\Varid{extract_S}\;(\Varid{h_{State}^\prime}\;(\Varid{gen}\;\Varid{x})))\;\Varid{stack})}
    \item \ensuremath{\Varid{run_{State}}\;(\Varid{h_{State}^\prime}\;(\Varid{alg}\;\Varid{x}))\;(\Conid{S}\;\Varid{xs}\;\Varid{stack})\mathrel{=}\Varid{run_{State}}\;(\Varid{h_{State}^\prime}\;\Varid{pop_S})\;(\Conid{S}\;(\Varid{xs}+\!\!+\Varid{extract_S}\;(\Varid{h_{State}^\prime}\;(\Varid{alg}\;\Varid{x})))\;\Varid{stack})}
\end{enumerate}

For the case \ensuremath{\Varid{p}\mathrel{=}\Varid{gen}\;\Varid{x}}, we calculate as follows:
\indentbegin \begin{hscode}\SaveRestoreHook
\column{B}{@{}>{\hspre}l<{\hspost}@{}}%
\column{3}{@{}>{\hspre}l<{\hspost}@{}}%
\column{6}{@{}>{\hspre}l<{\hspost}@{}}%
\column{E}{@{}>{\hspre}l<{\hspost}@{}}%
\>[6]{}\Varid{run_{State}}\;(\Varid{h_{State}^\prime}\;(\Varid{gen}\;\Varid{x}))\;(\Conid{S}\;\Varid{xs}\;\Varid{stack}){}\<[E]%
\\
\>[3]{}\mathrel{=}\mbox{\commentbegin ~  definition of \ensuremath{\Varid{gen}}   \commentend}{}\<[E]%
\\
\>[3]{}\hsindent{3}{}\<[6]%
\>[6]{}\Varid{run_{State}}\;(\Varid{h_{State}^\prime}\;(\Varid{append_S}\;\Varid{x}\;\Varid{pop_S}))\;(\Conid{S}\;\Varid{xs}\;\Varid{stack}){}\<[E]%
\\
\>[3]{}\mathrel{=}\mbox{\commentbegin ~  \Cref{eq:eval-append}   \commentend}{}\<[E]%
\\
\>[3]{}\hsindent{3}{}\<[6]%
\>[6]{}\Varid{run_{State}}\;(\Varid{h_{State}^\prime}\;\Varid{pop_S})\;(\Conid{S}\;(\Varid{xs}+\!\!+[\mskip1.5mu \Varid{x}\mskip1.5mu])\;\Varid{stack}){}\<[E]%
\\
\>[3]{}\mathrel{=}\mbox{\commentbegin ~  definition of \ensuremath{\Varid{\eta}}   \commentend}{}\<[E]%
\\
\>[3]{}\hsindent{3}{}\<[6]%
\>[6]{}\Varid{run_{State}}\;(\Varid{h_{State}^\prime}\;\Varid{pop_S})\;(\Conid{S}\;(\Varid{xs}+\!\!+\Varid{\eta}\;\Varid{x})\;\Varid{stack}){}\<[E]%
\\
\>[3]{}\mathrel{=}\mbox{\commentbegin ~  extract-gen   \commentend}{}\<[E]%
\\
\>[3]{}\hsindent{3}{}\<[6]%
\>[6]{}\Varid{run_{State}}\;(\Varid{h_{State}^\prime}\;\Varid{pop_S})\;(\Conid{S}\;(\Varid{xs}+\!\!+\Varid{extract_S}\;(\Varid{h_{State}^\prime}\;(\Varid{gen}\;\Varid{x})))\;\Varid{stack}){}\<[E]%
\ColumnHook
\end{hscode}\resethooks
\indentend %

For the case \ensuremath{\Varid{p}\mathrel{=}\Varid{alg}\;\Varid{x}}, we proceed with a case analysis on \ensuremath{\Varid{x}}.

\noindent \mbox{\underline{case \ensuremath{\Conid{Fail}}}}
\indentbegin \begin{hscode}\SaveRestoreHook
\column{B}{@{}>{\hspre}l<{\hspost}@{}}%
\column{3}{@{}>{\hspre}l<{\hspost}@{}}%
\column{6}{@{}>{\hspre}l<{\hspost}@{}}%
\column{E}{@{}>{\hspre}l<{\hspost}@{}}%
\>[6]{}\Varid{run_{State}}\;(\Varid{h_{State}^\prime}\;(\Varid{alg}\;\Conid{Fail}))\;(\Conid{S}\;\Varid{xs}\;\Varid{stack}){}\<[E]%
\\
\>[3]{}\mathrel{=}\mbox{\commentbegin ~  definition of \ensuremath{\Varid{alg}}   \commentend}{}\<[E]%
\\
\>[3]{}\hsindent{3}{}\<[6]%
\>[6]{}\Varid{run_{State}}\;(\Varid{h_{State}^\prime}\;(\Varid{pop_S}))\;(\Conid{S}\;\Varid{xs}\;\Varid{stack}){}\<[E]%
\\
\>[3]{}\mathrel{=}\mbox{\commentbegin ~  definition of \ensuremath{[\mskip1.5mu \mskip1.5mu]}   \commentend}{}\<[E]%
\\
\>[3]{}\hsindent{3}{}\<[6]%
\>[6]{}\Varid{run_{State}}\;(\Varid{h_{State}^\prime}\;\Varid{pop_S})\;(\Conid{S}\;(\Varid{xs}+\!\!+[\mskip1.5mu \mskip1.5mu])\;\Varid{stack}){}\<[E]%
\\
\>[3]{}\mathrel{=}\mbox{\commentbegin ~  extract-alg1   \commentend}{}\<[E]%
\\
\>[3]{}\hsindent{3}{}\<[6]%
\>[6]{}\Varid{run_{State}}\;(\Varid{h_{State}^\prime}\;\Varid{pop_S})\;(\Conid{S}\;(\Varid{xs}+\!\!+\Varid{extract_S}\;(\Varid{h_{State}^\prime}\;(\Varid{alg}\;\Conid{Fail})))\;\Varid{stack}){}\<[E]%
\ColumnHook
\end{hscode}\resethooks
\indentend \noindent \mbox{\underline{case \ensuremath{\Conid{Or}\;\Varid{p}_{1}\;\Varid{p}_{2}}}}

\indentbegin \begin{hscode}\SaveRestoreHook
\column{B}{@{}>{\hspre}l<{\hspost}@{}}%
\column{3}{@{}>{\hspre}l<{\hspost}@{}}%
\column{6}{@{}>{\hspre}l<{\hspost}@{}}%
\column{E}{@{}>{\hspre}l<{\hspost}@{}}%
\>[6]{}\Varid{run_{State}}\;(\Varid{h_{State}^\prime}\;(\Varid{alg}\;(\Conid{Or}\;\Varid{p}_{1}\;\Varid{p}_{2})))\;(\Conid{S}\;\Varid{xs}\;\Varid{stack}){}\<[E]%
\\
\>[3]{}\mathrel{=}\mbox{\commentbegin ~  definition of \ensuremath{\Varid{alg}}   \commentend}{}\<[E]%
\\
\>[3]{}\hsindent{3}{}\<[6]%
\>[6]{}\Varid{run_{State}}\;(\Varid{h_{State}^\prime}\;(\Varid{push_S}\;\Varid{p}_{2}\;\Varid{p}_{1}))\;(\Conid{S}\;\Varid{xs}\;\Varid{stack}){}\<[E]%
\\
\>[3]{}\mathrel{=}\mbox{\commentbegin ~  \Cref{eq:eval-push}   \commentend}{}\<[E]%
\\
\>[3]{}\hsindent{3}{}\<[6]%
\>[6]{}\Varid{run_{State}}\;(\Varid{h_{State}^\prime}\;\Varid{p}_{1})\;(\Conid{S}\;\Varid{xs}\;(\Varid{p}_{2}\mathbin{:}\Varid{stack})){}\<[E]%
\\
\>[3]{}\mathrel{=}\mbox{\commentbegin ~  induction hypothesis   \commentend}{}\<[E]%
\\
\>[3]{}\hsindent{3}{}\<[6]%
\>[6]{}\Varid{run_{State}}\;(\Varid{h_{State}^\prime}\;\Varid{pop_S})\;(\Conid{S}\;(\Varid{xs}+\!\!+\Varid{extract_S}\;(\Varid{h_{State}^\prime}\;\Varid{p}_{1}))\;(\Varid{p}_{2}\mathbin{:}\Varid{stack})){}\<[E]%
\\
\>[3]{}\mathrel{=}\mbox{\commentbegin ~  \Cref{eq:eval-pop2}   \commentend}{}\<[E]%
\\
\>[3]{}\hsindent{3}{}\<[6]%
\>[6]{}\Varid{run_{State}}\;(\Varid{h_{State}^\prime}\;\Varid{p}_{2})\;(\Conid{S}\;(\Varid{xs}+\!\!+\Varid{extract_S}\;(\Varid{h_{State}^\prime}\;\Varid{p}_{1}))\;\Varid{stack}){}\<[E]%
\\
\>[3]{}\mathrel{=}\mbox{\commentbegin ~  induction hypothesis   \commentend}{}\<[E]%
\\
\>[3]{}\hsindent{3}{}\<[6]%
\>[6]{}\Varid{run_{State}}\;(\Varid{h_{State}^\prime}\;\Varid{pop_S})\;(\Conid{S}\;(\Varid{xs}+\!\!+\Varid{extract_S}\;(\Varid{h_{State}^\prime}\;\Varid{p}_{1})+\!\!+\Varid{extract_S}\;(\Varid{h_{State}^\prime}\;\Varid{p}_{2}))\;\Varid{stack}){}\<[E]%
\\
\>[3]{}\mathrel{=}\mbox{\commentbegin ~  extract-alg2   \commentend}{}\<[E]%
\\
\>[3]{}\hsindent{3}{}\<[6]%
\>[6]{}\Varid{run_{State}}\;(\Varid{h_{State}^\prime}\;\Varid{pop_S})\;(\Conid{S}\;(\Varid{xs}+\!\!+\Varid{h_{State}^\prime}\;(\Varid{alg}\;(\Conid{Or}\;\Varid{p}_{1}\;\Varid{p}_{2})))\;\Varid{stack}){}\<[E]%
\ColumnHook
\end{hscode}\resethooks
\indentend %
\end{proof}

The following four lemmas characterise the behaviours of stack
operations.

\begin{lemma}[evaluation-append]\label{eq:eval-append}~\indentbegin \begin{hscode}\SaveRestoreHook
\column{B}{@{}>{\hspre}l<{\hspost}@{}}%
\column{3}{@{}>{\hspre}l<{\hspost}@{}}%
\column{E}{@{}>{\hspre}l<{\hspost}@{}}%
\>[3]{}\Varid{run_{State}}\;(\Varid{h_{State}^\prime}\;(\Varid{append_S}\;\Varid{x}\;\Varid{p}))\;(\Conid{S}\;\Varid{xs}\;\Varid{stack})\mathrel{=}\Varid{run_{State}}\;(\Varid{h_{State}^\prime}\;\Varid{p})\;(\Conid{S}\;(\Varid{xs}+\!\!+[\mskip1.5mu \Varid{x}\mskip1.5mu])\;\Varid{stack}){}\<[E]%
\ColumnHook
\end{hscode}\resethooks
\indentend \end{lemma}
\begin{proof}~\indentbegin \begin{hscode}\SaveRestoreHook
\column{B}{@{}>{\hspre}l<{\hspost}@{}}%
\column{3}{@{}>{\hspre}l<{\hspost}@{}}%
\column{6}{@{}>{\hspre}l<{\hspost}@{}}%
\column{E}{@{}>{\hspre}l<{\hspost}@{}}%
\>[6]{}\Varid{run_{State}}\;(\Varid{h_{State}^\prime}\;(\Varid{append_S}\;\Varid{x}\;\Varid{p}))\;(\Conid{S}\;\Varid{xs}\;\Varid{stack}){}\<[E]%
\\
\>[3]{}\mathrel{=}\mbox{\commentbegin ~  definition of \ensuremath{\Varid{append_S}}   \commentend}{}\<[E]%
\ColumnHook
\end{hscode}\resethooks
\indentend %
\indentbegin \begin{hscode}\SaveRestoreHook
\column{B}{@{}>{\hspre}l<{\hspost}@{}}%
\column{3}{@{}>{\hspre}l<{\hspost}@{}}%
\column{6}{@{}>{\hspre}l<{\hspost}@{}}%
\column{8}{@{}>{\hspre}l<{\hspost}@{}}%
\column{E}{@{}>{\hspre}l<{\hspost}@{}}%
\>[6]{}\Varid{run_{State}}\;(\Varid{h_{State}^\prime}\;(\Varid{get}>\!\!>\!\!=\lambda (\Conid{S}\;\Varid{xs}\;\Varid{stack})\to \Varid{put}\;(\Conid{S}\;(\Varid{xs}+\!\!+[\mskip1.5mu \Varid{x}\mskip1.5mu])\;\Varid{stack})>\!\!>\Varid{p}))\;(\Conid{S}\;\Varid{xs}\;\Varid{stack}){}\<[E]%
\\
\>[3]{}\mathrel{=}\mbox{\commentbegin ~  definition of \ensuremath{\Varid{get}}   \commentend}{}\<[E]%
\\
\>[3]{}\hsindent{3}{}\<[6]%
\>[6]{}\Varid{run_{State}}\;(\Varid{h_{State}^\prime}\;(\Conid{Op}\;(\Conid{Get}\;\Varid{\eta})>\!\!>\!\!=\lambda (\Conid{S}\;\Varid{xs}\;\Varid{stack})\to \Varid{put}\;(\Conid{S}\;(\Varid{xs}+\!\!+[\mskip1.5mu \Varid{x}\mskip1.5mu])\;\Varid{stack})>\!\!>\Varid{p}))\;(\Conid{S}\;\Varid{xs}\;\Varid{stack}){}\<[E]%
\\
\>[3]{}\mathrel{=}\mbox{\commentbegin ~  definition of \ensuremath{(>\!\!>\!\!=)} for free monad and Law \ref{eq:monad-ret-bind}: return-bind   \commentend}{}\<[E]%
\\
\>[3]{}\hsindent{3}{}\<[6]%
\>[6]{}\Varid{run_{State}}\;(\Varid{h_{State}^\prime}\;(\Conid{Op}\;(\Conid{Get}\;(\lambda (\Conid{S}\;\Varid{xs}\;\Varid{stack})\to \Varid{put}\;(\Conid{S}\;(\Varid{xs}+\!\!+[\mskip1.5mu \Varid{x}\mskip1.5mu])\;\Varid{stack})>\!\!>\Varid{p}))))\;(\Conid{S}\;\Varid{xs}\;\Varid{stack}){}\<[E]%
\\
\>[3]{}\mathrel{=}\mbox{\commentbegin ~  definition of \ensuremath{\Varid{h_{State}^\prime}}   \commentend}{}\<[E]%
\\
\>[3]{}\hsindent{3}{}\<[6]%
\>[6]{}\Varid{run_{State}}\;(\Conid{State}\;(\lambda \Varid{s}\to \Varid{run_{State}}\;(\Varid{h_{State}^\prime}\;((\lambda (\Conid{S}\;\Varid{xs}\;\Varid{stack})\to \Varid{put}\;(\Conid{S}\;(\Varid{xs}+\!\!+[\mskip1.5mu \Varid{x}\mskip1.5mu])\;\Varid{stack})>\!\!>\Varid{p})\;\Varid{s}))\;\Varid{s}))\;{}\<[E]%
\\
\>[6]{}\hsindent{2}{}\<[8]%
\>[8]{}(\Conid{S}\;\Varid{xs}\;\Varid{stack}){}\<[E]%
\\
\>[3]{}\mathrel{=}\mbox{\commentbegin ~  definition of \ensuremath{\Varid{run_{State}}}   \commentend}{}\<[E]%
\\
\>[3]{}\hsindent{3}{}\<[6]%
\>[6]{}(\lambda \Varid{s}\to \Varid{run_{State}}\;(\Varid{h_{State}^\prime}\;((\lambda (\Conid{S}\;\Varid{xs}\;\Varid{stack})\to \Varid{put}\;(\Conid{S}\;(\Varid{xs}+\!\!+[\mskip1.5mu \Varid{x}\mskip1.5mu])\;\Varid{stack})>\!\!>\Varid{p})\;\Varid{s}))\;\Varid{s})\;(\Conid{S}\;\Varid{xs}\;\Varid{stack}){}\<[E]%
\\
\>[3]{}\mathrel{=}\mbox{\commentbegin ~  function application   \commentend}{}\<[E]%
\\
\>[3]{}\hsindent{3}{}\<[6]%
\>[6]{}\Varid{run_{State}}\;(\Varid{h_{State}^\prime}\;((\lambda (\Conid{S}\;\Varid{xs}\;\Varid{stack})\to \Varid{put}\;(\Conid{S}\;(\Varid{xs}+\!\!+[\mskip1.5mu \Varid{x}\mskip1.5mu])\;\Varid{stack})>\!\!>\Varid{p})\;(\Conid{S}\;\Varid{xs}\;\Varid{stack})))\;(\Conid{S}\;\Varid{xs}\;\Varid{stack}){}\<[E]%
\\
\>[3]{}\mathrel{=}\mbox{\commentbegin ~  function application   \commentend}{}\<[E]%
\\
\>[3]{}\hsindent{3}{}\<[6]%
\>[6]{}\Varid{run_{State}}\;(\Varid{h_{State}^\prime}\;(\Varid{put}\;(\Conid{S}\;(\Varid{xs}+\!\!+[\mskip1.5mu \Varid{x}\mskip1.5mu])\;\Varid{stack})>\!\!>\Varid{p}))\;(\Conid{S}\;\Varid{xs}\;\Varid{stack}){}\<[E]%
\\
\>[3]{}\mathrel{=}\mbox{\commentbegin ~  definition of \ensuremath{\Varid{put}}   \commentend}{}\<[E]%
\\
\>[3]{}\hsindent{3}{}\<[6]%
\>[6]{}\Varid{run_{State}}\;(\Varid{h_{State}^\prime}\;(\Conid{Op}\;(\Conid{Put}\;(\Conid{S}\;(\Varid{xs}+\!\!+[\mskip1.5mu \Varid{x}\mskip1.5mu])\;\Varid{stack})\;(\Varid{\eta}\;()))>\!\!>\Varid{p}))\;(\Conid{S}\;\Varid{xs}\;\Varid{stack}){}\<[E]%
\\
\>[3]{}\mathrel{=}\mbox{\commentbegin ~  definition of \ensuremath{(>\!\!>)} for free monad and Law \ref{eq:monad-ret-bind}: return-bind   \commentend}{}\<[E]%
\\
\>[3]{}\hsindent{3}{}\<[6]%
\>[6]{}\Varid{run_{State}}\;(\Varid{h_{State}^\prime}\;(\Conid{Op}\;(\Conid{Put}\;(\Conid{S}\;(\Varid{xs}+\!\!+[\mskip1.5mu \Varid{x}\mskip1.5mu])\;\Varid{stack})\;\Varid{p})))\;(\Conid{S}\;\Varid{xs}\;\Varid{stack}){}\<[E]%
\\
\>[3]{}\mathrel{=}\mbox{\commentbegin ~  definition of \ensuremath{\Varid{h_{State}^\prime}}   \commentend}{}\<[E]%
\\
\>[3]{}\hsindent{3}{}\<[6]%
\>[6]{}\Varid{run_{State}}\;(\Conid{State}\;(\lambda \Varid{s}\to \Varid{run_{State}}\;(\Varid{h_{State}^\prime}\;\Varid{p})\;(\Conid{S}\;(\Varid{xs}+\!\!+[\mskip1.5mu \Varid{x}\mskip1.5mu])\;\Varid{stack})))\;(\Conid{S}\;\Varid{xs}\;\Varid{stack}){}\<[E]%
\\
\>[3]{}\mathrel{=}\mbox{\commentbegin ~  definition of \ensuremath{\Varid{run_{State}}}   \commentend}{}\<[E]%
\\
\>[3]{}\hsindent{3}{}\<[6]%
\>[6]{}(\lambda \Varid{s}\to \Varid{run_{State}}\;(\Varid{h_{State}^\prime}\;\Varid{p})\;(\Conid{S}\;(\Varid{xs}+\!\!+[\mskip1.5mu \Varid{x}\mskip1.5mu])\;\Varid{stack}))\;(\Conid{S}\;\Varid{xs}\;\Varid{stack}){}\<[E]%
\\
\>[3]{}\mathrel{=}\mbox{\commentbegin ~  function application   \commentend}{}\<[E]%
\\
\>[3]{}\hsindent{3}{}\<[6]%
\>[6]{}\Varid{run_{State}}\;(\Varid{h_{State}^\prime}\;\Varid{p})\;(\Conid{S}\;(\Varid{xs}+\!\!+[\mskip1.5mu \Varid{x}\mskip1.5mu])\;\Varid{stack}){}\<[E]%
\ColumnHook
\end{hscode}\resethooks
\indentend \end{proof}

\begin{lemma}[evaluation-pop1]\label{eq:eval-pop1}~\indentbegin \begin{hscode}\SaveRestoreHook
\column{B}{@{}>{\hspre}l<{\hspost}@{}}%
\column{3}{@{}>{\hspre}l<{\hspost}@{}}%
\column{E}{@{}>{\hspre}l<{\hspost}@{}}%
\>[3]{}\Varid{run_{State}}\;(\Varid{h_{State}^\prime}\;\Varid{pop_S})\;(\Conid{S}\;\Varid{xs}\;[\mskip1.5mu \mskip1.5mu])\mathrel{=}((),\Conid{S}\;\Varid{xs}\;[\mskip1.5mu \mskip1.5mu]){}\<[E]%
\ColumnHook
\end{hscode}\resethooks
\indentend \end{lemma}
\begin{proof}
\indentbegin \begin{hscode}\SaveRestoreHook
\column{B}{@{}>{\hspre}l<{\hspost}@{}}%
\column{3}{@{}>{\hspre}l<{\hspost}@{}}%
\column{6}{@{}>{\hspre}l<{\hspost}@{}}%
\column{8}{@{}>{\hspre}l<{\hspost}@{}}%
\column{23}{@{}>{\hspre}l<{\hspost}@{}}%
\column{32}{@{}>{\hspre}l<{\hspost}@{}}%
\column{E}{@{}>{\hspre}l<{\hspost}@{}}%
\>[6]{}\Varid{run_{State}}\;(\Varid{h_{State}^\prime}\;\Varid{pop_S})\;(\Conid{S}\;\Varid{xs}\;[\mskip1.5mu \mskip1.5mu]){}\<[E]%
\\
\>[3]{}\mathrel{=}\mbox{\commentbegin ~  definition of \ensuremath{\Varid{pop_S}}   \commentend}{}\<[E]%
\\
\>[3]{}\hsindent{3}{}\<[6]%
\>[6]{}\Varid{run_{State}}\;(\Varid{h_{State}^\prime}\;(\Varid{get}>\!\!>\!\!=\lambda (\Conid{S}\;\Varid{xs}\;\Varid{stack})\to {}\<[E]%
\\
\>[6]{}\hsindent{2}{}\<[8]%
\>[8]{}\mathbf{case}\;\Varid{stack}\;\mathbf{of}\;{}\<[23]%
\>[23]{}[\mskip1.5mu \mskip1.5mu]{}\<[32]%
\>[32]{}\to \Varid{\eta}\;(){}\<[E]%
\\
\>[23]{}\Varid{op}\mathbin{:}\Varid{ps}{}\<[32]%
\>[32]{}\to \mathbf{do}\;\Varid{put}\;(\Conid{S}\;\Varid{xs}\;\Varid{ps});\Varid{op}))\;(\Conid{S}\;\Varid{xs}\;[\mskip1.5mu \mskip1.5mu]){}\<[E]%
\\
\>[3]{}\mathrel{=}\mbox{\commentbegin ~  definition of \ensuremath{\Varid{get}}   \commentend}{}\<[E]%
\\
\>[3]{}\hsindent{3}{}\<[6]%
\>[6]{}\Varid{run_{State}}\;(\Varid{h_{State}^\prime}\;(\Conid{Op}\;(\Conid{Get}\;\Varid{\eta})>\!\!>\!\!=\lambda (\Conid{S}\;\Varid{xs}\;\Varid{stack})\to {}\<[E]%
\\
\>[6]{}\hsindent{2}{}\<[8]%
\>[8]{}\mathbf{case}\;\Varid{stack}\;\mathbf{of}\;{}\<[23]%
\>[23]{}[\mskip1.5mu \mskip1.5mu]{}\<[32]%
\>[32]{}\to \Varid{\eta}\;(){}\<[E]%
\\
\>[23]{}\Varid{op}\mathbin{:}\Varid{ps}{}\<[32]%
\>[32]{}\to \mathbf{do}\;\Varid{put}\;(\Conid{S}\;\Varid{xs}\;\Varid{ps});\Varid{op}))\;(\Conid{S}\;\Varid{xs}\;[\mskip1.5mu \mskip1.5mu]){}\<[E]%
\\
\>[3]{}\mathrel{=}\mbox{\commentbegin ~  definition of \ensuremath{(>\!\!>\!\!=)} for free monad and Law \ref{eq:monad-ret-bind}: return-bind   \commentend}{}\<[E]%
\\
\>[3]{}\hsindent{3}{}\<[6]%
\>[6]{}\Varid{run_{State}}\;(\Varid{h_{State}^\prime}\;(\Conid{Op}\;(\Conid{Get}\;(\lambda (\Conid{S}\;\Varid{xs}\;\Varid{stack})\to {}\<[E]%
\\
\>[6]{}\hsindent{2}{}\<[8]%
\>[8]{}\mathbf{case}\;\Varid{stack}\;\mathbf{of}\;{}\<[23]%
\>[23]{}[\mskip1.5mu \mskip1.5mu]{}\<[32]%
\>[32]{}\to \Varid{\eta}\;(){}\<[E]%
\\
\>[23]{}\Varid{op}\mathbin{:}\Varid{ps}{}\<[32]%
\>[32]{}\to \mathbf{do}\;\Varid{put}\;(\Conid{S}\;\Varid{xs}\;\Varid{ps});\Varid{op}))))\;(\Conid{S}\;\Varid{xs}\;[\mskip1.5mu \mskip1.5mu]){}\<[E]%
\\
\>[3]{}\mathrel{=}\mbox{\commentbegin ~  definition of \ensuremath{\Varid{h_{State}^\prime}}   \commentend}{}\<[E]%
\\
\>[3]{}\hsindent{3}{}\<[6]%
\>[6]{}\Varid{run_{State}}\;(\Conid{State}\;(\lambda \Varid{s}\to \Varid{run_{State}}\;(\Varid{h_{State}^\prime}\;((\lambda (\Conid{S}\;\Varid{xs}\;\Varid{stack})\to {}\<[E]%
\\
\>[6]{}\hsindent{2}{}\<[8]%
\>[8]{}\mathbf{case}\;\Varid{stack}\;\mathbf{of}\;{}\<[23]%
\>[23]{}[\mskip1.5mu \mskip1.5mu]{}\<[32]%
\>[32]{}\to \Varid{\eta}\;(){}\<[E]%
\\
\>[23]{}\Varid{op}\mathbin{:}\Varid{ps}{}\<[32]%
\>[32]{}\to \mathbf{do}\;\Varid{put}\;(\Conid{S}\;\Varid{xs}\;\Varid{ps});\Varid{op})\;\Varid{s}))\;\Varid{s}))\;(\Conid{S}\;\Varid{xs}\;[\mskip1.5mu \mskip1.5mu]){}\<[E]%
\\
\>[3]{}\mathrel{=}\mbox{\commentbegin ~  definition of \ensuremath{\Varid{run_{State}}}   \commentend}{}\<[E]%
\\
\>[3]{}\hsindent{3}{}\<[6]%
\>[6]{}(\lambda \Varid{s}\to \Varid{run_{State}}\;(\Varid{h_{State}^\prime}\;((\lambda (\Conid{S}\;\Varid{xs}\;\Varid{stack})\to {}\<[E]%
\\
\>[6]{}\hsindent{2}{}\<[8]%
\>[8]{}\mathbf{case}\;\Varid{stack}\;\mathbf{of}\;{}\<[23]%
\>[23]{}[\mskip1.5mu \mskip1.5mu]{}\<[32]%
\>[32]{}\to \Varid{\eta}\;(){}\<[E]%
\\
\>[23]{}\Varid{op}\mathbin{:}\Varid{ps}{}\<[32]%
\>[32]{}\to \mathbf{do}\;\Varid{put}\;(\Conid{S}\;\Varid{xs}\;\Varid{ps});\Varid{op})\;\Varid{s}))\;\Varid{s})\;(\Conid{S}\;\Varid{xs}\;[\mskip1.5mu \mskip1.5mu]){}\<[E]%
\\
\>[3]{}\mathrel{=}\mbox{\commentbegin ~  function application   \commentend}{}\<[E]%
\\
\>[3]{}\hsindent{3}{}\<[6]%
\>[6]{}\Varid{run_{State}}\;(\Varid{h_{State}^\prime}\;((\lambda (\Conid{S}\;\Varid{xs}\;\Varid{stack})\to {}\<[E]%
\\
\>[6]{}\hsindent{2}{}\<[8]%
\>[8]{}\mathbf{case}\;\Varid{stack}\;\mathbf{of}\;{}\<[23]%
\>[23]{}[\mskip1.5mu \mskip1.5mu]{}\<[32]%
\>[32]{}\to \Varid{\eta}\;(){}\<[E]%
\\
\>[23]{}\Varid{op}\mathbin{:}\Varid{ps}{}\<[32]%
\>[32]{}\to \mathbf{do}\;\Varid{put}\;(\Conid{S}\;\Varid{xs}\;\Varid{ps});\Varid{op})\;(\Conid{S}\;\Varid{xs}\;[\mskip1.5mu \mskip1.5mu])))\;(\Conid{S}\;\Varid{xs}\;[\mskip1.5mu \mskip1.5mu]){}\<[E]%
\\
\>[3]{}\mathrel{=}\mbox{\commentbegin ~  function application, \ensuremath{\mathbf{case}}-analysis   \commentend}{}\<[E]%
\\
\>[3]{}\hsindent{3}{}\<[6]%
\>[6]{}\Varid{run_{State}}\;(\Varid{h_{State}^\prime}\;(\Varid{\eta}\;()))\;(\Conid{S}\;\Varid{xs}\;[\mskip1.5mu \mskip1.5mu]){}\<[E]%
\\
\>[3]{}\mathrel{=}\mbox{\commentbegin ~  definition of \ensuremath{\Varid{h_{State}^\prime}}   \commentend}{}\<[E]%
\\
\>[3]{}\hsindent{3}{}\<[6]%
\>[6]{}\Varid{run_{State}}\;(\Conid{State}\;(\lambda \Varid{s}\to ((),\Varid{s})))\;(\Conid{S}\;\Varid{xs}\;[\mskip1.5mu \mskip1.5mu]){}\<[E]%
\\
\>[3]{}\mathrel{=}\mbox{\commentbegin ~  definition of \ensuremath{\Varid{run_{State}}}, function application   \commentend}{}\<[E]%
\\
\>[3]{}\hsindent{3}{}\<[6]%
\>[6]{}((),\Conid{S}\;\Varid{xs}\;[\mskip1.5mu \mskip1.5mu]){}\<[E]%
\ColumnHook
\end{hscode}\resethooks
\indentend \end{proof}

\begin{lemma}[evaluation-pop2]\label{eq:eval-pop2}~\indentbegin \begin{hscode}\SaveRestoreHook
\column{B}{@{}>{\hspre}l<{\hspost}@{}}%
\column{3}{@{}>{\hspre}l<{\hspost}@{}}%
\column{E}{@{}>{\hspre}l<{\hspost}@{}}%
\>[3]{}\Varid{run_{State}}\;(\Varid{h_{State}^\prime}\;\Varid{pop_S})\;(\Conid{S}\;\Varid{xs}\;(\Varid{q}\mathbin{:}\Varid{stack}))\mathrel{=}\Varid{run_{State}}\;(\Varid{h_{State}^\prime}\;\Varid{q})\;(\Conid{S}\;\Varid{xs}\;\Varid{stack}){}\<[E]%
\ColumnHook
\end{hscode}\resethooks
\indentend \end{lemma}
\begin{proof}~\indentbegin \begin{hscode}\SaveRestoreHook
\column{B}{@{}>{\hspre}l<{\hspost}@{}}%
\column{3}{@{}>{\hspre}l<{\hspost}@{}}%
\column{6}{@{}>{\hspre}l<{\hspost}@{}}%
\column{8}{@{}>{\hspre}l<{\hspost}@{}}%
\column{23}{@{}>{\hspre}l<{\hspost}@{}}%
\column{32}{@{}>{\hspre}l<{\hspost}@{}}%
\column{E}{@{}>{\hspre}l<{\hspost}@{}}%
\>[6]{}\Varid{run_{State}}\;(\Varid{h_{State}^\prime}\;\Varid{pop_S})\;(\Conid{S}\;\Varid{xs}\;(\Varid{q}\mathbin{:}\Varid{stack})){}\<[E]%
\\
\>[3]{}\mathrel{=}\mbox{\commentbegin ~  definition of \ensuremath{\Varid{pop_S}}   \commentend}{}\<[E]%
\\
\>[3]{}\hsindent{3}{}\<[6]%
\>[6]{}\Varid{run_{State}}\;(\Varid{h_{State}^\prime}\;(\Varid{get}>\!\!>\!\!=\lambda (\Conid{S}\;\Varid{xs}\;\Varid{stack})\to {}\<[E]%
\\
\>[6]{}\hsindent{2}{}\<[8]%
\>[8]{}\mathbf{case}\;\Varid{stack}\;\mathbf{of}\;{}\<[23]%
\>[23]{}[\mskip1.5mu \mskip1.5mu]{}\<[32]%
\>[32]{}\to \Varid{\eta}\;(){}\<[E]%
\\
\>[23]{}\Varid{op}\mathbin{:}\Varid{ps}{}\<[32]%
\>[32]{}\to \mathbf{do}\;\Varid{put}\;(\Conid{S}\;\Varid{xs}\;\Varid{ps});\Varid{op}))\;(\Conid{S}\;\Varid{xs}\;(\Varid{q}\mathbin{:}\Varid{stack})){}\<[E]%
\\
\>[3]{}\mathrel{=}\mbox{\commentbegin ~  definition of \ensuremath{\Varid{get}}   \commentend}{}\<[E]%
\\
\>[3]{}\hsindent{3}{}\<[6]%
\>[6]{}\Varid{run_{State}}\;(\Varid{h_{State}^\prime}\;(\Conid{Op}\;(\Conid{Get}\;\Varid{\eta})>\!\!>\!\!=\lambda (\Conid{S}\;\Varid{xs}\;\Varid{stack})\to {}\<[E]%
\\
\>[6]{}\hsindent{2}{}\<[8]%
\>[8]{}\mathbf{case}\;\Varid{stack}\;\mathbf{of}\;{}\<[23]%
\>[23]{}[\mskip1.5mu \mskip1.5mu]{}\<[32]%
\>[32]{}\to \Varid{\eta}\;(){}\<[E]%
\\
\>[23]{}\Varid{op}\mathbin{:}\Varid{ps}{}\<[32]%
\>[32]{}\to \mathbf{do}\;\Varid{put}\;(\Conid{S}\;\Varid{xs}\;\Varid{ps});\Varid{op}))\;(\Conid{S}\;\Varid{xs}\;(\Varid{q}\mathbin{:}\Varid{stack})){}\<[E]%
\\
\>[3]{}\mathrel{=}\mbox{\commentbegin ~  definition of \ensuremath{(>\!\!>\!\!=)} for free monad and Law \ref{eq:monad-ret-bind}: return-bind   \commentend}{}\<[E]%
\\
\>[3]{}\hsindent{3}{}\<[6]%
\>[6]{}\Varid{run_{State}}\;(\Varid{h_{State}^\prime}\;(\Conid{Op}\;(\Conid{Get}\;(\lambda (\Conid{S}\;\Varid{xs}\;\Varid{stack})\to {}\<[E]%
\\
\>[6]{}\hsindent{2}{}\<[8]%
\>[8]{}\mathbf{case}\;\Varid{stack}\;\mathbf{of}\;{}\<[23]%
\>[23]{}[\mskip1.5mu \mskip1.5mu]{}\<[32]%
\>[32]{}\to \Varid{\eta}\;(){}\<[E]%
\\
\>[23]{}\Varid{op}\mathbin{:}\Varid{ps}{}\<[32]%
\>[32]{}\to \mathbf{do}\;\Varid{put}\;(\Conid{S}\;\Varid{xs}\;\Varid{ps});\Varid{op}))))\;(\Conid{S}\;\Varid{xs}\;(\Varid{q}\mathbin{:}\Varid{stack})){}\<[E]%
\\
\>[3]{}\mathrel{=}\mbox{\commentbegin ~  definition of \ensuremath{\Varid{h_{State}^\prime}}   \commentend}{}\<[E]%
\\
\>[3]{}\hsindent{3}{}\<[6]%
\>[6]{}\Varid{run_{State}}\;(\Conid{State}\;(\lambda \Varid{s}\to \Varid{run_{State}}\;(\Varid{h_{State}^\prime}\;((\lambda (\Conid{S}\;\Varid{xs}\;\Varid{stack})\to {}\<[E]%
\\
\>[6]{}\hsindent{2}{}\<[8]%
\>[8]{}\mathbf{case}\;\Varid{stack}\;\mathbf{of}\;{}\<[23]%
\>[23]{}[\mskip1.5mu \mskip1.5mu]{}\<[32]%
\>[32]{}\to \Varid{\eta}\;(){}\<[E]%
\\
\>[23]{}\Varid{op}\mathbin{:}\Varid{ps}{}\<[32]%
\>[32]{}\to \mathbf{do}\;\Varid{put}\;(\Conid{S}\;\Varid{xs}\;\Varid{ps});\Varid{op})\;\Varid{s}))\;\Varid{s}))\;(\Conid{S}\;\Varid{xs}\;(\Varid{q}\mathbin{:}\Varid{stack})){}\<[E]%
\\
\>[3]{}\mathrel{=}\mbox{\commentbegin ~  definition of \ensuremath{\Varid{run_{State}}}   \commentend}{}\<[E]%
\\
\>[3]{}\hsindent{3}{}\<[6]%
\>[6]{}(\lambda \Varid{s}\to \Varid{run_{State}}\;(\Varid{h_{State}^\prime}\;((\lambda (\Conid{S}\;\Varid{xs}\;\Varid{stack})\to {}\<[E]%
\\
\>[6]{}\hsindent{2}{}\<[8]%
\>[8]{}\mathbf{case}\;\Varid{stack}\;\mathbf{of}\;{}\<[23]%
\>[23]{}[\mskip1.5mu \mskip1.5mu]{}\<[32]%
\>[32]{}\to \Varid{\eta}\;(){}\<[E]%
\\
\>[23]{}\Varid{op}\mathbin{:}\Varid{ps}{}\<[32]%
\>[32]{}\to \mathbf{do}\;\Varid{put}\;(\Conid{S}\;\Varid{xs}\;\Varid{ps});\Varid{op})\;\Varid{s}))\;\Varid{s})\;(\Conid{S}\;\Varid{xs}\;(\Varid{q}\mathbin{:}\Varid{stack})){}\<[E]%
\\
\>[3]{}\mathrel{=}\mbox{\commentbegin ~  function application   \commentend}{}\<[E]%
\\
\>[3]{}\hsindent{3}{}\<[6]%
\>[6]{}\Varid{run_{State}}\;(\Varid{h_{State}^\prime}\;((\lambda (\Conid{S}\;\Varid{xs}\;\Varid{stack})\to {}\<[E]%
\\
\>[6]{}\hsindent{2}{}\<[8]%
\>[8]{}\mathbf{case}\;\Varid{stack}\;\mathbf{of}\;{}\<[23]%
\>[23]{}[\mskip1.5mu \mskip1.5mu]{}\<[32]%
\>[32]{}\to \Varid{\eta}\;(){}\<[E]%
\\
\>[23]{}\Varid{op}\mathbin{:}\Varid{ps}{}\<[32]%
\>[32]{}\to \mathbf{do}\;\Varid{put}\;(\Conid{S}\;\Varid{xs}\;\Varid{ps});\Varid{op})\;(\Conid{S}\;\Varid{xs}\;(\Varid{q}\mathbin{:}\Varid{stack}))))\;(\Conid{S}\;\Varid{xs}\;(\Varid{q}\mathbin{:}\Varid{stack})){}\<[E]%
\\
\>[3]{}\mathrel{=}\mbox{\commentbegin ~  function application, \ensuremath{\mathbf{case}}-analysis   \commentend}{}\<[E]%
\ColumnHook
\end{hscode}\resethooks
\indentend %
\indentbegin \begin{hscode}\SaveRestoreHook
\column{B}{@{}>{\hspre}l<{\hspost}@{}}%
\column{3}{@{}>{\hspre}l<{\hspost}@{}}%
\column{6}{@{}>{\hspre}l<{\hspost}@{}}%
\column{E}{@{}>{\hspre}l<{\hspost}@{}}%
\>[6]{}\Varid{run_{State}}\;(\Varid{h_{State}^\prime}\;(\Varid{put}\;(\Conid{S}\;\Varid{xs}\;\Varid{stack})>\!\!>\Varid{q}))\;(\Conid{S}\;\Varid{xs}\;(\Varid{q}\mathbin{:}\Varid{stack})){}\<[E]%
\\
\>[3]{}\mathrel{=}\mbox{\commentbegin ~  definition of \ensuremath{\Varid{put}}   \commentend}{}\<[E]%
\\
\>[3]{}\hsindent{3}{}\<[6]%
\>[6]{}\Varid{run_{State}}\;(\Varid{h_{State}^\prime}\;(\Conid{Op}\;(\Conid{Put}\;(\Conid{S}\;\Varid{xs}\;\Varid{stack})\;(\Varid{\eta}\;()))>\!\!>\Varid{q}))\;(\Conid{S}\;\Varid{xs}\;(\Varid{q}\mathbin{:}\Varid{stack})){}\<[E]%
\\
\>[3]{}\mathrel{=}\mbox{\commentbegin ~  definition of \ensuremath{(>\!\!>)} for free monad and Law \ref{eq:monad-ret-bind}: return-bind   \commentend}{}\<[E]%
\\
\>[3]{}\hsindent{3}{}\<[6]%
\>[6]{}\Varid{run_{State}}\;(\Varid{h_{State}^\prime}\;(\Conid{Op}\;(\Conid{Put}\;(\Conid{S}\;\Varid{xs}\;\Varid{stack})\;\Varid{q})))\;(\Conid{S}\;\Varid{xs}\;(\Varid{q}\mathbin{:}\Varid{stack})){}\<[E]%
\\
\>[3]{}\mathrel{=}\mbox{\commentbegin ~  definition of \ensuremath{\Varid{h_{State}^\prime}}   \commentend}{}\<[E]%
\\
\>[3]{}\hsindent{3}{}\<[6]%
\>[6]{}\Varid{run_{State}}\;(\Conid{State}\;(\lambda \Varid{s}\to \Varid{run_{State}}\;(\Varid{h_{State}^\prime}\;\Varid{q})\;(\Conid{S}\;\Varid{xs}\;\Varid{stack})))\;(\Conid{S}\;\Varid{xs}\;(\Varid{q}\mathbin{:}\Varid{stack})){}\<[E]%
\\
\>[3]{}\mathrel{=}\mbox{\commentbegin ~  definition of \ensuremath{\Varid{run_{State}}}   \commentend}{}\<[E]%
\\
\>[3]{}\hsindent{3}{}\<[6]%
\>[6]{}(\lambda \Varid{s}\to \Varid{run_{State}}\;(\Varid{h_{State}^\prime}\;\Varid{q})\;(\Conid{S}\;\Varid{xs}\;\Varid{stack}))\;(\Conid{S}\;\Varid{xs}\;(\Varid{q}\mathbin{:}\Varid{stack})){}\<[E]%
\\
\>[3]{}\mathrel{=}\mbox{\commentbegin ~  function application   \commentend}{}\<[E]%
\\
\>[3]{}\hsindent{3}{}\<[6]%
\>[6]{}\Varid{run_{State}}\;(\Varid{h_{State}^\prime}\;\Varid{q})\;(\Conid{S}\;\Varid{xs}\;\Varid{stack}){}\<[E]%
\ColumnHook
\end{hscode}\resethooks
\indentend \end{proof}

\begin{lemma}[evaluation-push]\label{eq:eval-push}~\indentbegin \begin{hscode}\SaveRestoreHook
\column{B}{@{}>{\hspre}l<{\hspost}@{}}%
\column{3}{@{}>{\hspre}l<{\hspost}@{}}%
\column{E}{@{}>{\hspre}l<{\hspost}@{}}%
\>[3]{}\Varid{run_{State}}\;(\Varid{h_{State}^\prime}\;(\Varid{push_S}\;\Varid{q}\;\Varid{p}))\;(\Conid{S}\;\Varid{xs}\;\Varid{stack})\mathrel{=}\Varid{run_{State}}\;(\Varid{h_{State}^\prime}\;\Varid{p})\;(\Conid{S}\;\Varid{xs}\;(\Varid{q}\mathbin{:}\Varid{stack})){}\<[E]%
\ColumnHook
\end{hscode}\resethooks
\indentend \end{lemma}
\begin{proof}~\indentbegin \begin{hscode}\SaveRestoreHook
\column{B}{@{}>{\hspre}l<{\hspost}@{}}%
\column{3}{@{}>{\hspre}l<{\hspost}@{}}%
\column{6}{@{}>{\hspre}l<{\hspost}@{}}%
\column{E}{@{}>{\hspre}l<{\hspost}@{}}%
\>[6]{}\Varid{run_{State}}\;(\Varid{h_{State}^\prime}\;(\Varid{push_S}\;\Varid{q}\;\Varid{p}))\;(\Conid{S}\;\Varid{xs}\;\Varid{stack}){}\<[E]%
\\
\>[3]{}\mathrel{=}\mbox{\commentbegin ~  definition of \ensuremath{\Varid{push_S}}   \commentend}{}\<[E]%
\ColumnHook
\end{hscode}\resethooks
\indentend %
\indentbegin \begin{hscode}\SaveRestoreHook
\column{B}{@{}>{\hspre}l<{\hspost}@{}}%
\column{3}{@{}>{\hspre}l<{\hspost}@{}}%
\column{6}{@{}>{\hspre}l<{\hspost}@{}}%
\column{8}{@{}>{\hspre}l<{\hspost}@{}}%
\column{E}{@{}>{\hspre}l<{\hspost}@{}}%
\>[6]{}\Varid{run_{State}}\;(\Varid{h_{State}^\prime}\;(\Varid{get}>\!\!>\!\!=\lambda (\Conid{S}\;\Varid{xs}\;\Varid{stack})\to \Varid{put}\;(\Conid{S}\;\Varid{xs}\;(\Varid{q}\mathbin{:}\Varid{stack}))>\!\!>\Varid{p}))\;(\Conid{S}\;\Varid{xs}\;\Varid{stack}){}\<[E]%
\\
\>[3]{}\mathrel{=}\mbox{\commentbegin ~  definition of \ensuremath{\Varid{get}}   \commentend}{}\<[E]%
\\
\>[3]{}\hsindent{3}{}\<[6]%
\>[6]{}\Varid{run_{State}}\;(\Varid{h_{State}^\prime}\;(\Conid{Op}\;(\Conid{Get}\;\Varid{\eta})>\!\!>\!\!=\lambda (\Conid{S}\;\Varid{xs}\;\Varid{stack})\to \Varid{put}\;(\Conid{S}\;\Varid{xs}\;(\Varid{q}\mathbin{:}\Varid{stack}))>\!\!>\Varid{p}))\;(\Conid{S}\;\Varid{xs}\;\Varid{stack}){}\<[E]%
\\
\>[3]{}\mathrel{=}\mbox{\commentbegin ~  definition of \ensuremath{(>\!\!>\!\!=)} for free monad and Law \ref{eq:monad-ret-bind}: return-bind   \commentend}{}\<[E]%
\\
\>[3]{}\hsindent{3}{}\<[6]%
\>[6]{}\Varid{run_{State}}\;(\Varid{h_{State}^\prime}\;(\Conid{Op}\;(\Conid{Get}\;(\lambda (\Conid{S}\;\Varid{xs}\;\Varid{stack})\to \Varid{put}\;(\Conid{S}\;\Varid{xs}\;(\Varid{q}\mathbin{:}\Varid{stack}))>\!\!>\Varid{p}))))\;(\Conid{S}\;\Varid{xs}\;\Varid{stack}){}\<[E]%
\\
\>[3]{}\mathrel{=}\mbox{\commentbegin ~  definition of \ensuremath{\Varid{h_{State}^\prime}}   \commentend}{}\<[E]%
\\
\>[3]{}\hsindent{3}{}\<[6]%
\>[6]{}\Varid{run_{State}}\;(\Conid{State}\;(\lambda \Varid{s}\to \Varid{run_{State}}\;(\Varid{h_{State}^\prime}\;((\lambda (\Conid{S}\;\Varid{xs}\;\Varid{stack})\to \Varid{put}\;(\Conid{S}\;\Varid{xs}\;(\Varid{q}\mathbin{:}\Varid{stack}))>\!\!>\Varid{p})\;\Varid{s}))\;\Varid{s}))\;{}\<[E]%
\\
\>[6]{}\hsindent{2}{}\<[8]%
\>[8]{}(\Conid{S}\;\Varid{xs}\;\Varid{stack}){}\<[E]%
\\
\>[3]{}\mathrel{=}\mbox{\commentbegin ~  definition of \ensuremath{\Varid{run_{State}}}   \commentend}{}\<[E]%
\\
\>[3]{}\hsindent{3}{}\<[6]%
\>[6]{}(\lambda \Varid{s}\to \Varid{run_{State}}\;(\Varid{h_{State}^\prime}\;((\lambda (\Conid{S}\;\Varid{xs}\;\Varid{stack})\to \Varid{put}\;(\Conid{S}\;\Varid{xs}\;(\Varid{q}\mathbin{:}\Varid{stack}))>\!\!>\Varid{p})\;\Varid{s}))\;\Varid{s})\;(\Conid{S}\;\Varid{xs}\;\Varid{stack}){}\<[E]%
\\
\>[3]{}\mathrel{=}\mbox{\commentbegin ~  function application   \commentend}{}\<[E]%
\\
\>[3]{}\hsindent{3}{}\<[6]%
\>[6]{}\Varid{run_{State}}\;(\Varid{h_{State}^\prime}\;((\lambda (\Conid{S}\;\Varid{xs}\;\Varid{stack})\to \Varid{put}\;(\Conid{S}\;\Varid{xs}\;(\Varid{q}\mathbin{:}\Varid{stack}))>\!\!>\Varid{p})\;(\Conid{S}\;\Varid{xs}\;\Varid{stack})))\;(\Conid{S}\;\Varid{xs}\;\Varid{stack}){}\<[E]%
\\
\>[3]{}\mathrel{=}\mbox{\commentbegin ~  function application   \commentend}{}\<[E]%
\\
\>[3]{}\hsindent{3}{}\<[6]%
\>[6]{}\Varid{run_{State}}\;(\Varid{h_{State}^\prime}\;(\Varid{put}\;(\Conid{S}\;\Varid{xs}\;(\Varid{q}\mathbin{:}\Varid{stack}))>\!\!>\Varid{p}))\;(\Conid{S}\;\Varid{xs}\;\Varid{stack}){}\<[E]%
\\
\>[3]{}\mathrel{=}\mbox{\commentbegin ~  definition of \ensuremath{\Varid{put}}   \commentend}{}\<[E]%
\\
\>[3]{}\hsindent{3}{}\<[6]%
\>[6]{}\Varid{run_{State}}\;(\Varid{h_{State}^\prime}\;(\Conid{Op}\;(\Conid{Put}\;(\Conid{S}\;\Varid{xs}\;(\Varid{q}\mathbin{:}\Varid{stack}))\;(\Varid{\eta}\;()))>\!\!>\Varid{p}))\;(\Conid{S}\;\Varid{xs}\;\Varid{stack}){}\<[E]%
\\
\>[3]{}\mathrel{=}\mbox{\commentbegin ~  definition of \ensuremath{(>\!\!>)} for free monad and Law \ref{eq:monad-ret-bind}: return-bind   \commentend}{}\<[E]%
\\
\>[3]{}\hsindent{3}{}\<[6]%
\>[6]{}\Varid{run_{State}}\;(\Varid{h_{State}^\prime}\;(\Conid{Op}\;(\Conid{Put}\;(\Conid{S}\;\Varid{xs}\;(\Varid{q}\mathbin{:}\Varid{stack}))\;\Varid{p})))\;(\Conid{S}\;\Varid{xs}\;\Varid{stack}){}\<[E]%
\\
\>[3]{}\mathrel{=}\mbox{\commentbegin ~  definition of \ensuremath{\Varid{h_{State}^\prime}}   \commentend}{}\<[E]%
\\
\>[3]{}\hsindent{3}{}\<[6]%
\>[6]{}\Varid{run_{State}}\;(\Conid{State}\;(\lambda \Varid{s}\to \Varid{run_{State}}\;(\Varid{h_{State}^\prime}\;\Varid{p})\;(\Conid{S}\;\Varid{xs}\;(\Varid{q}\mathbin{:}\Varid{stack}))))\;(\Conid{S}\;\Varid{xs}\;\Varid{stack}){}\<[E]%
\\
\>[3]{}\mathrel{=}\mbox{\commentbegin ~  definition of \ensuremath{\Varid{run_{State}}}   \commentend}{}\<[E]%
\\
\>[3]{}\hsindent{3}{}\<[6]%
\>[6]{}(\lambda \Varid{s}\to \Varid{run_{State}}\;(\Varid{h_{State}^\prime}\;\Varid{p})\;(\Conid{S}\;\Varid{xs}\;(\Varid{q}\mathbin{:}\Varid{stack})))\;(\Conid{S}\;\Varid{xs}\;\Varid{stack}){}\<[E]%
\\
\>[3]{}\mathrel{=}\mbox{\commentbegin ~  function application   \commentend}{}\<[E]%
\\
\>[3]{}\hsindent{3}{}\<[6]%
\>[6]{}\Varid{run_{State}}\;(\Varid{h_{State}^\prime}\;\Varid{p})\;(\Conid{S}\;\Varid{xs}\;(\Varid{q}\mathbin{:}\Varid{stack})){}\<[E]%
\ColumnHook
\end{hscode}\resethooks
\indentend \end{proof}

\subsection{Combining with Other Effects}
\label{app:in-combination-with-other-effects}

This section proves the following theorem in \Cref{sec:nondet2state}.

\nondetState*

\begin{proof}
The proof is very similar to that of \Cref{thm:nondet-stateS} in
\Cref{app:runnd-hnd}.

We start with expanding the definition of \ensuremath{\Varid{run_{ND+f}}}:\indentbegin \begin{hscode}\SaveRestoreHook
\column{B}{@{}>{\hspre}l<{\hspost}@{}}%
\column{3}{@{}>{\hspre}l<{\hspost}@{}}%
\column{E}{@{}>{\hspre}l<{\hspost}@{}}%
\>[3]{}\Varid{extract_{SS}}\hsdot{\circ }{.}\Varid{h_{State}}\hsdot{\circ }{.}\Varid{nondet2state}\mathrel{=}\Varid{h_{ND}}{}\<[E]%
\ColumnHook
\end{hscode}\resethooks
\indentend We use the fold fusion law {\bf fusion-post'}~(\ref{eq:fusion-post-strong}) to
fuse the left-hand side.
Since the right-hand side is already a fold, to prove the equation we
just need to check the components of the fold \ensuremath{\Varid{h_{ND}}} satisfy the
conditions of the fold fusion. The conditions can be splitted into
the following three equations:

\[\ba{rl}
    &\ensuremath{(\Varid{extract_{SS}}\hsdot{\circ }{.}\Varid{h_{State}})\hsdot{\circ }{.}\Varid{gen}\mathrel{=}\Varid{gen_{ND+f}}} \\
    &\ensuremath{(\Varid{extract_{SS}}\hsdot{\circ }{.}\Varid{h_{State}})\hsdot{\circ }{.}\Varid{alg}\hsdot{\circ }{.}\Varid{fmap}\;\Varid{nondet2state_S}}\\
 \ensuremath{\mathrel{=}}&  \ensuremath{\Varid{alg_{ND+f}}\hsdot{\circ }{.}\Varid{fmap}\;(\Varid{extract_{SS}}\hsdot{\circ }{.}\Varid{h_{State}})\hsdot{\circ }{.}\Varid{fmap}\;\Varid{nondet2state_S}} \\
    &\ensuremath{(\Varid{extract_{SS}}\hsdot{\circ }{.}\Varid{h_{State}})\hsdot{\circ }{.}\Varid{fwd}\hsdot{\circ }{.}\Varid{fmap}\;\Varid{nondet2state_S}}\\
 \ensuremath{\mathrel{=}}&  \ensuremath{\Varid{fwd_{ND+f}}\hsdot{\circ }{.}\Varid{fmap}\;(\Varid{extract_{SS}}\hsdot{\circ }{.}\Varid{h_{State}})\hsdot{\circ }{.}\Varid{fmap}\;\Varid{nondet2state_S}}
\ea\]

For brevity, we omit the last common part \ensuremath{\Varid{fmap}\;\Varid{nondet2state_S}} of
the second equation in the following proof. Instead, we assume that
the input is in the codomain of \ensuremath{\Varid{fmap}\;\Varid{nondet2state_S}}.

For the first equation, we calculate as follows:
\indentbegin \begin{hscode}\SaveRestoreHook
\column{B}{@{}>{\hspre}l<{\hspost}@{}}%
\column{3}{@{}>{\hspre}l<{\hspost}@{}}%
\column{6}{@{}>{\hspre}l<{\hspost}@{}}%
\column{E}{@{}>{\hspre}l<{\hspost}@{}}%
\>[6]{}\Varid{extract_{SS}}\;(\Varid{h_{State}}\;(\Varid{gen}\;\Varid{x})){}\<[E]%
\\
\>[3]{}\mathrel{=}\mbox{\commentbegin ~  definition of \ensuremath{\Varid{gen}}   \commentend}{}\<[E]%
\\
\>[3]{}\hsindent{3}{}\<[6]%
\>[6]{}\Varid{extract_{SS}}\;(\Varid{h_{State}}\;(\Varid{append_{SS}}\;\Varid{x}\;\Varid{pop_{SS}})){}\<[E]%
\\
\>[3]{}\mathrel{=}\mbox{\commentbegin ~  definition of \ensuremath{\Varid{extract_{SS}}}   \commentend}{}\<[E]%
\\
\>[3]{}\hsindent{3}{}\<[6]%
\>[6]{}\Varid{results_{SS}}\hsdot{\circ }{.}\Varid{snd}\mathbin{\langle\hspace{1.6pt}\mathclap{\raisebox{0.1pt}{\scalebox{1}{\$}}}\hspace{1.6pt}\rangle}\Varid{run_{StateT}}\;(\Varid{h_{State}}\;(\Varid{append_{SS}}\;\Varid{x}\;\Varid{pop_{SS}}))\;(\Conid{SS}\;[\mskip1.5mu \mskip1.5mu]\;[\mskip1.5mu \mskip1.5mu]){}\<[E]%
\\
\>[3]{}\mathrel{=}\mbox{\commentbegin ~  \Cref{eq:eval-append-f}   \commentend}{}\<[E]%
\\
\>[3]{}\hsindent{3}{}\<[6]%
\>[6]{}\Varid{results_{SS}}\hsdot{\circ }{.}\Varid{snd}\mathbin{\langle\hspace{1.6pt}\mathclap{\raisebox{0.1pt}{\scalebox{1}{\$}}}\hspace{1.6pt}\rangle}\Varid{run_{StateT}}\;(\Varid{h_{State}}\;\Varid{pop_{SS}})\;(\Conid{SS}\;([\mskip1.5mu \mskip1.5mu]+\!\!+[\mskip1.5mu \Varid{x}\mskip1.5mu])\;[\mskip1.5mu \mskip1.5mu]){}\<[E]%
\\
\>[3]{}\mathrel{=}\mbox{\commentbegin ~  definition of \ensuremath{(+\!\!+)}   \commentend}{}\<[E]%
\\
\>[3]{}\hsindent{3}{}\<[6]%
\>[6]{}\Varid{results_{SS}}\hsdot{\circ }{.}\Varid{snd}\mathbin{\langle\hspace{1.6pt}\mathclap{\raisebox{0.1pt}{\scalebox{1}{\$}}}\hspace{1.6pt}\rangle}\Varid{run_{StateT}}\;(\Varid{h_{State}}\;\Varid{pop_{SS}})\;(\Conid{SS}\;[\mskip1.5mu \Varid{x}\mskip1.5mu]\;[\mskip1.5mu \mskip1.5mu]){}\<[E]%
\\
\>[3]{}\mathrel{=}\mbox{\commentbegin ~  \Cref{eq:eval-pop1-f}   \commentend}{}\<[E]%
\\
\>[3]{}\hsindent{3}{}\<[6]%
\>[6]{}\Varid{results_{SS}}\hsdot{\circ }{.}\Varid{snd}\mathbin{\langle\hspace{1.6pt}\mathclap{\raisebox{0.1pt}{\scalebox{1}{\$}}}\hspace{1.6pt}\rangle}\Varid{\eta}\;((),\Conid{SS}\;[\mskip1.5mu \Varid{x}\mskip1.5mu]\;[\mskip1.5mu \mskip1.5mu]){}\<[E]%
\\
\>[3]{}\mathrel{=}\mbox{\commentbegin ~  evaluation of \ensuremath{\Varid{snd},\Varid{results_{SS}}}   \commentend}{}\<[E]%
\\
\>[3]{}\hsindent{3}{}\<[6]%
\>[6]{}\Varid{\eta}\;[\mskip1.5mu \Varid{x}\mskip1.5mu]{}\<[E]%
\\
\>[3]{}\mathrel{=}\mbox{\commentbegin ~  definition of \ensuremath{\Varid{\eta}} for free monad   \commentend}{}\<[E]%
\\
\>[3]{}\hsindent{3}{}\<[6]%
\>[6]{}\Conid{Var}\;[\mskip1.5mu \Varid{x}\mskip1.5mu]{}\<[E]%
\\
\>[3]{}\mathrel{=}\mbox{\commentbegin ~  definition of \ensuremath{\Varid{\eta}_{[]}}   \commentend}{}\<[E]%
\\
\>[3]{}\hsindent{3}{}\<[6]%
\>[6]{}(\Conid{Var}\hsdot{\circ }{.}\Varid{\eta})\;\Varid{x}{}\<[E]%
\\
\>[3]{}\mathrel{=}\mbox{\commentbegin ~  definition of \ensuremath{\Varid{gen_{ND+f}}}   \commentend}{}\<[E]%
\\
\>[3]{}\hsindent{3}{}\<[6]%
\>[6]{}\Varid{gen_{ND+f}}\;\Varid{x}{}\<[E]%
\ColumnHook
\end{hscode}\resethooks
\indentend %

For the second equation, we proceed with a case analysis on the input.

\noindent \mbox{\underline{case \ensuremath{\Conid{Fail}}}}

\indentbegin \begin{hscode}\SaveRestoreHook
\column{B}{@{}>{\hspre}l<{\hspost}@{}}%
\column{3}{@{}>{\hspre}l<{\hspost}@{}}%
\column{6}{@{}>{\hspre}l<{\hspost}@{}}%
\column{E}{@{}>{\hspre}l<{\hspost}@{}}%
\>[6]{}\Varid{extract_{SS}}\;(\Varid{h_{State}}\;(\Varid{alg}\;\Conid{Fail})){}\<[E]%
\\
\>[3]{}\mathrel{=}\mbox{\commentbegin ~  definition of \ensuremath{\Varid{alg}}   \commentend}{}\<[E]%
\\
\>[3]{}\hsindent{3}{}\<[6]%
\>[6]{}\Varid{extract_{SS}}\;(\Varid{h_{State}}\;\Varid{pop_{SS}}){}\<[E]%
\\
\>[3]{}\mathrel{=}\mbox{\commentbegin ~  definition of \ensuremath{\Varid{extract_{SS}}}   \commentend}{}\<[E]%
\\
\>[3]{}\hsindent{3}{}\<[6]%
\>[6]{}\Varid{results_{SS}}\hsdot{\circ }{.}\Varid{snd}\mathbin{\langle\hspace{1.6pt}\mathclap{\raisebox{0.1pt}{\scalebox{1}{\$}}}\hspace{1.6pt}\rangle}\Varid{run_{StateT}}\;(\Varid{h_{State}}\;\Varid{pop_{SS}})\;(\Conid{SS}\;[\mskip1.5mu \mskip1.5mu]\;[\mskip1.5mu \mskip1.5mu]){}\<[E]%
\\
\>[3]{}\mathrel{=}\mbox{\commentbegin ~  \Cref{eq:eval-pop1-f}   \commentend}{}\<[E]%
\\
\>[3]{}\hsindent{3}{}\<[6]%
\>[6]{}\Varid{results_{SS}}\hsdot{\circ }{.}\Varid{snd}\mathbin{\langle\hspace{1.6pt}\mathclap{\raisebox{0.1pt}{\scalebox{1}{\$}}}\hspace{1.6pt}\rangle}\Varid{\eta}\;((),\Conid{SS}\;[\mskip1.5mu \mskip1.5mu]\;[\mskip1.5mu \mskip1.5mu]){}\<[E]%
\\
\>[3]{}\mathrel{=}\mbox{\commentbegin ~  evaluation of \ensuremath{\Varid{snd},\Varid{results_{SS}}}   \commentend}{}\<[E]%
\\
\>[3]{}\hsindent{3}{}\<[6]%
\>[6]{}\Varid{\eta}\;[\mskip1.5mu \mskip1.5mu]{}\<[E]%
\\
\>[3]{}\mathrel{=}\mbox{\commentbegin ~  definition of \ensuremath{\Varid{\eta}} for free monad   \commentend}{}\<[E]%
\\
\>[3]{}\hsindent{3}{}\<[6]%
\>[6]{}\Conid{Var}\;[\mskip1.5mu \mskip1.5mu]{}\<[E]%
\\
\>[3]{}\mathrel{=}\mbox{\commentbegin ~  definition of \ensuremath{\Varid{alg_{ND+f}}}   \commentend}{}\<[E]%
\\
\>[3]{}\hsindent{3}{}\<[6]%
\>[6]{}\Varid{alg_{ND+f}}\;\Conid{Fail}{}\<[E]%
\\
\>[3]{}\mathrel{=}\mbox{\commentbegin ~  definition of \ensuremath{\Varid{fmap}}   \commentend}{}\<[E]%
\\
\>[3]{}\hsindent{3}{}\<[6]%
\>[6]{}(\Varid{alg_{ND+f}}\hsdot{\circ }{.}\Varid{fmap}\;(\Varid{extract_{SS}}\hsdot{\circ }{.}\Varid{h_{State}}))\;\Conid{Fail}{}\<[E]%
\ColumnHook
\end{hscode}\resethooks
\indentend %

\noindent \mbox{\underline{case \ensuremath{\Conid{Inl}\;(\Conid{Or}\;\Varid{p}\;\Varid{q})}}}

\indentbegin \begin{hscode}\SaveRestoreHook
\column{B}{@{}>{\hspre}l<{\hspost}@{}}%
\column{3}{@{}>{\hspre}l<{\hspost}@{}}%
\column{6}{@{}>{\hspre}l<{\hspost}@{}}%
\column{8}{@{}>{\hspre}l<{\hspost}@{}}%
\column{E}{@{}>{\hspre}l<{\hspost}@{}}%
\>[6]{}\Varid{extract_{SS}}\;(\Varid{h_{State}}\;(\Varid{alg}\;(\Conid{Or}\;\Varid{p}\;\Varid{q}))){}\<[E]%
\\
\>[3]{}\mathrel{=}\mbox{\commentbegin ~  definition of \ensuremath{\Varid{alg}}   \commentend}{}\<[E]%
\\
\>[3]{}\hsindent{3}{}\<[6]%
\>[6]{}\Varid{extract_{SS}}\;(\Varid{h_{State}}\;(\Varid{push_{SS}}\;\Varid{q}\;\Varid{p})){}\<[E]%
\\
\>[3]{}\mathrel{=}\mbox{\commentbegin ~  definition of \ensuremath{\Varid{extract_{SS}}}   \commentend}{}\<[E]%
\\
\>[3]{}\hsindent{3}{}\<[6]%
\>[6]{}\Varid{results_{SS}}\hsdot{\circ }{.}\Varid{snd}\mathbin{\langle\hspace{1.6pt}\mathclap{\raisebox{0.1pt}{\scalebox{1}{\$}}}\hspace{1.6pt}\rangle}\Varid{run_{StateT}}\;(\Varid{h_{State}}\;(\Varid{push_{SS}}\;\Varid{q}\;\Varid{p}))\;(\Conid{SS}\;[\mskip1.5mu \mskip1.5mu]\;[\mskip1.5mu \mskip1.5mu]){}\<[E]%
\\
\>[3]{}\mathrel{=}\mbox{\commentbegin ~  \Cref{eq:eval-push-f}   \commentend}{}\<[E]%
\\
\>[3]{}\hsindent{3}{}\<[6]%
\>[6]{}\Varid{results_{SS}}\hsdot{\circ }{.}\Varid{snd}\mathbin{\langle\hspace{1.6pt}\mathclap{\raisebox{0.1pt}{\scalebox{1}{\$}}}\hspace{1.6pt}\rangle}\Varid{run_{StateT}}\;(\Varid{h_{State}}\;\Varid{p})\;(\Conid{SS}\;[\mskip1.5mu \mskip1.5mu]\;[\mskip1.5mu \Varid{q}\mskip1.5mu]){}\<[E]%
\\
\>[3]{}\mathrel{=}\mbox{\commentbegin ~  \Cref{eq:pop-extract-f}   \commentend}{}\<[E]%
\\
\>[3]{}\hsindent{3}{}\<[6]%
\>[6]{}\Varid{results_{SS}}\hsdot{\circ }{.}\Varid{snd}\mathbin{\langle\hspace{1.6pt}\mathclap{\raisebox{0.1pt}{\scalebox{1}{\$}}}\hspace{1.6pt}\rangle}{}\<[E]%
\\
\>[6]{}\hsindent{2}{}\<[8]%
\>[8]{}\mathbf{do}\;\{\mskip1.5mu \Varid{p'}\leftarrow \Varid{extract_{SS}}\;(\Varid{h_{State}}\;\Varid{p});\Varid{run_{StateT}}\;(\Varid{h_{State}}\;\Varid{pop_{SS}})\;(\Conid{SS}\;([\mskip1.5mu \mskip1.5mu]+\!\!+\Varid{p'})\;[\mskip1.5mu \Varid{q}\mskip1.5mu])\mskip1.5mu\}{}\<[E]%
\\
\>[3]{}\mathrel{=}\mbox{\commentbegin ~  definition of \ensuremath{(+\!\!+)}   \commentend}{}\<[E]%
\\
\>[3]{}\hsindent{3}{}\<[6]%
\>[6]{}\Varid{results_{SS}}\hsdot{\circ }{.}\Varid{snd}\mathbin{\langle\hspace{1.6pt}\mathclap{\raisebox{0.1pt}{\scalebox{1}{\$}}}\hspace{1.6pt}\rangle}\mathbf{do}\;\{\mskip1.5mu \Varid{p'}\leftarrow \Varid{extract_{SS}}\;(\Varid{h_{State}}\;\Varid{p});\Varid{run_{StateT}}\;(\Varid{h_{State}}\;\Varid{pop_{SS}})\;(\Conid{SS}\;\Varid{p'}\;[\mskip1.5mu \Varid{q}\mskip1.5mu])\mskip1.5mu\}{}\<[E]%
\\
\>[3]{}\mathrel{=}\mbox{\commentbegin ~  \Cref{eq:eval-pop2-f}   \commentend}{}\<[E]%
\\
\>[3]{}\hsindent{3}{}\<[6]%
\>[6]{}\Varid{results_{SS}}\hsdot{\circ }{.}\Varid{snd}\mathbin{\langle\hspace{1.6pt}\mathclap{\raisebox{0.1pt}{\scalebox{1}{\$}}}\hspace{1.6pt}\rangle}\mathbf{do}\;\{\mskip1.5mu \Varid{p'}\leftarrow \Varid{extract_{SS}}\;(\Varid{h_{State}}\;\Varid{p});\Varid{run_{StateT}}\;(\Varid{h_{State}}\;\Varid{q})\;(\Conid{SS}\;\Varid{p'}\;[\mskip1.5mu \mskip1.5mu])\mskip1.5mu\}{}\<[E]%
\\
\>[3]{}\mathrel{=}\mbox{\commentbegin ~  \Cref{eq:pop-extract-f}   \commentend}{}\<[E]%
\\
\>[3]{}\hsindent{3}{}\<[6]%
\>[6]{}\Varid{results_{SS}}\hsdot{\circ }{.}\Varid{snd}\mathbin{\langle\hspace{1.6pt}\mathclap{\raisebox{0.1pt}{\scalebox{1}{\$}}}\hspace{1.6pt}\rangle}\mathbf{do}\;\{\mskip1.5mu \Varid{p'}\leftarrow \Varid{extract_{SS}}\;(\Varid{h_{State}}\;\Varid{p});{}\<[E]%
\\
\>[6]{}\hsindent{2}{}\<[8]%
\>[8]{}\mathbf{do}\;\{\mskip1.5mu \Varid{q'}\leftarrow \Varid{extract_{SS}}\;(\Varid{h_{State}}\;\Varid{q});\Varid{run_{StateT}}\;(\Varid{h_{State}}\;\Varid{pop_{SS}})\;(\Conid{SS}\;(\Varid{p'}+\!\!+\Varid{q'})\;[\mskip1.5mu \mskip1.5mu])\mskip1.5mu\}\mskip1.5mu\}{}\<[E]%
\\
\>[3]{}\mathrel{=}\mbox{\commentbegin ~  Law (\ref{eq:monad-assoc}) for \ensuremath{\mathbf{do}}-notation   \commentend}{}\<[E]%
\\
\>[3]{}\hsindent{3}{}\<[6]%
\>[6]{}\Varid{results_{SS}}\hsdot{\circ }{.}\Varid{snd}\mathbin{\langle\hspace{1.6pt}\mathclap{\raisebox{0.1pt}{\scalebox{1}{\$}}}\hspace{1.6pt}\rangle}\mathbf{do}\;\{\mskip1.5mu \Varid{p'}\leftarrow \Varid{extract_{SS}}\;(\Varid{h_{State}}\;\Varid{p});\Varid{q'}\leftarrow \Varid{extract_{SS}}\;(\Varid{h_{State}}\;\Varid{q});{}\<[E]%
\\
\>[6]{}\hsindent{2}{}\<[8]%
\>[8]{}\Varid{run_{StateT}}\;(\Varid{h_{State}}\;\Varid{pop_{SS}})\;(\Conid{SS}\;(\Varid{p'}+\!\!+\Varid{q'})\;[\mskip1.5mu \mskip1.5mu])\mskip1.5mu\}{}\<[E]%
\\
\>[3]{}\mathrel{=}\mbox{\commentbegin ~  \Cref{eq:eval-pop1-f}   \commentend}{}\<[E]%
\\
\>[3]{}\hsindent{3}{}\<[6]%
\>[6]{}\Varid{results_{SS}}\hsdot{\circ }{.}\Varid{snd}\mathbin{\langle\hspace{1.6pt}\mathclap{\raisebox{0.1pt}{\scalebox{1}{\$}}}\hspace{1.6pt}\rangle}{}\<[E]%
\\
\>[6]{}\hsindent{2}{}\<[8]%
\>[8]{}\mathbf{do}\;\{\mskip1.5mu \Varid{p'}\leftarrow \Varid{extract_{SS}}\;(\Varid{h_{State}}\;\Varid{p});\Varid{q'}\leftarrow \Varid{extract_{SS}}\;(\Varid{h_{State}}\;\Varid{q});\Varid{\eta}\;((),\Conid{SS}\;(\Varid{p'}+\!\!+\Varid{q'})\;[\mskip1.5mu \mskip1.5mu])\mskip1.5mu\}{}\<[E]%
\\
\>[3]{}\mathrel{=}\mbox{\commentbegin ~  evaluation of \ensuremath{\Varid{snd},\Varid{results_{SS}}}   \commentend}{}\<[E]%
\\
\>[3]{}\hsindent{3}{}\<[6]%
\>[6]{}\mathbf{do}\;\{\mskip1.5mu \Varid{p'}\leftarrow \Varid{extract_{SS}}\;(\Varid{h_{State}}\;\Varid{p});\Varid{q'}\leftarrow \Varid{extract_{SS}}\;(\Varid{h_{State}}\;\Varid{q});\Varid{\eta}\;(\Varid{p'}+\!\!+\Varid{q'})\mskip1.5mu\}{}\<[E]%
\\
\>[3]{}\mathrel{=}\mbox{\commentbegin ~  definition of \ensuremath{\Varid{liftM2}}   \commentend}{}\<[E]%
\\
\>[3]{}\hsindent{3}{}\<[6]%
\>[6]{}\Varid{liftM2}\;(+\!\!+)\;((\Varid{extract_{SS}}\hsdot{\circ }{.}\Varid{h_{State}})\;\Varid{p})\;((\Varid{extract_{SS}}\hsdot{\circ }{.}\Varid{h_{State}})\;\Varid{q}){}\<[E]%
\\
\>[3]{}\mathrel{=}\mbox{\commentbegin ~  definition of \ensuremath{\Varid{alg_{ND+f}}}   \commentend}{}\<[E]%
\\
\>[3]{}\hsindent{3}{}\<[6]%
\>[6]{}\Varid{alg_{ND+f}}\;(\Conid{Or}\;((\Varid{extract_{SS}}\hsdot{\circ }{.}\Varid{h_{State}})\;\Varid{p})\;((\Varid{extract_{SS}}\hsdot{\circ }{.}\Varid{h_{State}})\;\Varid{q})){}\<[E]%
\\
\>[3]{}\mathrel{=}\mbox{\commentbegin ~  definition of \ensuremath{\Varid{fmap}}   \commentend}{}\<[E]%
\\
\>[3]{}\hsindent{3}{}\<[6]%
\>[6]{}(\Varid{alg_{ND+f}}\hsdot{\circ }{.}\Varid{fmap}\;(\Varid{extract_{SS}}\hsdot{\circ }{.}\Varid{h_{State}}))\;(\Conid{Or}\;\Varid{p}\;\Varid{q}){}\<[E]%
\ColumnHook
\end{hscode}\resethooks
\indentend %

For the last equation, we calculate as follows:

\indentbegin \begin{hscode}\SaveRestoreHook
\column{B}{@{}>{\hspre}l<{\hspost}@{}}%
\column{3}{@{}>{\hspre}l<{\hspost}@{}}%
\column{6}{@{}>{\hspre}l<{\hspost}@{}}%
\column{8}{@{}>{\hspre}l<{\hspost}@{}}%
\column{E}{@{}>{\hspre}l<{\hspost}@{}}%
\>[6]{}\Varid{extract_{SS}}\;(\Varid{h_{State}}\;(\Varid{fwd}\;\Varid{y})){}\<[E]%
\\
\>[3]{}\mathrel{=}\mbox{\commentbegin ~  definition of \ensuremath{\Varid{fwd}}   \commentend}{}\<[E]%
\\
\>[3]{}\hsindent{3}{}\<[6]%
\>[6]{}\Varid{extract_{SS}}\;(\Varid{h_{State}}\;(\Conid{Op}\;(\Conid{Inr}\;\Varid{y}))){}\<[E]%
\\
\>[3]{}\mathrel{=}\mbox{\commentbegin ~  definition of \ensuremath{\Varid{h_{State}}}   \commentend}{}\<[E]%
\\
\>[3]{}\hsindent{3}{}\<[6]%
\>[6]{}\Varid{extract_{SS}}\;(\Conid{StateT}\mathbin{\$}\lambda \Varid{s}\to \Conid{Op}\mathbin{\$}\Varid{fmap}\;(\lambda \Varid{k}\to \Varid{run_{StateT}}\;\Varid{k}\;\Varid{s})\;(\Varid{fmap}\;\Varid{h_{State}}\;\Varid{y})){}\<[E]%
\\
\>[3]{}\mathrel{=}\mbox{\commentbegin ~  definition of \ensuremath{\Varid{extract_{SS}}}   \commentend}{}\<[E]%
\\
\>[3]{}\hsindent{3}{}\<[6]%
\>[6]{}\Varid{results_{SS}}\hsdot{\circ }{.}\Varid{snd}\mathbin{\langle\hspace{1.6pt}\mathclap{\raisebox{0.1pt}{\scalebox{1}{\$}}}\hspace{1.6pt}\rangle}{}\<[E]%
\\
\>[6]{}\hsindent{2}{}\<[8]%
\>[8]{}\Varid{run_{StateT}}\;(\Conid{StateT}\mathbin{\$}\lambda \Varid{s}\to \Conid{Op}\mathbin{\$}\Varid{fmap}\;(\lambda \Varid{k}\to \Varid{run_{StateT}}\;\Varid{k}\;\Varid{s})\;(\Varid{fmap}\;\Varid{h_{State}}\;\Varid{y}))\;(\Conid{SS}\;[\mskip1.5mu \mskip1.5mu]\;[\mskip1.5mu \mskip1.5mu]){}\<[E]%
\\
\>[3]{}\mathrel{=}\mbox{\commentbegin ~  definition of \ensuremath{\Varid{run_{StateT}}}   \commentend}{}\<[E]%
\\
\>[3]{}\hsindent{3}{}\<[6]%
\>[6]{}\Varid{results_{SS}}\hsdot{\circ }{.}\Varid{snd}\mathbin{\langle\hspace{1.6pt}\mathclap{\raisebox{0.1pt}{\scalebox{1}{\$}}}\hspace{1.6pt}\rangle}(\lambda \Varid{s}\to \Conid{Op}\mathbin{\$}\Varid{fmap}\;(\lambda \Varid{k}\to \Varid{run_{StateT}}\;\Varid{k}\;\Varid{s})\;(\Varid{fmap}\;\Varid{h_{State}}\;\Varid{y}))\;(\Conid{SS}\;[\mskip1.5mu \mskip1.5mu]\;[\mskip1.5mu \mskip1.5mu]){}\<[E]%
\\
\>[3]{}\mathrel{=}\mbox{\commentbegin ~  function application   \commentend}{}\<[E]%
\\
\>[3]{}\hsindent{3}{}\<[6]%
\>[6]{}\Varid{results_{SS}}\hsdot{\circ }{.}\Varid{snd}\mathbin{\langle\hspace{1.6pt}\mathclap{\raisebox{0.1pt}{\scalebox{1}{\$}}}\hspace{1.6pt}\rangle}\Conid{Op}\mathbin{\$}\Varid{fmap}\;(\lambda \Varid{k}\to \Varid{run_{StateT}}\;\Varid{k}\;(\Conid{SS}\;[\mskip1.5mu \mskip1.5mu]\;[\mskip1.5mu \mskip1.5mu]))\;(\Varid{fmap}\;\Varid{h_{State}}\;\Varid{y})){}\<[E]%
\\
\>[3]{}\mathrel{=}\mbox{\commentbegin ~  definition of \ensuremath{\mathbin{\langle\hspace{1.6pt}\mathclap{\raisebox{0.1pt}{\scalebox{1}{\$}}}\hspace{1.6pt}\rangle}}   \commentend}{}\<[E]%
\\
\>[3]{}\hsindent{3}{}\<[6]%
\>[6]{}\Conid{Op}\;(\Varid{fmap}\;(\lambda \Varid{k}\to \Varid{results_{SS}}\hsdot{\circ }{.}\Varid{snd}\mathbin{\langle\hspace{1.6pt}\mathclap{\raisebox{0.1pt}{\scalebox{1}{\$}}}\hspace{1.6pt}\rangle}\Varid{run_{StateT}}\;\Varid{k}\;(\Conid{SS}\;[\mskip1.5mu \mskip1.5mu]\;[\mskip1.5mu \mskip1.5mu]))\;(\Varid{fmap}\;\Varid{h_{State}}\;\Varid{y})){}\<[E]%
\\
\>[3]{}\mathrel{=}\mbox{\commentbegin ~  definition of \ensuremath{\Varid{fwd_{ND+f}}}   \commentend}{}\<[E]%
\\
\>[3]{}\hsindent{3}{}\<[6]%
\>[6]{}\Varid{fwd_{ND+f}}\;(\Varid{fmap}\;(\lambda \Varid{k}\to \Varid{results_{SS}}\hsdot{\circ }{.}\Varid{snd}\mathbin{\langle\hspace{1.6pt}\mathclap{\raisebox{0.1pt}{\scalebox{1}{\$}}}\hspace{1.6pt}\rangle}\Varid{run_{StateT}}\;\Varid{k}\;(\Conid{SS}\;[\mskip1.5mu \mskip1.5mu]\;[\mskip1.5mu \mskip1.5mu]))\;(\Varid{fmap}\;\Varid{h_{State}}\;\Varid{y})){}\<[E]%
\\
\>[3]{}\mathrel{=}\mbox{\commentbegin ~  definition of \ensuremath{\Varid{extract_{SS}}}   \commentend}{}\<[E]%
\\
\>[3]{}\hsindent{3}{}\<[6]%
\>[6]{}\Varid{fwd_{ND+f}}\;(\Varid{fmap}\;\Varid{extract_{SS}}\;(\Varid{fmap}\;\Varid{h_{State}}\;\Varid{y})){}\<[E]%
\\
\>[3]{}\mathrel{=}\mbox{\commentbegin ~  Law (\ref{eq:functor-composition})   \commentend}{}\<[E]%
\\
\>[3]{}\hsindent{3}{}\<[6]%
\>[6]{}\Varid{fwd_{ND+f}}\;(\Varid{fmap}\;(\Varid{extract_{SS}}\hsdot{\circ }{.}\Varid{h_{State}})\;\Varid{y}){}\<[E]%
\ColumnHook
\end{hscode}\resethooks
\indentend %

\end{proof}

In the above proof we have used several lemmas. Now we prove them.

\begin{lemma}[pop-extract of \ensuremath{\Conid{SS}}]\label{eq:pop-extract-f}
~\indentbegin \begin{hscode}\SaveRestoreHook
\column{B}{@{}>{\hspre}l<{\hspost}@{}}%
\column{3}{@{}>{\hspre}c<{\hspost}@{}}%
\column{3E}{@{}l@{}}%
\column{6}{@{}>{\hspre}l<{\hspost}@{}}%
\column{E}{@{}>{\hspre}l<{\hspost}@{}}%
\>[6]{}\Varid{run_{StateT}}\;(\Varid{h_{State}}\;\Varid{p})\;(\Conid{SS}\;\Varid{xs}\;\Varid{stack}){}\<[E]%
\\
\>[3]{}\mathrel{=}{}\<[3E]%
\>[6]{}\mathbf{do}\;\{\mskip1.5mu \Varid{p'}\leftarrow \Varid{extract_{SS}}\;(\Varid{h_{State}}\;\Varid{p});\Varid{run_{StateT}}\;(\Varid{h_{State}}\;\Varid{pop_{SS}})\;(\Conid{SS}\;(\Varid{xs}+\!\!+\Varid{p'})\;\Varid{stack})\mskip1.5mu\}{}\<[E]%
\ColumnHook
\end{hscode}\resethooks
\indentend holds for all \ensuremath{\Varid{p}} in the codomain of the function \ensuremath{\Varid{nondet2state}}.
\end{lemma}

\begin{proof} ~
The proof structure is similar to that of \Cref{eq:pop-extract}.
We prove this lemma by structural induction on \ensuremath{\Varid{p}\mathbin{::}\Conid{Free}\;(\Varid{State_{F}}\;(\Conid{SS}\;\Varid{f}\;\Varid{a})\mathrel{{:}{+}{:}}\Varid{f})\;()}.
For each inductive case of \ensuremath{\Varid{p}}, we not only assume this lemma holds
for its sub-terms (this is the standard induction hypothesis), but
also assume \Cref{thm:nondet-state} holds for \ensuremath{\Varid{p}} and its sub-terms.
This is sound because in the proof of \Cref{thm:nondet-state}, for
\ensuremath{(\Varid{extract_{SS}}\hsdot{\circ }{.}\Varid{h_{State}}\hsdot{\circ }{.}\Varid{nondet2state})\;\Varid{p}\mathrel{=}\Varid{h_{ND+f}}\;\Varid{p}}, we only apply
\Cref{eq:pop-extract-f} to the sub-terms of \ensuremath{\Varid{p}}, which is already
included in the induction hypothesis so there is no circular argument.

Since we assume \Cref{thm:nondet-state} holds for \ensuremath{\Varid{p}} and its
sub-terms, we can use several useful properties proved in the
sub-cases of the proof of \Cref{thm:nondet-state}. We list them here
for easy reference:
\begin{itemize}
\item {extract-gen-ext}:
\ensuremath{\Varid{extract_{SS}}\hsdot{\circ }{.}\Varid{h_{State}}\hsdot{\circ }{.}\Varid{gen}\mathrel{=}\Conid{Var}\hsdot{\circ }{.}\Varid{\eta}}
\item {extract-alg1-ext}:
\ensuremath{\Varid{extract_{SS}}\;(\Varid{h_{State}}\;(\Varid{alg}\;\Conid{Fail}))\mathrel{=}\Conid{Var}\;[\mskip1.5mu \mskip1.5mu]}
\item {extract-alg2-ext}:
\ensuremath{\Varid{extract_{SS}}\;(\Varid{h_{State}}\;(\Varid{alg}\;(\Conid{Or}\;\Varid{p}\;\Varid{q})))\mathrel{=}\Varid{liftM2}\;(+\!\!+)\;(\Varid{extract_{SS}}\;(\Varid{h_{State}}\;\Varid{p}))\;(\Varid{extract_{SS}}\;(\Varid{h_{State}}\;\Varid{q}))}
\item {extract-fwd}
\ensuremath{\Varid{extract_{SS}}\;(\Varid{h_{State}}\;(\Varid{fwd}\;\Varid{y}))\mathrel{=}\Varid{fwd_{ND+f}}\;(\Varid{fmap}\;(\Varid{extract_{SS}}\hsdot{\circ }{.}\Varid{h_{State}})\;\Varid{y})}
\end{itemize}

We proceed by structural induction on \ensuremath{\Varid{p}}.
Note that for all \ensuremath{\Varid{p}} in the codomain of \ensuremath{\Varid{nondet2state}}, it is either
generated by the \ensuremath{\Varid{gen}}, \ensuremath{\Varid{alg}}, or \ensuremath{\Varid{fwd}} of \ensuremath{\Varid{nondet2state}}.  Thus, we
only need to prove the following three equations where \ensuremath{\Varid{x}} is in the
codomain of \ensuremath{\Varid{fmap}\;\Varid{nondet2state_S}} and \ensuremath{\Varid{p}\mathrel{=}\Varid{gen}\;\Varid{x}}, \ensuremath{\Varid{p}\mathrel{=}\Varid{alg}\;\Varid{x}}, and \ensuremath{\Varid{p}\mathrel{=}\Varid{fwd}\;\Varid{x}}, respectively.
\begin{enumerate}
    \item\indentbegin \begin{hscode}\SaveRestoreHook
\column{B}{@{}>{\hspre}l<{\hspost}@{}}%
\column{3}{@{}>{\hspre}c<{\hspost}@{}}%
\column{3E}{@{}l@{}}%
\column{6}{@{}>{\hspre}l<{\hspost}@{}}%
\column{E}{@{}>{\hspre}l<{\hspost}@{}}%
\>[6]{}\Varid{run_{StateT}}\;(\Varid{h_{State}}\;(\Varid{gen}\;\Varid{x}))\;(\Conid{SS}\;\Varid{xs}\;\Varid{stack}){}\<[E]%
\\
\>[3]{}\mathrel{=}{}\<[3E]%
\>[6]{}\mathbf{do}\;\{\mskip1.5mu \Varid{p'}\leftarrow \Varid{extract_{SS}}\;(\Varid{h_{State}}\;(\Varid{gen}\;\Varid{x}));\Varid{run_{StateT}}\;(\Varid{h_{State}}\;\Varid{pop_{SS}})\;(\Conid{SS}\;(\Varid{xs}+\!\!+\Varid{p'})\;\Varid{stack})\mskip1.5mu\}{}\<[E]%
\ColumnHook
\end{hscode}\resethooks
\indentend     \item\indentbegin \begin{hscode}\SaveRestoreHook
\column{B}{@{}>{\hspre}l<{\hspost}@{}}%
\column{3}{@{}>{\hspre}c<{\hspost}@{}}%
\column{3E}{@{}l@{}}%
\column{6}{@{}>{\hspre}l<{\hspost}@{}}%
\column{E}{@{}>{\hspre}l<{\hspost}@{}}%
\>[6]{}\Varid{run_{StateT}}\;(\Varid{h_{State}}\;(\Varid{alg}\;\Varid{x}))\;(\Conid{SS}\;\Varid{xs}\;\Varid{stack}){}\<[E]%
\\
\>[3]{}\mathrel{=}{}\<[3E]%
\>[6]{}\mathbf{do}\;\{\mskip1.5mu \Varid{p'}\leftarrow \Varid{extract_{SS}}\;(\Varid{h_{State}}\;(\Varid{alg}\;\Varid{x}));\Varid{run_{StateT}}\;(\Varid{h_{State}}\;\Varid{pop_{SS}})\;(\Conid{SS}\;(\Varid{xs}+\!\!+\Varid{p'})\;\Varid{stack})\mskip1.5mu\}{}\<[E]%
\ColumnHook
\end{hscode}\resethooks
\indentend     \item\indentbegin \begin{hscode}\SaveRestoreHook
\column{B}{@{}>{\hspre}l<{\hspost}@{}}%
\column{3}{@{}>{\hspre}c<{\hspost}@{}}%
\column{3E}{@{}l@{}}%
\column{6}{@{}>{\hspre}l<{\hspost}@{}}%
\column{E}{@{}>{\hspre}l<{\hspost}@{}}%
\>[6]{}\Varid{run_{StateT}}\;(\Varid{h_{State}}\;(\Varid{fwd}\;\Varid{x}))\;(\Conid{SS}\;\Varid{xs}\;\Varid{stack}){}\<[E]%
\\
\>[3]{}\mathrel{=}{}\<[3E]%
\>[6]{}\mathbf{do}\;\{\mskip1.5mu \Varid{p'}\leftarrow \Varid{extract_{SS}}\;(\Varid{h_{State}}\;(\Varid{fwd}\;\Varid{x}));\Varid{run_{StateT}}\;(\Varid{h_{State}}\;\Varid{pop_{SS}})\;(\Conid{SS}\;(\Varid{xs}+\!\!+\Varid{p'})\;\Varid{stack})\mskip1.5mu\}{}\<[E]%
\ColumnHook
\end{hscode}\resethooks
\indentend \end{enumerate}

For the case \ensuremath{\Varid{p}\mathrel{=}\Varid{gen}\;\Varid{x}}, we calculate as follows:
\indentbegin \begin{hscode}\SaveRestoreHook
\column{B}{@{}>{\hspre}l<{\hspost}@{}}%
\column{3}{@{}>{\hspre}l<{\hspost}@{}}%
\column{6}{@{}>{\hspre}l<{\hspost}@{}}%
\column{E}{@{}>{\hspre}l<{\hspost}@{}}%
\>[6]{}\Varid{run_{StateT}}\;(\Varid{h_{State}}\;(\Varid{gen}\;\Varid{x}))\;(\Conid{SS}\;\Varid{xs}\;\Varid{stack}){}\<[E]%
\\
\>[3]{}\mathrel{=}\mbox{\commentbegin ~  definition of \ensuremath{\Varid{gen}}   \commentend}{}\<[E]%
\\
\>[3]{}\hsindent{3}{}\<[6]%
\>[6]{}\Varid{run_{StateT}}\;(\Varid{h_{State}}\;(\Varid{append_{SS}}\;\Varid{x}\;\Varid{pop_{SS}}))\;(\Conid{SS}\;\Varid{xs}\;\Varid{stack}){}\<[E]%
\\
\>[3]{}\mathrel{=}\mbox{\commentbegin ~  \Cref{eq:eval-append-f}   \commentend}{}\<[E]%
\\
\>[3]{}\hsindent{3}{}\<[6]%
\>[6]{}\Varid{run_{StateT}}\;(\Varid{h_{State}}\;\Varid{pop_{SS}})\;(\Conid{SS}\;(\Varid{xs}+\!\!+[\mskip1.5mu \Varid{x}\mskip1.5mu])\;\Varid{stack}){}\<[E]%
\\
\>[3]{}\mathrel{=}\mbox{\commentbegin ~  definition of \ensuremath{\Varid{\eta}_{[]}}   \commentend}{}\<[E]%
\\
\>[3]{}\hsindent{3}{}\<[6]%
\>[6]{}\Varid{run_{StateT}}\;(\Varid{h_{State}}\;\Varid{pop_{SS}})\;(\Conid{SS}\;(\Varid{xs}+\!\!+\Varid{\eta}\;\Varid{x})\;\Varid{stack}){}\<[E]%
\\
\>[3]{}\mathrel{=}\mbox{\commentbegin ~  definition of \ensuremath{\Conid{Var}} and reformulation   \commentend}{}\<[E]%
\\
\>[3]{}\hsindent{3}{}\<[6]%
\>[6]{}\mathbf{do}\;\{\mskip1.5mu \Varid{p'}\leftarrow \Conid{Var}\;(\Varid{\eta}\;\Varid{x});\Varid{run_{StateT}}\;(\Varid{h_{State}}\;\Varid{pop_{SS}})\;(\Conid{SS}\;(\Varid{xs}+\!\!+\Varid{p'})\;\Varid{stack})\mskip1.5mu\}{}\<[E]%
\\
\>[3]{}\mathrel{=}\mbox{\commentbegin ~  extract-gen-ext   \commentend}{}\<[E]%
\\
\>[3]{}\hsindent{3}{}\<[6]%
\>[6]{}\mathbf{do}\;\{\mskip1.5mu \Varid{p'}\leftarrow \Varid{extract_{SS}}\;(\Varid{h_{State}}\;(\Varid{gen}\;\Varid{x}));\Varid{run_{StateT}}\;(\Varid{h_{State}}\;\Varid{pop_{SS}})\;(\Conid{SS}\;(\Varid{xs}+\!\!+\Varid{p'})\;\Varid{stack})\mskip1.5mu\}{}\<[E]%
\ColumnHook
\end{hscode}\resethooks
\indentend For the case \ensuremath{\Varid{p}\mathrel{=}\Varid{alg}\;\Varid{x}}, we proceed with a case analysis on \ensuremath{\Varid{x}}.

\noindent \mbox{\underline{case \ensuremath{\Conid{Fail}}}}
\indentbegin \begin{hscode}\SaveRestoreHook
\column{B}{@{}>{\hspre}l<{\hspost}@{}}%
\column{3}{@{}>{\hspre}l<{\hspost}@{}}%
\column{6}{@{}>{\hspre}l<{\hspost}@{}}%
\column{E}{@{}>{\hspre}l<{\hspost}@{}}%
\>[6]{}\Varid{run_{StateT}}\;(\Varid{h_{State}}\;(\Varid{alg}\;\Conid{Fail}))\;(\Conid{SS}\;\Varid{xs}\;\Varid{stack}){}\<[E]%
\\
\>[3]{}\mathrel{=}\mbox{\commentbegin ~  definition of \ensuremath{\Varid{alg}}   \commentend}{}\<[E]%
\\
\>[3]{}\hsindent{3}{}\<[6]%
\>[6]{}\Varid{run_{StateT}}\;(\Varid{h_{State}}\;\Varid{pop_{SS}})\;(\Conid{SS}\;\Varid{xs}\;\Varid{stack}){}\<[E]%
\\
\>[3]{}\mathrel{=}\mbox{\commentbegin ~  definition of \ensuremath{[\mskip1.5mu \mskip1.5mu]}   \commentend}{}\<[E]%
\\
\>[3]{}\hsindent{3}{}\<[6]%
\>[6]{}\Varid{run_{StateT}}\;(\Varid{h_{State}}\;\Varid{pop_{SS}})\;(\Conid{SS}\;(\Varid{xs}+\!\!+[\mskip1.5mu \mskip1.5mu])\;\Varid{stack}){}\<[E]%
\\
\>[3]{}\mathrel{=}\mbox{\commentbegin ~  definition of \ensuremath{\Conid{Var}}   \commentend}{}\<[E]%
\\
\>[3]{}\hsindent{3}{}\<[6]%
\>[6]{}\mathbf{do}\;\{\mskip1.5mu \Varid{p'}\leftarrow \Conid{Var}\;[\mskip1.5mu \mskip1.5mu];\Varid{run_{StateT}}\;(\Varid{h_{State}}\;\Varid{pop_{SS}})\;(\Conid{SS}\;(\Varid{xs}+\!\!+\Varid{p'})\;\Varid{stack})\mskip1.5mu\}{}\<[E]%
\\
\>[3]{}\mathrel{=}\mbox{\commentbegin ~  extract-alg1-ext   \commentend}{}\<[E]%
\\
\>[3]{}\hsindent{3}{}\<[6]%
\>[6]{}\mathbf{do}\;\{\mskip1.5mu \Varid{p'}\leftarrow \Varid{extract_{SS}}\;(\Varid{h_{State}}\;(\Varid{alg}\;\Conid{Fail}));\Varid{run_{StateT}}\;(\Varid{h_{State}}\;\Varid{pop_{SS}})\;(\Conid{SS}\;(\Varid{xs}+\!\!+\Varid{p'})\;\Varid{stack})\mskip1.5mu\}{}\<[E]%
\ColumnHook
\end{hscode}\resethooks
\indentend \noindent \mbox{\underline{case \ensuremath{\Conid{Or}\;\Varid{p}_{1}\;\Varid{p}_{2}}}}

\indentbegin \begin{hscode}\SaveRestoreHook
\column{B}{@{}>{\hspre}l<{\hspost}@{}}%
\column{3}{@{}>{\hspre}l<{\hspost}@{}}%
\column{6}{@{}>{\hspre}l<{\hspost}@{}}%
\column{8}{@{}>{\hspre}l<{\hspost}@{}}%
\column{E}{@{}>{\hspre}l<{\hspost}@{}}%
\>[6]{}\Varid{run_{StateT}}\;(\Varid{h_{State}}\;(\Varid{alg}\;(\Conid{Or}\;\Varid{p}_{1}\;\Varid{p}_{2})))\;(\Conid{SS}\;\Varid{xs}\;\Varid{stack}){}\<[E]%
\\
\>[3]{}\mathrel{=}\mbox{\commentbegin ~  definition of \ensuremath{\Varid{alg}}   \commentend}{}\<[E]%
\\
\>[3]{}\hsindent{3}{}\<[6]%
\>[6]{}\Varid{run_{StateT}}\;(\Varid{h_{State}}\;(\Varid{push_{SS}}\;\Varid{p}_{2}\;\Varid{p}_{1}))\;(\Conid{SS}\;\Varid{xs}\;\Varid{stack}){}\<[E]%
\\
\>[3]{}\mathrel{=}\mbox{\commentbegin ~  \Cref{eq:eval-push-f}   \commentend}{}\<[E]%
\\
\>[3]{}\hsindent{3}{}\<[6]%
\>[6]{}\Varid{run_{StateT}}\;(\Varid{h_{State}}\;\Varid{p}_{1})\;(\Conid{SS}\;\Varid{xs}\;(\Varid{p}_{2}\mathbin{:}\Varid{stack})){}\<[E]%
\\
\>[3]{}\mathrel{=}\mbox{\commentbegin ~  induction hypothesis: pop-extract of \ensuremath{\Varid{p}_{1}}   \commentend}{}\<[E]%
\\
\>[3]{}\hsindent{3}{}\<[6]%
\>[6]{}\mathbf{do}\;\{\mskip1.5mu \Varid{p}_{1}'\leftarrow \Varid{extract_{SS}}\;(\Varid{h_{State}}\;\Varid{p}_{1});\Varid{run_{StateT}}\;(\Varid{h_{State}}\;\Varid{pop_{SS}})\;(\Conid{SS}\;(\Varid{xs}+\!\!+\Varid{p}_{1}')\;(\Varid{p}_{2}\mathbin{:}\Varid{stack}))\mskip1.5mu\}{}\<[E]%
\\
\>[3]{}\mathrel{=}\mbox{\commentbegin ~  \Cref{eq:eval-pop2-f}   \commentend}{}\<[E]%
\\
\>[3]{}\hsindent{3}{}\<[6]%
\>[6]{}\mathbf{do}\;\{\mskip1.5mu \Varid{p}_{1}'\leftarrow \Varid{extract_{SS}}\;(\Varid{h_{State}}\;\Varid{p}_{1});\Varid{run_{StateT}}\;(\Varid{h_{State}}\;\Varid{p}_{2})\;(\Conid{SS}\;(\Varid{xs}+\!\!+\Varid{p}_{1}')\;\Varid{stack})\mskip1.5mu\}{}\<[E]%
\\
\>[3]{}\mathrel{=}\mbox{\commentbegin ~  induction hypothesis: pop-extract of \ensuremath{\Varid{p}_{2}}   \commentend}{}\<[E]%
\\
\>[3]{}\hsindent{3}{}\<[6]%
\>[6]{}\mathbf{do}\;\{\mskip1.5mu \Varid{p}_{2}'\leftarrow \Varid{extract_{SS}}\;(\Varid{h_{State}}\;\Varid{p}_{2});\mathbf{do}\;\{\mskip1.5mu \Varid{p}_{1}'\leftarrow \Varid{extract_{SS}}\;(\Varid{h_{State}}\;\Varid{p}_{1});{}\<[E]%
\\
\>[6]{}\hsindent{2}{}\<[8]%
\>[8]{}\Varid{run_{StateT}}\;(\Varid{h_{State}}\;\Varid{pop_{SS}})\;(\Conid{SS}\;(\Varid{xs}+\!\!+\Varid{p}_{1}'+\!\!+\Varid{p}_{2}')\;\Varid{stack})\mskip1.5mu\}\mskip1.5mu\}{}\<[E]%
\\
\>[3]{}\mathrel{=}\mbox{\commentbegin ~  Law (\ref{eq:monad-assoc}) with \ensuremath{\mathbf{do}}-notation   \commentend}{}\<[E]%
\\
\>[3]{}\hsindent{3}{}\<[6]%
\>[6]{}\mathbf{do}\;\{\mskip1.5mu \Varid{p}_{2}'\leftarrow \Varid{extract_{SS}}\;(\Varid{h_{State}}\;\Varid{p}_{2});\Varid{p}_{1}'\leftarrow \Varid{extract_{SS}}\;(\Varid{h_{State}}\;\Varid{p}_{1});{}\<[E]%
\\
\>[6]{}\hsindent{2}{}\<[8]%
\>[8]{}\Varid{run_{StateT}}\;(\Varid{h_{State}}\;\Varid{pop_{SS}})\;(\Conid{SS}\;(\Varid{xs}+\!\!+\Varid{p}_{1}'+\!\!+\Varid{p}_{2}')\;\Varid{stack})\mskip1.5mu\}{}\<[E]%
\\
\>[3]{}\mathrel{=}\mbox{\commentbegin ~  definition of \ensuremath{\Varid{liftM2}}   \commentend}{}\<[E]%
\\
\>[3]{}\hsindent{3}{}\<[6]%
\>[6]{}\mathbf{do}\;\{\mskip1.5mu \Varid{p'}\leftarrow \Varid{liftM2}\;(+\!\!+)\;(\Varid{extract_{SS}}\;(\Varid{h_{State}}\;\Varid{p}_{2}))\;(\Varid{extract_{SS}}\;(\Varid{h_{State}}\;\Varid{p}_{1}));{}\<[E]%
\\
\>[6]{}\hsindent{2}{}\<[8]%
\>[8]{}\Varid{run_{StateT}}\;(\Varid{h_{State}}\;\Varid{pop_{SS}})\;(\Conid{SS}\;(\Varid{xs}+\!\!+\Varid{p'})\;\Varid{stack})\mskip1.5mu\}{}\<[E]%
\\
\>[3]{}\mathrel{=}\mbox{\commentbegin ~  extract-alg2-ext   \commentend}{}\<[E]%
\\
\>[3]{}\hsindent{3}{}\<[6]%
\>[6]{}\mathbf{do}\;\{\mskip1.5mu \Varid{p'}\leftarrow \Varid{extract_{SS}}\;(\Varid{h_{State}}\;(\Varid{alg}\;(\Conid{Or}\;\Varid{p}_{1}\;\Varid{p}_{2})));\Varid{run_{StateT}}\;(\Varid{h_{State}}\;\Varid{pop_{SS}})\;(\Conid{SS}\;(\Varid{xs}+\!\!+\Varid{p'})\;\Varid{stack})\mskip1.5mu\}{}\<[E]%
\ColumnHook
\end{hscode}\resethooks
\indentend %
For the case \ensuremath{\Varid{p}\mathrel{=}\Varid{fwd}\;\Varid{x}}, we proceed with a case analysis on \ensuremath{\Varid{x}}.
\indentbegin \begin{hscode}\SaveRestoreHook
\column{B}{@{}>{\hspre}l<{\hspost}@{}}%
\column{3}{@{}>{\hspre}l<{\hspost}@{}}%
\column{5}{@{}>{\hspre}l<{\hspost}@{}}%
\column{6}{@{}>{\hspre}l<{\hspost}@{}}%
\column{E}{@{}>{\hspre}l<{\hspost}@{}}%
\>[6]{}\Varid{run_{StateT}}\;(\Varid{h_{State}}\;(\Varid{fwd}\;\Varid{x}))\;(\Conid{SS}\;\Varid{xs}\;\Varid{stack}){}\<[E]%
\\
\>[3]{}\mathrel{=}\mbox{\commentbegin ~  definition of \ensuremath{\Varid{fwd}}   \commentend}{}\<[E]%
\\
\>[3]{}\hsindent{3}{}\<[6]%
\>[6]{}\Varid{run_{StateT}}\;(\Varid{h_{State}}\;(\Conid{Op}\;(\Conid{Inr}\;\Varid{x})))\;(\Conid{SS}\;\Varid{xs}\;\Varid{stack}){}\<[E]%
\\
\>[3]{}\mathrel{=}\mbox{\commentbegin ~  definition of \ensuremath{\Varid{h_{State}}}   \commentend}{}\<[E]%
\\
\>[3]{}\hsindent{3}{}\<[6]%
\>[6]{}\Varid{run_{StateT}}\;(\Conid{StateT}\;(\lambda \Varid{s}\to \Conid{Op}\;(\Varid{fmap}\;(\lambda \Varid{y}\to \Varid{run_{StateT}}\;(\Varid{h_{State}}\;\Varid{y}\;\Varid{s}))\;\Varid{x})))\;(\Conid{SS}\;\Varid{xs}\;\Varid{stack}){}\<[E]%
\\
\>[3]{}\mathrel{=}\mbox{\commentbegin ~  definition of \ensuremath{\Varid{run_{StateT}}}   \commentend}{}\<[E]%
\\
\>[3]{}\hsindent{3}{}\<[6]%
\>[6]{}\Conid{Op}\;(\Varid{fmap}\;(\lambda \Varid{y}\to \Varid{run_{StateT}}\;(\Varid{h_{State}}\;\Varid{y}\;(\Conid{SS}\;\Varid{xs}\;\Varid{stack})))\;\Varid{x}){}\<[E]%
\\
\>[3]{}\mathrel{=}\mbox{\commentbegin ~  induction hypothesis  \commentend}{}\<[E]%
\\
\>[3]{}\hsindent{2}{}\<[5]%
\>[5]{}\Conid{Op}\;(\Varid{fmap}\;(\lambda \Varid{y}\to \mathbf{do}\;\{\mskip1.5mu \Varid{p'}\leftarrow \Varid{extract_{SS}}\;(\Varid{h_{State}}\;\Varid{y});\Varid{run_{StateT}}\;(\Varid{h_{State}}\;\Varid{pop_{SS}})\;(\Conid{SS}\;(\Varid{xs}+\!\!+\Varid{p'})\;\Varid{stack})\mskip1.5mu\})\;\Varid{x}){}\<[E]%
\\
\>[3]{}\mathrel{=}\mbox{\commentbegin ~  definition of \ensuremath{>\!\!>\!\!=}  \commentend}{}\<[E]%
\\
\>[3]{}\hsindent{2}{}\<[5]%
\>[5]{}\mathbf{do}\;\{\mskip1.5mu \Varid{p'}\leftarrow \Conid{Op}\;(\Varid{fmap}\;(\lambda \Varid{y}\to \Varid{extract_{SS}}\;(\Varid{h_{State}}\;\Varid{y}))\;\Varid{x});\Varid{run_{StateT}}\;(\Varid{h_{State}}\;\Varid{pop_{SS}})\;(\Conid{SS}\;(\Varid{xs}+\!\!+\Varid{p'})\;\Varid{stack})\mskip1.5mu\}{}\<[E]%
\\
\>[3]{}\mathrel{=}\mbox{\commentbegin ~  definition of \ensuremath{\Varid{fwd_{ND+f}}}  \commentend}{}\<[E]%
\\
\>[3]{}\hsindent{2}{}\<[5]%
\>[5]{}\mathbf{do}\;\{\mskip1.5mu \Varid{p'}\leftarrow \Varid{fwd_{ND+f}}\;(\Varid{fmap}\;(\lambda \Varid{y}\to \Varid{extract_{SS}}\;(\Varid{h_{State}}\;\Varid{y}))\;\Varid{x});\Varid{run_{StateT}}\;(\Varid{h_{State}}\;\Varid{pop_{SS}})\;(\Conid{SS}\;(\Varid{xs}+\!\!+\Varid{p'})\;\Varid{stack})\mskip1.5mu\}{}\<[E]%
\\
\>[3]{}\mathrel{=}\mbox{\commentbegin ~  extract-fwd  \commentend}{}\<[E]%
\\
\>[3]{}\mathrel{=}{}\<[6]%
\>[6]{}\mathbf{do}\;\{\mskip1.5mu \Varid{p'}\leftarrow \Varid{extract_{SS}}\;(\Varid{h_{State}}\;(\Varid{fwd}\;\Varid{x}));\Varid{run_{StateT}}\;(\Varid{h_{State}}\;\Varid{pop_{SS}})\;(\Conid{SS}\;(\Varid{xs}+\!\!+\Varid{p'})\;\Varid{stack})\mskip1.5mu\}{}\<[E]%
\ColumnHook
\end{hscode}\resethooks
\indentend %
\end{proof}

The following four lemmas characterise the behaviours of stack
operations.

\begin{lemma}[evaluation-append-ext]\label{eq:eval-append-f}~\indentbegin \begin{hscode}\SaveRestoreHook
\column{B}{@{}>{\hspre}l<{\hspost}@{}}%
\column{3}{@{}>{\hspre}l<{\hspost}@{}}%
\column{E}{@{}>{\hspre}l<{\hspost}@{}}%
\>[3]{}\Varid{run_{StateT}}\;(\Varid{h_{State}}\;(\Varid{append_{SS}}\;\Varid{x}\;\Varid{p}))\;(\Conid{SS}\;\Varid{xs}\;\Varid{stack})\mathrel{=}\Varid{run_{StateT}}\;(\Varid{h_{State}}\;\Varid{p})\;(\Conid{SS}\;(\Varid{xs}+\!\!+[\mskip1.5mu \Varid{x}\mskip1.5mu])\;\Varid{stack}){}\<[E]%
\ColumnHook
\end{hscode}\resethooks
\indentend \end{lemma}
\begin{proof}~\indentbegin \begin{hscode}\SaveRestoreHook
\column{B}{@{}>{\hspre}l<{\hspost}@{}}%
\column{3}{@{}>{\hspre}l<{\hspost}@{}}%
\column{6}{@{}>{\hspre}l<{\hspost}@{}}%
\column{E}{@{}>{\hspre}l<{\hspost}@{}}%
\>[6]{}\Varid{run_{StateT}}\;(\Varid{h_{State}}\;(\Varid{append_{SS}}\;\Varid{x}\;\Varid{p}))\;(\Conid{SS}\;\Varid{xs}\;\Varid{stack}){}\<[E]%
\\
\>[3]{}\mathrel{=}\mbox{\commentbegin ~  definition of \ensuremath{\Varid{append_{SS}}}   \commentend}{}\<[E]%
\ColumnHook
\end{hscode}\resethooks
\indentend %
\indentbegin \begin{hscode}\SaveRestoreHook
\column{B}{@{}>{\hspre}l<{\hspost}@{}}%
\column{3}{@{}>{\hspre}l<{\hspost}@{}}%
\column{6}{@{}>{\hspre}l<{\hspost}@{}}%
\column{8}{@{}>{\hspre}l<{\hspost}@{}}%
\column{E}{@{}>{\hspre}l<{\hspost}@{}}%
\>[6]{}\Varid{run_{StateT}}\;(\Varid{h_{State}}\;(\Varid{get}>\!\!>\!\!=\lambda (\Conid{SS}\;\Varid{xs}\;\Varid{stack})\to \Varid{put}\;(\Conid{SS}\;(\Varid{xs}+\!\!+[\mskip1.5mu \Varid{x}\mskip1.5mu])\;\Varid{stack})>\!\!>\Varid{p}))\;(\Conid{SS}\;\Varid{xs}\;\Varid{stack}){}\<[E]%
\\
\>[3]{}\mathrel{=}\mbox{\commentbegin ~  definition of \ensuremath{\Varid{get}}   \commentend}{}\<[E]%
\\
\>[3]{}\hsindent{3}{}\<[6]%
\>[6]{}\Varid{run_{StateT}}\;(\Varid{h_{State}}\;(\Conid{Op}\;(\Conid{Inl}\;(\Conid{Get}\;\Varid{\eta}))>\!\!>\!\!=\lambda (\Conid{SS}\;\Varid{xs}\;\Varid{stack})\to \Varid{put}\;(\Conid{SS}\;(\Varid{xs}+\!\!+[\mskip1.5mu \Varid{x}\mskip1.5mu])\;\Varid{stack})>\!\!>\Varid{p}))\;{}\<[E]%
\\
\>[6]{}\hsindent{2}{}\<[8]%
\>[8]{}(\Conid{SS}\;\Varid{xs}\;\Varid{stack}){}\<[E]%
\\
\>[3]{}\mathrel{=}\mbox{\commentbegin ~  definition of \ensuremath{(>\!\!>\!\!=)} for free monad and Law \ref{eq:monad-ret-bind}: return-bind   \commentend}{}\<[E]%
\\
\>[3]{}\hsindent{3}{}\<[6]%
\>[6]{}\Varid{run_{StateT}}\;(\Varid{h_{State}}\;(\Conid{Op}\;(\Conid{Inl}\;(\Conid{Get}\;(\lambda (\Conid{SS}\;\Varid{xs}\;\Varid{stack})\to \Varid{put}\;(\Conid{SS}\;(\Varid{xs}+\!\!+[\mskip1.5mu \Varid{x}\mskip1.5mu])\;\Varid{stack})>\!\!>\Varid{p})))))\;(\Conid{SS}\;\Varid{xs}\;\Varid{stack}){}\<[E]%
\\
\>[3]{}\mathrel{=}\mbox{\commentbegin ~  definition of \ensuremath{\Varid{h_{State}}}   \commentend}{}\<[E]%
\\
\>[3]{}\hsindent{3}{}\<[6]%
\>[6]{}\Varid{run_{StateT}}\;(\Conid{StateT}\;(\lambda \Varid{s}\to \Varid{run_{StateT}}\;(\Varid{h_{State}}\;((\lambda (\Conid{SS}\;\Varid{xs}\;\Varid{stack})\to \Varid{put}\;(\Conid{SS}\;(\Varid{xs}+\!\!+[\mskip1.5mu \Varid{x}\mskip1.5mu])\;\Varid{stack})>\!\!>\Varid{p})\;\Varid{s}))\;\Varid{s}))\;{}\<[E]%
\\
\>[6]{}\hsindent{2}{}\<[8]%
\>[8]{}(\Conid{SS}\;\Varid{xs}\;\Varid{stack}){}\<[E]%
\\
\>[3]{}\mathrel{=}\mbox{\commentbegin ~  definition of \ensuremath{\Varid{run_{StateT}}}   \commentend}{}\<[E]%
\\
\>[3]{}\hsindent{3}{}\<[6]%
\>[6]{}(\lambda \Varid{s}\to \Varid{run_{StateT}}\;(\Varid{h_{State}}\;((\lambda (\Conid{SS}\;\Varid{xs}\;\Varid{stack})\to \Varid{put}\;(\Conid{SS}\;(\Varid{xs}+\!\!+[\mskip1.5mu \Varid{x}\mskip1.5mu])\;\Varid{stack})>\!\!>\Varid{p})\;\Varid{s}))\;\Varid{s})\;(\Conid{SS}\;\Varid{xs}\;\Varid{stack}){}\<[E]%
\\
\>[3]{}\mathrel{=}\mbox{\commentbegin ~  function application   \commentend}{}\<[E]%
\\
\>[3]{}\hsindent{3}{}\<[6]%
\>[6]{}\Varid{run_{StateT}}\;(\Varid{h_{State}}\;((\lambda (\Conid{SS}\;\Varid{xs}\;\Varid{stack})\to \Varid{put}\;(\Conid{SS}\;(\Varid{xs}+\!\!+[\mskip1.5mu \Varid{x}\mskip1.5mu])\;\Varid{stack})>\!\!>\Varid{p})\;(\Conid{SS}\;\Varid{xs}\;\Varid{stack})))\;(\Conid{SS}\;\Varid{xs}\;\Varid{stack}){}\<[E]%
\\
\>[3]{}\mathrel{=}\mbox{\commentbegin ~  function application   \commentend}{}\<[E]%
\\
\>[3]{}\hsindent{3}{}\<[6]%
\>[6]{}\Varid{run_{StateT}}\;(\Varid{h_{State}}\;(\Varid{put}\;(\Conid{SS}\;(\Varid{xs}+\!\!+[\mskip1.5mu \Varid{x}\mskip1.5mu])\;\Varid{stack})>\!\!>\Varid{p}))\;(\Conid{SS}\;\Varid{xs}\;\Varid{stack}){}\<[E]%
\\
\>[3]{}\mathrel{=}\mbox{\commentbegin ~  definition of \ensuremath{\Varid{put}}   \commentend}{}\<[E]%
\\
\>[3]{}\hsindent{3}{}\<[6]%
\>[6]{}\Varid{run_{StateT}}\;(\Varid{h_{State}}\;(\Conid{Op}\;(\Conid{Inl}\;(\Conid{Put}\;(\Conid{SS}\;(\Varid{xs}+\!\!+[\mskip1.5mu \Varid{x}\mskip1.5mu])\;\Varid{stack})\;(\Varid{\eta}\;())))>\!\!>\Varid{p}))\;(\Conid{SS}\;\Varid{xs}\;\Varid{stack}){}\<[E]%
\\
\>[3]{}\mathrel{=}\mbox{\commentbegin ~  definition of \ensuremath{(>\!\!>)} for free monad and Law \ref{eq:monad-ret-bind}: return-bind   \commentend}{}\<[E]%
\\
\>[3]{}\hsindent{3}{}\<[6]%
\>[6]{}\Varid{run_{StateT}}\;(\Varid{h_{State}}\;(\Conid{Op}\;(\Conid{Inl}\;(\Conid{Put}\;(\Conid{SS}\;(\Varid{xs}+\!\!+[\mskip1.5mu \Varid{x}\mskip1.5mu])\;\Varid{stack})\;\Varid{p}))))\;(\Conid{SS}\;\Varid{xs}\;\Varid{stack}){}\<[E]%
\\
\>[3]{}\mathrel{=}\mbox{\commentbegin ~  definition of \ensuremath{\Varid{h_{State}}}   \commentend}{}\<[E]%
\\
\>[3]{}\hsindent{3}{}\<[6]%
\>[6]{}\Varid{run_{StateT}}\;(\Conid{StateT}\;(\lambda \Varid{s}\to \Varid{run_{StateT}}\;(\Varid{h_{State}}\;\Varid{p})\;(\Conid{SS}\;(\Varid{xs}+\!\!+[\mskip1.5mu \Varid{x}\mskip1.5mu])\;\Varid{stack})))\;(\Conid{SS}\;\Varid{xs}\;\Varid{stack}){}\<[E]%
\\
\>[3]{}\mathrel{=}\mbox{\commentbegin ~  definition of \ensuremath{\Varid{run_{StateT}}}   \commentend}{}\<[E]%
\\
\>[3]{}\hsindent{3}{}\<[6]%
\>[6]{}(\lambda \Varid{s}\to \Varid{run_{StateT}}\;(\Varid{h_{State}}\;\Varid{p})\;(\Conid{SS}\;(\Varid{xs}+\!\!+[\mskip1.5mu \Varid{x}\mskip1.5mu])\;\Varid{stack}))\;(\Conid{SS}\;\Varid{xs}\;\Varid{stack}){}\<[E]%
\\
\>[3]{}\mathrel{=}\mbox{\commentbegin ~  function application   \commentend}{}\<[E]%
\\
\>[3]{}\hsindent{3}{}\<[6]%
\>[6]{}\Varid{run_{StateT}}\;(\Varid{h_{State}}\;\Varid{p})\;(\Conid{SS}\;(\Varid{xs}+\!\!+[\mskip1.5mu \Varid{x}\mskip1.5mu])\;\Varid{stack}){}\<[E]%
\ColumnHook
\end{hscode}\resethooks
\indentend \end{proof}

\begin{lemma}[evaluation-pop1-ext]\label{eq:eval-pop1-f}~\indentbegin \begin{hscode}\SaveRestoreHook
\column{B}{@{}>{\hspre}l<{\hspost}@{}}%
\column{3}{@{}>{\hspre}l<{\hspost}@{}}%
\column{E}{@{}>{\hspre}l<{\hspost}@{}}%
\>[3]{}\Varid{run_{StateT}}\;(\Varid{h_{State}}\;\Varid{pop_{SS}})\;(\Conid{SS}\;\Varid{xs}\;[\mskip1.5mu \mskip1.5mu])\mathrel{=}\Varid{\eta}\;((),\Conid{SS}\;\Varid{xs}\;[\mskip1.5mu \mskip1.5mu]){}\<[E]%
\ColumnHook
\end{hscode}\resethooks
\indentend \end{lemma}
\begin{proof}
\indentbegin \begin{hscode}\SaveRestoreHook
\column{B}{@{}>{\hspre}l<{\hspost}@{}}%
\column{3}{@{}>{\hspre}l<{\hspost}@{}}%
\column{6}{@{}>{\hspre}l<{\hspost}@{}}%
\column{8}{@{}>{\hspre}l<{\hspost}@{}}%
\column{23}{@{}>{\hspre}l<{\hspost}@{}}%
\column{32}{@{}>{\hspre}l<{\hspost}@{}}%
\column{E}{@{}>{\hspre}l<{\hspost}@{}}%
\>[6]{}\Varid{run_{StateT}}\;(\Varid{h_{State}}\;\Varid{pop_{SS}})\;(\Conid{SS}\;\Varid{xs}\;[\mskip1.5mu \mskip1.5mu]){}\<[E]%
\\
\>[3]{}\mathrel{=}\mbox{\commentbegin ~  definition of \ensuremath{\Varid{pop_{SS}}}   \commentend}{}\<[E]%
\\
\>[3]{}\hsindent{3}{}\<[6]%
\>[6]{}\Varid{run_{StateT}}\;(\Varid{h_{State}}\;(\Varid{get}>\!\!>\!\!=\lambda (\Conid{SS}\;\Varid{xs}\;\Varid{stack})\to {}\<[E]%
\\
\>[6]{}\hsindent{2}{}\<[8]%
\>[8]{}\mathbf{case}\;\Varid{stack}\;\mathbf{of}\;{}\<[23]%
\>[23]{}[\mskip1.5mu \mskip1.5mu]{}\<[32]%
\>[32]{}\to \Varid{\eta}\;(){}\<[E]%
\\
\>[23]{}\Varid{op}\mathbin{:}\Varid{ps}{}\<[32]%
\>[32]{}\to \mathbf{do}\;\Varid{put}\;(\Conid{SS}\;\Varid{xs}\;\Varid{ps});\Varid{op}))\;(\Conid{SS}\;\Varid{xs}\;[\mskip1.5mu \mskip1.5mu]){}\<[E]%
\\
\>[3]{}\mathrel{=}\mbox{\commentbegin ~  definition of \ensuremath{\Varid{get}}   \commentend}{}\<[E]%
\\
\>[3]{}\hsindent{3}{}\<[6]%
\>[6]{}\Varid{run_{StateT}}\;(\Varid{h_{State}}\;(\Conid{Op}\;(\Conid{Inl}\;(\Conid{Get}\;\Varid{\eta}))>\!\!>\!\!=\lambda (\Conid{SS}\;\Varid{xs}\;\Varid{stack})\to {}\<[E]%
\\
\>[6]{}\hsindent{2}{}\<[8]%
\>[8]{}\mathbf{case}\;\Varid{stack}\;\mathbf{of}\;{}\<[23]%
\>[23]{}[\mskip1.5mu \mskip1.5mu]{}\<[32]%
\>[32]{}\to \Varid{\eta}\;(){}\<[E]%
\\
\>[23]{}\Varid{op}\mathbin{:}\Varid{ps}{}\<[32]%
\>[32]{}\to \mathbf{do}\;\Varid{put}\;(\Conid{SS}\;\Varid{xs}\;\Varid{ps});\Varid{op}))\;(\Conid{SS}\;\Varid{xs}\;[\mskip1.5mu \mskip1.5mu]){}\<[E]%
\\
\>[3]{}\mathrel{=}\mbox{\commentbegin ~  definition of \ensuremath{(>\!\!>\!\!=)} for free monad and Law \ref{eq:monad-ret-bind}: return-bind   \commentend}{}\<[E]%
\\
\>[3]{}\hsindent{3}{}\<[6]%
\>[6]{}\Varid{run_{StateT}}\;(\Varid{h_{State}}\;(\Conid{Op}\;(\Conid{Inl}\;(\Conid{Get}\;(\lambda (\Conid{SS}\;\Varid{xs}\;\Varid{stack})\to {}\<[E]%
\\
\>[6]{}\hsindent{2}{}\<[8]%
\>[8]{}\mathbf{case}\;\Varid{stack}\;\mathbf{of}\;{}\<[23]%
\>[23]{}[\mskip1.5mu \mskip1.5mu]{}\<[32]%
\>[32]{}\to \Varid{\eta}\;(){}\<[E]%
\\
\>[23]{}\Varid{op}\mathbin{:}\Varid{ps}{}\<[32]%
\>[32]{}\to \mathbf{do}\;\Varid{put}\;(\Conid{SS}\;\Varid{xs}\;\Varid{ps});\Varid{op})))))\;(\Conid{SS}\;\Varid{xs}\;[\mskip1.5mu \mskip1.5mu]){}\<[E]%
\\
\>[3]{}\mathrel{=}\mbox{\commentbegin ~  definition of \ensuremath{\Varid{h_{State}}}   \commentend}{}\<[E]%
\\
\>[3]{}\hsindent{3}{}\<[6]%
\>[6]{}\Varid{run_{StateT}}\;(\Conid{StateT}\;(\lambda \Varid{s}\to \Varid{run_{StateT}}\;(\Varid{h_{State}}\;((\lambda (\Conid{SS}\;\Varid{xs}\;\Varid{stack})\to {}\<[E]%
\\
\>[6]{}\hsindent{2}{}\<[8]%
\>[8]{}\mathbf{case}\;\Varid{stack}\;\mathbf{of}\;{}\<[23]%
\>[23]{}[\mskip1.5mu \mskip1.5mu]{}\<[32]%
\>[32]{}\to \Varid{\eta}\;(){}\<[E]%
\\
\>[23]{}\Varid{op}\mathbin{:}\Varid{ps}{}\<[32]%
\>[32]{}\to \mathbf{do}\;\Varid{put}\;(\Conid{SS}\;\Varid{xs}\;\Varid{ps});\Varid{op})\;\Varid{s}))\;\Varid{s}))\;(\Conid{SS}\;\Varid{xs}\;[\mskip1.5mu \mskip1.5mu]){}\<[E]%
\\
\>[3]{}\mathrel{=}\mbox{\commentbegin ~  definition of \ensuremath{\Varid{run_{StateT}}}   \commentend}{}\<[E]%
\\
\>[3]{}\hsindent{3}{}\<[6]%
\>[6]{}(\lambda \Varid{s}\to \Varid{run_{StateT}}\;(\Varid{h_{State}}\;((\lambda (\Conid{SS}\;\Varid{xs}\;\Varid{stack})\to {}\<[E]%
\\
\>[6]{}\hsindent{2}{}\<[8]%
\>[8]{}\mathbf{case}\;\Varid{stack}\;\mathbf{of}\;{}\<[23]%
\>[23]{}[\mskip1.5mu \mskip1.5mu]{}\<[32]%
\>[32]{}\to \Varid{\eta}\;(){}\<[E]%
\\
\>[23]{}\Varid{op}\mathbin{:}\Varid{ps}{}\<[32]%
\>[32]{}\to \mathbf{do}\;\Varid{put}\;(\Conid{SS}\;\Varid{xs}\;\Varid{ps});\Varid{op})\;\Varid{s}))\;\Varid{s})\;(\Conid{SS}\;\Varid{xs}\;[\mskip1.5mu \mskip1.5mu]){}\<[E]%
\\
\>[3]{}\mathrel{=}\mbox{\commentbegin ~  function application   \commentend}{}\<[E]%
\\
\>[3]{}\hsindent{3}{}\<[6]%
\>[6]{}\Varid{run_{StateT}}\;(\Varid{h_{State}}\;((\lambda (\Conid{SS}\;\Varid{xs}\;\Varid{stack})\to {}\<[E]%
\\
\>[6]{}\hsindent{2}{}\<[8]%
\>[8]{}\mathbf{case}\;\Varid{stack}\;\mathbf{of}\;{}\<[23]%
\>[23]{}[\mskip1.5mu \mskip1.5mu]{}\<[32]%
\>[32]{}\to \Varid{\eta}\;(){}\<[E]%
\\
\>[23]{}\Varid{op}\mathbin{:}\Varid{ps}{}\<[32]%
\>[32]{}\to \mathbf{do}\;\Varid{put}\;(\Conid{SS}\;\Varid{xs}\;\Varid{ps});\Varid{op})\;(\Conid{SS}\;\Varid{xs}\;[\mskip1.5mu \mskip1.5mu])))\;(\Conid{SS}\;\Varid{xs}\;[\mskip1.5mu \mskip1.5mu]){}\<[E]%
\\
\>[3]{}\mathrel{=}\mbox{\commentbegin ~  function application, \ensuremath{\mathbf{case}}-analysis   \commentend}{}\<[E]%
\\
\>[3]{}\hsindent{3}{}\<[6]%
\>[6]{}\Varid{run_{StateT}}\;(\Varid{h_{State}}\;(\Varid{\eta}\;()))\;(\Conid{SS}\;\Varid{xs}\;[\mskip1.5mu \mskip1.5mu]){}\<[E]%
\\
\>[3]{}\mathrel{=}\mbox{\commentbegin ~  definition of \ensuremath{\Varid{h_{State}}}   \commentend}{}\<[E]%
\\
\>[3]{}\hsindent{3}{}\<[6]%
\>[6]{}\Varid{run_{StateT}}\;(\Conid{StateT}\;(\lambda \Varid{s}\to \Varid{\eta}\;((),\Varid{s})))\;(\Conid{SS}\;\Varid{xs}\;[\mskip1.5mu \mskip1.5mu]){}\<[E]%
\\
\>[3]{}\mathrel{=}\mbox{\commentbegin ~  definition of \ensuremath{\Varid{run_{StateT}}}   \commentend}{}\<[E]%
\\
\>[3]{}\hsindent{3}{}\<[6]%
\>[6]{}(\lambda \Varid{s}\to \Varid{\eta}\;((),\Varid{s}))\;(\Conid{SS}\;\Varid{xs}\;[\mskip1.5mu \mskip1.5mu]){}\<[E]%
\\
\>[3]{}\mathrel{=}\mbox{\commentbegin ~  function application   \commentend}{}\<[E]%
\\
\>[3]{}\hsindent{3}{}\<[6]%
\>[6]{}((),\Conid{SS}\;\Varid{xs}\;[\mskip1.5mu \mskip1.5mu]){}\<[E]%
\ColumnHook
\end{hscode}\resethooks
\indentend \end{proof}

\begin{lemma}[evaluation-pop2-ext]\label{eq:eval-pop2-f}~\indentbegin \begin{hscode}\SaveRestoreHook
\column{B}{@{}>{\hspre}l<{\hspost}@{}}%
\column{3}{@{}>{\hspre}l<{\hspost}@{}}%
\column{E}{@{}>{\hspre}l<{\hspost}@{}}%
\>[3]{}\Varid{run_{StateT}}\;(\Varid{h_{State}}\;\Varid{pop_{SS}})\;(\Conid{SS}\;\Varid{xs}\;(\Varid{q}\mathbin{:}\Varid{stack}))\mathrel{=}\Varid{run_{StateT}}\;(\Varid{h_{State}}\;\Varid{q})\;(\Conid{SS}\;\Varid{xs}\;\Varid{stack}){}\<[E]%
\ColumnHook
\end{hscode}\resethooks
\indentend \end{lemma}
\begin{proof}~\indentbegin \begin{hscode}\SaveRestoreHook
\column{B}{@{}>{\hspre}l<{\hspost}@{}}%
\column{3}{@{}>{\hspre}l<{\hspost}@{}}%
\column{6}{@{}>{\hspre}l<{\hspost}@{}}%
\column{8}{@{}>{\hspre}l<{\hspost}@{}}%
\column{23}{@{}>{\hspre}l<{\hspost}@{}}%
\column{32}{@{}>{\hspre}l<{\hspost}@{}}%
\column{E}{@{}>{\hspre}l<{\hspost}@{}}%
\>[6]{}\Varid{run_{StateT}}\;(\Varid{h_{State}}\;\Varid{pop_{SS}})\;(\Conid{SS}\;\Varid{xs}\;(\Varid{q}\mathbin{:}\Varid{stack})){}\<[E]%
\\
\>[3]{}\mathrel{=}\mbox{\commentbegin ~  definition of \ensuremath{\Varid{pop_{SS}}}   \commentend}{}\<[E]%
\\
\>[3]{}\hsindent{3}{}\<[6]%
\>[6]{}\Varid{run_{StateT}}\;(\Varid{h_{State}}\;(\Varid{get}>\!\!>\!\!=\lambda (\Conid{SS}\;\Varid{xs}\;\Varid{stack})\to {}\<[E]%
\\
\>[6]{}\hsindent{2}{}\<[8]%
\>[8]{}\mathbf{case}\;\Varid{stack}\;\mathbf{of}\;{}\<[23]%
\>[23]{}[\mskip1.5mu \mskip1.5mu]{}\<[32]%
\>[32]{}\to \Varid{\eta}\;(){}\<[E]%
\\
\>[23]{}\Varid{op}\mathbin{:}\Varid{ps}{}\<[32]%
\>[32]{}\to \mathbf{do}\;\Varid{put}\;(\Conid{SS}\;\Varid{xs}\;\Varid{ps});\Varid{op}))\;(\Conid{SS}\;\Varid{xs}\;(\Varid{q}\mathbin{:}\Varid{stack})){}\<[E]%
\\
\>[3]{}\mathrel{=}\mbox{\commentbegin ~  definition of \ensuremath{\Varid{get}}   \commentend}{}\<[E]%
\\
\>[3]{}\hsindent{3}{}\<[6]%
\>[6]{}\Varid{run_{StateT}}\;(\Varid{h_{State}}\;(\Conid{Op}\;(\Conid{Inl}\;(\Conid{Get}\;\Varid{\eta}))>\!\!>\!\!=\lambda (\Conid{SS}\;\Varid{xs}\;\Varid{stack})\to {}\<[E]%
\\
\>[6]{}\hsindent{2}{}\<[8]%
\>[8]{}\mathbf{case}\;\Varid{stack}\;\mathbf{of}\;{}\<[23]%
\>[23]{}[\mskip1.5mu \mskip1.5mu]{}\<[32]%
\>[32]{}\to \Varid{\eta}\;(){}\<[E]%
\\
\>[23]{}\Varid{op}\mathbin{:}\Varid{ps}{}\<[32]%
\>[32]{}\to \mathbf{do}\;\Varid{put}\;(\Conid{SS}\;\Varid{xs}\;\Varid{ps});\Varid{op}))\;(\Conid{SS}\;\Varid{xs}\;(\Varid{q}\mathbin{:}\Varid{stack})){}\<[E]%
\\
\>[3]{}\mathrel{=}\mbox{\commentbegin ~  definition of \ensuremath{(>\!\!>\!\!=)} for free monad and Law \ref{eq:monad-ret-bind}: return-bind   \commentend}{}\<[E]%
\\
\>[3]{}\hsindent{3}{}\<[6]%
\>[6]{}\Varid{run_{StateT}}\;(\Varid{h_{State}}\;(\Conid{Op}\;(\Conid{Inl}\;(\Conid{Get}\;(\lambda (\Conid{SS}\;\Varid{xs}\;\Varid{stack})\to {}\<[E]%
\\
\>[6]{}\hsindent{2}{}\<[8]%
\>[8]{}\mathbf{case}\;\Varid{stack}\;\mathbf{of}\;{}\<[23]%
\>[23]{}[\mskip1.5mu \mskip1.5mu]{}\<[32]%
\>[32]{}\to \Varid{\eta}\;(){}\<[E]%
\\
\>[23]{}\Varid{op}\mathbin{:}\Varid{ps}{}\<[32]%
\>[32]{}\to \mathbf{do}\;\Varid{put}\;(\Conid{SS}\;\Varid{xs}\;\Varid{ps});\Varid{op})))))\;(\Conid{SS}\;\Varid{xs}\;(\Varid{q}\mathbin{:}\Varid{stack})){}\<[E]%
\\
\>[3]{}\mathrel{=}\mbox{\commentbegin ~  definition of \ensuremath{\Varid{h_{State}}}   \commentend}{}\<[E]%
\\
\>[3]{}\hsindent{3}{}\<[6]%
\>[6]{}\Varid{run_{StateT}}\;(\Conid{StateT}\;(\lambda \Varid{s}\to \Varid{run_{StateT}}\;(\Varid{h_{State}}\;((\lambda (\Conid{SS}\;\Varid{xs}\;\Varid{stack})\to {}\<[E]%
\\
\>[6]{}\hsindent{2}{}\<[8]%
\>[8]{}\mathbf{case}\;\Varid{stack}\;\mathbf{of}\;{}\<[23]%
\>[23]{}[\mskip1.5mu \mskip1.5mu]{}\<[32]%
\>[32]{}\to \Varid{\eta}\;(){}\<[E]%
\\
\>[23]{}\Varid{op}\mathbin{:}\Varid{ps}{}\<[32]%
\>[32]{}\to \mathbf{do}\;\Varid{put}\;(\Conid{SS}\;\Varid{xs}\;\Varid{ps});\Varid{op})\;\Varid{s}))\;\Varid{s}))\;(\Conid{SS}\;\Varid{xs}\;(\Varid{q}\mathbin{:}\Varid{stack})){}\<[E]%
\\
\>[3]{}\mathrel{=}\mbox{\commentbegin ~  definition of \ensuremath{\Varid{run_{StateT}}}   \commentend}{}\<[E]%
\\
\>[3]{}\hsindent{3}{}\<[6]%
\>[6]{}(\lambda \Varid{s}\to \Varid{run_{StateT}}\;(\Varid{h_{State}}\;((\lambda (\Conid{SS}\;\Varid{xs}\;\Varid{stack})\to {}\<[E]%
\\
\>[6]{}\hsindent{2}{}\<[8]%
\>[8]{}\mathbf{case}\;\Varid{stack}\;\mathbf{of}\;{}\<[23]%
\>[23]{}[\mskip1.5mu \mskip1.5mu]{}\<[32]%
\>[32]{}\to \Varid{\eta}\;(){}\<[E]%
\\
\>[23]{}\Varid{op}\mathbin{:}\Varid{ps}{}\<[32]%
\>[32]{}\to \mathbf{do}\;\Varid{put}\;(\Conid{SS}\;\Varid{xs}\;\Varid{ps});\Varid{op})\;\Varid{s}))\;\Varid{s})\;(\Conid{SS}\;\Varid{xs}\;(\Varid{q}\mathbin{:}\Varid{stack})){}\<[E]%
\\
\>[3]{}\mathrel{=}\mbox{\commentbegin ~  function application   \commentend}{}\<[E]%
\\
\>[3]{}\hsindent{3}{}\<[6]%
\>[6]{}\Varid{run_{StateT}}\;(\Varid{h_{State}}\;((\lambda (\Conid{SS}\;\Varid{xs}\;\Varid{stack})\to {}\<[E]%
\\
\>[6]{}\hsindent{2}{}\<[8]%
\>[8]{}\mathbf{case}\;\Varid{stack}\;\mathbf{of}\;{}\<[23]%
\>[23]{}[\mskip1.5mu \mskip1.5mu]{}\<[32]%
\>[32]{}\to \Varid{\eta}\;(){}\<[E]%
\\
\>[23]{}\Varid{op}\mathbin{:}\Varid{ps}{}\<[32]%
\>[32]{}\to \mathbf{do}\;\Varid{put}\;(\Conid{SS}\;\Varid{xs}\;\Varid{ps});\Varid{op})\;(\Conid{SS}\;\Varid{xs}\;(\Varid{q}\mathbin{:}\Varid{stack}))))\;(\Conid{SS}\;\Varid{xs}\;(\Varid{q}\mathbin{:}\Varid{stack})){}\<[E]%
\\
\>[3]{}\mathrel{=}\mbox{\commentbegin ~  function application, \ensuremath{\mathbf{case}}-analysis   \commentend}{}\<[E]%
\ColumnHook
\end{hscode}\resethooks
\indentend %
\indentbegin \begin{hscode}\SaveRestoreHook
\column{B}{@{}>{\hspre}l<{\hspost}@{}}%
\column{3}{@{}>{\hspre}l<{\hspost}@{}}%
\column{6}{@{}>{\hspre}l<{\hspost}@{}}%
\column{E}{@{}>{\hspre}l<{\hspost}@{}}%
\>[6]{}\Varid{run_{StateT}}\;(\Varid{h_{State}}\;(\Varid{put}\;(\Conid{SS}\;\Varid{xs}\;\Varid{stack})>\!\!>\Varid{q}))\;(\Conid{SS}\;\Varid{xs}\;(\Varid{q}\mathbin{:}\Varid{stack})){}\<[E]%
\\
\>[3]{}\mathrel{=}\mbox{\commentbegin ~  definition of \ensuremath{\Varid{put}}   \commentend}{}\<[E]%
\\
\>[3]{}\hsindent{3}{}\<[6]%
\>[6]{}\Varid{run_{StateT}}\;(\Varid{h_{State}}\;(\Conid{Op}\;(\Conid{Inl}\;(\Conid{Put}\;(\Conid{SS}\;\Varid{xs}\;\Varid{stack})\;(\Varid{\eta}\;())))>\!\!>\Varid{q}))\;(\Conid{SS}\;\Varid{xs}\;(\Varid{q}\mathbin{:}\Varid{stack})){}\<[E]%
\\
\>[3]{}\mathrel{=}\mbox{\commentbegin ~  definition of \ensuremath{(>\!\!>)} for free monad and Law \ref{eq:monad-ret-bind}: return-bind   \commentend}{}\<[E]%
\\
\>[3]{}\hsindent{3}{}\<[6]%
\>[6]{}\Varid{run_{StateT}}\;(\Varid{h_{State}}\;(\Conid{Op}\;(\Conid{Inl}\;(\Conid{Put}\;(\Conid{SS}\;\Varid{xs}\;\Varid{stack})\;\Varid{q}))))\;(\Conid{SS}\;\Varid{xs}\;(\Varid{q}\mathbin{:}\Varid{stack})){}\<[E]%
\\
\>[3]{}\mathrel{=}\mbox{\commentbegin ~  definition of \ensuremath{\Varid{h_{State}}}   \commentend}{}\<[E]%
\\
\>[3]{}\hsindent{3}{}\<[6]%
\>[6]{}\Varid{run_{StateT}}\;(\Conid{StateT}\;(\lambda \Varid{s}\to \Varid{run_{StateT}}\;(\Varid{h_{State}}\;\Varid{q})\;(\Conid{SS}\;\Varid{xs}\;\Varid{stack})))\;(\Conid{SS}\;\Varid{xs}\;(\Varid{q}\mathbin{:}\Varid{stack})){}\<[E]%
\\
\>[3]{}\mathrel{=}\mbox{\commentbegin ~  definition of \ensuremath{\Varid{run_{StateT}}}   \commentend}{}\<[E]%
\\
\>[3]{}\hsindent{3}{}\<[6]%
\>[6]{}(\lambda \Varid{s}\to \Varid{run_{StateT}}\;(\Varid{h_{State}}\;\Varid{q})\;(\Conid{SS}\;\Varid{xs}\;\Varid{stack}))\;(\Conid{SS}\;\Varid{xs}\;(\Varid{q}\mathbin{:}\Varid{stack})){}\<[E]%
\\
\>[3]{}\mathrel{=}\mbox{\commentbegin ~  function application   \commentend}{}\<[E]%
\\
\>[3]{}\hsindent{3}{}\<[6]%
\>[6]{}\Varid{run_{StateT}}\;(\Varid{h_{State}}\;\Varid{q})\;(\Conid{SS}\;\Varid{xs}\;\Varid{stack}){}\<[E]%
\ColumnHook
\end{hscode}\resethooks
\indentend \end{proof}

\begin{lemma}[evaluation-push-ext]\label{eq:eval-push-f}~\indentbegin \begin{hscode}\SaveRestoreHook
\column{B}{@{}>{\hspre}l<{\hspost}@{}}%
\column{3}{@{}>{\hspre}l<{\hspost}@{}}%
\column{E}{@{}>{\hspre}l<{\hspost}@{}}%
\>[3]{}\Varid{run_{StateT}}\;(\Varid{h_{State}}\;(\Varid{push_{SS}}\;\Varid{q}\;\Varid{p}))\;(\Conid{SS}\;\Varid{xs}\;\Varid{stack})\mathrel{=}\Varid{run_{StateT}}\;(\Varid{h_{State}}\;\Varid{p})\;(\Conid{SS}\;\Varid{xs}\;(\Varid{q}\mathbin{:}\Varid{stack})){}\<[E]%
\ColumnHook
\end{hscode}\resethooks
\indentend \end{lemma}
\begin{proof}~\indentbegin \begin{hscode}\SaveRestoreHook
\column{B}{@{}>{\hspre}l<{\hspost}@{}}%
\column{3}{@{}>{\hspre}l<{\hspost}@{}}%
\column{6}{@{}>{\hspre}l<{\hspost}@{}}%
\column{E}{@{}>{\hspre}l<{\hspost}@{}}%
\>[6]{}\Varid{run_{StateT}}\;(\Varid{h_{State}}\;(\Varid{push_{SS}}\;\Varid{q}\;\Varid{p}))\;(\Conid{SS}\;\Varid{xs}\;\Varid{stack}){}\<[E]%
\\
\>[3]{}\mathrel{=}\mbox{\commentbegin ~  definition of \ensuremath{\Varid{push_{SS}}}   \commentend}{}\<[E]%
\ColumnHook
\end{hscode}\resethooks
\indentend %
\indentbegin \begin{hscode}\SaveRestoreHook
\column{B}{@{}>{\hspre}l<{\hspost}@{}}%
\column{3}{@{}>{\hspre}l<{\hspost}@{}}%
\column{6}{@{}>{\hspre}l<{\hspost}@{}}%
\column{8}{@{}>{\hspre}l<{\hspost}@{}}%
\column{E}{@{}>{\hspre}l<{\hspost}@{}}%
\>[6]{}\Varid{run_{StateT}}\;(\Varid{h_{State}}\;(\Varid{get}>\!\!>\!\!=\lambda (\Conid{SS}\;\Varid{xs}\;\Varid{stack})\to \Varid{put}\;(\Conid{SS}\;\Varid{xs}\;(\Varid{q}\mathbin{:}\Varid{stack}))>\!\!>\Varid{p}))\;(\Conid{SS}\;\Varid{xs}\;\Varid{stack}){}\<[E]%
\\
\>[3]{}\mathrel{=}\mbox{\commentbegin ~  definition of \ensuremath{\Varid{get}}   \commentend}{}\<[E]%
\\
\>[3]{}\hsindent{3}{}\<[6]%
\>[6]{}\Varid{run_{StateT}}\;(\Varid{h_{State}}\;(\Conid{Op}\;(\Conid{Inl}\;(\Conid{Get}\;\Varid{\eta}))>\!\!>\!\!=\lambda (\Conid{SS}\;\Varid{xs}\;\Varid{stack})\to \Varid{put}\;(\Conid{SS}\;\Varid{xs}\;(\Varid{q}\mathbin{:}\Varid{stack}))>\!\!>\Varid{p}))\;(\Conid{SS}\;\Varid{xs}\;\Varid{stack}){}\<[E]%
\\
\>[3]{}\mathrel{=}\mbox{\commentbegin ~  definition of \ensuremath{(>\!\!>\!\!=)} for free monad and Law \ref{eq:monad-ret-bind}: return-bind   \commentend}{}\<[E]%
\\
\>[3]{}\hsindent{3}{}\<[6]%
\>[6]{}\Varid{run_{StateT}}\;(\Varid{h_{State}}\;(\Conid{Op}\;(\Conid{Inl}\;(\Conid{Get}\;(\lambda (\Conid{SS}\;\Varid{xs}\;\Varid{stack})\to \Varid{put}\;(\Conid{SS}\;\Varid{xs}\;(\Varid{q}\mathbin{:}\Varid{stack}))>\!\!>\Varid{p})))))\;(\Conid{SS}\;\Varid{xs}\;\Varid{stack}){}\<[E]%
\\
\>[3]{}\mathrel{=}\mbox{\commentbegin ~  definition of \ensuremath{\Varid{h_{State}}}   \commentend}{}\<[E]%
\\
\>[3]{}\hsindent{3}{}\<[6]%
\>[6]{}\Varid{run_{StateT}}\;(\Conid{StateT}\;(\lambda \Varid{s}\to \Varid{run_{StateT}}\;(\Varid{h_{State}}\;((\lambda (\Conid{SS}\;\Varid{xs}\;\Varid{stack})\to \Varid{put}\;(\Conid{SS}\;\Varid{xs}\;(\Varid{q}\mathbin{:}\Varid{stack}))>\!\!>\Varid{p})\;\Varid{s}))\;\Varid{s}))\;{}\<[E]%
\\
\>[6]{}\hsindent{2}{}\<[8]%
\>[8]{}(\Conid{SS}\;\Varid{xs}\;\Varid{stack}){}\<[E]%
\\
\>[3]{}\mathrel{=}\mbox{\commentbegin ~  definition of \ensuremath{\Varid{run_{StateT}}}   \commentend}{}\<[E]%
\\
\>[3]{}\hsindent{3}{}\<[6]%
\>[6]{}(\lambda \Varid{s}\to \Varid{run_{StateT}}\;(\Varid{h_{State}}\;((\lambda (\Conid{SS}\;\Varid{xs}\;\Varid{stack})\to \Varid{put}\;(\Conid{SS}\;\Varid{xs}\;(\Varid{q}\mathbin{:}\Varid{stack}))>\!\!>\Varid{p})\;\Varid{s}))\;\Varid{s})\;(\Conid{SS}\;\Varid{xs}\;\Varid{stack}){}\<[E]%
\\
\>[3]{}\mathrel{=}\mbox{\commentbegin ~  function application   \commentend}{}\<[E]%
\\
\>[3]{}\hsindent{3}{}\<[6]%
\>[6]{}\Varid{run_{StateT}}\;(\Varid{h_{State}}\;((\lambda (\Conid{SS}\;\Varid{xs}\;\Varid{stack})\to \Varid{put}\;(\Conid{SS}\;\Varid{xs}\;(\Varid{q}\mathbin{:}\Varid{stack}))>\!\!>\Varid{p})\;(\Conid{SS}\;\Varid{xs}\;\Varid{stack})))\;(\Conid{SS}\;\Varid{xs}\;\Varid{stack}){}\<[E]%
\\
\>[3]{}\mathrel{=}\mbox{\commentbegin ~  function application   \commentend}{}\<[E]%
\\
\>[3]{}\hsindent{3}{}\<[6]%
\>[6]{}\Varid{run_{StateT}}\;(\Varid{h_{State}}\;(\Varid{put}\;(\Conid{SS}\;\Varid{xs}\;(\Varid{q}\mathbin{:}\Varid{stack}))>\!\!>\Varid{p}))\;(\Conid{SS}\;\Varid{xs}\;\Varid{stack}){}\<[E]%
\\
\>[3]{}\mathrel{=}\mbox{\commentbegin ~  definition of \ensuremath{\Varid{put}}   \commentend}{}\<[E]%
\\
\>[3]{}\hsindent{3}{}\<[6]%
\>[6]{}\Varid{run_{StateT}}\;(\Varid{h_{State}}\;(\Conid{Op}\;(\Conid{Inl}\;(\Conid{Put}\;(\Conid{SS}\;\Varid{xs}\;(\Varid{q}\mathbin{:}\Varid{stack}))\;(\Varid{\eta}\;())))>\!\!>\Varid{p}))\;(\Conid{SS}\;\Varid{xs}\;\Varid{stack}){}\<[E]%
\\
\>[3]{}\mathrel{=}\mbox{\commentbegin ~  definition of \ensuremath{(>\!\!>)} for free monad and Law \ref{eq:monad-ret-bind}: return-bind   \commentend}{}\<[E]%
\\
\>[3]{}\hsindent{3}{}\<[6]%
\>[6]{}\Varid{run_{StateT}}\;(\Varid{h_{State}}\;(\Conid{Op}\;(\Conid{Inl}\;(\Conid{Put}\;(\Conid{SS}\;\Varid{xs}\;(\Varid{q}\mathbin{:}\Varid{stack}))\;\Varid{p}))))\;(\Conid{SS}\;\Varid{xs}\;\Varid{stack}){}\<[E]%
\\
\>[3]{}\mathrel{=}\mbox{\commentbegin ~  definition of \ensuremath{\Varid{h_{State}}}   \commentend}{}\<[E]%
\\
\>[3]{}\hsindent{3}{}\<[6]%
\>[6]{}\Varid{run_{StateT}}\;(\Conid{StateT}\;(\lambda \Varid{s}\to \Varid{run_{StateT}}\;(\Varid{h_{State}}\;\Varid{p})\;(\Conid{SS}\;\Varid{xs}\;(\Varid{q}\mathbin{:}\Varid{stack}))))\;(\Conid{SS}\;\Varid{xs}\;\Varid{stack}){}\<[E]%
\\
\>[3]{}\mathrel{=}\mbox{\commentbegin ~  definition of \ensuremath{\Varid{run_{StateT}}}   \commentend}{}\<[E]%
\\
\>[3]{}\hsindent{3}{}\<[6]%
\>[6]{}(\lambda \Varid{s}\to \Varid{run_{StateT}}\;(\Varid{h_{State}}\;\Varid{p})\;(\Conid{SS}\;\Varid{xs}\;(\Varid{q}\mathbin{:}\Varid{stack})))\;(\Conid{SS}\;\Varid{xs}\;\Varid{stack}){}\<[E]%
\\
\>[3]{}\mathrel{=}\mbox{\commentbegin ~  function application   \commentend}{}\<[E]%
\\
\>[3]{}\hsindent{3}{}\<[6]%
\>[6]{}\Varid{run_{StateT}}\;(\Varid{h_{State}}\;\Varid{p})\;(\Conid{SS}\;\Varid{xs}\;(\Varid{q}\mathbin{:}\Varid{stack})){}\<[E]%
\ColumnHook
\end{hscode}\resethooks
\indentend \end{proof}

\section{Proofs for Modelling Two States with One State}
\label{app:states-state}

In this section, we prove the following theorem in
\Cref{sec:multiple-states}.

\statesState*

\begin{proof}
Instead of proving it directly, we show the correctness of the
isomorphism of \ensuremath{\Varid{nest}} and \ensuremath{\Varid{flatten}}, and prove the following equation:\indentbegin \begin{hscode}\SaveRestoreHook
\column{B}{@{}>{\hspre}l<{\hspost}@{}}%
\column{3}{@{}>{\hspre}l<{\hspost}@{}}%
\column{E}{@{}>{\hspre}l<{\hspost}@{}}%
\>[3]{}\Varid{flatten}\hsdot{\circ }{.}\Varid{h_{States}}\mathrel{=}\Varid{h_{State}}\hsdot{\circ }{.}\Varid{states2state}{}\<[E]%
\ColumnHook
\end{hscode}\resethooks
\indentend %
It is obvious that \ensuremath{\alpha} and \ensuremath{\alpha^{-1}} form an isomorphism, i.e.,
\ensuremath{\alpha\hsdot{\circ }{.}\alpha^{-1}\mathrel{=}\Varid{id}} and \ensuremath{\alpha^{-1}\hsdot{\circ }{.}\alpha\mathrel{=}\Varid{id}}. We show that \ensuremath{\Varid{nest}}
and \ensuremath{\Varid{flatten}} form an isomorphism by calculation.
For all \ensuremath{\Varid{t}\mathbin{::}\Conid{StateT}\;\Varid{s}_{1}\;(\Conid{StateT}\;\Varid{s}_{2}\;(\Conid{Free}\;\Varid{f}))\;\Varid{a}}, we show that \ensuremath{(\Varid{nest}\hsdot{\circ }{.}\Varid{flatten})\;\Varid{t}\mathrel{=}\Varid{t}}.
\indentbegin \begin{hscode}\SaveRestoreHook
\column{B}{@{}>{\hspre}l<{\hspost}@{}}%
\column{3}{@{}>{\hspre}l<{\hspost}@{}}%
\column{6}{@{}>{\hspre}l<{\hspost}@{}}%
\column{8}{@{}>{\hspre}l<{\hspost}@{}}%
\column{E}{@{}>{\hspre}l<{\hspost}@{}}%
\>[6]{}(\Varid{nest}\hsdot{\circ }{.}\Varid{flatten})\;\Varid{t}{}\<[E]%
\\
\>[3]{}\mathrel{=}\mbox{\commentbegin ~  definition of \ensuremath{\Varid{flatten}}   \commentend}{}\<[E]%
\\
\>[3]{}\hsindent{3}{}\<[6]%
\>[6]{}\Varid{nest}\mathbin{\$}\Conid{StateT}\mathbin{\$}\lambda (\Varid{s}_{1},\Varid{s}_{2})\to \alpha\mathbin{\langle\hspace{1.6pt}\mathclap{\raisebox{0.1pt}{\scalebox{1}{\$}}}\hspace{1.6pt}\rangle}\Varid{run_{StateT}}\;(\Varid{run_{StateT}}\;\Varid{t}\;\Varid{s}_{1})\;\Varid{s}_{2}{}\<[E]%
\\
\>[3]{}\mathrel{=}\mbox{\commentbegin ~  definition of \ensuremath{\Varid{nest}}   \commentend}{}\<[E]%
\\
\>[3]{}\hsindent{3}{}\<[6]%
\>[6]{}\Conid{StateT}\mathbin{\$}\lambda \Varid{s}_{1}\to \Conid{StateT}\mathbin{\$}\lambda \Varid{s}_{2}\to \alpha^{-1}\mathbin{\langle\hspace{1.6pt}\mathclap{\raisebox{0.1pt}{\scalebox{1}{\$}}}\hspace{1.6pt}\rangle}{}\<[E]%
\\
\>[6]{}\hsindent{2}{}\<[8]%
\>[8]{}\Varid{run_{StateT}}\;(\Conid{StateT}\mathbin{\$}\lambda (\Varid{s}_{1},\Varid{s}_{2})\to \alpha\mathbin{\langle\hspace{1.6pt}\mathclap{\raisebox{0.1pt}{\scalebox{1}{\$}}}\hspace{1.6pt}\rangle}\Varid{run_{StateT}}\;(\Varid{run_{StateT}}\;\Varid{t}\;\Varid{s}_{1})\;\Varid{s}_{2})\;(\Varid{s}_{1},\Varid{s}_{2}){}\<[E]%
\\
\>[3]{}\mathrel{=}\mbox{\commentbegin ~  definition of \ensuremath{\Varid{run_{StateT}}}   \commentend}{}\<[E]%
\\
\>[3]{}\hsindent{3}{}\<[6]%
\>[6]{}\Conid{StateT}\mathbin{\$}\lambda \Varid{s}_{1}\to \Conid{StateT}\mathbin{\$}\lambda \Varid{s}_{2}\to \alpha^{-1}\mathbin{\langle\hspace{1.6pt}\mathclap{\raisebox{0.1pt}{\scalebox{1}{\$}}}\hspace{1.6pt}\rangle}{}\<[E]%
\\
\>[6]{}\hsindent{2}{}\<[8]%
\>[8]{}(\lambda (\Varid{s}_{1},\Varid{s}_{2})\to \alpha\mathbin{\langle\hspace{1.6pt}\mathclap{\raisebox{0.1pt}{\scalebox{1}{\$}}}\hspace{1.6pt}\rangle}\Varid{run_{StateT}}\;(\Varid{run_{StateT}}\;\Varid{t}\;\Varid{s}_{1})\;\Varid{s}_{2})\;(\Varid{s}_{1},\Varid{s}_{2}){}\<[E]%
\\
\>[3]{}\mathrel{=}\mbox{\commentbegin ~  function application   \commentend}{}\<[E]%
\\
\>[3]{}\hsindent{3}{}\<[6]%
\>[6]{}\Conid{StateT}\mathbin{\$}\lambda \Varid{s}_{1}\to \Conid{StateT}\mathbin{\$}\lambda \Varid{s}_{2}\to \alpha^{-1}\mathbin{\langle\hspace{1.6pt}\mathclap{\raisebox{0.1pt}{\scalebox{1}{\$}}}\hspace{1.6pt}\rangle}(\alpha\mathbin{\langle\hspace{1.6pt}\mathclap{\raisebox{0.1pt}{\scalebox{1}{\$}}}\hspace{1.6pt}\rangle}\Varid{run_{StateT}}\;(\Varid{run_{StateT}}\;\Varid{t}\;\Varid{s}_{1})\;\Varid{s}_{2}){}\<[E]%
\\
\>[3]{}\mathrel{=}\mbox{\commentbegin ~  \Cref{eq:functor-composition}   \commentend}{}\<[E]%
\\
\>[3]{}\hsindent{3}{}\<[6]%
\>[6]{}\Conid{StateT}\mathbin{\$}\lambda \Varid{s}_{1}\to \Conid{StateT}\mathbin{\$}\lambda \Varid{s}_{2}\to (\Varid{fmap}\;(\alpha^{-1}\hsdot{\circ }{.}\alpha)\;(\Varid{run_{StateT}}\;(\Varid{run_{StateT}}\;\Varid{t}\;\Varid{s}_{1})\;\Varid{s}_{2})){}\<[E]%
\\
\>[3]{}\mathrel{=}\mbox{\commentbegin ~  \ensuremath{\alpha^{-1}\hsdot{\circ }{.}\alpha\mathrel{=}\Varid{id}}   \commentend}{}\<[E]%
\\
\>[3]{}\hsindent{3}{}\<[6]%
\>[6]{}\Conid{StateT}\mathbin{\$}\lambda \Varid{s}_{1}\to \Conid{StateT}\mathbin{\$}\lambda \Varid{s}_{2}\to (\Varid{fmap}\;\Varid{id}\;(\Varid{run_{StateT}}\;(\Varid{run_{StateT}}\;\Varid{t}\;\Varid{s}_{1})\;\Varid{s}_{2})){}\<[E]%
\\
\>[3]{}\mathrel{=}\mbox{\commentbegin ~  \Cref{eq:functor-identity}   \commentend}{}\<[E]%
\\
\>[3]{}\hsindent{3}{}\<[6]%
\>[6]{}\Conid{StateT}\mathbin{\$}\lambda \Varid{s}_{1}\to \Conid{StateT}\mathbin{\$}\lambda \Varid{s}_{2}\to (\Varid{run_{StateT}}\;(\Varid{run_{StateT}}\;\Varid{t}\;\Varid{s}_{1})\;\Varid{s}_{2}){}\<[E]%
\\
\>[3]{}\mathrel{=}\mbox{\commentbegin ~  \ensuremath{\Varid{\eta}}-reduction and reformulation   \commentend}{}\<[E]%
\\
\>[3]{}\hsindent{3}{}\<[6]%
\>[6]{}\Conid{StateT}\mathbin{\$}\lambda \Varid{s}_{1}\to (\Conid{StateT}\hsdot{\circ }{.}\Varid{run_{StateT}})\;(\Varid{run_{StateT}}\;\Varid{t}\;\Varid{s}_{1}){}\<[E]%
\\
\>[3]{}\mathrel{=}\mbox{\commentbegin ~  \ensuremath{\Conid{StateT}\hsdot{\circ }{.}\Varid{run_{StateT}}\mathrel{=}\Varid{id}}   \commentend}{}\<[E]%
\\
\>[3]{}\hsindent{3}{}\<[6]%
\>[6]{}\Conid{StateT}\mathbin{\$}\lambda \Varid{s}_{1}\to \Varid{run_{StateT}}\;\Varid{t}\;\Varid{s}_{1}{}\<[E]%
\\
\>[3]{}\mathrel{=}\mbox{\commentbegin ~  \ensuremath{\Varid{\eta}}-reduction   \commentend}{}\<[E]%
\\
\>[3]{}\hsindent{3}{}\<[6]%
\>[6]{}\Conid{StateT}\mathbin{\$}\Varid{run_{StateT}}\;\Varid{t}{}\<[E]%
\\
\>[3]{}\mathrel{=}\mbox{\commentbegin ~  \ensuremath{\Conid{StateT}\hsdot{\circ }{.}\Varid{run_{StateT}}\mathrel{=}\Varid{id}}   \commentend}{}\<[E]%
\\
\>[3]{}\hsindent{3}{}\<[6]%
\>[6]{}\Varid{t}{}\<[E]%
\ColumnHook
\end{hscode}\resethooks
\indentend For all \ensuremath{\Varid{t}\mathbin{::}\Conid{StateT}\;(\Varid{s}_{1},\Varid{s}_{2})\;(\Conid{Free}\;\Varid{f})\;\Varid{a}}, we show that \ensuremath{(\Varid{flatten}\hsdot{\circ }{.}\Varid{nest})\;\Varid{t}\mathrel{=}\Varid{t}}.
\indentbegin \begin{hscode}\SaveRestoreHook
\column{B}{@{}>{\hspre}l<{\hspost}@{}}%
\column{3}{@{}>{\hspre}l<{\hspost}@{}}%
\column{6}{@{}>{\hspre}l<{\hspost}@{}}%
\column{8}{@{}>{\hspre}l<{\hspost}@{}}%
\column{E}{@{}>{\hspre}l<{\hspost}@{}}%
\>[6]{}(\Varid{flatten}\hsdot{\circ }{.}\Varid{nest})\;\Varid{t}\mathrel{=}\Varid{t}{}\<[E]%
\\
\>[3]{}\mathrel{=}\mbox{\commentbegin ~  definition of \ensuremath{\Varid{nest}}   \commentend}{}\<[E]%
\\
\>[3]{}\hsindent{3}{}\<[6]%
\>[6]{}\Varid{flatten}\mathbin{\$}\Conid{StateT}\mathbin{\$}\lambda \Varid{s}_{1}\to \Conid{StateT}\mathbin{\$}\lambda \Varid{s}_{2}\to \alpha^{-1}\mathbin{\langle\hspace{1.6pt}\mathclap{\raisebox{0.1pt}{\scalebox{1}{\$}}}\hspace{1.6pt}\rangle}\Varid{run_{StateT}}\;\Varid{t}\;(\Varid{s}_{1},\Varid{s}_{2}){}\<[E]%
\\
\>[3]{}\mathrel{=}\mbox{\commentbegin ~  definition of \ensuremath{\Varid{flatten}}   \commentend}{}\<[E]%
\\
\>[3]{}\hsindent{3}{}\<[6]%
\>[6]{}\Conid{StateT}\mathbin{\$}\lambda (\Varid{s}_{1},\Varid{s}_{2})\to \alpha\mathbin{\langle\hspace{1.6pt}\mathclap{\raisebox{0.1pt}{\scalebox{1}{\$}}}\hspace{1.6pt}\rangle}{}\<[E]%
\\
\>[6]{}\hsindent{2}{}\<[8]%
\>[8]{}\Varid{run_{StateT}}\;(\Varid{run_{StateT}}\;(\Conid{StateT}\mathbin{\$}\lambda \Varid{s}_{1}\to \Conid{StateT}\mathbin{\$}\lambda \Varid{s}_{2}\to \alpha^{-1}\mathbin{\langle\hspace{1.6pt}\mathclap{\raisebox{0.1pt}{\scalebox{1}{\$}}}\hspace{1.6pt}\rangle}\Varid{run_{StateT}}\;\Varid{t}\;(\Varid{s}_{1},\Varid{s}_{2}))\;\Varid{s}_{1})\;\Varid{s}_{2}{}\<[E]%
\\
\>[3]{}\mathrel{=}\mbox{\commentbegin ~  definition of \ensuremath{\Varid{run_{StateT}}}   \commentend}{}\<[E]%
\\
\>[3]{}\hsindent{3}{}\<[6]%
\>[6]{}\Conid{StateT}\mathbin{\$}\lambda (\Varid{s}_{1},\Varid{s}_{2})\to \alpha\mathbin{\langle\hspace{1.6pt}\mathclap{\raisebox{0.1pt}{\scalebox{1}{\$}}}\hspace{1.6pt}\rangle}{}\<[E]%
\\
\>[6]{}\hsindent{2}{}\<[8]%
\>[8]{}\Varid{run_{StateT}}\;((\lambda \Varid{s}_{1}\to \Conid{StateT}\mathbin{\$}\lambda \Varid{s}_{2}\to \alpha^{-1}\mathbin{\langle\hspace{1.6pt}\mathclap{\raisebox{0.1pt}{\scalebox{1}{\$}}}\hspace{1.6pt}\rangle}\Varid{run_{StateT}}\;\Varid{t}\;(\Varid{s}_{1},\Varid{s}_{2}))\;\Varid{s}_{1})\;\Varid{s}_{2}{}\<[E]%
\\
\>[3]{}\mathrel{=}\mbox{\commentbegin ~  function application   \commentend}{}\<[E]%
\\
\>[3]{}\hsindent{3}{}\<[6]%
\>[6]{}\Conid{StateT}\mathbin{\$}\lambda (\Varid{s}_{1},\Varid{s}_{2})\to \alpha\mathbin{\langle\hspace{1.6pt}\mathclap{\raisebox{0.1pt}{\scalebox{1}{\$}}}\hspace{1.6pt}\rangle}\Varid{run_{StateT}}\;(\Conid{StateT}\mathbin{\$}\lambda \Varid{s}_{2}\to \alpha^{-1}\mathbin{\langle\hspace{1.6pt}\mathclap{\raisebox{0.1pt}{\scalebox{1}{\$}}}\hspace{1.6pt}\rangle}\Varid{run_{StateT}}\;\Varid{t}\;(\Varid{s}_{1},\Varid{s}_{2}))\;\Varid{s}_{2}{}\<[E]%
\\
\>[3]{}\mathrel{=}\mbox{\commentbegin ~  definition of \ensuremath{\Varid{run_{StateT}}}   \commentend}{}\<[E]%
\\
\>[3]{}\hsindent{3}{}\<[6]%
\>[6]{}\Conid{StateT}\mathbin{\$}\lambda (\Varid{s}_{1},\Varid{s}_{2})\to \alpha\mathbin{\langle\hspace{1.6pt}\mathclap{\raisebox{0.1pt}{\scalebox{1}{\$}}}\hspace{1.6pt}\rangle}(\lambda \Varid{s}_{2}\to \alpha^{-1}\mathbin{\langle\hspace{1.6pt}\mathclap{\raisebox{0.1pt}{\scalebox{1}{\$}}}\hspace{1.6pt}\rangle}\Varid{run_{StateT}}\;\Varid{t}\;(\Varid{s}_{1},\Varid{s}_{2}))\;\Varid{s}_{2}{}\<[E]%
\\
\>[3]{}\mathrel{=}\mbox{\commentbegin ~  function application   \commentend}{}\<[E]%
\\
\>[3]{}\hsindent{3}{}\<[6]%
\>[6]{}\Conid{StateT}\mathbin{\$}\lambda (\Varid{s}_{1},\Varid{s}_{2})\to \alpha\mathbin{\langle\hspace{1.6pt}\mathclap{\raisebox{0.1pt}{\scalebox{1}{\$}}}\hspace{1.6pt}\rangle}(\alpha^{-1}\mathbin{\langle\hspace{1.6pt}\mathclap{\raisebox{0.1pt}{\scalebox{1}{\$}}}\hspace{1.6pt}\rangle}\Varid{run_{StateT}}\;\Varid{t}\;(\Varid{s}_{1},\Varid{s}_{2})){}\<[E]%
\\
\>[3]{}\mathrel{=}\mbox{\commentbegin ~  definition of \ensuremath{\mathbin{\langle\hspace{1.6pt}\mathclap{\raisebox{0.1pt}{\scalebox{1}{\$}}}\hspace{1.6pt}\rangle}}   \commentend}{}\<[E]%
\\
\>[3]{}\hsindent{3}{}\<[6]%
\>[6]{}\Conid{StateT}\mathbin{\$}\lambda (\Varid{s}_{1},\Varid{s}_{2})\to \Varid{fmap}\;\alpha\;(\Varid{fmap}\;\alpha^{-1}\;(\Varid{run_{StateT}}\;\Varid{t}\;(\Varid{s}_{1},\Varid{s}_{2}))){}\<[E]%
\\
\>[3]{}\mathrel{=}\mbox{\commentbegin ~  \Cref{eq:functor-composition}   \commentend}{}\<[E]%
\\
\>[3]{}\hsindent{3}{}\<[6]%
\>[6]{}\Conid{StateT}\mathbin{\$}\lambda (\Varid{s}_{1},\Varid{s}_{2})\to \Varid{fmap}\;(\alpha\hsdot{\circ }{.}\alpha^{-1})\;(\Varid{run_{StateT}}\;\Varid{t}\;(\Varid{s}_{1},\Varid{s}_{2})){}\<[E]%
\\
\>[3]{}\mathrel{=}\mbox{\commentbegin ~  \ensuremath{\alpha\hsdot{\circ }{.}\alpha^{-1}\mathrel{=}\Varid{id}}   \commentend}{}\<[E]%
\\
\>[3]{}\hsindent{3}{}\<[6]%
\>[6]{}\Conid{StateT}\mathbin{\$}\lambda (\Varid{s}_{1},\Varid{s}_{2})\to \Varid{fmap}\;\Varid{id}\;(\Varid{run_{StateT}}\;\Varid{t}\;(\Varid{s}_{1},\Varid{s}_{2})){}\<[E]%
\\
\>[3]{}\mathrel{=}\mbox{\commentbegin ~  \Cref{eq:functor-identity}   \commentend}{}\<[E]%
\\
\>[3]{}\hsindent{3}{}\<[6]%
\>[6]{}\Conid{StateT}\mathbin{\$}\lambda (\Varid{s}_{1},\Varid{s}_{2})\to \Varid{run_{StateT}}\;\Varid{t}\;(\Varid{s}_{1},\Varid{s}_{2}){}\<[E]%
\\
\>[3]{}\mathrel{=}\mbox{\commentbegin ~  \ensuremath{\Varid{\eta}}-reduction   \commentend}{}\<[E]%
\\
\>[3]{}\hsindent{3}{}\<[6]%
\>[6]{}\Conid{StateT}\mathbin{\$}\Varid{run_{StateT}}\;\Varid{t}{}\<[E]%
\\
\>[3]{}\mathrel{=}\mbox{\commentbegin ~  \ensuremath{\Conid{StateT}\hsdot{\circ }{.}\Varid{run_{StateT}}\mathrel{=}\Varid{id}}   \commentend}{}\<[E]%
\\
\>[3]{}\hsindent{3}{}\<[6]%
\>[6]{}\Varid{t}{}\<[E]%
\ColumnHook
\end{hscode}\resethooks
\indentend Then, we first calculate the LHS \ensuremath{\Varid{flatten}\hsdot{\circ }{.}\Varid{h_{States}}} into one function
\ensuremath{\Varid{h_{States}^\prime}} which is defined as\indentbegin \begin{hscode}\SaveRestoreHook
\column{B}{@{}>{\hspre}l<{\hspost}@{}}%
\column{3}{@{}>{\hspre}l<{\hspost}@{}}%
\column{E}{@{}>{\hspre}l<{\hspost}@{}}%
\>[3]{}\Varid{h_{States}^\prime}\mathbin{::}\Conid{Functor}\;\Varid{f}\Rightarrow \Conid{Free}\;(\Varid{State_{F}}\;\Varid{s}_{1}\mathrel{{:}{+}{:}}\Varid{State_{F}}\;\Varid{s}_{2}\mathrel{{:}{+}{:}}\Varid{f})\;\Varid{a}\to \Conid{StateT}\;(\Varid{s}_{1},\Varid{s}_{2})\;(\Conid{Free}\;\Varid{f})\;\Varid{a}{}\<[E]%
\\
\>[3]{}\Varid{h_{States}^\prime}\;\Varid{t}\mathrel{=}\Conid{StateT}\mathbin{\$}\lambda (\Varid{s}_{1},\Varid{s}_{2})\to \alpha\mathbin{\langle\hspace{1.6pt}\mathclap{\raisebox{0.1pt}{\scalebox{1}{\$}}}\hspace{1.6pt}\rangle}\Varid{run_{StateT}}\;(\Varid{h_{State}}\;(\Varid{run_{StateT}}\;(\Varid{h_{State}}\;\Varid{t})\;\Varid{s}_{1}))\;\Varid{s}_{2}{}\<[E]%
\ColumnHook
\end{hscode}\resethooks
\indentend %
For all \ensuremath{\Varid{t}\mathbin{::}\Conid{Free}\;(\Varid{State_{F}}\;\Varid{s}_{1}\mathrel{{:}{+}{:}}\Varid{State_{F}}\;\Varid{s}_{2}\mathrel{{:}{+}{:}}\Varid{f})\;\Varid{a}}, we show the equation
\ensuremath{(\Varid{flatten}\hsdot{\circ }{.}\Varid{h_{States}})\;\Varid{t}\mathrel{=}\Varid{h_{States}^\prime}\;\Varid{t}} by the following calculation.
\indentbegin \begin{hscode}\SaveRestoreHook
\column{B}{@{}>{\hspre}l<{\hspost}@{}}%
\column{3}{@{}>{\hspre}l<{\hspost}@{}}%
\column{6}{@{}>{\hspre}l<{\hspost}@{}}%
\column{8}{@{}>{\hspre}l<{\hspost}@{}}%
\column{E}{@{}>{\hspre}l<{\hspost}@{}}%
\>[6]{}(\Varid{flatten}\hsdot{\circ }{.}\Varid{h_{States}})\;\Varid{t}{}\<[E]%
\\
\>[3]{}\mathrel{=}\mbox{\commentbegin ~  definition of \ensuremath{\Varid{h_{States}}}   \commentend}{}\<[E]%
\\
\>[3]{}\hsindent{3}{}\<[6]%
\>[6]{}(\Varid{flatten}\hsdot{\circ }{.}(\lambda \Varid{t}\to \Conid{StateT}\;(\Varid{h_{State}}\hsdot{\circ }{.}\Varid{run_{StateT}}\;(\Varid{h_{State}}\;\Varid{t}))))\;\Varid{t}{}\<[E]%
\\
\>[3]{}\mathrel{=}\mbox{\commentbegin ~  function application   \commentend}{}\<[E]%
\\
\>[3]{}\hsindent{3}{}\<[6]%
\>[6]{}\Varid{flatten}\;(\Conid{StateT}\;(\Varid{h_{State}}\hsdot{\circ }{.}\Varid{run_{StateT}}\;(\Varid{h_{State}}\;\Varid{t}))){}\<[E]%
\\
\>[3]{}\mathrel{=}\mbox{\commentbegin ~  definition of \ensuremath{\Varid{flatten}}   \commentend}{}\<[E]%
\\
\>[3]{}\hsindent{3}{}\<[6]%
\>[6]{}\Conid{StateT}\mathbin{\$}\lambda (\Varid{s}_{1},\Varid{s}_{2})\to \alpha\mathbin{\langle\hspace{1.6pt}\mathclap{\raisebox{0.1pt}{\scalebox{1}{\$}}}\hspace{1.6pt}\rangle}{}\<[E]%
\\
\>[6]{}\hsindent{2}{}\<[8]%
\>[8]{}\Varid{run_{StateT}}\;(\Varid{run_{StateT}}\;(\Conid{StateT}\;(\Varid{h_{State}}\hsdot{\circ }{.}\Varid{run_{StateT}}\;(\Varid{h_{State}}\;\Varid{t})))\;\Varid{s}_{1})\;\Varid{s}_{2}{}\<[E]%
\\
\>[3]{}\mathrel{=}\mbox{\commentbegin ~  definition of \ensuremath{\Varid{run_{StateT}}}   \commentend}{}\<[E]%
\\
\>[3]{}\hsindent{3}{}\<[6]%
\>[6]{}\Conid{StateT}\mathbin{\$}\lambda (\Varid{s}_{1},\Varid{s}_{2})\to \alpha\mathbin{\langle\hspace{1.6pt}\mathclap{\raisebox{0.1pt}{\scalebox{1}{\$}}}\hspace{1.6pt}\rangle}{}\<[E]%
\\
\>[6]{}\hsindent{2}{}\<[8]%
\>[8]{}\Varid{run_{StateT}}\;((\Varid{h_{State}}\hsdot{\circ }{.}\Varid{run_{StateT}}\;(\Varid{h_{State}}\;\Varid{t}))\;\Varid{s}_{1})\;\Varid{s}_{2}{}\<[E]%
\\
\>[3]{}\mathrel{=}\mbox{\commentbegin ~  definition of \ensuremath{\Varid{h_{States}^\prime}}   \commentend}{}\<[E]%
\\
\>[3]{}\Varid{h_{States}^\prime}\;\Varid{t}{}\<[E]%
\ColumnHook
\end{hscode}\resethooks
\indentend Now we only need to show that for any input \ensuremath{\Varid{t}\mathbin{::}\Conid{Free}\;(\Varid{State_{F}}\;\Varid{s}_{1}\mathrel{{:}{+}{:}}\Varid{State_{F}}\;\Varid{s}_{2}\mathrel{{:}{+}{:}}\Varid{f})\;\Varid{a}}, the equation \ensuremath{\Varid{h_{States}^\prime}\;\Varid{t}\mathrel{=}(\Varid{h_{State}}\hsdot{\circ }{.}\Varid{states2state})\;\Varid{t}} holds.
Note that both sides use folds.
We can proceed with either fold fusion, as what we have done in the
proofs of other theorems, or a direct structural induction on the
input \ensuremath{\Varid{t}}. Although using fold fusion makes the proof simpler than
using structural induction, we opt for the latter here to show that
the our methods of defining effects and translations based on
algebraic effects and handlers also work well with structural
induction.

\noindent \mbox{\underline{case \ensuremath{\Varid{t}\mathrel{=}\Conid{Var}\;\Varid{x}}}}\indentbegin \begin{hscode}\SaveRestoreHook
\column{B}{@{}>{\hspre}l<{\hspost}@{}}%
\column{3}{@{}>{\hspre}l<{\hspost}@{}}%
\column{5}{@{}>{\hspre}l<{\hspost}@{}}%
\column{6}{@{}>{\hspre}l<{\hspost}@{}}%
\column{27}{@{}>{\hspre}c<{\hspost}@{}}%
\column{27E}{@{}l@{}}%
\column{31}{@{}>{\hspre}l<{\hspost}@{}}%
\column{E}{@{}>{\hspre}l<{\hspost}@{}}%
\>[6]{}(\Varid{h_{State}}\hsdot{\circ }{.}\Varid{states2state})\;(\Conid{Var}\;\Varid{x}){}\<[E]%
\\
\>[3]{}\mathrel{=}\mbox{\commentbegin ~  definition of \ensuremath{\Varid{states2state}}   \commentend}{}\<[E]%
\\
\>[3]{}\hsindent{3}{}\<[6]%
\>[6]{}\Varid{h_{State}}\;(\Conid{Var}\;\Varid{x}){}\<[E]%
\\
\>[3]{}\mathrel{=}\mbox{\commentbegin ~  definition of \ensuremath{\Varid{h_{State}}}   \commentend}{}\<[E]%
\\
\>[3]{}\hsindent{2}{}\<[5]%
\>[5]{}\Conid{StateT}\mathbin{\$}\lambda \Varid{s}\to \Varid{\eta}\;(\Varid{x},\Varid{s}){}\<[E]%
\\
\>[3]{}\mathrel{=}\mbox{\commentbegin ~  let \ensuremath{\Varid{s}\mathrel{=}(\Varid{s}_{1},\Varid{s}_{2})}   \commentend}{}\<[E]%
\\
\>[3]{}\hsindent{3}{}\<[6]%
\>[6]{}\Conid{StateT}\mathbin{\$}\lambda (\Varid{s}_{1},\Varid{s}_{2})\to \Conid{Var}\;(\Varid{x},(\Varid{s}_{1},\Varid{s}_{2})){}\<[E]%
\\
\>[3]{}\mathrel{=}\mbox{\commentbegin ~  definition of \ensuremath{\alpha}   \commentend}{}\<[E]%
\\
\>[3]{}\hsindent{3}{}\<[6]%
\>[6]{}\Conid{StateT}\mathbin{\$}\lambda (\Varid{s}_{1},\Varid{s}_{2}){}\<[27]%
\>[27]{}\to {}\<[27E]%
\>[31]{}\Conid{Var}\;(\alpha\;((\Varid{x},\Varid{s}_{1}),\Varid{s}_{2})){}\<[E]%
\\
\>[3]{}\mathrel{=}\mbox{\commentbegin ~  definition of \ensuremath{\Varid{fmap}}   \commentend}{}\<[E]%
\\
\>[3]{}\hsindent{3}{}\<[6]%
\>[6]{}\Conid{StateT}\mathbin{\$}\lambda (\Varid{s}_{1},\Varid{s}_{2}){}\<[27]%
\>[27]{}\to {}\<[27E]%
\>[31]{}\Varid{fmap}\;\alpha\mathbin{\$}\Conid{Var}\;((\Varid{x},\Varid{s}_{1}),\Varid{s}_{2}){}\<[E]%
\\
\>[3]{}\mathrel{=}\mbox{\commentbegin ~  definition of \ensuremath{\Varid{\eta}}   \commentend}{}\<[E]%
\\
\>[3]{}\hsindent{3}{}\<[6]%
\>[6]{}\Conid{StateT}\mathbin{\$}\lambda (\Varid{s}_{1},\Varid{s}_{2}){}\<[27]%
\>[27]{}\to {}\<[27E]%
\>[31]{}\Varid{fmap}\;\alpha\mathbin{\$}\Varid{\eta}\;((\Varid{x},\Varid{s}_{1}),\Varid{s}_{2}){}\<[E]%
\\
\>[3]{}\mathrel{=}\mbox{\commentbegin ~  \ensuremath{\Varid{\beta}}-expansion   \commentend}{}\<[E]%
\\
\>[3]{}\hsindent{3}{}\<[6]%
\>[6]{}\Conid{StateT}\mathbin{\$}\lambda (\Varid{s}_{1},\Varid{s}_{2}){}\<[27]%
\>[27]{}\to {}\<[27E]%
\>[31]{}\Varid{fmap}\;\alpha\mathbin{\$}(\lambda \Varid{s}\to \Varid{\eta}\;((\Varid{x},\Varid{s}_{1}),\Varid{s}))\;\Varid{s}_{2}{}\<[E]%
\\
\>[3]{}\mathrel{=}\mbox{\commentbegin ~  definition of \ensuremath{\Varid{run_{StateT}}}   \commentend}{}\<[E]%
\\
\>[3]{}\hsindent{3}{}\<[6]%
\>[6]{}\Conid{StateT}\mathbin{\$}\lambda (\Varid{s}_{1},\Varid{s}_{2}){}\<[27]%
\>[27]{}\to {}\<[27E]%
\>[31]{}\Varid{fmap}\;\alpha\mathbin{\$}\Varid{run_{StateT}}\;(\Conid{StateT}\mathbin{\$}\lambda \Varid{s}\to \Varid{\eta}\;((\Varid{x},\Varid{s}_{1}),\Varid{s}))\;\Varid{s}_{2}{}\<[E]%
\\
\>[3]{}\mathrel{=}\mbox{\commentbegin ~  definition of \ensuremath{\Varid{h_{State}}}   \commentend}{}\<[E]%
\\
\>[3]{}\hsindent{3}{}\<[6]%
\>[6]{}\Conid{StateT}\mathbin{\$}\lambda (\Varid{s}_{1},\Varid{s}_{2}){}\<[27]%
\>[27]{}\to {}\<[27E]%
\>[31]{}\Varid{fmap}\;\alpha\mathbin{\$}\Varid{run_{StateT}}\;(\Varid{h_{State}}\;(\Varid{\eta}\;(\Varid{x},\Varid{s}_{1})))\;\Varid{s}_{2}{}\<[E]%
\\
\>[3]{}\mathrel{=}\mbox{\commentbegin ~  \ensuremath{\Varid{\beta}}-expansion   \commentend}{}\<[E]%
\\
\>[3]{}\hsindent{3}{}\<[6]%
\>[6]{}\Conid{StateT}\mathbin{\$}\lambda (\Varid{s}_{1},\Varid{s}_{2}){}\<[27]%
\>[27]{}\to {}\<[27E]%
\>[31]{}\Varid{fmap}\;\alpha\mathbin{\$}\Varid{run_{StateT}}\;(\Varid{h_{State}}\;((\lambda \Varid{s}\to \Varid{\eta}\;(\Varid{x},\Varid{s}))\;\Varid{s}_{1}))\;\Varid{s}_{2}{}\<[E]%
\\
\>[3]{}\mathrel{=}\mbox{\commentbegin ~  definition of \ensuremath{\Varid{run_{StateT}}}   \commentend}{}\<[E]%
\\
\>[3]{}\hsindent{3}{}\<[6]%
\>[6]{}\Conid{StateT}\mathbin{\$}\lambda (\Varid{s}_{1},\Varid{s}_{2}){}\<[27]%
\>[27]{}\to {}\<[27E]%
\>[31]{}\Varid{fmap}\;\alpha\mathbin{\$}\Varid{run_{StateT}}\;(\Varid{h_{State}}\;(\Varid{run_{StateT}}\;(\Conid{StateT}\mathbin{\$}\lambda \Varid{s}\to \Varid{\eta}\;(\Varid{x},\Varid{s}))\;\Varid{s}_{1}))\;\Varid{s}_{2}{}\<[E]%
\\
\>[3]{}\mathrel{=}\mbox{\commentbegin ~  definition of \ensuremath{\Varid{h_{State}}}   \commentend}{}\<[E]%
\\
\>[3]{}\hsindent{3}{}\<[6]%
\>[6]{}\Conid{StateT}\mathbin{\$}\lambda (\Varid{s}_{1},\Varid{s}_{2}){}\<[27]%
\>[27]{}\to {}\<[27E]%
\>[31]{}\Varid{fmap}\;\alpha\mathbin{\$}\Varid{run_{StateT}}\;(\Varid{h_{State}}\;(\Varid{run_{StateT}}\;(\Varid{h_{State}}\;(\Conid{Var}\;\Varid{x}))\;\Varid{s}_{1}))\;\Varid{s}_{2}{}\<[E]%
\\
\>[3]{}\mathrel{=}\mbox{\commentbegin ~  definition of \ensuremath{\mathbin{\langle\hspace{1.6pt}\mathclap{\raisebox{0.1pt}{\scalebox{1}{\$}}}\hspace{1.6pt}\rangle}}   \commentend}{}\<[E]%
\\
\>[3]{}\hsindent{3}{}\<[6]%
\>[6]{}\Conid{StateT}\mathbin{\$}\lambda (\Varid{s}_{1},\Varid{s}_{2}){}\<[27]%
\>[27]{}\to {}\<[27E]%
\>[31]{}\alpha\mathbin{\langle\hspace{1.6pt}\mathclap{\raisebox{0.1pt}{\scalebox{1}{\$}}}\hspace{1.6pt}\rangle}\Varid{run_{StateT}}\;(\Varid{h_{State}}\;(\Varid{run_{StateT}}\;(\Varid{h_{State}}\;(\Conid{Var}\;\Varid{x}))\;\Varid{s}_{1}))\;\Varid{s}_{2}{}\<[E]%
\\
\>[3]{}\mathrel{=}\mbox{\commentbegin ~  definition of \ensuremath{\Varid{h_{States}^\prime}}   \commentend}{}\<[E]%
\\
\>[3]{}\hsindent{3}{}\<[6]%
\>[6]{}\Varid{h_{States}^\prime}\;(\Conid{Var}\;\Varid{x}){}\<[E]%
\ColumnHook
\end{hscode}\resethooks
\indentend %

\noindent \mbox{\underline{case \ensuremath{\Varid{t}\mathrel{=}\Conid{Op}\;(\Conid{Inl}\;(\Conid{Get}\;\Varid{k}))}}}

Induction hypothesis: \ensuremath{\Varid{h_{States}^\prime}\;(\Varid{k}\;\Varid{s})\mathrel{=}(\Varid{h_{State}}\hsdot{\circ }{.}\Varid{states2state})\;(\Varid{k}\;\Varid{s})} for any \ensuremath{\Varid{s}}.
\indentbegin \begin{hscode}\SaveRestoreHook
\column{B}{@{}>{\hspre}l<{\hspost}@{}}%
\column{3}{@{}>{\hspre}l<{\hspost}@{}}%
\column{6}{@{}>{\hspre}l<{\hspost}@{}}%
\column{8}{@{}>{\hspre}l<{\hspost}@{}}%
\column{27}{@{}>{\hspre}c<{\hspost}@{}}%
\column{27E}{@{}l@{}}%
\column{31}{@{}>{\hspre}l<{\hspost}@{}}%
\column{42}{@{}>{\hspre}c<{\hspost}@{}}%
\column{42E}{@{}l@{}}%
\column{46}{@{}>{\hspre}l<{\hspost}@{}}%
\column{61}{@{}>{\hspre}c<{\hspost}@{}}%
\column{61E}{@{}l@{}}%
\column{65}{@{}>{\hspre}l<{\hspost}@{}}%
\column{E}{@{}>{\hspre}l<{\hspost}@{}}%
\>[6]{}(\Varid{h_{State}}\hsdot{\circ }{.}\Varid{states2state})\;(\Conid{Op}\;(\Conid{Inl}\;(\Conid{Get}\;\Varid{k}))){}\<[E]%
\\
\>[3]{}\mathrel{=}\mbox{\commentbegin ~  definition of \ensuremath{\Varid{states2state}}   \commentend}{}\<[E]%
\\
\>[3]{}\hsindent{3}{}\<[6]%
\>[6]{}\Varid{h_{State}}\mathbin{\$}\Varid{get}>\!\!>\!\!=\lambda (\Varid{s}_{1},{}\<[31]%
\>[31]{}\anonymous )\to \Varid{states2state}\;(\Varid{k}\;\Varid{s}_{1}){}\<[E]%
\\
\>[3]{}\mathrel{=}\mbox{\commentbegin ~  definition of \ensuremath{\Varid{get}}   \commentend}{}\<[E]%
\\
\>[3]{}\hsindent{3}{}\<[6]%
\>[6]{}\Varid{h_{State}}\mathbin{\$}\Conid{Op}\;(\Conid{Inl}\;(\Conid{Get}\;\Varid{\eta}))>\!\!>\!\!=\lambda (\Varid{s}_{1},\anonymous )\to \Varid{states2state}\;(\Varid{k}\;\Varid{s}_{1}){}\<[E]%
\\
\>[3]{}\mathrel{=}\mbox{\commentbegin ~  definition of \ensuremath{(>\!\!>\!\!=)}   \commentend}{}\<[E]%
\\
\>[3]{}\hsindent{3}{}\<[6]%
\>[6]{}\Varid{h_{State}}\;(\Conid{Op}\;(\Conid{Inl}\;(\Conid{Get}\;(\lambda (\Varid{s}_{1},\anonymous )\to \Varid{states2state}\;(\Varid{k}\;\Varid{s}_{1}))))){}\<[E]%
\\
\>[3]{}\mathrel{=}\mbox{\commentbegin ~  definition of \ensuremath{\Varid{h_{State}}}   \commentend}{}\<[E]%
\\
\>[3]{}\hsindent{3}{}\<[6]%
\>[6]{}\Conid{StateT}\mathbin{\$}\lambda \Varid{s}\to \Varid{run_{StateT}}\;((\lambda (\Varid{s}_{1},\anonymous )\to \Varid{h_{State}}\;(\Varid{states2state}\;(\Varid{k}\;\Varid{s}_{1})))\;\Varid{s})\;\Varid{s}{}\<[E]%
\\
\>[3]{}\mathrel{=}\mbox{\commentbegin ~  let \ensuremath{\Varid{s}\mathrel{=}(\Varid{s}_{1},\Varid{s}_{2})}   \commentend}{}\<[E]%
\\
\>[3]{}\hsindent{3}{}\<[6]%
\>[6]{}\Conid{StateT}\mathbin{\$}\lambda (\Varid{s}_{1},\Varid{s}_{2})\to \Varid{run_{StateT}}\;((\lambda (\Varid{s}_{1},\anonymous )\to \Varid{h_{State}}\;(\Varid{states2state}\;(\Varid{k}\;\Varid{s}_{1})))\;(\Varid{s}_{1},\Varid{s}_{2}))\;(\Varid{s}_{1},\Varid{s}_{2}){}\<[E]%
\\
\>[3]{}\mathrel{=}\mbox{\commentbegin ~  function application   \commentend}{}\<[E]%
\\
\>[3]{}\hsindent{3}{}\<[6]%
\>[6]{}\Conid{StateT}\mathbin{\$}\lambda (\Varid{s}_{1},\Varid{s}_{2})\to \Varid{run_{StateT}}\;(\Varid{h_{State}}\;(\Varid{states2state}\;(\Varid{k}\;\Varid{s}_{1})))\;(\Varid{s}_{1},\Varid{s}_{2}){}\<[E]%
\\
\>[3]{}\mathrel{=}\mbox{\commentbegin ~  induction hypothesis   \commentend}{}\<[E]%
\\
\>[3]{}\hsindent{3}{}\<[6]%
\>[6]{}\Conid{StateT}\mathbin{\$}\lambda (\Varid{s}_{1},\Varid{s}_{2})\to \Varid{run_{StateT}}\;(\Varid{h_{States}^\prime}\;(\Varid{k}\;\Varid{s}_{1}))\;(\Varid{s}_{1},\Varid{s}_{2}){}\<[E]%
\\
\>[3]{}\mathrel{=}\mbox{\commentbegin ~  definition of \ensuremath{\Varid{h_{States}^\prime}}   \commentend}{}\<[E]%
\\
\>[3]{}\hsindent{3}{}\<[6]%
\>[6]{}\Conid{StateT}\mathbin{\$}\lambda (\Varid{s}_{1},\Varid{s}_{2})\to \Varid{run_{StateT}}\;(\Conid{StateT}\mathbin{\$}\lambda (\Varid{s}_{1},\Varid{s}_{2}){}\<[61]%
\>[61]{}\to {}\<[61E]%
\>[65]{}\alpha\mathbin{\langle\hspace{1.6pt}\mathclap{\raisebox{0.1pt}{\scalebox{1}{\$}}}\hspace{1.6pt}\rangle}{}\<[E]%
\\
\>[6]{}\hsindent{2}{}\<[8]%
\>[8]{}\Varid{run_{StateT}}\;(\Varid{h_{State}}\;(\Varid{run_{StateT}}\;(\Varid{h_{State}}\;(\Varid{k}\;\Varid{s}_{1}))\;\Varid{s}_{1}))\;\Varid{s}_{2})\;(\Varid{s}_{1},\Varid{s}_{2}){}\<[E]%
\\
\>[3]{}\mathrel{=}\mbox{\commentbegin ~  definition of \ensuremath{\Varid{run_{StateT}}}   \commentend}{}\<[E]%
\\
\>[3]{}\hsindent{3}{}\<[6]%
\>[6]{}\Conid{StateT}\mathbin{\$}\lambda (\Varid{s}_{1},\Varid{s}_{2})\to (\lambda (\Varid{s}_{1},\Varid{s}_{2}){}\<[42]%
\>[42]{}\to {}\<[42E]%
\>[46]{}\alpha\mathbin{\langle\hspace{1.6pt}\mathclap{\raisebox{0.1pt}{\scalebox{1}{\$}}}\hspace{1.6pt}\rangle}{}\<[E]%
\\
\>[6]{}\hsindent{2}{}\<[8]%
\>[8]{}\Varid{run_{StateT}}\;(\Varid{h_{State}}\;(\Varid{run_{StateT}}\;(\Varid{h_{State}}\;(\Varid{k}\;\Varid{s}_{1}))\;\Varid{s}_{1}))\;\Varid{s}_{2})\;(\Varid{s}_{1},\Varid{s}_{2}){}\<[E]%
\\
\>[3]{}\mathrel{=}\mbox{\commentbegin ~  function application  \commentend}{}\<[E]%
\\
\>[3]{}\hsindent{3}{}\<[6]%
\>[6]{}\Conid{StateT}\mathbin{\$}\lambda (\Varid{s}_{1},\Varid{s}_{2}){}\<[27]%
\>[27]{}\to {}\<[27E]%
\>[31]{}\alpha\mathbin{\langle\hspace{1.6pt}\mathclap{\raisebox{0.1pt}{\scalebox{1}{\$}}}\hspace{1.6pt}\rangle}\Varid{run_{StateT}}\;(\Varid{h_{State}}\;(\Varid{run_{StateT}}\;(\Varid{h_{State}}\;(\Varid{k}\;\Varid{s}_{1}))\;\Varid{s}_{1}))\;\Varid{s}_{2}{}\<[E]%
\\
\>[3]{}\mathrel{=}\mbox{\commentbegin ~  \ensuremath{\Varid{\beta}}-expansion   \commentend}{}\<[E]%
\\
\>[3]{}\hsindent{3}{}\<[6]%
\>[6]{}\Conid{StateT}\mathbin{\$}\lambda (\Varid{s}_{1},\Varid{s}_{2}){}\<[27]%
\>[27]{}\to {}\<[27E]%
\>[31]{}\alpha\mathbin{\langle\hspace{1.6pt}\mathclap{\raisebox{0.1pt}{\scalebox{1}{\$}}}\hspace{1.6pt}\rangle}{}\<[E]%
\\
\>[6]{}\hsindent{2}{}\<[8]%
\>[8]{}\Varid{run_{StateT}}\;(\Varid{h_{State}}\;((\lambda \Varid{s}\to \Varid{run_{StateT}}\;(\Varid{h_{State}}\;(\Varid{k}\;\Varid{s}))\;\Varid{s})\;\Varid{s}_{1}))\;\Varid{s}_{2}{}\<[E]%
\\
\>[3]{}\mathrel{=}\mbox{\commentbegin ~  definition of \ensuremath{\Varid{run_{StateT}}}   \commentend}{}\<[E]%
\\
\>[3]{}\hsindent{3}{}\<[6]%
\>[6]{}\Conid{StateT}\mathbin{\$}\lambda (\Varid{s}_{1},\Varid{s}_{2}){}\<[27]%
\>[27]{}\to {}\<[27E]%
\>[31]{}\alpha\mathbin{\langle\hspace{1.6pt}\mathclap{\raisebox{0.1pt}{\scalebox{1}{\$}}}\hspace{1.6pt}\rangle}{}\<[E]%
\\
\>[6]{}\hsindent{2}{}\<[8]%
\>[8]{}\Varid{run_{StateT}}\;(\Varid{h_{State}}\;(\Varid{run_{StateT}}\;(\Conid{StateT}\mathbin{\$}\lambda \Varid{s}\to \Varid{run_{StateT}}\;(\Varid{h_{State}}\;(\Varid{k}\;\Varid{s}))\;\Varid{s})\;\Varid{s}_{1}))\;\Varid{s}_{2}{}\<[E]%
\\
\>[3]{}\mathrel{=}\mbox{\commentbegin ~  definition of \ensuremath{\Varid{h_{State}}}   \commentend}{}\<[E]%
\\
\>[3]{}\hsindent{3}{}\<[6]%
\>[6]{}\Conid{StateT}\mathbin{\$}\lambda (\Varid{s}_{1},\Varid{s}_{2}){}\<[27]%
\>[27]{}\to {}\<[27E]%
\>[31]{}\alpha\mathbin{\langle\hspace{1.6pt}\mathclap{\raisebox{0.1pt}{\scalebox{1}{\$}}}\hspace{1.6pt}\rangle}\Varid{run_{StateT}}\;(\Varid{h_{State}}\;(\Varid{run_{StateT}}\;(\Varid{h_{State}}\;(\Conid{Op}\;(\Conid{Inl}\;(\Conid{Get}\;\Varid{k}))))\;\Varid{s}_{1}))\;\Varid{s}_{2}{}\<[E]%
\\
\>[3]{}\mathrel{=}\mbox{\commentbegin ~  definition of \ensuremath{\Varid{h_{States}^\prime}}   \commentend}{}\<[E]%
\\
\>[3]{}\hsindent{3}{}\<[6]%
\>[6]{}\Varid{h_{States}^\prime}\;(\Conid{Op}\;(\Conid{Inl}\;(\Conid{Get}\;\Varid{k}))){}\<[E]%
\ColumnHook
\end{hscode}\resethooks
\indentend \noindent \mbox{\underline{case \ensuremath{\Varid{t}\mathrel{=}\Conid{Op}\;(\Conid{Inl}\;(\Conid{Put}\;\Varid{s}\;\Varid{k}))}}}

Induction hypothesis: \ensuremath{\Varid{h_{States}^\prime}\;\Varid{k}\mathrel{=}(\Varid{h_{State}}\hsdot{\circ }{.}\Varid{states2state})\;\Varid{k}}.
\indentbegin \begin{hscode}\SaveRestoreHook
\column{B}{@{}>{\hspre}l<{\hspost}@{}}%
\column{3}{@{}>{\hspre}l<{\hspost}@{}}%
\column{6}{@{}>{\hspre}l<{\hspost}@{}}%
\column{8}{@{}>{\hspre}l<{\hspost}@{}}%
\column{9}{@{}>{\hspre}l<{\hspost}@{}}%
\column{27}{@{}>{\hspre}c<{\hspost}@{}}%
\column{27E}{@{}l@{}}%
\column{31}{@{}>{\hspre}l<{\hspost}@{}}%
\column{34}{@{}>{\hspre}l<{\hspost}@{}}%
\column{42}{@{}>{\hspre}c<{\hspost}@{}}%
\column{42E}{@{}l@{}}%
\column{46}{@{}>{\hspre}l<{\hspost}@{}}%
\column{E}{@{}>{\hspre}l<{\hspost}@{}}%
\>[6]{}(\Varid{h_{State}}\hsdot{\circ }{.}\Varid{states2state})\;(\Conid{Op}\;(\Conid{Inl}\;(\Conid{Put}\;\Varid{s}\;\Varid{k}))){}\<[E]%
\\
\>[3]{}\mathrel{=}\mbox{\commentbegin ~  definition of \ensuremath{\Varid{states2state}}   \commentend}{}\<[E]%
\\
\>[3]{}\hsindent{3}{}\<[6]%
\>[6]{}\Varid{h_{State}}\mathbin{\$}\Varid{get}>\!\!>\!\!=\lambda (\anonymous ,\Varid{s}_{2}){}\<[34]%
\>[34]{}\to \Varid{put}\;(\Varid{s},\Varid{s}_{2})>\!\!>(\Varid{states2state}\;\Varid{k}){}\<[E]%
\\
\>[3]{}\mathrel{=}\mbox{\commentbegin ~  definition of \ensuremath{\Varid{get}} and \ensuremath{\Varid{put}}   \commentend}{}\<[E]%
\\
\>[3]{}\hsindent{3}{}\<[6]%
\>[6]{}\Varid{h_{State}}\mathbin{\$}\Conid{Op}\;(\Conid{Inl}\;(\Conid{Get}\;\Varid{\eta}))>\!\!>\!\!=\lambda (\anonymous ,\Varid{s}_{2})\to \Conid{Op}\;(\Conid{Inl}\;(\Conid{Put}\;(\Varid{s},\Varid{s}_{2})\;(\Varid{\eta}\;())))>\!\!>(\Varid{states2state}\;\Varid{k}){}\<[E]%
\\
\>[3]{}\mathrel{=}\mbox{\commentbegin ~  definition of \ensuremath{(>\!\!>\!\!=)} and \ensuremath{(>\!\!>)}   \commentend}{}\<[E]%
\\
\>[3]{}\hsindent{3}{}\<[6]%
\>[6]{}\Varid{h_{State}}\mathbin{\$}\Conid{Op}\;(\Conid{Inl}\;(\Conid{Get}\;(\lambda (\anonymous ,\Varid{s}_{2})\to \Conid{Op}\;(\Conid{Inl}\;(\Conid{Put}\;(\Varid{s},\Varid{s}_{2})\;(\Varid{states2state}\;\Varid{k})))))){}\<[E]%
\\
\>[3]{}\mathrel{=}\mbox{\commentbegin ~  definition of \ensuremath{\Varid{h_{State}}}   \commentend}{}\<[E]%
\\
\>[3]{}\hsindent{3}{}\<[6]%
\>[6]{}\Varid{alg_{S}}\;(\Conid{Get}\;(\lambda (\anonymous ,\Varid{s}_{2})\to \Varid{h_{State}}\;(\Conid{Op}\;(\Conid{Inl}\;(\Conid{Put}\;(\Varid{s},\Varid{s}_{2})\;(\Varid{states2state}\;\Varid{k})))))){}\<[E]%
\\
\>[3]{}\mathrel{=}\mbox{\commentbegin ~  definition of \ensuremath{\Varid{alg_{S}}}   \commentend}{}\<[E]%
\\
\>[3]{}\hsindent{3}{}\<[6]%
\>[6]{}\Conid{StateT}\mathbin{\$}\lambda \Varid{s'}\to \Varid{run_{StateT}}{}\<[E]%
\\
\>[6]{}\hsindent{2}{}\<[8]%
\>[8]{}((\lambda (\anonymous ,\Varid{s}_{2})\to \Varid{h_{State}}\;(\Conid{Op}\;(\Conid{Inl}\;(\Conid{Put}\;(\Varid{s},\Varid{s}_{2})\;(\Varid{states2state}\;\Varid{k})))))\;\Varid{s'})\;\Varid{s'}{}\<[E]%
\\
\>[3]{}\mathrel{=}\mbox{\commentbegin ~  let \ensuremath{\Varid{s'}\mathrel{=}(\Varid{s}_{1},\Varid{s}_{2})}   \commentend}{}\<[E]%
\\
\>[3]{}\hsindent{3}{}\<[6]%
\>[6]{}\Conid{StateT}\mathbin{\$}\lambda (\Varid{s}_{1},\Varid{s}_{2})\to \Varid{run_{StateT}}{}\<[E]%
\\
\>[6]{}\hsindent{2}{}\<[8]%
\>[8]{}((\lambda (\anonymous ,\Varid{s}_{2})\to \Varid{h_{State}}\;(\Conid{Op}\;(\Conid{Inl}\;(\Conid{Put}\;(\Varid{s},\Varid{s}_{2})\;(\Varid{states2state}\;\Varid{k})))))\;(\Varid{s}_{1},\Varid{s}_{2}))\;(\Varid{s}_{1},\Varid{s}_{2}){}\<[E]%
\\
\>[3]{}\mathrel{=}\mbox{\commentbegin ~  function application   \commentend}{}\<[E]%
\\
\>[3]{}\hsindent{3}{}\<[6]%
\>[6]{}\Conid{StateT}\mathbin{\$}\lambda (\Varid{s}_{1},\Varid{s}_{2})\to \Varid{run_{StateT}}{}\<[E]%
\\
\>[6]{}\hsindent{2}{}\<[8]%
\>[8]{}(\Varid{h_{State}}\;(\Conid{Op}\;(\Conid{Inl}\;(\Conid{Put}\;(\Varid{s},\Varid{s}_{2})\;(\Varid{states2state}\;\Varid{k})))))\;(\Varid{s}_{1},\Varid{s}_{2}){}\<[E]%
\\
\>[3]{}\mathrel{=}\mbox{\commentbegin ~  definition of \ensuremath{\Varid{h_{State}}}   \commentend}{}\<[E]%
\\
\>[3]{}\hsindent{3}{}\<[6]%
\>[6]{}\Conid{StateT}\mathbin{\$}\lambda (\Varid{s}_{1},\Varid{s}_{2})\to \Varid{run_{StateT}}{}\<[E]%
\\
\>[6]{}\hsindent{2}{}\<[8]%
\>[8]{}(\Conid{StateT}\mathbin{\$}\lambda \anonymous \to \Varid{run_{StateT}}\;(\Varid{h_{State}}\;(\Varid{states2state}\;\Varid{k}))\;(\Varid{s},\Varid{s}_{2}))\;(\Varid{s}_{1},\Varid{s}_{2}){}\<[E]%
\\
\>[3]{}\mathrel{=}\mbox{\commentbegin ~  definition of \ensuremath{\Varid{run_{StateT}}}   \commentend}{}\<[E]%
\\
\>[3]{}\hsindent{3}{}\<[6]%
\>[6]{}\Conid{StateT}\mathbin{\$}\lambda (\Varid{s}_{1},\Varid{s}_{2})\to (\lambda \anonymous \to \Varid{run_{StateT}}\;(\Varid{h_{State}}\;(\Varid{states2state}\;\Varid{k}))\;(\Varid{s},\Varid{s}_{2}))\;(\Varid{s}_{1},\Varid{s}_{2}){}\<[E]%
\\
\>[3]{}\mathrel{=}\mbox{\commentbegin ~  function application   \commentend}{}\<[E]%
\\
\>[3]{}\hsindent{3}{}\<[6]%
\>[6]{}\Conid{StateT}\mathbin{\$}\lambda (\Varid{s}_{1},\Varid{s}_{2})\to \Varid{run_{StateT}}\;(\Varid{h_{State}}\;(\Varid{states2state}\;\Varid{k}))\;(\Varid{s},\Varid{s}_{2}){}\<[E]%
\\
\>[3]{}\mathrel{=}\mbox{\commentbegin ~  induction hypothesis   \commentend}{}\<[E]%
\\
\>[3]{}\hsindent{3}{}\<[6]%
\>[6]{}\Conid{StateT}\mathbin{\$}\lambda (\Varid{s}_{1},\Varid{s}_{2})\to \Varid{run_{StateT}}\;(\Varid{h_{States}^\prime}\;\Varid{k})\;(\Varid{s},\Varid{s}_{2}){}\<[E]%
\\
\>[3]{}\mathrel{=}\mbox{\commentbegin ~  definition of \ensuremath{\Varid{h_{States}^\prime}}  \commentend}{}\<[E]%
\\
\>[3]{}\hsindent{3}{}\<[6]%
\>[6]{}\Conid{StateT}\mathbin{\$}\lambda (\Varid{s}_{1},\Varid{s}_{2})\to \Varid{run_{StateT}}\;(\Conid{StateT}\mathbin{\$}\lambda (\Varid{s}_{1},\Varid{s}_{2})\to \alpha\mathbin{\langle\hspace{1.6pt}\mathclap{\raisebox{0.1pt}{\scalebox{1}{\$}}}\hspace{1.6pt}\rangle}{}\<[E]%
\\
\>[6]{}\hsindent{3}{}\<[9]%
\>[9]{}\Varid{run_{StateT}}\;(\Varid{h_{State}}\;(\Varid{run_{StateT}}\;(\Varid{h_{State}}\;\Varid{k})\;\Varid{s}_{1}))\;\Varid{s}_{2})\;(\Varid{s},\Varid{s}_{2}){}\<[E]%
\\
\>[3]{}\mathrel{=}\mbox{\commentbegin ~  definition of \ensuremath{\Varid{run_{StateT}}}  \commentend}{}\<[E]%
\\
\>[3]{}\hsindent{3}{}\<[6]%
\>[6]{}\Conid{StateT}\mathbin{\$}\lambda (\Varid{s}_{1},\Varid{s}_{2})\to (\lambda (\Varid{s}_{1},\Varid{s}_{2}){}\<[42]%
\>[42]{}\to {}\<[42E]%
\>[46]{}\alpha\mathbin{\langle\hspace{1.6pt}\mathclap{\raisebox{0.1pt}{\scalebox{1}{\$}}}\hspace{1.6pt}\rangle}{}\<[E]%
\\
\>[6]{}\hsindent{2}{}\<[8]%
\>[8]{}\Varid{run_{StateT}}\;(\Varid{h_{State}}\;(\Varid{run_{StateT}}\;(\Varid{h_{State}}\;\Varid{k})\;\Varid{s}_{1}))\;\Varid{s}_{2})\;(\Varid{s},\Varid{s}_{2}){}\<[E]%
\\
\>[3]{}\mathrel{=}\mbox{\commentbegin ~  function application  \commentend}{}\<[E]%
\\
\>[3]{}\hsindent{3}{}\<[6]%
\>[6]{}\Conid{StateT}\mathbin{\$}\lambda (\Varid{s}_{1},\Varid{s}_{2})\to \alpha\mathbin{\langle\hspace{1.6pt}\mathclap{\raisebox{0.1pt}{\scalebox{1}{\$}}}\hspace{1.6pt}\rangle}{}\<[E]%
\\
\>[6]{}\hsindent{2}{}\<[8]%
\>[8]{}\Varid{run_{StateT}}\;(\Varid{h_{State}}\;(\Varid{run_{StateT}}\;(\Varid{h_{State}}\;\Varid{k})\;\Varid{s}))\;\Varid{s}_{2}{}\<[E]%
\\
\>[3]{}\mathrel{=}\mbox{\commentbegin ~  \ensuremath{\Varid{\beta}}-expansion  \commentend}{}\<[E]%
\\
\>[3]{}\hsindent{3}{}\<[6]%
\>[6]{}\Conid{StateT}\mathbin{\$}\lambda (\Varid{s}_{1},\Varid{s}_{2}){}\<[27]%
\>[27]{}\to {}\<[27E]%
\>[31]{}\alpha\mathbin{\langle\hspace{1.6pt}\mathclap{\raisebox{0.1pt}{\scalebox{1}{\$}}}\hspace{1.6pt}\rangle}{}\<[E]%
\\
\>[6]{}\hsindent{2}{}\<[8]%
\>[8]{}\Varid{run_{StateT}}\;(\Varid{h_{State}}\;((\lambda \Varid{s'}\to \Varid{run_{StateT}}\;(\Varid{h_{State}}\;\Varid{k})\;\Varid{s})\;\Varid{s}_{1}))\;\Varid{s}_{2}{}\<[E]%
\\
\>[3]{}\mathrel{=}\mbox{\commentbegin ~  definition of \ensuremath{\Varid{run_{StateT}}}  \commentend}{}\<[E]%
\\
\>[3]{}\hsindent{3}{}\<[6]%
\>[6]{}\Conid{StateT}\mathbin{\$}\lambda (\Varid{s}_{1},\Varid{s}_{2}){}\<[27]%
\>[27]{}\to {}\<[27E]%
\>[31]{}\alpha\mathbin{\langle\hspace{1.6pt}\mathclap{\raisebox{0.1pt}{\scalebox{1}{\$}}}\hspace{1.6pt}\rangle}{}\<[E]%
\\
\>[6]{}\hsindent{2}{}\<[8]%
\>[8]{}\Varid{run_{StateT}}\;(\Varid{h_{State}}\;(\Varid{run_{StateT}}\;(\Conid{StateT}\mathbin{\$}\lambda \Varid{s'}\to \Varid{run_{StateT}}\;(\Varid{h_{State}}\;\Varid{k})\;\Varid{s})\;\Varid{s}_{1}))\;\Varid{s}_{2}{}\<[E]%
\\
\>[3]{}\mathrel{=}\mbox{\commentbegin ~  definition of \ensuremath{\Varid{h_{State}}}   \commentend}{}\<[E]%
\\
\>[3]{}\hsindent{3}{}\<[6]%
\>[6]{}\Conid{StateT}\mathbin{\$}\lambda (\Varid{s}_{1},\Varid{s}_{2}){}\<[27]%
\>[27]{}\to {}\<[27E]%
\>[31]{}\alpha\mathbin{\langle\hspace{1.6pt}\mathclap{\raisebox{0.1pt}{\scalebox{1}{\$}}}\hspace{1.6pt}\rangle}{}\<[E]%
\\
\>[6]{}\hsindent{2}{}\<[8]%
\>[8]{}\Varid{run_{StateT}}\;(\Varid{h_{State}}\;(\Varid{run_{StateT}}\;(\Varid{h_{State}}\;(\Conid{Op}\;(\Conid{Inl}\;(\Conid{Put}\;\Varid{s}\;\Varid{k}))))\;\Varid{s}_{1}))\;\Varid{s}_{2}{}\<[E]%
\\
\>[3]{}\mathrel{=}\mbox{\commentbegin ~  definition of \ensuremath{\Varid{h_{States}^\prime}}   \commentend}{}\<[E]%
\\
\>[3]{}\hsindent{3}{}\<[6]%
\>[6]{}\Varid{h_{States}^\prime}\;(\Conid{Op}\;(\Conid{Inl}\;(\Conid{Put}\;\Varid{s}\;\Varid{k}))){}\<[E]%
\ColumnHook
\end{hscode}\resethooks
\indentend \noindent \mbox{\underline{case \ensuremath{\Varid{t}\mathrel{=}\Conid{Op}\;(\Conid{Inr}\;(\Conid{Inl}\;(\Conid{Get}\;\Varid{k})))}}}

Induction hypothesis: \ensuremath{\Varid{h_{States}^\prime}\;(\Varid{k}\;\Varid{s})\mathrel{=}(\Varid{h_{State}}\hsdot{\circ }{.}\Varid{states2state})\;(\Varid{k}\;\Varid{s})}
for any \ensuremath{\Varid{s}}.
\indentbegin \begin{hscode}\SaveRestoreHook
\column{B}{@{}>{\hspre}l<{\hspost}@{}}%
\column{3}{@{}>{\hspre}l<{\hspost}@{}}%
\column{6}{@{}>{\hspre}l<{\hspost}@{}}%
\column{8}{@{}>{\hspre}l<{\hspost}@{}}%
\column{27}{@{}>{\hspre}c<{\hspost}@{}}%
\column{27E}{@{}l@{}}%
\column{30}{@{}>{\hspre}l<{\hspost}@{}}%
\column{31}{@{}>{\hspre}l<{\hspost}@{}}%
\column{42}{@{}>{\hspre}c<{\hspost}@{}}%
\column{42E}{@{}l@{}}%
\column{46}{@{}>{\hspre}l<{\hspost}@{}}%
\column{61}{@{}>{\hspre}c<{\hspost}@{}}%
\column{61E}{@{}l@{}}%
\column{65}{@{}>{\hspre}l<{\hspost}@{}}%
\column{E}{@{}>{\hspre}l<{\hspost}@{}}%
\>[6]{}(\Varid{h_{State}}\hsdot{\circ }{.}\Varid{states2state})\;(\Conid{Op}\;(\Conid{Inr}\;(\Conid{Inl}\;(\Conid{Get}\;\Varid{k})))){}\<[E]%
\\
\>[3]{}\mathrel{=}\mbox{\commentbegin ~  definition of \ensuremath{\Varid{states2state}}   \commentend}{}\<[E]%
\\
\>[3]{}\hsindent{3}{}\<[6]%
\>[6]{}\Varid{h_{State}}\mathbin{\$}\Varid{get}>\!\!>\!\!=\lambda (\anonymous ,{}\<[30]%
\>[30]{}\Varid{s}_{2})\to \Varid{states2state}\;(\Varid{k}\;\Varid{s}_{2}){}\<[E]%
\\
\>[3]{}\mathrel{=}\mbox{\commentbegin ~  definition of \ensuremath{\Varid{get}}   \commentend}{}\<[E]%
\\
\>[3]{}\hsindent{3}{}\<[6]%
\>[6]{}\Varid{h_{State}}\mathbin{\$}\Conid{Op}\;(\Conid{Inl}\;(\Conid{Get}\;\Varid{\eta}))>\!\!>\!\!=\lambda (\anonymous ,\Varid{s}_{2})\to \Varid{states2state}\;(\Varid{k}\;\Varid{s}_{2}){}\<[E]%
\\
\>[3]{}\mathrel{=}\mbox{\commentbegin ~  definition of \ensuremath{(>\!\!>\!\!=)} for free monad   \commentend}{}\<[E]%
\\
\>[3]{}\hsindent{3}{}\<[6]%
\>[6]{}\Varid{h_{State}}\;(\Conid{Op}\;(\Conid{Inl}\;(\Conid{Get}\;(\lambda (\anonymous ,\Varid{s}_{2})\to \Varid{states2state}\;(\Varid{k}\;\Varid{s}_{2}))))){}\<[E]%
\\
\>[3]{}\mathrel{=}\mbox{\commentbegin ~  definition of \ensuremath{\Varid{h_{State}}}   \commentend}{}\<[E]%
\\
\>[3]{}\hsindent{3}{}\<[6]%
\>[6]{}\Conid{StateT}\mathbin{\$}\lambda \Varid{s}\to \Varid{run_{StateT}}\;((\lambda (\anonymous ,\Varid{s}_{2})\to \Varid{h_{State}}\;(\Varid{states2state}\;(\Varid{k}\;\Varid{s}_{2})))\;\Varid{s})\;\Varid{s}{}\<[E]%
\\
\>[3]{}\mathrel{=}\mbox{\commentbegin ~  let \ensuremath{\Varid{s}\mathrel{=}(\Varid{s}_{1},\Varid{s}_{2})}   \commentend}{}\<[E]%
\\
\>[3]{}\hsindent{3}{}\<[6]%
\>[6]{}\Conid{StateT}\mathbin{\$}\lambda (\Varid{s}_{1},\Varid{s}_{2})\to \Varid{run_{StateT}}\;((\lambda (\anonymous ,\Varid{s}_{2})\to \Varid{h_{State}}\;(\Varid{states2state}\;(\Varid{k}\;\Varid{s}_{2})))\;(\Varid{s}_{1},\Varid{s}_{2}))\;(\Varid{s}_{1},\Varid{s}_{2}){}\<[E]%
\\
\>[3]{}\mathrel{=}\mbox{\commentbegin ~  function application   \commentend}{}\<[E]%
\\
\>[3]{}\hsindent{3}{}\<[6]%
\>[6]{}\Conid{StateT}\mathbin{\$}\lambda (\Varid{s}_{1},\Varid{s}_{2})\to \Varid{run_{StateT}}\;(\Varid{h_{State}}\;(\Varid{states2state}\;(\Varid{k}\;\Varid{s}_{2})))\;(\Varid{s}_{1},\Varid{s}_{2}){}\<[E]%
\\
\>[3]{}\mathrel{=}\mbox{\commentbegin ~  induction hypothesis   \commentend}{}\<[E]%
\\
\>[3]{}\hsindent{3}{}\<[6]%
\>[6]{}\Conid{StateT}\mathbin{\$}\lambda (\Varid{s}_{1},\Varid{s}_{2})\to \Varid{run_{StateT}}\;(\Varid{h_{States}^\prime}\;(\Varid{k}\;\Varid{s}_{2}))\;(\Varid{s}_{1},\Varid{s}_{2}){}\<[E]%
\\
\>[3]{}\mathrel{=}\mbox{\commentbegin ~  definition of \ensuremath{\Varid{h_{States}^\prime}}  \commentend}{}\<[E]%
\\
\>[3]{}\hsindent{3}{}\<[6]%
\>[6]{}\Conid{StateT}\mathbin{\$}\lambda (\Varid{s}_{1},\Varid{s}_{2})\to \Varid{run_{StateT}}\;(\Conid{StateT}\mathbin{\$}\lambda (\Varid{s}_{1},\Varid{s}_{2}){}\<[61]%
\>[61]{}\to {}\<[61E]%
\>[65]{}\alpha\mathbin{\langle\hspace{1.6pt}\mathclap{\raisebox{0.1pt}{\scalebox{1}{\$}}}\hspace{1.6pt}\rangle}{}\<[E]%
\\
\>[6]{}\hsindent{2}{}\<[8]%
\>[8]{}\Varid{run_{StateT}}\;(\Varid{h_{State}}\;(\Varid{run_{StateT}}\;(\Varid{h_{State}}\;(\Varid{k}\;\Varid{s}_{2}))\;\Varid{s}_{1}))\;\Varid{s}_{2})\;(\Varid{s}_{1},\Varid{s}_{2}){}\<[E]%
\\
\>[3]{}\mathrel{=}\mbox{\commentbegin ~  definition of \ensuremath{\Varid{run_{StateT}}}  \commentend}{}\<[E]%
\\
\>[3]{}\hsindent{3}{}\<[6]%
\>[6]{}\Conid{StateT}\mathbin{\$}\lambda (\Varid{s}_{1},\Varid{s}_{2})\to (\lambda (\Varid{s}_{1},\Varid{s}_{2}){}\<[42]%
\>[42]{}\to {}\<[42E]%
\>[46]{}\alpha\mathbin{\langle\hspace{1.6pt}\mathclap{\raisebox{0.1pt}{\scalebox{1}{\$}}}\hspace{1.6pt}\rangle}{}\<[E]%
\\
\>[6]{}\hsindent{2}{}\<[8]%
\>[8]{}\Varid{run_{StateT}}\;(\Varid{h_{State}}\;(\Varid{run_{StateT}}\;(\Varid{h_{State}}\;(\Varid{k}\;\Varid{s}_{2}))\;\Varid{s}_{1}))\;\Varid{s}_{2})\;(\Varid{s}_{1},\Varid{s}_{2}){}\<[E]%
\\
\>[3]{}\mathrel{=}\mbox{\commentbegin ~  function application  \commentend}{}\<[E]%
\\
\>[3]{}\hsindent{3}{}\<[6]%
\>[6]{}\Conid{StateT}\mathbin{\$}\lambda (\Varid{s}_{1},\Varid{s}_{2}){}\<[27]%
\>[27]{}\to {}\<[27E]%
\>[31]{}\alpha\mathbin{\langle\hspace{1.6pt}\mathclap{\raisebox{0.1pt}{\scalebox{1}{\$}}}\hspace{1.6pt}\rangle}{}\<[E]%
\\
\>[6]{}\hsindent{2}{}\<[8]%
\>[8]{}\Varid{run_{StateT}}\;(\Varid{h_{State}}\;(\Varid{run_{StateT}}\;(\Varid{h_{State}}\;(\Varid{k}\;\Varid{s}_{2}))\;\Varid{s}_{1}))\;\Varid{s}_{2}{}\<[E]%
\\
\>[3]{}\mathrel{=}\mbox{\commentbegin ~  reformulation   \commentend}{}\<[E]%
\\
\>[3]{}\hsindent{3}{}\<[6]%
\>[6]{}\Conid{StateT}\mathbin{\$}\lambda (\Varid{s}_{1},\Varid{s}_{2}){}\<[27]%
\>[27]{}\to {}\<[27E]%
\>[31]{}\alpha\mathbin{\langle\hspace{1.6pt}\mathclap{\raisebox{0.1pt}{\scalebox{1}{\$}}}\hspace{1.6pt}\rangle}{}\<[E]%
\\
\>[6]{}\hsindent{2}{}\<[8]%
\>[8]{}(\Varid{run_{StateT}}\;((\Varid{h_{State}}\hsdot{\circ }{.}(\lambda \Varid{k}\to \Varid{run_{StateT}}\;\Varid{k}\;\Varid{s}_{1})\hsdot{\circ }{.}\Varid{h_{State}}\hsdot{\circ }{.}\Varid{k})\;\Varid{s}_{2})\;\Varid{s}_{2}){}\<[E]%
\\
\>[3]{}\mathrel{=}\mbox{\commentbegin ~  \ensuremath{\Varid{\beta}}-expansion   \commentend}{}\<[E]%
\\
\>[3]{}\hsindent{3}{}\<[6]%
\>[6]{}\Conid{StateT}\mathbin{\$}\lambda (\Varid{s}_{1},\Varid{s}_{2}){}\<[27]%
\>[27]{}\to {}\<[27E]%
\>[31]{}\alpha\mathbin{\langle\hspace{1.6pt}\mathclap{\raisebox{0.1pt}{\scalebox{1}{\$}}}\hspace{1.6pt}\rangle}{}\<[E]%
\\
\>[6]{}\hsindent{2}{}\<[8]%
\>[8]{}(\lambda \Varid{s}\to \Varid{run_{StateT}}\;((\Varid{h_{State}}\hsdot{\circ }{.}(\lambda \Varid{k}\to \Varid{run_{StateT}}\;\Varid{k}\;\Varid{s}_{1})\hsdot{\circ }{.}\Varid{h_{State}}\hsdot{\circ }{.}\Varid{k})\;\Varid{s})\;\Varid{s})\;\Varid{s}_{2}{}\<[E]%
\\
\>[3]{}\mathrel{=}\mbox{\commentbegin ~  definition of \ensuremath{\Varid{run_{StateT}}}  \commentend}{}\<[E]%
\\
\>[3]{}\hsindent{3}{}\<[6]%
\>[6]{}\Conid{StateT}\mathbin{\$}\lambda (\Varid{s}_{1},\Varid{s}_{2}){}\<[27]%
\>[27]{}\to {}\<[27E]%
\>[31]{}\alpha\mathbin{\langle\hspace{1.6pt}\mathclap{\raisebox{0.1pt}{\scalebox{1}{\$}}}\hspace{1.6pt}\rangle}{}\<[E]%
\\
\>[6]{}\hsindent{2}{}\<[8]%
\>[8]{}\Varid{run_{StateT}}\;(\Conid{StateT}\mathbin{\$}\lambda \Varid{s}\to \Varid{run_{StateT}}\;((\Varid{h_{State}}\hsdot{\circ }{.}(\lambda \Varid{k}\to \Varid{run_{StateT}}\;\Varid{k}\;\Varid{s}_{1})\hsdot{\circ }{.}\Varid{h_{State}}\hsdot{\circ }{.}\Varid{k})\;\Varid{s})\;\Varid{s})\;\Varid{s}_{2}{}\<[E]%
\\
\>[3]{}\mathrel{=}\mbox{\commentbegin ~  definition of \ensuremath{\Varid{h_{State}}}   \commentend}{}\<[E]%
\\
\>[3]{}\hsindent{3}{}\<[6]%
\>[6]{}\Conid{StateT}\mathbin{\$}\lambda (\Varid{s}_{1},\Varid{s}_{2}){}\<[27]%
\>[27]{}\to {}\<[27E]%
\>[31]{}\alpha\mathbin{\langle\hspace{1.6pt}\mathclap{\raisebox{0.1pt}{\scalebox{1}{\$}}}\hspace{1.6pt}\rangle}{}\<[E]%
\\
\>[6]{}\hsindent{2}{}\<[8]%
\>[8]{}\Varid{run_{StateT}}\;(\Varid{h_{State}}\;(\Conid{Op}\;(\Conid{Inl}\;(\Conid{Get}\;((\lambda \Varid{k}\to \Varid{run_{StateT}}\;\Varid{k}\;\Varid{s}_{1})\hsdot{\circ }{.}\Varid{h_{State}}\hsdot{\circ }{.}\Varid{k})))))\;\Varid{s}_{2}{}\<[E]%
\\
\>[3]{}\mathrel{=}\mbox{\commentbegin ~  definition of \ensuremath{\Varid{fmap}}   \commentend}{}\<[E]%
\\
\>[3]{}\hsindent{3}{}\<[6]%
\>[6]{}\Conid{StateT}\mathbin{\$}\lambda (\Varid{s}_{1},\Varid{s}_{2}){}\<[27]%
\>[27]{}\to {}\<[27E]%
\>[31]{}\alpha\mathbin{\langle\hspace{1.6pt}\mathclap{\raisebox{0.1pt}{\scalebox{1}{\$}}}\hspace{1.6pt}\rangle}{}\<[E]%
\\
\>[6]{}\hsindent{2}{}\<[8]%
\>[8]{}\Varid{run_{StateT}}\;(\Varid{h_{State}}\;(\Conid{Op}\mathbin{\$}\Varid{fmap}\;(\lambda \Varid{k}\to \Varid{run_{StateT}}\;\Varid{k}\;\Varid{s}_{1})\;(\Conid{Inl}\;(\Conid{Get}\;(\Varid{h_{State}}\hsdot{\circ }{.}\Varid{k})))))\;\Varid{s}_{2}{}\<[E]%
\\
\>[3]{}\mathrel{=}\mbox{\commentbegin ~  \ensuremath{\Varid{\beta}}-expansion   \commentend}{}\<[E]%
\\
\>[3]{}\hsindent{3}{}\<[6]%
\>[6]{}\Conid{StateT}\mathbin{\$}\lambda (\Varid{s}_{1},\Varid{s}_{2}){}\<[27]%
\>[27]{}\to {}\<[27E]%
\>[31]{}\alpha\mathbin{\langle\hspace{1.6pt}\mathclap{\raisebox{0.1pt}{\scalebox{1}{\$}}}\hspace{1.6pt}\rangle}{}\<[E]%
\\
\>[6]{}\hsindent{2}{}\<[8]%
\>[8]{}\Varid{run_{StateT}}\;(\Varid{h_{State}}\;((\lambda \Varid{s}\to \Conid{Op}\mathbin{\$}\Varid{fmap}\;(\lambda \Varid{k}\to \Varid{run_{StateT}}\;\Varid{k}\;\Varid{s})\;(\Conid{Inl}\;(\Conid{Get}\;(\Varid{h_{State}}\hsdot{\circ }{.}\Varid{k}))))\;\Varid{s}_{1}))\;\Varid{s}_{2}{}\<[E]%
\\
\>[3]{}\mathrel{=}\mbox{\commentbegin ~  definition of \ensuremath{\Varid{run_{StateT}}}   \commentend}{}\<[E]%
\\
\>[3]{}\hsindent{3}{}\<[6]%
\>[6]{}\Conid{StateT}\mathbin{\$}\lambda (\Varid{s}_{1},\Varid{s}_{2}){}\<[27]%
\>[27]{}\to {}\<[27E]%
\>[31]{}\alpha\mathbin{\langle\hspace{1.6pt}\mathclap{\raisebox{0.1pt}{\scalebox{1}{\$}}}\hspace{1.6pt}\rangle}\Varid{run_{StateT}}\;(\Varid{h_{State}}{}\<[E]%
\\
\>[6]{}\hsindent{2}{}\<[8]%
\>[8]{}(\Varid{run_{StateT}}\;(\Conid{StateT}\mathbin{\$}\lambda \Varid{s}\to \Conid{Op}\mathbin{\$}\Varid{fmap}\;(\lambda \Varid{k}\to \Varid{run_{StateT}}\;\Varid{k}\;\Varid{s})\;(\Conid{Inl}\;(\Conid{Get}\;(\Varid{h_{State}}\hsdot{\circ }{.}\Varid{k}))))\;\Varid{s}_{1}))\;\Varid{s}_{2}{}\<[E]%
\\
\>[3]{}\mathrel{=}\mbox{\commentbegin ~  definition of \ensuremath{\Varid{h_{State}}}   \commentend}{}\<[E]%
\\
\>[3]{}\hsindent{3}{}\<[6]%
\>[6]{}\Conid{StateT}\mathbin{\$}\lambda (\Varid{s}_{1},\Varid{s}_{2}){}\<[27]%
\>[27]{}\to {}\<[27E]%
\>[31]{}\alpha\mathbin{\langle\hspace{1.6pt}\mathclap{\raisebox{0.1pt}{\scalebox{1}{\$}}}\hspace{1.6pt}\rangle}{}\<[E]%
\\
\>[6]{}\hsindent{2}{}\<[8]%
\>[8]{}\Varid{run_{StateT}}\;(\Varid{h_{State}}\;(\Varid{run_{StateT}}\;(\Varid{h_{State}}\;(\Conid{Op}\;(\Conid{Inr}\;(\Conid{Inl}\;(\Conid{Get}\;\Varid{k})))))\;\Varid{s}_{1}))\;\Varid{s}_{2}{}\<[E]%
\\
\>[3]{}\mathrel{=}\mbox{\commentbegin ~  definition of \ensuremath{\Varid{h_{States}^\prime}}   \commentend}{}\<[E]%
\\
\>[3]{}\hsindent{3}{}\<[6]%
\>[6]{}\Varid{h_{States}^\prime}\;(\Conid{Op}\;(\Conid{Inr}\;(\Conid{Inl}\;(\Conid{Get}\;\Varid{k})))){}\<[E]%
\ColumnHook
\end{hscode}\resethooks
\indentend \noindent \mbox{\underline{case \ensuremath{\Varid{t}\mathrel{=}\Conid{Op}\;(\Conid{Inr}\;(\Conid{Inl}\;(\Conid{Put}\;\Varid{s}\;\Varid{k})))}}}

Induction hypothesis: \ensuremath{\Varid{h_{States}^\prime}\;\Varid{k}\mathrel{=}(\Varid{h_{State}}\hsdot{\circ }{.}\Varid{states2state})\;\Varid{k}}.
\indentbegin \begin{hscode}\SaveRestoreHook
\column{B}{@{}>{\hspre}l<{\hspost}@{}}%
\column{3}{@{}>{\hspre}l<{\hspost}@{}}%
\column{6}{@{}>{\hspre}l<{\hspost}@{}}%
\column{8}{@{}>{\hspre}l<{\hspost}@{}}%
\column{27}{@{}>{\hspre}c<{\hspost}@{}}%
\column{27E}{@{}l@{}}%
\column{31}{@{}>{\hspre}l<{\hspost}@{}}%
\column{42}{@{}>{\hspre}c<{\hspost}@{}}%
\column{42E}{@{}l@{}}%
\column{46}{@{}>{\hspre}l<{\hspost}@{}}%
\column{E}{@{}>{\hspre}l<{\hspost}@{}}%
\>[6]{}(\Varid{h_{State}}\hsdot{\circ }{.}\Varid{states2state})\;(\Conid{Op}\;(\Conid{Inr}\;(\Conid{Inl}\;(\Conid{Put}\;\Varid{s}\;\Varid{k})))){}\<[E]%
\\
\>[3]{}\mathrel{=}\mbox{\commentbegin ~  definition of \ensuremath{\Varid{states2state}}   \commentend}{}\<[E]%
\\
\>[3]{}\hsindent{3}{}\<[6]%
\>[6]{}\Varid{h_{State}}\mathbin{\$}\Varid{get}>\!\!>\!\!=\lambda (\Varid{s}_{1},\anonymous )\to \Varid{put}\;(\Varid{s}_{1},\Varid{s})>\!\!>(\Varid{states2state}\;\Varid{k}){}\<[E]%
\\
\>[3]{}\mathrel{=}\mbox{\commentbegin ~  definition of \ensuremath{\Varid{get}} and \ensuremath{\Varid{put}}   \commentend}{}\<[E]%
\\
\>[3]{}\hsindent{3}{}\<[6]%
\>[6]{}\Varid{h_{State}}\mathbin{\$}\Conid{Op}\;(\Conid{Inl}\;(\Conid{Get}\;\Varid{\eta}))>\!\!>\!\!=\lambda (\Varid{s}_{1},\anonymous )\to \Conid{Op}\;(\Conid{Inl}\;(\Conid{Put}\;(\Varid{s}_{1},\Varid{s})\;(\Varid{\eta}\;())))>\!\!>(\Varid{states2state}\;\Varid{k}){}\<[E]%
\\
\>[3]{}\mathrel{=}\mbox{\commentbegin ~  definition of \ensuremath{(>\!\!>\!\!=)} and \ensuremath{(>\!\!>)}   \commentend}{}\<[E]%
\\
\>[3]{}\hsindent{3}{}\<[6]%
\>[6]{}\Varid{h_{State}}\mathbin{\$}\Conid{Op}\;(\Conid{Inl}\;(\Conid{Get}\;(\lambda (\Varid{s}_{1},\anonymous )\to \Conid{Op}\;(\Conid{Inl}\;(\Conid{Put}\;(\Varid{s}_{1},\Varid{s})\;(\Varid{states2state}\;\Varid{k})))))){}\<[E]%
\\
\>[3]{}\mathrel{=}\mbox{\commentbegin ~  definition of \ensuremath{\Varid{h_{State}}}   \commentend}{}\<[E]%
\\
\>[3]{}\hsindent{3}{}\<[6]%
\>[6]{}\Conid{StateT}\mathbin{\$}\lambda \Varid{s'}\to \Varid{run_{StateT}}{}\<[E]%
\\
\>[6]{}\hsindent{2}{}\<[8]%
\>[8]{}((\lambda (\Varid{s}_{1},\anonymous )\to \Varid{h_{State}}\;(\Conid{Op}\;(\Conid{Inl}\;(\Conid{Put}\;(\Varid{s}_{1},\Varid{s})\;(\Varid{states2state}\;\Varid{k})))))\;\Varid{s'})\;\Varid{s'}{}\<[E]%
\\
\>[3]{}\mathrel{=}\mbox{\commentbegin ~  let \ensuremath{\Varid{s'}\mathrel{=}(\Varid{s}_{1},\Varid{s}_{2})}   \commentend}{}\<[E]%
\\
\>[3]{}\hsindent{3}{}\<[6]%
\>[6]{}\Conid{StateT}\mathbin{\$}\lambda (\Varid{s}_{1},\Varid{s}_{2})\to \Varid{run_{StateT}}{}\<[E]%
\\
\>[6]{}\hsindent{2}{}\<[8]%
\>[8]{}((\lambda (\Varid{s}_{1},\anonymous )\to \Varid{h_{State}}\;(\Conid{Op}\;(\Conid{Inl}\;(\Conid{Put}\;(\Varid{s}_{1},\Varid{s})\;(\Varid{states2state}\;\Varid{k})))))\;(\Varid{s}_{1},\Varid{s}_{2}))\;(\Varid{s}_{1},\Varid{s}_{2}){}\<[E]%
\\
\>[3]{}\mathrel{=}\mbox{\commentbegin ~  function application   \commentend}{}\<[E]%
\\
\>[3]{}\hsindent{3}{}\<[6]%
\>[6]{}\Conid{StateT}\mathbin{\$}\lambda (\Varid{s}_{1},\Varid{s}_{2})\to \Varid{run_{StateT}}{}\<[E]%
\\
\>[6]{}\hsindent{2}{}\<[8]%
\>[8]{}(\Varid{h_{State}}\;(\Conid{Op}\;(\Conid{Inl}\;(\Conid{Put}\;(\Varid{s}_{1},\Varid{s})\;(\Varid{states2state}\;\Varid{k})))))\;(\Varid{s}_{1},\Varid{s}_{2}){}\<[E]%
\\
\>[3]{}\mathrel{=}\mbox{\commentbegin ~  definition of \ensuremath{\Varid{h_{State}}}   \commentend}{}\<[E]%
\\
\>[3]{}\hsindent{3}{}\<[6]%
\>[6]{}\Conid{StateT}\mathbin{\$}\lambda (\Varid{s}_{1},\Varid{s}_{2})\to \Varid{run_{StateT}}\;(\Conid{StateT}\mathbin{\$}{}\<[E]%
\\
\>[6]{}\hsindent{2}{}\<[8]%
\>[8]{}\lambda \Varid{s'}\to \Varid{run_{StateT}}\;(\Varid{h_{State}}\;(\Varid{states2state}\;\Varid{k}))\;(\Varid{s}_{1},\Varid{s}))\;(\Varid{s}_{1},\Varid{s}_{2}){}\<[E]%
\\
\>[3]{}\mathrel{=}\mbox{\commentbegin ~  definition of \ensuremath{\Varid{run_{StateT}}}   \commentend}{}\<[E]%
\\
\>[3]{}\hsindent{3}{}\<[6]%
\>[6]{}\Conid{StateT}\mathbin{\$}\lambda (\Varid{s}_{1},\Varid{s}_{2})\to (\lambda \Varid{s'}\to \Varid{run_{StateT}}\;(\Varid{h_{State}}\;(\Varid{states2state}\;\Varid{k}))\;(\Varid{s}_{1},\Varid{s}))\;(\Varid{s}_{1},\Varid{s}_{2}){}\<[E]%
\\
\>[3]{}\mathrel{=}\mbox{\commentbegin ~  function application   \commentend}{}\<[E]%
\\
\>[3]{}\hsindent{3}{}\<[6]%
\>[6]{}\Conid{StateT}\mathbin{\$}\lambda (\Varid{s}_{1},\Varid{s}_{2})\to \Varid{run_{StateT}}\;(\Varid{h_{State}}\;(\Varid{states2state}\;\Varid{k}))\;(\Varid{s}_{1},\Varid{s}){}\<[E]%
\\
\>[3]{}\mathrel{=}\mbox{\commentbegin ~  induction hypothesis   \commentend}{}\<[E]%
\\
\>[3]{}\hsindent{3}{}\<[6]%
\>[6]{}\Conid{StateT}\mathbin{\$}\lambda (\Varid{s}_{1},\Varid{s}_{2})\to \Varid{run_{StateT}}\;(\Varid{h_{States}^\prime}\;\Varid{k})\;(\Varid{s}_{1},\Varid{s}){}\<[E]%
\\
\>[3]{}\mathrel{=}\mbox{\commentbegin ~  definition of \ensuremath{\Varid{h_{States}^\prime}}  \commentend}{}\<[E]%
\\
\>[3]{}\hsindent{3}{}\<[6]%
\>[6]{}\Conid{StateT}\mathbin{\$}\lambda (\Varid{s}_{1},\Varid{s}_{2})\to \Varid{run_{StateT}}\;(\Conid{StateT}\mathbin{\$}\lambda (\Varid{s}_{1},\Varid{s}_{2})\to \alpha\mathbin{\langle\hspace{1.6pt}\mathclap{\raisebox{0.1pt}{\scalebox{1}{\$}}}\hspace{1.6pt}\rangle}{}\<[E]%
\\
\>[6]{}\hsindent{2}{}\<[8]%
\>[8]{}\Varid{run_{StateT}}\;(\Varid{h_{State}}\;(\Varid{run_{StateT}}\;(\Varid{h_{State}}\;\Varid{k})\;\Varid{s}_{1}))\;\Varid{s}_{2})\;(\Varid{s}_{1},\Varid{s}){}\<[E]%
\\
\>[3]{}\mathrel{=}\mbox{\commentbegin ~  definition of \ensuremath{\Varid{run_{StateT}}}  \commentend}{}\<[E]%
\\
\>[3]{}\hsindent{3}{}\<[6]%
\>[6]{}\Conid{StateT}\mathbin{\$}\lambda (\Varid{s}_{1},\Varid{s}_{2})\to (\lambda (\Varid{s}_{1},\Varid{s}_{2}){}\<[42]%
\>[42]{}\to {}\<[42E]%
\>[46]{}\alpha\mathbin{\langle\hspace{1.6pt}\mathclap{\raisebox{0.1pt}{\scalebox{1}{\$}}}\hspace{1.6pt}\rangle}{}\<[E]%
\\
\>[6]{}\hsindent{2}{}\<[8]%
\>[8]{}\Varid{run_{StateT}}\;(\Varid{h_{State}}\;(\Varid{run_{StateT}}\;(\Varid{h_{State}}\;\Varid{k})\;\Varid{s}_{1}))\;\Varid{s}_{2})\;(\Varid{s}_{1},\Varid{s}){}\<[E]%
\\
\>[3]{}\mathrel{=}\mbox{\commentbegin ~  function application  \commentend}{}\<[E]%
\\
\>[3]{}\hsindent{3}{}\<[6]%
\>[6]{}\Conid{StateT}\mathbin{\$}\lambda (\Varid{s}_{1},\Varid{s}_{2}){}\<[27]%
\>[27]{}\to {}\<[27E]%
\>[31]{}\alpha\mathbin{\langle\hspace{1.6pt}\mathclap{\raisebox{0.1pt}{\scalebox{1}{\$}}}\hspace{1.6pt}\rangle}{}\<[E]%
\\
\>[6]{}\hsindent{2}{}\<[8]%
\>[8]{}\Varid{run_{StateT}}\;(\Varid{h_{State}}\;(\Varid{run_{StateT}}\;(\Varid{h_{State}}\;\Varid{k})\;\Varid{s}_{1}))\;\Varid{s}{}\<[E]%
\\
\>[3]{}\mathrel{=}\mbox{\commentbegin ~  \ensuremath{\Varid{\beta}}-expansion  \commentend}{}\<[E]%
\\
\>[3]{}\hsindent{3}{}\<[6]%
\>[6]{}\Conid{StateT}\mathbin{\$}\lambda (\Varid{s}_{1},\Varid{s}_{2}){}\<[27]%
\>[27]{}\to {}\<[27E]%
\>[31]{}\alpha\mathbin{\langle\hspace{1.6pt}\mathclap{\raisebox{0.1pt}{\scalebox{1}{\$}}}\hspace{1.6pt}\rangle}{}\<[E]%
\\
\>[6]{}\hsindent{2}{}\<[8]%
\>[8]{}(\lambda \Varid{s'}\to \Varid{run_{StateT}}\;(\Varid{h_{State}}\;(\Varid{run_{StateT}}\;(\Varid{h_{State}}\;\Varid{k})\;\Varid{s}_{1}))\;\Varid{s})\;\Varid{s}_{2}{}\<[E]%
\\
\>[3]{}\mathrel{=}\mbox{\commentbegin ~  definition of \ensuremath{\Varid{run_{StateT}}}  \commentend}{}\<[E]%
\\
\>[3]{}\hsindent{3}{}\<[6]%
\>[6]{}\Conid{StateT}\mathbin{\$}\lambda (\Varid{s}_{1},\Varid{s}_{2}){}\<[27]%
\>[27]{}\to {}\<[27E]%
\>[31]{}\alpha\mathbin{\langle\hspace{1.6pt}\mathclap{\raisebox{0.1pt}{\scalebox{1}{\$}}}\hspace{1.6pt}\rangle}{}\<[E]%
\\
\>[6]{}\hsindent{2}{}\<[8]%
\>[8]{}\Varid{run_{StateT}}\;(\Conid{StateT}\mathbin{\$}\lambda \Varid{s'}\to \Varid{run_{StateT}}\;(\Varid{h_{State}}\;(\Varid{run_{StateT}}\;(\Varid{h_{State}}\;\Varid{k})\;\Varid{s}_{1}))\;\Varid{s})\;\Varid{s}_{2}{}\<[E]%
\\
\>[3]{}\mathrel{=}\mbox{\commentbegin ~  definition of \ensuremath{\Varid{h_{State}}}  \commentend}{}\<[E]%
\\
\>[3]{}\hsindent{3}{}\<[6]%
\>[6]{}\Conid{StateT}\mathbin{\$}\lambda (\Varid{s}_{1},\Varid{s}_{2}){}\<[27]%
\>[27]{}\to {}\<[27E]%
\>[31]{}\alpha\mathbin{\langle\hspace{1.6pt}\mathclap{\raisebox{0.1pt}{\scalebox{1}{\$}}}\hspace{1.6pt}\rangle}{}\<[E]%
\\
\>[6]{}\hsindent{2}{}\<[8]%
\>[8]{}\Varid{run_{StateT}}\;(\Varid{h_{State}}\;(\Conid{Op}\;(\Conid{Inl}\;(\Conid{Put}\;\Varid{s}\;(\Varid{run_{StateT}}\;(\Varid{h_{State}}\;\Varid{k})\;\Varid{s}_{1})))))\;\Varid{s}_{2}{}\<[E]%
\\
\>[3]{}\mathrel{=}\mbox{\commentbegin ~  reformulation  \commentend}{}\<[E]%
\\
\>[3]{}\hsindent{3}{}\<[6]%
\>[6]{}\Conid{StateT}\mathbin{\$}\lambda (\Varid{s}_{1},\Varid{s}_{2}){}\<[27]%
\>[27]{}\to {}\<[27E]%
\>[31]{}\alpha\mathbin{\langle\hspace{1.6pt}\mathclap{\raisebox{0.1pt}{\scalebox{1}{\$}}}\hspace{1.6pt}\rangle}{}\<[E]%
\\
\>[6]{}\hsindent{2}{}\<[8]%
\>[8]{}\Varid{run_{StateT}}\;(\Varid{h_{State}}\;(\Conid{Op}\;(\Conid{Inl}\;(\Conid{Put}\;\Varid{s}\;((\lambda \Varid{k}\to \Varid{run_{StateT}}\;\Varid{k}\;\Varid{s}_{1})\;(\Varid{h_{State}}\;\Varid{k}))))))\;\Varid{s}_{2}{}\<[E]%
\\
\>[3]{}\mathrel{=}\mbox{\commentbegin ~  definition of \ensuremath{\Varid{fmap}}  \commentend}{}\<[E]%
\\
\>[3]{}\hsindent{3}{}\<[6]%
\>[6]{}\Conid{StateT}\mathbin{\$}\lambda (\Varid{s}_{1},\Varid{s}_{2}){}\<[27]%
\>[27]{}\to {}\<[27E]%
\>[31]{}\alpha\mathbin{\langle\hspace{1.6pt}\mathclap{\raisebox{0.1pt}{\scalebox{1}{\$}}}\hspace{1.6pt}\rangle}{}\<[E]%
\\
\>[6]{}\hsindent{2}{}\<[8]%
\>[8]{}\Varid{run_{StateT}}\;(\Varid{h_{State}}\;(\Conid{Op}\mathbin{\$}\Varid{fmap}\;(\lambda \Varid{k}\to \Varid{run_{StateT}}\;\Varid{k}\;\Varid{s}_{1})\;(\Conid{Inl}\;(\Conid{Put}\;\Varid{s}\;(\Varid{h_{State}}\;\Varid{k})))))\;\Varid{s}_{2}{}\<[E]%
\\
\>[3]{}\mathrel{=}\mbox{\commentbegin ~  \ensuremath{\Varid{\beta}}-expansion  \commentend}{}\<[E]%
\\
\>[3]{}\hsindent{3}{}\<[6]%
\>[6]{}\Conid{StateT}\mathbin{\$}\lambda (\Varid{s}_{1},\Varid{s}_{2}){}\<[27]%
\>[27]{}\to {}\<[27E]%
\>[31]{}\alpha\mathbin{\langle\hspace{1.6pt}\mathclap{\raisebox{0.1pt}{\scalebox{1}{\$}}}\hspace{1.6pt}\rangle}\Varid{run_{StateT}}{}\<[E]%
\\
\>[6]{}\hsindent{2}{}\<[8]%
\>[8]{}(\Varid{h_{State}}\;((\lambda \Varid{s'}\to \Conid{Op}\mathbin{\$}\Varid{fmap}\;(\lambda \Varid{k}\to \Varid{run_{StateT}}\;\Varid{k}\;\Varid{s'})\;(\Conid{Inl}\;(\Conid{Put}\;\Varid{s}\;(\Varid{h_{State}}\;\Varid{k}))))\;\Varid{s}_{1}))\;\Varid{s}_{2}{}\<[E]%
\\
\>[3]{}\mathrel{=}\mbox{\commentbegin ~  definition of \ensuremath{\Varid{run_{StateT}}}   \commentend}{}\<[E]%
\\
\>[3]{}\hsindent{3}{}\<[6]%
\>[6]{}\Conid{StateT}\mathbin{\$}\lambda (\Varid{s}_{1},\Varid{s}_{2}){}\<[27]%
\>[27]{}\to {}\<[27E]%
\>[31]{}\alpha\mathbin{\langle\hspace{1.6pt}\mathclap{\raisebox{0.1pt}{\scalebox{1}{\$}}}\hspace{1.6pt}\rangle}\Varid{run_{StateT}}\;(\Varid{h_{State}}\;(\Varid{run_{StateT}}\;(\Conid{StateT}\mathbin{\$}{}\<[E]%
\\
\>[6]{}\hsindent{2}{}\<[8]%
\>[8]{}\lambda \Varid{s'}\to \Conid{Op}\mathbin{\$}\Varid{fmap}\;(\lambda \Varid{k}\to \Varid{run_{StateT}}\;\Varid{k}\;\Varid{s'})\;(\Conid{Inl}\;(\Conid{Put}\;\Varid{s}\;(\Varid{h_{State}}\;\Varid{k}))))\;\Varid{s}_{1}))\;\Varid{s}_{2}{}\<[E]%
\\
\>[3]{}\mathrel{=}\mbox{\commentbegin ~  definition of \ensuremath{\Varid{h_{State}}}   \commentend}{}\<[E]%
\\
\>[3]{}\hsindent{3}{}\<[6]%
\>[6]{}\Conid{StateT}\mathbin{\$}\lambda (\Varid{s}_{1},\Varid{s}_{2}){}\<[27]%
\>[27]{}\to {}\<[27E]%
\>[31]{}\alpha\mathbin{\langle\hspace{1.6pt}\mathclap{\raisebox{0.1pt}{\scalebox{1}{\$}}}\hspace{1.6pt}\rangle}{}\<[E]%
\\
\>[6]{}\hsindent{2}{}\<[8]%
\>[8]{}\Varid{run_{StateT}}\;(\Varid{h_{State}}\;(\Varid{run_{StateT}}\;(\Varid{h_{State}}\;(\Conid{Op}\;(\Conid{Inr}\;(\Conid{Inl}\;(\Conid{Put}\;\Varid{s}\;\Varid{k})))))\;\Varid{s}_{1}))\;\Varid{s}_{2}{}\<[E]%
\\
\>[3]{}\mathrel{=}\mbox{\commentbegin ~  definition of \ensuremath{\Varid{h_{States}^\prime}}   \commentend}{}\<[E]%
\\
\>[3]{}\hsindent{3}{}\<[6]%
\>[6]{}\Varid{h_{States}^\prime}\;(\Conid{Op}\;(\Conid{Inr}\;(\Conid{Inl}\;(\Conid{Put}\;\Varid{s}\;\Varid{k})))){}\<[E]%
\ColumnHook
\end{hscode}\resethooks
\indentend \noindent \mbox{\underline{case \ensuremath{\Varid{t}\mathrel{=}\Conid{Op}\;(\Conid{Inr}\;(\Conid{Inr}\;\Varid{y}))}}}

Induction hypothesis: \ensuremath{\Varid{h_{States}^\prime}\;\Varid{y}\mathrel{=}(\Varid{h_{State}}\hsdot{\circ }{.}\Varid{states2state})\;\Varid{y}}.
\indentbegin \begin{hscode}\SaveRestoreHook
\column{B}{@{}>{\hspre}l<{\hspost}@{}}%
\column{3}{@{}>{\hspre}l<{\hspost}@{}}%
\column{6}{@{}>{\hspre}l<{\hspost}@{}}%
\column{8}{@{}>{\hspre}l<{\hspost}@{}}%
\column{9}{@{}>{\hspre}l<{\hspost}@{}}%
\column{22}{@{}>{\hspre}c<{\hspost}@{}}%
\column{22E}{@{}l@{}}%
\column{26}{@{}>{\hspre}l<{\hspost}@{}}%
\column{27}{@{}>{\hspre}c<{\hspost}@{}}%
\column{27E}{@{}l@{}}%
\column{31}{@{}>{\hspre}l<{\hspost}@{}}%
\column{E}{@{}>{\hspre}l<{\hspost}@{}}%
\>[6]{}(\Varid{h_{State}}\hsdot{\circ }{.}\Varid{states2state})\;(\Conid{Op}\;(\Conid{Inr}\;(\Conid{Inr}\;\Varid{y}))){}\<[E]%
\\
\>[3]{}\mathrel{=}\mbox{\commentbegin ~  definition of \ensuremath{\Varid{states2state}}   \commentend}{}\<[E]%
\\
\>[3]{}\hsindent{3}{}\<[6]%
\>[6]{}\Varid{h_{State}}\mathbin{\$}\Conid{Op}\;(\Conid{Inr}\;(\Varid{fmap}\;\Varid{states2state}\;\Varid{y})){}\<[E]%
\\
\>[3]{}\mathrel{=}\mbox{\commentbegin ~  definition of \ensuremath{\Varid{h_{State}}}   \commentend}{}\<[E]%
\\
\>[3]{}\hsindent{3}{}\<[6]%
\>[6]{}\Varid{fwd_{S}}\;(\Varid{fmap}\;(\Varid{h_{State}}\hsdot{\circ }{.}\Varid{states2state})\;\Varid{y}){}\<[E]%
\\
\>[3]{}\mathrel{=}\mbox{\commentbegin ~  induction hypothesis   \commentend}{}\<[E]%
\\
\>[3]{}\hsindent{3}{}\<[6]%
\>[6]{}\Varid{fwd_{S}}\;(\Varid{fmap}\;\Varid{h_{States}^\prime}\;\Varid{y}){}\<[E]%
\\
\>[3]{}\mathrel{=}\mbox{\commentbegin ~  definition of \ensuremath{\Varid{h_{States}^\prime}}   \commentend}{}\<[E]%
\\
\>[3]{}\hsindent{3}{}\<[6]%
\>[6]{}\Varid{fwd_{S}}\;(\Varid{fmap}\;(\lambda \Varid{t}\to \Conid{StateT}{}\<[E]%
\\
\>[6]{}\hsindent{2}{}\<[8]%
\>[8]{}\mathbin{\$}\lambda (\Varid{s}_{1},\Varid{s}_{2}){}\<[22]%
\>[22]{}\to {}\<[22E]%
\>[26]{}\alpha\mathbin{\langle\hspace{1.6pt}\mathclap{\raisebox{0.1pt}{\scalebox{1}{\$}}}\hspace{1.6pt}\rangle}\Varid{run_{StateT}}\;(\Varid{h_{State}}\;(\Varid{run_{StateT}}\;(\Varid{h_{State}}\;\Varid{t})\;\Varid{s}_{1}))\;\Varid{s}_{2})\;\Varid{y}){}\<[E]%
\\
\>[3]{}\mathrel{=}\mbox{\commentbegin ~  definition of \ensuremath{\Varid{fwd_{S}}}   \commentend}{}\<[E]%
\\
\>[3]{}\hsindent{3}{}\<[6]%
\>[6]{}\Conid{StateT}\mathbin{\$}\lambda \Varid{s}\to \Conid{Op}\mathbin{\$}\Varid{fmap}\;(\lambda \Varid{k}\to \Varid{run_{StateT}}\;\Varid{k}\;\Varid{s})\;(\Varid{fmap}\;(\lambda \Varid{t}\to \Conid{StateT}{}\<[E]%
\\
\>[6]{}\hsindent{2}{}\<[8]%
\>[8]{}\mathbin{\$}\lambda (\Varid{s}_{1},\Varid{s}_{2})\to \alpha\mathbin{\langle\hspace{1.6pt}\mathclap{\raisebox{0.1pt}{\scalebox{1}{\$}}}\hspace{1.6pt}\rangle}\Varid{run_{StateT}}\;(\Varid{h_{State}}\;(\Varid{run_{StateT}}\;(\Varid{h_{State}}\;\Varid{t})\;\Varid{s}_{1}))\;\Varid{s}_{2})\;\Varid{y}){}\<[E]%
\\
\>[3]{}\mathrel{=}\mbox{\commentbegin ~  \Cref{eq:functor-composition}   \commentend}{}\<[E]%
\\
\>[3]{}\hsindent{3}{}\<[6]%
\>[6]{}\Conid{StateT}\mathbin{\$}\lambda \Varid{s}\to \Conid{Op}\;(\Varid{fmap}\;((\lambda \Varid{k}\to \Varid{run_{StateT}}\;\Varid{k}\;\Varid{s})\hsdot{\circ }{.}(\lambda \Varid{t}\to \Conid{StateT}{}\<[E]%
\\
\>[6]{}\hsindent{2}{}\<[8]%
\>[8]{}\mathbin{\$}\lambda (\Varid{s}_{1},\Varid{s}_{2})\to \alpha\mathbin{\langle\hspace{1.6pt}\mathclap{\raisebox{0.1pt}{\scalebox{1}{\$}}}\hspace{1.6pt}\rangle}\Varid{run_{StateT}}\;(\Varid{h_{State}}\;(\Varid{run_{StateT}}\;(\Varid{h_{State}}\;\Varid{t})\;\Varid{s}_{1}))\;\Varid{s}_{2}))\;\Varid{y}){}\<[E]%
\\
\>[3]{}\mathrel{=}\mbox{\commentbegin ~  reformulation   \commentend}{}\<[E]%
\\
\>[3]{}\hsindent{3}{}\<[6]%
\>[6]{}\Conid{StateT}\mathbin{\$}\lambda \Varid{s}\to \Conid{Op}\;(\Varid{fmap}\;(\lambda \Varid{t}\to \Varid{run_{StateT}}\;(\Conid{StateT}{}\<[E]%
\\
\>[6]{}\hsindent{2}{}\<[8]%
\>[8]{}\mathbin{\$}\lambda (\Varid{s}_{1},\Varid{s}_{2})\to \alpha\mathbin{\langle\hspace{1.6pt}\mathclap{\raisebox{0.1pt}{\scalebox{1}{\$}}}\hspace{1.6pt}\rangle}\Varid{run_{StateT}}\;(\Varid{h_{State}}\;(\Varid{run_{StateT}}\;(\Varid{h_{State}}\;\Varid{t})\;\Varid{s}_{1}))\;\Varid{s}_{2})\;\Varid{s})\;\Varid{y}){}\<[E]%
\\
\>[3]{}\mathrel{=}\mbox{\commentbegin ~  definition of \ensuremath{\Varid{run_{StateT}}}   \commentend}{}\<[E]%
\\
\>[3]{}\hsindent{3}{}\<[6]%
\>[6]{}\Conid{StateT}\mathbin{\$}\lambda \Varid{s}\to \Conid{Op}\;(\Varid{fmap}\;(\lambda \Varid{t}\to {}\<[E]%
\\
\>[6]{}\hsindent{2}{}\<[8]%
\>[8]{}(\lambda (\Varid{s}_{1},\Varid{s}_{2})\to \alpha\mathbin{\langle\hspace{1.6pt}\mathclap{\raisebox{0.1pt}{\scalebox{1}{\$}}}\hspace{1.6pt}\rangle}\Varid{run_{StateT}}\;(\Varid{h_{State}}\;(\Varid{run_{StateT}}\;(\Varid{h_{State}}\;\Varid{t})\;\Varid{s}_{1}))\;\Varid{s}_{2})\;\Varid{s})\;\Varid{y}){}\<[E]%
\\
\>[3]{}\mathrel{=}\mbox{\commentbegin ~  let \ensuremath{\Varid{s}\mathrel{=}(\Varid{s}_{1},\Varid{s}_{2})}   \commentend}{}\<[E]%
\\
\>[3]{}\hsindent{3}{}\<[6]%
\>[6]{}\Conid{StateT}\mathbin{\$}\lambda (\Varid{s}_{1},\Varid{s}_{2})\to \Conid{Op}\;(\Varid{fmap}\;(\lambda \Varid{t}\to {}\<[E]%
\\
\>[6]{}\hsindent{2}{}\<[8]%
\>[8]{}(\lambda (\Varid{s}_{1},\Varid{s}_{2})\to \alpha\mathbin{\langle\hspace{1.6pt}\mathclap{\raisebox{0.1pt}{\scalebox{1}{\$}}}\hspace{1.6pt}\rangle}\Varid{run_{StateT}}\;(\Varid{h_{State}}\;(\Varid{run_{StateT}}\;(\Varid{h_{State}}\;\Varid{t})\;\Varid{s}_{1}))\;\Varid{s}_{2})\;(\Varid{s}_{1},\Varid{s}_{2}))\;\Varid{y}){}\<[E]%
\\
\>[3]{}\mathrel{=}\mbox{\commentbegin ~  function application   \commentend}{}\<[E]%
\\
\>[3]{}\hsindent{3}{}\<[6]%
\>[6]{}\Conid{StateT}\mathbin{\$}\lambda (\Varid{s}_{1},\Varid{s}_{2})\to \Conid{Op}\;(\Varid{fmap}\;(\lambda \Varid{t}\to \alpha\mathbin{\langle\hspace{1.6pt}\mathclap{\raisebox{0.1pt}{\scalebox{1}{\$}}}\hspace{1.6pt}\rangle}{}\<[E]%
\\
\>[6]{}\hsindent{3}{}\<[9]%
\>[9]{}\Varid{run_{StateT}}\;(\Varid{h_{State}}\;(\Varid{run_{StateT}}\;(\Varid{h_{State}}\;\Varid{t})\;\Varid{s}_{1}))\;\Varid{s}_{2})\;\Varid{y}){}\<[E]%
\\
\>[3]{}\mathrel{=}\mbox{\commentbegin ~  reformulation   \commentend}{}\<[E]%
\\
\>[3]{}\hsindent{3}{}\<[6]%
\>[6]{}\Conid{StateT}\mathbin{\$}\lambda (\Varid{s}_{1},\Varid{s}_{2})\to \Conid{Op}\;(\Varid{fmap}\;(\alpha\mathbin{\langle\hspace{1.6pt}\mathclap{\raisebox{0.1pt}{\scalebox{1}{\$}}}\hspace{1.6pt}\rangle}{}\<[E]%
\\
\>[6]{}\hsindent{2}{}\<[8]%
\>[8]{}\hsdot{\circ }{.}(\lambda \Varid{k}\to \Varid{run_{StateT}}\;\Varid{k}\;\Varid{s}_{2})\hsdot{\circ }{.}\Varid{h_{State}}\hsdot{\circ }{.}(\lambda \Varid{k}\to \Varid{run_{StateT}}\;\Varid{k}\;\Varid{s}_{1})\hsdot{\circ }{.}\Varid{h_{State}})\;\Varid{y}){}\<[E]%
\\
\>[3]{}\mathrel{=}\mbox{\commentbegin ~  definition of \ensuremath{\mathbin{\langle\hspace{1.6pt}\mathclap{\raisebox{0.1pt}{\scalebox{1}{\$}}}\hspace{1.6pt}\rangle}}   \commentend}{}\<[E]%
\\
\>[3]{}\hsindent{3}{}\<[6]%
\>[6]{}\Conid{StateT}\mathbin{\$}\lambda (\Varid{s}_{1},\Varid{s}_{2}){}\<[27]%
\>[27]{}\to {}\<[27E]%
\>[31]{}\alpha\mathbin{\langle\hspace{1.6pt}\mathclap{\raisebox{0.1pt}{\scalebox{1}{\$}}}\hspace{1.6pt}\rangle}{}\<[E]%
\\
\>[6]{}\hsindent{2}{}\<[8]%
\>[8]{}\Conid{Op}\;(\Varid{fmap}\;((\lambda \Varid{k}\to \Varid{run_{StateT}}\;\Varid{k}\;\Varid{s}_{2})\hsdot{\circ }{.}\Varid{h_{State}}\hsdot{\circ }{.}(\lambda \Varid{k}\to \Varid{run_{StateT}}\;\Varid{k}\;\Varid{s}_{1})\hsdot{\circ }{.}\Varid{h_{State}})\;\Varid{y}){}\<[E]%
\\
\>[3]{}\mathrel{=}\mbox{\commentbegin ~  \Cref{eq:functor-composition}   \commentend}{}\<[E]%
\\
\>[3]{}\hsindent{3}{}\<[6]%
\>[6]{}\Conid{StateT}\mathbin{\$}\lambda (\Varid{s}_{1},\Varid{s}_{2}){}\<[27]%
\>[27]{}\to {}\<[27E]%
\>[31]{}\alpha\mathbin{\langle\hspace{1.6pt}\mathclap{\raisebox{0.1pt}{\scalebox{1}{\$}}}\hspace{1.6pt}\rangle}{}\<[E]%
\\
\>[6]{}\hsindent{2}{}\<[8]%
\>[8]{}\Conid{Op}\mathbin{\$}\Varid{fmap}\;(\lambda \Varid{k}\to \Varid{run_{StateT}}\;\Varid{k}\;\Varid{s}_{2})\;(\Varid{fmap}\;(\Varid{h_{State}}\hsdot{\circ }{.}(\lambda \Varid{k}\to \Varid{run_{StateT}}\;\Varid{k}\;\Varid{s}_{1})\hsdot{\circ }{.}\Varid{h_{State}})\;\Varid{y}){}\<[E]%
\\
\>[3]{}\mathrel{=}\mbox{\commentbegin ~  \ensuremath{\Varid{\beta}}-expansion   \commentend}{}\<[E]%
\\
\>[3]{}\hsindent{3}{}\<[6]%
\>[6]{}\Conid{StateT}\mathbin{\$}\lambda (\Varid{s}_{1},\Varid{s}_{2}){}\<[27]%
\>[27]{}\to {}\<[27E]%
\>[31]{}\alpha\mathbin{\langle\hspace{1.6pt}\mathclap{\raisebox{0.1pt}{\scalebox{1}{\$}}}\hspace{1.6pt}\rangle}(\lambda \Varid{s}\to \Conid{Op}\mathbin{\$}{}\<[E]%
\\
\>[6]{}\hsindent{2}{}\<[8]%
\>[8]{}\Varid{fmap}\;(\lambda \Varid{k}\to \Varid{run_{StateT}}\;\Varid{k}\;\Varid{s})\;(\Varid{fmap}\;(\Varid{h_{State}}\hsdot{\circ }{.}(\lambda \Varid{k}\to \Varid{run_{StateT}}\;\Varid{k}\;\Varid{s}_{1})\hsdot{\circ }{.}\Varid{h_{State}})\;\Varid{y}))\;\Varid{s}_{2}{}\<[E]%
\\
\>[3]{}\mathrel{=}\mbox{\commentbegin ~  definition of \ensuremath{\Varid{run_{StateT}}}   \commentend}{}\<[E]%
\\
\>[3]{}\hsindent{3}{}\<[6]%
\>[6]{}\Conid{StateT}\mathbin{\$}\lambda (\Varid{s}_{1},\Varid{s}_{2}){}\<[27]%
\>[27]{}\to {}\<[27E]%
\>[31]{}\alpha\mathbin{\langle\hspace{1.6pt}\mathclap{\raisebox{0.1pt}{\scalebox{1}{\$}}}\hspace{1.6pt}\rangle}\Varid{run_{StateT}}\;(\Conid{StateT}\mathbin{\$}\lambda \Varid{s}\to \Conid{Op}\mathbin{\$}{}\<[E]%
\\
\>[6]{}\hsindent{2}{}\<[8]%
\>[8]{}\Varid{fmap}\;(\lambda \Varid{k}\to \Varid{run_{StateT}}\;\Varid{k}\;\Varid{s})\;(\Varid{fmap}\;(\Varid{h_{State}}\hsdot{\circ }{.}(\lambda \Varid{k}\to \Varid{run_{StateT}}\;\Varid{k}\;\Varid{s}_{1})\hsdot{\circ }{.}\Varid{h_{State}})\;\Varid{y}))\;\Varid{s}_{2}{}\<[E]%
\\
\>[3]{}\mathrel{=}\mbox{\commentbegin ~  definition of \ensuremath{\Varid{h_{State}}}   \commentend}{}\<[E]%
\\
\>[3]{}\hsindent{3}{}\<[6]%
\>[6]{}\Conid{StateT}\mathbin{\$}\lambda (\Varid{s}_{1},\Varid{s}_{2}){}\<[27]%
\>[27]{}\to {}\<[27E]%
\>[31]{}\alpha\mathbin{\langle\hspace{1.6pt}\mathclap{\raisebox{0.1pt}{\scalebox{1}{\$}}}\hspace{1.6pt}\rangle}\Varid{run_{StateT}}{}\<[E]%
\\
\>[6]{}\hsindent{2}{}\<[8]%
\>[8]{}(\Varid{h_{State}}\;(\Conid{Op}\;(\Conid{Inr}\;(\Varid{fmap}\;((\lambda \Varid{k}\to \Varid{run_{StateT}}\;\Varid{k}\;\Varid{s}_{1})\hsdot{\circ }{.}\Varid{h_{State}})\;\Varid{y}))))\;\Varid{s}_{2}{}\<[E]%
\\
\>[3]{}\mathrel{=}\mbox{\commentbegin ~  \Cref{eq:functor-composition}   \commentend}{}\<[E]%
\\
\>[3]{}\hsindent{3}{}\<[6]%
\>[6]{}\Conid{StateT}\mathbin{\$}\lambda (\Varid{s}_{1},\Varid{s}_{2}){}\<[27]%
\>[27]{}\to {}\<[27E]%
\>[31]{}\alpha\mathbin{\langle\hspace{1.6pt}\mathclap{\raisebox{0.1pt}{\scalebox{1}{\$}}}\hspace{1.6pt}\rangle}\Varid{run_{StateT}}{}\<[E]%
\\
\>[6]{}\hsindent{2}{}\<[8]%
\>[8]{}(\Varid{h_{State}}\;(\Conid{Op}\;(\Conid{Inr}\;(\Varid{fmap}\;(\lambda \Varid{k}\to \Varid{run_{StateT}}\;\Varid{k}\;\Varid{s}_{1})\;(\Varid{fmap}\;\Varid{h_{State}}\;\Varid{y})))))\;\Varid{s}_{2}{}\<[E]%
\\
\>[3]{}\mathrel{=}\mbox{\commentbegin ~  definition of \ensuremath{\Varid{fmap}}  \commentend}{}\<[E]%
\\
\>[3]{}\hsindent{3}{}\<[6]%
\>[6]{}\Conid{StateT}\mathbin{\$}\lambda (\Varid{s}_{1},\Varid{s}_{2}){}\<[27]%
\>[27]{}\to {}\<[27E]%
\>[31]{}\alpha\mathbin{\langle\hspace{1.6pt}\mathclap{\raisebox{0.1pt}{\scalebox{1}{\$}}}\hspace{1.6pt}\rangle}\Varid{run_{StateT}}{}\<[E]%
\\
\>[6]{}\hsindent{2}{}\<[8]%
\>[8]{}(\Varid{h_{State}}\;(\Conid{Op}\mathbin{\$}\Varid{fmap}\;(\lambda \Varid{k}\to \Varid{run_{StateT}}\;\Varid{k}\;\Varid{s}_{1})\;(\Conid{Inr}\;(\Varid{fmap}\;\Varid{h_{State}}\;\Varid{y}))))\;\Varid{s}_{2}{}\<[E]%
\\
\>[3]{}\mathrel{=}\mbox{\commentbegin ~  \ensuremath{\Varid{\beta}}-expansion  \commentend}{}\<[E]%
\\
\>[3]{}\hsindent{3}{}\<[6]%
\>[6]{}\Conid{StateT}\mathbin{\$}\lambda (\Varid{s}_{1},\Varid{s}_{2}){}\<[27]%
\>[27]{}\to {}\<[27E]%
\>[31]{}\alpha\mathbin{\langle\hspace{1.6pt}\mathclap{\raisebox{0.1pt}{\scalebox{1}{\$}}}\hspace{1.6pt}\rangle}\Varid{run_{StateT}}\;(\Varid{h_{State}}{}\<[E]%
\\
\>[6]{}\hsindent{2}{}\<[8]%
\>[8]{}((\lambda \Varid{s}\to \Conid{Op}\mathbin{\$}\Varid{fmap}\;(\lambda \Varid{k}\to \Varid{run_{StateT}}\;\Varid{k}\;\Varid{s})\;(\Conid{Inr}\;(\Varid{fmap}\;\Varid{h_{State}}\;\Varid{y})))\;\Varid{s}_{1}))\;\Varid{s}_{2}{}\<[E]%
\\
\>[3]{}\mathrel{=}\mbox{\commentbegin ~  definition of \ensuremath{\Varid{run_{StateT}}}   \commentend}{}\<[E]%
\\
\>[3]{}\hsindent{3}{}\<[6]%
\>[6]{}\Conid{StateT}\mathbin{\$}\lambda (\Varid{s}_{1},\Varid{s}_{2}){}\<[27]%
\>[27]{}\to {}\<[27E]%
\>[31]{}\alpha\mathbin{\langle\hspace{1.6pt}\mathclap{\raisebox{0.1pt}{\scalebox{1}{\$}}}\hspace{1.6pt}\rangle}\Varid{run_{StateT}}\;(\Varid{h_{State}}\;(\Varid{run_{StateT}}\;(\Conid{StateT}\mathbin{\$}{}\<[E]%
\\
\>[6]{}\hsindent{2}{}\<[8]%
\>[8]{}\lambda \Varid{s}\to \Conid{Op}\mathbin{\$}\Varid{fmap}\;(\lambda \Varid{k}\to \Varid{run_{StateT}}\;\Varid{k}\;\Varid{s})\;(\Conid{Inr}\;(\Varid{fmap}\;\Varid{h_{State}}\;\Varid{y})))\;\Varid{s}_{1}))\;\Varid{s}_{2}{}\<[E]%
\\
\>[3]{}\mathrel{=}\mbox{\commentbegin ~  definition of \ensuremath{\Varid{h_{State}}}   \commentend}{}\<[E]%
\\
\>[3]{}\hsindent{3}{}\<[6]%
\>[6]{}\Conid{StateT}\mathbin{\$}\lambda (\Varid{s}_{1},\Varid{s}_{2}){}\<[27]%
\>[27]{}\to {}\<[27E]%
\>[31]{}\alpha\mathbin{\langle\hspace{1.6pt}\mathclap{\raisebox{0.1pt}{\scalebox{1}{\$}}}\hspace{1.6pt}\rangle}{}\<[E]%
\\
\>[6]{}\hsindent{2}{}\<[8]%
\>[8]{}\Varid{run_{StateT}}\;(\Varid{h_{State}}\;(\Varid{run_{StateT}}\;(\Varid{h_{State}}\;(\Conid{Op}\;(\Conid{Inr}\;(\Conid{Inr}\;\Varid{y}))))\;\Varid{s}_{1}))\;\Varid{s}_{2}{}\<[E]%
\\
\>[3]{}\mathrel{=}\mbox{\commentbegin ~  definition of \ensuremath{\Varid{h_{States}^\prime}}   \commentend}{}\<[E]%
\\
\>[3]{}\hsindent{3}{}\<[6]%
\>[6]{}\Varid{h_{States}^\prime}\;(\Conid{Op}\;(\Conid{Inr}\;(\Conid{Inr}\;\Varid{y}))){}\<[E]%
\ColumnHook
\end{hscode}\resethooks
\indentend \end{proof}

\section{Proofs for the All in One Simulation}
\label{app:final-simulate}

In this section, we prove the correctness of the final simulation in
\Cref{sec:final-simulate}.

\finalSimulate*

\begin{proof}
We calculate as follows, using all our three previous theorems
\Cref{thm:local-global}, \Cref{thm:nondet-state},
\Cref{thm:states-state}, and an auxiliary lemma \Cref{lemma:final1}.
\indentbegin \begin{hscode}\SaveRestoreHook
\column{B}{@{}>{\hspre}l<{\hspost}@{}}%
\column{3}{@{}>{\hspre}l<{\hspost}@{}}%
\column{6}{@{}>{\hspre}l<{\hspost}@{}}%
\column{E}{@{}>{\hspre}l<{\hspost}@{}}%
\>[6]{}\Varid{simulate}{}\<[E]%
\\
\>[3]{}\mathrel{=}\mbox{\commentbegin ~  definition of \ensuremath{\Varid{simulate}}   \commentend}{}\<[E]%
\\
\>[3]{}\hsindent{3}{}\<[6]%
\>[6]{}\Varid{extract}\hsdot{\circ }{.}\Varid{h_{State}}\hsdot{\circ }{.}\Varid{states2state}\hsdot{\circ }{.}\Varid{nondet2state}\hsdot{\circ }{.}\Varid{(\Leftrightarrow)}\hsdot{\circ }{.}\Varid{local2global}{}\<[E]%
\\
\>[3]{}\mathrel{=}\mbox{\commentbegin ~  \Cref{thm:states-state}   \commentend}{}\<[E]%
\\
\>[3]{}\hsindent{3}{}\<[6]%
\>[6]{}\Varid{extract}\hsdot{\circ }{.}\Varid{flatten}\hsdot{\circ }{.}\Varid{h_{States}}\hsdot{\circ }{.}\Varid{nondet2state}\hsdot{\circ }{.}\Varid{(\Leftrightarrow)}\hsdot{\circ }{.}\Varid{local2global}{}\<[E]%
\\
\>[3]{}\mathrel{=}\mbox{\commentbegin ~  \Cref{lemma:final1}   \commentend}{}\<[E]%
\\
\>[3]{}\hsindent{3}{}\<[6]%
\>[6]{}\Varid{fmap}\;(\Varid{fmap}\;\Varid{fst})\hsdot{\circ }{.}\Varid{run_{StateT}}\hsdot{\circ }{.}\Varid{h_{State}}\hsdot{\circ }{.}\Varid{extract_{SS}}\hsdot{\circ }{.}\Varid{h_{State}}\hsdot{\circ }{.}\Varid{nondet2state}\hsdot{\circ }{.}\Varid{(\Leftrightarrow)}\hsdot{\circ }{.}\Varid{local2global}{}\<[E]%
\\
\>[3]{}\mathrel{=}\mbox{\commentbegin ~  definition of \ensuremath{\Varid{run_{ND+f}}}   \commentend}{}\<[E]%
\\
\>[3]{}\hsindent{3}{}\<[6]%
\>[6]{}\Varid{fmap}\;(\Varid{fmap}\;\Varid{fst})\hsdot{\circ }{.}\Varid{run_{StateT}}\hsdot{\circ }{.}\Varid{h_{State}}\hsdot{\circ }{.}\Varid{run_{ND+f}}\hsdot{\circ }{.}\Varid{(\Leftrightarrow)}\hsdot{\circ }{.}\Varid{local2global}{}\<[E]%
\\
\>[3]{}\mathrel{=}\mbox{\commentbegin ~  \Cref{thm:nondet-state}   \commentend}{}\<[E]%
\\
\>[3]{}\hsindent{3}{}\<[6]%
\>[6]{}\Varid{fmap}\;(\Varid{fmap}\;\Varid{fst})\hsdot{\circ }{.}\Varid{run_{StateT}}\hsdot{\circ }{.}\Varid{h_{State}}\hsdot{\circ }{.}\Varid{h_{ND+f}}\hsdot{\circ }{.}\Varid{(\Leftrightarrow)}\hsdot{\circ }{.}\Varid{local2global}{}\<[E]%
\\
\>[3]{}\mathrel{=}\mbox{\commentbegin ~  definition of \ensuremath{\Varid{h_{Global}}}   \commentend}{}\<[E]%
\\
\>[3]{}\hsindent{3}{}\<[6]%
\>[6]{}\Varid{h_{Global}}\hsdot{\circ }{.}\Varid{local2global}{}\<[E]%
\\
\>[3]{}\mathrel{=}\mbox{\commentbegin ~  \Cref{thm:local-global}   \commentend}{}\<[E]%
\\
\>[3]{}\hsindent{3}{}\<[6]%
\>[6]{}\Varid{h_{Local}}{}\<[E]%
\ColumnHook
\end{hscode}\resethooks
\indentend \end{proof}

\begin{lemma}\label{lemma:final1}\indentbegin \begin{hscode}\SaveRestoreHook
\column{B}{@{}>{\hspre}l<{\hspost}@{}}%
\column{6}{@{}>{\hspre}l<{\hspost}@{}}%
\column{E}{@{}>{\hspre}l<{\hspost}@{}}%
\>[6]{}\Varid{extract}\hsdot{\circ }{.}\Varid{flatten}\hsdot{\circ }{.}\Varid{h_{States}}\mathrel{=}\Varid{fmap}\;(\Varid{fmap}\;\Varid{fst})\hsdot{\circ }{.}\Varid{run_{StateT}}\hsdot{\circ }{.}\Varid{h_{State}}\hsdot{\circ }{.}\Varid{extract_{SS}}\hsdot{\circ }{.}\Varid{h_{State}}{}\<[E]%
\ColumnHook
\end{hscode}\resethooks
\indentend \end{lemma}

\begin{proof}
As shown in \Cref{app:states-state}, we can combine \ensuremath{\Varid{flatten}\hsdot{\circ }{.}\Varid{h_{States}}} into one function \ensuremath{\Varid{h_{States}^\prime}} defined as follows:\indentbegin \begin{hscode}\SaveRestoreHook
\column{B}{@{}>{\hspre}l<{\hspost}@{}}%
\column{3}{@{}>{\hspre}l<{\hspost}@{}}%
\column{E}{@{}>{\hspre}l<{\hspost}@{}}%
\>[3]{}\Varid{h_{States}^\prime}\mathbin{::}\Conid{Functor}\;\Varid{f}\Rightarrow \Conid{Free}\;(\Varid{State_{F}}\;\Varid{s}_{1}\mathrel{{:}{+}{:}}\Varid{State_{F}}\;\Varid{s}_{2}\mathrel{{:}{+}{:}}\Varid{f})\;\Varid{a}\to \Conid{StateT}\;(\Varid{s}_{1},\Varid{s}_{2})\;(\Conid{Free}\;\Varid{f})\;\Varid{a}{}\<[E]%
\\
\>[3]{}\Varid{h_{States}^\prime}\;\Varid{t}\mathrel{=}\Conid{StateT}\mathbin{\$}\lambda (\Varid{s}_{1},\Varid{s}_{2})\to \alpha\mathbin{\langle\hspace{1.6pt}\mathclap{\raisebox{0.1pt}{\scalebox{1}{\$}}}\hspace{1.6pt}\rangle}\Varid{run_{StateT}}\;(\Varid{h_{State}}\;(\Varid{run_{StateT}}\;(\Varid{h_{State}}\;\Varid{t})\;\Varid{s}_{1}))\;\Varid{s}_{2}{}\<[E]%
\ColumnHook
\end{hscode}\resethooks
\indentend Then we show that for any input \ensuremath{\Varid{t}\mathbin{::}\Conid{Free}\;(\Varid{State_{F}}\;(\Conid{SS}\;(\Varid{State_{F}}\;\Varid{s}\mathrel{{:}{+}{:}}\Varid{f})\;\Varid{a})\mathrel{{:}{+}{:}}(\Varid{State_{F}}\;\Varid{s}\mathrel{{:}{+}{:}}\Varid{f}))\;()}, we have \ensuremath{(\Varid{extract}\hsdot{\circ }{.}\Varid{h_{States}^\prime})\;\Varid{t}\mathrel{=}(\Varid{fmap}\;(\Varid{fmap}\;\Varid{fst})\hsdot{\circ }{.}\Varid{run_{StateT}}\hsdot{\circ }{.}\Varid{h_{State}}\hsdot{\circ }{.}\Varid{extract_{SS}}\hsdot{\circ }{.}\Varid{h_{State}})\;\Varid{t}} via the
following calculation.
\indentbegin \begin{hscode}\SaveRestoreHook
\column{B}{@{}>{\hspre}l<{\hspost}@{}}%
\column{3}{@{}>{\hspre}l<{\hspost}@{}}%
\column{6}{@{}>{\hspre}l<{\hspost}@{}}%
\column{8}{@{}>{\hspre}l<{\hspost}@{}}%
\column{E}{@{}>{\hspre}l<{\hspost}@{}}%
\>[6]{}(\Varid{extract}\hsdot{\circ }{.}\Varid{h_{States}^\prime})\;\Varid{t}{}\<[E]%
\\
\>[3]{}\mathrel{=}\mbox{\commentbegin ~  function application   \commentend}{}\<[E]%
\\
\>[3]{}\hsindent{3}{}\<[6]%
\>[6]{}\Varid{extract}\;(\Varid{h_{States}^\prime}\;\Varid{t}){}\<[E]%
\\
\>[3]{}\mathrel{=}\mbox{\commentbegin ~  definition of \ensuremath{\Varid{h_{States}^\prime}}   \commentend}{}\<[E]%
\\
\>[3]{}\hsindent{3}{}\<[6]%
\>[6]{}\Varid{extract}\;(\Conid{StateT}\mathbin{\$}\lambda (\Varid{s}_{1},\Varid{s}_{2})\to \alpha\mathbin{\langle\hspace{1.6pt}\mathclap{\raisebox{0.1pt}{\scalebox{1}{\$}}}\hspace{1.6pt}\rangle}{}\<[E]%
\\
\>[6]{}\hsindent{2}{}\<[8]%
\>[8]{}\Varid{run_{StateT}}\;(\Varid{h_{State}}\;(\Varid{run_{StateT}}\;(\Varid{h_{State}}\;\Varid{t})\;\Varid{s}_{1}))\;\Varid{s}_{2}){}\<[E]%
\\
\>[3]{}\mathrel{=}\mbox{\commentbegin ~  definition of \ensuremath{\Varid{extract}}   \commentend}{}\<[E]%
\\
\>[3]{}\hsindent{3}{}\<[6]%
\>[6]{}\lambda \Varid{s}\to \Varid{results_{SS}}\hsdot{\circ }{.}\Varid{fst}\hsdot{\circ }{.}\Varid{snd}\mathbin{\langle\hspace{1.6pt}\mathclap{\raisebox{0.1pt}{\scalebox{1}{\$}}}\hspace{1.6pt}\rangle}\Varid{run_{StateT}}\;(\Conid{StateT}\mathbin{\$}\lambda (\Varid{s}_{1},\Varid{s}_{2})\to \alpha\mathbin{\langle\hspace{1.6pt}\mathclap{\raisebox{0.1pt}{\scalebox{1}{\$}}}\hspace{1.6pt}\rangle}{}\<[E]%
\\
\>[6]{}\hsindent{2}{}\<[8]%
\>[8]{}\Varid{run_{StateT}}\;(\Varid{h_{State}}\;(\Varid{run_{StateT}}\;(\Varid{h_{State}}\;\Varid{t})\;\Varid{s}_{1}))\;\Varid{s}_{2})\;(\Conid{SS}\;[\mskip1.5mu \mskip1.5mu]\;[\mskip1.5mu \mskip1.5mu],\Varid{s}){}\<[E]%
\\
\>[3]{}\mathrel{=}\mbox{\commentbegin ~  definition of \ensuremath{\Varid{run_{StateT}}}   \commentend}{}\<[E]%
\\
\>[3]{}\hsindent{3}{}\<[6]%
\>[6]{}\lambda \Varid{s}\to \Varid{results_{SS}}\hsdot{\circ }{.}\Varid{fst}\hsdot{\circ }{.}\Varid{snd}\mathbin{\langle\hspace{1.6pt}\mathclap{\raisebox{0.1pt}{\scalebox{1}{\$}}}\hspace{1.6pt}\rangle}(\lambda (\Varid{s}_{1},\Varid{s}_{2})\to \alpha\mathbin{\langle\hspace{1.6pt}\mathclap{\raisebox{0.1pt}{\scalebox{1}{\$}}}\hspace{1.6pt}\rangle}{}\<[E]%
\\
\>[6]{}\hsindent{2}{}\<[8]%
\>[8]{}\Varid{run_{StateT}}\;(\Varid{h_{State}}\;(\Varid{run_{StateT}}\;(\Varid{h_{State}}\;\Varid{t})\;\Varid{s}_{1}))\;\Varid{s}_{2})\;(\Conid{SS}\;[\mskip1.5mu \mskip1.5mu]\;[\mskip1.5mu \mskip1.5mu],\Varid{s}){}\<[E]%
\\
\>[3]{}\mathrel{=}\mbox{\commentbegin ~  function application   \commentend}{}\<[E]%
\\
\>[3]{}\hsindent{3}{}\<[6]%
\>[6]{}\lambda \Varid{s}\to \Varid{results_{SS}}\hsdot{\circ }{.}\Varid{fst}\hsdot{\circ }{.}\Varid{snd}\mathbin{\langle\hspace{1.6pt}\mathclap{\raisebox{0.1pt}{\scalebox{1}{\$}}}\hspace{1.6pt}\rangle}(\alpha\mathbin{\langle\hspace{1.6pt}\mathclap{\raisebox{0.1pt}{\scalebox{1}{\$}}}\hspace{1.6pt}\rangle}{}\<[E]%
\\
\>[6]{}\hsindent{2}{}\<[8]%
\>[8]{}\Varid{run_{StateT}}\;(\Varid{h_{State}}\;(\Varid{run_{StateT}}\;(\Varid{h_{State}}\;\Varid{t})\;(\Conid{SS}\;[\mskip1.5mu \mskip1.5mu]\;[\mskip1.5mu \mskip1.5mu])))\;\Varid{s}){}\<[E]%
\\
\>[3]{}\mathrel{=}\mbox{\commentbegin ~  \Cref{eq:functor-composition}   \commentend}{}\<[E]%
\\
\>[3]{}\hsindent{3}{}\<[6]%
\>[6]{}\lambda \Varid{s}\to \Varid{results_{SS}}\hsdot{\circ }{.}\Varid{fst}\hsdot{\circ }{.}\Varid{snd}\hsdot{\circ }{.}\alpha\mathbin{\langle\hspace{1.6pt}\mathclap{\raisebox{0.1pt}{\scalebox{1}{\$}}}\hspace{1.6pt}\rangle}{}\<[E]%
\\
\>[6]{}\hsindent{2}{}\<[8]%
\>[8]{}\Varid{run_{StateT}}\;(\Varid{h_{State}}\;(\Varid{run_{StateT}}\;(\Varid{h_{State}}\;\Varid{t})\;(\Conid{SS}\;[\mskip1.5mu \mskip1.5mu]\;[\mskip1.5mu \mskip1.5mu])))\;\Varid{s}{}\<[E]%
\\
\>[3]{}\mathrel{=}\mbox{\commentbegin ~  \ensuremath{\Varid{fst}\hsdot{\circ }{.}\Varid{snd}\hsdot{\circ }{.}\alpha\mathrel{=}\Varid{snd}\hsdot{\circ }{.}\Varid{fst}}   \commentend}{}\<[E]%
\\
\>[3]{}\hsindent{3}{}\<[6]%
\>[6]{}\lambda \Varid{s}\to \Varid{results_{SS}}\hsdot{\circ }{.}\Varid{snd}\hsdot{\circ }{.}\Varid{fst}\mathbin{\langle\hspace{1.6pt}\mathclap{\raisebox{0.1pt}{\scalebox{1}{\$}}}\hspace{1.6pt}\rangle}\Varid{run_{StateT}}\;(\Varid{h_{State}}\;(\Varid{run_{StateT}}\;(\Varid{h_{State}}\;\Varid{t})\;(\Conid{SS}\;[\mskip1.5mu \mskip1.5mu]\;[\mskip1.5mu \mskip1.5mu])))\;\Varid{s}{}\<[E]%
\\
\>[3]{}\mathrel{=}\mbox{\commentbegin ~  \Cref{eq:functor-composition} and definition of \ensuremath{\mathbin{\langle\hspace{1.6pt}\mathclap{\raisebox{0.1pt}{\scalebox{1}{\$}}}\hspace{1.6pt}\rangle}}   \commentend}{}\<[E]%
\\
\>[3]{}\hsindent{3}{}\<[6]%
\>[6]{}\lambda \Varid{s}\to \Varid{fmap}\;(\Varid{results_{SS}}\hsdot{\circ }{.}\Varid{snd})\hsdot{\circ }{.}\Varid{fmap}\;\Varid{fst}\mathbin{\$}{}\<[E]%
\\
\>[6]{}\hsindent{2}{}\<[8]%
\>[8]{}\Varid{run_{StateT}}\;(\Varid{h_{State}}\;(\Varid{run_{StateT}}\;(\Varid{h_{State}}\;\Varid{t})\;(\Conid{SS}\;[\mskip1.5mu \mskip1.5mu]\;[\mskip1.5mu \mskip1.5mu])))\;\Varid{s}{}\<[E]%
\\
\>[3]{}\mathrel{=}\mbox{\commentbegin ~  definition of \ensuremath{\Varid{flip}} and reformulation   \commentend}{}\<[E]%
\\
\>[3]{}\hsindent{3}{}\<[6]%
\>[6]{}\lambda \Varid{s}\to \Varid{fmap}\;(\Varid{results_{SS}}\hsdot{\circ }{.}\Varid{snd})\hsdot{\circ }{.}\Varid{fmap}\;\Varid{fst}\mathbin{\$}{}\<[E]%
\\
\>[6]{}\hsindent{2}{}\<[8]%
\>[8]{}\Varid{flip}\;(\Varid{run_{StateT}}\hsdot{\circ }{.}\Varid{h_{State}})\;\Varid{s}\;(\Varid{run_{StateT}}\;(\Varid{h_{State}}\;\Varid{t})\;(\Conid{SS}\;[\mskip1.5mu \mskip1.5mu]\;[\mskip1.5mu \mskip1.5mu])){}\<[E]%
\\
\>[3]{}\mathrel{=}\mbox{\commentbegin ~  reformulation   \commentend}{}\<[E]%
\\
\>[3]{}\hsindent{3}{}\<[6]%
\>[6]{}\lambda \Varid{s}\to \Varid{fmap}\;(\Varid{results_{SS}}\hsdot{\circ }{.}\Varid{snd})\hsdot{\circ }{.}(\Varid{fmap}\;\Varid{fst}\hsdot{\circ }{.}\Varid{flip}\;(\Varid{run_{StateT}}\hsdot{\circ }{.}\Varid{h_{State}})\;\Varid{s})\mathbin{\$}{}\<[E]%
\\
\>[6]{}\hsindent{2}{}\<[8]%
\>[8]{}\Varid{run_{StateT}}\;(\Varid{h_{State}}\;\Varid{t})\;(\Conid{SS}\;[\mskip1.5mu \mskip1.5mu]\;[\mskip1.5mu \mskip1.5mu]){}\<[E]%
\\
\>[3]{}\mathrel{=}\mbox{\commentbegin ~  parametricity of free monads  \commentend}{}\<[E]%
\\
\>[3]{}\hsindent{3}{}\<[6]%
\>[6]{}\lambda \Varid{s}\to (\Varid{fmap}\;\Varid{fst}\hsdot{\circ }{.}\Varid{flip}\;(\Varid{run_{StateT}}\hsdot{\circ }{.}\Varid{h_{State}})\;\Varid{s})\hsdot{\circ }{.}\Varid{fmap}\;(\Varid{results_{SS}}\hsdot{\circ }{.}\Varid{snd})\mathbin{\$}{}\<[E]%
\\
\>[6]{}\hsindent{2}{}\<[8]%
\>[8]{}\Varid{run_{StateT}}\;(\Varid{h_{State}}\;\Varid{t})\;(\Conid{SS}\;[\mskip1.5mu \mskip1.5mu]\;[\mskip1.5mu \mskip1.5mu]){}\<[E]%
\\
\>[3]{}\mathrel{=}\mbox{\commentbegin ~  definition of \ensuremath{\mathbin{\langle\hspace{1.6pt}\mathclap{\raisebox{0.1pt}{\scalebox{1}{\$}}}\hspace{1.6pt}\rangle}}  \commentend}{}\<[E]%
\\
\>[3]{}\hsindent{3}{}\<[6]%
\>[6]{}\lambda \Varid{s}\to (\Varid{fmap}\;\Varid{fst}\hsdot{\circ }{.}\Varid{flip}\;(\Varid{run_{StateT}}\hsdot{\circ }{.}\Varid{h_{State}})\;\Varid{s})\mathbin{\$}\Varid{results_{SS}}\hsdot{\circ }{.}\Varid{snd}\mathbin{\langle\hspace{1.6pt}\mathclap{\raisebox{0.1pt}{\scalebox{1}{\$}}}\hspace{1.6pt}\rangle}{}\<[E]%
\\
\>[6]{}\hsindent{2}{}\<[8]%
\>[8]{}\Varid{run_{StateT}}\;(\Varid{h_{State}}\;\Varid{t})\;(\Conid{SS}\;[\mskip1.5mu \mskip1.5mu]\;[\mskip1.5mu \mskip1.5mu]){}\<[E]%
\\
\>[3]{}\mathrel{=}\mbox{\commentbegin ~  function application  \commentend}{}\<[E]%
\\
\>[3]{}\hsindent{3}{}\<[6]%
\>[6]{}\lambda \Varid{s}\to \Varid{fmap}\;\Varid{fst}\;(\Varid{run_{StateT}}\;(\Varid{h_{State}}\;(\Varid{results_{SS}}\hsdot{\circ }{.}\Varid{snd}\mathbin{\langle\hspace{1.6pt}\mathclap{\raisebox{0.1pt}{\scalebox{1}{\$}}}\hspace{1.6pt}\rangle}{}\<[E]%
\\
\>[6]{}\hsindent{2}{}\<[8]%
\>[8]{}\Varid{run_{StateT}}\;(\Varid{h_{State}}\;\Varid{t})\;(\Conid{SS}\;[\mskip1.5mu \mskip1.5mu]\;[\mskip1.5mu \mskip1.5mu])))\;\Varid{s}){}\<[E]%
\\
\>[3]{}\mathrel{=}\mbox{\commentbegin ~  definition of \ensuremath{\Varid{fmap}}  \commentend}{}\<[E]%
\\
\>[3]{}\hsindent{3}{}\<[6]%
\>[6]{}\lambda \Varid{s}\to \Varid{fmap}\;(\Varid{fmap}\;\Varid{fst})\;(\Varid{run_{StateT}}\;(\Varid{h_{State}}\;(\Varid{results_{SS}}\hsdot{\circ }{.}\Varid{snd}\mathbin{\langle\hspace{1.6pt}\mathclap{\raisebox{0.1pt}{\scalebox{1}{\$}}}\hspace{1.6pt}\rangle}{}\<[E]%
\\
\>[6]{}\hsindent{2}{}\<[8]%
\>[8]{}\Varid{run_{StateT}}\;(\Varid{h_{State}}\;\Varid{t})\;(\Conid{SS}\;[\mskip1.5mu \mskip1.5mu]\;[\mskip1.5mu \mskip1.5mu]))))\;\Varid{s}{}\<[E]%
\\
\>[3]{}\mathrel{=}\mbox{\commentbegin ~  reformulation  \commentend}{}\<[E]%
\\
\>[3]{}\hsindent{3}{}\<[6]%
\>[6]{}\lambda \Varid{s}\to (\Varid{fmap}\;(\Varid{fmap}\;\Varid{fst})\hsdot{\circ }{.}\Varid{run_{StateT}}\hsdot{\circ }{.}\Varid{h_{State}}\mathbin{\$}\Varid{results_{SS}}\hsdot{\circ }{.}\Varid{snd}\mathbin{\langle\hspace{1.6pt}\mathclap{\raisebox{0.1pt}{\scalebox{1}{\$}}}\hspace{1.6pt}\rangle}{}\<[E]%
\\
\>[6]{}\hsindent{2}{}\<[8]%
\>[8]{}\Varid{run_{StateT}}\;(\Varid{h_{State}}\;\Varid{t})\;(\Conid{SS}\;[\mskip1.5mu \mskip1.5mu]\;[\mskip1.5mu \mskip1.5mu]))\;\Varid{s}{}\<[E]%
\\
\>[3]{}\mathrel{=}\mbox{\commentbegin ~  \ensuremath{\Varid{\eta}}-reduction   \commentend}{}\<[E]%
\\
\>[3]{}\hsindent{3}{}\<[6]%
\>[6]{}\Varid{fmap}\;(\Varid{fmap}\;\Varid{fst})\hsdot{\circ }{.}\Varid{run_{StateT}}\hsdot{\circ }{.}\Varid{h_{State}}\mathbin{\$}\Varid{results_{SS}}\hsdot{\circ }{.}\Varid{snd}\mathbin{\langle\hspace{1.6pt}\mathclap{\raisebox{0.1pt}{\scalebox{1}{\$}}}\hspace{1.6pt}\rangle}\Varid{run_{StateT}}\;(\Varid{h_{State}}\;\Varid{t})\;(\Conid{SS}\;[\mskip1.5mu \mskip1.5mu]\;[\mskip1.5mu \mskip1.5mu]){}\<[E]%
\\
\>[3]{}\mathrel{=}\mbox{\commentbegin ~  definition of \ensuremath{\Varid{extract_{SS}}}   \commentend}{}\<[E]%
\\
\>[3]{}\hsindent{3}{}\<[6]%
\>[6]{}\Varid{fmap}\;(\Varid{fmap}\;\Varid{fst})\hsdot{\circ }{.}\Varid{run_{StateT}}\hsdot{\circ }{.}\Varid{h_{State}}\mathbin{\$}\Varid{extract_{SS}}\;(\Varid{h_{State}}\;\Varid{t}){}\<[E]%
\\
\>[3]{}\mathrel{=}\mbox{\commentbegin ~     \commentend}{}\<[E]%
\\
\>[3]{}\hsindent{3}{}\<[6]%
\>[6]{}(\Varid{fmap}\;(\Varid{fmap}\;\Varid{fst})\hsdot{\circ }{.}\Varid{run_{StateT}}\hsdot{\circ }{.}\Varid{h_{State}}\hsdot{\circ }{.}\Varid{extract_{SS}}\hsdot{\circ }{.}\Varid{h_{State}})\;\Varid{t}{}\<[E]%
\ColumnHook
\end{hscode}\resethooks
\indentend Note that in the above calculation, we use the
parametricity~\citep{Reynolds83, Wadler89} of free monads which is
formally stated as follows:\indentbegin \begin{hscode}\SaveRestoreHook
\column{B}{@{}>{\hspre}l<{\hspost}@{}}%
\column{3}{@{}>{\hspre}l<{\hspost}@{}}%
\column{E}{@{}>{\hspre}l<{\hspost}@{}}%
\>[3]{}\Varid{fmap}\;\Varid{f}\hsdot{\circ }{.}\Varid{g}\mathrel{=}\Varid{g}\hsdot{\circ }{.}\Varid{fmap}\;\Varid{f}{}\<[E]%
\ColumnHook
\end{hscode}\resethooks
\indentend for any \ensuremath{\Varid{g}\mathbin{::}\forall \Varid{a}\hsforall \hsdot{\circ }{.}\Conid{Free}\;\Conid{F}\;\Varid{a}\to \Conid{Free}\;\Conid{G}\;\Varid{a}} with two functors \ensuremath{\Conid{F}}
and \ensuremath{\Conid{G}}.

\end{proof}

\section{Proofs for Modelling Local State with Undo}
\label{app:modify-local-global}

In this section, we prove the following theorem in \Cref{sec:undo}.

\modifyLocalGlobal*

The proof structure is very similar to that in
\Cref{app:local-global}. We start with the following preliminary
fusion.

\paragraph*{Preliminary}
It is easy to see that \ensuremath{\Varid{run_{StateT}}\hsdot{\circ }{.}\Varid{h_{Modify}}} can be fused into a single
fold defined as follows:
\indentbegin \begin{hscode}\SaveRestoreHook
\column{B}{@{}>{\hspre}l<{\hspost}@{}}%
\column{3}{@{}>{\hspre}l<{\hspost}@{}}%
\column{5}{@{}>{\hspre}l<{\hspost}@{}}%
\column{11}{@{}>{\hspre}l<{\hspost}@{}}%
\column{14}{@{}>{\hspre}l<{\hspost}@{}}%
\column{26}{@{}>{\hspre}l<{\hspost}@{}}%
\column{E}{@{}>{\hspre}l<{\hspost}@{}}%
\>[B]{}\Varid{h_{Modify1}}{}\<[11]%
\>[11]{}\mathbin{::}(\Conid{Functor}\;\Varid{f},\Conid{Undo}\;\Varid{s}\;\Varid{r})\Rightarrow \Conid{Free}\;(\Varid{Modify_{F}}\;\Varid{s}\;\Varid{r}\mathrel{{:}{+}{:}}\Varid{f})\;\Varid{a}\to (\Varid{s}\to \Conid{Free}\;\Varid{f}\;(\Varid{a},\Varid{s})){}\<[E]%
\\
\>[B]{}\Varid{h_{Modify1}}{}\<[11]%
\>[11]{}\mathrel{=}{}\<[14]%
\>[14]{}\Varid{fold}\;\Varid{gen_{S}}\;(\Varid{alg_{S}}\mathbin{\#}\Varid{fwd_{S}}){}\<[E]%
\\
\>[B]{}\hsindent{3}{}\<[3]%
\>[3]{}\mathbf{where}{}\<[E]%
\\
\>[3]{}\hsindent{2}{}\<[5]%
\>[5]{}\Varid{gen_{S}}\;\Varid{x}\;{}\<[26]%
\>[26]{}\Varid{s}\mathrel{=}\Conid{Var}\;(\Varid{x},\Varid{s}){}\<[E]%
\\
\>[3]{}\hsindent{2}{}\<[5]%
\>[5]{}\Varid{alg_{S}}\;(\Conid{MGet}\;\Varid{k})\;{}\<[26]%
\>[26]{}\Varid{s}\mathrel{=}\Varid{k}\;\Varid{s}\;\Varid{s}{}\<[E]%
\\
\>[3]{}\hsindent{2}{}\<[5]%
\>[5]{}\Varid{alg_{S}}\;(\Conid{MUpdate}\;\Varid{r}\;\Varid{k})\;{}\<[26]%
\>[26]{}\Varid{s}\mathrel{=}\Varid{k}\;(\Varid{s}\mathbin{\oplus}\Varid{r}){}\<[E]%
\\
\>[3]{}\hsindent{2}{}\<[5]%
\>[5]{}\Varid{alg_{S}}\;(\Conid{MRestore}\;\Varid{r}\;\Varid{k})\;{}\<[26]%
\>[26]{}\Varid{s}\mathrel{=}\Varid{k}\;(\Varid{s}\mathbin{\ominus}\Varid{r}){}\<[E]%
\\
\>[3]{}\hsindent{2}{}\<[5]%
\>[5]{}\Varid{fwd_{S}}\;\Varid{y}\;{}\<[26]%
\>[26]{}\Varid{s}\mathrel{=}\Conid{Op}\;(\Varid{fmap}\;(\mathbin{\$}\Varid{s})\;\Varid{y}){}\<[E]%
\ColumnHook
\end{hscode}\resethooks
\indentend For brevity, we use \ensuremath{\Varid{h_{Modify1}}} to replace \ensuremath{\Varid{run_{StateT}}\hsdot{\circ }{.}\Varid{h_{Modify}}} in the
following proofs.

\subsection{Main Proof Structure}
The main proof structure of \Cref{thm:modify-local-global} is as
follows.
\begin{proof}
Both the left-hand side and the right-hand side of the equation consist of
function compositions involving one or more folds.
We apply fold fusion separately on both sides to contract each
into a single fold:
\begin{eqnarray*}
\ensuremath{\Varid{h_{GlobalM}}\hsdot{\circ }{.}\Varid{local2global_M}} & = & \ensuremath{\Varid{fold}\;\Varid{gen}_{\Varid{LHS}}\;(\Varid{alg}_{\Varid{LHS}}^{\Varid{S}}\mathbin{\#}\Varid{alg}_{\Varid{RHS}}^{\Varid{ND}}\mathbin{\#}\Varid{fwd}_{\Varid{LHS}})} \\
\ensuremath{\Varid{h_{LocalM}}}& = & \ensuremath{\Varid{fold}\;\Varid{gen}_{\Varid{RHS}}\;(\Varid{alg}_{\Varid{RHS}}^{\Varid{S}}\mathbin{\#}\Varid{alg}_{\Varid{RHS}}^{\Varid{ND}}\mathbin{\#}\Varid{fwd}_{\Varid{RHS}})}
\end{eqnarray*}
Finally, we show that both folds are equal by showing that their
corresponding parameters are equal:
\begin{eqnarray*}
\ensuremath{\Varid{gen}_{\Varid{LHS}}} & = & \ensuremath{\Varid{gen}_{\Varid{RHS}}} \\
\ensuremath{\Varid{alg}_{\Varid{LHS}}^{\Varid{S}}} & = & \ensuremath{\Varid{alg}_{\Varid{RHS}}^{\Varid{S}}} \\
\ensuremath{\Varid{alg}_{\Varid{LHS}}^{\Varid{ND}}} & = & \ensuremath{\Varid{alg}_{\Varid{RHS}}^{\Varid{ND}}} \\
\ensuremath{\Varid{fwd}_{\Varid{LHS}}} & = & \ensuremath{\Varid{fwd}_{\Varid{RHS}}}
\end{eqnarray*}
We elaborate each of these steps below.
\end{proof}

\subsection{Fusing the Right-Hand Side}
\label{app:modify-fusing-rhs}
We calculate as follows:
\indentbegin \begin{hscode}\SaveRestoreHook
\column{B}{@{}>{\hspre}l<{\hspost}@{}}%
\column{5}{@{}>{\hspre}l<{\hspost}@{}}%
\column{7}{@{}>{\hspre}l<{\hspost}@{}}%
\column{10}{@{}>{\hspre}l<{\hspost}@{}}%
\column{14}{@{}>{\hspre}l<{\hspost}@{}}%
\column{E}{@{}>{\hspre}l<{\hspost}@{}}%
\>[5]{}\Varid{h_{LocalM}}{}\<[E]%
\\
\>[B]{}\mathrel{=}\mbox{\commentbegin ~  definition  \commentend}{}\<[E]%
\\
\>[B]{}\hsindent{5}{}\<[5]%
\>[5]{}\Varid{h_{L}}\hsdot{\circ }{.}\Varid{h_{Modify1}}{}\<[E]%
\\
\>[5]{}\hsindent{2}{}\<[7]%
\>[7]{} \text{with } {}\<[E]%
\\
\>[7]{}\hsindent{3}{}\<[10]%
\>[10]{}\Varid{h_{L}}{}\<[14]%
\>[14]{}\mathbin{::}(\Conid{Functor}\;\Varid{f},\Conid{Functor}\;\Varid{g}){}\<[E]%
\\
\>[14]{}\Rightarrow \Varid{g}\;(\Conid{Free}\;(\Varid{Nondet_{F}}\mathrel{{:}{+}{:}}\Varid{f})\;(\Varid{a},\Varid{s}))\to \Varid{g}\;(\Conid{Free}\;\Varid{f}\;[\mskip1.5mu \Varid{a}\mskip1.5mu]){}\<[E]%
\\
\>[7]{}\hsindent{3}{}\<[10]%
\>[10]{}\Varid{h_{L}}{}\<[14]%
\>[14]{}\mathrel{=}\Varid{fmap}\;(\Varid{fmap}\;(\Varid{fmap}\;\Varid{fst})\hsdot{\circ }{.}\Varid{h_{ND+f}}){}\<[E]%
\\
\>[B]{}\mathrel{=}\mbox{\commentbegin ~  definition of \ensuremath{\Varid{h_{Modify1}}}   \commentend}{}\<[E]%
\\
\>[B]{}\hsindent{5}{}\<[5]%
\>[5]{}\Varid{h_{L}}\hsdot{\circ }{.}\Varid{fold}\;\Varid{gen_{S}}\;(\Varid{alg_{S}}\mathbin{\#}\Varid{fwd_{S}}){}\<[E]%
\\
\>[B]{}\mathrel{=}\mbox{\commentbegin ~  fold fusion-post (Equation \ref{eq:fusion-post})   \commentend}{}\<[E]%
\\
\>[B]{}\hsindent{5}{}\<[5]%
\>[5]{}\Varid{fold}\;\Varid{gen}_{\Varid{RHS}}\;(\Varid{alg}_{\Varid{RHS}}^{\Varid{S}}\mathbin{\#}\Varid{alg}_{\Varid{RHS}}^{\Varid{ND}}\mathbin{\#}\Varid{fwd}_{\Varid{RHS}}){}\<[E]%
\ColumnHook
\end{hscode}\resethooks
\indentend This last step is valid provided that the fusion conditions are satisfied:
\[\ba{rclr}
\ensuremath{\Varid{h_{L}}\hsdot{\circ }{.}\Varid{gen_{S}}} & = & \ensuremath{\Varid{gen}_{\Varid{RHS}}} &\refa{} \\
\ensuremath{\Varid{h_{L}}\hsdot{\circ }{.}(\Varid{alg_{S}}\mathbin{\#}\Varid{fwd_{S}})} & = & \ensuremath{(\Varid{alg}_{\Varid{RHS}}^{\Varid{S}}\mathbin{\#}\Varid{alg}_{\Varid{RHS}}^{\Varid{ND}}\mathbin{\#}\Varid{fwd}_{\Varid{RHS}})\hsdot{\circ }{.}\Varid{fmap}\;\Varid{h_{L}}} &\refb{}
\ea\]

For the first fusion condition \refa{}, we define \ensuremath{\Varid{gen}_{\Varid{RHS}}} as follows\indentbegin \begin{hscode}\SaveRestoreHook
\column{B}{@{}>{\hspre}l<{\hspost}@{}}%
\column{3}{@{}>{\hspre}l<{\hspost}@{}}%
\column{E}{@{}>{\hspre}l<{\hspost}@{}}%
\>[3]{}\Varid{gen}_{\Varid{RHS}}\mathbin{::}\Conid{Functor}\;\Varid{f}\Rightarrow \Varid{a}\to (\Varid{s}\to \Conid{Free}\;\Varid{f}\;[\mskip1.5mu \Varid{a}\mskip1.5mu]){}\<[E]%
\\
\>[3]{}\Varid{gen}_{\Varid{RHS}}\;\Varid{x}\mathrel{=}\lambda \Varid{s}\to \Conid{Var}\;[\mskip1.5mu \Varid{x}\mskip1.5mu]{}\<[E]%
\ColumnHook
\end{hscode}\resethooks
\indentend We show that \refa{} is satisfied by the following calculation.
\indentbegin \begin{hscode}\SaveRestoreHook
\column{B}{@{}>{\hspre}l<{\hspost}@{}}%
\column{3}{@{}>{\hspre}l<{\hspost}@{}}%
\column{5}{@{}>{\hspre}l<{\hspost}@{}}%
\column{E}{@{}>{\hspre}l<{\hspost}@{}}%
\>[5]{}\Varid{h_{L}}\;(\Varid{gen_{S}}\;\Varid{x}){}\<[E]%
\\
\>[3]{}\mathrel{=}\mbox{\commentbegin ~ definition of \ensuremath{\Varid{gen_{S}}}  \commentend}{}\<[E]%
\\
\>[3]{}\hsindent{2}{}\<[5]%
\>[5]{}\Varid{h_{L}}\;(\lambda \Varid{s}\to \Conid{Var}\;(\Varid{x},\Varid{s})){}\<[E]%
\\
\>[3]{}\mathrel{=}\mbox{\commentbegin ~ definition of \ensuremath{\Varid{h_{L}}}  \commentend}{}\<[E]%
\\
\>[3]{}\hsindent{2}{}\<[5]%
\>[5]{}\Varid{fmap}\;(\Varid{fmap}\;(\Varid{fmap}\;\Varid{fst})\hsdot{\circ }{.}\Varid{h_{ND+f}})\;(\lambda \Varid{s}\to \Conid{Var}\;(\Varid{x},\Varid{s})){}\<[E]%
\\
\>[3]{}\mathrel{=}\mbox{\commentbegin ~ definition of \ensuremath{\Varid{fmap}}  \commentend}{}\<[E]%
\\
\>[3]{}\hsindent{2}{}\<[5]%
\>[5]{}\lambda \Varid{s}\to \Varid{fmap}\;(\Varid{fmap}\;\Varid{fst})\;(\Varid{h_{ND+f}}\;(\Conid{Var}\;(\Varid{x},\Varid{s}))){}\<[E]%
\\
\>[3]{}\mathrel{=}\mbox{\commentbegin ~ definition of \ensuremath{\Varid{h_{ND+f}}}  \commentend}{}\<[E]%
\\
\>[3]{}\hsindent{2}{}\<[5]%
\>[5]{}\lambda \Varid{s}\to \Varid{fmap}\;(\Varid{fmap}\;\Varid{fst})\;(\Conid{Var}\;[\mskip1.5mu (\Varid{x},\Varid{s})\mskip1.5mu]){}\<[E]%
\\
\>[3]{}\mathrel{=}\mbox{\commentbegin ~ definition of \ensuremath{\Varid{fmap}} (twice)  \commentend}{}\<[E]%
\\
\>[3]{}\hsindent{2}{}\<[5]%
\>[5]{}\lambda \Varid{s}\to \Conid{Var}\;[\mskip1.5mu \Varid{x}\mskip1.5mu]{}\<[E]%
\\
\>[3]{}\mathrel{=}\mbox{\commentbegin ~ definition of \ensuremath{\Varid{gen}_{\Varid{RHS}}}   \commentend}{}\<[E]%
\\
\>[3]{}\mathrel{=}\Varid{gen}_{\Varid{RHS}}\;\Varid{x}{}\<[E]%
\ColumnHook
\end{hscode}\resethooks
\indentend By a straightforward case analysis on the two cases \ensuremath{\Conid{Inl}} and \ensuremath{\Conid{Inr}},
the second fusion condition \refb{} decomposes into two separate
conditions:
\[\ba{rclr}
\ensuremath{\Varid{h_{L}}\hsdot{\circ }{.}\Varid{alg_{S}}} & = & \ensuremath{\Varid{alg}_{\Varid{RHS}}^{\Varid{S}}\hsdot{\circ }{.}\Varid{fmap}\;\Varid{h_{L}}} &\refc{} \\
\ensuremath{\Varid{h_{L}}\hsdot{\circ }{.}\Varid{fwd_{S}}\hsdot{\circ }{.}\Conid{Inl}}& = & \ensuremath{\Varid{alg}_{\Varid{RHS}}^{\Varid{ND}}\hsdot{\circ }{.}\Varid{fmap}\;\Varid{h_{L}}} &\refd{}\\
\ensuremath{\Varid{h_{L}}\hsdot{\circ }{.}\Varid{fwd_{S}}\hsdot{\circ }{.}\Conid{Inr}}& = & \ensuremath{\Varid{fwd}_{\Varid{RHS}}\hsdot{\circ }{.}\Varid{fmap}\;\Varid{h_{L}}} &\refe{}
\ea\]

For the subcondition \refc{}, we define \ensuremath{\Varid{alg}_{\Varid{RHS}}^{\Varid{S}}} as follows.\indentbegin \begin{hscode}\SaveRestoreHook
\column{B}{@{}>{\hspre}l<{\hspost}@{}}%
\column{3}{@{}>{\hspre}l<{\hspost}@{}}%
\column{27}{@{}>{\hspre}l<{\hspost}@{}}%
\column{E}{@{}>{\hspre}l<{\hspost}@{}}%
\>[3]{}\Varid{alg}_{\Varid{RHS}}^{\Varid{S}}\mathbin{::}\Conid{Undo}\;\Varid{s}\;\Varid{r}\Rightarrow \Varid{Modify_{F}}\;\Varid{s}\;\Varid{r}\;(\Varid{s}\to \Varid{p})\to (\Varid{s}\to \Varid{p}){}\<[E]%
\\
\>[3]{}\Varid{alg}_{\Varid{RHS}}^{\Varid{S}}\;(\Conid{MGet}\;\Varid{k}){}\<[27]%
\>[27]{}\mathrel{=}\lambda \Varid{s}\to \Varid{k}\;\Varid{s}\;\Varid{s}{}\<[E]%
\\
\>[3]{}\Varid{alg}_{\Varid{RHS}}^{\Varid{S}}\;(\Conid{MUpdate}\;\Varid{r}\;\Varid{k}){}\<[27]%
\>[27]{}\mathrel{=}\lambda \Varid{s}\to \Varid{k}\;(\Varid{s}\mathbin{\oplus}\Varid{r}){}\<[E]%
\\
\>[3]{}\Varid{alg}_{\Varid{RHS}}^{\Varid{S}}\;(\Conid{MRestore}\;\Varid{r}\;\Varid{k}){}\<[27]%
\>[27]{}\mathrel{=}\lambda \Varid{s}\to \Varid{k}\;(\Varid{s}\mathbin{\oplus}\Varid{r}){}\<[E]%
\ColumnHook
\end{hscode}\resethooks
\indentend We prove its correctness by a case analysis on the shape of input \ensuremath{\Varid{t}\mathbin{::}\Varid{State_{F}}\;\Varid{s}\;(\Varid{s}\to \Conid{Free}\;(\Varid{Nondet_{F}}\mathrel{{:}{+}{:}}\Varid{f})\;(\Varid{a},\Varid{s}))}.

\vspace{.5\baselineskip}
\noindent \mbox{\underline{case \ensuremath{\Varid{t}\mathrel{=}\Conid{Get}\;\Varid{k}}}}\indentbegin \begin{hscode}\SaveRestoreHook
\column{B}{@{}>{\hspre}l<{\hspost}@{}}%
\column{3}{@{}>{\hspre}l<{\hspost}@{}}%
\column{5}{@{}>{\hspre}l<{\hspost}@{}}%
\column{E}{@{}>{\hspre}l<{\hspost}@{}}%
\>[5]{}\Varid{h_{L}}\;(\Varid{alg_{S}}\;(\Conid{Get}\;\Varid{k})){}\<[E]%
\\
\>[3]{}\mathrel{=}\mbox{\commentbegin ~  definition of \ensuremath{\Varid{alg_{S}}}  \commentend}{}\<[E]%
\\
\>[3]{}\hsindent{2}{}\<[5]%
\>[5]{}\Varid{h_{L}}\;(\lambda \Varid{s}\to \Varid{k}\;\Varid{s}\;\Varid{s}){}\<[E]%
\\
\>[3]{}\mathrel{=}\mbox{\commentbegin ~  definition of \ensuremath{\Varid{h_{L}}}  \commentend}{}\<[E]%
\\
\>[3]{}\hsindent{2}{}\<[5]%
\>[5]{}\Varid{fmap}\;(\Varid{fmap}\;(\Varid{fmap}\;\Varid{fst})\hsdot{\circ }{.}\Varid{h_{ND+f}})\;(\lambda \Varid{s}\to \Varid{k}\;\Varid{s}\;\Varid{s}){}\<[E]%
\\
\>[3]{}\mathrel{=}\mbox{\commentbegin ~  definition of \ensuremath{\Varid{fmap}}  \commentend}{}\<[E]%
\\
\>[3]{}\hsindent{2}{}\<[5]%
\>[5]{}\lambda \Varid{s}\to \Varid{fmap}\;(\Varid{fmap}\;\Varid{fst})\;(\Varid{h_{ND+f}}\;(\Varid{k}\;\Varid{s}\;\Varid{s})){}\<[E]%
\\
\>[3]{}\mathrel{=}\mbox{\commentbegin ~  beta-expansion (twice)  \commentend}{}\<[E]%
\\
\>[3]{}\mathrel{=}\lambda \Varid{s}\to (\lambda \Varid{s}_{1}\;\Varid{s}_{2}\to \Varid{fmap}\;(\Varid{fmap}\;\Varid{fst})\;(\Varid{h_{ND+f}}\;(\Varid{k}\;\Varid{s}_{2}\;\Varid{s}_{1})))\;\Varid{s}\;\Varid{s}{}\<[E]%
\\
\>[3]{}\mathrel{=}\mbox{\commentbegin ~  definition of \ensuremath{\Varid{fmap}} (twice)  \commentend}{}\<[E]%
\\
\>[3]{}\mathrel{=}\lambda \Varid{s}\to (\Varid{fmap}\;(\Varid{fmap}\;(\Varid{fmap}\;(\Varid{fmap}\;\Varid{fst})\hsdot{\circ }{.}\Varid{h_{ND+f}}))\;(\lambda \Varid{s}_{1}\;\Varid{s}_{2}\to \Varid{k}\;\Varid{s}_{2}\;\Varid{s}_{1}))\;\Varid{s}\;\Varid{s}{}\<[E]%
\\
\>[3]{}\mathrel{=}\mbox{\commentbegin ~  eta-expansion of \ensuremath{\Varid{k}}  \commentend}{}\<[E]%
\\
\>[3]{}\mathrel{=}\lambda \Varid{s}\to (\Varid{fmap}\;(\Varid{fmap}\;(\Varid{fmap}\;(\Varid{fmap}\;\Varid{fst})\hsdot{\circ }{.}\Varid{h_{ND+f}}))\;\Varid{k})\;\Varid{s}\;\Varid{s}{}\<[E]%
\\
\>[3]{}\mathrel{=}\mbox{\commentbegin ~  definition of \ensuremath{\Varid{algRHS}}  \commentend}{}\<[E]%
\\
\>[3]{}\mathrel{=}\Varid{alg}_{\Varid{RHS}}^{\Varid{S}}\;(\Conid{Get}\;(\Varid{fmap}\;(\Varid{fmap}\;(\Varid{fmap}\;(\Varid{fmap}\;\Varid{fst})\hsdot{\circ }{.}\Varid{h_{ND+f}}))\;\Varid{k})){}\<[E]%
\\
\>[3]{}\mathrel{=}\mbox{\commentbegin ~  definition of \ensuremath{\Varid{fmap}}  \commentend}{}\<[E]%
\\
\>[3]{}\mathrel{=}\Varid{alg}_{\Varid{RHS}}^{\Varid{S}}\;(\Varid{fmap}\;(\Varid{fmap}\;(\Varid{fmap}\;(\Varid{fmap}\;\Varid{fst})\hsdot{\circ }{.}\Varid{h_{ND+f}}))\;(\Conid{Get}\;\Varid{k})){}\<[E]%
\\
\>[3]{}\mathrel{=}\mbox{\commentbegin ~  definition of \ensuremath{\Varid{h_{L}}}  \commentend}{}\<[E]%
\\
\>[3]{}\mathrel{=}\Varid{alg}_{\Varid{RHS}}^{\Varid{S}}\;(\Varid{fmap}\;\Varid{h_{L}}\;(\Conid{Get}\;\Varid{k})){}\<[E]%
\ColumnHook
\end{hscode}\resethooks
\indentend \noindent \mbox{\underline{case \ensuremath{\Varid{t}\mathrel{=}\Conid{MUpdate}\;\Varid{r}\;\Varid{k}}}}\indentbegin \begin{hscode}\SaveRestoreHook
\column{B}{@{}>{\hspre}l<{\hspost}@{}}%
\column{3}{@{}>{\hspre}l<{\hspost}@{}}%
\column{5}{@{}>{\hspre}l<{\hspost}@{}}%
\column{E}{@{}>{\hspre}l<{\hspost}@{}}%
\>[5]{}\Varid{h_{L}}\;(\Varid{alg_{S}}\;(\Conid{MUpdate}\;\Varid{r}\;\Varid{k})){}\<[E]%
\\
\>[3]{}\mathrel{=}\mbox{\commentbegin ~  definition of \ensuremath{\Varid{alg_{S}}}  \commentend}{}\<[E]%
\\
\>[3]{}\hsindent{2}{}\<[5]%
\>[5]{}\Varid{h_{L}}\;(\lambda \Varid{s}\to \Varid{k}\;(\Varid{s}\mathbin{\oplus}\Varid{r})){}\<[E]%
\\
\>[3]{}\mathrel{=}\mbox{\commentbegin ~  definition of \ensuremath{\Varid{h_{L}}}  \commentend}{}\<[E]%
\\
\>[3]{}\hsindent{2}{}\<[5]%
\>[5]{}\Varid{fmap}\;(\Varid{fmap}\;(\Varid{fmap}\;\Varid{fst})\hsdot{\circ }{.}\Varid{h_{ND+f}})\;(\lambda \Varid{s}\to \Varid{k}\;(\Varid{s}\mathbin{\oplus}\Varid{r})){}\<[E]%
\\
\>[3]{}\mathrel{=}\mbox{\commentbegin ~  definition of \ensuremath{\Varid{fmap}}  \commentend}{}\<[E]%
\\
\>[3]{}\hsindent{2}{}\<[5]%
\>[5]{}\lambda \Varid{s}\to \Varid{fmap}\;(\Varid{fmap}\;\Varid{fst})\;(\Varid{h_{ND+f}}\;(\Varid{k}\;(\Varid{s}\mathbin{\oplus}\Varid{r}))){}\<[E]%
\\
\>[3]{}\mathrel{=}\mbox{\commentbegin ~  beta-expansion  \commentend}{}\<[E]%
\\
\>[3]{}\mathrel{=}\lambda \Varid{s}\to (\lambda \Varid{s}_{1}\to \Varid{fmap}\;(\Varid{fmap}\;\Varid{fst})\;(\Varid{h_{ND+f}}\;(\Varid{k}\;\Varid{s}_{1})))\;(\Varid{s}\mathbin{\oplus}\Varid{r}){}\<[E]%
\\
\>[3]{}\mathrel{=}\mbox{\commentbegin ~  definition of \ensuremath{\Varid{fmap}}  \commentend}{}\<[E]%
\\
\>[3]{}\mathrel{=}\lambda \Varid{s}\to (\Varid{fmap}\;(\Varid{fmap}\;(\Varid{fmap}\;\Varid{fst})\hsdot{\circ }{.}\Varid{h_{ND+f}})\;(\lambda \Varid{s}_{1}\to \Varid{k}\;\Varid{s}_{1}))\;(\Varid{s}\mathbin{\oplus}\Varid{r}){}\<[E]%
\\
\>[3]{}\mathrel{=}\mbox{\commentbegin ~  eta-expansion of \ensuremath{\Varid{k}}  \commentend}{}\<[E]%
\\
\>[3]{}\mathrel{=}\lambda \Varid{s}\to (\Varid{fmap}\;(\Varid{fmap}\;(\Varid{fmap}\;\Varid{fst})\hsdot{\circ }{.}\Varid{h_{ND+f}})\;\Varid{k})\;(\Varid{s}\mathbin{\oplus}\Varid{r}){}\<[E]%
\\
\>[3]{}\mathrel{=}\mbox{\commentbegin ~  definition of \ensuremath{\Varid{alg}_{\Varid{RHS}}^{\Varid{S}}}  \commentend}{}\<[E]%
\\
\>[3]{}\mathrel{=}\Varid{alg}_{\Varid{RHS}}^{\Varid{S}}\;(\Conid{MUpdate}\;\Varid{r}\;(\Varid{fmap}\;(\Varid{fmap}\;(\Varid{fmap}\;\Varid{fst})\hsdot{\circ }{.}\Varid{h_{ND+f}})\;\Varid{k})){}\<[E]%
\\
\>[3]{}\mathrel{=}\mbox{\commentbegin ~  definition of \ensuremath{\Varid{fmap}}  \commentend}{}\<[E]%
\\
\>[3]{}\mathrel{=}\Varid{alg}_{\Varid{RHS}}^{\Varid{S}}\;(\Varid{fmap}\;(\Varid{fmap}\;(\Varid{fmap}\;\Varid{fst})\hsdot{\circ }{.}\Varid{h_{ND+f}}))\;(\Conid{MUpdate}\;\Varid{r}\;\Varid{k})){}\<[E]%
\\
\>[3]{}\mathrel{=}\mbox{\commentbegin ~  definition of \ensuremath{\Varid{h_{L}}}  \commentend}{}\<[E]%
\\
\>[3]{}\mathrel{=}\Varid{alg}_{\Varid{RHS}}^{\Varid{S}}\;(\Varid{fmap}\;\Varid{h_{L}}\;(\Conid{MUpdate}\;\Varid{r}\;\Varid{k})){}\<[E]%
\ColumnHook
\end{hscode}\resethooks
\indentend \noindent \mbox{\underline{case \ensuremath{\Varid{t}\mathrel{=}\Conid{MRestore}\;\Varid{r}\;\Varid{k}}}}\indentbegin \begin{hscode}\SaveRestoreHook
\column{B}{@{}>{\hspre}l<{\hspost}@{}}%
\column{3}{@{}>{\hspre}l<{\hspost}@{}}%
\column{5}{@{}>{\hspre}l<{\hspost}@{}}%
\column{E}{@{}>{\hspre}l<{\hspost}@{}}%
\>[5]{}\Varid{h_{L}}\;(\Varid{alg_{S}}\;(\Conid{MRestore}\;\Varid{r}\;\Varid{k})){}\<[E]%
\\
\>[3]{}\mathrel{=}\mbox{\commentbegin ~  definition of \ensuremath{\Varid{alg_{S}}}  \commentend}{}\<[E]%
\\
\>[3]{}\hsindent{2}{}\<[5]%
\>[5]{}\Varid{h_{L}}\;(\lambda \Varid{s}\to \Varid{k}\;(\Varid{s}\mathbin{\ominus}\Varid{r})){}\<[E]%
\\
\>[3]{}\mathrel{=}\mbox{\commentbegin ~  definition of \ensuremath{\Varid{h_{L}}}  \commentend}{}\<[E]%
\\
\>[3]{}\hsindent{2}{}\<[5]%
\>[5]{}\Varid{fmap}\;(\Varid{fmap}\;(\Varid{fmap}\;\Varid{fst})\hsdot{\circ }{.}\Varid{h_{ND+f}})\;(\lambda \Varid{s}\to \Varid{k}\;(\Varid{s}\mathbin{\ominus}\Varid{r})){}\<[E]%
\\
\>[3]{}\mathrel{=}\mbox{\commentbegin ~  definition of \ensuremath{\Varid{fmap}}  \commentend}{}\<[E]%
\\
\>[3]{}\hsindent{2}{}\<[5]%
\>[5]{}\lambda \Varid{s}\to \Varid{fmap}\;(\Varid{fmap}\;\Varid{fst})\;(\Varid{h_{ND+f}}\;(\Varid{k}\;(\Varid{s}\mathbin{\ominus}\Varid{r}))){}\<[E]%
\\
\>[3]{}\mathrel{=}\mbox{\commentbegin ~  beta-expansion  \commentend}{}\<[E]%
\\
\>[3]{}\mathrel{=}\lambda \Varid{s}\to (\lambda \Varid{s}_{1}\to \Varid{fmap}\;(\Varid{fmap}\;\Varid{fst})\;(\Varid{h_{ND+f}}\;(\Varid{k}\;\Varid{s}_{1})))\;(\Varid{s}\mathbin{\ominus}\Varid{r}){}\<[E]%
\\
\>[3]{}\mathrel{=}\mbox{\commentbegin ~  definition of \ensuremath{\Varid{fmap}}  \commentend}{}\<[E]%
\\
\>[3]{}\mathrel{=}\lambda \Varid{s}\to (\Varid{fmap}\;(\Varid{fmap}\;(\Varid{fmap}\;\Varid{fst})\hsdot{\circ }{.}\Varid{h_{ND+f}})\;(\lambda \Varid{s}_{1}\to \Varid{k}\;\Varid{s}_{1}))\;(\Varid{s}\mathbin{\ominus}\Varid{r}){}\<[E]%
\\
\>[3]{}\mathrel{=}\mbox{\commentbegin ~  eta-expansion of \ensuremath{\Varid{k}}  \commentend}{}\<[E]%
\\
\>[3]{}\mathrel{=}\lambda \Varid{s}\to (\Varid{fmap}\;(\Varid{fmap}\;(\Varid{fmap}\;\Varid{fst})\hsdot{\circ }{.}\Varid{h_{ND+f}})\;\Varid{k})\;(\Varid{s}\mathbin{\ominus}\Varid{r}){}\<[E]%
\\
\>[3]{}\mathrel{=}\mbox{\commentbegin ~  definition of \ensuremath{\Varid{alg}_{\Varid{RHS}}^{\Varid{S}}}  \commentend}{}\<[E]%
\\
\>[3]{}\mathrel{=}\Varid{alg}_{\Varid{RHS}}^{\Varid{S}}\;(\Conid{MRestore}\;\Varid{r}\;(\Varid{fmap}\;(\Varid{fmap}\;(\Varid{fmap}\;\Varid{fst})\hsdot{\circ }{.}\Varid{h_{ND+f}})\;\Varid{k})){}\<[E]%
\\
\>[3]{}\mathrel{=}\mbox{\commentbegin ~  definition of \ensuremath{\Varid{fmap}}  \commentend}{}\<[E]%
\\
\>[3]{}\mathrel{=}\Varid{alg}_{\Varid{RHS}}^{\Varid{S}}\;(\Varid{fmap}\;(\Varid{fmap}\;(\Varid{fmap}\;\Varid{fst})\hsdot{\circ }{.}\Varid{h_{ND+f}}))\;(\Conid{MRestore}\;\Varid{r}\;\Varid{k})){}\<[E]%
\\
\>[3]{}\mathrel{=}\mbox{\commentbegin ~  definition of \ensuremath{\Varid{h_{L}}}  \commentend}{}\<[E]%
\\
\>[3]{}\mathrel{=}\Varid{alg}_{\Varid{RHS}}^{\Varid{S}}\;(\Varid{fmap}\;\Varid{h_{L}}\;(\Conid{MRestore}\;\Varid{r}\;\Varid{k})){}\<[E]%
\ColumnHook
\end{hscode}\resethooks
\indentend %

For the subcondition \refd{}, we define \ensuremath{\Varid{alg}_{\Varid{RHS}}^{\Varid{ND}}} as follows.\indentbegin \begin{hscode}\SaveRestoreHook
\column{B}{@{}>{\hspre}l<{\hspost}@{}}%
\column{3}{@{}>{\hspre}l<{\hspost}@{}}%
\column{22}{@{}>{\hspre}l<{\hspost}@{}}%
\column{E}{@{}>{\hspre}l<{\hspost}@{}}%
\>[3]{}\Varid{alg}_{\Varid{RHS}}^{\Varid{ND}}\mathbin{::}\Conid{Functor}\;\Varid{f}\Rightarrow \Varid{Nondet_{F}}\;(\Varid{s}\to \Conid{Free}\;\Varid{f}\;[\mskip1.5mu \Varid{a}\mskip1.5mu])\to (\Varid{s}\to \Conid{Free}\;\Varid{f}\;[\mskip1.5mu \Varid{a}\mskip1.5mu]){}\<[E]%
\\
\>[3]{}\Varid{alg}_{\Varid{RHS}}^{\Varid{ND}}\;\Conid{Fail}{}\<[22]%
\>[22]{}\mathrel{=}\lambda \Varid{s}\to \Conid{Var}\;[\mskip1.5mu \mskip1.5mu]{}\<[E]%
\\
\>[3]{}\Varid{alg}_{\Varid{RHS}}^{\Varid{ND}}\;(\Conid{Or}\;\Varid{p}\;\Varid{q}){}\<[22]%
\>[22]{}\mathrel{=}\lambda \Varid{s}\to \Varid{liftM2}\;(+\!\!+)\;(\Varid{p}\;\Varid{s})\;(\Varid{q}\;\Varid{s}){}\<[E]%
\ColumnHook
\end{hscode}\resethooks
\indentend To show its correctness, given \ensuremath{\Varid{op}\mathbin{::}\Varid{Nondet_{F}}\;(\Varid{s}\to \Conid{Free}\;(\Varid{Nondet_{F}}\mathrel{{:}{+}{:}}\Varid{f})\;(\Varid{a},\Varid{s}))} with \ensuremath{\Conid{Functor}\;\Varid{f}}, we calculate:
\indentbegin \begin{hscode}\SaveRestoreHook
\column{B}{@{}>{\hspre}l<{\hspost}@{}}%
\column{3}{@{}>{\hspre}l<{\hspost}@{}}%
\column{5}{@{}>{\hspre}l<{\hspost}@{}}%
\column{E}{@{}>{\hspre}l<{\hspost}@{}}%
\>[5]{}\Varid{h_{L}}\;(\Varid{fwd_{S}}\;(\Conid{Inl}\;\Varid{op})){}\<[E]%
\\
\>[3]{}\mathrel{=}\mbox{\commentbegin ~ definition of \ensuremath{\Varid{fwd_{S}}}  \commentend}{}\<[E]%
\\
\>[3]{}\hsindent{2}{}\<[5]%
\>[5]{}\Varid{h_{L}}\;(\lambda \Varid{s}\to \Conid{Op}\;(\Varid{fmap}\;(\mathbin{\$}\Varid{s})\;(\Conid{Inl}\;\Varid{op}))){}\<[E]%
\\
\>[3]{}\mathrel{=}\mbox{\commentbegin ~ definition of \ensuremath{\Varid{fmap}}  \commentend}{}\<[E]%
\\
\>[3]{}\hsindent{2}{}\<[5]%
\>[5]{}\Varid{h_{L}}\;(\lambda \Varid{s}\to \Conid{Op}\;(\Conid{Inl}\;(\Varid{fmap}\;(\mathbin{\$}\Varid{s})\;\Varid{op}))){}\<[E]%
\\
\>[3]{}\mathrel{=}\mbox{\commentbegin ~ definition of \ensuremath{\Varid{h_{L}}}  \commentend}{}\<[E]%
\\
\>[3]{}\hsindent{2}{}\<[5]%
\>[5]{}\Varid{fmap}\;(\Varid{fmap}\;(\Varid{fmap}\;\Varid{fst})\hsdot{\circ }{.}\Varid{h_{ND+f}})\;(\lambda \Varid{s}\to \Conid{Op}\;(\Conid{Inl}\;(\Varid{fmap}\;(\mathbin{\$}\Varid{s})\;\Varid{op}))){}\<[E]%
\\
\>[3]{}\mathrel{=}\mbox{\commentbegin ~ definition of \ensuremath{\Varid{fmap}}  \commentend}{}\<[E]%
\\
\>[3]{}\hsindent{2}{}\<[5]%
\>[5]{}\lambda \Varid{s}\to \Varid{fmap}\;(\Varid{fmap}\;\Varid{fst})\;(\Varid{h_{ND+f}}\;(\Conid{Op}\;(\Conid{Inl}\;(\Varid{fmap}\;(\mathbin{\$}\Varid{s})\;\Varid{op})))){}\<[E]%
\\
\>[3]{}\mathrel{=}\mbox{\commentbegin ~ definition of \ensuremath{\Varid{h_{ND+f}}}  \commentend}{}\<[E]%
\\
\>[3]{}\hsindent{2}{}\<[5]%
\>[5]{}\lambda \Varid{s}\to \Varid{fmap}\;(\Varid{fmap}\;\Varid{fst})\;(\Varid{alg_{ND+f}}\;(\Varid{fmap}\;\Varid{h_{ND+f}}\;(\Varid{fmap}\;(\mathbin{\$}\Varid{s})\;\Varid{op}))){}\<[E]%
\ColumnHook
\end{hscode}\resethooks
\indentend 
We proceed by a case analysis on \ensuremath{\Varid{op}}:

\noindent \mbox{\underline{case \ensuremath{\Varid{op}\mathrel{=}\Conid{Fail}}}}\indentbegin \begin{hscode}\SaveRestoreHook
\column{B}{@{}>{\hspre}l<{\hspost}@{}}%
\column{3}{@{}>{\hspre}l<{\hspost}@{}}%
\column{5}{@{}>{\hspre}l<{\hspost}@{}}%
\column{E}{@{}>{\hspre}l<{\hspost}@{}}%
\>[5]{}\lambda \Varid{s}\to \Varid{fmap}\;(\Varid{fmap}\;\Varid{fst})\;(\Varid{alg_{ND+f}}\;(\Varid{fmap}\;\Varid{h_{ND+f}}\;(\Varid{fmap}\;(\mathbin{\$}\Varid{s})\;\Conid{Fail}))){}\<[E]%
\\
\>[3]{}\mathrel{=}\mbox{\commentbegin ~ defintion of \ensuremath{\Varid{fmap}} (twice)  \commentend}{}\<[E]%
\\
\>[3]{}\hsindent{2}{}\<[5]%
\>[5]{}\lambda \Varid{s}\to \Varid{fmap}\;(\Varid{fmap}\;\Varid{fst})\;(\Varid{alg_{ND+f}}\;\Conid{Fail}){}\<[E]%
\\
\>[3]{}\mathrel{=}\mbox{\commentbegin ~ definition of \ensuremath{\Varid{alg_{ND+f}}}  \commentend}{}\<[E]%
\\
\>[3]{}\hsindent{2}{}\<[5]%
\>[5]{}\lambda \Varid{s}\to \Varid{fmap}\;(\Varid{fmap}\;\Varid{fst})\;(\Conid{Var}\;[\mskip1.5mu \mskip1.5mu]){}\<[E]%
\\
\>[3]{}\mathrel{=}\mbox{\commentbegin ~ definition of \ensuremath{\Varid{fmap}} (twice)  \commentend}{}\<[E]%
\\
\>[3]{}\hsindent{2}{}\<[5]%
\>[5]{}\lambda \Varid{s}\to \Conid{Var}\;[\mskip1.5mu \mskip1.5mu]{}\<[E]%
\\
\>[3]{}\mathrel{=}\mbox{\commentbegin ~ definition of \ensuremath{\Varid{alg}_{\Varid{RHS}}^{\Varid{ND}}}  \commentend}{}\<[E]%
\\
\>[3]{}\hsindent{2}{}\<[5]%
\>[5]{}\Varid{alg}_{\Varid{RHS}}^{\Varid{ND}}\;\Conid{Fail}{}\<[E]%
\\
\>[3]{}\mathrel{=}\mbox{\commentbegin  definition of \ensuremath{\Varid{fmap}}  \commentend}{}\<[E]%
\\
\>[3]{}\hsindent{2}{}\<[5]%
\>[5]{}\Varid{alg}_{\Varid{RHS}}^{\Varid{ND}}\;(\Varid{fmap}\;\Varid{h_{L}}\;\Varid{fail}){}\<[E]%
\ColumnHook
\end{hscode}\resethooks
\indentend 
\noindent \mbox{\underline{case \ensuremath{\Varid{op}\mathrel{=}\Conid{Or}\;\Varid{p}\;\Varid{q}}}}\indentbegin \begin{hscode}\SaveRestoreHook
\column{B}{@{}>{\hspre}l<{\hspost}@{}}%
\column{3}{@{}>{\hspre}l<{\hspost}@{}}%
\column{5}{@{}>{\hspre}l<{\hspost}@{}}%
\column{E}{@{}>{\hspre}l<{\hspost}@{}}%
\>[5]{}\lambda \Varid{s}\to \Varid{fmap}\;(\Varid{fmap}\;\Varid{fst})\;(\Varid{alg_{ND+f}}\;(\Varid{fmap}\;\Varid{h_{ND+f}}\;(\Varid{fmap}\;(\mathbin{\$}\Varid{s})\;(\Conid{Or}\;\Varid{p}\;\Varid{q})))){}\<[E]%
\\
\>[3]{}\mathrel{=}\mbox{\commentbegin ~ defintion of \ensuremath{\Varid{fmap}} (twice)  \commentend}{}\<[E]%
\\
\>[3]{}\hsindent{2}{}\<[5]%
\>[5]{}\lambda \Varid{s}\to \Varid{fmap}\;(\Varid{fmap}\;\Varid{fst})\;(\Varid{alg_{ND+f}}\;(\Conid{Or}\;(\Varid{h_{ND+f}}\;(\Varid{p}\;\Varid{s}))\;(\Varid{h_{ND+f}}\;(\Varid{q}\;\Varid{s})))){}\<[E]%
\\
\>[3]{}\mathrel{=}\mbox{\commentbegin ~ definition of \ensuremath{\Varid{alg_{ND+f}}}  \commentend}{}\<[E]%
\\
\>[3]{}\hsindent{2}{}\<[5]%
\>[5]{}\lambda \Varid{s}\to \Varid{fmap}\;(\Varid{fmap}\;\Varid{fst})\;(\Varid{liftM2}\;(+\!\!+)\;(\Varid{h_{ND+f}}\;(\Varid{p}\;\Varid{s}))\;(\Varid{h_{ND+f}}\;(\Varid{q}\;\Varid{s}))){}\<[E]%
\\
\>[3]{}\mathrel{=}\mbox{\commentbegin ~ Lemma~\ref{eq:liftM2-fst-comm}  \commentend}{}\<[E]%
\\
\>[3]{}\hsindent{2}{}\<[5]%
\>[5]{}\lambda \Varid{s}\to \Varid{liftM2}\;(+\!\!+)\;(\Varid{fmap}\;(\Varid{fmap}\;\Varid{fst})\;(\Varid{h_{ND+f}}\;(\Varid{p}\;\Varid{s})))\;(\Varid{fmap}\;(\Varid{fmap}\;\Varid{fst})\;(\Varid{h_{ND+f}}\;(\Varid{q}\;\Varid{s}))){}\<[E]%
\\
\>[3]{}\mathrel{=}\mbox{\commentbegin ~ definition of \ensuremath{\Varid{alg}_{\Varid{RHS}}^{\Varid{ND}}}  \commentend}{}\<[E]%
\\
\>[3]{}\hsindent{2}{}\<[5]%
\>[5]{}\Varid{alg}_{\Varid{RHS}}^{\Varid{ND}}\;(\Conid{Or}\;(\Varid{fmap}\;(\Varid{fmap}\;\Varid{fst})\hsdot{\circ }{.}\Varid{h_{ND+f}}\hsdot{\circ }{.}\Varid{p})\;(\Varid{fmap}\;(\Varid{fmap}\;\Varid{fst})\hsdot{\circ }{.}\Varid{h_{ND+f}}\hsdot{\circ }{.}\Varid{q})){}\<[E]%
\\
\>[3]{}\mathrel{=}\mbox{\commentbegin ~ defintion of \ensuremath{\Varid{fmap}} (twice)  \commentend}{}\<[E]%
\\
\>[3]{}\hsindent{2}{}\<[5]%
\>[5]{}\Varid{alg}_{\Varid{RHS}}^{\Varid{ND}}\;(\Varid{fmap}\;(\Varid{fmap}\;(\Varid{fmap}\;(\Varid{fmap}\;\Varid{fst})\hsdot{\circ }{.}\Varid{h_{ND+f}}))\;(\Conid{Or}\;\Varid{p}\;\Varid{q})){}\<[E]%
\\
\>[3]{}\mathrel{=}\mbox{\commentbegin ~ defintion of \ensuremath{\Varid{h_{L}}}  \commentend}{}\<[E]%
\\
\>[3]{}\hsindent{2}{}\<[5]%
\>[5]{}\Varid{alg}_{\Varid{RHS}}^{\Varid{ND}}\;(\Varid{fmap}\;\Varid{h_{L}}\;(\Conid{Or}\;\Varid{p}\;\Varid{q})){}\<[E]%
\ColumnHook
\end{hscode}\resethooks
\indentend %

For the last subcondition \refe{}, we define \ensuremath{\Varid{fwd}_{\Varid{RHS}}} as follows.\indentbegin \begin{hscode}\SaveRestoreHook
\column{B}{@{}>{\hspre}l<{\hspost}@{}}%
\column{3}{@{}>{\hspre}l<{\hspost}@{}}%
\column{E}{@{}>{\hspre}l<{\hspost}@{}}%
\>[3]{}\Varid{fwd}_{\Varid{RHS}}\mathbin{::}\Conid{Functor}\;\Varid{f}\Rightarrow \Varid{f}\;(\Varid{s}\to \Conid{Free}\;\Varid{f}\;[\mskip1.5mu \Varid{a}\mskip1.5mu])\to (\Varid{s}\to \Conid{Free}\;\Varid{f}\;[\mskip1.5mu \Varid{a}\mskip1.5mu]){}\<[E]%
\\
\>[3]{}\Varid{fwd}_{\Varid{RHS}}\;\Varid{op}\mathrel{=}\lambda \Varid{s}\to \Conid{Op}\;(\Varid{fmap}\;(\mathbin{\$}\Varid{s})\;\Varid{op}){}\<[E]%
\ColumnHook
\end{hscode}\resethooks
\indentend To show its correctness, given input \ensuremath{\Varid{op}\mathbin{::}\Varid{f}\;(\Varid{s}\to \Conid{Free}\;(\Varid{Nondet_{F}}\mathrel{{:}{+}{:}}\Varid{f})\;(\Varid{a},\Varid{s}))}, we calculate:
\indentbegin \begin{hscode}\SaveRestoreHook
\column{B}{@{}>{\hspre}l<{\hspost}@{}}%
\column{3}{@{}>{\hspre}l<{\hspost}@{}}%
\column{5}{@{}>{\hspre}l<{\hspost}@{}}%
\column{E}{@{}>{\hspre}l<{\hspost}@{}}%
\>[5]{}\Varid{h_{L}}\;(\Varid{fwd_{S}}\;(\Conid{Inr}\;\Varid{op})){}\<[E]%
\\
\>[3]{}\mathrel{=}\mbox{\commentbegin ~ definition of \ensuremath{\Varid{fwd_{S}}}  \commentend}{}\<[E]%
\\
\>[3]{}\hsindent{2}{}\<[5]%
\>[5]{}\Varid{h_{L}}\;(\lambda \Varid{s}\to \Conid{Op}\;(\Varid{fmap}\;(\mathbin{\$}\Varid{s})\;(\Conid{Inr}\;\Varid{op}))){}\<[E]%
\\
\>[3]{}\mathrel{=}\mbox{\commentbegin ~ definition of \ensuremath{\Varid{fmap}}  \commentend}{}\<[E]%
\\
\>[3]{}\hsindent{2}{}\<[5]%
\>[5]{}\Varid{h_{L}}\;(\lambda \Varid{s}\to \Conid{Op}\;(\Conid{Inr}\;(\Varid{fmap}\;(\mathbin{\$}\Varid{s})\;\Varid{op}))){}\<[E]%
\\
\>[3]{}\mathrel{=}\mbox{\commentbegin ~ definition of \ensuremath{\Varid{h_{L}}}  \commentend}{}\<[E]%
\\
\>[3]{}\hsindent{2}{}\<[5]%
\>[5]{}\Varid{fmap}\;(\Varid{fmap}\;(\Varid{fmap}\;\Varid{fst})\hsdot{\circ }{.}\Varid{h_{ND+f}})\;(\lambda \Varid{s}\to \Conid{Op}\;(\Conid{Inr}\;(\Varid{fmap}\;(\mathbin{\$}\Varid{s})\;\Varid{op}))){}\<[E]%
\\
\>[3]{}\mathrel{=}\mbox{\commentbegin ~ definition of \ensuremath{\Varid{fmap}}  \commentend}{}\<[E]%
\\
\>[3]{}\hsindent{2}{}\<[5]%
\>[5]{}\lambda \Varid{s}\to \Varid{fmap}\;(\Varid{fmap}\;\Varid{fst})\;(\Varid{h_{ND+f}}\;(\Conid{Op}\;(\Conid{Inr}\;(\Varid{fmap}\;(\mathbin{\$}\Varid{s})\;\Varid{op})))){}\<[E]%
\\
\>[3]{}\mathrel{=}\mbox{\commentbegin ~ definition of \ensuremath{\Varid{h_{ND+f}}}  \commentend}{}\<[E]%
\\
\>[3]{}\hsindent{2}{}\<[5]%
\>[5]{}\lambda \Varid{s}\to \Varid{fmap}\;(\Varid{fmap}\;\Varid{fst})\;(\Varid{fwd_{ND+f}}\;(\Varid{fmap}\;\Varid{h_{ND+f}}\;(\Varid{fmap}\;(\mathbin{\$}\Varid{s})\;\Varid{op}))){}\<[E]%
\\
\>[3]{}\mathrel{=}\mbox{\commentbegin ~ definition of \ensuremath{\Varid{fwd_{ND+f}}}  \commentend}{}\<[E]%
\\
\>[3]{}\hsindent{2}{}\<[5]%
\>[5]{}\lambda \Varid{s}\to \Varid{fmap}\;(\Varid{fmap}\;\Varid{fst})\;(\Conid{Op}\;(\Varid{fmap}\;\Varid{h_{ND+f}}\;(\Varid{fmap}\;(\mathbin{\$}\Varid{s})\;\Varid{op}))){}\<[E]%
\\
\>[3]{}\mathrel{=}\mbox{\commentbegin ~ definition of \ensuremath{\Varid{fmap}}  \commentend}{}\<[E]%
\\
\>[3]{}\hsindent{2}{}\<[5]%
\>[5]{}\lambda \Varid{s}\to \Conid{Op}\;(\Varid{fmap}\;(\Varid{fmap}\;(\Varid{fmap}\;\Varid{fst}))\;(\Varid{fmap}\;\Varid{h_{ND+f}}\;(\Varid{fmap}\;(\mathbin{\$}\Varid{s})\;\Varid{op}))){}\<[E]%
\\
\>[3]{}\mathrel{=}\mbox{\commentbegin ~ \ensuremath{\Varid{fmap}} fusion  \commentend}{}\<[E]%
\\
\>[3]{}\hsindent{2}{}\<[5]%
\>[5]{}\lambda \Varid{s}\to \Conid{Op}\;(\Varid{fmap}\;(\Varid{fmap}\;(\Varid{fmap}\;\Varid{fst})\hsdot{\circ }{.}\Varid{h_{ND+f}})\;(\Varid{fmap}\;(\mathbin{\$}\Varid{s})\;\Varid{op})){}\<[E]%
\\
\>[3]{}\mathrel{=}\mbox{\commentbegin ~ definition of \ensuremath{\Varid{h_{L}}}  \commentend}\mid {}\<[E]%
\\
\>[3]{}\hsindent{2}{}\<[5]%
\>[5]{}\lambda \Varid{s}\to \Conid{Op}\;(\Varid{h_{L}}\;(\Varid{fmap}\;(\mathbin{\$}\Varid{s})\;\Varid{op})){}\<[E]%
\\
\>[3]{}\mathrel{=}\mbox{\commentbegin ~ \Cref{eq:comm-app-fmap}  \commentend}{}\<[E]%
\\
\>[3]{}\hsindent{2}{}\<[5]%
\>[5]{}\lambda \Varid{s}\to \Conid{Op}\;(\Varid{fmap}\;((\mathbin{\$}\Varid{s})\hsdot{\circ }{.}\Varid{h_{L}})\;\Varid{op}){}\<[E]%
\\
\>[3]{}\mathrel{=}\mbox{\commentbegin ~ \ensuremath{\Varid{fmap}} fusion  \commentend}{}\<[E]%
\\
\>[3]{}\hsindent{2}{}\<[5]%
\>[5]{}\lambda \Varid{s}\to \Conid{Op}\;(\Varid{fmap}\;(\mathbin{\$}\Varid{s})\;(\Varid{fmap}\;\Varid{h_{L}}\;\Varid{op})){}\<[E]%
\\
\>[3]{}\mathrel{=}\mbox{\commentbegin ~ definition of \ensuremath{\Varid{fwd}_{\Varid{RHS}}}  \commentend}{}\<[E]%
\\
\>[3]{}\hsindent{2}{}\<[5]%
\>[5]{}\Varid{fwd}_{\Varid{RHS}}\;(\Varid{fmap}\;\Varid{h_{L}}\;\Varid{op}){}\<[E]%
\ColumnHook
\end{hscode}\resethooks
\indentend

\subsection{Fusing the Left-Hand Side}
\label{app:modify-fusing-lhs}

We calculate as follows:
\indentbegin \begin{hscode}\SaveRestoreHook
\column{B}{@{}>{\hspre}l<{\hspost}@{}}%
\column{5}{@{}>{\hspre}l<{\hspost}@{}}%
\column{7}{@{}>{\hspre}l<{\hspost}@{}}%
\column{9}{@{}>{\hspre}l<{\hspost}@{}}%
\column{34}{@{}>{\hspre}l<{\hspost}@{}}%
\column{E}{@{}>{\hspre}l<{\hspost}@{}}%
\>[5]{}\Varid{h_{GlobalM}}\hsdot{\circ }{.}\Varid{local2global_M}{}\<[E]%
\\
\>[B]{}\mathrel{=}\mbox{\commentbegin ~  definition of \ensuremath{\Varid{local2global_M}}  \commentend}{}\<[E]%
\\
\>[B]{}\hsindent{5}{}\<[5]%
\>[5]{}\Varid{h_{GlobalM}}\hsdot{\circ }{.}\Varid{fold}\;\Conid{Var}\;\Varid{alg}{}\<[E]%
\\
\>[5]{}\hsindent{2}{}\<[7]%
\>[7]{}\mathbf{where}{}\<[E]%
\\
\>[7]{}\hsindent{2}{}\<[9]%
\>[9]{}\Varid{alg}\;(\Conid{Inl}\;(\Conid{MUpdate}\;\Varid{r}\;\Varid{k})){}\<[34]%
\>[34]{}\mathrel{=}(\Varid{update}\;\Varid{r}\mathbin{\talloblong}\Varid{side}\;(\Varid{restore}\;\Varid{r}))>\!\!>\Varid{k}{}\<[E]%
\\
\>[7]{}\hsindent{2}{}\<[9]%
\>[9]{}\Varid{alg}\;\Varid{p}{}\<[34]%
\>[34]{}\mathrel{=}\Conid{Op}\;\Varid{p}{}\<[E]%
\\
\>[B]{}\mathrel{=}\mbox{\commentbegin ~  fold fusion-post' (Equation \ref{eq:fusion-post-strong})   \commentend}{}\<[E]%
\\
\>[B]{}\hsindent{5}{}\<[5]%
\>[5]{}\Varid{fold}\;\Varid{gen}_{\Varid{LHS}}\;(\Varid{alg}_{\Varid{LHS}}^{\Varid{S}}\mathbin{\#}\Varid{alg}_{\Varid{LHS}}^{\Varid{ND}}\mathbin{\#}\Varid{fwd}_{\Varid{LHS}}){}\<[E]%
\ColumnHook
\end{hscode}\resethooks
\indentend 
This last step is valid provided that the fusion conditions are satisfied:
\[\ba{rclr}
\ensuremath{\Varid{h_{GlobalM}}\hsdot{\circ }{.}\Conid{Var}} & = & \ensuremath{\Varid{gen}_{\Varid{LHS}}} &\refa{}\\
\ea\]
\vspace{-\baselineskip}
\[\ba{rlr}
   &\ensuremath{\Varid{h_{GlobalM}}\hsdot{\circ }{.}\Varid{alg}\hsdot{\circ }{.}\Varid{fmap}\;\Varid{local2global_M}} \\
= & \ensuremath{(\Varid{alg}_{\Varid{LHS}}^{\Varid{S}}\mathbin{\#}\Varid{alg}_{\Varid{LHS}}^{\Varid{ND}}\mathbin{\#}\Varid{fwd}_{\Varid{LHS}})\hsdot{\circ }{.}\Varid{fmap}\;\Varid{h_{GlobalM}}\hsdot{\circ }{.}\Varid{fmap}\;\Varid{local2global_M}} &\refb{}
\ea\]

The first subcondition \refa{} is met by\indentbegin \begin{hscode}\SaveRestoreHook
\column{B}{@{}>{\hspre}l<{\hspost}@{}}%
\column{3}{@{}>{\hspre}l<{\hspost}@{}}%
\column{E}{@{}>{\hspre}l<{\hspost}@{}}%
\>[3]{}\Varid{gen}_{\Varid{LHS}}\mathbin{::}\Conid{Functor}\;\Varid{f}\Rightarrow \Varid{a}\to (\Varid{s}\to \Conid{Free}\;\Varid{f}\;[\mskip1.5mu \Varid{a}\mskip1.5mu]){}\<[E]%
\\
\>[3]{}\Varid{gen}_{\Varid{LHS}}\;\Varid{x}\mathrel{=}\lambda \Varid{s}\to \Conid{Var}\;[\mskip1.5mu \Varid{x}\mskip1.5mu]{}\<[E]%
\ColumnHook
\end{hscode}\resethooks
\indentend as established in the following calculation:\indentbegin \begin{hscode}\SaveRestoreHook
\column{B}{@{}>{\hspre}l<{\hspost}@{}}%
\column{3}{@{}>{\hspre}l<{\hspost}@{}}%
\column{5}{@{}>{\hspre}l<{\hspost}@{}}%
\column{E}{@{}>{\hspre}l<{\hspost}@{}}%
\>[5]{}\Varid{h_{GlobalM}}\;(\Conid{Var}\;\Varid{x}){}\<[E]%
\\
\>[3]{}\mathrel{=}\mbox{\commentbegin ~ definition of \ensuremath{\Varid{h_{GlobalM}}}  \commentend}{}\<[E]%
\\
\>[3]{}\hsindent{2}{}\<[5]%
\>[5]{}\Varid{fmap}\;(\Varid{fmap}\;\Varid{fst})\;(\Varid{h_{Modify1}}\;(\Varid{h_{ND+f}}\;(\Varid{(\Leftrightarrow)}\;(\Conid{Var}\;\Varid{x})))){}\<[E]%
\\
\>[3]{}\mathrel{=}\mbox{\commentbegin ~ definition of \ensuremath{\Varid{(\Leftrightarrow)}}  \commentend}{}\<[E]%
\\
\>[3]{}\hsindent{2}{}\<[5]%
\>[5]{}\Varid{fmap}\;(\Varid{fmap}\;\Varid{fst})\;(\Varid{h_{Modify1}}\;(\Varid{h_{ND+f}}\;(\Conid{Var}\;\Varid{x}))){}\<[E]%
\\
\>[3]{}\mathrel{=}\mbox{\commentbegin ~ definition of \ensuremath{\Varid{h_{ND+f}}}  \commentend}{}\<[E]%
\\
\>[3]{}\hsindent{2}{}\<[5]%
\>[5]{}\Varid{fmap}\;(\Varid{fmap}\;\Varid{fst})\;(\Varid{h_{Modify1}}\;(\Conid{Var}\;[\mskip1.5mu \Varid{x}\mskip1.5mu])){}\<[E]%
\\
\>[3]{}\mathrel{=}\mbox{\commentbegin ~ definition of \ensuremath{\Varid{h_{Modify1}}}  \commentend}{}\<[E]%
\\
\>[3]{}\hsindent{2}{}\<[5]%
\>[5]{}\Varid{fmap}\;(\Varid{fmap}\;\Varid{fst})\;(\lambda \Varid{s}\to \Conid{Var}\;([\mskip1.5mu \Varid{x}\mskip1.5mu],\Varid{s})){}\<[E]%
\\
\>[3]{}\mathrel{=}\mbox{\commentbegin ~ definition of \ensuremath{\Varid{fmap}} (twice)  \commentend}{}\<[E]%
\\
\>[3]{}\hsindent{2}{}\<[5]%
\>[5]{}\lambda \Varid{s}\to \Conid{Var}\;[\mskip1.5mu \Varid{x}\mskip1.5mu]{}\<[E]%
\\
\>[3]{}\mathrel{=}\mbox{\commentbegin ~ definition of \ensuremath{\Varid{gen}_{\Varid{LHS}}}  \commentend}{}\<[E]%
\\
\>[3]{}\hsindent{2}{}\<[5]%
\>[5]{}\Varid{gen}_{\Varid{LHS}}\;\Varid{x}{}\<[E]%
\ColumnHook
\end{hscode}\resethooks
\indentend 
We can split the second fusion condition \refb{} in three
subconditions:
\[\ba{rclr}
\ensuremath{\Varid{h_{GlobalM}}\hsdot{\circ }{.}\Varid{alg}\hsdot{\circ }{.}\Conid{Inl}\hsdot{\circ }{.}\Varid{fmap}\;\Varid{local2global_M}} & = & \ensuremath{\Varid{alg}_{\Varid{LHS}}^{\Varid{S}}\hsdot{\circ }{.}\Varid{fmap}\;\Varid{h_{GlobalM}}\hsdot{\circ }{.}\Varid{fmap}\;\Varid{local2global_M}} &\refc{}\\
\ensuremath{\Varid{h_{GlobalM}}\hsdot{\circ }{.}\Varid{alg}\hsdot{\circ }{.}\Conid{Inr}\hsdot{\circ }{.}\Conid{Inl}\hsdot{\circ }{.}\Varid{fmap}\;\Varid{local2global_M}} & = & \ensuremath{\Varid{alg}_{\Varid{LHS}}^{\Varid{ND}}\hsdot{\circ }{.}\Varid{fmap}\;\Varid{h_{GlobalM}}\hsdot{\circ }{.}\Varid{fmap}\;\Varid{local2global_M}} &\refd{}\\
\ensuremath{\Varid{h_{GlobalM}}\hsdot{\circ }{.}\Varid{alg}\hsdot{\circ }{.}\Conid{Inr}\hsdot{\circ }{.}\Conid{Inr}\hsdot{\circ }{.}\Varid{fmap}\;\Varid{local2global_M}} & = & \ensuremath{\Varid{fwd}_{\Varid{LHS}}\hsdot{\circ }{.}\Varid{fmap}\;\Varid{h_{GlobalM}}\hsdot{\circ }{.}\Varid{fmap}\;\Varid{local2global_M}} &\refe{}
\ea\]

For brevity, we omit the last common part \ensuremath{\Varid{fmap}\;\Varid{local2global_M}} of
these equations in the following proofs. Instead, we assume that the
input is in the codomain of \ensuremath{\Varid{fmap}\;\Varid{local2global_M}}.
Also, we use the condition in \Cref{thm:modify-local-global} that the
input program does not use the \ensuremath{\Varid{restore}} operation.

For the first subcondition \refc{}, we can define \ensuremath{\Varid{alg}_{\Varid{LHS}}^{\Varid{S}}} as follows.\indentbegin \begin{hscode}\SaveRestoreHook
\column{B}{@{}>{\hspre}l<{\hspost}@{}}%
\column{3}{@{}>{\hspre}l<{\hspost}@{}}%
\column{27}{@{}>{\hspre}l<{\hspost}@{}}%
\column{30}{@{}>{\hspre}l<{\hspost}@{}}%
\column{E}{@{}>{\hspre}l<{\hspost}@{}}%
\>[3]{}\Varid{alg}_{\Varid{LHS}}^{\Varid{S}}\mathbin{::}(\Conid{Functor}\;\Varid{f},\Conid{Undo}\;\Varid{s}\;\Varid{r})\Rightarrow \Varid{Modify_{F}}\;\Varid{s}\;\Varid{r}\;(\Varid{s}\to \Conid{Free}\;\Varid{f}\;[\mskip1.5mu \Varid{a}\mskip1.5mu])\to (\Varid{s}\to \Conid{Free}\;\Varid{f}\;[\mskip1.5mu \Varid{a}\mskip1.5mu]){}\<[E]%
\\
\>[3]{}\Varid{alg}_{\Varid{LHS}}^{\Varid{S}}\;(\Conid{MGet}\;\Varid{k}){}\<[27]%
\>[27]{}\mathrel{=}{}\<[30]%
\>[30]{}\lambda \Varid{s}\to \Varid{k}\;\Varid{s}\;\Varid{s}{}\<[E]%
\\
\>[3]{}\Varid{alg}_{\Varid{LHS}}^{\Varid{S}}\;(\Conid{MUpdate}\;\Varid{r}\;\Varid{k}){}\<[27]%
\>[27]{}\mathrel{=}\lambda \Varid{s}\to \Varid{k}\;(\Varid{s}\mathbin{\oplus}\Varid{r}){}\<[E]%
\\
\>[3]{}\Varid{alg}_{\Varid{LHS}}^{\Varid{S}}\;(\Conid{MRestore}\;\Varid{r}\;\Varid{k}){}\<[27]%
\>[27]{}\mathrel{=}\lambda \Varid{s}\to \Varid{k}\;(\Varid{s}\mathbin{\ominus}\Varid{r}){}\<[E]%
\ColumnHook
\end{hscode}\resethooks
\indentend We prove it by a case analysis on the shape of input \ensuremath{\Varid{op}\mathbin{::}\Varid{Modify_{F}}\;\Varid{s}\;\Varid{r}\;(\Conid{Free}\;(\Varid{Modify_{F}}\;\Varid{s}\;\Varid{r}\mathrel{{:}{+}{:}}\Varid{Nondet_{F}}\mathrel{{:}{+}{:}}\Varid{f})\;\Varid{a})}.
Note that we only need to consider the case that \ensuremath{\Varid{op}} is of form \ensuremath{\Conid{MGet}\;\Varid{k}} or \ensuremath{\Conid{MUpdate}\;\Varid{r}\;\Varid{k}} where \ensuremath{\Varid{restore}} is also not used in the
continuation \ensuremath{\Varid{k}}.

\vspace{0.5\lineskip}

\noindent \mbox{\underline{case \ensuremath{\Varid{op}\mathrel{=}\Conid{MGet}\;\Varid{k}}}}\indentbegin \begin{hscode}\SaveRestoreHook
\column{B}{@{}>{\hspre}l<{\hspost}@{}}%
\column{3}{@{}>{\hspre}l<{\hspost}@{}}%
\column{5}{@{}>{\hspre}l<{\hspost}@{}}%
\column{E}{@{}>{\hspre}l<{\hspost}@{}}%
\>[5]{}\Varid{h_{GlobalM}}\;(\Varid{alg}\;(\Conid{Inl}\;(\Conid{MGet}\;\Varid{k}))){}\<[E]%
\\
\>[3]{}\mathrel{=}\mbox{\commentbegin ~ definition of \ensuremath{\Varid{alg}}  \commentend}{}\<[E]%
\\
\>[3]{}\hsindent{2}{}\<[5]%
\>[5]{}\Varid{h_{GlobalM}}\;(\Conid{Op}\;(\Conid{Inl}\;(\Conid{MGet}\;\Varid{k}))){}\<[E]%
\\
\>[3]{}\mathrel{=}\mbox{\commentbegin ~ definition of \ensuremath{\Varid{h_{GlobalM}}}  \commentend}{}\<[E]%
\\
\>[3]{}\hsindent{2}{}\<[5]%
\>[5]{}\Varid{fmap}\;(\Varid{fmap}\;\Varid{fst})\;(\Varid{h_{Modify1}}\;(\Varid{h_{ND+f}}\;(\Varid{(\Leftrightarrow)}\;(\Conid{Op}\;(\Conid{Inl}\;(\Conid{MGet}\;\Varid{k})))))){}\<[E]%
\\
\>[3]{}\mathrel{=}\mbox{\commentbegin ~ definition of \ensuremath{\Varid{(\Leftrightarrow)}}  \commentend}{}\<[E]%
\\
\>[3]{}\hsindent{2}{}\<[5]%
\>[5]{}\Varid{fmap}\;(\Varid{fmap}\;\Varid{fst})\;(\Varid{h_{Modify1}}\;(\Varid{h_{ND+f}}\;(\Conid{Op}\;(\Conid{Inr}\;(\Conid{Inl}\;(\Varid{fmap}\;\Varid{(\Leftrightarrow)}\;(\Conid{MGet}\;\Varid{k}))))))){}\<[E]%
\\
\>[3]{}\mathrel{=}\mbox{\commentbegin ~ definition of \ensuremath{\Varid{fmap}}  \commentend}{}\<[E]%
\\
\>[3]{}\hsindent{2}{}\<[5]%
\>[5]{}\Varid{fmap}\;(\Varid{fmap}\;\Varid{fst})\;(\Varid{h_{Modify1}}\;(\Varid{h_{ND+f}}\;(\Conid{Op}\;(\Conid{Inr}\;(\Conid{Inl}\;(\Conid{MGet}\;(\Varid{(\Leftrightarrow)}\hsdot{\circ }{.}\Varid{k}))))))){}\<[E]%
\\
\>[3]{}\mathrel{=}\mbox{\commentbegin ~ definition of \ensuremath{\Varid{h_{ND+f}}}  \commentend}{}\<[E]%
\\
\>[3]{}\hsindent{2}{}\<[5]%
\>[5]{}\Varid{fmap}\;(\Varid{fmap}\;\Varid{fst})\;(\Varid{h_{Modify1}}\;(\Conid{Op}\;(\Varid{fmap}\;\Varid{h_{ND+f}}\;(\Conid{Inl}\;(\Conid{MGet}\;(\Varid{(\Leftrightarrow)}\hsdot{\circ }{.}\Varid{k})))))){}\<[E]%
\\
\>[3]{}\mathrel{=}\mbox{\commentbegin ~ definition of \ensuremath{\Varid{fmap}}  \commentend}{}\<[E]%
\\
\>[3]{}\hsindent{2}{}\<[5]%
\>[5]{}\Varid{fmap}\;(\Varid{fmap}\;\Varid{fst})\;(\Varid{h_{Modify1}}\;(\Conid{Op}\;(\Conid{Inl}\;(\Conid{MGet}\;(\Varid{h_{ND+f}}\hsdot{\circ }{.}\Varid{(\Leftrightarrow)}\hsdot{\circ }{.}\Varid{k}))))){}\<[E]%
\\
\>[3]{}\mathrel{=}\mbox{\commentbegin ~ definition of \ensuremath{\Varid{h_{Modify1}}}  \commentend}{}\<[E]%
\\
\>[3]{}\hsindent{2}{}\<[5]%
\>[5]{}\Varid{fmap}\;(\Varid{fmap}\;\Varid{fst})\;(\lambda \Varid{s}\to (\Varid{h_{Modify1}}\hsdot{\circ }{.}\Varid{h_{ND+f}}\hsdot{\circ }{.}\Varid{(\Leftrightarrow)}\hsdot{\circ }{.}\Varid{k})\;\Varid{s}\;\Varid{s}){}\<[E]%
\\
\>[3]{}\mathrel{=}\mbox{\commentbegin ~ definition of \ensuremath{\Varid{fmap}}  \commentend}{}\<[E]%
\\
\>[3]{}\hsindent{2}{}\<[5]%
\>[5]{}\lambda \Varid{s}\to \Varid{fmap}\;\Varid{fst}\;((\Varid{h_{Modify1}}\hsdot{\circ }{.}\Varid{h_{ND+f}}\hsdot{\circ }{.}\Varid{(\Leftrightarrow)}\hsdot{\circ }{.}\Varid{k})\;\Varid{s}\;\Varid{s}){}\<[E]%
\\
\>[3]{}\mathrel{=}\mbox{\commentbegin ~ definition of \ensuremath{\Varid{fmap}}  \commentend}{}\<[E]%
\\
\>[3]{}\hsindent{2}{}\<[5]%
\>[5]{}\lambda \Varid{s}\to ((\Varid{fmap}\;(\Varid{fmap}\;\Varid{fst})\hsdot{\circ }{.}\Varid{h_{Modify1}}\hsdot{\circ }{.}\Varid{h_{ND+f}}\hsdot{\circ }{.}\Varid{(\Leftrightarrow)}\hsdot{\circ }{.}\Varid{k})\;\Varid{s}\;\Varid{s}){}\<[E]%
\\
\>[3]{}\mathrel{=}\mbox{\commentbegin ~ definition of \ensuremath{\Varid{h_{GlobalM}}}  \commentend}{}\<[E]%
\\
\>[3]{}\hsindent{2}{}\<[5]%
\>[5]{}\lambda \Varid{s}\to (\Varid{h_{GlobalM}}\hsdot{\circ }{.}\Varid{k})\;\Varid{s}\;\Varid{s}{}\<[E]%
\\
\>[3]{}\mathrel{=}\mbox{\commentbegin ~ definition of \ensuremath{\Varid{alg}_{\Varid{LHS}}^{\Varid{S}}}  \commentend}{}\<[E]%
\\
\>[3]{}\hsindent{2}{}\<[5]%
\>[5]{}\Varid{alg}_{\Varid{LHS}}^{\Varid{S}}\;(\Conid{MGet}\;(\Varid{h_{GlobalM}}\hsdot{\circ }{.}\Varid{k})){}\<[E]%
\\
\>[3]{}\mathrel{=}\mbox{\commentbegin ~ definition of \ensuremath{\Varid{fmap}}  \commentend}{}\<[E]%
\\
\>[3]{}\hsindent{2}{}\<[5]%
\>[5]{}\Varid{alg}_{\Varid{LHS}}^{\Varid{S}}\;(\Varid{fmap}\;\Varid{h_{GlobalM}}\;(\Conid{MGet}\;\Varid{k})){}\<[E]%
\ColumnHook
\end{hscode}\resethooks
\indentend \noindent \mbox{\underline{case \ensuremath{\Varid{op}\mathrel{=}\Conid{MUpdate}\;\Varid{r}\;\Varid{k}}}}
From \ensuremath{\Varid{op}} is in the codomain of \ensuremath{\Varid{fmap}\;\Varid{local2global_M}} we obtain \ensuremath{\Varid{k}} is
in the codomain of \ensuremath{\Varid{local2global_M}}.
\indentbegin \begin{hscode}\SaveRestoreHook
\column{B}{@{}>{\hspre}l<{\hspost}@{}}%
\column{3}{@{}>{\hspre}l<{\hspost}@{}}%
\column{5}{@{}>{\hspre}l<{\hspost}@{}}%
\column{7}{@{}>{\hspre}l<{\hspost}@{}}%
\column{11}{@{}>{\hspre}l<{\hspost}@{}}%
\column{16}{@{}>{\hspre}l<{\hspost}@{}}%
\column{18}{@{}>{\hspre}l<{\hspost}@{}}%
\column{21}{@{}>{\hspre}l<{\hspost}@{}}%
\column{26}{@{}>{\hspre}l<{\hspost}@{}}%
\column{33}{@{}>{\hspre}l<{\hspost}@{}}%
\column{E}{@{}>{\hspre}l<{\hspost}@{}}%
\>[5]{}\Varid{h_{GlobalM}}\;(\Varid{alg}\;(\Conid{Inl}\;(\Conid{MUpdate}\;\Varid{r}\;\Varid{k}))){}\<[E]%
\\
\>[3]{}\mathrel{=}\mbox{\commentbegin ~ definition of \ensuremath{\Varid{alg}}  \commentend}{}\<[E]%
\\
\>[3]{}\hsindent{2}{}\<[5]%
\>[5]{}\Varid{h_{GlobalM}}\;((\Varid{update}\;\Varid{r}\mathbin{\talloblong}\Varid{side}\;(\Varid{restore}\;\Varid{r}))>\!\!>\Varid{k}){}\<[E]%
\\
\>[3]{}\mathrel{=}\mbox{\commentbegin ~ definitions of \ensuremath{\Varid{side}}, \ensuremath{\Varid{update}}, \ensuremath{\Varid{restore}}, \ensuremath{(\talloblong)}, and \ensuremath{(>\!\!>)}  \commentend}{}\<[E]%
\\
\>[3]{}\hsindent{2}{}\<[5]%
\>[5]{}\Varid{h_{GlobalM}}\;(\Conid{Op}\;(\Conid{Inr}\;(\Conid{Inl}\;(\Conid{Or}\;{}\<[33]%
\>[33]{}(\Conid{Op}\;(\Conid{Inl}\;(\Conid{MUpdate}\;\Varid{r}\;\Varid{k})))\;{}\<[E]%
\\
\>[33]{}(\Conid{Op}\;(\Conid{Inl}\;(\Conid{MRestore}\;\Varid{r}\;(\Conid{Op}\;(\Conid{Inr}\;(\Conid{Inl}\;\Conid{Fail})))))))))){}\<[E]%
\\
\>[3]{}\mathrel{=}\mbox{\commentbegin ~ definition of \ensuremath{\Varid{h_{GlobalM}}}  \commentend}{}\<[E]%
\\
\>[3]{}\hsindent{2}{}\<[5]%
\>[5]{}\Varid{fmap}\;(\Varid{fmap}\;\Varid{fst})\;(\Varid{h_{Modify1}}\;(\Varid{h_{ND+f}}\;(\Varid{(\Leftrightarrow)}{}\<[E]%
\\
\>[5]{}\hsindent{2}{}\<[7]%
\>[7]{}(\Conid{Op}\;(\Conid{Inr}\;(\Conid{Inl}\;(\Conid{Or}\;{}\<[26]%
\>[26]{}(\Conid{Op}\;(\Conid{Inl}\;(\Conid{MUpdate}\;\Varid{r}\;\Varid{k})))\;{}\<[E]%
\\
\>[26]{}(\Conid{Op}\;(\Conid{Inl}\;(\Conid{MRestore}\;\Varid{r}\;(\Conid{Op}\;\Conid{Inr}\;((\Conid{Inl}\;\Conid{Fail}))))))))))))){}\<[E]%
\\
\>[3]{}\mathrel{=}\mbox{\commentbegin ~ definition of \ensuremath{\Varid{(\Leftrightarrow)}}  \commentend}{}\<[E]%
\\
\>[3]{}\hsindent{2}{}\<[5]%
\>[5]{}\Varid{fmap}\;(\Varid{fmap}\;\Varid{fst})\;(\Varid{h_{Modify1}}\;(\Varid{h_{ND+f}}\;({}\<[E]%
\\
\>[5]{}\hsindent{2}{}\<[7]%
\>[7]{}(\Conid{Op}\;(\Conid{Inl}\;(\Conid{Or}\;{}\<[21]%
\>[21]{}(\Conid{Op}\;(\Conid{Inr}\;(\Conid{Inl}\;(\Conid{MUpdate}\;\Varid{r}\;(\Varid{(\Leftrightarrow)}\;\Varid{k})))))\;{}\<[E]%
\\
\>[21]{}(\Conid{Op}\;(\Conid{Inr}\;(\Conid{Inl}\;(\Conid{MRestore}\;\Varid{r}\;(\Conid{Op}\;(\Conid{Inl}\;\Conid{Fail})))))))))))){}\<[E]%
\\
\>[3]{}\mathrel{=}\mbox{\commentbegin ~ definition of \ensuremath{\Varid{h_{ND+f}}}  \commentend}{}\<[E]%
\\
\>[3]{}\hsindent{2}{}\<[5]%
\>[5]{}\Varid{fmap}\;(\Varid{fmap}\;\Varid{fst})\;(\Varid{h_{Modify1}}\;({}\<[E]%
\\
\>[5]{}\hsindent{2}{}\<[7]%
\>[7]{}(\Varid{liftM2}\;(+\!\!+)\;{}\<[21]%
\>[21]{}(\Conid{Op}\;(\Conid{Inl}\;(\Conid{MUpdate}\;\Varid{r}\;(\Varid{h_{ND+f}}\;(\Varid{(\Leftrightarrow)}\;\Varid{k})))))\;{}\<[E]%
\\
\>[21]{}(\Conid{Op}\;(\Conid{Inl}\;(\Conid{MRestore}\;\Varid{r}\;(\Conid{Var}\;[\mskip1.5mu \mskip1.5mu]))))))){}\<[E]%
\\
\>[3]{}\mathrel{=}\mbox{\commentbegin ~ definition of \ensuremath{\Varid{liftM2}}  \commentend}{}\<[E]%
\\
\>[3]{}\hsindent{2}{}\<[5]%
\>[5]{}\Varid{fmap}\;(\Varid{fmap}\;\Varid{fst})\;(\Varid{h_{Modify1}}\;({}\<[E]%
\\
\>[5]{}\hsindent{2}{}\<[7]%
\>[7]{}\mathbf{do}\;{}\<[11]%
\>[11]{}\Varid{x}\leftarrow \Conid{Op}\;(\Conid{Inl}\;(\Conid{MUpdate}\;\Varid{r}\;(\Varid{h_{ND+f}}\;(\Varid{(\Leftrightarrow)}\;\Varid{k})))){}\<[E]%
\\
\>[11]{}\Varid{y}\leftarrow \Conid{Op}\;(\Conid{Inl}\;(\Conid{MRestore}\;\Varid{r}\;(\Conid{Var}\;[\mskip1.5mu \mskip1.5mu]))){}\<[E]%
\\
\>[11]{}\Conid{Var}\;(\Varid{x}+\!\!+\Varid{y}){}\<[E]%
\\
\>[3]{}\hsindent{2}{}\<[5]%
\>[5]{})){}\<[E]%
\\
\>[3]{}\mathrel{=}\mbox{\commentbegin ~ \Cref{lemma:dist-hModify1}  \commentend}{}\<[E]%
\\
\>[3]{}\hsindent{2}{}\<[5]%
\>[5]{}\Varid{fmap}\;(\Varid{fmap}\;\Varid{fst})\;(\lambda \Varid{t}\to {}\<[E]%
\\
\>[5]{}\hsindent{2}{}\<[7]%
\>[7]{}\mathbf{do}\;{}\<[11]%
\>[11]{}(\Varid{x},\Varid{t}_{1})\leftarrow \Varid{h_{Modify1}}\;(\Conid{Op}\;(\Conid{Inl}\;(\Conid{MUpdate}\;\Varid{r}\;(\Varid{h_{ND+f}}\;(\Varid{(\Leftrightarrow)}\;\Varid{k})))))\;\Varid{t}{}\<[E]%
\\
\>[11]{}(\Varid{y},\Varid{t}_{2})\leftarrow \Varid{h_{Modify1}}\;(\Conid{Op}\;(\Conid{Inl}\;(\Conid{MRestore}\;\Varid{r}\;(\Conid{Var}\;[\mskip1.5mu \mskip1.5mu]))))\;\Varid{t}_{1}{}\<[E]%
\\
\>[11]{}\Varid{h_{Modify1}}\;(\Conid{Var}\;(\Varid{x}+\!\!+\Varid{y}))\;\Varid{t}_{2}{}\<[E]%
\\
\>[3]{}\hsindent{2}{}\<[5]%
\>[5]{}){}\<[E]%
\\
\>[3]{}\mathrel{=}\mbox{\commentbegin ~ definition of \ensuremath{\Varid{h_{Modify1}}}  \commentend}{}\<[E]%
\\
\>[3]{}\hsindent{2}{}\<[5]%
\>[5]{}\Varid{fmap}\;(\Varid{fmap}\;\Varid{fst})\;{}\<[E]%
\\
\>[5]{}\hsindent{2}{}\<[7]%
\>[7]{}(\lambda \Varid{t}\to \mathbf{do}\;{}\<[18]%
\>[18]{}(\Varid{x},\anonymous ){}\<[26]%
\>[26]{}\leftarrow \Varid{h_{Modify1}}\;(\Varid{h_{ND+f}}\;(\Varid{(\Leftrightarrow)}\;\Varid{k}))\;(\Varid{t}\mathbin{\oplus}\Varid{r}){}\<[E]%
\\
\>[18]{}(\Varid{y},\Varid{t}_{2}){}\<[26]%
\>[26]{}\leftarrow \Conid{Var}\;([\mskip1.5mu \mskip1.5mu],\Varid{t}\mathbin{\ominus}\Varid{r}){}\<[E]%
\\
\>[18]{}\Conid{Var}\;(\Varid{x}+\!\!+\Varid{y},\Varid{t}_{2}){}\<[E]%
\\
\>[5]{}\hsindent{2}{}\<[7]%
\>[7]{}){}\<[E]%
\\
\>[3]{}\mathrel{=}\mbox{\commentbegin ~ monad law  \commentend}{}\<[E]%
\\
\>[3]{}\hsindent{2}{}\<[5]%
\>[5]{}\Varid{fmap}\;(\Varid{fmap}\;\Varid{fst})\;{}\<[E]%
\\
\>[5]{}\hsindent{2}{}\<[7]%
\>[7]{}(\lambda \Varid{t}\to \mathbf{do}\;{}\<[18]%
\>[18]{}(\Varid{x},\Varid{t}_{1})\leftarrow \Varid{h_{Modify1}}\;(\Varid{h_{ND+f}}\;(\Varid{(\Leftrightarrow)}\;\Varid{k}))\;(\Varid{t}\mathbin{\oplus}\Varid{r}){}\<[E]%
\\
\>[18]{}\Conid{Var}\;(\Varid{x}+\!\!+[\mskip1.5mu \mskip1.5mu],\Varid{t}_{1}\mathbin{\ominus}\Varid{r}){}\<[E]%
\\
\>[5]{}\hsindent{2}{}\<[7]%
\>[7]{}){}\<[E]%
\\
\>[3]{}\mathrel{=}\mbox{\commentbegin ~ right unit of \ensuremath{(+\!\!+)}  \commentend}{}\<[E]%
\\
\>[3]{}\hsindent{2}{}\<[5]%
\>[5]{}\Varid{fmap}\;(\Varid{fmap}\;\Varid{fst})\;{}\<[E]%
\\
\>[5]{}\hsindent{2}{}\<[7]%
\>[7]{}(\lambda \Varid{t}\to \mathbf{do}\;{}\<[18]%
\>[18]{}(\Varid{x},\Varid{t}_{1})\leftarrow \Varid{h_{Modify1}}\;(\Varid{h_{ND+f}}\;(\Varid{(\Leftrightarrow)}\;\Varid{k}))\;(\Varid{t}\mathbin{\oplus}\Varid{r}){}\<[E]%
\\
\>[18]{}\Conid{Var}\;(\Varid{x},\Varid{t}_{1}\mathbin{\ominus}\Varid{r}){}\<[E]%
\\
\>[5]{}\hsindent{2}{}\<[7]%
\>[7]{}){}\<[E]%
\\
\>[3]{}\mathrel{=}\mbox{\commentbegin ~ \Cref{lemma:modify-state-restore}  \commentend}{}\<[E]%
\\
\>[3]{}\hsindent{2}{}\<[5]%
\>[5]{}\Varid{fmap}\;(\Varid{fmap}\;\Varid{fst})\;{}\<[E]%
\\
\>[5]{}\hsindent{2}{}\<[7]%
\>[7]{}(\lambda \Varid{t}\to \mathbf{do}\;{}\<[18]%
\>[18]{}(\Varid{x},\Varid{t}\mathbin{\oplus}\Varid{r})\leftarrow \Varid{h_{Modify1}}\;(\Varid{h_{ND+f}}\;(\Varid{(\Leftrightarrow)}\;\Varid{k}))\;(\Varid{t}\mathbin{\oplus}\Varid{r}){}\<[E]%
\\
\>[18]{}\Conid{Var}\;(\Varid{x},(\Varid{t}\mathbin{\oplus}\Varid{r})\mathbin{\ominus}\Varid{r}){}\<[E]%
\\
\>[5]{}\hsindent{2}{}\<[7]%
\>[7]{}){}\<[E]%
\\
\>[3]{}\mathrel{=}\mbox{\commentbegin ~ \Cref{eq:plus-minus}  \commentend}{}\<[E]%
\\
\>[3]{}\hsindent{2}{}\<[5]%
\>[5]{}\Varid{fmap}\;(\Varid{fmap}\;\Varid{fst})\;{}\<[E]%
\\
\>[5]{}\hsindent{2}{}\<[7]%
\>[7]{}(\lambda \Varid{t}\to \mathbf{do}\;{}\<[18]%
\>[18]{}(\Varid{x},\anonymous )\leftarrow \Varid{h_{Modify1}}\;(\Varid{h_{ND+f}}\;(\Varid{(\Leftrightarrow)}\;\Varid{k}))\;(\Varid{t}\mathbin{\oplus}\Varid{r}){}\<[E]%
\\
\>[18]{}\Conid{Var}\;(\Varid{x},\Varid{t}){}\<[E]%
\\
\>[5]{}\hsindent{2}{}\<[7]%
\>[7]{}){}\<[E]%
\\
\>[3]{}\mathrel{=}\mbox{\commentbegin ~ definition of \ensuremath{\Varid{fmap}\;\Varid{fst}}  \commentend}{}\<[E]%
\\
\>[3]{}\hsindent{2}{}\<[5]%
\>[5]{}\Varid{fmap}\;(\Varid{fmap}\;\Varid{fst})\;{}\<[E]%
\\
\>[5]{}\hsindent{2}{}\<[7]%
\>[7]{}(\lambda \Varid{t}\to \mathbf{do}\;{}\<[18]%
\>[18]{}\Varid{x}\leftarrow \Varid{fmap}\;\Varid{fst}\;(\Varid{h_{Modify1}}\;(\Varid{h_{ND+f}}\;(\Varid{(\Leftrightarrow)}\;\Varid{k}))\;(\Varid{t}\mathbin{\oplus}\Varid{r})){}\<[E]%
\\
\>[18]{}\Conid{Var}\;(\Varid{x},\Varid{t}){}\<[E]%
\\
\>[5]{}\hsindent{2}{}\<[7]%
\>[7]{}){}\<[E]%
\\
\>[3]{}\mathrel{=}\mbox{\commentbegin ~ definition of \ensuremath{\Varid{fmap}}  \commentend}{}\<[E]%
\\
\>[3]{}\hsindent{2}{}\<[5]%
\>[5]{}\Varid{fmap}\;(\Varid{fmap}\;\Varid{fst})\;{}\<[E]%
\\
\>[5]{}\hsindent{2}{}\<[7]%
\>[7]{}(\lambda \Varid{t}\to \mathbf{do}\;{}\<[18]%
\>[18]{}\Varid{x}\leftarrow (\Varid{fmap}\;(\Varid{fmap}\;\Varid{fst})\;(\Varid{h_{Modify1}}\;(\Varid{h_{ND+f}}\;(\Varid{(\Leftrightarrow)}\;\Varid{k}))))\;(\Varid{t}\mathbin{\oplus}\Varid{r}){}\<[E]%
\\
\>[18]{}\Conid{Var}\;(\Varid{x},\Varid{t}){}\<[E]%
\\
\>[5]{}\hsindent{2}{}\<[7]%
\>[7]{}){}\<[E]%
\\
\>[3]{}\mathrel{=}\mbox{\commentbegin ~ definition of \ensuremath{\Varid{fmap}\;(\Varid{fmap}\;\Varid{fst})}  \commentend}{}\<[E]%
\\
\>[3]{}\hsindent{2}{}\<[5]%
\>[5]{}\lambda \Varid{t}\to \mathbf{do}\;{}\<[16]%
\>[16]{}\Varid{x}\leftarrow (\Varid{fmap}\;(\Varid{fmap}\;\Varid{fst})\;(\Varid{h_{Modify1}}\;(\Varid{h_{ND+f}}\;(\Varid{(\Leftrightarrow)}\;\Varid{k}))))\;(\Varid{t}\mathbin{\oplus}\Varid{r}){}\<[E]%
\\
\>[16]{}\Conid{Var}\;\Varid{x}{}\<[E]%
\\
\>[3]{}\mathrel{=}\mbox{\commentbegin ~ monad law  \commentend}{}\<[E]%
\\
\>[3]{}\hsindent{2}{}\<[5]%
\>[5]{}\lambda \Varid{t}\to (\Varid{fmap}\;(\Varid{fmap}\;\Varid{fst})\;(\Varid{h_{Modify1}}\;(\Varid{h_{ND+f}}\;(\Varid{(\Leftrightarrow)}\;\Varid{k}))))\;(\Varid{t}\mathbin{\oplus}\Varid{r}){}\<[E]%
\\
\>[3]{}\mathrel{=}\mbox{\commentbegin ~ definition of \ensuremath{\Varid{h_{GlobalM}}}  \commentend}{}\<[E]%
\\
\>[3]{}\hsindent{2}{}\<[5]%
\>[5]{}\lambda \Varid{t}\to (\Varid{h_{GlobalM}}\;\Varid{k})\;(\Varid{t}\mathbin{\oplus}\Varid{r}){}\<[E]%
\\
\>[3]{}\mathrel{=}\mbox{\commentbegin ~ definition of \ensuremath{\Varid{alg}_{\Varid{LHS}}^{\Varid{S}}}  \commentend}{}\<[E]%
\\
\>[3]{}\hsindent{2}{}\<[5]%
\>[5]{}\Varid{alg}_{\Varid{LHS}}^{\Varid{S}}\;(\Conid{MUpdate}\;\Varid{r}\;(\Varid{h_{GlobalM}}\;\Varid{k})){}\<[E]%
\\
\>[3]{}\mathrel{=}\mbox{\commentbegin ~ definition of \ensuremath{\Varid{fmap}}  \commentend}{}\<[E]%
\\
\>[3]{}\hsindent{2}{}\<[5]%
\>[5]{}\Varid{alg}_{\Varid{LHS}}^{\Varid{S}}\;(\Varid{fmap}\;\Varid{h_{GlobalM}}\;(\Conid{MUpdate}\;\Varid{r}\;\Varid{k})){}\<[E]%
\ColumnHook
\end{hscode}\resethooks
\indentend 
For the second subcondition \refd{}, we can define \ensuremath{\Varid{alg}_{\Varid{LHS}}^{\Varid{ND}}} as
follows.\indentbegin \begin{hscode}\SaveRestoreHook
\column{B}{@{}>{\hspre}l<{\hspost}@{}}%
\column{3}{@{}>{\hspre}l<{\hspost}@{}}%
\column{22}{@{}>{\hspre}l<{\hspost}@{}}%
\column{E}{@{}>{\hspre}l<{\hspost}@{}}%
\>[3]{}\Varid{alg}_{\Varid{LHS}}^{\Varid{ND}}\mathbin{::}\Conid{Functor}\;\Varid{f}\Rightarrow \Varid{Nondet_{F}}\;(\Varid{s}\to \Conid{Free}\;\Varid{f}\;[\mskip1.5mu \Varid{a}\mskip1.5mu])\to (\Varid{s}\to \Conid{Free}\;\Varid{f}\;[\mskip1.5mu \Varid{a}\mskip1.5mu]){}\<[E]%
\\
\>[3]{}\Varid{alg}_{\Varid{LHS}}^{\Varid{ND}}\;\Conid{Fail}{}\<[22]%
\>[22]{}\mathrel{=}\lambda \Varid{s}\to \Conid{Var}\;[\mskip1.5mu \mskip1.5mu]{}\<[E]%
\\
\>[3]{}\Varid{alg}_{\Varid{LHS}}^{\Varid{ND}}\;(\Conid{Or}\;\Varid{p}\;\Varid{q}){}\<[22]%
\>[22]{}\mathrel{=}\lambda \Varid{s}\to \Varid{liftM2}\;(+\!\!+)\;(\Varid{p}\;\Varid{s})\;(\Varid{q}\;\Varid{s}){}\<[E]%
\ColumnHook
\end{hscode}\resethooks
\indentend We prove it by a case analysis on the shape of input \ensuremath{\Varid{op}\mathbin{::}\Varid{Nondet_{F}}\;(\Conid{Free}\;(\Varid{Modify_{F}}\;\Varid{s}\;\Varid{r}\mathrel{{:}{+}{:}}\Varid{Nondet_{F}}\mathrel{{:}{+}{:}}\Varid{f})\;\Varid{a})}.

\noindent \mbox{\underline{case \ensuremath{\Varid{op}\mathrel{=}\Conid{Fail}}}}\indentbegin \begin{hscode}\SaveRestoreHook
\column{B}{@{}>{\hspre}l<{\hspost}@{}}%
\column{3}{@{}>{\hspre}l<{\hspost}@{}}%
\column{4}{@{}>{\hspre}l<{\hspost}@{}}%
\column{5}{@{}>{\hspre}l<{\hspost}@{}}%
\column{E}{@{}>{\hspre}l<{\hspost}@{}}%
\>[5]{}\Varid{h_{GlobalM}}\;(\Varid{alg}\;(\Conid{Inr}\;(\Conid{Inl}\;\Conid{Fail}))){}\<[E]%
\\
\>[3]{}\mathrel{=}\mbox{\commentbegin ~ definition of \ensuremath{\Varid{alg}}  \commentend}{}\<[E]%
\\
\>[3]{}\hsindent{2}{}\<[5]%
\>[5]{}\Varid{h_{GlobalM}}\;(\Conid{Op}\;(\Conid{Inr}\;(\Conid{Inl}\;\Conid{Fail}))){}\<[E]%
\\
\>[3]{}\mathrel{=}\mbox{\commentbegin ~ definition of \ensuremath{\Varid{h_{GlobalM}}}  \commentend}{}\<[E]%
\\
\>[3]{}\hsindent{2}{}\<[5]%
\>[5]{}\Varid{fmap}\;(\Varid{fmap}\;\Varid{fst})\;(\Varid{h_{Modify1}}\;(\Varid{h_{ND+f}}\;(\Varid{(\Leftrightarrow)}\;(\Conid{Op}\;(\Conid{Inr}\;(\Conid{Inl}\;\Conid{Fail})))))){}\<[E]%
\\
\>[3]{}\mathrel{=}\mbox{\commentbegin ~ definition of \ensuremath{\Varid{(\Leftrightarrow)}}  \commentend}{}\<[E]%
\\
\>[3]{}\hsindent{2}{}\<[5]%
\>[5]{}\Varid{fmap}\;(\Varid{fmap}\;\Varid{fst})\;(\Varid{h_{Modify1}}\;(\Varid{h_{ND+f}}\;(\Conid{Op}\;(\Conid{Inl}\;(\Varid{fmap}\;\Varid{(\Leftrightarrow)}\;\Conid{Fail}))))){}\<[E]%
\\
\>[3]{}\mathrel{=}\mbox{\commentbegin ~ definition of \ensuremath{\Varid{fmap}}  \commentend}{}\<[E]%
\\
\>[3]{}\hsindent{2}{}\<[5]%
\>[5]{}\Varid{fmap}\;(\Varid{fmap}\;\Varid{fst})\;(\Varid{h_{Modify1}}\;(\Varid{h_{ND+f}}\;(\Conid{Op}\;(\Conid{Inl}\;\Conid{Fail})))){}\<[E]%
\\
\>[3]{}\mathrel{=}\mbox{\commentbegin ~ definition of \ensuremath{\Varid{h_{ND+f}}}  \commentend}{}\<[E]%
\\
\>[3]{}\hsindent{2}{}\<[5]%
\>[5]{}\Varid{fmap}\;(\Varid{fmap}\;\Varid{fst})\;(\Varid{h_{Modify1}}\;(\Conid{Var}\;[\mskip1.5mu \mskip1.5mu])){}\<[E]%
\\
\>[3]{}\mathrel{=}\mbox{\commentbegin ~ definition of \ensuremath{\Varid{h_{Modify1}}}  \commentend}{}\<[E]%
\\
\>[3]{}\hsindent{2}{}\<[5]%
\>[5]{}\Varid{fmap}\;(\Varid{fmap}\;\Varid{fst})\;(\lambda \Varid{s}\to \Conid{Var}\;([\mskip1.5mu \mskip1.5mu],\Varid{s})){}\<[E]%
\\
\>[3]{}\mathrel{=}\mbox{\commentbegin ~ definition of \ensuremath{\Varid{fmap}} twice and \ensuremath{\Varid{fst}}  \commentend}{}\<[E]%
\\
\>[3]{}\hsindent{2}{}\<[5]%
\>[5]{}\lambda \Varid{s}\to \Conid{Var}\;[\mskip1.5mu \mskip1.5mu]{}\<[E]%
\\
\>[3]{}\mathrel{=}\mbox{\commentbegin ~ definition of \ensuremath{\Varid{alg}_{\Varid{RHS}}^{\Varid{ND}}}   \commentend}{}\<[E]%
\\
\>[3]{}\hsindent{1}{}\<[4]%
\>[4]{}\Varid{alg}_{\Varid{RHS}}^{\Varid{ND}}\;\Conid{Fail}{}\<[E]%
\\
\>[3]{}\mathrel{=}\mbox{\commentbegin ~ definition of \ensuremath{\Varid{fmap}}  \commentend}{}\<[E]%
\\
\>[3]{}\hsindent{1}{}\<[4]%
\>[4]{}\Varid{alg}_{\Varid{RHS}}^{\Varid{ND}}\;(\Varid{fmap}\;\Varid{h_{GlobalM}}\;\Conid{Fail}){}\<[E]%
\ColumnHook
\end{hscode}\resethooks
\indentend \noindent \mbox{\underline{case \ensuremath{\Varid{op}\mathrel{=}\Conid{Or}\;\Varid{p}\;\Varid{q}}}}
From \ensuremath{\Varid{op}} is in the codomain of \ensuremath{\Varid{fmap}\;\Varid{local2global_M}} we obtain \ensuremath{\Varid{p}} and
\ensuremath{\Varid{q}} are in the codomain of \ensuremath{\Varid{local2global_M}}.
\indentbegin \begin{hscode}\SaveRestoreHook
\column{B}{@{}>{\hspre}l<{\hspost}@{}}%
\column{3}{@{}>{\hspre}l<{\hspost}@{}}%
\column{5}{@{}>{\hspre}l<{\hspost}@{}}%
\column{16}{@{}>{\hspre}l<{\hspost}@{}}%
\column{33}{@{}>{\hspre}l<{\hspost}@{}}%
\column{36}{@{}>{\hspre}l<{\hspost}@{}}%
\column{E}{@{}>{\hspre}l<{\hspost}@{}}%
\>[5]{}\Varid{h_{GlobalM}}\;(\Varid{alg}\;(\Conid{Inr}\;(\Conid{Inl}\;(\Conid{Or}\;\Varid{p}\;\Varid{q})))){}\<[E]%
\\
\>[3]{}\mathrel{=}\mbox{\commentbegin ~ definition of \ensuremath{\Varid{alg}}  \commentend}{}\<[E]%
\\
\>[3]{}\hsindent{2}{}\<[5]%
\>[5]{}\Varid{h_{GlobalM}}\;(\Conid{Op}\;(\Conid{Inr}\;(\Conid{Inl}\;(\Conid{Or}\;\Varid{p}\;\Varid{q})))){}\<[E]%
\\
\>[3]{}\mathrel{=}\mbox{\commentbegin ~ definition of \ensuremath{\Varid{h_{GlobalM}}}  \commentend}{}\<[E]%
\\
\>[3]{}\hsindent{2}{}\<[5]%
\>[5]{}\Varid{fmap}\;(\Varid{fmap}\;\Varid{fst})\;(\Varid{h_{Modify1}}\;(\Varid{h_{ND+f}}\;(\Varid{(\Leftrightarrow)}\;(\Conid{Op}\;(\Conid{Inr}\;(\Conid{Inl}\;(\Conid{Or}\;\Varid{p}\;\Varid{q}))))))){}\<[E]%
\\
\>[3]{}\mathrel{=}\mbox{\commentbegin ~ definition of \ensuremath{\Varid{(\Leftrightarrow)}}  \commentend}{}\<[E]%
\\
\>[3]{}\hsindent{2}{}\<[5]%
\>[5]{}\Varid{fmap}\;(\Varid{fmap}\;\Varid{fst})\;(\Varid{h_{Modify1}}\;(\Varid{h_{ND+f}}\;(\Conid{Op}\;(\Conid{Inl}\;(\Varid{fmap}\;\Varid{(\Leftrightarrow)}\;(\Conid{Or}\;\Varid{p}\;\Varid{q})))))){}\<[E]%
\\
\>[3]{}\mathrel{=}\mbox{\commentbegin ~ definition of \ensuremath{\Varid{fmap}}  \commentend}{}\<[E]%
\\
\>[3]{}\hsindent{2}{}\<[5]%
\>[5]{}\Varid{fmap}\;(\Varid{fmap}\;\Varid{fst})\;(\Varid{h_{Modify1}}\;(\Varid{h_{ND+f}}\;(\Conid{Op}\;(\Conid{Inl}\;(\Conid{Or}\;(\Varid{(\Leftrightarrow)}\;\Varid{p})\;(\Varid{(\Leftrightarrow)}\;\Varid{q})))))){}\<[E]%
\\
\>[3]{}\mathrel{=}\mbox{\commentbegin ~ definition of \ensuremath{\Varid{h_{ND+f}}}  \commentend}{}\<[E]%
\\
\>[3]{}\hsindent{2}{}\<[5]%
\>[5]{}\Varid{fmap}\;(\Varid{fmap}\;\Varid{fst})\;(\Varid{h_{Modify1}}\;(\Varid{liftM2}\;(+\!\!+)\;(\Varid{h_{ND+f}}\;(\Varid{(\Leftrightarrow)}\;\Varid{p}))\;(\Varid{h_{ND+f}}\;(\Varid{(\Leftrightarrow)}\;\Varid{q})))){}\<[E]%
\\
\>[3]{}\mathrel{=}\mbox{\commentbegin ~ definition of \ensuremath{\Varid{liftM2}}  \commentend}{}\<[E]%
\\
\>[3]{}\hsindent{2}{}\<[5]%
\>[5]{}\Varid{fmap}\;(\Varid{fmap}\;\Varid{fst})\;(\Varid{h_{Modify1}}\;(\mathbf{do}\;{}\<[36]%
\>[36]{}\Varid{x}\leftarrow \Varid{h_{ND+f}}\;(\Varid{(\Leftrightarrow)}\;\Varid{p}){}\<[E]%
\\
\>[36]{}\Varid{y}\leftarrow \Varid{h_{ND+f}}\;(\Varid{(\Leftrightarrow)}\;\Varid{q}){}\<[E]%
\\
\>[36]{}\Varid{\eta}\;(\Varid{x}+\!\!+\Varid{y}))){}\<[E]%
\\
\>[3]{}\mathrel{=}\mbox{\commentbegin ~ \Cref{lemma:dist-hModify1}  \commentend}{}\<[E]%
\\
\>[3]{}\hsindent{2}{}\<[5]%
\>[5]{}\Varid{fmap}\;(\Varid{fmap}\;\Varid{fst})\;(\lambda \Varid{s}_{0}\to (\mathbf{do}\;{}\<[33]%
\>[33]{}(\Varid{x},\Varid{s}_{1})\leftarrow \Varid{h_{Modify1}}\;(\Varid{h_{ND+f}}\;(\Varid{(\Leftrightarrow)}\;\Varid{p}))\;\Varid{s}_{0}{}\<[E]%
\\
\>[33]{}(\Varid{y},\Varid{s}_{2})\leftarrow \Varid{h_{Modify1}}\;(\Varid{h_{ND+f}}\;(\Varid{(\Leftrightarrow)}\;\Varid{q}))\;\Varid{s}_{1}{}\<[E]%
\\
\>[33]{}\Varid{h_{Modify1}}\;(\Varid{\eta}\;(\Varid{x}+\!\!+\Varid{y}))\;\Varid{s}_{2})){}\<[E]%
\\
\>[3]{}\mathrel{=}\mbox{\commentbegin ~ definition of \ensuremath{\Varid{h_{Modify1}}}  \commentend}{}\<[E]%
\\
\>[3]{}\hsindent{2}{}\<[5]%
\>[5]{}\Varid{fmap}\;(\Varid{fmap}\;\Varid{fst})\;(\lambda \Varid{s}_{0}\to (\mathbf{do}\;{}\<[33]%
\>[33]{}(\Varid{x},\Varid{s}_{1})\leftarrow \Varid{h_{Modify1}}\;(\Varid{h_{ND+f}}\;(\Varid{(\Leftrightarrow)}\;\Varid{p}))\;\Varid{s}_{0}{}\<[E]%
\\
\>[33]{}(\Varid{y},\Varid{s}_{2})\leftarrow \Varid{h_{Modify1}}\;(\Varid{h_{ND+f}}\;(\Varid{(\Leftrightarrow)}\;\Varid{q}))\;\Varid{s}_{1}{}\<[E]%
\\
\>[33]{}\Conid{Var}\;(\Varid{x}+\!\!+\Varid{y},\Varid{s}_{2}))){}\<[E]%
\\
\>[3]{}\mathrel{=}\mbox{\commentbegin ~ \Cref{lemma:modify-state-restore}   \commentend}{}\<[E]%
\\
\>[3]{}\hsindent{2}{}\<[5]%
\>[5]{}\Varid{fmap}\;(\Varid{fmap}\;\Varid{fst})\;(\lambda \Varid{s}_{0}\to (\mathbf{do}\;{}\<[33]%
\>[33]{}(\Varid{x},\Varid{s}_{1})\leftarrow \mathbf{do}\;\{\mskip1.5mu (\Varid{x},\anonymous )\leftarrow \Varid{h_{Modify1}}\;(\Varid{h_{ND+f}}\;(\Varid{(\Leftrightarrow)}\;\Varid{p}))\;\Varid{s}_{0};\Varid{\eta}\;(\Varid{x},\Varid{s}_{0})\mskip1.5mu\}{}\<[E]%
\\
\>[33]{}(\Varid{y},\Varid{s}_{2})\leftarrow \mathbf{do}\;\{\mskip1.5mu (\Varid{y},\anonymous )\leftarrow \Varid{h_{Modify1}}\;(\Varid{h_{ND+f}}\;(\Varid{(\Leftrightarrow)}\;\Varid{q}))\;\Varid{s}_{1};\Varid{\eta}\;(\Varid{x},\Varid{s}_{1})\mskip1.5mu\}{}\<[E]%
\\
\>[33]{}\Conid{Var}\;(\Varid{x}+\!\!+\Varid{y},\Varid{s}_{2}))){}\<[E]%
\\
\>[3]{}\mathrel{=}\mbox{\commentbegin ~ monad laws  \commentend}{}\<[E]%
\\
\>[3]{}\hsindent{2}{}\<[5]%
\>[5]{}\Varid{fmap}\;(\Varid{fmap}\;\Varid{fst})\;(\lambda \Varid{s}_{0}\to (\mathbf{do}\;{}\<[33]%
\>[33]{}(\Varid{x},\anonymous )\leftarrow \Varid{h_{Modify1}}\;(\Varid{h_{ND+f}}\;(\Varid{(\Leftrightarrow)}\;\Varid{p}))\;\Varid{s}_{0}{}\<[E]%
\\
\>[33]{}(\Varid{y},\anonymous )\leftarrow \Varid{h_{Modify1}}\;(\Varid{h_{ND+f}}\;(\Varid{(\Leftrightarrow)}\;\Varid{q}))\;\Varid{s}_{0}{}\<[E]%
\\
\>[33]{}\Conid{Var}\;(\Varid{x}+\!\!+\Varid{y},\Varid{s}_{0}))){}\<[E]%
\\
\>[3]{}\mathrel{=}\mbox{\commentbegin ~ definition of \ensuremath{\Varid{fmap}} (twice) and \ensuremath{\Varid{fst}}  \commentend}{}\<[E]%
\\
\>[3]{}\hsindent{2}{}\<[5]%
\>[5]{}\lambda \Varid{s}_{0}\to (\mathbf{do}\;{}\<[16]%
\>[16]{}(\Varid{x},\anonymous )\leftarrow \Varid{h_{Modify1}}\;(\Varid{h_{ND+f}}\;(\Varid{(\Leftrightarrow)}\;\Varid{p}))\;\Varid{s}_{0}{}\<[E]%
\\
\>[16]{}(\Varid{y},\anonymous )\leftarrow \Varid{h_{Modify1}}\;(\Varid{h_{ND+f}}\;(\Varid{(\Leftrightarrow)}\;\Varid{q}))\;\Varid{s}_{0}{}\<[E]%
\\
\>[16]{}\Conid{Var}\;(\Varid{x}+\!\!+\Varid{y})){}\<[E]%
\\
\>[3]{}\mathrel{=}\mbox{\commentbegin ~ definition of \ensuremath{\Varid{fmap}}, \ensuremath{\Varid{fst}} and monad laws  \commentend}{}\<[E]%
\\
\>[3]{}\hsindent{2}{}\<[5]%
\>[5]{}\lambda \Varid{s}_{0}\to (\mathbf{do}\;{}\<[16]%
\>[16]{}\Varid{x}\leftarrow \Varid{fmap}\;\Varid{fst}\;(\Varid{h_{Modify1}}\;(\Varid{h_{ND+f}}\;(\Varid{(\Leftrightarrow)}\;\Varid{p}))\;\Varid{s}_{0}){}\<[E]%
\\
\>[16]{}\Varid{y}\leftarrow \Varid{fmap}\;\Varid{fst}\;(\Varid{h_{Modify1}}\;(\Varid{h_{ND+f}}\;(\Varid{(\Leftrightarrow)}\;\Varid{q}))\;\Varid{s}_{0}){}\<[E]%
\\
\>[16]{}\Conid{Var}\;(\Varid{x}+\!\!+\Varid{y})){}\<[E]%
\\
\>[3]{}\mathrel{=}\mbox{\commentbegin ~ definition of \ensuremath{\Varid{fmap}}  \commentend}{}\<[E]%
\\
\>[3]{}\hsindent{2}{}\<[5]%
\>[5]{}\lambda \Varid{s}_{0}\to (\mathbf{do}\;{}\<[16]%
\>[16]{}\Varid{x}\leftarrow \Varid{fmap}\;(\Varid{fmap}\;\Varid{fst})\;(\Varid{h_{Modify1}}\;(\Varid{h_{ND+f}}\;(\Varid{(\Leftrightarrow)}\;\Varid{p})))\;\Varid{s}_{0}{}\<[E]%
\\
\>[16]{}\Varid{y}\leftarrow \Varid{fmap}\;(\Varid{fmap}\;\Varid{fst})\;(\Varid{h_{Modify1}}\;(\Varid{h_{ND+f}}\;(\Varid{(\Leftrightarrow)}\;\Varid{q})))\;\Varid{s}_{0}{}\<[E]%
\\
\>[16]{}\Conid{Var}\;(\Varid{x}+\!\!+\Varid{y})){}\<[E]%
\\
\>[3]{}\mathrel{=}\mbox{\commentbegin ~ definition of \ensuremath{\Varid{h_{GlobalM}}}  \commentend}{}\<[E]%
\\
\>[3]{}\hsindent{2}{}\<[5]%
\>[5]{}\lambda \Varid{s}_{0}\to (\mathbf{do}\;{}\<[16]%
\>[16]{}\Varid{x}\leftarrow \Varid{h_{GlobalM}}\;\Varid{p}\;\Varid{s}_{0}{}\<[E]%
\\
\>[16]{}\Varid{y}\leftarrow \Varid{h_{GlobalM}}\;\Varid{q}\;\Varid{s}_{0}{}\<[E]%
\\
\>[16]{}\Conid{Var}\;(\Varid{x}+\!\!+\Varid{y})){}\<[E]%
\\
\>[3]{}\mathrel{=}\mbox{\commentbegin ~ definition of \ensuremath{\Varid{liftM2}}  \commentend}{}\<[E]%
\\
\>[3]{}\hsindent{2}{}\<[5]%
\>[5]{}\lambda \Varid{s}_{0}\to \Varid{liftM2}\;(+\!\!+)\;(\Varid{h_{GlobalM}}\;\Varid{p}\;\Varid{s}_{0})\;(\Varid{h_{GlobalM}}\;\Varid{q}\;\Varid{s}_{0}){}\<[E]%
\\
\>[3]{}\mathrel{=}\mbox{\commentbegin ~ definition of \ensuremath{\Varid{alg}_{\Varid{LHS}}^{\Varid{ND}}}   \commentend}{}\<[E]%
\\
\>[3]{}\hsindent{2}{}\<[5]%
\>[5]{}\Varid{alg}_{\Varid{LHS}}^{\Varid{ND}}\;(\Conid{Or}\;(\Varid{h_{GlobalM}}\;\Varid{p})\;(\Varid{h_{GlobalM}}\;\Varid{q})){}\<[E]%
\\
\>[3]{}\mathrel{=}\mbox{\commentbegin ~ definition of \ensuremath{\Varid{fmap}}  \commentend}{}\<[E]%
\\
\>[3]{}\hsindent{2}{}\<[5]%
\>[5]{}\Varid{alg}_{\Varid{LHS}}^{\Varid{ND}}\;(\Varid{fmap}\;\Varid{hGobal}\;(\Conid{Or}\;\Varid{p}\;\Varid{q})){}\<[E]%
\ColumnHook
\end{hscode}\resethooks
\indentend For the last subcondition \refe{}, we can define \ensuremath{\Varid{fwd}_{\Varid{LHS}}} as follows.\indentbegin \begin{hscode}\SaveRestoreHook
\column{B}{@{}>{\hspre}l<{\hspost}@{}}%
\column{3}{@{}>{\hspre}l<{\hspost}@{}}%
\column{E}{@{}>{\hspre}l<{\hspost}@{}}%
\>[3]{}\Varid{fwd}_{\Varid{LHS}}\mathbin{::}\Conid{Functor}\;\Varid{f}\Rightarrow \Varid{f}\;(\Varid{s}\to \Conid{Free}\;\Varid{f}\;[\mskip1.5mu \Varid{a}\mskip1.5mu])\to (\Varid{s}\to \Conid{Free}\;\Varid{f}\;[\mskip1.5mu \Varid{a}\mskip1.5mu]){}\<[E]%
\\
\>[3]{}\Varid{fwd}_{\Varid{LHS}}\;\Varid{op}\mathrel{=}\lambda \Varid{s}\to \Conid{Op}\;(\Varid{fmap}\;(\mathbin{\$}\Varid{s})\;\Varid{op}){}\<[E]%
\ColumnHook
\end{hscode}\resethooks
\indentend We prove it by the following calculation for input \ensuremath{\Varid{op}\mathbin{::}\Varid{f}\;(\Conid{Free}\;(\Varid{Modify_{F}}\;\Varid{s}\;\Varid{r}\mathrel{{:}{+}{:}}\Varid{Nondet_{F}}\mathrel{{:}{+}{:}}\Varid{f})\;\Varid{a})}.

\indentbegin \begin{hscode}\SaveRestoreHook
\column{B}{@{}>{\hspre}l<{\hspost}@{}}%
\column{3}{@{}>{\hspre}l<{\hspost}@{}}%
\column{5}{@{}>{\hspre}l<{\hspost}@{}}%
\column{E}{@{}>{\hspre}l<{\hspost}@{}}%
\>[5]{}\Varid{h_{GlobalM}}\;(\Varid{alg}\;(\Conid{Inr}\;(\Conid{Inr}\;\Varid{op}))){}\<[E]%
\\
\>[3]{}\mathrel{=}\mbox{\commentbegin ~ definition of \ensuremath{\Varid{alg}}  \commentend}{}\<[E]%
\\
\>[3]{}\hsindent{2}{}\<[5]%
\>[5]{}\Varid{h_{GlobalM}}\;(\Conid{Op}\;(\Conid{Inr}\;(\Conid{Inr}\;\Varid{op}))){}\<[E]%
\\
\>[3]{}\mathrel{=}\mbox{\commentbegin ~ definition of \ensuremath{\Varid{h_{GlobalM}}}  \commentend}{}\<[E]%
\\
\>[3]{}\hsindent{2}{}\<[5]%
\>[5]{}\Varid{fmap}\;(\Varid{fmap}\;\Varid{fst})\;(\Varid{h_{Modify1}}\;(\Varid{h_{ND+f}}\;(\Varid{(\Leftrightarrow)}\;(\Conid{Op}\;(\Conid{Inr}\;(\Conid{Inr}\;\Varid{op})))))){}\<[E]%
\\
\>[3]{}\mathrel{=}\mbox{\commentbegin ~ definition of \ensuremath{\Varid{(\Leftrightarrow)}}  \commentend}{}\<[E]%
\\
\>[3]{}\hsindent{2}{}\<[5]%
\>[5]{}\Varid{fmap}\;(\Varid{fmap}\;\Varid{fst})\;(\Varid{h_{Modify1}}\;(\Varid{h_{ND+f}}\;(\Conid{Op}\;(\Conid{Inr}\;(\Conid{Inr}\;(\Varid{fmap}\;\Varid{(\Leftrightarrow)}\;\Varid{op})))))){}\<[E]%
\\
\>[3]{}\mathrel{=}\mbox{\commentbegin ~ definition of \ensuremath{\Varid{h_{ND+f}}}  \commentend}{}\<[E]%
\\
\>[3]{}\hsindent{2}{}\<[5]%
\>[5]{}\Varid{fmap}\;(\Varid{fmap}\;\Varid{fst})\;(\Varid{h_{Modify1}}\;(\Conid{Op}\;(\Varid{fmap}\;\Varid{h_{ND+f}}\;(\Conid{Inr}\;(\Varid{fmap}\;\Varid{(\Leftrightarrow)}\;\Varid{op}))))){}\<[E]%
\\
\>[3]{}\mathrel{=}\mbox{\commentbegin ~ definition of \ensuremath{\Varid{fmap}}  \commentend}{}\<[E]%
\\
\>[3]{}\hsindent{2}{}\<[5]%
\>[5]{}\Varid{fmap}\;(\Varid{fmap}\;\Varid{fst})\;(\Varid{h_{Modify1}}\;(\Conid{Op}\;(\Conid{Inr}\;(\Varid{fmap}\;\Varid{h_{ND+f}}\;(\Varid{fmap}\;\Varid{(\Leftrightarrow)}\;\Varid{op}))))){}\<[E]%
\\
\>[3]{}\mathrel{=}\mbox{\commentbegin ~ \ensuremath{\Varid{fmap}} fusion  \commentend}{}\<[E]%
\\
\>[3]{}\hsindent{2}{}\<[5]%
\>[5]{}\Varid{fmap}\;(\Varid{fmap}\;\Varid{fst})\;(\Varid{h_{Modify1}}\;(\Conid{Op}\;(\Conid{Inr}\;(\Varid{fmap}\;(\Varid{h_{ND+f}}\hsdot{\circ }{.}\Varid{(\Leftrightarrow)})\;\Varid{op})))){}\<[E]%
\\
\>[3]{}\mathrel{=}\mbox{\commentbegin ~ definition of \ensuremath{\Varid{h_{Modify1}}}  \commentend}{}\<[E]%
\\
\>[3]{}\hsindent{2}{}\<[5]%
\>[5]{}\Varid{fmap}\;(\Varid{fmap}\;\Varid{fst})\;(\lambda \Varid{s}\to \Conid{Op}\;(\Varid{fmap}\;(\mathbin{\$}\Varid{s})\;(\Varid{fmap}\;\Varid{h_{Modify1}}\;(\Varid{fmap}\;(\Varid{h_{ND+f}}\hsdot{\circ }{.}\Varid{(\Leftrightarrow)})\;\Varid{op})))){}\<[E]%
\\
\>[3]{}\mathrel{=}\mbox{\commentbegin ~ \ensuremath{\Varid{fmap}} fusion  \commentend}{}\<[E]%
\\
\>[3]{}\hsindent{2}{}\<[5]%
\>[5]{}\Varid{fmap}\;(\Varid{fmap}\;\Varid{fst})\;(\lambda \Varid{s}\to \Conid{Op}\;(\Varid{fmap}\;(\mathbin{\$}\Varid{s})\;(\Varid{fmap}\;(\Varid{h_{Modify1}}\hsdot{\circ }{.}\Varid{h_{ND+f}}\hsdot{\circ }{.}\Varid{(\Leftrightarrow)})\;\Varid{op}))){}\<[E]%
\\
\>[3]{}\mathrel{=}\mbox{\commentbegin ~ definition of \ensuremath{\Varid{fmap}}  \commentend}{}\<[E]%
\\
\>[3]{}\hsindent{2}{}\<[5]%
\>[5]{}\lambda \Varid{s}\to \Varid{fmap}\;\Varid{fst}\;(\Conid{Op}\;(\Varid{fmap}\;(\mathbin{\$}\Varid{s})\;(\Varid{fmap}\;(\Varid{h_{Modify1}}\hsdot{\circ }{.}\Varid{h_{ND+f}}\hsdot{\circ }{.}\Varid{(\Leftrightarrow)})\;\Varid{op}))){}\<[E]%
\\
\>[3]{}\mathrel{=}\mbox{\commentbegin ~ definition of \ensuremath{\Varid{fmap}}  \commentend}{}\<[E]%
\\
\>[3]{}\hsindent{2}{}\<[5]%
\>[5]{}\lambda \Varid{s}\to \Conid{Op}\;(\Varid{fmap}\;(\Varid{fmap}\;\Varid{fst})\;(\Varid{fmap}\;(\mathbin{\$}\Varid{s})\;(\Varid{fmap}\;(\Varid{h_{Modify1}}\hsdot{\circ }{.}\Varid{h_{ND+f}}\hsdot{\circ }{.}\Varid{(\Leftrightarrow)})\;\Varid{op}))){}\<[E]%
\\
\>[3]{}\mathrel{=}\mbox{\commentbegin ~ \ensuremath{\Varid{fmap}} fusion  \commentend}{}\<[E]%
\\
\>[3]{}\hsindent{2}{}\<[5]%
\>[5]{}\lambda \Varid{s}\to \Conid{Op}\;(\Varid{fmap}\;(\Varid{fmap}\;\Varid{fst}\hsdot{\circ }{.}(\mathbin{\$}\Varid{s}))\;(\Varid{fmap}\;(\Varid{h_{Modify1}}\hsdot{\circ }{.}\Varid{h_{ND+f}}\hsdot{\circ }{.}\Varid{(\Leftrightarrow)})\;\Varid{op}))){}\<[E]%
\\
\>[3]{}\mathrel{=}\mbox{\commentbegin ~ \Cref{eq:comm-app-fmap}  \commentend}{}\<[E]%
\\
\>[3]{}\hsindent{2}{}\<[5]%
\>[5]{}\lambda \Varid{s}\to \Conid{Op}\;(\Varid{fmap}\;((\mathbin{\$}\Varid{s})\hsdot{\circ }{.}\Varid{fmap}\;(\Varid{fmap}\;\Varid{fst}))\;(\Varid{fmap}\;(\Varid{h_{Modify1}}\hsdot{\circ }{.}\Varid{h_{ND+f}}\hsdot{\circ }{.}\Varid{(\Leftrightarrow)})\;\Varid{op}))){}\<[E]%
\\
\>[3]{}\mathrel{=}\mbox{\commentbegin ~ \ensuremath{\Varid{fmap}} fission  \commentend}{}\<[E]%
\\
\>[3]{}\hsindent{2}{}\<[5]%
\>[5]{}\lambda \Varid{s}\to \Conid{Op}\;((\Varid{fmap}\;(\mathbin{\$}\Varid{s})\hsdot{\circ }{.}\Varid{fmap}\;(\Varid{fmap}\;(\Varid{fmap}\;\Varid{fst})))\;(\Varid{fmap}\;(\Varid{h_{Modify1}}\hsdot{\circ }{.}\Varid{h_{ND+f}}\hsdot{\circ }{.}\Varid{(\Leftrightarrow)})\;\Varid{op})){}\<[E]%
\\
\>[3]{}\mathrel{=}\mbox{\commentbegin ~ \ensuremath{\Varid{fmap}} fusion  \commentend}{}\<[E]%
\\
\>[3]{}\hsindent{2}{}\<[5]%
\>[5]{}\lambda \Varid{s}\to \Conid{Op}\;(\Varid{fmap}\;(\mathbin{\$}\Varid{s})\;(\Varid{fmap}\;(\Varid{fmap}\;(\Varid{fmap}\;\Varid{fst})\hsdot{\circ }{.}\Varid{h_{Modify1}}\hsdot{\circ }{.}\Varid{h_{ND+f}}\hsdot{\circ }{.}\Varid{(\Leftrightarrow)})\;\Varid{op})){}\<[E]%
\\
\>[3]{}\mathrel{=}\mbox{\commentbegin ~ definition of \ensuremath{\Varid{h_{GlobalM}}}  \commentend}{}\<[E]%
\\
\>[3]{}\hsindent{2}{}\<[5]%
\>[5]{}\lambda \Varid{s}\to \Conid{Op}\;(\Varid{fmap}\;(\mathbin{\$}\Varid{s})\;(\Varid{fmap}\;\Varid{h_{GlobalM}}\;\Varid{op})){}\<[E]%
\\
\>[3]{}\mathrel{=}\mbox{\commentbegin ~ definition of \ensuremath{\Varid{fwd}_{\Varid{LHS}}}   \commentend}{}\<[E]%
\\
\>[3]{}\hsindent{2}{}\<[5]%
\>[5]{}\Varid{fwd}_{\Varid{LHS}}\;(\Varid{fmap}\;\Varid{h_{GlobalM}}\;\Varid{op}){}\<[E]%
\ColumnHook
\end{hscode}\resethooks
\indentend 
\subsection{Equating the Fused Sides}

We observe that the following equations hold trivially.
\begin{eqnarray*}
\ensuremath{\Varid{gen}_{\Varid{LHS}}} & = & \ensuremath{\Varid{gen}_{\Varid{RHS}}} \\
\ensuremath{\Varid{alg}_{\Varid{LHS}}^{\Varid{S}}} & = & \ensuremath{\Varid{alg}_{\Varid{RHS}}^{\Varid{S}}} \\
\ensuremath{\Varid{alg}_{\Varid{LHS}}^{\Varid{ND}}} & = & \ensuremath{\Varid{alg}_{\Varid{RHS}}^{\Varid{ND}}} \\
\ensuremath{\Varid{fwd}_{\Varid{LHS}}} & = & \ensuremath{\Varid{fwd}_{\Varid{RHS}}}
\end{eqnarray*}

Therefore, the main theorem (\Cref{thm:modify-local-global}) holds.

\subsection{Key Lemma: State Restoration}

Similar to \Cref{app:local-global}, we have a key lemma saying that
\ensuremath{\Varid{local2global_M}} restores the initial state after a computation.

\begin{lemma}[State is Restored] \label{lemma:modify-state-restore} \ \\
For any program \ensuremath{\Varid{p}\mathbin{::}\Conid{Free}\;(\Varid{Modify_{F}}\;\Varid{s}\;\Varid{r}\mathrel{{:}{+}{:}}\Varid{Nondet_{F}}\mathrel{{:}{+}{:}}\Varid{f})\;\Varid{a}} that do
not use the operation \ensuremath{\Conid{OP}\;(\Conid{Inl}\;\Conid{MRestore}\;\anonymous \;\anonymous )}, we have
\[\ba{ll}
  &\ensuremath{\Varid{h_{Modify1}}\;(\Varid{h_{ND+f}}\;(\Varid{(\Leftrightarrow)}\;(\Varid{local2global_M}\;\Varid{p})))\;\Varid{s}} \\
= &\ensuremath{\mathbf{do}\;(\Varid{x},\anonymous )\leftarrow \Varid{h_{Modify1}}\;(\Varid{h_{ND+f}}\;(\Varid{(\Leftrightarrow)}\;(\Varid{local2global_M}\;\Varid{p})))\;\Varid{s};\Varid{\eta}\;(\Varid{x},\Varid{s})}
\ea\]
\end{lemma}

\begin{proof}
The proof follows the same structure of \Cref{lemma:state-restore}.
We proceed by induction on \ensuremath{\Varid{t}}.

\noindent \mbox{\underline{case \ensuremath{\Varid{t}\mathrel{=}\Conid{Var}\;\Varid{y}}}}\indentbegin \begin{hscode}\SaveRestoreHook
\column{B}{@{}>{\hspre}l<{\hspost}@{}}%
\column{3}{@{}>{\hspre}l<{\hspost}@{}}%
\column{6}{@{}>{\hspre}l<{\hspost}@{}}%
\column{E}{@{}>{\hspre}l<{\hspost}@{}}%
\>[6]{}\Varid{h_{Modify1}}\;(\Varid{h_{ND+f}}\;(\Varid{(\Leftrightarrow)}\;(\Varid{local2global_M}\;(\Conid{Var}\;\Varid{y}))))\;\Varid{s}{}\<[E]%
\\
\>[3]{}\mathrel{=}\mbox{\commentbegin ~  definition of \ensuremath{\Varid{local2global_M}}   \commentend}{}\<[E]%
\\
\>[3]{}\hsindent{3}{}\<[6]%
\>[6]{}\Varid{h_{Modify1}}\;(\Varid{h_{ND+f}}\;(\Varid{(\Leftrightarrow)}\;(\Conid{Var}\;\Varid{y})))\;\Varid{s}{}\<[E]%
\\
\>[3]{}\mathrel{=}\mbox{\commentbegin ~  definition of \ensuremath{\Varid{(\Leftrightarrow)}}   \commentend}{}\<[E]%
\\
\>[3]{}\hsindent{3}{}\<[6]%
\>[6]{}\Varid{h_{Modify1}}\;(\Varid{h_{ND+f}}\;(\Conid{Var}\;\Varid{y}))\;\Varid{s}{}\<[E]%
\\
\>[3]{}\mathrel{=}\mbox{\commentbegin ~  definition of \ensuremath{\Varid{h_{ND+f}}}   \commentend}{}\<[E]%
\\
\>[3]{}\hsindent{3}{}\<[6]%
\>[6]{}\Varid{h_{Modify1}}\;(\Conid{Var}\;[\mskip1.5mu \Varid{y}\mskip1.5mu])\;\Varid{s}{}\<[E]%
\\
\>[3]{}\mathrel{=}\mbox{\commentbegin ~  definition of \ensuremath{\Varid{h_{Modify1}}}   \commentend}{}\<[E]%
\\
\>[3]{}\hsindent{3}{}\<[6]%
\>[6]{}\Conid{Var}\;([\mskip1.5mu \Varid{y}\mskip1.5mu],\Varid{s}){}\<[E]%
\\
\>[3]{}\mathrel{=}\mbox{\commentbegin ~  monad law  \commentend}{}\<[E]%
\\
\>[3]{}\hsindent{3}{}\<[6]%
\>[6]{}\mathbf{do}\;(\Varid{x},\anonymous )\leftarrow \Conid{Var}\;([\mskip1.5mu \Varid{y}\mskip1.5mu],\Varid{s});\Conid{Var}\;(\Varid{x},\Varid{s}){}\<[E]%
\\
\>[3]{}\mathrel{=}\mbox{\commentbegin ~  definition of \ensuremath{\Varid{local2global_M},\Varid{h_{ND+f}},\Varid{(\Leftrightarrow)},\Varid{h_{Modify1}}} and \ensuremath{\Varid{\eta}}   \commentend}{}\<[E]%
\\
\>[3]{}\hsindent{3}{}\<[6]%
\>[6]{}\mathbf{do}\;(\Varid{x},\anonymous )\leftarrow \Varid{h_{Modify1}}\;(\Varid{h_{ND+f}}\;(\Varid{(\Leftrightarrow)}\;(\Varid{local2global_M}\;(\Conid{Var}\;\Varid{y}))))\;\Varid{s};\Varid{\eta}\;(\Varid{x},\Varid{s}){}\<[E]%
\ColumnHook
\end{hscode}\resethooks
\indentend \noindent \mbox{\underline{case \ensuremath{\Varid{t}\mathrel{=}\Conid{Op}\;(\Conid{Inl}\;(\Conid{MGet}\;\Varid{k}))}}}\indentbegin \begin{hscode}\SaveRestoreHook
\column{B}{@{}>{\hspre}l<{\hspost}@{}}%
\column{3}{@{}>{\hspre}l<{\hspost}@{}}%
\column{6}{@{}>{\hspre}l<{\hspost}@{}}%
\column{E}{@{}>{\hspre}l<{\hspost}@{}}%
\>[6]{}\Varid{h_{Modify1}}\;(\Varid{h_{ND+f}}\;(\Varid{(\Leftrightarrow)}\;(\Varid{local2global_M}\;(\Conid{Op}\;(\Conid{Inl}\;(\Conid{MGet}\;\Varid{k}))))))\;\Varid{s}{}\<[E]%
\\
\>[3]{}\mathrel{=}\mbox{\commentbegin ~  definition of \ensuremath{\Varid{local2global_M}}   \commentend}{}\<[E]%
\\
\>[3]{}\hsindent{3}{}\<[6]%
\>[6]{}\Varid{h_{Modify1}}\;(\Varid{h_{ND+f}}\;(\Varid{(\Leftrightarrow)}\;(\Conid{Op}\;(\Conid{Inl}\;(\Conid{MGet}\;(\Varid{local2global_M}\hsdot{\circ }{.}\Varid{k}))))))\;\Varid{s}{}\<[E]%
\\
\>[3]{}\mathrel{=}\mbox{\commentbegin ~  definition of \ensuremath{\Varid{(\Leftrightarrow)}}   \commentend}{}\<[E]%
\\
\>[3]{}\hsindent{3}{}\<[6]%
\>[6]{}\Varid{h_{Modify1}}\;(\Varid{h_{ND+f}}\;(\Conid{Op}\;(\Conid{Inr}\;(\Conid{Inl}\;(\Conid{MGet}\;(\Varid{(\Leftrightarrow)}\hsdot{\circ }{.}\Varid{local2global_M}\hsdot{\circ }{.}\Varid{k}))))))\;\Varid{s}{}\<[E]%
\\
\>[3]{}\mathrel{=}\mbox{\commentbegin ~  definition of \ensuremath{\Varid{h_{ND+f}}}   \commentend}{}\<[E]%
\\
\>[3]{}\hsindent{3}{}\<[6]%
\>[6]{}\Varid{h_{Modify1}}\;(\Conid{Op}\;(\Conid{Inl}\;(\Conid{MGet}\;(\Varid{h_{ND+f}}\hsdot{\circ }{.}\Varid{(\Leftrightarrow)}\hsdot{\circ }{.}\Varid{local2global_M}\hsdot{\circ }{.}\Varid{k}))))\;\Varid{s}{}\<[E]%
\\
\>[3]{}\mathrel{=}\mbox{\commentbegin ~  definition of \ensuremath{\Varid{h_{Modify1}}}   \commentend}{}\<[E]%
\\
\>[3]{}\hsindent{3}{}\<[6]%
\>[6]{}(\Varid{h_{Modify1}}\hsdot{\circ }{.}\Varid{h_{ND+f}}\hsdot{\circ }{.}\Varid{(\Leftrightarrow)}\hsdot{\circ }{.}\Varid{local2global_M}\hsdot{\circ }{.}\Varid{k})\;\Varid{s}\;\Varid{s}{}\<[E]%
\\
\>[3]{}\mathrel{=}\mbox{\commentbegin ~  definition of \ensuremath{(\hsdot{\circ }{.})}   \commentend}{}\<[E]%
\\
\>[3]{}\hsindent{3}{}\<[6]%
\>[6]{}(\Varid{h_{Modify1}}\;(\Varid{h_{ND+f}}\;(\Varid{(\Leftrightarrow)}\;(\Varid{local2global_M}\;(\Varid{k}\;\Varid{s})))))\;\Varid{s}{}\<[E]%
\\
\>[3]{}\mathrel{=}\mbox{\commentbegin ~  induction hypothesis   \commentend}{}\<[E]%
\\
\>[3]{}\hsindent{3}{}\<[6]%
\>[6]{}\mathbf{do}\;(\Varid{x},\anonymous )\leftarrow \Varid{h_{Modify1}}\;(\Varid{(\Leftrightarrow)}\;(\Varid{h_{ND+f}}\;(\Varid{local2global_M}\;(\Varid{k}\;\Varid{s}))))\;\Varid{s};\Varid{\eta}\;(\Varid{x},\Varid{s}){}\<[E]%
\\
\>[3]{}\mathrel{=}\mbox{\commentbegin ~  definition of \ensuremath{\Varid{local2global_M},\Varid{(\Leftrightarrow)},\Varid{h_{ND+f}},\Varid{h_{Modify1}}}   \commentend}{}\<[E]%
\\
\>[3]{}\hsindent{3}{}\<[6]%
\>[6]{}\mathbf{do}\;(\Varid{x},\anonymous )\leftarrow \Varid{h_{Modify1}}\;(\Varid{h_{ND+f}}\;(\Varid{local2global_M}\;(\Conid{Op}\;(\Conid{Inl}\;(\Conid{MGet}\;\Varid{k})))))\;\Varid{s};\Varid{\eta}\;(\Varid{x},\Varid{s}){}\<[E]%
\ColumnHook
\end{hscode}\resethooks
\indentend \noindent \mbox{\underline{case \ensuremath{\Varid{t}\mathrel{=}\Conid{Op}\;(\Conid{Inr}\;(\Conid{Inl}\;\Conid{Fail}))}}}\indentbegin \begin{hscode}\SaveRestoreHook
\column{B}{@{}>{\hspre}l<{\hspost}@{}}%
\column{3}{@{}>{\hspre}l<{\hspost}@{}}%
\column{6}{@{}>{\hspre}l<{\hspost}@{}}%
\column{E}{@{}>{\hspre}l<{\hspost}@{}}%
\>[6]{}\Varid{h_{Modify1}}\;(\Varid{h_{ND+f}}\;(\Varid{(\Leftrightarrow)}\;(\Varid{local2global_M}\;(\Conid{Op}\;(\Conid{Inr}\;(\Conid{Inl}\;\Conid{Fail}))))))\;\Varid{s}{}\<[E]%
\\
\>[3]{}\mathrel{=}\mbox{\commentbegin ~  definition of \ensuremath{\Varid{local2global_M}}   \commentend}{}\<[E]%
\\
\>[3]{}\hsindent{3}{}\<[6]%
\>[6]{}\Varid{h_{Modify1}}\;(\Varid{h_{ND+f}}\;(\Varid{(\Leftrightarrow)}\;(\Conid{Op}\;(\Conid{Inr}\;(\Conid{Inl}\;\Conid{Fail})))))\;\Varid{s}{}\<[E]%
\\
\>[3]{}\mathrel{=}\mbox{\commentbegin ~  definition of \ensuremath{\Varid{(\Leftrightarrow)}}   \commentend}{}\<[E]%
\\
\>[3]{}\hsindent{3}{}\<[6]%
\>[6]{}\Varid{h_{Modify1}}\;(\Varid{h_{ND+f}}\;(\Conid{Op}\;(\Conid{Inl}\;\Conid{Fail})))\;\Varid{s}{}\<[E]%
\\
\>[3]{}\mathrel{=}\mbox{\commentbegin ~  definition of \ensuremath{\Varid{h_{ND+f}}}   \commentend}{}\<[E]%
\\
\>[3]{}\hsindent{3}{}\<[6]%
\>[6]{}\Varid{h_{Modify1}}\;(\Conid{Var}\;[\mskip1.5mu \mskip1.5mu])\;\Varid{s}{}\<[E]%
\\
\>[3]{}\mathrel{=}\mbox{\commentbegin ~  definition of \ensuremath{\Varid{h_{Modify1}}}   \commentend}{}\<[E]%
\\
\>[3]{}\hsindent{3}{}\<[6]%
\>[6]{}\Conid{Var}\;([\mskip1.5mu \mskip1.5mu],\Varid{s}){}\<[E]%
\\
\>[3]{}\mathrel{=}\mbox{\commentbegin ~  monad law  \commentend}{}\<[E]%
\\
\>[3]{}\hsindent{3}{}\<[6]%
\>[6]{}\mathbf{do}\;(\Varid{x},\anonymous )\leftarrow \Conid{Var}\;([\mskip1.5mu \mskip1.5mu],\Varid{s});\Conid{Var}\;(\Varid{x},\Varid{s}){}\<[E]%
\\
\>[3]{}\mathrel{=}\mbox{\commentbegin ~  definition of \ensuremath{\Varid{local2global_M},\Varid{(\Leftrightarrow)},\Varid{h_{ND+f}},\Varid{h_{Modify1}}}   \commentend}{}\<[E]%
\\
\>[3]{}\hsindent{3}{}\<[6]%
\>[6]{}\mathbf{do}\;(\Varid{x},\anonymous )\leftarrow \Varid{h_{Modify1}}\;(\Varid{h_{ND+f}}\;(\Varid{(\Leftrightarrow)}\;(\Varid{local2global_M}\;(\Conid{Op}\;(\Conid{Inr}\;(\Conid{Inl}\;\Conid{Fail}))))))\;\Varid{s};\Varid{\eta}\;(\Varid{x},\Varid{s}){}\<[E]%
\ColumnHook
\end{hscode}\resethooks
\indentend \noindent \mbox{\underline{case \ensuremath{\Varid{t}\mathrel{=}\Conid{Op}\;(\Conid{Inl}\;(\Conid{MUpdate}\;\Varid{r}\;\Varid{k}))}}}\indentbegin \begin{hscode}\SaveRestoreHook
\column{B}{@{}>{\hspre}l<{\hspost}@{}}%
\column{3}{@{}>{\hspre}l<{\hspost}@{}}%
\column{6}{@{}>{\hspre}l<{\hspost}@{}}%
\column{8}{@{}>{\hspre}l<{\hspost}@{}}%
\column{10}{@{}>{\hspre}l<{\hspost}@{}}%
\column{15}{@{}>{\hspre}l<{\hspost}@{}}%
\column{16}{@{}>{\hspre}l<{\hspost}@{}}%
\column{21}{@{}>{\hspre}l<{\hspost}@{}}%
\column{26}{@{}>{\hspre}l<{\hspost}@{}}%
\column{E}{@{}>{\hspre}l<{\hspost}@{}}%
\>[6]{}\Varid{h_{Modify1}}\;(\Varid{h_{ND+f}}\;(\Varid{(\Leftrightarrow)}\;(\Varid{local2global_M}\;(\Conid{Op}\;(\Conid{Inl}\;(\Conid{Put}\;\Varid{t}\;\Varid{k}))))))\;\Varid{s}{}\<[E]%
\\
\>[3]{}\mathrel{=}\mbox{\commentbegin ~  definition of \ensuremath{\Varid{local2global_M}}   \commentend}{}\<[E]%
\\
\>[3]{}\hsindent{3}{}\<[6]%
\>[6]{}\Varid{h_{Modify1}}\;(\Varid{h_{ND+f}}\;(\Varid{(\Leftrightarrow)}\;((\Varid{update}\;\Varid{r}\mathbin{\talloblong}\Varid{side}\;(\Varid{restore}\;\Varid{r}))>\!\!>\Varid{local2global_M}\;\Varid{k})))\;\Varid{s}{}\<[E]%
\\
\>[3]{}\mathrel{=}\mbox{\commentbegin ~  definition of \ensuremath{(\talloblong)}, \ensuremath{\Varid{update}}, \ensuremath{\Varid{restore}}, \ensuremath{\Varid{side}} and \ensuremath{(>\!\!>\!\!=)}   \commentend}{}\<[E]%
\\
\>[3]{}\hsindent{3}{}\<[6]%
\>[6]{}\Varid{h_{Modify1}}\;(\Varid{h_{ND+f}}\;(\Varid{(\Leftrightarrow)}\;({}\<[E]%
\\
\>[6]{}\hsindent{2}{}\<[8]%
\>[8]{}\Conid{Op}\;(\Conid{Inr}\;(\Conid{Inl}\;(\Conid{Or}\;{}\<[26]%
\>[26]{}(\Conid{Op}\;(\Conid{Inl}\;(\Conid{MUpdate}\;\Varid{r}\;(\Varid{local2global_M}\;\Varid{k}))))\;{}\<[E]%
\\
\>[26]{}(\Conid{Op}\;(\Conid{Inl}\;(\Conid{MRestore}\;\Varid{r}\;(\Conid{Op}\;(\Conid{Inr}\;(\Conid{Inl}\;\Conid{Fail}))))))))))))\;\Varid{s}{}\<[E]%
\\
\>[3]{}\mathrel{=}\mbox{\commentbegin ~  definition of \ensuremath{\Varid{(\Leftrightarrow)}}   \commentend}{}\<[E]%
\\
\>[3]{}\hsindent{3}{}\<[6]%
\>[6]{}\Varid{h_{Modify1}}\;(\Varid{h_{ND+f}}\;({}\<[E]%
\\
\>[6]{}\hsindent{2}{}\<[8]%
\>[8]{}\Conid{Op}\;(\Conid{Inl}\;(\Conid{Or}\;{}\<[21]%
\>[21]{}(\Conid{Op}\;(\Conid{Inr}\;(\Conid{Inl}\;(\Conid{MUpdate}\;\Varid{r}\;(\Varid{(\Leftrightarrow)}\;(\Varid{local2global_M}\;\Varid{k}))))))\;{}\<[E]%
\\
\>[21]{}(\Conid{Op}\;(\Conid{Inr}\;(\Conid{Inl}\;(\Conid{MRestore}\;\Varid{r}\;(\Conid{Op}\;(\Conid{Inl}\;\Conid{Fail}))))))))))\;\Varid{s}{}\<[E]%
\\
\>[3]{}\mathrel{=}\mbox{\commentbegin ~  definition of \ensuremath{\Varid{h_{ND+f}}}   \commentend}{}\<[E]%
\\
\>[3]{}\hsindent{3}{}\<[6]%
\>[6]{}\Varid{h_{Modify1}}\;({}\<[E]%
\\
\>[6]{}\hsindent{2}{}\<[8]%
\>[8]{}\Varid{liftM2}\;(+\!\!+)\;{}\<[21]%
\>[21]{}(\Conid{Op}\;(\Conid{Inl}\;(\Conid{MUpdate}\;\Varid{r}\;(\Varid{h_{ND+f}}\;(\Varid{(\Leftrightarrow)}\;(\Varid{local2global_M}\;\Varid{k}))))))\;{}\<[E]%
\\
\>[21]{}(\Conid{Op}\;(\Conid{Inl}\;(\Conid{MRestore}\;\Varid{r}\;(\Conid{Var}\;[\mskip1.5mu \mskip1.5mu])))))\;\Varid{s}{}\<[E]%
\\
\>[3]{}\mathrel{=}\mbox{\commentbegin ~  definition of \ensuremath{\Varid{liftM2}}  \commentend}{}\<[E]%
\\
\>[3]{}\hsindent{3}{}\<[6]%
\>[6]{}\Varid{h_{Modify1}}\;{}\<[16]%
\>[16]{}(\mathbf{do}\;{}\<[21]%
\>[21]{}\Varid{x}\leftarrow \Conid{Op}\;(\Conid{Inl}\;(\Conid{MUpdate}\;\Varid{r}\;(\Varid{h_{ND+f}}\;(\Varid{(\Leftrightarrow)}\;(\Varid{local2global_M}\;\Varid{k}))))){}\<[E]%
\\
\>[21]{}\Varid{y}\leftarrow \Conid{Op}\;(\Conid{Inl}\;(\Conid{MRestore}\;\Varid{r}\;(\Conid{Var}\;[\mskip1.5mu \mskip1.5mu]))){}\<[E]%
\\
\>[21]{}\Conid{Var}\;(\Varid{x}+\!\!+\Varid{y}){}\<[E]%
\\
\>[6]{}\hsindent{9}{}\<[15]%
\>[15]{})\;\Varid{s}{}\<[E]%
\\
\>[3]{}\mathrel{=}\mbox{\commentbegin ~  Lemma~\ref{lemma:dist-hModify1}  \commentend}{}\<[E]%
\\
\>[3]{}\hsindent{3}{}\<[6]%
\>[6]{}\mathbf{do}\;{}\<[10]%
\>[10]{}(\Varid{x},\Varid{s}_{1})\leftarrow \Varid{h_{Modify1}}\;(\Conid{Op}\;(\Conid{Inl}\;(\Conid{MUpdate}\;\Varid{r}\;(\Varid{h_{ND+f}}\;(\Varid{(\Leftrightarrow)}\;(\Varid{local2global_M}\;\Varid{k}))))))\;\Varid{s}{}\<[E]%
\\
\>[10]{}(\Varid{y},\Varid{s}_{2})\leftarrow \Varid{h_{Modify1}}\;(\Conid{Op}\;(\Conid{Inl}\;(\Conid{MRestore}\;\Varid{r}\;(\Conid{Var}\;[\mskip1.5mu \mskip1.5mu]))))\;\Varid{s}_{1}{}\<[E]%
\\
\>[10]{}\Conid{Var}\;(\Varid{x}+\!\!+\Varid{y},\Varid{s}_{2}){}\<[E]%
\\
\>[3]{}\mathrel{=}\mbox{\commentbegin ~  definition of \ensuremath{\Varid{h_{Modify1}}}  \commentend}{}\<[E]%
\\
\>[3]{}\hsindent{3}{}\<[6]%
\>[6]{}\mathbf{do}\;{}\<[10]%
\>[10]{}(\Varid{x},\Varid{s}_{1})\leftarrow \Varid{h_{Modify1}}\;(\Varid{h_{ND+f}}\;(\Varid{(\Leftrightarrow)}\;(\Varid{local2global_M}\;\Varid{k})))\;(\Varid{s}\mathbin{\oplus}\Varid{r}){}\<[E]%
\\
\>[10]{}(\Varid{y},\Varid{s}_{2})\leftarrow \Conid{Var}\;([\mskip1.5mu \mskip1.5mu],\Varid{s}_{1}\mathbin{\ominus}\Varid{r}){}\<[E]%
\\
\>[10]{}\Conid{Var}\;(\Varid{x}+\!\!+\Varid{y},\Varid{s}_{2}){}\<[E]%
\\
\>[3]{}\mathrel{=}\mbox{\commentbegin ~  monad laws  \commentend}{}\<[E]%
\\
\>[3]{}\hsindent{3}{}\<[6]%
\>[6]{}\mathbf{do}\;{}\<[10]%
\>[10]{}(\Varid{x},\Varid{s}_{1})\leftarrow \Varid{h_{Modify1}}\;(\Varid{h_{ND+f}}\;(\Varid{(\Leftrightarrow)}\;(\Varid{local2global_M}\;\Varid{k})))\;(\Varid{s}\mathbin{\oplus}\Varid{r}){}\<[E]%
\\
\>[10]{}\Conid{Var}\;(\Varid{x}+\!\!+[\mskip1.5mu \mskip1.5mu],\Varid{s}_{1}\mathbin{\ominus}\Varid{r}){}\<[E]%
\\
\>[3]{}\mathrel{=}\mbox{\commentbegin ~  right unit of \ensuremath{(+\!\!+)}  \commentend}{}\<[E]%
\\
\>[3]{}\hsindent{3}{}\<[6]%
\>[6]{}\mathbf{do}\;{}\<[10]%
\>[10]{}(\Varid{x},\Varid{s}_{1})\leftarrow \Varid{h_{Modify1}}\;(\Varid{h_{ND+f}}\;(\Varid{(\Leftrightarrow)}\;(\Varid{local2global_M}\;\Varid{k})))\;(\Varid{s}\mathbin{\oplus}\Varid{r}){}\<[E]%
\\
\>[10]{}\Conid{Var}\;(\Varid{x},\Varid{s}_{1}\mathbin{\ominus}\Varid{r}){}\<[E]%
\\
\>[3]{}\mathrel{=}\mbox{\commentbegin ~  induction hypothesis  \commentend}{}\<[E]%
\\
\>[3]{}\hsindent{3}{}\<[6]%
\>[6]{}\mathbf{do}\;{}\<[10]%
\>[10]{}(\Varid{x},\Varid{s}_{1})\leftarrow \mathbf{do}\;\{\mskip1.5mu (\Varid{x},\anonymous )\leftarrow \Varid{h_{Modify1}}\;(\Varid{h_{ND+f}}\;(\Varid{(\Leftrightarrow)}\;(\Varid{local2global_M}\;\Varid{k})))\;(\Varid{s}\mathbin{\oplus}\Varid{r});\Varid{\eta}\;(\Varid{x},\Varid{s}\mathbin{\oplus}\Varid{r})\mskip1.5mu\}{}\<[E]%
\\
\>[10]{}\Conid{Var}\;(\Varid{x},\Varid{s}_{1}\mathbin{\ominus}\Varid{r}){}\<[E]%
\\
\>[3]{}\mathrel{=}\mbox{\commentbegin ~  monad laws   \commentend}{}\<[E]%
\\
\>[3]{}\hsindent{3}{}\<[6]%
\>[6]{}\mathbf{do}\;{}\<[10]%
\>[10]{}(\Varid{x},\anonymous )\leftarrow \Varid{h_{Modify1}}\;(\Varid{h_{ND+f}}\;(\Varid{(\Leftrightarrow)}\;(\Varid{local2global_M}\;\Varid{k})))\;(\Varid{s}\mathbin{\oplus}\Varid{r}){}\<[E]%
\\
\>[10]{}\Conid{Var}\;(\Varid{x},(\Varid{s}\mathbin{\oplus}\Varid{r})\mathbin{\ominus}\Varid{r}){}\<[E]%
\\
\>[3]{}\mathrel{=}\mbox{\commentbegin ~  \Cref{eq:plus-minus}   \commentend}{}\<[E]%
\\
\>[3]{}\hsindent{3}{}\<[6]%
\>[6]{}\mathbf{do}\;{}\<[10]%
\>[10]{}(\Varid{x},\anonymous )\leftarrow \Varid{h_{Modify1}}\;(\Varid{h_{ND+f}}\;(\Varid{(\Leftrightarrow)}\;(\Varid{local2global_M}\;\Varid{k})))\;(\Varid{s}\mathbin{\oplus}\Varid{r}){}\<[E]%
\\
\>[10]{}\Conid{Var}\;(\Varid{x},\Varid{s}){}\<[E]%
\\
\>[3]{}\mathrel{=}\mbox{\commentbegin ~  monad laws   \commentend}{}\<[E]%
\\
\>[3]{}\hsindent{3}{}\<[6]%
\>[6]{}\mathbf{do}\;{}\<[10]%
\>[10]{}(\Varid{x},\anonymous )\leftarrow \mathbf{do}\;\{\mskip1.5mu (\Varid{x},\anonymous )\leftarrow \Varid{h_{Modify1}}\;(\Varid{h_{ND+f}}\;(\Varid{(\Leftrightarrow)}\;(\Varid{local2global_M}\;\Varid{k})))\;(\Varid{s}\mathbin{\oplus}\Varid{r});\Varid{\eta}\;(\Varid{x},\Varid{s})\mskip1.5mu\}{}\<[E]%
\\
\>[10]{}\Conid{Var}\;(\Varid{x},\Varid{s}){}\<[E]%
\\
\>[3]{}\mathrel{=}\mbox{\commentbegin ~  deriviation in reverse   \commentend}{}\<[E]%
\\
\>[3]{}\hsindent{3}{}\<[6]%
\>[6]{}\mathbf{do}\;{}\<[10]%
\>[10]{}(\Varid{x},\anonymous )\leftarrow \Varid{h_{Modify1}}\;(\Varid{h_{ND+f}}\;(\Varid{(\Leftrightarrow)}\;(\Varid{local2global_M}\;(\Conid{Op}\;(\Conid{Inl}\;(\Conid{MUpdate}\;\Varid{r}\;\Varid{k}))))))\;\Varid{s}{}\<[E]%
\\
\>[10]{}\Conid{Var}\;(\Varid{x},\Varid{s}){}\<[E]%
\ColumnHook
\end{hscode}\resethooks
\indentend \noindent \mbox{\underline{case \ensuremath{\Varid{t}\mathrel{=}\Conid{Op}\;(\Conid{Inr}\;(\Conid{Inl}\;(\Conid{Or}\;\Varid{p}\;\Varid{q})))}}}\indentbegin \begin{hscode}\SaveRestoreHook
\column{B}{@{}>{\hspre}l<{\hspost}@{}}%
\column{3}{@{}>{\hspre}l<{\hspost}@{}}%
\column{6}{@{}>{\hspre}l<{\hspost}@{}}%
\column{10}{@{}>{\hspre}l<{\hspost}@{}}%
\column{12}{@{}>{\hspre}l<{\hspost}@{}}%
\column{14}{@{}>{\hspre}l<{\hspost}@{}}%
\column{16}{@{}>{\hspre}l<{\hspost}@{}}%
\column{20}{@{}>{\hspre}l<{\hspost}@{}}%
\column{E}{@{}>{\hspre}l<{\hspost}@{}}%
\>[6]{}\Varid{h_{Modify1}}\;(\Varid{h_{ND+f}}\;(\Varid{(\Leftrightarrow)}\;(\Varid{local2global_M}\;(\Conid{Op}\;(\Conid{Inr}\;(\Conid{Inl}\;(\Conid{Or}\;\Varid{p}\;\Varid{q})))))))\;\Varid{s}{}\<[E]%
\\
\>[3]{}\mathrel{=}\mbox{\commentbegin ~  definition of \ensuremath{\Varid{local2global_M}}   \commentend}{}\<[E]%
\\
\>[3]{}\hsindent{3}{}\<[6]%
\>[6]{}\Varid{h_{Modify1}}\;(\Varid{h_{ND+f}}\;(\Varid{(\Leftrightarrow)}\;(\Conid{Op}\;(\Conid{Inr}\;(\Conid{Inl}\;(\Conid{Or}\;(\Varid{local2global_M}\;\Varid{p})\;(\Varid{local2global_M}\;\Varid{q})))))))\;\Varid{s}{}\<[E]%
\\
\>[3]{}\mathrel{=}\mbox{\commentbegin ~  definition of \ensuremath{\Varid{(\Leftrightarrow)}}   \commentend}{}\<[E]%
\\
\>[3]{}\hsindent{3}{}\<[6]%
\>[6]{}\Varid{h_{Modify1}}\;(\Varid{h_{ND+f}}\;(\Conid{Op}\;(\Conid{Inl}\;(\Conid{Or}\;(\Varid{(\Leftrightarrow)}\;(\Varid{local2global_M}\;\Varid{p}))\;(\Varid{(\Leftrightarrow)}\;(\Varid{local2global_M}\;\Varid{q}))))))\;\Varid{s}{}\<[E]%
\\
\>[3]{}\mathrel{=}\mbox{\commentbegin ~  definition of \ensuremath{\Varid{h_{ND+f}}}   \commentend}{}\<[E]%
\\
\>[3]{}\hsindent{3}{}\<[6]%
\>[6]{}\Varid{h_{Modify1}}\;(\Varid{liftM2}\;(+\!\!+)\;(\Varid{h_{ND+f}}\;(\Varid{(\Leftrightarrow)}\;(\Varid{local2global_M}\;\Varid{p})))\;(\Varid{h_{ND+f}}\;(\Varid{(\Leftrightarrow)}\;(\Varid{local2global_M}\;\Varid{q}))))\;\Varid{s}{}\<[E]%
\\
\>[3]{}\mathrel{=}\mbox{\commentbegin ~  definition of \ensuremath{\Varid{liftM2}}   \commentend}{}\<[E]%
\\
\>[3]{}\hsindent{3}{}\<[6]%
\>[6]{}\Varid{h_{Modify1}}\;(\mathbf{do}\;{}\<[20]%
\>[20]{}\Varid{x}\leftarrow \Varid{h_{ND+f}}\;(\Varid{(\Leftrightarrow)}\;(\Varid{local2global_M}\;\Varid{p})){}\<[E]%
\\
\>[20]{}\Varid{y}\leftarrow \Varid{h_{ND+f}}\;(\Varid{(\Leftrightarrow)}\;(\Varid{local2global_M}\;\Varid{q})){}\<[E]%
\\
\>[20]{}\Conid{Var}\;(\Varid{x}+\!\!+\Varid{y}){}\<[E]%
\\
\>[6]{}\hsindent{8}{}\<[14]%
\>[14]{})\;\Varid{s}{}\<[E]%
\\
\>[3]{}\mathrel{=}\mbox{\commentbegin ~  Lemma~\ref{lemma:dist-hModify1}   \commentend}{}\<[E]%
\\
\>[3]{}\hsindent{3}{}\<[6]%
\>[6]{}\mathbf{do}\;{}\<[10]%
\>[10]{}(\Varid{x},\Varid{s}_{1})\leftarrow \Varid{h_{Modify1}}\;(\Varid{h_{ND+f}}\;(\Varid{(\Leftrightarrow)}\;(\Varid{local2global_M}\;\Varid{p})))\;\Varid{s}{}\<[E]%
\\
\>[10]{}(\Varid{y},\Varid{s}_{2})\leftarrow \Varid{h_{Modify1}}\;(\Varid{h_{ND+f}}\;(\Varid{(\Leftrightarrow)}\;(\Varid{local2global_M}\;\Varid{q})))\;\Varid{s}_{1}{}\<[E]%
\\
\>[10]{}\Varid{h_{Modify1}}\;(\Conid{Var}\;(\Varid{x}+\!\!+\Varid{y}))\;\Varid{s}_{2}{}\<[E]%
\\
\>[3]{}\mathrel{=}\mbox{\commentbegin ~  induction hypothesis   \commentend}{}\<[E]%
\\
\>[3]{}\hsindent{3}{}\<[6]%
\>[6]{}\mathbf{do}\;{}\<[10]%
\>[10]{}(\Varid{x},\Varid{s}_{1})\leftarrow \mathbf{do}\;\{\mskip1.5mu (\Varid{x},\anonymous )\leftarrow \Varid{h_{Modify1}}\;(\Varid{h_{ND+f}}\;(\Varid{(\Leftrightarrow)}\;(\Varid{local2global_M}\;\Varid{p})))\;\Varid{s};\Varid{\eta}\;(\Varid{x},\Varid{s})\mskip1.5mu\}{}\<[E]%
\\
\>[10]{}(\Varid{y},\Varid{s}_{2})\leftarrow \mathbf{do}\;\{\mskip1.5mu (\Varid{y},\anonymous )\leftarrow \Varid{h_{Modify1}}\;(\Varid{h_{ND+f}}\;(\Varid{(\Leftrightarrow)}\;(\Varid{local2global_M}\;\Varid{q})))\;\Varid{s}_{1};\Varid{\eta}\;(\Varid{y},\Varid{s}_{1})\mskip1.5mu\}{}\<[E]%
\\
\>[10]{}\Varid{h_{Modify1}}\;(\Conid{Var}\;(\Varid{x}+\!\!+\Varid{y}))\;\Varid{s}_{2}{}\<[E]%
\\
\>[3]{}\mathrel{=}\mbox{\commentbegin ~  monad laws   \commentend}{}\<[E]%
\\
\>[3]{}\hsindent{3}{}\<[6]%
\>[6]{}\mathbf{do}\;{}\<[10]%
\>[10]{}(\Varid{x},\anonymous )\leftarrow \Varid{h_{Modify1}}\;(\Varid{h_{ND+f}}\;(\Varid{(\Leftrightarrow)}\;(\Varid{local2global_M}\;\Varid{p})))\;\Varid{s}{}\<[E]%
\\
\>[10]{}(\Varid{y},\anonymous )\leftarrow \Varid{h_{Modify1}}\;(\Varid{h_{ND+f}}\;(\Varid{(\Leftrightarrow)}\;(\Varid{local2global_M}\;\Varid{q})))\;\Varid{s}_{1}{}\<[E]%
\\
\>[10]{}\Varid{h_{Modify1}}\;(\Conid{Var}\;(\Varid{x}+\!\!+\Varid{y}))\;\Varid{s}{}\<[E]%
\\
\>[3]{}\mathrel{=}\mbox{\commentbegin ~  definition of \ensuremath{\Varid{h_{Modify1}}}  \commentend}{}\<[E]%
\\
\>[3]{}\hsindent{3}{}\<[6]%
\>[6]{}\mathbf{do}\;{}\<[10]%
\>[10]{}(\Varid{x},\anonymous )\leftarrow \Varid{h_{Modify1}}\;(\Varid{h_{ND+f}}\;(\Varid{(\Leftrightarrow)}\;(\Varid{local2global_M}\;\Varid{p})))\;\Varid{s}{}\<[E]%
\\
\>[10]{}(\Varid{y},\anonymous )\leftarrow \Varid{h_{Modify1}}\;(\Varid{h_{ND+f}}\;(\Varid{(\Leftrightarrow)}\;(\Varid{local2global_M}\;\Varid{q})))\;\Varid{s}{}\<[E]%
\\
\>[10]{}\Varid{\eta}\;(\Varid{x}+\!\!+\Varid{y},\Varid{s}){}\<[E]%
\\
\>[3]{}\mathrel{=}\mbox{\commentbegin ~  monad laws   \commentend}{}\<[E]%
\\
\>[3]{}\hsindent{3}{}\<[6]%
\>[6]{}\mathbf{do}\;{}\<[10]%
\>[10]{}(\Varid{x},\anonymous )\leftarrow ({}\<[E]%
\\
\>[10]{}\hsindent{2}{}\<[12]%
\>[12]{}\mathbf{do}\;{}\<[16]%
\>[16]{}(\Varid{x},\anonymous )\leftarrow \Varid{h_{Modify1}}\;(\Varid{h_{ND+f}}\;(\Varid{(\Leftrightarrow)}\;(\Varid{local2global_M}\;\Varid{p})))\;\Varid{s}{}\<[E]%
\\
\>[16]{}(\Varid{y},\anonymous )\leftarrow \Varid{h_{Modify1}}\;(\Varid{h_{ND+f}}\;(\Varid{(\Leftrightarrow)}\;(\Varid{local2global_M}\;\Varid{q})))\;\Varid{s}_{1}{}\<[E]%
\\
\>[16]{}\Varid{\eta}\;(\Varid{x}+\!\!+\Varid{y},\Varid{s}){}\<[E]%
\\
\>[10]{}\hsindent{2}{}\<[12]%
\>[12]{}){}\<[E]%
\\
\>[10]{}\Varid{\eta}\;(\Varid{x},\Varid{s}){}\<[E]%
\\
\>[3]{}\mathrel{=}\mbox{\commentbegin ~  derivation in reverse (similar to before)   \commentend}{}\<[E]%
\\
\>[3]{}\hsindent{3}{}\<[6]%
\>[6]{}\mathbf{do}\;{}\<[10]%
\>[10]{}(\Varid{x},\anonymous )\leftarrow \Varid{h_{Modify1}}\;(\Varid{h_{ND+f}}\;(\Varid{(\Leftrightarrow)}\;(\Varid{local2global_M}\;(\Conid{Op}\;(\Conid{Inr}\;(\Conid{Inl}\;(\Conid{Or}\;\Varid{p}\;\Varid{q})))))))\;\Varid{s}{}\<[E]%
\\
\>[10]{}\Varid{\eta}\;(\Varid{x},\Varid{s}){}\<[E]%
\ColumnHook
\end{hscode}\resethooks
\indentend %

\noindent \mbox{\underline{case \ensuremath{\Varid{t}\mathrel{=}\Conid{Op}\;(\Conid{Inr}\;(\Conid{Inr}\;\Varid{y}))}}}\indentbegin \begin{hscode}\SaveRestoreHook
\column{B}{@{}>{\hspre}l<{\hspost}@{}}%
\column{3}{@{}>{\hspre}l<{\hspost}@{}}%
\column{6}{@{}>{\hspre}l<{\hspost}@{}}%
\column{10}{@{}>{\hspre}l<{\hspost}@{}}%
\column{E}{@{}>{\hspre}l<{\hspost}@{}}%
\>[6]{}\Varid{h_{Modify1}}\;(\Varid{h_{ND+f}}\;(\Varid{(\Leftrightarrow)}\;(\Varid{local2global_M}\;(\Conid{Op}\;(\Conid{Inr}\;(\Conid{Inr}\;\Varid{y}))))))\;\Varid{s}{}\<[E]%
\\
\>[3]{}\mathrel{=}\mbox{\commentbegin ~  definition of \ensuremath{\Varid{local2global_M}}   \commentend}{}\<[E]%
\\
\>[3]{}\hsindent{3}{}\<[6]%
\>[6]{}\Varid{h_{Modify1}}\;(\Varid{h_{ND+f}}\;(\Varid{(\Leftrightarrow)}\;(\Conid{Op}\;(\Conid{Inr}\;(\Conid{Inr}\;(\Varid{fmap}\;\Varid{local2global_M}\;\Varid{y}))))))\;\Varid{s}{}\<[E]%
\\
\>[3]{}\mathrel{=}\mbox{\commentbegin ~  definition of \ensuremath{\Varid{(\Leftrightarrow)}}; \ensuremath{\Varid{fmap}} fusion   \commentend}{}\<[E]%
\\
\>[3]{}\hsindent{3}{}\<[6]%
\>[6]{}\Varid{h_{Modify1}}\;(\Varid{h_{ND+f}}\;(\Conid{Op}\;(\Conid{Inr}\;(\Conid{Inr}\;(\Varid{fmap}\;(\Varid{(\Leftrightarrow)}\hsdot{\circ }{.}\Varid{local2global_M})\;\Varid{y})))))\;\Varid{s}{}\<[E]%
\\
\>[3]{}\mathrel{=}\mbox{\commentbegin ~  definition of \ensuremath{\Varid{h_{ND+f}}}; \ensuremath{\Varid{fmap}} fusion   \commentend}{}\<[E]%
\\
\>[3]{}\hsindent{3}{}\<[6]%
\>[6]{}\Varid{h_{Modify1}}\;(\Conid{Op}\;(\Conid{Inr}\;(\Varid{fmap}\;(\Varid{h_{ND+f}}\hsdot{\circ }{.}\Varid{(\Leftrightarrow)}\hsdot{\circ }{.}\Varid{local2global_M})\;\Varid{y})))\;\Varid{s}{}\<[E]%
\\
\>[3]{}\mathrel{=}\mbox{\commentbegin ~  definition of \ensuremath{\Varid{h_{Modify1}}}; \ensuremath{\Varid{fmap}} fusion   \commentend}{}\<[E]%
\\
\>[3]{}\hsindent{3}{}\<[6]%
\>[6]{}\Conid{Op}\;(\Varid{fmap}\;((\mathbin{\$}\Varid{s})\hsdot{\circ }{.}\Varid{h_{Modify1}}\hsdot{\circ }{.}\Varid{h_{ND+f}}\hsdot{\circ }{.}\Varid{(\Leftrightarrow)}\hsdot{\circ }{.}\Varid{local2global_M})\;\Varid{y}){}\<[E]%
\\
\>[3]{}\mathrel{=}\mbox{\commentbegin ~  induction hypothesis   \commentend}{}\<[E]%
\\
\>[3]{}\hsindent{3}{}\<[6]%
\>[6]{}\Conid{Op}\;(\Varid{fmap}\;((>\!\!>\!\!=\lambda (\Varid{x},\anonymous )\to \Varid{\eta}\;(\Varid{x},\Varid{s}))\hsdot{\circ }{.}(\mathbin{\$}\Varid{s})\hsdot{\circ }{.}\Varid{h_{Modify1}}\hsdot{\circ }{.}\Varid{h_{ND+f}}\hsdot{\circ }{.}\Varid{(\Leftrightarrow)}\hsdot{\circ }{.}\Varid{local2global_M})\;\Varid{y}){}\<[E]%
\\
\>[3]{}\mathrel{=}\mbox{\commentbegin ~  \ensuremath{\Varid{fmap}} fission; definition of \ensuremath{(>\!\!>\!\!=)}   \commentend}{}\<[E]%
\\
\>[3]{}\hsindent{3}{}\<[6]%
\>[6]{}\mathbf{do}\;{}\<[10]%
\>[10]{}(\Varid{x},\anonymous )\leftarrow \Conid{Op}\;(\Varid{fmap}\;((\mathbin{\$}\Varid{s})\hsdot{\circ }{.}\Varid{h_{Modify1}}\hsdot{\circ }{.}\Varid{h_{ND+f}}\hsdot{\circ }{.}\Varid{(\Leftrightarrow)}\hsdot{\circ }{.}\Varid{local2global_M})\;\Varid{y}){}\<[E]%
\\
\>[10]{}\Varid{\eta}\;(\Varid{x},\Varid{s}){}\<[E]%
\\
\>[3]{}\mathrel{=}\mbox{\commentbegin ~  deriviation in reverse (similar to before)   \commentend}{}\<[E]%
\\
\>[3]{}\hsindent{3}{}\<[6]%
\>[6]{}\mathbf{do}\;{}\<[10]%
\>[10]{}(\Varid{x},\anonymous )\leftarrow \Varid{h_{Modify1}}\;(\Varid{h_{ND+f}}\;(\Varid{(\Leftrightarrow)}\;(\Varid{local2global_M}\;(\Conid{Op}\;(\Conid{Inr}\;(\Conid{Inr}\;\Varid{y}))))))\;\Varid{s}{}\<[E]%
\\
\>[10]{}\Varid{\eta}\;(\Varid{x},\Varid{s}){}\<[E]%
\ColumnHook
\end{hscode}\resethooks
\indentend \end{proof}

\subsection{Auxiliary Lemmas}

The derivations above made use of several auxliary lemmas.
We prove them here.

\begin{lemma}[Distributivity of \ensuremath{\Varid{h_{Modify1}}}] \label{lemma:dist-hModify1} \ \\
\[
\ensuremath{\Varid{h_{Modify1}}\;(\Varid{p}>\!\!>\!\!=\Varid{k})\;\Varid{s}} \quad =\quad \ensuremath{\Varid{h_{Modify1}}\;\Varid{p}\;\Varid{s}>\!\!>\!\!=\lambda (\Varid{x},\Varid{s'})\to \Varid{h_{Modify1}}\;(\Varid{k}\;\Varid{x})\;\Varid{s'}}
\]
\end{lemma}

\begin{proof}
The proof follows the same structure of \Cref{lemma:dist-hState1}.
We proceed by induction on \ensuremath{\Varid{p}}.

\noindent \mbox{\underline{case \ensuremath{\Varid{p}\mathrel{=}\Conid{Var}\;\Varid{x}}}}
\indentbegin \begin{hscode}\SaveRestoreHook
\column{B}{@{}>{\hspre}l<{\hspost}@{}}%
\column{3}{@{}>{\hspre}l<{\hspost}@{}}%
\column{6}{@{}>{\hspre}l<{\hspost}@{}}%
\column{E}{@{}>{\hspre}l<{\hspost}@{}}%
\>[6]{}\Varid{h_{Modify1}}\;(\Conid{Var}\;\Varid{x}>\!\!>\!\!=\Varid{k})\;\Varid{s}{}\<[E]%
\\
\>[3]{}\mathrel{=}\mbox{\commentbegin ~  monad law   \commentend}{}\<[E]%
\\
\>[3]{}\hsindent{3}{}\<[6]%
\>[6]{}\Varid{h_{Modify1}}\;(\Varid{k}\;\Varid{x})\;\Varid{s}{}\<[E]%
\\
\>[3]{}\mathrel{=}\mbox{\commentbegin ~  monad law   \commentend}{}\<[E]%
\\
\>[3]{}\hsindent{3}{}\<[6]%
\>[6]{}\Varid{\eta}\;(\Varid{x},\Varid{s})>\!\!>\!\!=\lambda (\Varid{x},\Varid{s'})\to \Varid{h_{Modify1}}\;(\Varid{k}\;\Varid{x})\;\Varid{s'}{}\<[E]%
\\
\>[3]{}\mathrel{=}\mbox{\commentbegin ~  definition of \ensuremath{\Varid{h_{Modify1}}}   \commentend}{}\<[E]%
\\
\>[3]{}\hsindent{3}{}\<[6]%
\>[6]{}\Varid{h_{Modify1}}\;(\Conid{Var}\;\Varid{x})\;\Varid{s}>\!\!>\!\!=\lambda (\Varid{x},\Varid{s'})\to \Varid{h_{Modify1}}\;(\Varid{k}\;\Varid{x})\;\Varid{s'}{}\<[E]%
\ColumnHook
\end{hscode}\resethooks
\indentend \noindent \mbox{\underline{case \ensuremath{\Varid{p}\mathrel{=}\Conid{Op}\;(\Conid{Inl}\;(\Conid{MGet}\;\Varid{p}))}}}
\indentbegin \begin{hscode}\SaveRestoreHook
\column{B}{@{}>{\hspre}l<{\hspost}@{}}%
\column{3}{@{}>{\hspre}l<{\hspost}@{}}%
\column{6}{@{}>{\hspre}l<{\hspost}@{}}%
\column{E}{@{}>{\hspre}l<{\hspost}@{}}%
\>[6]{}\Varid{h_{Modify1}}\;(\Conid{Op}\;(\Conid{Inl}\;(\Conid{MGet}\;\Varid{p}))>\!\!>\!\!=\Varid{k})\;\Varid{s}{}\<[E]%
\\
\>[3]{}\mathrel{=}\mbox{\commentbegin ~  definition of \ensuremath{(>\!\!>\!\!=)} for free monad   \commentend}{}\<[E]%
\\
\>[3]{}\hsindent{3}{}\<[6]%
\>[6]{}\Varid{h_{Modify1}}\;(\Conid{Op}\;(\Varid{fmap}\;(>\!\!>\!\!=\Varid{k})\;(\Conid{Inl}\;(\Conid{MGet}\;\Varid{p}))))\;\Varid{s}{}\<[E]%
\\
\>[3]{}\mathrel{=}\mbox{\commentbegin ~  definition of \ensuremath{\Varid{fmap}} for coproduct \ensuremath{(\mathrel{{:}{+}{:}})}   \commentend}{}\<[E]%
\\
\>[3]{}\hsindent{3}{}\<[6]%
\>[6]{}\Varid{h_{Modify1}}\;(\Conid{Op}\;(\Conid{Inl}\;(\Varid{fmap}\;(>\!\!>\!\!=\Varid{k})\;(\Conid{MGet}\;\Varid{p}))))\;\Varid{s}{}\<[E]%
\\
\>[3]{}\mathrel{=}\mbox{\commentbegin ~  definition of \ensuremath{\Varid{fmap}} for \ensuremath{\Conid{MGet}}   \commentend}{}\<[E]%
\\
\>[3]{}\hsindent{3}{}\<[6]%
\>[6]{}\Varid{h_{Modify1}}\;(\Conid{Op}\;(\Conid{Inl}\;(\Conid{MGet}\;(\lambda \Varid{x}\to \Varid{p}\;\Varid{s}>\!\!>\!\!=\Varid{k}))))\;\Varid{s}{}\<[E]%
\\
\>[3]{}\mathrel{=}\mbox{\commentbegin ~  definition of \ensuremath{\Varid{h_{Modify1}}}   \commentend}{}\<[E]%
\\
\>[3]{}\hsindent{3}{}\<[6]%
\>[6]{}\Varid{h_{Modify1}}\;(\Varid{p}\;\Varid{s}>\!\!>\!\!=\Varid{k})\;\Varid{s}{}\<[E]%
\\
\>[3]{}\mathrel{=}\mbox{\commentbegin ~  induction hypothesis   \commentend}{}\<[E]%
\\
\>[3]{}\hsindent{3}{}\<[6]%
\>[6]{}\Varid{h_{Modify1}}\;(\Varid{p}\;\Varid{s})\;\Varid{s}>\!\!>\!\!=\lambda (\Varid{x},\Varid{s'})\to \Varid{h_{Modify1}}\;(\Varid{k}\;\Varid{x})\;\Varid{s'}{}\<[E]%
\\
\>[3]{}\mathrel{=}\mbox{\commentbegin ~  definition of \ensuremath{\Varid{h_{Modify1}}}   \commentend}{}\<[E]%
\\
\>[3]{}\hsindent{3}{}\<[6]%
\>[6]{}\Varid{h_{Modify1}}\;(\Conid{Op}\;(\Conid{Inl}\;(\Conid{MGet}\;\Varid{p})))\;\Varid{s}>\!\!>\!\!=\lambda (\Varid{x},\Varid{s'})\to \Varid{h_{Modify1}}\;(\Varid{k}\;\Varid{x})\;\Varid{s'}{}\<[E]%
\ColumnHook
\end{hscode}\resethooks
\indentend \noindent \mbox{\underline{case \ensuremath{\Varid{p}\mathrel{=}\Conid{Op}\;(\Conid{Inl}\;(\Conid{MUpdate}\;\Varid{r}\;\Varid{p}))}}}
\indentbegin \begin{hscode}\SaveRestoreHook
\column{B}{@{}>{\hspre}l<{\hspost}@{}}%
\column{3}{@{}>{\hspre}l<{\hspost}@{}}%
\column{6}{@{}>{\hspre}l<{\hspost}@{}}%
\column{E}{@{}>{\hspre}l<{\hspost}@{}}%
\>[6]{}\Varid{h_{Modify1}}\;(\Conid{Op}\;(\Conid{Inl}\;(\Conid{MUpdate}\;\Varid{r}\;\Varid{p}))>\!\!>\!\!=\Varid{k})\;\Varid{s}{}\<[E]%
\\
\>[3]{}\mathrel{=}\mbox{\commentbegin ~  definition of \ensuremath{(>\!\!>\!\!=)} for free monad   \commentend}{}\<[E]%
\\
\>[3]{}\hsindent{3}{}\<[6]%
\>[6]{}\Varid{h_{Modify1}}\;(\Conid{Op}\;(\Varid{fmap}\;(>\!\!>\!\!=\Varid{k})\;(\Conid{Inl}\;(\Conid{MUpdate}\;\Varid{r}\;\Varid{p}))))\;\Varid{s}{}\<[E]%
\\
\>[3]{}\mathrel{=}\mbox{\commentbegin ~  definition of \ensuremath{\Varid{fmap}} for coproduct \ensuremath{(\mathrel{{:}{+}{:}})}   \commentend}{}\<[E]%
\\
\>[3]{}\hsindent{3}{}\<[6]%
\>[6]{}\Varid{h_{Modify1}}\;(\Conid{Op}\;(\Conid{Inl}\;(\Varid{fmap}\;(>\!\!>\!\!=\Varid{k})\;(\Conid{MUpdate}\;\Varid{r}\;\Varid{p}))))\;\Varid{s}{}\<[E]%
\\
\>[3]{}\mathrel{=}\mbox{\commentbegin ~  definition of \ensuremath{\Varid{fmap}} for \ensuremath{\Conid{MUpdate}}   \commentend}{}\<[E]%
\\
\>[3]{}\hsindent{3}{}\<[6]%
\>[6]{}\Varid{h_{Modify1}}\;(\Conid{Op}\;(\Conid{Inl}\;(\Conid{MUpdate}\;\Varid{r}\;(\Varid{p}>\!\!>\!\!=\Varid{k}))))\;\Varid{s}{}\<[E]%
\\
\>[3]{}\mathrel{=}\mbox{\commentbegin ~  definition of \ensuremath{\Varid{h_{Modify1}}}   \commentend}{}\<[E]%
\\
\>[3]{}\hsindent{3}{}\<[6]%
\>[6]{}\Varid{h_{Modify1}}\;(\Varid{p}>\!\!>\!\!=\Varid{k})\;(\Varid{s}\mathbin{\oplus}\Varid{r}){}\<[E]%
\\
\>[3]{}\mathrel{=}\mbox{\commentbegin ~  induction hypothesis   \commentend}{}\<[E]%
\\
\>[3]{}\hsindent{3}{}\<[6]%
\>[6]{}\Varid{h_{Modify1}}\;\Varid{p}\;(\Varid{s}\mathbin{\oplus}\Varid{r})>\!\!>\!\!=\lambda (\Varid{x},\Varid{s'})\to \Varid{h_{Modify1}}\;(\Varid{k}\;\Varid{x})\;\Varid{s'}{}\<[E]%
\\
\>[3]{}\mathrel{=}\mbox{\commentbegin ~  definition of \ensuremath{\Varid{h_{Modify1}}}   \commentend}{}\<[E]%
\\
\>[3]{}\hsindent{3}{}\<[6]%
\>[6]{}\Varid{h_{Modify1}}\;(\Conid{Op}\;(\Conid{Inl}\;(\Conid{MUpdate}\;\Varid{r}\;\Varid{p})))\;\Varid{s}>\!\!>\!\!=\lambda (\Varid{x},\Varid{s'})\to \Varid{h_{Modify1}}\;(\Varid{k}\;\Varid{x})\;\Varid{s'}{}\<[E]%
\ColumnHook
\end{hscode}\resethooks
\indentend \noindent \mbox{\underline{case \ensuremath{\Varid{p}\mathrel{=}\Conid{Op}\;(\Conid{Inl}\;(\Conid{MRestore}\;\Varid{r}\;\Varid{p}))}}}
\indentbegin \begin{hscode}\SaveRestoreHook
\column{B}{@{}>{\hspre}l<{\hspost}@{}}%
\column{3}{@{}>{\hspre}l<{\hspost}@{}}%
\column{6}{@{}>{\hspre}l<{\hspost}@{}}%
\column{E}{@{}>{\hspre}l<{\hspost}@{}}%
\>[6]{}\Varid{h_{Modify1}}\;(\Conid{Op}\;(\Conid{Inl}\;(\Conid{MRestore}\;\Varid{r}\;\Varid{p}))>\!\!>\!\!=\Varid{k})\;\Varid{s}{}\<[E]%
\\
\>[3]{}\mathrel{=}\mbox{\commentbegin ~  definition of \ensuremath{(>\!\!>\!\!=)} for free monad   \commentend}{}\<[E]%
\\
\>[3]{}\hsindent{3}{}\<[6]%
\>[6]{}\Varid{h_{Modify1}}\;(\Conid{Op}\;(\Varid{fmap}\;(>\!\!>\!\!=\Varid{k})\;(\Conid{Inl}\;(\Conid{MRestore}\;\Varid{r}\;\Varid{p}))))\;\Varid{s}{}\<[E]%
\\
\>[3]{}\mathrel{=}\mbox{\commentbegin ~  definition of \ensuremath{\Varid{fmap}} for coproduct \ensuremath{(\mathrel{{:}{+}{:}})}   \commentend}{}\<[E]%
\\
\>[3]{}\hsindent{3}{}\<[6]%
\>[6]{}\Varid{h_{Modify1}}\;(\Conid{Op}\;(\Conid{Inl}\;(\Varid{fmap}\;(>\!\!>\!\!=\Varid{k})\;(\Conid{MRestore}\;\Varid{r}\;\Varid{p}))))\;\Varid{s}{}\<[E]%
\\
\>[3]{}\mathrel{=}\mbox{\commentbegin ~  definition of \ensuremath{\Varid{fmap}} for \ensuremath{\Conid{MRestore}}   \commentend}{}\<[E]%
\\
\>[3]{}\hsindent{3}{}\<[6]%
\>[6]{}\Varid{h_{Modify1}}\;(\Conid{Op}\;(\Conid{Inl}\;(\Conid{MRestore}\;\Varid{r}\;(\Varid{p}>\!\!>\!\!=\Varid{k}))))\;\Varid{s}{}\<[E]%
\\
\>[3]{}\mathrel{=}\mbox{\commentbegin ~  definition of \ensuremath{\Varid{h_{Modify1}}}   \commentend}{}\<[E]%
\\
\>[3]{}\hsindent{3}{}\<[6]%
\>[6]{}\Varid{h_{Modify1}}\;(\Varid{p}>\!\!>\!\!=\Varid{k})\;(\Varid{s}\mathbin{\ominus}\Varid{r}){}\<[E]%
\\
\>[3]{}\mathrel{=}\mbox{\commentbegin ~  induction hypothesis   \commentend}{}\<[E]%
\\
\>[3]{}\hsindent{3}{}\<[6]%
\>[6]{}\Varid{h_{Modify1}}\;\Varid{p}\;(\Varid{s}\mathbin{\ominus}\Varid{r})>\!\!>\!\!=\lambda (\Varid{x},\Varid{s'})\to \Varid{h_{Modify1}}\;(\Varid{k}\;\Varid{x})\;\Varid{s'}{}\<[E]%
\\
\>[3]{}\mathrel{=}\mbox{\commentbegin ~  definition of \ensuremath{\Varid{h_{Modify1}}}   \commentend}{}\<[E]%
\\
\>[3]{}\hsindent{3}{}\<[6]%
\>[6]{}\Varid{h_{Modify1}}\;(\Conid{Op}\;(\Conid{Inl}\;(\Conid{MRestore}\;\Varid{r}\;\Varid{p})))\;\Varid{s}>\!\!>\!\!=\lambda (\Varid{x},\Varid{s'})\to \Varid{h_{Modify1}}\;(\Varid{k}\;\Varid{x})\;\Varid{s'}{}\<[E]%
\ColumnHook
\end{hscode}\resethooks
\indentend \noindent \mbox{\underline{case \ensuremath{\Varid{p}\mathrel{=}\Conid{Op}\;(\Conid{Inr}\;\Varid{y})}}}
\indentbegin \begin{hscode}\SaveRestoreHook
\column{B}{@{}>{\hspre}l<{\hspost}@{}}%
\column{3}{@{}>{\hspre}l<{\hspost}@{}}%
\column{6}{@{}>{\hspre}l<{\hspost}@{}}%
\column{E}{@{}>{\hspre}l<{\hspost}@{}}%
\>[6]{}\Varid{h_{Modify1}}\;(\Conid{Op}\;(\Conid{Inr}\;\Varid{y})>\!\!>\!\!=\Varid{k})\;\Varid{s}{}\<[E]%
\\
\>[3]{}\mathrel{=}\mbox{\commentbegin ~  definition of \ensuremath{(>\!\!>\!\!=)} for free monad   \commentend}{}\<[E]%
\\
\>[3]{}\hsindent{3}{}\<[6]%
\>[6]{}\Varid{h_{Modify1}}\;(\Conid{Op}\;(\Varid{fmap}\;(>\!\!>\!\!=\Varid{k})\;(\Conid{Inr}\;\Varid{y})))\;\Varid{s}{}\<[E]%
\\
\>[3]{}\mathrel{=}\mbox{\commentbegin ~  definition of \ensuremath{\Varid{fmap}} for coproduct \ensuremath{(\mathrel{{:}{+}{:}})}   \commentend}{}\<[E]%
\\
\>[3]{}\hsindent{3}{}\<[6]%
\>[6]{}\Varid{h_{Modify1}}\;(\Conid{Op}\;(\Conid{Inr}\;(\Varid{fmap}\;(>\!\!>\!\!=\Varid{k})\;\Varid{y})))\;\Varid{s}{}\<[E]%
\\
\>[3]{}\mathrel{=}\mbox{\commentbegin ~  definition of \ensuremath{\Varid{h_{Modify1}}}   \commentend}{}\<[E]%
\\
\>[3]{}\hsindent{3}{}\<[6]%
\>[6]{}\Conid{Op}\;(\Varid{fmap}\;(\lambda \Varid{x}\to \Varid{h_{Modify1}}\;\Varid{x}\;\Varid{s})\;(\Varid{fmap}\;(>\!\!>\!\!=\Varid{k})\;\Varid{y})){}\<[E]%
\\
\>[3]{}\mathrel{=}\mbox{\commentbegin ~  \ensuremath{\Varid{fmap}} fusion   \commentend}{}\<[E]%
\\
\>[3]{}\hsindent{3}{}\<[6]%
\>[6]{}\Conid{Op}\;(\Varid{fmap}\;((\lambda \Varid{x}\to \Varid{h_{Modify1}}\;(\Varid{x}>\!\!>\!\!=\Varid{k})\;\Varid{s}))\;\Varid{y}){}\<[E]%
\\
\>[3]{}\mathrel{=}\mbox{\commentbegin ~  induction hypothesis   \commentend}{}\<[E]%
\\
\>[3]{}\hsindent{3}{}\<[6]%
\>[6]{}\Conid{Op}\;(\Varid{fmap}\;(\lambda \Varid{x}\to \Varid{h_{Modify1}}\;\Varid{x}\;\Varid{s}>\!\!>\!\!=\lambda (\Varid{x'},\Varid{s'})\to \Varid{h_{Modify1}}\;(\Varid{k}\;\Varid{x'})\;\Varid{s'})\;\Varid{y}){}\<[E]%
\\
\>[3]{}\mathrel{=}\mbox{\commentbegin ~  \ensuremath{\Varid{fmap}} fission  \commentend}{}\<[E]%
\\
\>[3]{}\hsindent{3}{}\<[6]%
\>[6]{}\Conid{Op}\;(\Varid{fmap}\;(\lambda \Varid{x}\to \Varid{x}>\!\!>\!\!=\lambda (\Varid{x'},\Varid{s'})\to \Varid{h_{Modify1}}\;(\Varid{k}\;\Varid{x'})\;\Varid{s'})\;(\Varid{fmap}\;(\lambda \Varid{x}\to \Varid{h_{Modify1}}\;\Varid{x}\;\Varid{s})\;\Varid{y})){}\<[E]%
\\
\>[3]{}\mathrel{=}\mbox{\commentbegin ~  definition of \ensuremath{(>\!\!>\!\!=)}  \commentend}{}\<[E]%
\\
\>[3]{}\hsindent{3}{}\<[6]%
\>[6]{}\Conid{Op}\;((\Varid{fmap}\;(\lambda \Varid{x}\to \Varid{h_{Modify1}}\;\Varid{x}\;\Varid{s})\;\Varid{y}))>\!\!>\!\!=\lambda (\Varid{x'},\Varid{s'})\to \Varid{h_{Modify1}}\;(\Varid{k}\;\Varid{x'})\;\Varid{s'}{}\<[E]%
\\
\>[3]{}\mathrel{=}\mbox{\commentbegin ~  definition of \ensuremath{\Varid{h_{Modify1}}}   \commentend}{}\<[E]%
\\
\>[3]{}\hsindent{3}{}\<[6]%
\>[6]{}\Conid{Op}\;(\Conid{Inr}\;\Varid{y})\;\Varid{s}>\!\!>\!\!=\lambda (\Varid{x'},\Varid{s'})\to \Varid{h_{Modify1}}\;(\Varid{k}\;\Varid{x'})\;\Varid{s'}{}\<[E]%
\ColumnHook
\end{hscode}\resethooks
\indentend \end{proof}

\section{Proofs for Modelling Local State with Trail Stack}
\label{app:immutable-trail-stack}

In this section, we prove the following theorem in
\Cref{sec:trail-stack}.

\localTrail*

The proof follows a similar structure to those in
\Cref{app:local-global} and \Cref{app:modify-local-global}.

As in \Cref{app:modify-local-global}, we fuse \ensuremath{\Varid{run_{StateT}}\hsdot{\circ }{.}\Varid{h_{Modify}}}
into \ensuremath{\Varid{h_{Modify1}}} and use it instead in the following proofs.

\subsection{Main Proof Structure}
The main proof structure of \Cref{thm:trail-local-global} is as follows.
\begin{proof}
The left-hand side is expanded to
\indentbegin \begin{hscode}\SaveRestoreHook
\column{B}{@{}>{\hspre}l<{\hspost}@{}}%
\column{3}{@{}>{\hspre}l<{\hspost}@{}}%
\column{E}{@{}>{\hspre}l<{\hspost}@{}}%
\>[3]{}\Varid{h_{GlobalT}}\mathrel{=}\Varid{fmap}\;(\Varid{fmap}\;\Varid{fst}\hsdot{\circ }{.}\Varid{flip}\;\Varid{run_{StateT}}\;(\Conid{Stack}\;[\mskip1.5mu \mskip1.5mu])\hsdot{\circ }{.}\Varid{h_{State}})\hsdot{\circ }{.}\Varid{h_{GlobalM}}\hsdot{\circ }{.}\Varid{local2trail}{}\<[E]%
\ColumnHook
\end{hscode}\resethooks
\indentend Both the left-hand side and the right-hand side of the equation
consist of function compositions involving one or more folds.
We apply fold fusion separately on both sides to contract each
into a single fold:
\begin{eqnarray*}
\ensuremath{\Varid{h_{GlobalT}}} & = & \ensuremath{\Varid{fold}\;\Varid{gen}_{\Varid{LHS}}\;(\Varid{alg}_{\Varid{LHS}}^{\Varid{S}}\mathbin{\#}\Varid{alg}_{\Varid{RHS}}^{\Varid{ND}}\mathbin{\#}\Varid{fwd}_{\Varid{LHS}})} \\
\ensuremath{\Varid{h_{LocalM}}}& = & \ensuremath{\Varid{fold}\;\Varid{gen}_{\Varid{RHS}}\;(\Varid{alg}_{\Varid{RHS}}^{\Varid{S}}\mathbin{\#}\Varid{alg}_{\Varid{RHS}}^{\Varid{ND}}\mathbin{\#}\Varid{fwd}_{\Varid{RHS}})}
\end{eqnarray*}
Finally, we show that both folds are equal by showing that their
corresponding parameters are equal:
\begin{eqnarray*}
\ensuremath{\Varid{gen}_{\Varid{LHS}}} & = & \ensuremath{\Varid{gen}_{\Varid{RHS}}} \\
\ensuremath{\Varid{alg}_{\Varid{LHS}}^{\Varid{S}}} & = & \ensuremath{\Varid{alg}_{\Varid{RHS}}^{\Varid{S}}} \\
\ensuremath{\Varid{alg}_{\Varid{LHS}}^{\Varid{ND}}} & = & \ensuremath{\Varid{alg}_{\Varid{RHS}}^{\Varid{ND}}} \\
\ensuremath{\Varid{fwd}_{\Varid{LHS}}} & = & \ensuremath{\Varid{fwd}_{\Varid{RHS}}}
\end{eqnarray*}
We elaborate each of these steps below.
\end{proof}

\subsection{Fusing the Right-Hand Side}

We have already fused \ensuremath{\Varid{h_{LocalM}}} in \Cref{app:modify-fusing-rhs}.
We just show the result here for easy reference.
\indentbegin \begin{hscode}\SaveRestoreHook
\column{B}{@{}>{\hspre}l<{\hspost}@{}}%
\column{3}{@{}>{\hspre}l<{\hspost}@{}}%
\column{5}{@{}>{\hspre}l<{\hspost}@{}}%
\column{24}{@{}>{\hspre}l<{\hspost}@{}}%
\column{29}{@{}>{\hspre}l<{\hspost}@{}}%
\column{E}{@{}>{\hspre}l<{\hspost}@{}}%
\>[B]{}\Varid{h_{LocalM}}\mathrel{=}\Varid{fold}\;\Varid{gen}_{\Varid{RHS}}\;(\Varid{alg}_{\Varid{RHS}}^{\Varid{S}}\mathbin{\#}\Varid{alg}_{\Varid{RHS}}^{\Varid{ND}}\mathbin{\#}\Varid{fwd}_{\Varid{RHS}}){}\<[E]%
\\
\>[B]{}\hsindent{3}{}\<[3]%
\>[3]{}\mathbf{where}{}\<[E]%
\\
\>[3]{}\hsindent{2}{}\<[5]%
\>[5]{}\Varid{gen}_{\Varid{RHS}}\mathbin{::}\Conid{Functor}\;\Varid{f}\Rightarrow \Varid{a}\to (\Varid{s}\to \Conid{Free}\;\Varid{f}\;[\mskip1.5mu \Varid{a}\mskip1.5mu]){}\<[E]%
\\
\>[3]{}\hsindent{2}{}\<[5]%
\>[5]{}\Varid{gen}_{\Varid{RHS}}\;\Varid{x}\mathrel{=}\lambda \Varid{s}\to \Conid{Var}\;[\mskip1.5mu \Varid{x}\mskip1.5mu]{}\<[E]%
\\
\>[3]{}\hsindent{2}{}\<[5]%
\>[5]{}\Varid{alg}_{\Varid{RHS}}^{\Varid{S}}\mathbin{::}\Conid{Undo}\;\Varid{s}\;\Varid{r}\Rightarrow \Varid{State_{F}}\;\Varid{s}\;(\Varid{s}\to \Varid{p})\to (\Varid{s}\to \Varid{p}){}\<[E]%
\\
\>[3]{}\hsindent{2}{}\<[5]%
\>[5]{}\Varid{alg}_{\Varid{RHS}}^{\Varid{S}}\;(\Conid{MGet}\;\Varid{k}){}\<[29]%
\>[29]{}\mathrel{=}\lambda \Varid{s}\to \Varid{k}\;\Varid{s}\;\Varid{s}{}\<[E]%
\\
\>[3]{}\hsindent{2}{}\<[5]%
\>[5]{}\Varid{alg}_{\Varid{RHS}}^{\Varid{S}}\;(\Conid{MUpdate}\;\Varid{r}\;\Varid{k}){}\<[29]%
\>[29]{}\mathrel{=}\lambda \Varid{s}\to \Varid{k}\;(\Varid{s}\mathbin{\oplus}\Varid{r}){}\<[E]%
\\
\>[3]{}\hsindent{2}{}\<[5]%
\>[5]{}\Varid{alg}_{\Varid{RHS}}^{\Varid{S}}\;(\Conid{MRestore}\;\Varid{r}\;\Varid{k}){}\<[29]%
\>[29]{}\mathrel{=}\lambda \Varid{s}\to \Varid{k}\;(\Varid{s}\mathbin{\oplus}\Varid{r}){}\<[E]%
\\
\>[3]{}\hsindent{2}{}\<[5]%
\>[5]{}\Varid{alg}_{\Varid{RHS}}^{\Varid{ND}}\mathbin{::}\Conid{Functor}\;\Varid{f}\Rightarrow \Varid{Nondet_{F}}\;(\Varid{s}\to \Conid{Free}\;\Varid{f}\;[\mskip1.5mu \Varid{a}\mskip1.5mu])\to (\Varid{s}\to \Conid{Free}\;\Varid{f}\;[\mskip1.5mu \Varid{a}\mskip1.5mu]){}\<[E]%
\\
\>[3]{}\hsindent{2}{}\<[5]%
\>[5]{}\Varid{alg}_{\Varid{RHS}}^{\Varid{ND}}\;\Conid{Fail}{}\<[24]%
\>[24]{}\mathrel{=}\lambda \Varid{s}\to \Conid{Var}\;[\mskip1.5mu \mskip1.5mu]{}\<[E]%
\\
\>[3]{}\hsindent{2}{}\<[5]%
\>[5]{}\Varid{alg}_{\Varid{RHS}}^{\Varid{ND}}\;(\Conid{Or}\;\Varid{p}\;\Varid{q}){}\<[24]%
\>[24]{}\mathrel{=}\lambda \Varid{s}\to \Varid{liftM2}\;(+\!\!+)\;(\Varid{p}\;\Varid{s})\;(\Varid{q}\;\Varid{s}){}\<[E]%
\\
\>[3]{}\hsindent{2}{}\<[5]%
\>[5]{}\Varid{fwd}_{\Varid{RHS}}\mathbin{::}\Conid{Functor}\;\Varid{f}\Rightarrow \Varid{f}\;(\Varid{s}\to \Conid{Free}\;\Varid{f}\;[\mskip1.5mu \Varid{a}\mskip1.5mu])\to (\Varid{s}\to \Conid{Free}\;\Varid{f}\;[\mskip1.5mu \Varid{a}\mskip1.5mu]){}\<[E]%
\\
\>[3]{}\hsindent{2}{}\<[5]%
\>[5]{}\Varid{fwd}_{\Varid{RHS}}\;\Varid{op}\mathrel{=}\lambda \Varid{s}\to \Conid{Op}\;(\Varid{fmap}\;(\mathbin{\$}\Varid{s})\;\Varid{op}){}\<[E]%
\ColumnHook
\end{hscode}\resethooks
\indentend 

\subsection{Fusing the Left-Hand Side}
\label{sec:trail-fusing-lhs}

As in \Cref{app:local-global}, we fuse \ensuremath{\Varid{run_{StateT}}\hsdot{\circ }{.}\Varid{h_{State}}}
into \ensuremath{\Varid{h_{State1}}}.
For brevity, we define\indentbegin \begin{hscode}\SaveRestoreHook
\column{B}{@{}>{\hspre}l<{\hspost}@{}}%
\column{3}{@{}>{\hspre}l<{\hspost}@{}}%
\column{E}{@{}>{\hspre}l<{\hspost}@{}}%
\>[3]{}\Varid{runStack}\mathrel{=}\Varid{fmap}\;\Varid{fst}\hsdot{\circ }{.}\Varid{flip}\;\Varid{h_{State1}}\;(\Conid{Stack}\;[\mskip1.5mu \mskip1.5mu]){}\<[E]%
\ColumnHook
\end{hscode}\resethooks
\indentend %
The left-hand side is simplified to\indentbegin \begin{hscode}\SaveRestoreHook
\column{B}{@{}>{\hspre}l<{\hspost}@{}}%
\column{3}{@{}>{\hspre}l<{\hspost}@{}}%
\column{E}{@{}>{\hspre}l<{\hspost}@{}}%
\>[3]{}\Varid{fmap}\;\Varid{runStack}\hsdot{\circ }{.}\Varid{h_{GlobalM}}\hsdot{\circ }{.}\Varid{local2trail}{}\<[E]%
\ColumnHook
\end{hscode}\resethooks
\indentend We calculate as follows:
\indentbegin \begin{hscode}\SaveRestoreHook
\column{B}{@{}>{\hspre}l<{\hspost}@{}}%
\column{5}{@{}>{\hspre}l<{\hspost}@{}}%
\column{7}{@{}>{\hspre}l<{\hspost}@{}}%
\column{9}{@{}>{\hspre}l<{\hspost}@{}}%
\column{25}{@{}>{\hspre}l<{\hspost}@{}}%
\column{27}{@{}>{\hspre}l<{\hspost}@{}}%
\column{29}{@{}>{\hspre}l<{\hspost}@{}}%
\column{E}{@{}>{\hspre}l<{\hspost}@{}}%
\>[5]{}\Varid{fmap}\;\Varid{runStack}\hsdot{\circ }{.}\Varid{h_{GlobalM}}\hsdot{\circ }{.}\Varid{local2trail}{}\<[E]%
\\
\>[B]{}\mathrel{=}\mbox{\commentbegin ~  definition of \ensuremath{\Varid{local2trail}}  \commentend}{}\<[E]%
\\
\>[B]{}\hsindent{5}{}\<[5]%
\>[5]{}\Varid{fmap}\;\Varid{runStack}\hsdot{\circ }{.}\Varid{h_{GlobalM}}\hsdot{\circ }{.}\Varid{fold}\;\Conid{Var}\;(\Varid{alg}_{1}\mathbin{\#}\Varid{alg}_{2}\mathbin{\#}\Varid{fwd}){}\<[E]%
\\
\>[5]{}\hsindent{2}{}\<[7]%
\>[7]{}\mathbf{where}{}\<[E]%
\\
\>[7]{}\hsindent{2}{}\<[9]%
\>[9]{}\Varid{alg}_{1}\;(\Conid{MUpdate}\;\Varid{r}\;\Varid{k}){}\<[29]%
\>[29]{}\mathrel{=}\Varid{pushStack}\;(\Conid{Left}\;\Varid{r})>\!\!>\Varid{update}\;\Varid{r}>\!\!>\Varid{k}{}\<[E]%
\\
\>[7]{}\hsindent{2}{}\<[9]%
\>[9]{}\Varid{alg}_{1}\;\Varid{p}{}\<[29]%
\>[29]{}\mathrel{=}\Conid{Op}\hsdot{\circ }{.}\Conid{Inl}\mathbin{\$}\Varid{p}{}\<[E]%
\\
\>[7]{}\hsindent{2}{}\<[9]%
\>[9]{}\Varid{alg}_{2}\;(\Conid{Or}\;\Varid{p}\;\Varid{q}){}\<[29]%
\>[29]{}\mathrel{=}(\Varid{pushStack}\;(\Conid{Right}\;())>\!\!>\Varid{p})\mathbin{\talloblong}(\Varid{untrail}>\!\!>\Varid{q}){}\<[E]%
\\
\>[7]{}\hsindent{2}{}\<[9]%
\>[9]{}\Varid{alg}_{2}\;\Varid{p}{}\<[29]%
\>[29]{}\mathrel{=}\Conid{Op}\hsdot{\circ }{.}\Conid{Inr}\hsdot{\circ }{.}\Conid{Inl}\mathbin{\$}\Varid{p}{}\<[E]%
\\
\>[7]{}\hsindent{2}{}\<[9]%
\>[9]{}\Varid{fwd}\;\Varid{p}{}\<[29]%
\>[29]{}\mathrel{=}\Conid{Op}\hsdot{\circ }{.}\Conid{Inr}\hsdot{\circ }{.}\Conid{Inr}\hsdot{\circ }{.}\Conid{Inr}\mathbin{\$}\Varid{p}{}\<[E]%
\\
\>[7]{}\hsindent{2}{}\<[9]%
\>[9]{}\Varid{untrail}\mathrel{=}\mathbf{do}\;{}\<[25]%
\>[25]{}\Varid{top}\leftarrow \Varid{popStack};{}\<[E]%
\\
\>[25]{}\mathbf{case}\;\Varid{top}\;\mathbf{of}{}\<[E]%
\\
\>[25]{}\hsindent{2}{}\<[27]%
\>[27]{}\Conid{Nothing}\to \Varid{\eta}\;(){}\<[E]%
\\
\>[25]{}\hsindent{2}{}\<[27]%
\>[27]{}\Conid{Just}\;(\Conid{Right}\;())\to \Varid{\eta}\;(){}\<[E]%
\\
\>[25]{}\hsindent{2}{}\<[27]%
\>[27]{}\Conid{Just}\;(\Conid{Left}\;\Varid{r})\to \mathbf{do}\;\Varid{restore}\;\Varid{r};\Varid{untrail}{}\<[E]%
\\
\>[B]{}\mathrel{=}\mbox{\commentbegin ~  fold fusion-post' (Equation \ref{eq:fusion-post-strong})   \commentend}{}\<[E]%
\\
\>[B]{}\hsindent{5}{}\<[5]%
\>[5]{}\Varid{fold}\;\Varid{gen}_{\Varid{LHS}}\;(\Varid{alg}_{\Varid{LHS}}^{\Varid{S}}\mathbin{\#}\Varid{alg}_{\Varid{LHS}}^{\Varid{ND}}\mathbin{\#}\Varid{fwd}_{\Varid{LHS}}){}\<[E]%
\ColumnHook
\end{hscode}\resethooks
\indentend 
This last step is valid provided that the fusion conditions are satisfied:
\[\ba{rclr}
\ensuremath{\Varid{fmap}\;\Varid{runStack}\hsdot{\circ }{.}\Varid{h_{GlobalM}}\hsdot{\circ }{.}\Conid{Var}} & = & \ensuremath{\Varid{gen}_{\Varid{LHS}}} &\refa{}\\
\ea\]
\vspace{-\baselineskip}
\[\ba{rlr}
   &\ensuremath{\Varid{fmap}\;\Varid{runStack}\hsdot{\circ }{.}\Varid{h_{GlobalM}}\hsdot{\circ }{.}(\Varid{alg}_{1}\mathbin{\#}\Varid{alg}_{2}\mathbin{\#}\Varid{fwd})\hsdot{\circ }{.}\Varid{fmap}\;\Varid{local2trail}} \\
= & \ensuremath{(\Varid{alg}_{\Varid{LHS}}^{\Varid{S}}\mathbin{\#}\Varid{alg}_{\Varid{LHS}}^{\Varid{ND}}\mathbin{\#}\Varid{fwd}_{\Varid{LHS}})\hsdot{\circ }{.}\Varid{fmap}\;(\Varid{fmap}\;\Varid{runStack}\hsdot{\circ }{.}\Varid{h_{GlobalM}})\hsdot{\circ }{.}\Varid{fmap}\;\Varid{local2trail}} &\refb{}
\ea\]

The first subcondition \refa{} is met by\indentbegin \begin{hscode}\SaveRestoreHook
\column{B}{@{}>{\hspre}l<{\hspost}@{}}%
\column{3}{@{}>{\hspre}l<{\hspost}@{}}%
\column{E}{@{}>{\hspre}l<{\hspost}@{}}%
\>[3]{}\Varid{gen}_{\Varid{LHS}}\mathbin{::}\Conid{Functor}\;\Varid{f}\Rightarrow \Varid{a}\to (\Varid{s}\to \Conid{Free}\;\Varid{f}\;[\mskip1.5mu \Varid{a}\mskip1.5mu]){}\<[E]%
\\
\>[3]{}\Varid{gen}_{\Varid{LHS}}\;\Varid{x}\mathrel{=}\lambda \Varid{s}\to \Conid{Var}\;[\mskip1.5mu \Varid{x}\mskip1.5mu]{}\<[E]%
\ColumnHook
\end{hscode}\resethooks
\indentend as established in the following calculation:\indentbegin \begin{hscode}\SaveRestoreHook
\column{B}{@{}>{\hspre}l<{\hspost}@{}}%
\column{3}{@{}>{\hspre}l<{\hspost}@{}}%
\column{5}{@{}>{\hspre}l<{\hspost}@{}}%
\column{E}{@{}>{\hspre}l<{\hspost}@{}}%
\>[5]{}\Varid{fmap}\;\Varid{runStack}\mathbin{\$}\Varid{h_{GlobalM}}\;(\Conid{Var}\;\Varid{x}){}\<[E]%
\\
\>[3]{}\mathrel{=}\mbox{\commentbegin ~ definition of \ensuremath{\Varid{h_{GlobalM}}}  \commentend}{}\<[E]%
\\
\>[3]{}\hsindent{2}{}\<[5]%
\>[5]{}\Varid{fmap}\;\Varid{runStack}\mathbin{\$}\Varid{fmap}\;(\Varid{fmap}\;\Varid{fst})\;(\Varid{h_{Modify1}}\;(\Varid{h_{ND+f}}\;(\Varid{(\Leftrightarrow)}\;(\Conid{Var}\;\Varid{x})))){}\<[E]%
\\
\>[3]{}\mathrel{=}\mbox{\commentbegin ~ definition of \ensuremath{\Varid{(\Leftrightarrow)}}  \commentend}{}\<[E]%
\\
\>[3]{}\hsindent{2}{}\<[5]%
\>[5]{}\Varid{fmap}\;\Varid{runStack}\mathbin{\$}\Varid{fmap}\;(\Varid{fmap}\;\Varid{fst})\;(\Varid{h_{Modify1}}\;(\Varid{h_{ND+f}}\;(\Conid{Var}\;\Varid{x}))){}\<[E]%
\\
\>[3]{}\mathrel{=}\mbox{\commentbegin ~ definition of \ensuremath{\Varid{h_{ND+f}}}  \commentend}{}\<[E]%
\\
\>[3]{}\hsindent{2}{}\<[5]%
\>[5]{}\Varid{fmap}\;\Varid{runStack}\mathbin{\$}\Varid{fmap}\;(\Varid{fmap}\;\Varid{fst})\;(\Varid{h_{Modify1}}\;(\Conid{Var}\;[\mskip1.5mu \Varid{x}\mskip1.5mu])){}\<[E]%
\\
\>[3]{}\mathrel{=}\mbox{\commentbegin ~ definition of \ensuremath{\Varid{h_{Modify1}}}  \commentend}{}\<[E]%
\\
\>[3]{}\hsindent{2}{}\<[5]%
\>[5]{}\Varid{fmap}\;\Varid{runStack}\mathbin{\$}\Varid{fmap}\;(\Varid{fmap}\;\Varid{fst})\;(\lambda \Varid{s}\to \Conid{Var}\;([\mskip1.5mu \Varid{x}\mskip1.5mu],\Varid{s})){}\<[E]%
\\
\>[3]{}\mathrel{=}\mbox{\commentbegin ~ definition of \ensuremath{\Varid{fmap}} (twice)  \commentend}{}\<[E]%
\\
\>[3]{}\hsindent{2}{}\<[5]%
\>[5]{}\Varid{fmap}\;\Varid{runStack}\mathbin{\$}\lambda \Varid{s}\to \Conid{Var}\;[\mskip1.5mu \Varid{x}\mskip1.5mu]{}\<[E]%
\\
\>[3]{}\mathrel{=}\mbox{\commentbegin ~ definition of \ensuremath{\Varid{fmap}}  \commentend}{}\<[E]%
\\
\>[3]{}\hsindent{2}{}\<[5]%
\>[5]{}\lambda \Varid{s}\to \Varid{runStack}\mathbin{\$}\Conid{Var}\;[\mskip1.5mu \Varid{x}\mskip1.5mu]{}\<[E]%
\\
\>[3]{}\mathrel{=}\mbox{\commentbegin ~ definition of \ensuremath{\Varid{runStack}}  \commentend}{}\<[E]%
\\
\>[3]{}\hsindent{2}{}\<[5]%
\>[5]{}\lambda \Varid{s}\to \Varid{fmap}\;\Varid{fst}\hsdot{\circ }{.}\Varid{flip}\;\Varid{h_{State1}}\;(\Conid{Stack}\;[\mskip1.5mu \mskip1.5mu])\mathbin{\$}\Conid{Var}\;[\mskip1.5mu \Varid{x}\mskip1.5mu]{}\<[E]%
\\
\>[3]{}\mathrel{=}\mbox{\commentbegin ~ definition of \ensuremath{\Varid{h_{State1}}}  \commentend}{}\<[E]%
\\
\>[3]{}\hsindent{2}{}\<[5]%
\>[5]{}\lambda \Varid{s}\to \Varid{fmap}\;\Varid{fst}\mathbin{\$}(\lambda \Varid{s}\to \Conid{Var}\;([\mskip1.5mu \Varid{x}\mskip1.5mu],\Varid{s}))\;(\Conid{Stack}\;[\mskip1.5mu \mskip1.5mu]){}\<[E]%
\\
\>[3]{}\mathrel{=}\mbox{\commentbegin ~ function application  \commentend}{}\<[E]%
\\
\>[3]{}\hsindent{2}{}\<[5]%
\>[5]{}\lambda \Varid{s}\to \Varid{fmap}\;\Varid{fst}\;(\Conid{Var}\;([\mskip1.5mu \Varid{x}\mskip1.5mu],\Conid{Stack}\;[\mskip1.5mu \mskip1.5mu])){}\<[E]%
\\
\>[3]{}\mathrel{=}\mbox{\commentbegin ~ definition of \ensuremath{\Varid{fmap}}  \commentend}{}\<[E]%
\\
\>[3]{}\hsindent{2}{}\<[5]%
\>[5]{}\lambda \Varid{s}\to \Conid{Var}\;[\mskip1.5mu \Varid{x}\mskip1.5mu]{}\<[E]%
\\
\>[3]{}\mathrel{=}\mbox{\commentbegin ~ definition of \ensuremath{\Varid{gen}_{\Varid{LHS}}}  \commentend}{}\<[E]%
\\
\>[3]{}\hsindent{2}{}\<[5]%
\>[5]{}\Varid{gen}_{\Varid{LHS}}\;\Varid{x}{}\<[E]%
\ColumnHook
\end{hscode}\resethooks
\indentend 
We can split the second fusion condition \refb{} in three
subconditions:
\[\ba{rlr}
  & \ensuremath{\Varid{fmap}\;\Varid{runStack}\hsdot{\circ }{.}\Varid{h_{GlobalM}}\hsdot{\circ }{.}\Varid{alg}_{1}\hsdot{\circ }{.}\Varid{fmap}\;\Varid{local2trail}} \\
= & \ensuremath{\Varid{alg}_{\Varid{LHS}}^{\Varid{S}}\hsdot{\circ }{.}\Varid{fmap}\;(\Varid{fmap}\;\Varid{runStack}\hsdot{\circ }{.}\Varid{h_{GlobalM}})\hsdot{\circ }{.}\Varid{fmap}\;\Varid{local2trail}} &\refc{}\\
  & \ensuremath{\Varid{fmap}\;\Varid{runStack}\hsdot{\circ }{.}\Varid{h_{GlobalM}}\hsdot{\circ }{.}\Varid{h_{GlobalM}}\hsdot{\circ }{.}\Varid{alg}_{2}\hsdot{\circ }{.}\Varid{fmap}\;\Varid{local2trail}} \\
= & \ensuremath{\Varid{alg}_{\Varid{LHS}}^{\Varid{ND}}\hsdot{\circ }{.}\Varid{fmap}\;(\Varid{fmap}\;\Varid{runStack}\hsdot{\circ }{.}\Varid{h_{GlobalM}})\hsdot{\circ }{.}\Varid{fmap}\;\Varid{local2trail}} &\refd{}\\
  & \ensuremath{\Varid{fmap}\;\Varid{runStack}\hsdot{\circ }{.}\Varid{h_{GlobalM}}\hsdot{\circ }{.}\Varid{h_{GlobalM}}\hsdot{\circ }{.}\Varid{fwd}\hsdot{\circ }{.}\Varid{fmap}\;\Varid{local2trail}} \\
= & \ensuremath{\Varid{fwd}_{\Varid{LHS}}\hsdot{\circ }{.}\Varid{fmap}\;(\Varid{fmap}\;\Varid{runStack}\hsdot{\circ }{.}\Varid{h_{GlobalM}})\hsdot{\circ }{.}\Varid{fmap}\;\Varid{local2trail}} &\refe{}
\ea\]

For brevity, we omit the last common part \ensuremath{\Varid{fmap}\;\Varid{local2global_M}} of
these equations. Instead, we assume that the input is in the codomain
of \ensuremath{\Varid{fmap}\;\Varid{local2global_M}}.

For the first subcondition \refc{}, we define \ensuremath{\Varid{alg}_{\Varid{LHS}}^{\Varid{S}}} as follows.\indentbegin \begin{hscode}\SaveRestoreHook
\column{B}{@{}>{\hspre}l<{\hspost}@{}}%
\column{3}{@{}>{\hspre}l<{\hspost}@{}}%
\column{27}{@{}>{\hspre}l<{\hspost}@{}}%
\column{E}{@{}>{\hspre}l<{\hspost}@{}}%
\>[3]{}\Varid{alg}_{\Varid{LHS}}^{\Varid{S}}\mathbin{::}(\Conid{Functor}\;\Varid{f},\Conid{Undo}\;\Varid{s}\;\Varid{r})\Rightarrow \Varid{Modify_{F}}\;\Varid{s}\;\Varid{r}\;(\Varid{s}\to \Conid{Free}\;\Varid{f}\;[\mskip1.5mu \Varid{a}\mskip1.5mu])\to (\Varid{s}\to \Conid{Free}\;\Varid{f}\;[\mskip1.5mu \Varid{a}\mskip1.5mu]){}\<[E]%
\\
\>[3]{}\Varid{alg}_{\Varid{LHS}}^{\Varid{S}}\;(\Conid{MGet}\;\Varid{k}){}\<[27]%
\>[27]{}\mathrel{=}\lambda \Varid{s}\to \Varid{k}\;\Varid{s}\;\Varid{s}{}\<[E]%
\\
\>[3]{}\Varid{alg}_{\Varid{LHS}}^{\Varid{S}}\;(\Conid{MUpdate}\;\Varid{r}\;\Varid{k}){}\<[27]%
\>[27]{}\mathrel{=}\lambda \Varid{s}\to \Varid{k}\;(\Varid{s}\mathbin{\oplus}\Varid{r}){}\<[E]%
\\
\>[3]{}\Varid{alg}_{\Varid{LHS}}^{\Varid{S}}\;(\Conid{MRestore}\;\Varid{r}\;\Varid{k}){}\<[27]%
\>[27]{}\mathrel{=}\lambda \Varid{s}\to \Varid{k}\;(\Varid{s}\mathbin{\ominus}\Varid{r}){}\<[E]%
\ColumnHook
\end{hscode}\resethooks
\indentend We prove it by a case analysis on the shape of input \ensuremath{\Varid{op}\mathbin{::}\Varid{Modify_{F}}\;\Varid{s}\;\Varid{r}\;(\Conid{Free}\;(\Varid{Modify_{F}}\;\Varid{s}\;\Varid{r}\mathrel{{:}{+}{:}}\Varid{Nondet_{F}}\mathrel{{:}{+}{:}}\Varid{f})\;\Varid{a})}.
We use the condition in \Cref{thm:modify-local-global} that the input
program does not use the \ensuremath{\Varid{restore}} operation.
We only need to consider the case that \ensuremath{\Varid{op}} is of form \ensuremath{\Conid{MGet}\;\Varid{k}} or
\ensuremath{\Conid{MUpdate}\;\Varid{r}\;\Varid{k}} where \ensuremath{\Varid{restore}} is also not used in the continuation
\ensuremath{\Varid{k}}.

\vspace{0.5\lineskip}

\noindent \mbox{\underline{case \ensuremath{\Varid{op}\mathrel{=}\Conid{MGet}\;\Varid{k}}}}
In the corresponding case of \Cref{app:modify-fusing-lhs}, we have
calculated that \ensuremath{\Varid{h_{GlobalM}}\;(\Conid{Op}\;(\Conid{Inl}\;(\Conid{MGet}\;\Varid{k})))\mathrel{=}\lambda \Varid{s}\to (\Varid{h_{GlobalM}}\hsdot{\circ }{.}\Varid{k})\;\Varid{s}\;\Varid{s}} \refs{}.
\indentbegin \begin{hscode}\SaveRestoreHook
\column{B}{@{}>{\hspre}l<{\hspost}@{}}%
\column{3}{@{}>{\hspre}l<{\hspost}@{}}%
\column{5}{@{}>{\hspre}l<{\hspost}@{}}%
\column{E}{@{}>{\hspre}l<{\hspost}@{}}%
\>[5]{}\Varid{fmap}\;\Varid{runStack}\mathbin{\$}\Varid{h_{GlobalM}}\;(\Varid{alg}_{1}\;(\Conid{MGet}\;\Varid{k})){}\<[E]%
\\
\>[3]{}\mathrel{=}\mbox{\commentbegin ~ definition of \ensuremath{\Varid{alg}_{1}}  \commentend}{}\<[E]%
\\
\>[3]{}\hsindent{2}{}\<[5]%
\>[5]{}\Varid{fmap}\;\Varid{runStack}\mathbin{\$}\Varid{h_{GlobalM}}\;(\Conid{Op}\;(\Conid{Inl}\;(\Conid{MGet}\;\Varid{k}))){}\<[E]%
\\
\>[3]{}\mathrel{=}\mbox{\commentbegin ~ Equation \refs{}  \commentend}{}\<[E]%
\\
\>[3]{}\hsindent{2}{}\<[5]%
\>[5]{}\Varid{fmap}\;\Varid{runStack}\mathbin{\$}\lambda \Varid{s}\to (\Varid{h_{GlobalM}}\hsdot{\circ }{.}\Varid{k})\;\Varid{s}\;\Varid{s}{}\<[E]%
\\
\>[3]{}\mathrel{=}\mbox{\commentbegin ~ definition of \ensuremath{\Varid{fmap}}  \commentend}{}\<[E]%
\\
\>[3]{}\hsindent{2}{}\<[5]%
\>[5]{}\lambda \Varid{s}\to \Varid{runStack}\mathbin{\$}(\Varid{h_{GlobalM}}\hsdot{\circ }{.}\Varid{k})\;\Varid{s}\;\Varid{s}{}\<[E]%
\\
\>[3]{}\mathrel{=}\mbox{\commentbegin ~ definition of \ensuremath{\Varid{fmap}}  \commentend}{}\<[E]%
\\
\>[3]{}\hsindent{2}{}\<[5]%
\>[5]{}\lambda \Varid{s}\to (\Varid{fmap}\;\Varid{runStack}\hsdot{\circ }{.}\Varid{h_{GlobalM}}\hsdot{\circ }{.}\Varid{k})\;\Varid{s}\;\Varid{s}{}\<[E]%
\\
\>[3]{}\mathrel{=}\mbox{\commentbegin ~ definition of \ensuremath{\Varid{alg}_{\Varid{LHS}}^{\Varid{S}}}  \commentend}{}\<[E]%
\\
\>[3]{}\hsindent{2}{}\<[5]%
\>[5]{}\Varid{alg}_{\Varid{LHS}}^{\Varid{S}}\;(\Conid{MGet}\;(\Varid{fmap}\;\Varid{runStack}\hsdot{\circ }{.}\Varid{h_{GlobalM}}\hsdot{\circ }{.}\Varid{k})){}\<[E]%
\\
\>[3]{}\mathrel{=}\mbox{\commentbegin ~ definition of \ensuremath{\Varid{fmap}}  \commentend}{}\<[E]%
\\
\>[3]{}\hsindent{2}{}\<[5]%
\>[5]{}\Varid{alg}_{\Varid{LHS}}^{\Varid{S}}\;(\Varid{fmap}\;(\Varid{fmap}\;\Varid{runStack}\hsdot{\circ }{.}\Varid{h_{GlobalM}})\;(\Conid{MGet}\;\Varid{k})){}\<[E]%
\ColumnHook
\end{hscode}\resethooks
\indentend \noindent \mbox{\underline{case \ensuremath{\Varid{op}\mathrel{=}\Conid{MUpdate}\;\Varid{r}\;\Varid{k}}}}
From \ensuremath{\Varid{op}} is in the codomain of \ensuremath{\Varid{fmap}\;\Varid{local2global_M}} we obtain \ensuremath{\Varid{k}} is
in the codomain of \ensuremath{\Varid{local2global_M}}.
\indentbegin \begin{hscode}\SaveRestoreHook
\column{B}{@{}>{\hspre}l<{\hspost}@{}}%
\column{3}{@{}>{\hspre}l<{\hspost}@{}}%
\column{5}{@{}>{\hspre}l<{\hspost}@{}}%
\column{7}{@{}>{\hspre}l<{\hspost}@{}}%
\column{9}{@{}>{\hspre}l<{\hspost}@{}}%
\column{11}{@{}>{\hspre}l<{\hspost}@{}}%
\column{E}{@{}>{\hspre}l<{\hspost}@{}}%
\>[5]{}\Varid{fmap}\;\Varid{runStack}\hsdot{\circ }{.}\Varid{h_{GlobalM}}\mathbin{\$}\Varid{alg}_{1}\;(\Conid{MUpdate}\;\Varid{r}\;\Varid{k}){}\<[E]%
\\
\>[3]{}\mathrel{=}\mbox{\commentbegin ~ definition of \ensuremath{\Varid{alg}_{1}}  \commentend}{}\<[E]%
\\
\>[3]{}\hsindent{2}{}\<[5]%
\>[5]{}\Varid{fmap}\;\Varid{runStack}\hsdot{\circ }{.}\Varid{h_{GlobalM}}\mathbin{\$}\Varid{pushStack}\;(\Conid{Left}\;\Varid{r})>\!\!>\Varid{update}\;\Varid{r}>\!\!>\Varid{k}{}\<[E]%
\\
\>[3]{}\mathrel{=}\mbox{\commentbegin ~ definition of \ensuremath{\Varid{pushStack}}  \commentend}{}\<[E]%
\\
\>[3]{}\hsindent{2}{}\<[5]%
\>[5]{}\Varid{fmap}\;\Varid{runStack}\hsdot{\circ }{.}\Varid{h_{GlobalM}}\mathbin{\$}\mathbf{do}{}\<[E]%
\\
\>[5]{}\hsindent{2}{}\<[7]%
\>[7]{}\Conid{Stack}\;\Varid{xs}\leftarrow \Varid{get}{}\<[E]%
\\
\>[5]{}\hsindent{2}{}\<[7]%
\>[7]{}\Varid{put}\;(\Conid{Stack}\;(\Conid{Left}\;\Varid{r}\mathbin{:}\Varid{xs})){}\<[E]%
\\
\>[5]{}\hsindent{2}{}\<[7]%
\>[7]{}\Varid{update}\;\Varid{r}>\!\!>\Varid{k}{}\<[E]%
\\
\>[3]{}\mathrel{=}\mbox{\commentbegin ~ definition of \ensuremath{\Varid{get}}, \ensuremath{\Varid{put}}, and \ensuremath{\Varid{update}}  \commentend}{}\<[E]%
\\
\>[3]{}\hsindent{2}{}\<[5]%
\>[5]{}\Varid{fmap}\;\Varid{runStack}\hsdot{\circ }{.}\Varid{h_{GlobalM}}\mathbin{\$}{}\<[E]%
\\
\>[5]{}\hsindent{2}{}\<[7]%
\>[7]{}\Conid{Op}\hsdot{\circ }{.}\Conid{Inr}\hsdot{\circ }{.}\Conid{Inr}\hsdot{\circ }{.}\Conid{Inl}\mathbin{\$}\Conid{Get}\;(\lambda (\Conid{Stack}\;\Varid{xs})\to {}\<[E]%
\\
\>[7]{}\hsindent{2}{}\<[9]%
\>[9]{}\Conid{Op}\hsdot{\circ }{.}\Conid{Inr}\hsdot{\circ }{.}\Conid{Inr}\hsdot{\circ }{.}\Conid{Inl}\mathbin{\$}\Conid{Put}\;(\Conid{Stack}\;(\Conid{Left}\;\Varid{r}\mathbin{:}\Varid{xs}))\;({}\<[E]%
\\
\>[9]{}\hsindent{2}{}\<[11]%
\>[11]{}\Conid{Op}\hsdot{\circ }{.}\Conid{Inl}\mathbin{\$}\Conid{MUpdate}\;\Varid{r}\;\Varid{k})){}\<[E]%
\\
\>[3]{}\mathrel{=}\mbox{\commentbegin ~ definition of \ensuremath{\Varid{h_{GlobalM}}}  \commentend}{}\<[E]%
\\
\>[3]{}\hsindent{2}{}\<[5]%
\>[5]{}\Varid{fmap}\;\Varid{runStack}\hsdot{\circ }{.}\Varid{fmap}\;(\Varid{fmap}\;\Varid{fst})\hsdot{\circ }{.}\Varid{h_{Modify1}}\hsdot{\circ }{.}\Varid{h_{ND+f}}\hsdot{\circ }{.}\Varid{(\Leftrightarrow)}\mathbin{\$}{}\<[E]%
\\
\>[5]{}\hsindent{2}{}\<[7]%
\>[7]{}\Conid{Op}\hsdot{\circ }{.}\Conid{Inr}\hsdot{\circ }{.}\Conid{Inr}\hsdot{\circ }{.}\Conid{Inl}\mathbin{\$}\Conid{Get}\;(\lambda (\Conid{Stack}\;\Varid{xs})\to {}\<[E]%
\\
\>[7]{}\hsindent{2}{}\<[9]%
\>[9]{}\Conid{Op}\hsdot{\circ }{.}\Conid{Inr}\hsdot{\circ }{.}\Conid{Inr}\hsdot{\circ }{.}\Conid{Inl}\mathbin{\$}\Conid{Put}\;(\Conid{Stack}\;(\Conid{Left}\;\Varid{r}\mathbin{:}\Varid{xs}))\;({}\<[E]%
\\
\>[9]{}\hsindent{2}{}\<[11]%
\>[11]{}\Conid{Op}\hsdot{\circ }{.}\Conid{Inl}\mathbin{\$}\Conid{MUpdate}\;\Varid{r}\;\Varid{k})){}\<[E]%
\\
\>[3]{}\mathrel{=}\mbox{\commentbegin ~ definition of \ensuremath{\Varid{(\Leftrightarrow)}}  \commentend}{}\<[E]%
\\
\>[3]{}\hsindent{2}{}\<[5]%
\>[5]{}\Varid{fmap}\;\Varid{runStack}\hsdot{\circ }{.}\Varid{fmap}\;(\Varid{fmap}\;\Varid{fst})\hsdot{\circ }{.}\Varid{h_{Modify1}}\hsdot{\circ }{.}\Varid{h_{ND+f}}\mathbin{\$}{}\<[E]%
\\
\>[5]{}\hsindent{2}{}\<[7]%
\>[7]{}\Conid{Op}\hsdot{\circ }{.}\Conid{Inr}\hsdot{\circ }{.}\Conid{Inr}\hsdot{\circ }{.}\Conid{Inl}\mathbin{\$}\Conid{Get}\;(\lambda (\Conid{Stack}\;\Varid{xs})\to {}\<[E]%
\\
\>[7]{}\hsindent{2}{}\<[9]%
\>[9]{}\Conid{Op}\hsdot{\circ }{.}\Conid{Inr}\hsdot{\circ }{.}\Conid{Inr}\hsdot{\circ }{.}\Conid{Inl}\mathbin{\$}\Conid{Put}\;(\Conid{Stack}\;(\Conid{Left}\;\Varid{r}\mathbin{:}\Varid{xs}))\;({}\<[E]%
\\
\>[9]{}\hsindent{2}{}\<[11]%
\>[11]{}\Conid{Op}\hsdot{\circ }{.}\Conid{Inr}\hsdot{\circ }{.}\Conid{Inl}\mathbin{\$}\Conid{MUpdate}\;\Varid{r}\;(\Varid{(\Leftrightarrow)}\;\Varid{k}))){}\<[E]%
\\
\>[3]{}\mathrel{=}\mbox{\commentbegin ~ definition of \ensuremath{\Varid{h_{ND+f}}}  \commentend}{}\<[E]%
\\
\>[3]{}\hsindent{2}{}\<[5]%
\>[5]{}\Varid{fmap}\;\Varid{runStack}\hsdot{\circ }{.}\Varid{fmap}\;(\Varid{fmap}\;\Varid{fst})\hsdot{\circ }{.}\Varid{h_{Modify1}}\mathbin{\$}{}\<[E]%
\\
\>[5]{}\hsindent{2}{}\<[7]%
\>[7]{}\Conid{Op}\hsdot{\circ }{.}\Conid{Inr}\hsdot{\circ }{.}\Conid{Inl}\mathbin{\$}\Conid{Get}\;(\lambda (\Conid{Stack}\;\Varid{xs})\to {}\<[E]%
\\
\>[7]{}\hsindent{2}{}\<[9]%
\>[9]{}\Conid{Op}\hsdot{\circ }{.}\Conid{Inr}\hsdot{\circ }{.}\Conid{Inl}\mathbin{\$}\Conid{Put}\;(\Conid{Stack}\;(\Conid{Left}\;\Varid{r}\mathbin{:}\Varid{xs}))\;({}\<[E]%
\\
\>[9]{}\hsindent{2}{}\<[11]%
\>[11]{}\Conid{Op}\hsdot{\circ }{.}\Conid{Inl}\mathbin{\$}\Conid{MUpdate}\;\Varid{r}\;(\Varid{h_{ND+f}}\hsdot{\circ }{.}\Varid{(\Leftrightarrow)}\mathbin{\$}\Varid{k}))){}\<[E]%
\\
\>[3]{}\mathrel{=}\mbox{\commentbegin ~ definition of \ensuremath{\Varid{h_{Modify1}}}  \commentend}{}\<[E]%
\\
\>[3]{}\hsindent{2}{}\<[5]%
\>[5]{}\Varid{fmap}\;\Varid{runStack}\hsdot{\circ }{.}\Varid{fmap}\;(\Varid{fmap}\;\Varid{fst})\mathbin{\$}\lambda \Varid{s}\to {}\<[E]%
\\
\>[5]{}\hsindent{2}{}\<[7]%
\>[7]{}\Conid{Op}\hsdot{\circ }{.}\Conid{Inl}\mathbin{\$}\Conid{Get}\;(\lambda (\Conid{Stack}\;\Varid{xs})\to {}\<[E]%
\\
\>[7]{}\hsindent{2}{}\<[9]%
\>[9]{}\Conid{Op}\hsdot{\circ }{.}\Conid{Inl}\mathbin{\$}\Conid{Put}\;(\Conid{Stack}\;(\Conid{Left}\;\Varid{r}\mathbin{:}\Varid{xs}))\;({}\<[E]%
\\
\>[9]{}\hsindent{2}{}\<[11]%
\>[11]{}(\Varid{h_{Modify1}}\hsdot{\circ }{.}\Varid{h_{ND+f}}\hsdot{\circ }{.}\Varid{(\Leftrightarrow)}\mathbin{\$}\Varid{k})\;(\Varid{s}\mathbin{\oplus}\Varid{r}))){}\<[E]%
\\
\>[3]{}\mathrel{=}\mbox{\commentbegin ~ definition of \ensuremath{\Varid{fmap}\;(\Varid{fmap}\;\Varid{fst})}  \commentend}{}\<[E]%
\\
\>[3]{}\hsindent{2}{}\<[5]%
\>[5]{}\Varid{fmap}\;\Varid{runStack}\mathbin{\$}\lambda \Varid{s}\to {}\<[E]%
\\
\>[5]{}\hsindent{2}{}\<[7]%
\>[7]{}\Conid{Op}\hsdot{\circ }{.}\Conid{Inl}\mathbin{\$}\Conid{Get}\;(\lambda (\Conid{Stack}\;\Varid{xs})\to {}\<[E]%
\\
\>[7]{}\hsindent{2}{}\<[9]%
\>[9]{}\Conid{Op}\hsdot{\circ }{.}\Conid{Inl}\mathbin{\$}\Conid{Put}\;(\Conid{Stack}\;(\Conid{Left}\;\Varid{r}\mathbin{:}\Varid{xs}))\;({}\<[E]%
\\
\>[9]{}\hsindent{2}{}\<[11]%
\>[11]{}(\Varid{fmap}\;(\Varid{fmap}\;\Varid{fst})\hsdot{\circ }{.}\Varid{h_{Modify1}}\hsdot{\circ }{.}\Varid{h_{ND+f}}\hsdot{\circ }{.}\Varid{(\Leftrightarrow)}\mathbin{\$}\Varid{k})\;(\Varid{s}\mathbin{\oplus}\Varid{r}))){}\<[E]%
\\
\>[3]{}\mathrel{=}\mbox{\commentbegin ~ definition of \ensuremath{\Varid{fmap}}  \commentend}{}\<[E]%
\\
\>[3]{}\hsindent{2}{}\<[5]%
\>[5]{}\lambda \Varid{s}\to \Varid{runStack}\mathbin{\$}{}\<[E]%
\\
\>[5]{}\hsindent{2}{}\<[7]%
\>[7]{}\Conid{Op}\hsdot{\circ }{.}\Conid{Inl}\mathbin{\$}\Conid{Get}\;(\lambda (\Conid{Stack}\;\Varid{xs})\to {}\<[E]%
\\
\>[7]{}\hsindent{2}{}\<[9]%
\>[9]{}\Conid{Op}\hsdot{\circ }{.}\Conid{Inl}\mathbin{\$}\Conid{Put}\;(\Conid{Stack}\;(\Conid{Left}\;\Varid{r}\mathbin{:}\Varid{xs}))\;({}\<[E]%
\\
\>[9]{}\hsindent{2}{}\<[11]%
\>[11]{}(\Varid{fmap}\;(\Varid{fmap}\;\Varid{fst})\hsdot{\circ }{.}\Varid{h_{Modify1}}\hsdot{\circ }{.}\Varid{h_{ND+f}}\hsdot{\circ }{.}\Varid{(\Leftrightarrow)}\mathbin{\$}\Varid{k})\;(\Varid{s}\mathbin{\oplus}\Varid{r}))){}\<[E]%
\\
\>[3]{}\mathrel{=}\mbox{\commentbegin ~ definition of \ensuremath{\Varid{h_{GlobalM}}}  \commentend}{}\<[E]%
\\
\>[3]{}\hsindent{2}{}\<[5]%
\>[5]{}\lambda \Varid{s}\to \Varid{runStack}\mathbin{\$}{}\<[E]%
\\
\>[5]{}\hsindent{2}{}\<[7]%
\>[7]{}\Conid{Op}\hsdot{\circ }{.}\Conid{Inl}\mathbin{\$}\Conid{Get}\;(\lambda (\Conid{Stack}\;\Varid{xs})\to {}\<[E]%
\\
\>[7]{}\hsindent{2}{}\<[9]%
\>[9]{}\Conid{Op}\hsdot{\circ }{.}\Conid{Inl}\mathbin{\$}\Conid{Put}\;(\Conid{Stack}\;(\Conid{Left}\;\Varid{r}\mathbin{:}\Varid{xs}))\;({}\<[E]%
\\
\>[9]{}\hsindent{2}{}\<[11]%
\>[11]{}(\Varid{h_{GlobalM}}\;\Varid{k})\;(\Varid{s}\mathbin{\oplus}\Varid{r}))){}\<[E]%
\\
\>[3]{}\mathrel{=}\mbox{\commentbegin ~ definition of \ensuremath{\Varid{runStack}}  \commentend}{}\<[E]%
\\
\>[3]{}\hsindent{2}{}\<[5]%
\>[5]{}\lambda \Varid{s}\to \Varid{fmap}\;\Varid{fst}\hsdot{\circ }{.}\Varid{flip}\;\Varid{h_{State1}}\;(\Conid{Stack}\;[\mskip1.5mu \mskip1.5mu])\mathbin{\$}{}\<[E]%
\\
\>[5]{}\hsindent{2}{}\<[7]%
\>[7]{}\Conid{Op}\hsdot{\circ }{.}\Conid{Inl}\mathbin{\$}\Conid{Get}\;(\lambda (\Conid{Stack}\;\Varid{xs})\to {}\<[E]%
\\
\>[7]{}\hsindent{2}{}\<[9]%
\>[9]{}\Conid{Op}\hsdot{\circ }{.}\Conid{Inl}\mathbin{\$}\Conid{Put}\;(\Conid{Stack}\;(\Conid{Left}\;\Varid{r}\mathbin{:}\Varid{xs}))\;({}\<[E]%
\\
\>[9]{}\hsindent{2}{}\<[11]%
\>[11]{}(\Varid{h_{GlobalM}}\;\Varid{k})\;(\Varid{s}\mathbin{\oplus}\Varid{r}))){}\<[E]%
\\
\>[3]{}\mathrel{=}\mbox{\commentbegin ~ definition of \ensuremath{\Varid{h_{State1}}}  \commentend}{}\<[E]%
\\
\>[3]{}\hsindent{2}{}\<[5]%
\>[5]{}\lambda \Varid{s}\to \Varid{fmap}\;\Varid{fst}\mathbin{\$}(\lambda \Varid{t}\to {}\<[E]%
\\
\>[5]{}\hsindent{2}{}\<[7]%
\>[7]{}(\lambda (\Conid{Stack}\;\Varid{xs})\to \lambda \anonymous \to {}\<[E]%
\\
\>[7]{}\hsindent{2}{}\<[9]%
\>[9]{}((\Varid{fmap}\;\Varid{h_{State1}}\hsdot{\circ }{.}\Varid{h_{GlobalM}}\mathbin{\$}\Varid{k})\;(\Varid{s}\mathbin{\oplus}\Varid{r}))\;(\Conid{Stack}\;(\Conid{Left}\;\Varid{r}\mathbin{:}\Varid{xs})){}\<[E]%
\\
\>[5]{}\hsindent{2}{}\<[7]%
\>[7]{})\;\Varid{t}\;\Varid{t}{}\<[E]%
\\
\>[3]{}\hsindent{2}{}\<[5]%
\>[5]{})(\Conid{Stack}\;[\mskip1.5mu \mskip1.5mu]){}\<[E]%
\\
\>[3]{}\mathrel{=}\mbox{\commentbegin ~ function application  \commentend}{}\<[E]%
\\
\>[3]{}\hsindent{2}{}\<[5]%
\>[5]{}\lambda \Varid{s}\to \Varid{fmap}\;\Varid{fst}\mathbin{\$}{}\<[E]%
\\
\>[5]{}\hsindent{2}{}\<[7]%
\>[7]{}(\lambda (\Conid{Stack}\;\Varid{xs})\to \lambda \anonymous \to {}\<[E]%
\\
\>[7]{}\hsindent{2}{}\<[9]%
\>[9]{}((\Varid{fmap}\;\Varid{h_{State1}}\hsdot{\circ }{.}\Varid{h_{GlobalM}}\mathbin{\$}\Varid{k})\;(\Varid{s}\mathbin{\oplus}\Varid{r}))\;(\Conid{Stack}\;(\Conid{Left}\;\Varid{r}\mathbin{:}\Varid{xs})){}\<[E]%
\\
\>[5]{}\hsindent{2}{}\<[7]%
\>[7]{})\;(\Conid{Stack}\;[\mskip1.5mu \mskip1.5mu])\;(\Conid{Stack}\;[\mskip1.5mu \mskip1.5mu]){}\<[E]%
\\
\>[3]{}\mathrel{=}\mbox{\commentbegin ~ function application  \commentend}{}\<[E]%
\\
\>[3]{}\hsindent{2}{}\<[5]%
\>[5]{}\lambda \Varid{s}\to \Varid{fmap}\;\Varid{fst}\mathbin{\$}{}\<[E]%
\\
\>[5]{}\hsindent{2}{}\<[7]%
\>[7]{}(\lambda \anonymous \to {}\<[E]%
\\
\>[7]{}\hsindent{2}{}\<[9]%
\>[9]{}((\Varid{fmap}\;\Varid{h_{State1}}\hsdot{\circ }{.}\Varid{h_{GlobalM}}\mathbin{\$}\Varid{k})\;(\Varid{s}\mathbin{\oplus}\Varid{r}))\;(\Conid{Stack}\;(\Conid{Left}\;\Varid{r}\mathbin{:}[\mskip1.5mu \mskip1.5mu])){}\<[E]%
\\
\>[5]{}\hsindent{2}{}\<[7]%
\>[7]{})\;(\Conid{Stack}\;[\mskip1.5mu \mskip1.5mu]){}\<[E]%
\\
\>[3]{}\mathrel{=}\mbox{\commentbegin ~ function application  \commentend}{}\<[E]%
\\
\>[3]{}\hsindent{2}{}\<[5]%
\>[5]{}\lambda \Varid{s}\to \Varid{fmap}\;\Varid{fst}\mathbin{\$}{}\<[E]%
\\
\>[5]{}\hsindent{2}{}\<[7]%
\>[7]{}((\Varid{fmap}\;\Varid{h_{State1}}\hsdot{\circ }{.}\Varid{h_{GlobalM}}\mathbin{\$}\Varid{k})\;(\Varid{s}\mathbin{\oplus}\Varid{r}))\;(\Conid{Stack}\;(\Conid{Left}\;\Varid{r}\mathbin{:}[\mskip1.5mu \mskip1.5mu])){}\<[E]%
\\
\>[3]{}\mathrel{=}\mbox{\commentbegin ~ function application  \commentend}{}\<[E]%
\\
\>[3]{}\hsindent{2}{}\<[5]%
\>[5]{}\lambda \Varid{s}\to \Varid{fmap}\;\Varid{fst}\mathbin{\$}{}\<[E]%
\\
\>[5]{}\hsindent{2}{}\<[7]%
\>[7]{}((\Varid{fmap}\;\Varid{h_{State1}}\hsdot{\circ }{.}\Varid{h_{GlobalM}}\mathbin{\$}\Varid{k})\;(\Varid{s}\mathbin{\oplus}\Varid{r}))\;(\Conid{Stack}\;(\Conid{Left}\;\Varid{r}\mathbin{:}[\mskip1.5mu \mskip1.5mu])){}\<[E]%
\\
\>[3]{}\mathrel{=}\mbox{\commentbegin ~ definition of \ensuremath{\Varid{flip}} and reformulation  \commentend}{}\<[E]%
\\
\>[3]{}\hsindent{2}{}\<[5]%
\>[5]{}\lambda \Varid{s}\to (\Varid{fmap}\;(\Varid{fmap}\;\Varid{fst}\hsdot{\circ }{.}\Varid{flip}\;\Varid{h_{State1}}\;(\Conid{Stack}\;[\mskip1.5mu \Conid{Left}\;\Varid{r}\mskip1.5mu]))\hsdot{\circ }{.}\Varid{h_{GlobalM}}\mathbin{\$}\Varid{k})\;(\Varid{s}\mathbin{\oplus}\Varid{r}){}\<[E]%
\\
\>[3]{}\mathrel{=}\mbox{\commentbegin ~ \Cref{lemma:trail-stack-tracks-state} and definition of \ensuremath{\Varid{fmap}} and \ensuremath{\Varid{fst}}  \commentend}{}\<[E]%
\\
\>[3]{}\hsindent{2}{}\<[5]%
\>[5]{}\lambda \Varid{s}\to (\Varid{fmap}\;(\Varid{fmap}\;\Varid{fst}\hsdot{\circ }{.}\Varid{flip}\;\Varid{h_{State1}}\;(\Conid{Stack}\;[\mskip1.5mu \mskip1.5mu]))\hsdot{\circ }{.}\Varid{h_{GlobalM}}\mathbin{\$}\Varid{k})\;(\Varid{s}\mathbin{\oplus}\Varid{r}){}\<[E]%
\\
\>[3]{}\mathrel{=}\mbox{\commentbegin ~ definition of \ensuremath{\Varid{runStack}}  \commentend}{}\<[E]%
\\
\>[3]{}\hsindent{2}{}\<[5]%
\>[5]{}\lambda \Varid{s}\to (\Varid{fmap}\;\Varid{runStack}\hsdot{\circ }{.}\Varid{h_{GlobalM}}\mathbin{\$}\Varid{k})\;(\Varid{s}\mathbin{\oplus}\Varid{r}){}\<[E]%
\\
\>[3]{}\mathrel{=}\mbox{\commentbegin ~ definition of \ensuremath{\Varid{alg}_{\Varid{LHS}}^{\Varid{S}}}  \commentend}{}\<[E]%
\\
\>[3]{}\hsindent{2}{}\<[5]%
\>[5]{}\Varid{alg}_{\Varid{LHS}}^{\Varid{S}}\;(\Conid{MUpdate}\;\Varid{r}\;(\Varid{fmap}\;\Varid{runStack}\hsdot{\circ }{.}\Varid{h_{GlobalM}}\mathbin{\$}\Varid{k})){}\<[E]%
\\
\>[3]{}\mathrel{=}\mbox{\commentbegin ~ definition of \ensuremath{\Varid{fmap}}  \commentend}{}\<[E]%
\\
\>[3]{}\hsindent{2}{}\<[5]%
\>[5]{}\Varid{alg}_{\Varid{LHS}}^{\Varid{S}}\;(\Varid{fmap}\;(\Varid{fmap}\;\Varid{runStack}\hsdot{\circ }{.}\Varid{h_{GlobalM}})\;(\Conid{MUpdate}\;\Varid{r}\;\Varid{k})){}\<[E]%
\ColumnHook
\end{hscode}\resethooks
\indentend 
For the second subcondition \refd{}, we can define \ensuremath{\Varid{alg}_{\Varid{LHS}}^{\Varid{ND}}} as
follows.\indentbegin \begin{hscode}\SaveRestoreHook
\column{B}{@{}>{\hspre}l<{\hspost}@{}}%
\column{3}{@{}>{\hspre}l<{\hspost}@{}}%
\column{22}{@{}>{\hspre}l<{\hspost}@{}}%
\column{E}{@{}>{\hspre}l<{\hspost}@{}}%
\>[3]{}\Varid{alg}_{\Varid{LHS}}^{\Varid{ND}}\mathbin{::}\Conid{Functor}\;\Varid{f}\Rightarrow \Varid{Nondet_{F}}\;(\Varid{s}\to \Conid{Free}\;\Varid{f}\;[\mskip1.5mu \Varid{a}\mskip1.5mu])\to (\Varid{s}\to \Conid{Free}\;\Varid{f}\;[\mskip1.5mu \Varid{a}\mskip1.5mu]){}\<[E]%
\\
\>[3]{}\Varid{alg}_{\Varid{LHS}}^{\Varid{ND}}\;\Conid{Fail}{}\<[22]%
\>[22]{}\mathrel{=}\lambda \Varid{s}\to \Conid{Var}\;[\mskip1.5mu \mskip1.5mu]{}\<[E]%
\\
\>[3]{}\Varid{alg}_{\Varid{LHS}}^{\Varid{ND}}\;(\Conid{Or}\;\Varid{p}\;\Varid{q}){}\<[22]%
\>[22]{}\mathrel{=}\lambda \Varid{s}\to \Varid{liftM2}\;(+\!\!+)\;(\Varid{p}\;\Varid{s})\;(\Varid{q}\;\Varid{s}){}\<[E]%
\ColumnHook
\end{hscode}\resethooks
\indentend We prove it by a case analysis on the shape of input \ensuremath{\Varid{op}\mathbin{::}\Varid{Nondet_{F}}\;(\Conid{Free}\;(\Varid{Modify_{F}}\;\Varid{s}\;\Varid{r}\mathrel{{:}{+}{:}}\Varid{Nondet_{F}}\mathrel{{:}{+}{:}}\Varid{f})\;\Varid{a})}.

\noindent \mbox{\underline{case \ensuremath{\Varid{op}\mathrel{=}\Conid{Fail}}}}
In the corresponding case of \Cref{app:modify-fusing-lhs}, we have
calculated that \ensuremath{\Varid{h_{GlobalM}}\;(\Conid{Op}\;(\Conid{Inr}\;(\Conid{Inl}\;\Conid{Fail})))\mathrel{=}\lambda \Varid{s}\to \Conid{Var}\;[\mskip1.5mu \mskip1.5mu]} \refs{}.
\indentbegin \begin{hscode}\SaveRestoreHook
\column{B}{@{}>{\hspre}l<{\hspost}@{}}%
\column{3}{@{}>{\hspre}l<{\hspost}@{}}%
\column{4}{@{}>{\hspre}l<{\hspost}@{}}%
\column{5}{@{}>{\hspre}l<{\hspost}@{}}%
\column{E}{@{}>{\hspre}l<{\hspost}@{}}%
\>[5]{}\Varid{fmap}\;\Varid{runStack}\mathbin{\$}\Varid{h_{GlobalM}}\;(\Varid{alg}_{2}\;(\Conid{Fail})){}\<[E]%
\\
\>[3]{}\mathrel{=}\mbox{\commentbegin ~ definition of \ensuremath{\Varid{alg}_{2}}  \commentend}{}\<[E]%
\\
\>[3]{}\hsindent{2}{}\<[5]%
\>[5]{}\Varid{fmap}\;\Varid{runStack}\mathbin{\$}\Varid{h_{GlobalM}}\;(\Conid{Op}\;(\Conid{Inr}\;(\Conid{Inl}\;\Conid{Fail}))){}\<[E]%
\\
\>[3]{}\mathrel{=}\mbox{\commentbegin ~ Equation \refs{}  \commentend}{}\<[E]%
\\
\>[3]{}\hsindent{2}{}\<[5]%
\>[5]{}\Varid{fmap}\;\Varid{runStack}\mathbin{\$}\lambda \Varid{s}\to \Conid{Var}\;[\mskip1.5mu \mskip1.5mu]{}\<[E]%
\\
\>[3]{}\mathrel{=}\mbox{\commentbegin ~ definition of \ensuremath{\Varid{fmap}}  \commentend}{}\<[E]%
\\
\>[3]{}\hsindent{2}{}\<[5]%
\>[5]{}\lambda \Varid{s}\to \Varid{runStack}\mathbin{\$}\Conid{Var}\;[\mskip1.5mu \mskip1.5mu]{}\<[E]%
\\
\>[3]{}\mathrel{=}\mbox{\commentbegin ~ definition of \ensuremath{\Varid{runStack}}  \commentend}{}\<[E]%
\\
\>[3]{}\hsindent{2}{}\<[5]%
\>[5]{}\lambda \Varid{s}\to \Varid{fmap}\;\Varid{fst}\hsdot{\circ }{.}\Varid{flip}\;\Varid{h_{State1}}\;(\Conid{Stack}\;[\mskip1.5mu \mskip1.5mu])\mathbin{\$}\Conid{Var}\;[\mskip1.5mu \mskip1.5mu]{}\<[E]%
\\
\>[3]{}\mathrel{=}\mbox{\commentbegin ~ definition of \ensuremath{\Varid{h_{State1}}}  \commentend}{}\<[E]%
\\
\>[3]{}\hsindent{2}{}\<[5]%
\>[5]{}\lambda \Varid{s}\to \Varid{fmap}\;\Varid{fst}\mathbin{\$}\Conid{Var}\;([\mskip1.5mu \mskip1.5mu],\Conid{Stack}\;[\mskip1.5mu \mskip1.5mu]){}\<[E]%
\\
\>[3]{}\mathrel{=}\mbox{\commentbegin ~ definition of \ensuremath{\Varid{fmap}}  \commentend}{}\<[E]%
\\
\>[3]{}\hsindent{2}{}\<[5]%
\>[5]{}\lambda \Varid{s}\to \Conid{Var}\;[\mskip1.5mu \mskip1.5mu]{}\<[E]%
\\
\>[3]{}\mathrel{=}\mbox{\commentbegin ~ definition of \ensuremath{\Varid{alg}_{\Varid{RHS}}^{\Varid{ND}}}   \commentend}{}\<[E]%
\\
\>[3]{}\hsindent{1}{}\<[4]%
\>[4]{}\Varid{alg}_{\Varid{RHS}}^{\Varid{ND}}\;\Conid{Fail}{}\<[E]%
\\
\>[3]{}\mathrel{=}\mbox{\commentbegin ~ definition of \ensuremath{\Varid{fmap}}  \commentend}{}\<[E]%
\\
\>[3]{}\hsindent{1}{}\<[4]%
\>[4]{}\Varid{alg}_{\Varid{RHS}}^{\Varid{ND}}\;(\Varid{fmap}\;(\Varid{fmap}\;\Varid{runStack}\hsdot{\circ }{.}\Varid{h_{GlobalM}})\;\Conid{Fail}){}\<[E]%
\ColumnHook
\end{hscode}\resethooks
\indentend 
\noindent \mbox{\underline{case \ensuremath{\Varid{op}\mathrel{=}\Conid{Or}\;\Varid{p}\;\Varid{q}}}}
From \ensuremath{\Varid{op}} is in the codomain of \ensuremath{\Varid{fmap}\;\Varid{local2global_M}} we obtain \ensuremath{\Varid{p}} and
\ensuremath{\Varid{q}} are in the codomain of \ensuremath{\Varid{local2global_M}}.
\indentbegin \begin{hscode}\SaveRestoreHook
\column{B}{@{}>{\hspre}l<{\hspost}@{}}%
\column{3}{@{}>{\hspre}l<{\hspost}@{}}%
\column{5}{@{}>{\hspre}l<{\hspost}@{}}%
\column{7}{@{}>{\hspre}l<{\hspost}@{}}%
\column{12}{@{}>{\hspre}l<{\hspost}@{}}%
\column{13}{@{}>{\hspre}l<{\hspost}@{}}%
\column{16}{@{}>{\hspre}l<{\hspost}@{}}%
\column{18}{@{}>{\hspre}l<{\hspost}@{}}%
\column{E}{@{}>{\hspre}l<{\hspost}@{}}%
\>[5]{}\Varid{fmap}\;\Varid{runStack}\hsdot{\circ }{.}\Varid{h_{GlobalM}}\mathbin{\$}\Varid{alg}_{2}\;(\Conid{Or}\;\Varid{p}\;\Varid{q}){}\<[E]%
\\
\>[3]{}\mathrel{=}\mbox{\commentbegin ~ definition of \ensuremath{\Varid{alg}_{2}}  \commentend}{}\<[E]%
\\
\>[3]{}\hsindent{2}{}\<[5]%
\>[5]{}\Varid{fmap}\;\Varid{runStack}\hsdot{\circ }{.}\Varid{h_{GlobalM}}\mathbin{\$}(\Varid{pushStack}\;(\Conid{Right}\;())>\!\!>\Varid{p})\mathbin{\talloblong}(\Varid{untrail}>\!\!>\Varid{q}){}\<[E]%
\\
\>[3]{}\mathrel{=}\mbox{\commentbegin ~ definition of \ensuremath{(\talloblong)}  \commentend}{}\<[E]%
\\
\>[3]{}\hsindent{2}{}\<[5]%
\>[5]{}\Varid{fmap}\;\Varid{runStack}\hsdot{\circ }{.}\Varid{h_{GlobalM}}\mathbin{\$}\Conid{Op}\hsdot{\circ }{.}\Conid{Inr}\hsdot{\circ }{.}\Conid{Inl}\mathbin{\$}\Conid{Or}{}\<[E]%
\\
\>[5]{}\hsindent{2}{}\<[7]%
\>[7]{}(\Varid{pushStack}\;(\Conid{Right}\;())>\!\!>\Varid{p})\;(\Varid{untrail}>\!\!>\Varid{q}){}\<[E]%
\\
\>[3]{}\mathrel{=}\mbox{\commentbegin ~ definition of \ensuremath{\Varid{h_{GlobalM}}}  \commentend}{}\<[E]%
\\
\>[3]{}\hsindent{2}{}\<[5]%
\>[5]{}\Varid{fmap}\;\Varid{runStack}\hsdot{\circ }{.}\Varid{fmap}\;(\Varid{fmap}\;\Varid{fst})\hsdot{\circ }{.}\Varid{h_{Modify1}}\hsdot{\circ }{.}\Varid{h_{ND+f}}\hsdot{\circ }{.}\Varid{(\Leftrightarrow)}\mathbin{\$}\Conid{Op}\hsdot{\circ }{.}\Conid{Inr}\hsdot{\circ }{.}\Conid{Inl}\mathbin{\$}\Conid{Or}{}\<[E]%
\\
\>[5]{}\hsindent{2}{}\<[7]%
\>[7]{}(\Varid{pushStack}\;(\Conid{Right}\;())>\!\!>\Varid{p})\;(\Varid{untrail}>\!\!>\Varid{q}){}\<[E]%
\\
\>[3]{}\mathrel{=}\mbox{\commentbegin ~ definition of \ensuremath{\Varid{(\Leftrightarrow)}}  \commentend}{}\<[E]%
\\
\>[3]{}\hsindent{2}{}\<[5]%
\>[5]{}\Varid{fmap}\;\Varid{runStack}\hsdot{\circ }{.}\Varid{fmap}\;(\Varid{fmap}\;\Varid{fst})\hsdot{\circ }{.}\Varid{h_{Modify1}}\hsdot{\circ }{.}\Varid{h_{ND+f}}\mathbin{\$}\Conid{Op}\hsdot{\circ }{.}\Conid{Inl}\mathbin{\$}\Conid{Or}{}\<[E]%
\\
\>[5]{}\hsindent{2}{}\<[7]%
\>[7]{}(\Varid{pushStack}\;(\Conid{Right}\;())>\!\!>\Varid{(\Leftrightarrow)}\;\Varid{p})\;(\Varid{untrail}>\!\!>\Varid{(\Leftrightarrow)}\;\Varid{q}){}\<[E]%
\\
\>[3]{}\mathrel{=}\mbox{\commentbegin ~ definition of \ensuremath{\Varid{h_{ND+f}}} and \ensuremath{\Varid{liftM2}}  \commentend}{}\<[E]%
\\
\>[3]{}\hsindent{2}{}\<[5]%
\>[5]{}\Varid{fmap}\;\Varid{runStack}\hsdot{\circ }{.}\Varid{fmap}\;(\Varid{fmap}\;\Varid{fst})\hsdot{\circ }{.}\Varid{h_{Modify1}}\mathbin{\$}\mathbf{do}{}\<[E]%
\\
\>[5]{}\hsindent{2}{}\<[7]%
\>[7]{}\Varid{x}\leftarrow \Varid{h_{ND+f}}\;(\Varid{(\Leftrightarrow)}\;(\Varid{pushStack}\;(\Conid{Right}\;())))>\!\!>\Varid{h_{ND+f}}\;(\Varid{(\Leftrightarrow)}\;\Varid{p}){}\<[E]%
\\
\>[5]{}\hsindent{2}{}\<[7]%
\>[7]{}\Varid{y}\leftarrow \Varid{h_{ND+f}}\;(\Varid{(\Leftrightarrow)}\;\Varid{untrail})>\!\!>\Varid{h_{ND+f}}\;(\Varid{(\Leftrightarrow)}\;\Varid{q}){}\<[E]%
\\
\>[5]{}\hsindent{2}{}\<[7]%
\>[7]{}\Varid{\eta}\;(\Varid{x}+\!\!+\Varid{y}){}\<[E]%
\\
\>[3]{}\mathrel{=}\mbox{\commentbegin ~ monad law  \commentend}{}\<[E]%
\\
\>[3]{}\hsindent{2}{}\<[5]%
\>[5]{}\Varid{fmap}\;\Varid{runStack}\hsdot{\circ }{.}\Varid{fmap}\;(\Varid{fmap}\;\Varid{fst})\hsdot{\circ }{.}\Varid{h_{Modify1}}\mathbin{\$}\mathbf{do}{}\<[E]%
\\
\>[5]{}\hsindent{2}{}\<[7]%
\>[7]{}\Varid{h_{ND+f}}\;(\Varid{(\Leftrightarrow)}\;(\Varid{pushStack}\;(\Conid{Right}\;()))){}\<[E]%
\\
\>[5]{}\hsindent{2}{}\<[7]%
\>[7]{}\Varid{x}\leftarrow \Varid{h_{ND+f}}\;(\Varid{(\Leftrightarrow)}\;\Varid{p}){}\<[E]%
\\
\>[5]{}\hsindent{2}{}\<[7]%
\>[7]{}\Varid{h_{ND+f}}\;(\Varid{(\Leftrightarrow)}\;\Varid{untrail}){}\<[E]%
\\
\>[5]{}\hsindent{2}{}\<[7]%
\>[7]{}\Varid{y}\leftarrow \Varid{h_{ND+f}}\;(\Varid{(\Leftrightarrow)}\;\Varid{q}){}\<[E]%
\\
\>[5]{}\hsindent{2}{}\<[7]%
\>[7]{}\Varid{\eta}\;(\Varid{x}+\!\!+\Varid{y}){}\<[E]%
\\
\>[3]{}\mathrel{=}\mbox{\commentbegin ~ definition of \ensuremath{\Varid{h_{Modify1}}} and \Cref{lemma:dist-hModify1}  \commentend}{}\<[E]%
\\
\>[3]{}\hsindent{2}{}\<[5]%
\>[5]{}\Varid{fmap}\;\Varid{runStack}\hsdot{\circ }{.}\Varid{fmap}\;(\Varid{fmap}\;\Varid{fst})\mathbin{\$}\lambda \Varid{s}\to \mathbf{do}{}\<[E]%
\\
\>[5]{}\hsindent{2}{}\<[7]%
\>[7]{}(\anonymous ,\Varid{s}_{1}){}\<[16]%
\>[16]{}\leftarrow \Varid{h_{Modify1}}\;(\Varid{h_{ND+f}}\;(\Varid{(\Leftrightarrow)}\;(\Varid{pushStack}\;(\Conid{Right}\;()))))\;\Varid{s}{}\<[E]%
\\
\>[5]{}\hsindent{2}{}\<[7]%
\>[7]{}(\Varid{x},\Varid{s}_{2}){}\<[16]%
\>[16]{}\leftarrow \Varid{h_{Modify1}}\;(\Varid{h_{ND+f}}\;(\Varid{(\Leftrightarrow)}\;\Varid{p}))\;\Varid{s}_{1}{}\<[E]%
\\
\>[5]{}\hsindent{2}{}\<[7]%
\>[7]{}(\anonymous ,\Varid{s3}){}\<[16]%
\>[16]{}\leftarrow \Varid{h_{Modify1}}\;(\Varid{h_{ND+f}}\;(\Varid{(\Leftrightarrow)}\;\Varid{untrail}))\;\Varid{s}_{2}{}\<[E]%
\\
\>[5]{}\hsindent{2}{}\<[7]%
\>[7]{}(\Varid{y},\Varid{s4}){}\<[16]%
\>[16]{}\leftarrow \Varid{h_{Modify1}}\;(\Varid{h_{ND+f}}\;(\Varid{(\Leftrightarrow)}\;\Varid{q}))\;\Varid{s3}{}\<[E]%
\\
\>[5]{}\hsindent{2}{}\<[7]%
\>[7]{}\Varid{\eta}\;(\Varid{x}+\!\!+\Varid{y},\Varid{s4}){}\<[E]%
\\
\>[3]{}\mathrel{=}\mbox{\commentbegin ~ definition of \ensuremath{\Varid{fmap}} (twice)  \commentend}{}\<[E]%
\\
\>[3]{}\hsindent{2}{}\<[5]%
\>[5]{}\Varid{fmap}\;\Varid{runStack}\mathbin{\$}\lambda \Varid{s}\to \mathbf{do}{}\<[E]%
\\
\>[5]{}\hsindent{2}{}\<[7]%
\>[7]{}(\anonymous ,\Varid{s}_{1})\leftarrow \Varid{h_{Modify1}}\;(\Varid{h_{ND+f}}\;(\Varid{(\Leftrightarrow)}\;(\Varid{pushStack}\;(\Conid{Right}\;()))))\;\Varid{s}{}\<[E]%
\\
\>[5]{}\hsindent{2}{}\<[7]%
\>[7]{}(\Varid{x},\Varid{s}_{2})\leftarrow \Varid{h_{Modify1}}\;(\Varid{h_{ND+f}}\;(\Varid{(\Leftrightarrow)}\;\Varid{p}))\;\Varid{s}_{1}{}\<[E]%
\\
\>[5]{}\hsindent{2}{}\<[7]%
\>[7]{}(\anonymous ,\Varid{s3})\leftarrow \Varid{h_{Modify1}}\;(\Varid{h_{ND+f}}\;(\Varid{(\Leftrightarrow)}\;\Varid{untrail}))\;\Varid{s}_{2}{}\<[E]%
\\
\>[5]{}\hsindent{2}{}\<[7]%
\>[7]{}(\Varid{y},{}\<[12]%
\>[12]{}\anonymous )\leftarrow \Varid{h_{Modify1}}\;(\Varid{h_{ND+f}}\;(\Varid{(\Leftrightarrow)}\;\Varid{q}))\;\Varid{s3}{}\<[E]%
\\
\>[5]{}\hsindent{2}{}\<[7]%
\>[7]{}\Varid{\eta}\;(\Varid{x}+\!\!+\Varid{y}){}\<[E]%
\\
\>[3]{}\mathrel{=}\mbox{\commentbegin ~ definition of \ensuremath{\Varid{fmap}} and \ensuremath{\Varid{runStack}}  \commentend}{}\<[E]%
\\
\>[3]{}\hsindent{2}{}\<[5]%
\>[5]{}\lambda \Varid{s}\to \Varid{fmap}\;\Varid{fst}\hsdot{\circ }{.}\Varid{flip}\;\Varid{h_{State1}}\;(\Conid{Stack}\;[\mskip1.5mu \mskip1.5mu])\mathbin{\$}\mathbf{do}{}\<[E]%
\\
\>[5]{}\hsindent{2}{}\<[7]%
\>[7]{}(\anonymous ,\Varid{s}_{1})\leftarrow \Varid{h_{Modify1}}\;(\Varid{h_{ND+f}}\;(\Varid{(\Leftrightarrow)}\;(\Varid{pushStack}\;(\Conid{Right}\;()))))\;\Varid{s}{}\<[E]%
\\
\>[5]{}\hsindent{2}{}\<[7]%
\>[7]{}(\Varid{x},\Varid{s}_{2})\leftarrow \Varid{h_{Modify1}}\;(\Varid{h_{ND+f}}\;(\Varid{(\Leftrightarrow)}\;\Varid{p}))\;\Varid{s}_{1}{}\<[E]%
\\
\>[5]{}\hsindent{2}{}\<[7]%
\>[7]{}(\anonymous ,\Varid{s3})\leftarrow \Varid{h_{Modify1}}\;(\Varid{h_{ND+f}}\;(\Varid{(\Leftrightarrow)}\;\Varid{untrail}))\;\Varid{s}_{2}{}\<[E]%
\\
\>[5]{}\hsindent{2}{}\<[7]%
\>[7]{}(\Varid{y},{}\<[12]%
\>[12]{}\anonymous )\leftarrow \Varid{h_{Modify1}}\;(\Varid{h_{ND+f}}\;(\Varid{(\Leftrightarrow)}\;\Varid{q}))\;\Varid{s3}{}\<[E]%
\\
\>[5]{}\hsindent{2}{}\<[7]%
\>[7]{}\Varid{\eta}\;(\Varid{x}+\!\!+\Varid{y}){}\<[E]%
\\
\>[3]{}\mathrel{=}\mbox{\commentbegin ~ definition of \ensuremath{\Varid{h_{State1}}} and \Cref{lemma:dist-hState1}  \commentend}{}\<[E]%
\\
\>[3]{}\hsindent{2}{}\<[5]%
\>[5]{}\lambda \Varid{s}\to \Varid{fmap}\;\Varid{fst}\mathbin{\$}(\lambda \Varid{t}\to \mathbf{do}{}\<[E]%
\\
\>[5]{}\hsindent{2}{}\<[7]%
\>[7]{}((\anonymous ,\Varid{s}_{1}),\Varid{t}_{1})\leftarrow \Varid{h_{State1}}\;(\Varid{h_{Modify1}}\;(\Varid{h_{ND+f}}\;(\Varid{(\Leftrightarrow)}\;(\Varid{pushStack}\;(\Conid{Right}\;()))))\;\Varid{s})\;\Varid{t}{}\<[E]%
\\
\>[5]{}\hsindent{2}{}\<[7]%
\>[7]{}((\Varid{x},\Varid{s}_{2}),\Varid{t}_{2})\leftarrow \Varid{h_{State1}}\;(\Varid{h_{Modify1}}\;(\Varid{h_{ND+f}}\;(\Varid{(\Leftrightarrow)}\;\Varid{p}))\;\Varid{s}_{1})\;\Varid{t}_{1}{}\<[E]%
\\
\>[5]{}\hsindent{2}{}\<[7]%
\>[7]{}((\anonymous ,\Varid{s3}),\Varid{t3})\leftarrow \Varid{h_{State1}}\;(\Varid{h_{Modify1}}\;(\Varid{h_{ND+f}}\;(\Varid{(\Leftrightarrow)}\;\Varid{untrail}))\;\Varid{s}_{2})\;\Varid{t}_{2}{}\<[E]%
\\
\>[5]{}\hsindent{2}{}\<[7]%
\>[7]{}((\Varid{y},{}\<[13]%
\>[13]{}\anonymous ),\Varid{t4})\leftarrow \Varid{h_{State1}}\;(\Varid{h_{Modify1}}\;(\Varid{h_{ND+f}}\;(\Varid{(\Leftrightarrow)}\;\Varid{q}))\;\Varid{s3})\;\Varid{t3}{}\<[E]%
\\
\>[5]{}\hsindent{2}{}\<[7]%
\>[7]{}\Varid{\eta}\;(\Varid{x}+\!\!+\Varid{y},\Varid{t4}))\;(\Conid{Stack}\;[\mskip1.5mu \mskip1.5mu]){}\<[E]%
\\
\>[3]{}\mathrel{=}\mbox{\commentbegin ~ function application  \commentend}{}\<[E]%
\\
\>[3]{}\hsindent{2}{}\<[5]%
\>[5]{}\lambda \Varid{s}\to \Varid{fmap}\;\Varid{fst}\mathbin{\$}\mathbf{do}{}\<[E]%
\\
\>[5]{}\hsindent{2}{}\<[7]%
\>[7]{}((\anonymous ,\Varid{s}_{1}),\Varid{t}_{1})\leftarrow \Varid{h_{State1}}\;(\Varid{h_{Modify1}}\;(\Varid{h_{ND+f}}\;(\Varid{(\Leftrightarrow)}\;(\Varid{pushStack}\;(\Conid{Right}\;()))))\;\Varid{s})\;(\Conid{Stack}\;[\mskip1.5mu \mskip1.5mu]){}\<[E]%
\\
\>[5]{}\hsindent{2}{}\<[7]%
\>[7]{}((\Varid{x},\Varid{s}_{2}),\Varid{t}_{2})\leftarrow \Varid{h_{State1}}\;(\Varid{h_{Modify1}}\;(\Varid{h_{ND+f}}\;(\Varid{(\Leftrightarrow)}\;\Varid{p}))\;\Varid{s}_{1})\;\Varid{t}_{1}{}\<[E]%
\\
\>[5]{}\hsindent{2}{}\<[7]%
\>[7]{}((\anonymous ,\Varid{s3}),\Varid{t3})\leftarrow \Varid{h_{State1}}\;(\Varid{h_{Modify1}}\;(\Varid{h_{ND+f}}\;(\Varid{(\Leftrightarrow)}\;\Varid{untrail}))\;\Varid{s}_{2})\;\Varid{t}_{2}{}\<[E]%
\\
\>[5]{}\hsindent{2}{}\<[7]%
\>[7]{}((\Varid{y},{}\<[13]%
\>[13]{}\anonymous ),\Varid{t4})\leftarrow \Varid{h_{State1}}\;(\Varid{h_{Modify1}}\;(\Varid{h_{ND+f}}\;(\Varid{(\Leftrightarrow)}\;\Varid{q}))\;\Varid{s3})\;\Varid{t3}{}\<[E]%
\\
\>[5]{}\hsindent{2}{}\<[7]%
\>[7]{}\Varid{\eta}\;(\Varid{x}+\!\!+\Varid{y},\Varid{t4}){}\<[E]%
\\
\>[3]{}\mathrel{=}\mbox{\commentbegin ~ definition of \ensuremath{\Varid{fmap}}  \commentend}{}\<[E]%
\\
\>[3]{}\hsindent{2}{}\<[5]%
\>[5]{}\lambda \Varid{s}\to \mathbf{do}{}\<[E]%
\\
\>[5]{}\hsindent{2}{}\<[7]%
\>[7]{}((\anonymous ,\Varid{s}_{1}),\Varid{t}_{1})\leftarrow \Varid{h_{State1}}\;(\Varid{h_{Modify1}}\;(\Varid{h_{ND+f}}\;(\Varid{(\Leftrightarrow)}\;(\Varid{pushStack}\;(\Conid{Right}\;()))))\;\Varid{s})\;(\Conid{Stack}\;[\mskip1.5mu \mskip1.5mu]){}\<[E]%
\\
\>[5]{}\hsindent{2}{}\<[7]%
\>[7]{}((\Varid{x},\Varid{s}_{2}),\Varid{t}_{2})\leftarrow \Varid{h_{State1}}\;(\Varid{h_{Modify1}}\;(\Varid{h_{ND+f}}\;(\Varid{(\Leftrightarrow)}\;\Varid{p}))\;\Varid{s}_{1})\;\Varid{t}_{1}{}\<[E]%
\\
\>[5]{}\hsindent{2}{}\<[7]%
\>[7]{}((\anonymous ,\Varid{s3}),\Varid{t3})\leftarrow \Varid{h_{State1}}\;(\Varid{h_{Modify1}}\;(\Varid{h_{ND+f}}\;(\Varid{(\Leftrightarrow)}\;\Varid{untrail}))\;\Varid{s}_{2})\;\Varid{t}_{2}{}\<[E]%
\\
\>[5]{}\hsindent{2}{}\<[7]%
\>[7]{}((\Varid{y},{}\<[13]%
\>[13]{}\anonymous ),{}\<[18]%
\>[18]{}\anonymous )\leftarrow \Varid{h_{State1}}\;(\Varid{h_{Modify1}}\;(\Varid{h_{ND+f}}\;(\Varid{(\Leftrightarrow)}\;\Varid{q}))\;\Varid{s3})\;\Varid{t3}{}\<[E]%
\\
\>[5]{}\hsindent{2}{}\<[7]%
\>[7]{}\Varid{\eta}\;(\Varid{x}+\!\!+\Varid{y}){}\<[E]%
\\
\>[3]{}\mathrel{=}\mbox{\commentbegin ~ \Cref{lemma:state-stack-restore}  \commentend}{}\<[E]%
\\
\>[3]{}\hsindent{2}{}\<[5]%
\>[5]{}\lambda \Varid{s}\to \mathbf{do}{}\<[E]%
\\
\>[5]{}\hsindent{2}{}\<[7]%
\>[7]{}((\Varid{x},{}\<[13]%
\>[13]{}\anonymous ),{}\<[18]%
\>[18]{}\anonymous )\leftarrow \Varid{h_{State1}}\;(\Varid{h_{Modify1}}\;(\Varid{h_{ND+f}}\;(\Varid{(\Leftrightarrow)}\;\Varid{p}))\;\Varid{s})\;(\Conid{Stack}\;[\mskip1.5mu \mskip1.5mu]){}\<[E]%
\\
\>[5]{}\hsindent{2}{}\<[7]%
\>[7]{}((\Varid{y},{}\<[13]%
\>[13]{}\anonymous ),{}\<[18]%
\>[18]{}\anonymous )\leftarrow \Varid{h_{State1}}\;(\Varid{h_{Modify1}}\;(\Varid{h_{ND+f}}\;(\Varid{(\Leftrightarrow)}\;\Varid{q}))\;\Varid{s})\;(\Conid{Stack}\;[\mskip1.5mu \mskip1.5mu]){}\<[E]%
\\
\>[5]{}\hsindent{2}{}\<[7]%
\>[7]{}\Varid{\eta}\;(\Varid{x}+\!\!+\Varid{y}){}\<[E]%
\\
\>[3]{}\mathrel{=}\mbox{\commentbegin ~ definition of \ensuremath{\Varid{runStack}}  \commentend}{}\<[E]%
\\
\>[3]{}\hsindent{2}{}\<[5]%
\>[5]{}\lambda \Varid{s}\to \mathbf{do}{}\<[E]%
\\
\>[5]{}\hsindent{2}{}\<[7]%
\>[7]{}(\Varid{x},\anonymous )\leftarrow \Varid{runStack}\;(\Varid{h_{Modify1}}\;(\Varid{h_{ND+f}}\;(\Varid{(\Leftrightarrow)}\;\Varid{p}))\;\Varid{s}){}\<[E]%
\\
\>[5]{}\hsindent{2}{}\<[7]%
\>[7]{}(\Varid{y},\anonymous )\leftarrow \Varid{runStack}\;(\Varid{h_{Modify1}}\;(\Varid{h_{ND+f}}\;(\Varid{(\Leftrightarrow)}\;\Varid{q}))\;\Varid{s}){}\<[E]%
\\
\>[5]{}\hsindent{2}{}\<[7]%
\>[7]{}\Varid{\eta}\;(\Varid{x}+\!\!+\Varid{y}){}\<[E]%
\\
\>[3]{}\mathrel{=}\mbox{\commentbegin ~ definition of \ensuremath{\Varid{h_{GlobalM}}}  \commentend}{}\<[E]%
\\
\>[3]{}\hsindent{2}{}\<[5]%
\>[5]{}\lambda \Varid{s}\to \mathbf{do}{}\<[E]%
\\
\>[5]{}\hsindent{2}{}\<[7]%
\>[7]{}\Varid{x}\leftarrow \Varid{runStack}\;(\Varid{h_{GlobalM}}\;\Varid{p}\;\Varid{s}){}\<[E]%
\\
\>[5]{}\hsindent{2}{}\<[7]%
\>[7]{}\Varid{y}\leftarrow \Varid{runStack}\;(\Varid{h_{GlobalM}}\;\Varid{q}\;\Varid{s}){}\<[E]%
\\
\>[5]{}\hsindent{2}{}\<[7]%
\>[7]{}\Varid{\eta}\;(\Varid{x}+\!\!+\Varid{y}){}\<[E]%
\\
\>[3]{}\mathrel{=}\mbox{\commentbegin ~ definition of \ensuremath{\Varid{fmap}}  \commentend}{}\<[E]%
\\
\>[3]{}\hsindent{2}{}\<[5]%
\>[5]{}\lambda \Varid{s}\to \mathbf{do}{}\<[E]%
\\
\>[5]{}\hsindent{2}{}\<[7]%
\>[7]{}\Varid{x}\leftarrow (\Varid{fmap}\;\Varid{runStack}\hsdot{\circ }{.}\Varid{hGobalM})\;\Varid{p}\;\Varid{s}{}\<[E]%
\\
\>[5]{}\hsindent{2}{}\<[7]%
\>[7]{}\Varid{y}\leftarrow (\Varid{fmap}\;\Varid{runStack}\hsdot{\circ }{.}\Varid{hGobalM})\;\Varid{q}\;\Varid{s}{}\<[E]%
\\
\>[5]{}\hsindent{2}{}\<[7]%
\>[7]{}\Varid{\eta}\;(\Varid{x}+\!\!+\Varid{y})){}\<[E]%
\\
\>[3]{}\mathrel{=}\mbox{\commentbegin ~ definition of \ensuremath{\Varid{liftM2}}  \commentend}{}\<[E]%
\\
\>[3]{}\hsindent{2}{}\<[5]%
\>[5]{}\lambda \Varid{s}\to \Varid{liftM2}\;(+\!\!+)\;((\Varid{fmap}\;\Varid{runStack}\hsdot{\circ }{.}\Varid{hGobalM})\;\Varid{p}\;\Varid{s})\;((\Varid{fmap}\;\Varid{runStack}\hsdot{\circ }{.}\Varid{hGobalM})\;\Varid{q}\;\Varid{s}){}\<[E]%
\\
\>[3]{}\mathrel{=}\mbox{\commentbegin ~ definition of \ensuremath{\Varid{alg}_{\Varid{LHS}}^{\Varid{ND}}}   \commentend}{}\<[E]%
\\
\>[3]{}\hsindent{2}{}\<[5]%
\>[5]{}\Varid{alg}_{\Varid{LHS}}^{\Varid{ND}}\;(\Conid{Or}\;((\Varid{fmap}\;\Varid{runStack}\hsdot{\circ }{.}\Varid{hGobalM})\;\Varid{p})\;((\Varid{fmap}\;\Varid{runStack}\hsdot{\circ }{.}\Varid{hGobalM})\;\Varid{q})){}\<[E]%
\\
\>[3]{}\mathrel{=}\mbox{\commentbegin ~ definition of \ensuremath{\Varid{fmap}}  \commentend}{}\<[E]%
\\
\>[3]{}\hsindent{2}{}\<[5]%
\>[5]{}\Varid{alg}_{\Varid{LHS}}^{\Varid{ND}}\;(\Varid{fmap}\;(\Varid{fmap}\;\Varid{runStack}\hsdot{\circ }{.}\Varid{hGobalM})\;(\Conid{Or}\;\Varid{p}\;\Varid{q})){}\<[E]%
\ColumnHook
\end{hscode}\resethooks
\indentend 
For the last subcondition \refe{}, we can define \ensuremath{\Varid{fwd}_{\Varid{LHS}}} as follows.\indentbegin \begin{hscode}\SaveRestoreHook
\column{B}{@{}>{\hspre}l<{\hspost}@{}}%
\column{3}{@{}>{\hspre}l<{\hspost}@{}}%
\column{E}{@{}>{\hspre}l<{\hspost}@{}}%
\>[3]{}\Varid{fwd}_{\Varid{LHS}}\mathbin{::}\Conid{Functor}\;\Varid{f}\Rightarrow \Varid{f}\;(\Varid{s}\to \Conid{Free}\;\Varid{f}\;[\mskip1.5mu \Varid{a}\mskip1.5mu])\to (\Varid{s}\to \Conid{Free}\;\Varid{f}\;[\mskip1.5mu \Varid{a}\mskip1.5mu]){}\<[E]%
\\
\>[3]{}\Varid{fwd}_{\Varid{LHS}}\;\Varid{op}\mathrel{=}\lambda \Varid{s}\to \Conid{Op}\;(\Varid{fmap}\;(\mathbin{\$}\Varid{s})\;\Varid{op}){}\<[E]%
\ColumnHook
\end{hscode}\resethooks
\indentend We prove it by the following calculation for input \ensuremath{\Varid{op}\mathbin{::}\Varid{f}\;(\Conid{Free}\;(\Varid{Modify_{F}}\;\Varid{s}\;\Varid{r}\mathrel{{:}{+}{:}}\Varid{Nondet_{F}}\mathrel{{:}{+}{:}}\Varid{f})\;\Varid{a})}.
In the corresponding case of \Cref{app:modify-fusing-lhs}, we have
calculated that \ensuremath{\Varid{h_{GlobalM}}\;(\Conid{Op}\;(\Conid{Inr}\;(\Conid{Inr}\;\Varid{op})))\mathrel{=}\lambda \Varid{s}\to \Conid{Op}\;(\Varid{fmap}\;(\mathbin{\$}\Varid{s})\;(\Varid{fmap}\;\Varid{h_{GlobalM}}\;\Varid{op}))} \refs{}.
\indentbegin \begin{hscode}\SaveRestoreHook
\column{B}{@{}>{\hspre}l<{\hspost}@{}}%
\column{3}{@{}>{\hspre}l<{\hspost}@{}}%
\column{5}{@{}>{\hspre}l<{\hspost}@{}}%
\column{7}{@{}>{\hspre}l<{\hspost}@{}}%
\column{E}{@{}>{\hspre}l<{\hspost}@{}}%
\>[5]{}\Varid{fmap}\;\Varid{runStack}\mathbin{\$}\Varid{h_{GlobalM}}\;(\Varid{fwd}\;\Varid{op}){}\<[E]%
\\
\>[3]{}\mathrel{=}\mbox{\commentbegin ~ definition of \ensuremath{\Varid{fwd}}  \commentend}{}\<[E]%
\\
\>[3]{}\hsindent{2}{}\<[5]%
\>[5]{}\Varid{fmap}\;\Varid{runStack}\mathbin{\$}\Varid{h_{GlobalM}}\;(\Conid{Op}\hsdot{\circ }{.}\Conid{Inr}\hsdot{\circ }{.}\Conid{Inr}\hsdot{\circ }{.}\Conid{Inr}\mathbin{\$}\Varid{op}){}\<[E]%
\\
\>[3]{}\mathrel{=}\mbox{\commentbegin ~ Equation \refs{}  \commentend}{}\<[E]%
\\
\>[3]{}\hsindent{2}{}\<[5]%
\>[5]{}\Varid{fmap}\;\Varid{runStack}\mathbin{\$}\lambda \Varid{s}\to \Conid{Op}\;(\Varid{fmap}\;(\mathbin{\$}\Varid{s})\;(\Varid{fmap}\;\Varid{h_{GlobalM}}\;(\Conid{Inr}\;\Varid{op}))){}\<[E]%
\\
\>[3]{}\mathrel{=}\mbox{\commentbegin ~ fmap fusion  \commentend}{}\<[E]%
\\
\>[3]{}\hsindent{2}{}\<[5]%
\>[5]{}\Varid{fmap}\;\Varid{runStack}\mathbin{\$}\lambda \Varid{s}\to \Conid{Op}\;(\Varid{fmap}\;((\mathbin{\$}\Varid{s})\hsdot{\circ }{.}\Varid{h_{GlobalM}})\;(\Conid{Inr}\;\Varid{op})){}\<[E]%
\\
\>[3]{}\mathrel{=}\mbox{\commentbegin ~ reformulation  \commentend}{}\<[E]%
\\
\>[3]{}\hsindent{2}{}\<[5]%
\>[5]{}\Varid{fmap}\;\Varid{runStack}\mathbin{\$}\lambda \Varid{s}\to \Conid{Op}\;(\Varid{fmap}\;(\lambda \Varid{x}\to \Varid{h_{GlobalM}}\;\Varid{x}\;\Varid{s})\;(\Conid{Inr}\;\Varid{op})){}\<[E]%
\\
\>[3]{}\mathrel{=}\mbox{\commentbegin ~ definition of \ensuremath{\Varid{fmap}}  \commentend}{}\<[E]%
\\
\>[3]{}\hsindent{2}{}\<[5]%
\>[5]{}\lambda \Varid{s}\to \Varid{runStack}\mathbin{\$}\Conid{Op}\;(\Varid{fmap}\;(\lambda \Varid{x}\to \Varid{h_{GlobalM}}\;\Varid{x}\;\Varid{s})\;(\Conid{Inr}\;\Varid{op})){}\<[E]%
\\
\>[3]{}\mathrel{=}\mbox{\commentbegin ~ definition of \ensuremath{\Varid{runStack}}  \commentend}{}\<[E]%
\\
\>[3]{}\hsindent{2}{}\<[5]%
\>[5]{}\lambda \Varid{s}\to \Varid{fmap}\;\Varid{fst}\hsdot{\circ }{.}\Varid{flip}\;\Varid{h_{State1}}\;(\Conid{Stack}\;[\mskip1.5mu \mskip1.5mu])\mathbin{\$}{}\<[E]%
\\
\>[5]{}\hsindent{2}{}\<[7]%
\>[7]{}\Conid{Op}\;(\Varid{fmap}\;(\lambda \Varid{x}\to \Varid{h_{GlobalM}}\;\Varid{x}\;\Varid{s})\;(\Conid{Inr}\;\Varid{op})){}\<[E]%
\\
\>[3]{}\mathrel{=}\mbox{\commentbegin ~ definition of \ensuremath{\Varid{h_{State1}}}  \commentend}{}\<[E]%
\\
\>[3]{}\hsindent{2}{}\<[5]%
\>[5]{}\lambda \Varid{s}\to \Varid{fmap}\;\Varid{fst}\mathbin{\$}(\lambda \Varid{t}\to {}\<[E]%
\\
\>[5]{}\hsindent{2}{}\<[7]%
\>[7]{}\Conid{Op}\;(\Varid{fmap}\;(\mathbin{\$}\Varid{t})\hsdot{\circ }{.}\Varid{fmap}\;(\Varid{h_{State1}})\mathbin{\$}\Varid{fmap}\;(\lambda \Varid{x}\to \Varid{h_{GlobalM}}\;\Varid{x}\;\Varid{s})\;\Varid{op}))\;(\Conid{Stack}\;[\mskip1.5mu \mskip1.5mu]){}\<[E]%
\\
\>[3]{}\mathrel{=}\mbox{\commentbegin ~ fmap fusion and reformulation  \commentend}{}\<[E]%
\\
\>[3]{}\hsindent{2}{}\<[5]%
\>[5]{}\lambda \Varid{s}\to \Varid{fmap}\;\Varid{fst}\mathbin{\$}(\lambda \Varid{t}\to {}\<[E]%
\\
\>[5]{}\hsindent{2}{}\<[7]%
\>[7]{}\Conid{Op}\;(\Varid{fmap}\;(\lambda \Varid{x}\to \Varid{h_{State1}}\;(\Varid{h_{GlobalM}}\;\Varid{x}\;\Varid{s})\;\Varid{t})\;\Varid{op}))\;(\Conid{Stack}\;[\mskip1.5mu \mskip1.5mu]){}\<[E]%
\\
\>[3]{}\mathrel{=}\mbox{\commentbegin ~ function application  \commentend}{}\<[E]%
\\
\>[3]{}\hsindent{2}{}\<[5]%
\>[5]{}\lambda \Varid{s}\to \Varid{fmap}\;\Varid{fst}\mathbin{\$}{}\<[E]%
\\
\>[5]{}\hsindent{2}{}\<[7]%
\>[7]{}\Conid{Op}\;(\Varid{fmap}\;(\lambda \Varid{x}\to \Varid{h_{State1}}\;(\Varid{h_{GlobalM}}\;\Varid{x}\;\Varid{s})\;(\Conid{Stack}\;[\mskip1.5mu \mskip1.5mu]))\;\Varid{op}){}\<[E]%
\\
\>[3]{}\mathrel{=}\mbox{\commentbegin ~ definition of \ensuremath{\Varid{fmap}}  \commentend}{}\<[E]%
\\
\>[3]{}\hsindent{2}{}\<[5]%
\>[5]{}\lambda \Varid{s}\to \Conid{Op}\;(\Varid{fmap}\;(\lambda \Varid{x}\to \Varid{fmap}\;\Varid{fst}\;(\Varid{h_{State1}}\;(\Varid{h_{GlobalM}}\;\Varid{x}\;\Varid{s})\;(\Conid{Stack}\;[\mskip1.5mu \mskip1.5mu])))\;\Varid{op}){}\<[E]%
\\
\>[3]{}\mathrel{=}\mbox{\commentbegin ~ reformulation  \commentend}{}\<[E]%
\\
\>[3]{}\hsindent{2}{}\<[5]%
\>[5]{}\lambda \Varid{s}\to \Conid{Op}\;(\Varid{fmap}\;(\lambda \Varid{x}\to \Varid{fmap}\;(\Varid{fmap}\;\Varid{fst}\hsdot{\circ }{.}\Varid{flip}\;\Varid{h_{State1}}\;(\Conid{Stack}\;[\mskip1.5mu \mskip1.5mu]))\hsdot{\circ }{.}\Varid{h_{GlobalM}}\mathbin{\$}\Varid{x}\;\Varid{s})\;\Varid{op}){}\<[E]%
\\
\>[3]{}\mathrel{=}\mbox{\commentbegin ~ reformulation  \commentend}{}\<[E]%
\\
\>[3]{}\hsindent{2}{}\<[5]%
\>[5]{}\lambda \Varid{s}\to \Conid{Op}\;(\Varid{fmap}\;(\lambda \Varid{x}\to (\Varid{fmap}\;\Varid{runStack}\hsdot{\circ }{.}\Varid{h_{GlobalM}}\mathbin{\$}\Varid{x})\;\Varid{s})\;\Varid{op}){}\<[E]%
\\
\>[3]{}\mathrel{=}\mbox{\commentbegin ~ fmap fission  \commentend}{}\<[E]%
\\
\>[3]{}\hsindent{2}{}\<[5]%
\>[5]{}\lambda \Varid{s}\to \Conid{Op}\;(\Varid{fmap}\;(\mathbin{\$}\Varid{s})\;(\Varid{fmap}\;(\Varid{fmap}\;\Varid{runStack}\hsdot{\circ }{.}\Varid{h_{GlobalM}})\;\Varid{op})){}\<[E]%
\\
\>[3]{}\mathrel{=}\mbox{\commentbegin ~ definition of \ensuremath{\Varid{fwd}_{\Varid{LHS}}}   \commentend}{}\<[E]%
\\
\>[3]{}\hsindent{2}{}\<[5]%
\>[5]{}\Varid{fwd}_{\Varid{LHS}}\;(\Varid{fmap}\;(\Varid{fmap}\;\Varid{runStack}\hsdot{\circ }{.}\Varid{h_{GlobalM}})\;\Varid{op}){}\<[E]%
\ColumnHook
\end{hscode}\resethooks
\indentend 
\subsection{Equating the Fused Sides}

We observe that the following equations hold trivially.
\begin{eqnarray*}
\ensuremath{\Varid{gen}_{\Varid{LHS}}} & = & \ensuremath{\Varid{gen}_{\Varid{RHS}}} \\
\ensuremath{\Varid{alg}_{\Varid{LHS}}^{\Varid{S}}} & = & \ensuremath{\Varid{alg}_{\Varid{RHS}}^{\Varid{S}}} \\
\ensuremath{\Varid{alg}_{\Varid{LHS}}^{\Varid{ND}}} & = & \ensuremath{\Varid{alg}_{\Varid{RHS}}^{\Varid{ND}}} \\
\ensuremath{\Varid{fwd}_{\Varid{LHS}}} & = & \ensuremath{\Varid{fwd}_{\Varid{RHS}}}
\end{eqnarray*}

Therefore, the main theorem (\Cref{thm:trail-local-global}) holds.

\subsection{Lemmas}

In this section, we prove the lemmas used in \Cref{sec:trail-fusing-lhs}.

The following lemma shows the relationship between the state and trail
stack. Intuitively, the trail stack contains all the deltas (updates)
that have not been restored in the program. Previous elements in the
trail stack do not influence the result and state of programs.
\begin{lemma}[Trail stack tracks state]~ \label{lemma:trail-stack-tracks-state}
For \ensuremath{\Varid{t}\mathbin{::}\Conid{Stack}\;(\Conid{Either}\;\Varid{r}\;())}, \ensuremath{\Varid{s}\mathbin{::}\Varid{s}}, and \ensuremath{\Varid{p}\mathbin{::}\Conid{Free}\;(\Varid{Modify_{F}}\;\Varid{s}\;\Varid{r}\mathrel{{:}{+}{:}}\Varid{Nondet_{F}}\mathrel{{:}{+}{:}}\Varid{f})\;\Varid{a}} which does not use the \ensuremath{\Varid{restore}} operation, we
have
\indentbegin \begin{hscode}\SaveRestoreHook
\column{B}{@{}>{\hspre}l<{\hspost}@{}}%
\column{3}{@{}>{\hspre}c<{\hspost}@{}}%
\column{3E}{@{}l@{}}%
\column{6}{@{}>{\hspre}l<{\hspost}@{}}%
\column{10}{@{}>{\hspre}l<{\hspost}@{}}%
\column{E}{@{}>{\hspre}l<{\hspost}@{}}%
\>[6]{}\Varid{h_{State1}}\;((\Varid{h_{Modify1}}\hsdot{\circ }{.}\Varid{h_{ND+f}}\hsdot{\circ }{.}\Varid{(\Leftrightarrow)})\;(\Varid{local2trail}\;\Varid{p})\;\Varid{s})\;\Varid{t}{}\<[E]%
\\
\>[3]{}\mathrel{=}{}\<[3E]%
\\
\>[3]{}\hsindent{3}{}\<[6]%
\>[6]{}\mathbf{do}\;{}\<[10]%
\>[10]{}((\Varid{x},\anonymous ),\anonymous )\leftarrow \Varid{h_{State1}}\;((\Varid{h_{Modify1}}\hsdot{\circ }{.}\Varid{h_{ND+f}}\hsdot{\circ }{.}\Varid{(\Leftrightarrow)})\;(\Varid{local2trail}\;\Varid{p})\;\Varid{s})\;(\Conid{Stack}\;[\mskip1.5mu \mskip1.5mu]){}\<[E]%
\\
\>[10]{}\Varid{\eta}\;((\Varid{x},\Varid{fplus}\;\Varid{s}\;\Varid{ys}),\Varid{extend}\;\Varid{t}\;\Varid{ys}){}\<[E]%
\ColumnHook
\end{hscode}\resethooks
\indentend %
for some \ensuremath{\Varid{ys}\mathrel{=}[\mskip1.5mu \Conid{Left}\;\Varid{r}_{\Varid{n}},\mathbin{...},\Conid{Left}\;\Varid{r\char95 1}\mskip1.5mu]}. The functions \ensuremath{\Varid{extend}} and
\ensuremath{\Varid{fplus}} are defined as follows:
\indentbegin \begin{hscode}\SaveRestoreHook
\column{B}{@{}>{\hspre}l<{\hspost}@{}}%
\column{E}{@{}>{\hspre}l<{\hspost}@{}}%
\>[B]{}\Varid{extend}\mathbin{::}\Conid{Stack}\;\Varid{s}\to [\mskip1.5mu \Varid{s}\mskip1.5mu]\to \Conid{Stack}\;\Varid{s}{}\<[E]%
\\
\>[B]{}\Varid{extend}\;(\Conid{Stack}\;\Varid{xs})\;\Varid{ys}\mathrel{=}\Conid{Stack}\;(\Varid{ys}+\!\!+\Varid{xs}){}\<[E]%
\\
\>[B]{}\Varid{fplus}\mathbin{::}\Conid{Undo}\;\Varid{s}\;\Varid{r}\Rightarrow \Varid{s}\to [\mskip1.5mu \Conid{Either}\;\Varid{r}\;\Varid{b}\mskip1.5mu]\to \Varid{s}{}\<[E]%
\\
\>[B]{}\Varid{fplus}\;\Varid{s}\;\Varid{ys}\mathrel{=}\Varid{foldr}\;(\lambda (\Conid{Left}\;\Varid{r})\;\Varid{s}\to \Varid{s}\mathbin{\oplus}\Varid{r})\;\Varid{s}\;\Varid{ys}{}\<[E]%
\ColumnHook
\end{hscode}\resethooks
\indentend \end{lemma}

Note that an immediate corollary of
\Cref{lemma:trail-stack-tracks-state} is that in addition to replacing
the stack \ensuremath{\Varid{t}} with the empty stack \ensuremath{\Conid{Stack}\;[\mskip1.5mu \mskip1.5mu]}, we can also replace it
with any other stack. The following equation holds.
\indentbegin \begin{hscode}\SaveRestoreHook
\column{B}{@{}>{\hspre}l<{\hspost}@{}}%
\column{3}{@{}>{\hspre}c<{\hspost}@{}}%
\column{3E}{@{}l@{}}%
\column{6}{@{}>{\hspre}l<{\hspost}@{}}%
\column{10}{@{}>{\hspre}l<{\hspost}@{}}%
\column{E}{@{}>{\hspre}l<{\hspost}@{}}%
\>[6]{}\Varid{h_{State1}}\;((\Varid{h_{Modify1}}\hsdot{\circ }{.}\Varid{h_{ND+f}}\hsdot{\circ }{.}\Varid{(\Leftrightarrow)})\;(\Varid{local2trail}\;\Varid{p})\;\Varid{s})\;\Varid{t}{}\<[E]%
\\
\>[3]{}\mathrel{=}{}\<[3E]%
\\
\>[3]{}\hsindent{3}{}\<[6]%
\>[6]{}\mathbf{do}\;{}\<[10]%
\>[10]{}((\Varid{x},\anonymous ),\anonymous )\leftarrow \Varid{h_{State1}}\;((\Varid{h_{Modify1}}\hsdot{\circ }{.}\Varid{h_{ND+f}}\hsdot{\circ }{.}\Varid{(\Leftrightarrow)})\;(\Varid{local2trail}\;\Varid{p})\;\Varid{s})\;\Varid{t'}{}\<[E]%
\\
\>[10]{}\Varid{\eta}\;((\Varid{x},\Varid{fplus}\;\Varid{s}\;\Varid{ys}),\Varid{extend}\;\Varid{t}\;\Varid{ys}){}\<[E]%
\ColumnHook
\end{hscode}\resethooks
\indentend %
We will also use this corollary in the proofs.

\begin{proof}

We proceed by induction on \ensuremath{\Varid{p}}.

\noindent \mbox{\underline{case \ensuremath{\Varid{p}\mathrel{=}\Conid{Var}\;\Varid{y}}}}
\indentbegin \begin{hscode}\SaveRestoreHook
\column{B}{@{}>{\hspre}l<{\hspost}@{}}%
\column{3}{@{}>{\hspre}l<{\hspost}@{}}%
\column{4}{@{}>{\hspre}l<{\hspost}@{}}%
\column{5}{@{}>{\hspre}l<{\hspost}@{}}%
\column{8}{@{}>{\hspre}l<{\hspost}@{}}%
\column{E}{@{}>{\hspre}l<{\hspost}@{}}%
\>[5]{}\Varid{h_{State1}}\;((\Varid{h_{Modify1}}\hsdot{\circ }{.}\Varid{h_{ND+f}}\hsdot{\circ }{.}\Varid{(\Leftrightarrow)})\;(\Varid{local2trail}\;(\Conid{Var}\;\Varid{y}))\;\Varid{s})\;\Varid{t}{}\<[E]%
\\
\>[3]{}\mathrel{=}\mbox{\commentbegin ~ definition of \ensuremath{\Varid{local2trail}}  \commentend}{}\<[E]%
\\
\>[3]{}\hsindent{2}{}\<[5]%
\>[5]{}\Varid{h_{State1}}\;((\Varid{h_{Modify1}}\hsdot{\circ }{.}\Varid{h_{ND+f}}\hsdot{\circ }{.}\Varid{(\Leftrightarrow)})\;(\Conid{Var}\;\Varid{y})\;\Varid{s})\;\Varid{t}{}\<[E]%
\\
\>[3]{}\mathrel{=}\mbox{\commentbegin ~ definition of \ensuremath{\Varid{(\Leftrightarrow)}}  \commentend}{}\<[E]%
\\
\>[3]{}\hsindent{2}{}\<[5]%
\>[5]{}\Varid{h_{State1}}\;((\Varid{h_{Modify1}}\hsdot{\circ }{.}\Varid{h_{ND+f}})\;(\Conid{Var}\;\Varid{y})\;\Varid{s})\;\Varid{t}{}\<[E]%
\\
\>[3]{}\mathrel{=}\mbox{\commentbegin ~ definition of \ensuremath{\Varid{h_{ND+f}}}  \commentend}{}\<[E]%
\\
\>[3]{}\hsindent{2}{}\<[5]%
\>[5]{}\Varid{h_{State1}}\;(\Varid{h_{Modify1}}\;(\Conid{Var}\;[\mskip1.5mu \Varid{y}\mskip1.5mu])\;\Varid{s})\;\Varid{t}{}\<[E]%
\\
\>[3]{}\mathrel{=}\mbox{\commentbegin ~ definition of \ensuremath{\Varid{h_{Modify1}}} and functiona application  \commentend}{}\<[E]%
\\
\>[3]{}\hsindent{2}{}\<[5]%
\>[5]{}\Varid{h_{State1}}\;(\Conid{Var}\;([\mskip1.5mu \Varid{y}\mskip1.5mu],\Varid{s}))\;\Varid{t}{}\<[E]%
\\
\>[3]{}\mathrel{=}\mbox{\commentbegin ~ definition of \ensuremath{\Varid{h_{State1}}} and functiona application  \commentend}{}\<[E]%
\\
\>[3]{}\hsindent{1}{}\<[4]%
\>[4]{}\Conid{Var}\;(([\mskip1.5mu \Varid{y}\mskip1.5mu],\Varid{s}),\Varid{t}){}\<[E]%
\\
\>[3]{}\mathrel{=}\mbox{\commentbegin ~ similar derivation in reverse   \commentend}{}\<[E]%
\\
\>[3]{}\hsindent{1}{}\<[4]%
\>[4]{}\mathbf{do}\;{}\<[8]%
\>[8]{}((\Varid{x},\anonymous ),\anonymous )\leftarrow \Varid{h_{State1}}\;((\Varid{h_{Modify1}}\hsdot{\circ }{.}\Varid{h_{ND+f}}\hsdot{\circ }{.}\Varid{(\Leftrightarrow)})\;(\Varid{local2trail}\;(\Conid{Var}\;\Varid{y}))\;\Varid{s})\;(\Conid{Stack}\;[\mskip1.5mu \mskip1.5mu]){}\<[E]%
\\
\>[8]{}\Conid{Var}\;((\Varid{x},\Varid{s}),\Varid{t}){}\<[E]%
\ColumnHook
\end{hscode}\resethooks
\indentend 
\noindent \mbox{\underline{case \ensuremath{\Varid{t}\mathrel{=}\Conid{Op}\hsdot{\circ }{.}\Conid{Inr}\hsdot{\circ }{.}\Conid{Inl}\mathbin{\$}\Conid{Fail}}}}
\indentbegin \begin{hscode}\SaveRestoreHook
\column{B}{@{}>{\hspre}l<{\hspost}@{}}%
\column{3}{@{}>{\hspre}l<{\hspost}@{}}%
\column{4}{@{}>{\hspre}l<{\hspost}@{}}%
\column{5}{@{}>{\hspre}l<{\hspost}@{}}%
\column{8}{@{}>{\hspre}l<{\hspost}@{}}%
\column{25}{@{}>{\hspre}l<{\hspost}@{}}%
\column{E}{@{}>{\hspre}l<{\hspost}@{}}%
\>[5]{}\Varid{h_{State1}}\;((\Varid{h_{Modify1}}\hsdot{\circ }{.}\Varid{h_{ND+f}}\hsdot{\circ }{.}\Varid{(\Leftrightarrow)})\;(\Varid{local2trail}\;(\Conid{Op}\hsdot{\circ }{.}\Conid{Inr}\hsdot{\circ }{.}\Conid{Inl}\mathbin{\$}\Conid{Fail}))\;\Varid{s})\;\Varid{t}{}\<[E]%
\\
\>[3]{}\mathrel{=}\mbox{\commentbegin ~ definition of \ensuremath{\Varid{local2trail}}  \commentend}{}\<[E]%
\\
\>[3]{}\hsindent{2}{}\<[5]%
\>[5]{}\Varid{h_{State1}}\;((\Varid{h_{Modify1}}\hsdot{\circ }{.}\Varid{h_{ND+f}}\hsdot{\circ }{.}\Varid{(\Leftrightarrow)})\;(\Conid{Op}\hsdot{\circ }{.}\Conid{Inr}\hsdot{\circ }{.}\Conid{Inl}\mathbin{\$}\Conid{Fail})\;\Varid{s})\;\Varid{t}{}\<[E]%
\\
\>[3]{}\mathrel{=}\mbox{\commentbegin ~ definition of \ensuremath{\Varid{(\Leftrightarrow)}} and \ensuremath{\Varid{h_{ND+f}}}  \commentend}{}\<[E]%
\\
\>[3]{}\hsindent{2}{}\<[5]%
\>[5]{}\Varid{h_{State1}}\;(\Varid{h_{Modify1}}\;(\Conid{Var}\;[\mskip1.5mu \mskip1.5mu])\;\Varid{s})\;\Varid{t}{}\<[E]%
\\
\>[3]{}\mathrel{=}\mbox{\commentbegin ~ definition of \ensuremath{\Varid{h_{Modify1}}} and function application  \commentend}{}\<[E]%
\\
\>[3]{}\hsindent{2}{}\<[5]%
\>[5]{}\Varid{h_{State1}}\;(\Conid{Var}\;([\mskip1.5mu \mskip1.5mu],\Varid{s}))\;\Varid{t}{}\<[E]%
\\
\>[3]{}\mathrel{=}\mbox{\commentbegin ~ definition of \ensuremath{\Varid{h_{State1}}} and functiona application  \commentend}{}\<[E]%
\\
\>[3]{}\hsindent{1}{}\<[4]%
\>[4]{}\Conid{Var}\;(([\mskip1.5mu \mskip1.5mu],\Varid{s}),\Varid{t}){}\<[E]%
\\
\>[3]{}\mathrel{=}\mbox{\commentbegin ~ similar derivation in reverse   \commentend}{}\<[E]%
\\
\>[3]{}\hsindent{1}{}\<[4]%
\>[4]{}\mathbf{do}\;{}\<[8]%
\>[8]{}((\Varid{x},\anonymous ),\anonymous )\leftarrow \Varid{h_{State1}}\;((\Varid{h_{Modify1}}\hsdot{\circ }{.}\Varid{h_{ND+f}}\hsdot{\circ }{.}\Varid{(\Leftrightarrow)}){}\<[E]%
\\
\>[8]{}\hsindent{17}{}\<[25]%
\>[25]{}(\Varid{local2trail}\;(\Conid{Op}\hsdot{\circ }{.}\Conid{Inr}\hsdot{\circ }{.}\Conid{Inl}\mathbin{\$}\Conid{Fail}))\;\Varid{s})\;(\Conid{Stack}\;[\mskip1.5mu \mskip1.5mu]){}\<[E]%
\\
\>[8]{}\Conid{Var}\;((\Varid{x},\Varid{s}),\Varid{t}){}\<[E]%
\ColumnHook
\end{hscode}\resethooks
\indentend 
\noindent \mbox{\underline{case \ensuremath{\Varid{t}\mathrel{=}\Conid{Op}\;(\Conid{Inl}\;(\Conid{MGet}\;\Varid{k}))}}}
\indentbegin \begin{hscode}\SaveRestoreHook
\column{B}{@{}>{\hspre}l<{\hspost}@{}}%
\column{3}{@{}>{\hspre}l<{\hspost}@{}}%
\column{5}{@{}>{\hspre}l<{\hspost}@{}}%
\column{9}{@{}>{\hspre}l<{\hspost}@{}}%
\column{26}{@{}>{\hspre}l<{\hspost}@{}}%
\column{E}{@{}>{\hspre}l<{\hspost}@{}}%
\>[5]{}\Varid{h_{State1}}\;((\Varid{h_{Modify1}}\hsdot{\circ }{.}\Varid{h_{ND+f}}\hsdot{\circ }{.}\Varid{(\Leftrightarrow)})\;(\Varid{local2trail}\;(\Conid{Op}\;(\Conid{Inl}\;(\Conid{MGet}\;\Varid{k}))))\;\Varid{s})\;\Varid{t}{}\<[E]%
\\
\>[3]{}\mathrel{=}\mbox{\commentbegin ~ definition of \ensuremath{\Varid{local2trail}}  \commentend}{}\<[E]%
\\
\>[3]{}\hsindent{2}{}\<[5]%
\>[5]{}\Varid{h_{State1}}\;((\Varid{h_{Modify1}}\hsdot{\circ }{.}\Varid{h_{ND+f}}\hsdot{\circ }{.}\Varid{(\Leftrightarrow)})\;(\Conid{Op}\;(\Conid{Inl}\;(\Conid{MGet}\;(\Varid{local2trail}\hsdot{\circ }{.}\Varid{k}))))\;\Varid{s})\;\Varid{t}{}\<[E]%
\\
\>[3]{}\mathrel{=}\mbox{\commentbegin ~ definition of \ensuremath{\Varid{(\Leftrightarrow)}} and \ensuremath{\Varid{h_{ND+f}}}  \commentend}{}\<[E]%
\\
\>[3]{}\hsindent{2}{}\<[5]%
\>[5]{}\Varid{h_{State1}}\;(\Varid{h_{Modify1}}\;(\Conid{Op}\;(\Conid{Inl}\;(\Conid{MGet}\;(\Varid{h_{ND+f}}\hsdot{\circ }{.}\Varid{(\Leftrightarrow)}\hsdot{\circ }{.}\Varid{local2trail}\hsdot{\circ }{.}\Varid{k}))))\;\Varid{s})\;\Varid{t}{}\<[E]%
\\
\>[3]{}\mathrel{=}\mbox{\commentbegin ~ definition of \ensuremath{\Varid{h_{Modify1}}} and function application  \commentend}{}\<[E]%
\\
\>[3]{}\hsindent{2}{}\<[5]%
\>[5]{}\Varid{h_{State1}}\;((\Varid{h_{Modify1}}\hsdot{\circ }{.}\Varid{h_{ND+f}}\hsdot{\circ }{.}\Varid{(\Leftrightarrow)}\hsdot{\circ }{.}\Varid{local2trail}\hsdot{\circ }{.}\Varid{k})\;\Varid{s}\;\Varid{s})\;\Varid{t}{}\<[E]%
\\
\>[3]{}\mathrel{=}\mbox{\commentbegin ~ reformulation  \commentend}{}\<[E]%
\\
\>[3]{}\hsindent{2}{}\<[5]%
\>[5]{}\Varid{h_{State1}}\;((\Varid{h_{Modify1}}\hsdot{\circ }{.}\Varid{h_{ND+f}}\hsdot{\circ }{.}\Varid{(\Leftrightarrow)})\;(\Varid{local2trail}\;(\Varid{k}\;\Varid{s}))\;\Varid{s})\;\Varid{t}{}\<[E]%
\\
\>[3]{}\mathrel{=}\mbox{\commentbegin ~ induction hypothesis on \ensuremath{\Varid{k}\;\Varid{s}}  \commentend}{}\<[E]%
\\
\>[3]{}\hsindent{2}{}\<[5]%
\>[5]{}\mathbf{do}\;{}\<[9]%
\>[9]{}((\Varid{x},\anonymous ),\anonymous )\leftarrow \Varid{h_{State1}}\;((\Varid{h_{Modify1}}\hsdot{\circ }{.}\Varid{h_{ND+f}}\hsdot{\circ }{.}\Varid{(\Leftrightarrow)})\;(\Varid{local2trail}\;(\Varid{k}\;\Varid{s}))\;\Varid{s})\;(\Conid{Stack}\;[\mskip1.5mu \mskip1.5mu]){}\<[E]%
\\
\>[9]{}\Varid{\eta}\;((\Varid{x},\Varid{fplus}\;\Varid{s}\;\Varid{ys}),\Varid{extend}\;\Varid{t}\;\Varid{ys}){}\<[E]%
\\
\>[3]{}\mathrel{=}\mbox{\commentbegin ~ similar derivation in reverse  \commentend}{}\<[E]%
\\
\>[3]{}\hsindent{2}{}\<[5]%
\>[5]{}\mathbf{do}\;{}\<[9]%
\>[9]{}((\Varid{x},\anonymous ),\anonymous )\leftarrow \Varid{h_{State1}}\;((\Varid{h_{Modify1}}\hsdot{\circ }{.}\Varid{h_{ND+f}}\hsdot{\circ }{.}\Varid{(\Leftrightarrow)}){}\<[E]%
\\
\>[9]{}\hsindent{17}{}\<[26]%
\>[26]{}(\Varid{local2trail}\;(\Conid{Op}\;(\Conid{Inl}\;(\Conid{MGet}\;\Varid{k}))))\;\Varid{s})\;(\Conid{Stack}\;[\mskip1.5mu \mskip1.5mu]){}\<[E]%
\\
\>[9]{}\Varid{\eta}\;((\Varid{x},\Varid{fplus}\;\Varid{s}\;\Varid{ys}),\Varid{extend}\;\Varid{t}\;\Varid{ys}){}\<[E]%
\ColumnHook
\end{hscode}\resethooks
\indentend 
\noindent \mbox{\underline{case \ensuremath{\Varid{t}\mathrel{=}\Conid{Op}\;(\Conid{Inl}\;(\Conid{MUpdate}\;\Varid{r}\;\Varid{k}))}}}
\indentbegin \begin{hscode}\SaveRestoreHook
\column{B}{@{}>{\hspre}l<{\hspost}@{}}%
\column{3}{@{}>{\hspre}l<{\hspost}@{}}%
\column{5}{@{}>{\hspre}l<{\hspost}@{}}%
\column{7}{@{}>{\hspre}l<{\hspost}@{}}%
\column{9}{@{}>{\hspre}l<{\hspost}@{}}%
\column{11}{@{}>{\hspre}l<{\hspost}@{}}%
\column{24}{@{}>{\hspre}l<{\hspost}@{}}%
\column{E}{@{}>{\hspre}l<{\hspost}@{}}%
\>[5]{}\Varid{h_{State1}}\;((\Varid{h_{Modify1}}\hsdot{\circ }{.}\Varid{h_{ND+f}}\hsdot{\circ }{.}\Varid{(\Leftrightarrow)})\;(\Varid{local2trail}\;(\Conid{Op}\;(\Conid{Inl}\;(\Conid{MUpdate}\;\Varid{r}\;\Varid{k}))))\;\Varid{s})\;\Varid{t}{}\<[E]%
\\
\>[3]{}\mathrel{=}\mbox{\commentbegin ~ definition of \ensuremath{\Varid{local2trail}}  \commentend}{}\<[E]%
\\
\>[3]{}\hsindent{2}{}\<[5]%
\>[5]{}\Varid{h_{State1}}\;((\Varid{h_{Modify1}}\hsdot{\circ }{.}\Varid{h_{ND+f}}\hsdot{\circ }{.}\Varid{(\Leftrightarrow)})\;(\mathbf{do}{}\<[E]%
\\
\>[5]{}\hsindent{2}{}\<[7]%
\>[7]{}\Varid{pushStack}\;(\Conid{Left}\;\Varid{r}){}\<[E]%
\\
\>[5]{}\hsindent{2}{}\<[7]%
\>[7]{}\Varid{update}\;\Varid{r}{}\<[E]%
\\
\>[5]{}\hsindent{2}{}\<[7]%
\>[7]{}\Varid{local2trail}\;\Varid{k}{}\<[E]%
\\
\>[3]{}\hsindent{2}{}\<[5]%
\>[5]{})\Varid{s})\;\Varid{t}{}\<[E]%
\\
\>[3]{}\mathrel{=}\mbox{\commentbegin ~ definition of \ensuremath{\Varid{(\Leftrightarrow)}} and \ensuremath{\Varid{h_{ND+f}}}  \commentend}{}\<[E]%
\\
\>[3]{}\hsindent{2}{}\<[5]%
\>[5]{}\Varid{h_{State1}}\;(\Varid{h_{Modify1}}\;(\mathbf{do}{}\<[E]%
\\
\>[5]{}\hsindent{2}{}\<[7]%
\>[7]{}\Varid{h_{ND+f}}\hsdot{\circ }{.}\Varid{(\Leftrightarrow)}\mathbin{\$}\Varid{pushStack}\;(\Conid{Left}\;\Varid{r}){}\<[E]%
\\
\>[5]{}\hsindent{2}{}\<[7]%
\>[7]{}\Varid{h_{ND+f}}\hsdot{\circ }{.}\Varid{(\Leftrightarrow)}\mathbin{\$}\Varid{update}\;\Varid{r}{}\<[E]%
\\
\>[5]{}\hsindent{2}{}\<[7]%
\>[7]{}\Varid{h_{ND+f}}\hsdot{\circ }{.}\Varid{(\Leftrightarrow)}\hsdot{\circ }{.}\Varid{local2trail}\mathbin{\$}\Varid{k}{}\<[E]%
\\
\>[3]{}\hsindent{2}{}\<[5]%
\>[5]{})\Varid{s})\;\Varid{t}{}\<[E]%
\\
\>[3]{}\mathrel{=}\mbox{\commentbegin ~ definition of \ensuremath{\Varid{h_{Modify1}}}, \Cref{lemma:dist-hModify1} and function application  \commentend}{}\<[E]%
\\
\>[3]{}\hsindent{2}{}\<[5]%
\>[5]{}\Varid{h_{State1}}\;(\mathbf{do}{}\<[E]%
\\
\>[5]{}\hsindent{2}{}\<[7]%
\>[7]{}(\anonymous ,\Varid{s}_{1})\leftarrow (\Varid{h_{Modify1}}\hsdot{\circ }{.}\Varid{h_{ND+f}}\hsdot{\circ }{.}\Varid{(\Leftrightarrow)}\mathbin{\$}\Varid{pushStack}\;(\Conid{Left}\;\Varid{r}))\;\Varid{s}{}\<[E]%
\\
\>[5]{}\hsindent{2}{}\<[7]%
\>[7]{}(\anonymous ,\Varid{s}_{2})\leftarrow (\Varid{h_{Modify1}}\hsdot{\circ }{.}\Varid{h_{ND+f}}\hsdot{\circ }{.}\Varid{(\Leftrightarrow)}\mathbin{\$}\Varid{update}\;\Varid{r})\;\Varid{s}_{1}{}\<[E]%
\\
\>[5]{}\hsindent{2}{}\<[7]%
\>[7]{}(\Varid{h_{Modify1}}\hsdot{\circ }{.}\Varid{h_{ND+f}}\hsdot{\circ }{.}\Varid{(\Leftrightarrow)}\hsdot{\circ }{.}\Varid{local2trail}\mathbin{\$}\Varid{k})\;\Varid{s}_{2}{}\<[E]%
\\
\>[3]{}\hsindent{2}{}\<[5]%
\>[5]{})\;\Varid{t}{}\<[E]%
\\
\>[3]{}\mathrel{=}\mbox{\commentbegin ~ definition of \ensuremath{\Varid{h_{Modify1}}} and \ensuremath{\Varid{update}}  \commentend}{}\<[E]%
\\
\>[3]{}\hsindent{2}{}\<[5]%
\>[5]{}\Varid{h_{State1}}\;(\mathbf{do}{}\<[E]%
\\
\>[5]{}\hsindent{2}{}\<[7]%
\>[7]{}(\anonymous ,\Varid{s}_{1})\leftarrow (\Varid{h_{Modify1}}\hsdot{\circ }{.}\Varid{h_{ND+f}}\hsdot{\circ }{.}\Varid{(\Leftrightarrow)}\mathbin{\$}\Varid{pushStack}\;(\Conid{Left}\;\Varid{r}))\;\Varid{s}{}\<[E]%
\\
\>[5]{}\hsindent{2}{}\<[7]%
\>[7]{}(\anonymous ,\Varid{s}_{2})\leftarrow \Varid{\eta}\;([\mskip1.5mu ()\mskip1.5mu],\Varid{s}_{1}\mathbin{\oplus}\Varid{r}){}\<[E]%
\\
\>[5]{}\hsindent{2}{}\<[7]%
\>[7]{}(\Varid{h_{Modify1}}\hsdot{\circ }{.}\Varid{h_{ND+f}}\hsdot{\circ }{.}\Varid{(\Leftrightarrow)}\hsdot{\circ }{.}\Varid{local2trail}\mathbin{\$}\Varid{k})\;\Varid{s}_{2}{}\<[E]%
\\
\>[3]{}\hsindent{2}{}\<[5]%
\>[5]{})\;\Varid{t}{}\<[E]%
\\
\>[3]{}\mathrel{=}\mbox{\commentbegin ~ monad law  \commentend}{}\<[E]%
\\
\>[3]{}\hsindent{2}{}\<[5]%
\>[5]{}\Varid{h_{State1}}\;(\mathbf{do}{}\<[E]%
\\
\>[5]{}\hsindent{2}{}\<[7]%
\>[7]{}(\anonymous ,\Varid{s}_{1})\leftarrow (\Varid{h_{Modify1}}\hsdot{\circ }{.}\Varid{h_{ND+f}}\hsdot{\circ }{.}\Varid{(\Leftrightarrow)}\mathbin{\$}\Varid{pushStack}\;(\Conid{Left}\;\Varid{r}))\;\Varid{s}{}\<[E]%
\\
\>[5]{}\hsindent{2}{}\<[7]%
\>[7]{}(\Varid{h_{Modify1}}\hsdot{\circ }{.}\Varid{h_{ND+f}}\hsdot{\circ }{.}\Varid{(\Leftrightarrow)}\hsdot{\circ }{.}\Varid{local2trail}\mathbin{\$}\Varid{k})\;(\Varid{s}_{1}\mathbin{\oplus}\Varid{r}){}\<[E]%
\\
\>[3]{}\hsindent{2}{}\<[5]%
\>[5]{})\;\Varid{t}{}\<[E]%
\\
\>[3]{}\mathrel{=}\mbox{\commentbegin ~ definition of \ensuremath{\Varid{h_{State1}}}, \Cref{lemma:dist-hState1}, and function application  \commentend}{}\<[E]%
\\
\>[3]{}\hsindent{2}{}\<[5]%
\>[5]{}\mathbf{do}\;{}\<[9]%
\>[9]{}((\anonymous ,\Varid{s}_{1}),\Varid{t}_{1})\leftarrow \Varid{h_{State1}}\;((\Varid{h_{Modify1}}\hsdot{\circ }{.}\Varid{h_{ND+f}}\hsdot{\circ }{.}\Varid{(\Leftrightarrow)}\mathbin{\$}\Varid{pushStack}\;(\Conid{Left}\;\Varid{r}))\;\Varid{s})\;\Varid{t}{}\<[E]%
\\
\>[9]{}\Varid{h_{State1}}\;((\Varid{h_{Modify1}}\hsdot{\circ }{.}\Varid{h_{ND+f}}\hsdot{\circ }{.}\Varid{(\Leftrightarrow)}\hsdot{\circ }{.}\Varid{local2trail}\mathbin{\$}\Varid{k})\;(\Varid{s}_{1}\mathbin{\oplus}\Varid{r}))\;\Varid{t}_{1}{}\<[E]%
\\
\>[3]{}\mathrel{=}\mbox{\commentbegin ~ definition of \ensuremath{\Varid{pushStack}}  \commentend}{}\<[E]%
\\
\>[3]{}\hsindent{2}{}\<[5]%
\>[5]{}\mathbf{do}\;{}\<[9]%
\>[9]{}((\anonymous ,\Varid{s}_{1}),\Varid{t}_{1})\leftarrow \Varid{h_{State1}}\;((\Varid{h_{Modify1}}\hsdot{\circ }{.}\Varid{h_{ND+f}}\hsdot{\circ }{.}\Varid{(\Leftrightarrow)}\mathbin{\$}\mathbf{do}{}\<[E]%
\\
\>[9]{}\hsindent{2}{}\<[11]%
\>[11]{}\Conid{Stack}\;\Varid{xs}\leftarrow \Varid{get}{}\<[E]%
\\
\>[9]{}\hsindent{2}{}\<[11]%
\>[11]{}\Varid{put}\;(\Conid{Stack}\;(\Conid{Left}\;\Varid{r}\mathbin{:}\Varid{xs})))\;\Varid{s})\;\Varid{t}{}\<[E]%
\\
\>[9]{}\Varid{h_{State1}}\;((\Varid{h_{Modify1}}\hsdot{\circ }{.}\Varid{h_{ND+f}}\hsdot{\circ }{.}\Varid{(\Leftrightarrow)}\hsdot{\circ }{.}\Varid{local2trail}\mathbin{\$}\Varid{k})\;(\Varid{s}_{1}\mathbin{\oplus}\Varid{r}))\;\Varid{t}_{1}{}\<[E]%
\\
\>[3]{}\mathrel{=}\mbox{\commentbegin ~ definition of \ensuremath{\Varid{h_{State1}}}, \ensuremath{\Varid{h_{Modify1}}}, \ensuremath{\Varid{h_{ND+f}}}, \ensuremath{\Varid{(\Leftrightarrow)}}, \ensuremath{\Varid{get}}, and \ensuremath{\Varid{put}}; let \ensuremath{\Varid{t}\mathrel{=}\Conid{Stack}\;\Varid{xs}}  \commentend}{}\<[E]%
\\
\>[3]{}\hsindent{2}{}\<[5]%
\>[5]{}\mathbf{do}\;{}\<[9]%
\>[9]{}((\anonymous ,\Varid{s}_{1}),\Varid{t}_{1})\leftarrow \Varid{\eta}\;(([\mskip1.5mu ()\mskip1.5mu],\Varid{s}),\Conid{Stack}\;(\Conid{Left}\;\Varid{r}\mathbin{:}\Varid{xs})){}\<[E]%
\\
\>[9]{}\Varid{h_{State1}}\;((\Varid{h_{Modify1}}\hsdot{\circ }{.}\Varid{h_{ND+f}}\hsdot{\circ }{.}\Varid{(\Leftrightarrow)})\;(\Varid{local2trail}\;\Varid{k})\;(\Varid{s}_{1}\mathbin{\oplus}\Varid{r}))\;\Varid{t}_{1}{}\<[E]%
\\
\>[3]{}\mathrel{=}\mbox{\commentbegin ~ monad law  \commentend}{}\<[E]%
\\
\>[3]{}\hsindent{2}{}\<[5]%
\>[5]{}\Varid{h_{State1}}\;((\Varid{h_{Modify1}}\hsdot{\circ }{.}\Varid{h_{ND+f}}\hsdot{\circ }{.}\Varid{(\Leftrightarrow)})\;(\Varid{local2trail}\;\Varid{k})\;(\Varid{s}\mathbin{\oplus}\Varid{r}))\;(\Conid{Stack}\;(\Conid{Left}\;\Varid{r}\mathbin{:}\Varid{xs})){}\<[E]%
\\
\>[3]{}\mathrel{=}\mbox{\commentbegin ~ by induction hypothesis on \ensuremath{\Varid{k}}, for some \ensuremath{\Varid{ys}} the equation holds  \commentend}{}\<[E]%
\\
\>[3]{}\hsindent{2}{}\<[5]%
\>[5]{}\mathbf{do}\;{}\<[9]%
\>[9]{}((\Varid{x},\anonymous ),\anonymous )\leftarrow \Varid{h_{State1}}\;((\Varid{h_{Modify1}}\hsdot{\circ }{.}\Varid{h_{ND+f}}\hsdot{\circ }{.}\Varid{(\Leftrightarrow)}){}\<[E]%
\\
\>[9]{}\hsindent{15}{}\<[24]%
\>[24]{}(\Varid{local2trail}\;\Varid{k})\;(\Varid{s}\mathbin{\oplus}\Varid{r}))\;(\Conid{Stack}\;[\mskip1.5mu \Conid{Left}\;\Varid{r}\mskip1.5mu]){}\<[E]%
\\
\>[9]{}\Varid{\eta}\;((\Varid{x},\Varid{fplus}\;(\Varid{s}\mathbin{\oplus}\Varid{r})\;\Varid{ys}),\Varid{extend}\;(\Conid{Stack}\;(\Conid{Left}\;\Varid{r}\mathbin{:}\Varid{xs}))\;\Varid{ys}){}\<[E]%
\\
\>[3]{}\mathrel{=}\mbox{\commentbegin ~ definition of \ensuremath{\Varid{fplus}} and \ensuremath{\Varid{extend}}  \commentend}{}\<[E]%
\\
\>[3]{}\hsindent{2}{}\<[5]%
\>[5]{}\mathbf{do}\;{}\<[9]%
\>[9]{}((\Varid{x},\anonymous ),\anonymous )\leftarrow \Varid{h_{State1}}\;((\Varid{h_{Modify1}}\hsdot{\circ }{.}\Varid{h_{ND+f}}\hsdot{\circ }{.}\Varid{(\Leftrightarrow)}){}\<[E]%
\\
\>[9]{}\hsindent{15}{}\<[24]%
\>[24]{}(\Varid{local2trail}\;\Varid{k})\;(\Varid{s}\mathbin{\oplus}\Varid{r}))\;(\Conid{Stack}\;[\mskip1.5mu \Conid{Left}\;\Varid{r}\mskip1.5mu]){}\<[E]%
\\
\>[9]{}\Varid{\eta}\;((\Varid{x},\Varid{fplus}\;\Varid{s}\;(\Varid{ys}+\!\!+[\mskip1.5mu \Conid{Left}\;\Varid{r}\mskip1.5mu])),\Varid{extend}\;(\Conid{Stack}\;\Varid{xs})\;(\Varid{ys}+\!\!+[\mskip1.5mu \Conid{Left}\;\Varid{r}\mskip1.5mu])){}\<[E]%
\\
\>[3]{}\mathrel{=}\mbox{\commentbegin ~ let \ensuremath{\Varid{ys'}\mathrel{=}\Varid{ys}+\!\!+[\mskip1.5mu \Conid{Left}\;\Varid{r}\mskip1.5mu]}; Equation \ensuremath{\Varid{t}\mathrel{=}\Conid{Stack}\;\Varid{xs}}  \commentend}{}\<[E]%
\\
\>[3]{}\hsindent{2}{}\<[5]%
\>[5]{}\mathbf{do}\;{}\<[9]%
\>[9]{}((\Varid{x},\anonymous ),\anonymous )\leftarrow \Varid{h_{State1}}\;((\Varid{h_{Modify1}}\hsdot{\circ }{.}\Varid{h_{ND+f}}\hsdot{\circ }{.}\Varid{(\Leftrightarrow)}){}\<[E]%
\\
\>[9]{}\hsindent{15}{}\<[24]%
\>[24]{}(\Varid{local2trail}\;\Varid{k})\;(\Varid{s}\mathbin{\oplus}\Varid{r}))\;(\Conid{Stack}\;[\mskip1.5mu \Conid{Left}\;\Varid{r}\mskip1.5mu]){}\<[E]%
\\
\>[9]{}\Varid{\eta}\;((\Varid{x},\Varid{fplus}\;\Varid{s}\;\Varid{ys'}),\Varid{extend}\;\Varid{t}\;\Varid{ys'}){}\<[E]%
\\
\>[3]{}\mathrel{=}\mbox{\commentbegin ~ similar derivation in reverse  \commentend}{}\<[E]%
\\
\>[3]{}\hsindent{2}{}\<[5]%
\>[5]{}\mathbf{do}\;{}\<[9]%
\>[9]{}((\Varid{x},\anonymous ),\anonymous )\leftarrow \Varid{h_{State1}}\;((\Varid{h_{Modify1}}\hsdot{\circ }{.}\Varid{h_{ND+f}}\hsdot{\circ }{.}\Varid{(\Leftrightarrow)}){}\<[E]%
\\
\>[9]{}\hsindent{15}{}\<[24]%
\>[24]{}(\Varid{local2trail}\;(\Conid{Op}\;(\Conid{Inl}\;(\Conid{MUpdate}\;\Varid{r}\;\Varid{k}))))\;\Varid{s})\;(\Conid{Stack}\;[\mskip1.5mu \mskip1.5mu]){}\<[E]%
\\
\>[9]{}\Varid{\eta}\;((\Varid{x},\Varid{fplus}\;\Varid{s}\;\Varid{ys'}),\Varid{extend}\;\Varid{t}\;\Varid{ys'}){}\<[E]%
\ColumnHook
\end{hscode}\resethooks
\indentend 
\noindent \mbox{\underline{case \ensuremath{\Varid{t}\mathrel{=}\Conid{Op}\hsdot{\circ }{.}\Conid{Inr}\hsdot{\circ }{.}\Conid{Inl}\mathbin{\$}\Conid{Or}\;\Varid{p}\;\Varid{q}}}}
\indentbegin \begin{hscode}\SaveRestoreHook
\column{B}{@{}>{\hspre}l<{\hspost}@{}}%
\column{3}{@{}>{\hspre}l<{\hspost}@{}}%
\column{5}{@{}>{\hspre}l<{\hspost}@{}}%
\column{7}{@{}>{\hspre}l<{\hspost}@{}}%
\column{9}{@{}>{\hspre}l<{\hspost}@{}}%
\column{11}{@{}>{\hspre}l<{\hspost}@{}}%
\column{15}{@{}>{\hspre}l<{\hspost}@{}}%
\column{16}{@{}>{\hspre}l<{\hspost}@{}}%
\column{20}{@{}>{\hspre}l<{\hspost}@{}}%
\column{24}{@{}>{\hspre}l<{\hspost}@{}}%
\column{28}{@{}>{\hspre}l<{\hspost}@{}}%
\column{30}{@{}>{\hspre}l<{\hspost}@{}}%
\column{E}{@{}>{\hspre}l<{\hspost}@{}}%
\>[5]{}\Varid{h_{State1}}\;((\Varid{h_{Modify1}}\hsdot{\circ }{.}\Varid{h_{ND+f}}\hsdot{\circ }{.}\Varid{(\Leftrightarrow)})\;(\Varid{local2trail}\;(\Conid{Op}\hsdot{\circ }{.}\Conid{Inr}\hsdot{\circ }{.}\Conid{Inl}\mathbin{\$}\Conid{Or}\;\Varid{p}\;\Varid{q}))\;\Varid{s})\;\Varid{t}{}\<[E]%
\\
\>[3]{}\mathrel{=}\mbox{\commentbegin ~ definition of \ensuremath{\Varid{local2trail}}  \commentend}{}\<[E]%
\\
\>[3]{}\hsindent{2}{}\<[5]%
\>[5]{}\Varid{h_{State1}}\;((\Varid{h_{Modify1}}\hsdot{\circ }{.}\Varid{h_{ND+f}}\hsdot{\circ }{.}\Varid{(\Leftrightarrow)})\;({}\<[E]%
\\
\>[5]{}\hsindent{2}{}\<[7]%
\>[7]{}(\Varid{pushStack}\;(\Conid{Right}\;())>\!\!>\Varid{local2trail}\;\Varid{p})\mathbin{\talloblong}(\Varid{untrail}>\!\!>\Varid{local2trail}\;\Varid{q}))\;\Varid{s})\;\Varid{t}{}\<[E]%
\\
\>[3]{}\mathrel{=}\mbox{\commentbegin ~ definition of \ensuremath{(\talloblong)}  \commentend}{}\<[E]%
\\
\>[3]{}\hsindent{2}{}\<[5]%
\>[5]{}\Varid{h_{State1}}\;((\Varid{h_{Modify1}}\hsdot{\circ }{.}\Varid{h_{ND+f}}\hsdot{\circ }{.}\Varid{(\Leftrightarrow)})\;(\Conid{Op}\hsdot{\circ }{.}\Conid{Inr}\hsdot{\circ }{.}\Conid{Inl}\mathbin{\$}\Conid{Or}{}\<[E]%
\\
\>[5]{}\hsindent{2}{}\<[7]%
\>[7]{}(\Varid{pushStack}\;(\Conid{Right}\;())>\!\!>\Varid{local2trail}\;\Varid{p}){}\<[E]%
\\
\>[5]{}\hsindent{2}{}\<[7]%
\>[7]{}(\Varid{untrail}>\!\!>\Varid{local2trail}\;\Varid{q}))\;\Varid{s})\;\Varid{t}{}\<[E]%
\\
\>[3]{}\mathrel{=}\mbox{\commentbegin ~ definition of \ensuremath{\Varid{h_{ND+f}}}, \ensuremath{\Varid{(\Leftrightarrow)}}, and \ensuremath{\Varid{liftM2}}  \commentend}{}\<[E]%
\\
\>[3]{}\hsindent{2}{}\<[5]%
\>[5]{}\Varid{h_{State1}}\;(\Varid{h_{Modify1}}\;(\mathbf{do}{}\<[E]%
\\
\>[5]{}\hsindent{2}{}\<[7]%
\>[7]{}\Varid{x}\leftarrow \Varid{h_{ND+f}}\;(\Varid{(\Leftrightarrow)}\;(\Varid{pushStack}\;(\Conid{Right}\;())))>\!\!>\Varid{h_{ND+f}}\;(\Varid{(\Leftrightarrow)}\;(\Varid{local2trail}\;\Varid{p})){}\<[E]%
\\
\>[5]{}\hsindent{2}{}\<[7]%
\>[7]{}\Varid{y}\leftarrow \Varid{h_{ND+f}}\;(\Varid{(\Leftrightarrow)}\;\Varid{untrail})>\!\!>\Varid{h_{ND+f}}\;(\Varid{(\Leftrightarrow)}\;(\Varid{local2trail}\;\Varid{q})){}\<[E]%
\\
\>[5]{}\hsindent{2}{}\<[7]%
\>[7]{}\Varid{\eta}\;(\Varid{x}+\!\!+\Varid{y}){}\<[E]%
\\
\>[3]{}\hsindent{2}{}\<[5]%
\>[5]{})\Varid{s})\;\Varid{t}{}\<[E]%
\\
\>[3]{}\mathrel{=}\mbox{\commentbegin ~ monad law  \commentend}{}\<[E]%
\\
\>[3]{}\hsindent{2}{}\<[5]%
\>[5]{}\Varid{h_{State1}}\;(\Varid{h_{Modify1}}\;(\mathbf{do}{}\<[E]%
\\
\>[5]{}\hsindent{2}{}\<[7]%
\>[7]{}\Varid{h_{ND+f}}\;(\Varid{(\Leftrightarrow)}\;(\Varid{pushStack}\;(\Conid{Right}\;()))){}\<[E]%
\\
\>[5]{}\hsindent{2}{}\<[7]%
\>[7]{}\Varid{x}\leftarrow \Varid{h_{ND+f}}\;(\Varid{(\Leftrightarrow)}\;(\Varid{local2trail}\;\Varid{p})){}\<[E]%
\\
\>[5]{}\hsindent{2}{}\<[7]%
\>[7]{}\Varid{h_{ND+f}}\;(\Varid{(\Leftrightarrow)}\;\Varid{untrail}){}\<[E]%
\\
\>[5]{}\hsindent{2}{}\<[7]%
\>[7]{}\Varid{y}\leftarrow \Varid{h_{ND+f}}\;(\Varid{(\Leftrightarrow)}\;(\Varid{local2trail}\;\Varid{q})){}\<[E]%
\\
\>[5]{}\hsindent{2}{}\<[7]%
\>[7]{}\Varid{\eta}\;(\Varid{x}+\!\!+\Varid{y}){}\<[E]%
\\
\>[3]{}\hsindent{2}{}\<[5]%
\>[5]{})\Varid{s})\;\Varid{t}{}\<[E]%
\\
\>[3]{}\mathrel{=}\mbox{\commentbegin ~ definition of \ensuremath{\Varid{h_{Modify1}}}, \Cref{lemma:dist-hModify1}, and function application  \commentend}{}\<[E]%
\\
\>[3]{}\hsindent{2}{}\<[5]%
\>[5]{}\Varid{h_{State1}}\;(\mathbf{do}{}\<[E]%
\\
\>[5]{}\hsindent{2}{}\<[7]%
\>[7]{}(\anonymous ,\Varid{s}_{1}){}\<[16]%
\>[16]{}\leftarrow \Varid{h_{Modify1}}\;(\Varid{h_{ND+f}}\;(\Varid{(\Leftrightarrow)}\;(\Varid{pushStack}\;(\Conid{Right}\;()))))\;\Varid{s}{}\<[E]%
\\
\>[5]{}\hsindent{2}{}\<[7]%
\>[7]{}(\Varid{x},\Varid{s}_{2}){}\<[16]%
\>[16]{}\leftarrow \Varid{h_{Modify1}}\;(\Varid{h_{ND+f}}\;(\Varid{(\Leftrightarrow)}\;(\Varid{local2trail}\;\Varid{p})))\;\Varid{s}_{1}{}\<[E]%
\\
\>[5]{}\hsindent{2}{}\<[7]%
\>[7]{}(\anonymous ,\Varid{s3}){}\<[16]%
\>[16]{}\leftarrow \Varid{h_{Modify1}}\;(\Varid{h_{ND+f}}\;(\Varid{(\Leftrightarrow)}\;\Varid{untrail}))\;\Varid{s}_{2}{}\<[E]%
\\
\>[5]{}\hsindent{2}{}\<[7]%
\>[7]{}(\Varid{y},\Varid{s4}){}\<[16]%
\>[16]{}\leftarrow \Varid{h_{Modify1}}\;(\Varid{h_{ND+f}}\;(\Varid{(\Leftrightarrow)}\;(\Varid{local2trail}\;\Varid{q})))\;\Varid{s3}{}\<[E]%
\\
\>[5]{}\hsindent{2}{}\<[7]%
\>[7]{}\Varid{\eta}\;(\Varid{x}+\!\!+\Varid{y},\Varid{s4}){}\<[E]%
\\
\>[3]{}\hsindent{2}{}\<[5]%
\>[5]{})\;\Varid{t}{}\<[E]%
\\
\>[3]{}\mathrel{=}\mbox{\commentbegin ~ definition of \ensuremath{\Varid{h_{State1}}}, \Cref{lemma:dist-hState1}, and function application  \commentend}{}\<[E]%
\\
\>[3]{}\hsindent{2}{}\<[5]%
\>[5]{}\mathbf{do}\;{}\<[9]%
\>[9]{}((\anonymous ,\Varid{s}_{1}),\Varid{t}_{1}){}\<[24]%
\>[24]{}\leftarrow \Varid{h_{State1}}\;(\Varid{h_{Modify1}}\;(\Varid{h_{ND+f}}\;(\Varid{(\Leftrightarrow)}\;(\Varid{pushStack}\;(\Conid{Right}\;()))))\;\Varid{s})\;\Varid{t}{}\<[E]%
\\
\>[9]{}((\Varid{x},\Varid{s}_{2}),\Varid{t}_{2}){}\<[24]%
\>[24]{}\leftarrow \Varid{h_{State1}}\;(\Varid{h_{Modify1}}\;(\Varid{h_{ND+f}}\;(\Varid{(\Leftrightarrow)}\;(\Varid{local2trail}\;\Varid{p})))\;\Varid{s}_{1})\;\Varid{t}_{1}{}\<[E]%
\\
\>[9]{}((\anonymous ,\Varid{s3}),\Varid{t3}){}\<[24]%
\>[24]{}\leftarrow \Varid{h_{State1}}\;(\Varid{h_{Modify1}}\;(\Varid{h_{ND+f}}\;(\Varid{(\Leftrightarrow)}\;\Varid{untrail}))\;\Varid{s}_{2})\;\Varid{t}_{2}{}\<[E]%
\\
\>[9]{}((\Varid{y},\Varid{s4}),\Varid{t4}){}\<[24]%
\>[24]{}\leftarrow \Varid{h_{State1}}\;(\Varid{h_{Modify1}}\;(\Varid{h_{ND+f}}\;(\Varid{(\Leftrightarrow)}\;(\Varid{local2trail}\;\Varid{q})))\;\Varid{s3})\;\Varid{t3}{}\<[E]%
\\
\>[9]{}\Varid{\eta}\;((\Varid{x}+\!\!+\Varid{y},\Varid{s4}),\Varid{t4}){}\<[E]%
\\
\>[3]{}\mathrel{=}\mbox{\commentbegin ~ definition of \ensuremath{\Varid{pushStack}}  \commentend}{}\<[E]%
\\
\>[3]{}\hsindent{2}{}\<[5]%
\>[5]{}\mathbf{do}\;{}\<[9]%
\>[9]{}((\anonymous ,\Varid{s}_{1}),\Varid{t}_{1}){}\<[24]%
\>[24]{}\leftarrow \Varid{h_{State1}}\;(\Varid{h_{Modify1}}\;(\Varid{h_{ND+f}}\;(\Varid{(\Leftrightarrow)}\;(\mathbf{do}{}\<[E]%
\\
\>[9]{}\hsindent{2}{}\<[11]%
\>[11]{}\Conid{Stack}\;\Varid{xs}\leftarrow \Varid{get}{}\<[E]%
\\
\>[9]{}\hsindent{2}{}\<[11]%
\>[11]{}\Varid{put}\;(\Conid{Stack}\;(\Conid{Right}\;()\mathbin{:}\Varid{xs})))))\;\Varid{s})\;\Varid{t}{}\<[E]%
\\
\>[9]{}((\Varid{x},\Varid{s}_{2}),\Varid{t}_{2}){}\<[24]%
\>[24]{}\leftarrow \Varid{h_{State1}}\;(\Varid{h_{Modify1}}\;(\Varid{h_{ND+f}}\;(\Varid{(\Leftrightarrow)}\;(\Varid{local2trail}\;\Varid{p})))\;\Varid{s}_{1})\;\Varid{t}_{1}{}\<[E]%
\\
\>[9]{}((\anonymous ,\Varid{s3}),\Varid{t3}){}\<[24]%
\>[24]{}\leftarrow \Varid{h_{State1}}\;(\Varid{h_{Modify1}}\;(\Varid{h_{ND+f}}\;(\Varid{(\Leftrightarrow)}\;\Varid{untrail}))\;\Varid{s}_{2})\;\Varid{t}_{2}{}\<[E]%
\\
\>[9]{}((\Varid{y},\Varid{s4}),\Varid{t4}){}\<[24]%
\>[24]{}\leftarrow \Varid{h_{State1}}\;(\Varid{h_{Modify1}}\;(\Varid{h_{ND+f}}\;(\Varid{(\Leftrightarrow)}\;(\Varid{local2trail}\;\Varid{q})))\;\Varid{s3})\;\Varid{t3}{}\<[E]%
\\
\>[9]{}\Varid{\eta}\;((\Varid{x}+\!\!+\Varid{y},\Varid{s4}),\Varid{t4}){}\<[E]%
\\
\>[3]{}\mathrel{=}\mbox{\commentbegin ~ definition of \ensuremath{\Varid{h_{State1}},\Varid{h_{Modify1}},\Varid{h_{ND+f}},\Varid{(\Leftrightarrow)},\Varid{get},\Varid{put}}; let \ensuremath{\Varid{t}\mathrel{=}\Conid{Stack}\;\Varid{xs}}  \commentend}{}\<[E]%
\\
\>[3]{}\hsindent{2}{}\<[5]%
\>[5]{}\mathbf{do}\;{}\<[9]%
\>[9]{}((\anonymous ,\Varid{s}_{1}),\Varid{t}_{1}){}\<[24]%
\>[24]{}\leftarrow \Varid{\eta}\;(([\mskip1.5mu ()\mskip1.5mu],\Varid{s}),\Conid{Stack}\;(\Conid{Right}\;()\mathbin{:}\Varid{xs})){}\<[E]%
\\
\>[9]{}((\Varid{x},\Varid{s}_{2}),\Varid{t}_{2}){}\<[24]%
\>[24]{}\leftarrow \Varid{h_{State1}}\;(\Varid{h_{Modify1}}\;(\Varid{h_{ND+f}}\;(\Varid{(\Leftrightarrow)}\;(\Varid{local2trail}\;\Varid{p})))\;\Varid{s}_{1})\;\Varid{t}_{1}{}\<[E]%
\\
\>[9]{}((\anonymous ,\Varid{s3}),\Varid{t3}){}\<[24]%
\>[24]{}\leftarrow \Varid{h_{State1}}\;(\Varid{h_{Modify1}}\;(\Varid{h_{ND+f}}\;(\Varid{(\Leftrightarrow)}\;\Varid{untrail}))\;\Varid{s}_{2})\;\Varid{t}_{2}{}\<[E]%
\\
\>[9]{}((\Varid{y},\Varid{s4}),\Varid{t4}){}\<[24]%
\>[24]{}\leftarrow \Varid{h_{State1}}\;(\Varid{h_{Modify1}}\;(\Varid{h_{ND+f}}\;(\Varid{(\Leftrightarrow)}\;(\Varid{local2trail}\;\Varid{q})))\;\Varid{s3})\;\Varid{t3}{}\<[E]%
\\
\>[9]{}\Varid{\eta}\;((\Varid{x}+\!\!+\Varid{y},\Varid{s4}),\Varid{t4}){}\<[E]%
\\
\>[3]{}\mathrel{=}\mbox{\commentbegin ~ monad law  \commentend}{}\<[E]%
\\
\>[3]{}\hsindent{2}{}\<[5]%
\>[5]{}\mathbf{do}\;{}\<[9]%
\>[9]{}((\Varid{x},\Varid{s}_{2}),\Varid{t}_{2}){}\<[24]%
\>[24]{}\leftarrow {}\<[28]%
\>[28]{}\Varid{h_{State1}}\;(\Varid{h_{Modify1}}\;(\Varid{h_{ND+f}}\;(\Varid{(\Leftrightarrow)}{}\<[E]%
\\
\>[28]{}\hsindent{2}{}\<[30]%
\>[30]{}(\Varid{local2trail}\;\Varid{p})))\;\Varid{s})\;(\Conid{Stack}\;(\Conid{Right}\;()\mathbin{:}\Varid{xs})){}\<[E]%
\\
\>[9]{}((\anonymous ,\Varid{s3}),\Varid{t3}){}\<[24]%
\>[24]{}\leftarrow {}\<[28]%
\>[28]{}\Varid{h_{State1}}\;(\Varid{h_{Modify1}}\;(\Varid{h_{ND+f}}\;(\Varid{(\Leftrightarrow)}\;\Varid{untrail}))\;\Varid{s}_{2})\;\Varid{t}_{2}{}\<[E]%
\\
\>[9]{}((\Varid{y},\Varid{s4}),\Varid{t4}){}\<[24]%
\>[24]{}\leftarrow {}\<[28]%
\>[28]{}\Varid{h_{State1}}\;(\Varid{h_{Modify1}}\;(\Varid{h_{ND+f}}\;(\Varid{(\Leftrightarrow)}\;(\Varid{local2trail}\;\Varid{q})))\;\Varid{s3})\;\Varid{t3}{}\<[E]%
\\
\>[9]{}\Varid{\eta}\;((\Varid{x}+\!\!+\Varid{y},\Varid{s4}),\Varid{t4}){}\<[E]%
\\
\>[3]{}\mathrel{=}\mbox{\commentbegin ~ reformulation  \commentend}{}\<[E]%
\\
\>[3]{}\hsindent{2}{}\<[5]%
\>[5]{}\mathbf{do}\;{}\<[9]%
\>[9]{}((\Varid{x},\Varid{s}_{2}),\Varid{t}_{2}){}\<[24]%
\>[24]{}\leftarrow {}\<[28]%
\>[28]{}\Varid{h_{State1}}\;((\Varid{h_{Modify1}}\hsdot{\circ }{.}\Varid{h_{ND+f}}\hsdot{\circ }{.}\Varid{(\Leftrightarrow)}){}\<[E]%
\\
\>[28]{}\hsindent{2}{}\<[30]%
\>[30]{}(\Varid{local2trail}\;\Varid{p})\;\Varid{s})\;(\Conid{Stack}\;(\Conid{Right}\;()\mathbin{:}\Varid{xs})){}\<[E]%
\\
\>[9]{}((\anonymous ,\Varid{s3}),\Varid{t3}){}\<[24]%
\>[24]{}\leftarrow {}\<[28]%
\>[28]{}\Varid{h_{State1}}\;((\Varid{h_{Modify1}}\hsdot{\circ }{.}\Varid{h_{ND+f}}\hsdot{\circ }{.}\Varid{(\Leftrightarrow)})\;\Varid{untrail}\;\Varid{s}_{2})\;\Varid{t}_{2}{}\<[E]%
\\
\>[9]{}((\Varid{y},\Varid{s4}),\Varid{t4}){}\<[24]%
\>[24]{}\leftarrow {}\<[28]%
\>[28]{}\Varid{h_{State1}}\;((\Varid{h_{Modify1}}\hsdot{\circ }{.}\Varid{h_{ND+f}}\hsdot{\circ }{.}\Varid{(\Leftrightarrow)})\;(\Varid{local2trail}\;\Varid{q})\;\Varid{s3})\;\Varid{t3}{}\<[E]%
\\
\>[9]{}\Varid{\eta}\;((\Varid{x}+\!\!+\Varid{y},\Varid{s4}),\Varid{t4}){}\<[E]%
\\
\>[3]{}\mathrel{=}\mbox{\commentbegin ~ by induction hypothesis on \ensuremath{\Varid{p}}, for some \ensuremath{\Varid{ys}} the equation holds   \commentend}{}\<[E]%
\\
\>[3]{}\hsindent{2}{}\<[5]%
\>[5]{}\mathbf{do}\;{}\<[9]%
\>[9]{}((\Varid{x},{}\<[15]%
\>[15]{}\anonymous ),{}\<[20]%
\>[20]{}\anonymous ){}\<[24]%
\>[24]{}\leftarrow {}\<[28]%
\>[28]{}\Varid{h_{State1}}\;((\Varid{h_{Modify1}}\hsdot{\circ }{.}\Varid{h_{ND+f}}\hsdot{\circ }{.}\Varid{(\Leftrightarrow)}){}\<[E]%
\\
\>[28]{}\hsindent{2}{}\<[30]%
\>[30]{}(\Varid{local2trail}\;\Varid{p})\;\Varid{s})\;(\Conid{Stack}\;(\Conid{Right}\;())){}\<[E]%
\\
\>[9]{}((\anonymous ,\Varid{s3}),\Varid{t3}){}\<[24]%
\>[24]{}\leftarrow {}\<[28]%
\>[28]{}\Varid{h_{State1}}\;((\Varid{h_{Modify1}}\hsdot{\circ }{.}\Varid{h_{ND+f}}\hsdot{\circ }{.}\Varid{(\Leftrightarrow)}){}\<[E]%
\\
\>[28]{}\hsindent{2}{}\<[30]%
\>[30]{}\Varid{untrail}\;(\Varid{fplus}\;\Varid{s}\;\Varid{ys}))\;(\Varid{extend}\;(\Conid{Stack}\;(\Conid{Right}\;()\mathbin{:}\Varid{xs}))\;\Varid{ys}){}\<[E]%
\\
\>[9]{}((\Varid{y},\Varid{s4}),\Varid{t4}){}\<[24]%
\>[24]{}\leftarrow {}\<[28]%
\>[28]{}\Varid{h_{State1}}\;((\Varid{h_{Modify1}}\hsdot{\circ }{.}\Varid{h_{ND+f}}\hsdot{\circ }{.}\Varid{(\Leftrightarrow)})\;(\Varid{local2trail}\;\Varid{q})\;\Varid{s3})\;\Varid{t3}{}\<[E]%
\\
\>[9]{}\Varid{\eta}\;((\Varid{x}+\!\!+\Varid{y},\Varid{s4}),\Varid{t4}){}\<[E]%
\\
\>[3]{}\mathrel{=}\mbox{\commentbegin ~ \Cref{lemma:undoTrail-undos}   \commentend}{}\<[E]%
\\
\>[3]{}\hsindent{2}{}\<[5]%
\>[5]{}\mathbf{do}\;{}\<[9]%
\>[9]{}((\Varid{x},{}\<[15]%
\>[15]{}\anonymous ),{}\<[20]%
\>[20]{}\anonymous ){}\<[24]%
\>[24]{}\leftarrow {}\<[28]%
\>[28]{}\Varid{h_{State1}}\;((\Varid{h_{Modify1}}\hsdot{\circ }{.}\Varid{h_{ND+f}}\hsdot{\circ }{.}\Varid{(\Leftrightarrow)}){}\<[E]%
\\
\>[28]{}\hsindent{2}{}\<[30]%
\>[30]{}(\Varid{local2trail}\;\Varid{p})\;\Varid{s})\;(\Conid{Stack}\;(\Conid{Right}\;())){}\<[E]%
\\
\>[9]{}((\Varid{y},\Varid{s4}),\Varid{t4}){}\<[24]%
\>[24]{}\leftarrow {}\<[28]%
\>[28]{}\Varid{h_{State1}}\;((\Varid{h_{Modify1}}\hsdot{\circ }{.}\Varid{h_{ND+f}}\hsdot{\circ }{.}\Varid{(\Leftrightarrow)}){}\<[E]%
\\
\>[28]{}\hsindent{2}{}\<[30]%
\>[30]{}(\Varid{local2trail}\;\Varid{q})\;(\Varid{fminus}\;(\Varid{fplus}\;\Varid{s}\;\Varid{ys})\;\Varid{ys}))\;(\Conid{Stack}\;\Varid{xs}){}\<[E]%
\\
\>[9]{}\Varid{\eta}\;((\Varid{x}+\!\!+\Varid{y},\Varid{s4}),\Varid{t4}){}\<[E]%
\\
\>[3]{}\mathrel{=}\mbox{\commentbegin ~ \Cref{eq:plus-minus} gives \ensuremath{\Varid{fminus}\;(\Varid{fplus}\;\Varid{s}\;\Varid{ys})\;\Varid{ys}\mathrel{=}\Varid{s}}   \commentend}{}\<[E]%
\\
\>[3]{}\hsindent{2}{}\<[5]%
\>[5]{}\mathbf{do}\;{}\<[9]%
\>[9]{}((\Varid{x},{}\<[15]%
\>[15]{}\anonymous ),{}\<[20]%
\>[20]{}\anonymous ){}\<[24]%
\>[24]{}\leftarrow {}\<[28]%
\>[28]{}\Varid{h_{State1}}\;((\Varid{h_{Modify1}}\hsdot{\circ }{.}\Varid{h_{ND+f}}\hsdot{\circ }{.}\Varid{(\Leftrightarrow)}){}\<[E]%
\\
\>[28]{}\hsindent{2}{}\<[30]%
\>[30]{}(\Varid{local2trail}\;\Varid{p})\;\Varid{s})\;(\Conid{Stack}\;(\Conid{Right}\;())){}\<[E]%
\\
\>[9]{}((\Varid{y},\Varid{s4}),\Varid{t4}){}\<[24]%
\>[24]{}\leftarrow {}\<[28]%
\>[28]{}\Varid{h_{State1}}\;((\Varid{h_{Modify1}}\hsdot{\circ }{.}\Varid{h_{ND+f}}\hsdot{\circ }{.}\Varid{(\Leftrightarrow)}){}\<[E]%
\\
\>[28]{}\hsindent{2}{}\<[30]%
\>[30]{}(\Varid{local2trail}\;\Varid{q})\;\Varid{s})\;(\Conid{Stack}\;\Varid{xs}){}\<[E]%
\\
\>[9]{}\Varid{\eta}\;((\Varid{x}+\!\!+\Varid{y},\Varid{s4}),\Varid{t4}){}\<[E]%
\\
\>[3]{}\mathrel{=}\mbox{\commentbegin ~ by induction hypothesis on \ensuremath{\Varid{p}}, for some \ensuremath{\Varid{ys'}} the equation holds   \commentend}{}\<[E]%
\\
\>[3]{}\hsindent{2}{}\<[5]%
\>[5]{}\mathbf{do}\;{}\<[9]%
\>[9]{}((\Varid{x},{}\<[15]%
\>[15]{}\anonymous ),{}\<[20]%
\>[20]{}\anonymous ){}\<[24]%
\>[24]{}\leftarrow {}\<[28]%
\>[28]{}\Varid{h_{State1}}\;((\Varid{h_{Modify1}}\hsdot{\circ }{.}\Varid{h_{ND+f}}\hsdot{\circ }{.}\Varid{(\Leftrightarrow)}){}\<[E]%
\\
\>[28]{}\hsindent{2}{}\<[30]%
\>[30]{}(\Varid{local2trail}\;\Varid{p})\;\Varid{s})\;(\Conid{Stack}\;(\Conid{Right}\;())){}\<[E]%
\\
\>[9]{}((\Varid{y},{}\<[15]%
\>[15]{}\anonymous ),{}\<[20]%
\>[20]{}\anonymous ){}\<[24]%
\>[24]{}\leftarrow {}\<[28]%
\>[28]{}\Varid{h_{State1}}\;((\Varid{h_{Modify1}}\hsdot{\circ }{.}\Varid{h_{ND+f}}\hsdot{\circ }{.}\Varid{(\Leftrightarrow)}){}\<[E]%
\\
\>[28]{}\hsindent{2}{}\<[30]%
\>[30]{}(\Varid{local2trail}\;\Varid{q})\;\Varid{s})\;(\Conid{Stack}\;[\mskip1.5mu \mskip1.5mu]){}\<[E]%
\\
\>[9]{}\Varid{\eta}\;((\Varid{x}+\!\!+\Varid{y},\Varid{fplus}\;\Varid{s}\;\Varid{ys'}),\Varid{extend}\;(\Conid{Stack}\;\Varid{xs})\;\Varid{ys'}){}\<[E]%
\\
\>[3]{}\mathrel{=}\mbox{\commentbegin ~ similar derivation in reverse   \commentend}{}\<[E]%
\\
\>[3]{}\hsindent{2}{}\<[5]%
\>[5]{}\mathbf{do}\;{}\<[9]%
\>[9]{}((\Varid{x},{}\<[15]%
\>[15]{}\anonymous ),{}\<[20]%
\>[20]{}\anonymous ){}\<[24]%
\>[24]{}\leftarrow {}\<[28]%
\>[28]{}\Varid{h_{State1}}\;((\Varid{h_{Modify1}}\hsdot{\circ }{.}\Varid{h_{ND+f}}\hsdot{\circ }{.}\Varid{(\Leftrightarrow)}){}\<[E]%
\\
\>[28]{}\hsindent{2}{}\<[30]%
\>[30]{}(\Varid{local2trail}\;(\Conid{Op}\hsdot{\circ }{.}\Conid{Inr}\hsdot{\circ }{.}\Conid{Inl}\mathbin{\$}\Conid{Or}\;\Varid{p}\;\Varid{q}))\;\Varid{s})\;(\Conid{Stack}\;[\mskip1.5mu \mskip1.5mu]){}\<[E]%
\\
\>[9]{}\Varid{\eta}\;((\Varid{x},\Varid{fplus}\;\Varid{s}\;\Varid{ys'}),\Varid{extend}\;\Varid{t}\;\Varid{ys'}){}\<[E]%
\ColumnHook
\end{hscode}\resethooks
\indentend 
\noindent \mbox{\underline{case \ensuremath{\Varid{t}\mathrel{=}\Conid{Op}\hsdot{\circ }{.}\Conid{Inr}\hsdot{\circ }{.}\Conid{Inr}\mathbin{\$}\Varid{y}}}}
\indentbegin \begin{hscode}\SaveRestoreHook
\column{B}{@{}>{\hspre}l<{\hspost}@{}}%
\column{3}{@{}>{\hspre}l<{\hspost}@{}}%
\column{5}{@{}>{\hspre}l<{\hspost}@{}}%
\column{7}{@{}>{\hspre}l<{\hspost}@{}}%
\column{9}{@{}>{\hspre}l<{\hspost}@{}}%
\column{11}{@{}>{\hspre}l<{\hspost}@{}}%
\column{E}{@{}>{\hspre}l<{\hspost}@{}}%
\>[5]{}\Varid{h_{State1}}\;((\Varid{h_{Modify1}}\hsdot{\circ }{.}\Varid{h_{ND+f}}\hsdot{\circ }{.}\Varid{(\Leftrightarrow)})\;(\Varid{local2trail}\;(\Conid{Op}\hsdot{\circ }{.}\Conid{Inr}\hsdot{\circ }{.}\Conid{Inr}\mathbin{\$}\Varid{y}))\;\Varid{s})\;\Varid{t}{}\<[E]%
\\
\>[3]{}\mathrel{=}\mbox{\commentbegin ~ definition of \ensuremath{\Varid{local2trail}}  \commentend}{}\<[E]%
\\
\>[3]{}\hsindent{2}{}\<[5]%
\>[5]{}\Varid{h_{State1}}\;((\Varid{h_{Modify1}}\hsdot{\circ }{.}\Varid{h_{ND+f}}\hsdot{\circ }{.}\Varid{(\Leftrightarrow)})\;(\Conid{Op}\hsdot{\circ }{.}\Conid{Inr}\hsdot{\circ }{.}\Conid{Inr}\hsdot{\circ }{.}\Conid{Inr}\hsdot{\circ }{.}\Varid{fmap}\;\Varid{local2trail}\mathbin{\$}\Varid{y})\;\Varid{s})\;\Varid{t}{}\<[E]%
\\
\>[3]{}\mathrel{=}\mbox{\commentbegin ~ definition of \ensuremath{\Varid{(\Leftrightarrow)}} and \ensuremath{\Varid{h_{ND+f}}}  \commentend}{}\<[E]%
\\
\>[3]{}\hsindent{2}{}\<[5]%
\>[5]{}\Varid{h_{State1}}\;(\Varid{h_{Modify1}}\;(\Conid{Op}\hsdot{\circ }{.}\Conid{Inr}\hsdot{\circ }{.}\Conid{Inr}\hsdot{\circ }{.}\Varid{fmap}\;(\Varid{h_{ND+f}}\hsdot{\circ }{.}\Varid{(\Leftrightarrow)}\hsdot{\circ }{.}\Varid{local2trail})\mathbin{\$}\Varid{y})\;\Varid{s})\;\Varid{t}{}\<[E]%
\\
\>[3]{}\mathrel{=}\mbox{\commentbegin ~ definition of \ensuremath{\Varid{h_{Modify1}}} and function application  \commentend}{}\<[E]%
\\
\>[3]{}\hsindent{2}{}\<[5]%
\>[5]{}\Varid{h_{State1}}\;(\Conid{Op}\hsdot{\circ }{.}\Conid{Inr}\hsdot{\circ }{.}\Varid{fmap}\;((\mathbin{\$}\Varid{s})\hsdot{\circ }{.}\Varid{h_{Modify1}}\hsdot{\circ }{.}\Varid{h_{ND+f}}\hsdot{\circ }{.}\Varid{(\Leftrightarrow)}\hsdot{\circ }{.}\Varid{local2trail})\mathbin{\$}\Varid{y})\;\Varid{t}{}\<[E]%
\\
\>[3]{}\mathrel{=}\mbox{\commentbegin ~ definition of \ensuremath{\Varid{h_{State1}}} and function application  \commentend}{}\<[E]%
\\
\>[3]{}\hsindent{2}{}\<[5]%
\>[5]{}\Conid{Op}\hsdot{\circ }{.}\Varid{fmap}\;((\mathbin{\$}\Varid{t})\hsdot{\circ }{.}\Varid{h_{State1}}\hsdot{\circ }{.}(\mathbin{\$}\Varid{s})\hsdot{\circ }{.}\Varid{h_{Modify1}}\hsdot{\circ }{.}\Varid{h_{ND+f}}\hsdot{\circ }{.}\Varid{(\Leftrightarrow)}\hsdot{\circ }{.}\Varid{local2trail})\mathbin{\$}\Varid{y}{}\<[E]%
\\
\>[3]{}\mathrel{=}\mbox{\commentbegin ~ reformulation  \commentend}{}\<[E]%
\\
\>[3]{}\hsindent{2}{}\<[5]%
\>[5]{}\Conid{Op}\hsdot{\circ }{.}\Varid{fmap}\;(\lambda \Varid{k}\to \Varid{h_{State1}}\;((\Varid{h_{Modify1}}\hsdot{\circ }{.}\Varid{h_{ND+f}}\hsdot{\circ }{.}\Varid{(\Leftrightarrow)})\;(\Varid{local2trail}\;\Varid{k})\;\Varid{s})\;\Varid{t})\mathbin{\$}\Varid{y}{}\<[E]%
\\
\>[3]{}\mathrel{=}\mbox{\commentbegin ~ by induction hypothesis on \ensuremath{\Varid{y}}, for some \ensuremath{\Varid{ys}} the equation holds  \commentend}{}\<[E]%
\\
\>[3]{}\hsindent{2}{}\<[5]%
\>[5]{}\Conid{Op}\hsdot{\circ }{.}\Varid{fmap}\;(\lambda \Varid{k}\to \mathbf{do}{}\<[E]%
\\
\>[5]{}\hsindent{2}{}\<[7]%
\>[7]{}((\Varid{x},\anonymous ),\anonymous )\leftarrow \Varid{h_{State1}}\;((\Varid{h_{Modify1}}\hsdot{\circ }{.}\Varid{h_{ND+f}}\hsdot{\circ }{.}\Varid{(\Leftrightarrow)})\;(\Varid{local2trail}\;\Varid{k})\;\Varid{s})\;\Varid{t}{}\<[E]%
\\
\>[5]{}\hsindent{2}{}\<[7]%
\>[7]{}\Varid{\eta}\;((\Varid{x},\Varid{fplus}\;\Varid{s}\;\Varid{ys}),\Varid{extend}\;\Varid{t}\;\Varid{ys}))\mathbin{\$}\Varid{y}{}\<[E]%
\\
\>[3]{}\mathrel{=}\mbox{\commentbegin ~ definition of free monad  \commentend}{}\<[E]%
\\
\>[3]{}\hsindent{2}{}\<[5]%
\>[5]{}\mathbf{do}\;{}\<[9]%
\>[9]{}((\Varid{x},\anonymous ),\anonymous )\leftarrow \Conid{Op}\hsdot{\circ }{.}\Varid{fmap}\;(\lambda \Varid{k}\to {}\<[E]%
\\
\>[9]{}\hsindent{2}{}\<[11]%
\>[11]{}\Varid{h_{State1}}\;((\Varid{h_{Modify1}}\hsdot{\circ }{.}\Varid{h_{ND+f}}\hsdot{\circ }{.}\Varid{(\Leftrightarrow)})\;(\Varid{local2trail}\;\Varid{k})\;\Varid{s})\;\Varid{t})\mathbin{\$}\Varid{y}{}\<[E]%
\\
\>[9]{}\Varid{\eta}\;((\Varid{x},\Varid{fplus}\;\Varid{s}\;\Varid{ys}),\Varid{extend}\;\Varid{t}\;\Varid{ys}){}\<[E]%
\\
\>[3]{}\mathrel{=}\mbox{\commentbegin ~ similar derivation in reverse  \commentend}{}\<[E]%
\\
\>[3]{}\hsindent{2}{}\<[5]%
\>[5]{}\mathbf{do}\;{}\<[9]%
\>[9]{}((\Varid{x},\anonymous ),\anonymous )\leftarrow \Varid{h_{State1}}\;((\Varid{h_{Modify1}}\hsdot{\circ }{.}\Varid{h_{ND+f}}\hsdot{\circ }{.}\Varid{(\Leftrightarrow)})\;(\Varid{local2trail}\;(\Conid{Op}\hsdot{\circ }{.}\Conid{Inr}\hsdot{\circ }{.}\Conid{Inr}\mathbin{\$}\Varid{y}))\;\Varid{s})\;\Varid{t}{}\<[E]%
\\
\>[9]{}\Varid{\eta}\;((\Varid{x},\Varid{fplus}\;\Varid{s}\;\Varid{ys}),\Varid{extend}\;\Varid{t}\;\Varid{ys}){}\<[E]%
\ColumnHook
\end{hscode}\resethooks
\indentend \end{proof}

The following lemma shows that the \ensuremath{\Varid{untrail}} function restores all
the updates in the trail stack until it reaches a time stamp \ensuremath{\Conid{Right}\;()}.
\begin{lemma}[UndoTrail undos]~ \label{lemma:undoTrail-undos}
For \ensuremath{\Varid{t}\mathrel{=}\Conid{Stack}\;(\Varid{ys}+\!\!+(\Conid{Right}\;()\mathbin{:}\Varid{xs}))} and \ensuremath{\Varid{ys}\mathrel{=}[\mskip1.5mu \Conid{Left}\;\Varid{r}_{1},\mathbin{...},\Conid{Left}\;\Varid{r}_{\Varid{n}}\mskip1.5mu]}, we have\indentbegin \begin{hscode}\SaveRestoreHook
\column{B}{@{}>{\hspre}l<{\hspost}@{}}%
\column{3}{@{}>{\hspre}c<{\hspost}@{}}%
\column{3E}{@{}l@{}}%
\column{6}{@{}>{\hspre}l<{\hspost}@{}}%
\column{E}{@{}>{\hspre}l<{\hspost}@{}}%
\>[6]{}\Varid{h_{State1}}\;((\Varid{h_{Modify1}}\hsdot{\circ }{.}\Varid{h_{ND+f}}\hsdot{\circ }{.}\Varid{(\Leftrightarrow)})\;\Varid{untrail}\;\Varid{s})\;\Varid{t}{}\<[E]%
\\
\>[3]{}\mathrel{=}{}\<[3E]%
\\
\>[3]{}\hsindent{3}{}\<[6]%
\>[6]{}\Varid{\eta}\;(([\mskip1.5mu ()\mskip1.5mu],\Varid{fminus}\;\Varid{s}\;\Varid{ys}),\Conid{Stack}\;\Varid{xs}){}\<[E]%
\ColumnHook
\end{hscode}\resethooks
\indentend The function \ensuremath{\Varid{fminus}} is defined as follows:
\indentbegin \begin{hscode}\SaveRestoreHook
\column{B}{@{}>{\hspre}l<{\hspost}@{}}%
\column{E}{@{}>{\hspre}l<{\hspost}@{}}%
\>[B]{}\Varid{fminus}\mathbin{::}\Conid{Undo}\;\Varid{s}\;\Varid{r}\Rightarrow \Varid{s}\to [\mskip1.5mu \Conid{Either}\;\Varid{r}\;\Varid{b}\mskip1.5mu]\to \Varid{s}{}\<[E]%
\\
\>[B]{}\Varid{fminus}\;\Varid{s}\;\Varid{ys}\mathrel{=}\Varid{foldl}\;(\lambda \Varid{s}\;(\Conid{Left}\;\Varid{r})\to \Varid{s}\mathbin{\ominus}\Varid{r})\;\Varid{s}\;\Varid{ys}{}\<[E]%
\ColumnHook
\end{hscode}\resethooks
\indentend \end{lemma}

\begin{proof}~

We first calculate as follows:
\indentbegin \begin{hscode}\SaveRestoreHook
\column{B}{@{}>{\hspre}l<{\hspost}@{}}%
\column{3}{@{}>{\hspre}l<{\hspost}@{}}%
\column{5}{@{}>{\hspre}l<{\hspost}@{}}%
\column{7}{@{}>{\hspre}l<{\hspost}@{}}%
\column{9}{@{}>{\hspre}l<{\hspost}@{}}%
\column{18}{@{}>{\hspre}l<{\hspost}@{}}%
\column{19}{@{}>{\hspre}l<{\hspost}@{}}%
\column{20}{@{}>{\hspre}l<{\hspost}@{}}%
\column{27}{@{}>{\hspre}l<{\hspost}@{}}%
\column{29}{@{}>{\hspre}l<{\hspost}@{}}%
\column{32}{@{}>{\hspre}l<{\hspost}@{}}%
\column{E}{@{}>{\hspre}l<{\hspost}@{}}%
\>[5]{}\Varid{h_{State1}}\;((\Varid{h_{Modify1}}\hsdot{\circ }{.}\Varid{h_{ND+f}}\hsdot{\circ }{.}\Varid{(\Leftrightarrow)})\;\Varid{untrail}\;\Varid{s})\;\Varid{t}{}\<[E]%
\\
\>[3]{}\mathrel{=}\mbox{\commentbegin ~ definition of \ensuremath{\Varid{untrail}}  \commentend}{}\<[E]%
\\
\>[3]{}\hsindent{2}{}\<[5]%
\>[5]{}\Varid{h_{State1}}\;((\Varid{h_{Modify1}}\hsdot{\circ }{.}\Varid{h_{ND+f}}\hsdot{\circ }{.}\Varid{(\Leftrightarrow)})\;(\mathbf{do}{}\<[E]%
\\
\>[5]{}\hsindent{2}{}\<[7]%
\>[7]{}\Varid{top}\leftarrow \Varid{popStack}{}\<[E]%
\\
\>[5]{}\hsindent{2}{}\<[7]%
\>[7]{}\mathbf{case}\;\Varid{top}\;\mathbf{of}{}\<[E]%
\\
\>[7]{}\hsindent{2}{}\<[9]%
\>[9]{}\Conid{Nothing}\to \Varid{\eta}\;(){}\<[E]%
\\
\>[7]{}\hsindent{2}{}\<[9]%
\>[9]{}\Conid{Just}\;(\Conid{Right}\;())\to \Varid{\eta}\;(){}\<[E]%
\\
\>[7]{}\hsindent{2}{}\<[9]%
\>[9]{}\Conid{Just}\;(\Conid{Left}\;\Varid{r})\to \mathbf{do}\;\Varid{restore}\;\Varid{r};\Varid{untrail}{}\<[E]%
\\
\>[3]{}\hsindent{2}{}\<[5]%
\>[5]{})\Varid{s})\;\Varid{t}{}\<[E]%
\\
\>[3]{}\mathrel{=}\mbox{\commentbegin ~ definition of \ensuremath{\Varid{popStack}}  \commentend}{}\<[E]%
\\
\>[3]{}\hsindent{2}{}\<[5]%
\>[5]{}\Varid{h_{State1}}\;((\Varid{h_{Modify1}}\hsdot{\circ }{.}\Varid{h_{ND+f}}\hsdot{\circ }{.}\Varid{(\Leftrightarrow)})\;(\mathbf{do}{}\<[E]%
\\
\>[5]{}\hsindent{2}{}\<[7]%
\>[7]{}\Varid{top}\leftarrow \mathbf{do}\;{}\<[18]%
\>[18]{}\Conid{Stack}\;\Varid{xs}\leftarrow \Varid{get}{}\<[E]%
\\
\>[18]{}\mathbf{case}\;\Varid{xs}\;\mathbf{of}{}\<[E]%
\\
\>[18]{}\hsindent{2}{}\<[20]%
\>[20]{}[\mskip1.5mu \mskip1.5mu]{}\<[29]%
\>[29]{}\to \Varid{\eta}\;\Conid{Nothing}{}\<[E]%
\\
\>[18]{}\hsindent{2}{}\<[20]%
\>[20]{}(\Varid{x}\mathbin{:}\Varid{xs'}){}\<[29]%
\>[29]{}\to \mathbf{do}\;\Varid{put}\;(\Conid{Stack}\;\Varid{xs'});\Varid{\eta}\;(\Conid{Just}\;\Varid{x}){}\<[E]%
\\
\>[5]{}\hsindent{2}{}\<[7]%
\>[7]{}\mathbf{case}\;\Varid{top}\;\mathbf{of}{}\<[E]%
\\
\>[7]{}\hsindent{2}{}\<[9]%
\>[9]{}\Conid{Nothing}\to \Varid{\eta}\;(){}\<[E]%
\\
\>[7]{}\hsindent{2}{}\<[9]%
\>[9]{}\Conid{Just}\;(\Conid{Right}\;())\to \Varid{\eta}\;(){}\<[E]%
\\
\>[7]{}\hsindent{2}{}\<[9]%
\>[9]{}\Conid{Just}\;(\Conid{Left}\;\Varid{r})\to \mathbf{do}\;\Varid{restore}\;\Varid{r};\Varid{untrail}{}\<[E]%
\\
\>[3]{}\hsindent{2}{}\<[5]%
\>[5]{})\Varid{s})\;\Varid{t}{}\<[E]%
\\
\>[3]{}\mathrel{=}\mbox{\commentbegin ~ monad law and case split  \commentend}{}\<[E]%
\\
\>[3]{}\hsindent{2}{}\<[5]%
\>[5]{}\Varid{h_{State1}}\;((\Varid{h_{Modify1}}\hsdot{\circ }{.}\Varid{h_{ND+f}}\hsdot{\circ }{.}\Varid{(\Leftrightarrow)})\;(\mathbf{do}{}\<[E]%
\\
\>[5]{}\hsindent{2}{}\<[7]%
\>[7]{}\Conid{Stack}\;\Varid{xs}\leftarrow \Varid{get}{}\<[E]%
\\
\>[5]{}\hsindent{2}{}\<[7]%
\>[7]{}\mathbf{case}\;\Varid{xs}\;\mathbf{of}{}\<[E]%
\\
\>[7]{}\hsindent{2}{}\<[9]%
\>[9]{}[\mskip1.5mu \mskip1.5mu]{}\<[27]%
\>[27]{}\to \Varid{\eta}\;(){}\<[E]%
\\
\>[7]{}\hsindent{2}{}\<[9]%
\>[9]{}(\Conid{Right}\;()\mathbin{:}\Varid{xs'}){}\<[27]%
\>[27]{}\to \mathbf{do}\;\Varid{put}\;(\Conid{Stack}\;\Varid{xs'});\Varid{\eta}\;(){}\<[E]%
\\
\>[7]{}\hsindent{2}{}\<[9]%
\>[9]{}(\Conid{Left}\;\Varid{r}{}\<[19]%
\>[19]{}\mathbin{:}\Varid{xs'}){}\<[27]%
\>[27]{}\to \mathbf{do}\;\Varid{put}\;(\Conid{Stack}\;\Varid{xs'});\Varid{restore}\;\Varid{r};\Varid{untrail}{}\<[E]%
\\
\>[3]{}\hsindent{2}{}\<[5]%
\>[5]{})\Varid{s})\;\Varid{t}{}\<[E]%
\\
\>[3]{}\mathrel{=}\mbox{\commentbegin ~ definition of \ensuremath{\Varid{get}}  \commentend}{}\<[E]%
\\
\>[3]{}\hsindent{2}{}\<[5]%
\>[5]{}\Varid{h_{State1}}\;((\Varid{h_{Modify1}}\hsdot{\circ }{.}\Varid{h_{ND+f}}\hsdot{\circ }{.}\Varid{(\Leftrightarrow)})\;(\Conid{Op}\hsdot{\circ }{.}\Conid{Inr}\hsdot{\circ }{.}\Conid{Inr}\hsdot{\circ }{.}\Conid{Inl}\hsdot{\circ }{.}\Conid{Get}\mathbin{\$}\lambda (\Conid{Stack}\;\Varid{xs})\to {}\<[E]%
\\
\>[5]{}\hsindent{2}{}\<[7]%
\>[7]{}\mathbf{case}\;\Varid{xs}\;\mathbf{of}{}\<[E]%
\\
\>[7]{}\hsindent{2}{}\<[9]%
\>[9]{}[\mskip1.5mu \mskip1.5mu]{}\<[27]%
\>[27]{}\to \Varid{\eta}\;(){}\<[E]%
\\
\>[7]{}\hsindent{2}{}\<[9]%
\>[9]{}(\Conid{Right}\;()\mathbin{:}\Varid{xs'}){}\<[27]%
\>[27]{}\to \mathbf{do}\;\Varid{put}\;(\Conid{Stack}\;\Varid{xs'});\Varid{\eta}\;(){}\<[E]%
\\
\>[7]{}\hsindent{2}{}\<[9]%
\>[9]{}(\Conid{Left}\;\Varid{r}{}\<[19]%
\>[19]{}\mathbin{:}\Varid{xs'}){}\<[27]%
\>[27]{}\to \mathbf{do}\;\Varid{put}\;(\Conid{Stack}\;\Varid{xs'});\Varid{restore}\;\Varid{r};\Varid{untrail}{}\<[E]%
\\
\>[3]{}\hsindent{2}{}\<[5]%
\>[5]{})\Varid{s})\;\Varid{t}{}\<[E]%
\\
\>[3]{}\mathrel{=}\mbox{\commentbegin ~ definition of \ensuremath{\Varid{h_{ND+f}}} and \ensuremath{\Varid{(\Leftrightarrow)}}  \commentend}{}\<[E]%
\\
\>[3]{}\hsindent{2}{}\<[5]%
\>[5]{}\Varid{h_{State1}}\;(\Varid{h_{Modify1}}\;(\Conid{Op}\hsdot{\circ }{.}\Conid{Inr}\hsdot{\circ }{.}\Conid{Inl}\hsdot{\circ }{.}\Conid{Get}\mathbin{\$}\lambda (\Conid{Stack}\;\Varid{xs})\to {}\<[E]%
\\
\>[5]{}\hsindent{2}{}\<[7]%
\>[7]{}\mathbf{case}\;\Varid{xs}\;\mathbf{of}{}\<[E]%
\\
\>[7]{}\hsindent{2}{}\<[9]%
\>[9]{}[\mskip1.5mu \mskip1.5mu]{}\<[27]%
\>[27]{}\to \Varid{\eta}\;[\mskip1.5mu ()\mskip1.5mu]{}\<[E]%
\\
\>[7]{}\hsindent{2}{}\<[9]%
\>[9]{}(\Conid{Right}\;()\mathbin{:}\Varid{xs'}){}\<[27]%
\>[27]{}\to \Varid{h_{ND+f}}\hsdot{\circ }{.}\Varid{(\Leftrightarrow)}\mathbin{\$}\mathbf{do}\;\Varid{put}\;(\Conid{Stack}\;\Varid{xs'});\Varid{\eta}\;(){}\<[E]%
\\
\>[7]{}\hsindent{2}{}\<[9]%
\>[9]{}(\Conid{Left}\;\Varid{r}{}\<[19]%
\>[19]{}\mathbin{:}\Varid{xs'}){}\<[27]%
\>[27]{}\to \Varid{h_{ND+f}}\hsdot{\circ }{.}\Varid{(\Leftrightarrow)}\mathbin{\$}\mathbf{do}\;\Varid{put}\;(\Conid{Stack}\;\Varid{xs'});\Varid{restore}\;\Varid{r};\Varid{untrail}{}\<[E]%
\\
\>[3]{}\hsindent{2}{}\<[5]%
\>[5]{})\Varid{s})\;\Varid{t}{}\<[E]%
\\
\>[3]{}\mathrel{=}\mbox{\commentbegin ~ definition of \ensuremath{\Varid{h_{Modify1}}} and function application  \commentend}{}\<[E]%
\\
\>[3]{}\hsindent{2}{}\<[5]%
\>[5]{}\Varid{h_{State1}}\;(\Conid{Op}\hsdot{\circ }{.}\Conid{Inl}\hsdot{\circ }{.}\Conid{Get}\mathbin{\$}\lambda (\Conid{Stack}\;\Varid{xs})\to {}\<[E]%
\\
\>[5]{}\hsindent{2}{}\<[7]%
\>[7]{}\mathbf{case}\;\Varid{xs}\;\mathbf{of}{}\<[E]%
\\
\>[7]{}\hsindent{2}{}\<[9]%
\>[9]{}[\mskip1.5mu \mskip1.5mu]{}\<[27]%
\>[27]{}\to \Varid{\eta}\;([\mskip1.5mu ()\mskip1.5mu],\Varid{s}){}\<[E]%
\\
\>[7]{}\hsindent{2}{}\<[9]%
\>[9]{}(\Conid{Right}\;()\mathbin{:}\Varid{xs'}){}\<[27]%
\>[27]{}\to \Varid{h_{Modify1}}\;(\Varid{h_{ND+f}}\hsdot{\circ }{.}\Varid{(\Leftrightarrow)}\mathbin{\$}\mathbf{do}\;\Varid{put}\;(\Conid{Stack}\;\Varid{xs'});\Varid{\eta}\;())\;\Varid{s}{}\<[E]%
\\
\>[7]{}\hsindent{2}{}\<[9]%
\>[9]{}(\Conid{Left}\;\Varid{r}{}\<[19]%
\>[19]{}\mathbin{:}\Varid{xs'}){}\<[27]%
\>[27]{}\to \Varid{h_{Modify1}}\;(\Varid{h_{ND+f}}\hsdot{\circ }{.}\Varid{(\Leftrightarrow)}\mathbin{\$}\mathbf{do}\;\Varid{put}\;(\Conid{Stack}\;\Varid{xs'});\Varid{restore}\;\Varid{r};\Varid{untrail})\;\Varid{s}{}\<[E]%
\\
\>[3]{}\hsindent{2}{}\<[5]%
\>[5]{})\;\Varid{t}{}\<[E]%
\\
\>[3]{}\mathrel{=}\mbox{\commentbegin ~ definition of \ensuremath{\Varid{h_{State1}}} and function application; let \ensuremath{\Varid{t}\mathrel{=}\Conid{Stack}\;(\Varid{ys}+\!\!+(\Conid{Right}\;()\mathbin{:}\Varid{xs}))}  \commentend}{}\<[E]%
\\
\>[3]{}\hsindent{4}{}\<[7]%
\>[7]{}\mathbf{case}\;(\Varid{ys}+\!\!+(\Conid{Right}\;()\mathbin{:}\Varid{xs}))\;\mathbf{of}{}\<[E]%
\\
\>[7]{}\hsindent{2}{}\<[9]%
\>[9]{}[\mskip1.5mu \mskip1.5mu]{}\<[27]%
\>[27]{}\to \Varid{\eta}\;(([\mskip1.5mu ()\mskip1.5mu],\Varid{s}),\Varid{t}){}\<[E]%
\\
\>[7]{}\hsindent{2}{}\<[9]%
\>[9]{}(\Conid{Right}\;()\mathbin{:}\Varid{xs'}){}\<[27]%
\>[27]{}\to \Varid{\eta}\;(([\mskip1.5mu ()\mskip1.5mu],\Varid{s}),\Conid{Stack}\;\Varid{xs'}){}\<[E]%
\\
\>[7]{}\hsindent{2}{}\<[9]%
\>[9]{}(\Conid{Left}\;\Varid{r}{}\<[19]%
\>[19]{}\mathbin{:}\Varid{xs'}){}\<[27]%
\>[27]{}\to \Varid{h_{State1}}\;(\Varid{h_{Modify1}}\;(\Varid{h_{ND+f}}\hsdot{\circ }{.}\Varid{(\Leftrightarrow)}\mathbin{\$}{}\<[E]%
\\
\>[27]{}\hsindent{5}{}\<[32]%
\>[32]{}\mathbf{do}\;\Varid{put}\;(\Conid{Stack}\;\Varid{xs'});\Varid{restore}\;\Varid{r};\Varid{untrail})\;\Varid{s})\;\Varid{t}{}\<[E]%
\ColumnHook
\end{hscode}\resethooks
\indentend Then, we proceed by an induction on the structure of \ensuremath{\Varid{ys}}.

\noindent \mbox{\underline{case \ensuremath{\Varid{ys}\mathrel{=}[\mskip1.5mu \mskip1.5mu]}}}\indentbegin \begin{hscode}\SaveRestoreHook
\column{B}{@{}>{\hspre}l<{\hspost}@{}}%
\column{3}{@{}>{\hspre}l<{\hspost}@{}}%
\column{5}{@{}>{\hspre}l<{\hspost}@{}}%
\column{7}{@{}>{\hspre}l<{\hspost}@{}}%
\column{17}{@{}>{\hspre}l<{\hspost}@{}}%
\column{25}{@{}>{\hspre}l<{\hspost}@{}}%
\column{30}{@{}>{\hspre}l<{\hspost}@{}}%
\column{E}{@{}>{\hspre}l<{\hspost}@{}}%
\>[5]{}\mathbf{case}\;(\Varid{ys}+\!\!+(\Conid{Right}\;()\mathbin{:}\Varid{xs}))\;\mathbf{of}{}\<[E]%
\\
\>[5]{}\hsindent{2}{}\<[7]%
\>[7]{}[\mskip1.5mu \mskip1.5mu]{}\<[25]%
\>[25]{}\to \Varid{\eta}\;(([\mskip1.5mu ()\mskip1.5mu],\Varid{s}),\Varid{t}){}\<[E]%
\\
\>[5]{}\hsindent{2}{}\<[7]%
\>[7]{}(\Conid{Right}\;()\mathbin{:}\Varid{xs'}){}\<[25]%
\>[25]{}\to \Varid{\eta}\;(([\mskip1.5mu ()\mskip1.5mu],\Varid{s}),\Conid{Stack}\;\Varid{xs'}){}\<[E]%
\\
\>[5]{}\hsindent{2}{}\<[7]%
\>[7]{}(\Conid{Left}\;\Varid{r}{}\<[17]%
\>[17]{}\mathbin{:}\Varid{xs'}){}\<[25]%
\>[25]{}\to \Varid{h_{State1}}\;(\Varid{h_{Modify1}}\;(\Varid{h_{ND+f}}\hsdot{\circ }{.}\Varid{(\Leftrightarrow)}\mathbin{\$}{}\<[E]%
\\
\>[25]{}\hsindent{5}{}\<[30]%
\>[30]{}\mathbf{do}\;\Varid{put}\;(\Conid{Stack}\;\Varid{xs'});\Varid{restore}\;\Varid{r};\Varid{untrail})\;\Varid{s})\;\Varid{t}{}\<[E]%
\\
\>[3]{}\mathrel{=}\mbox{\commentbegin ~ case split  \commentend}{}\<[E]%
\\
\>[3]{}\hsindent{2}{}\<[5]%
\>[5]{}\Varid{\eta}\;(([\mskip1.5mu ()\mskip1.5mu],\Varid{s}),\Conid{Stack}\;\Varid{xs}){}\<[E]%
\\
\>[3]{}\mathrel{=}\mbox{\commentbegin ~ definition of \ensuremath{\Varid{fminus}}  \commentend}{}\<[E]%
\\
\>[3]{}\hsindent{2}{}\<[5]%
\>[5]{}\Varid{\eta}\;(([\mskip1.5mu ()\mskip1.5mu],\Varid{fminus}\;\Varid{s}\;[\mskip1.5mu \mskip1.5mu]),\Conid{Stack}\;\Varid{xs}){}\<[E]%
\ColumnHook
\end{hscode}\resethooks
\indentend \noindent \mbox{\underline{case \ensuremath{\Varid{ys}\mathrel{=}(\Conid{Left}\;\Varid{r}\mathbin{:}\Varid{ys'})}}}\indentbegin \begin{hscode}\SaveRestoreHook
\column{B}{@{}>{\hspre}l<{\hspost}@{}}%
\column{3}{@{}>{\hspre}l<{\hspost}@{}}%
\column{5}{@{}>{\hspre}l<{\hspost}@{}}%
\column{7}{@{}>{\hspre}l<{\hspost}@{}}%
\column{17}{@{}>{\hspre}l<{\hspost}@{}}%
\column{25}{@{}>{\hspre}l<{\hspost}@{}}%
\column{30}{@{}>{\hspre}l<{\hspost}@{}}%
\column{E}{@{}>{\hspre}l<{\hspost}@{}}%
\>[5]{}\mathbf{case}\;(\Varid{ys}+\!\!+(\Conid{Right}\;()\mathbin{:}\Varid{xs}))\;\mathbf{of}{}\<[E]%
\\
\>[5]{}\hsindent{2}{}\<[7]%
\>[7]{}[\mskip1.5mu \mskip1.5mu]{}\<[25]%
\>[25]{}\to \Varid{\eta}\;(([\mskip1.5mu ()\mskip1.5mu],\Varid{s}),\Varid{t}){}\<[E]%
\\
\>[5]{}\hsindent{2}{}\<[7]%
\>[7]{}(\Conid{Right}\;()\mathbin{:}\Varid{xs'}){}\<[25]%
\>[25]{}\to \Varid{\eta}\;(([\mskip1.5mu ()\mskip1.5mu],\Varid{s}),\Conid{Stack}\;\Varid{xs'}){}\<[E]%
\\
\>[5]{}\hsindent{2}{}\<[7]%
\>[7]{}(\Conid{Left}\;\Varid{r}{}\<[17]%
\>[17]{}\mathbin{:}\Varid{xs'}){}\<[25]%
\>[25]{}\to \Varid{h_{State1}}\;(\Varid{h_{Modify1}}\;(\Varid{h_{ND+f}}\hsdot{\circ }{.}\Varid{(\Leftrightarrow)}\mathbin{\$}{}\<[E]%
\\
\>[25]{}\hsindent{5}{}\<[30]%
\>[30]{}\mathbf{do}\;\Varid{put}\;(\Conid{Stack}\;\Varid{xs'});\Varid{restore}\;\Varid{r};\Varid{untrail})\;\Varid{s})\;\Varid{t}{}\<[E]%
\\
\>[3]{}\mathrel{=}\mbox{\commentbegin ~ case split  \commentend}{}\<[E]%
\\
\>[3]{}\hsindent{2}{}\<[5]%
\>[5]{}\Varid{h_{State1}}\;(\Varid{h_{Modify1}}\;(\Varid{h_{ND+f}}\hsdot{\circ }{.}\Varid{(\Leftrightarrow)}\mathbin{\$}{}\<[E]%
\\
\>[5]{}\hsindent{2}{}\<[7]%
\>[7]{}\mathbf{do}\;\Varid{put}\;(\Conid{Stack}\;(\Varid{ys'}+\!\!+(\Conid{Right}\;()\mathbin{:}\Varid{xs})));\Varid{restore}\;\Varid{r};\Varid{untrail})\;\Varid{s})\;\Varid{t}{}\<[E]%
\\
\>[3]{}\mathrel{=}\mbox{\commentbegin ~ definition of \ensuremath{\Varid{h_{State1}},\Varid{h_{Modify1}},\Varid{h_{ND+f}},\Varid{(\Leftrightarrow)}} and reformulation  \commentend}{}\<[E]%
\\
\>[3]{}\hsindent{2}{}\<[5]%
\>[5]{}\Varid{h_{State1}}\;((\Varid{h_{Modify1}}\hsdot{\circ }{.}\Varid{h_{ND+f}}\hsdot{\circ }{.}\Varid{(\Leftrightarrow)})\;\Varid{untrail}\;(\Varid{s}\mathbin{\ominus}\Varid{r}))\;{}\<[E]%
\\
\>[5]{}\hsindent{2}{}\<[7]%
\>[7]{}(\Conid{Stack}\;(\Varid{ys'}+\!\!+(\Conid{Right}\;()\mathbin{:}\Varid{xs}))){}\<[E]%
\\
\>[3]{}\mathrel{=}\mbox{\commentbegin ~ induction hypothesis on \ensuremath{\Varid{ys'}}  \commentend}{}\<[E]%
\\
\>[3]{}\hsindent{2}{}\<[5]%
\>[5]{}\Varid{\eta}\;(([\mskip1.5mu ()\mskip1.5mu],\Varid{fminus}\;(\Varid{s}\mathbin{\ominus}\Varid{r})\;\Varid{ys'}),\Conid{Stack}\;\Varid{xs}){}\<[E]%
\\
\>[3]{}\mathrel{=}\mbox{\commentbegin ~ definition of \ensuremath{\Varid{fminus}}   \commentend}{}\<[E]%
\\
\>[3]{}\hsindent{2}{}\<[5]%
\>[5]{}\Varid{\eta}\;(([\mskip1.5mu ()\mskip1.5mu],\Varid{fminus}\;\Varid{s}\;(\Conid{Left}\;\Varid{r}\mathbin{:}\Varid{ys'})),\Conid{Stack}\;\Varid{xs}){}\<[E]%
\\
\>[3]{}\mathrel{=}\mbox{\commentbegin ~ definition of \ensuremath{\Varid{ys}}   \commentend}{}\<[E]%
\\
\>[3]{}\hsindent{2}{}\<[5]%
\>[5]{}\Varid{\eta}\;(([\mskip1.5mu ()\mskip1.5mu],\Varid{fminus}\;\Varid{s}\;\Varid{ys}),\Conid{Stack}\;\Varid{xs}){}\<[E]%
\ColumnHook
\end{hscode}\resethooks
\indentend \end{proof}

The following lemma is obvious from
\Cref{lemma:trail-stack-tracks-state} and \Cref{lemma:undoTrail-undos}.
It shows that we can restore the previous state and stack by pushing a
time stamp on the trail stack and use the function \ensuremath{\Varid{untrail}}
afterwards.

\begin{lemma}[State and stack are restored]~ \label{lemma:state-stack-restore}
For \ensuremath{\Varid{t}\mathbin{::}\Conid{Stack}\;(\Conid{Either}\;\Varid{r}\;())}, \ensuremath{\Varid{s}\mathbin{::}\Varid{s}}, and \ensuremath{\Varid{p}\mathbin{::}\Conid{Free}\;(\Varid{Modify_{F}}\;\Varid{s}\;\Varid{r}\mathrel{{:}{+}{:}}\Varid{Nondet_{F}}\mathrel{{:}{+}{:}}\Varid{f})\;\Varid{a}} which does not use the \ensuremath{\Varid{restore}} operation, we
have
\indentbegin \begin{hscode}\SaveRestoreHook
\column{B}{@{}>{\hspre}l<{\hspost}@{}}%
\column{3}{@{}>{\hspre}c<{\hspost}@{}}%
\column{3E}{@{}l@{}}%
\column{4}{@{}>{\hspre}l<{\hspost}@{}}%
\column{8}{@{}>{\hspre}l<{\hspost}@{}}%
\column{14}{@{}>{\hspre}l<{\hspost}@{}}%
\column{19}{@{}>{\hspre}l<{\hspost}@{}}%
\column{23}{@{}>{\hspre}l<{\hspost}@{}}%
\column{E}{@{}>{\hspre}l<{\hspost}@{}}%
\>[4]{}\mathbf{do}\;{}\<[8]%
\>[8]{}((\anonymous ,\Varid{s}_{1}),\Varid{t}_{1}){}\<[23]%
\>[23]{}\leftarrow \Varid{h_{State1}}\;((\Varid{h_{Modify1}}\hsdot{\circ }{.}\Varid{h_{ND+f}}\hsdot{\circ }{.}\Varid{(\Leftrightarrow)})\;(\Varid{pushStack}\;(\Conid{Right}\;()))\;\Varid{s})\;\Varid{t}{}\<[E]%
\\
\>[8]{}((\Varid{x},\Varid{s}_{2}),\Varid{t}_{2}){}\<[23]%
\>[23]{}\leftarrow \Varid{h_{State1}}\;((\Varid{h_{Modify1}}\hsdot{\circ }{.}\Varid{h_{ND+f}}\hsdot{\circ }{.}\Varid{(\Leftrightarrow)})\;(\Varid{local2trail}\;\Varid{p})\;\Varid{s}_{1})\;\Varid{t}_{1}{}\<[E]%
\\
\>[8]{}((\anonymous ,\Varid{s3}),\Varid{t3}){}\<[23]%
\>[23]{}\leftarrow \Varid{h_{State1}}\;((\Varid{h_{Modify1}}\hsdot{\circ }{.}\Varid{h_{ND+f}}\hsdot{\circ }{.}\Varid{(\Leftrightarrow)})\;\Varid{untrail}\;\Varid{s}_{2})\;\Varid{t}_{2}{}\<[E]%
\\
\>[8]{}\Varid{\eta}\;((\Varid{x},\Varid{s3}),\Varid{t3}){}\<[E]%
\\
\>[3]{}\mathrel{=}{}\<[3E]%
\\
\>[3]{}\hsindent{1}{}\<[4]%
\>[4]{}\mathbf{do}\;{}\<[8]%
\>[8]{}((\Varid{x},{}\<[14]%
\>[14]{}\anonymous ),{}\<[19]%
\>[19]{}\anonymous ){}\<[23]%
\>[23]{}\leftarrow \Varid{h_{State1}}\;((\Varid{h_{Modify1}}\hsdot{\circ }{.}\Varid{h_{ND+f}}\hsdot{\circ }{.}\Varid{(\Leftrightarrow)})\;(\Varid{local2trail}\;\Varid{p})\;\Varid{s})\;(\Conid{Stack}\;[\mskip1.5mu \mskip1.5mu]){}\<[E]%
\\
\>[8]{}\Varid{\eta}\;((\Varid{x},\Varid{s}),\Varid{t}){}\<[E]%
\ColumnHook
\end{hscode}\resethooks
\indentend \end{lemma}

\begin{proof}~
Suppose \ensuremath{\Varid{t}\mathrel{=}\Conid{Stack}\;\Varid{xs}}. We calculate as follows.
\indentbegin \begin{hscode}\SaveRestoreHook
\column{B}{@{}>{\hspre}l<{\hspost}@{}}%
\column{3}{@{}>{\hspre}l<{\hspost}@{}}%
\column{4}{@{}>{\hspre}l<{\hspost}@{}}%
\column{8}{@{}>{\hspre}l<{\hspost}@{}}%
\column{23}{@{}>{\hspre}l<{\hspost}@{}}%
\column{28}{@{}>{\hspre}l<{\hspost}@{}}%
\column{30}{@{}>{\hspre}l<{\hspost}@{}}%
\column{45}{@{}>{\hspre}l<{\hspost}@{}}%
\column{E}{@{}>{\hspre}l<{\hspost}@{}}%
\>[4]{}\mathbf{do}\;{}\<[8]%
\>[8]{}((\anonymous ,\Varid{s}_{1}),\Varid{t}_{1}){}\<[23]%
\>[23]{}\leftarrow \Varid{h_{State1}}\;((\Varid{h_{Modify1}}\hsdot{\circ }{.}\Varid{h_{ND+f}}\hsdot{\circ }{.}\Varid{(\Leftrightarrow)})\;(\Varid{pushStack}\;(\Conid{Right}\;()))\;\Varid{s})\;(\Conid{Stack}\;\Varid{xs}){}\<[E]%
\\
\>[8]{}((\Varid{x},\Varid{s}_{2}),\Varid{t}_{2}){}\<[23]%
\>[23]{}\leftarrow \Varid{h_{State1}}\;((\Varid{h_{Modify1}}\hsdot{\circ }{.}\Varid{h_{ND+f}}\hsdot{\circ }{.}\Varid{(\Leftrightarrow)})\;(\Varid{local2trail}\;\Varid{p})\;\Varid{s}_{1})\;\Varid{t}_{1}{}\<[E]%
\\
\>[8]{}((\anonymous ,\Varid{s3}),\Varid{t3}){}\<[23]%
\>[23]{}\leftarrow \Varid{h_{State1}}\;((\Varid{h_{Modify1}}\hsdot{\circ }{.}\Varid{h_{ND+f}}\hsdot{\circ }{.}\Varid{(\Leftrightarrow)})\;\Varid{untrail}\;\Varid{s}_{2})\;\Varid{t}_{2}{}\<[E]%
\\
\>[8]{}\Varid{\eta}\;((\Varid{x},\Varid{s3}),\Varid{t3}){}\<[E]%
\\
\>[3]{}\mathrel{=}\mbox{\commentbegin ~ definition of \ensuremath{\Varid{h_{State1}},\Varid{h_{Modify1}},\Varid{h_{ND+f}},\Varid{(\Leftrightarrow)},\Varid{pushStack}}   \commentend}{}\<[E]%
\\
\>[3]{}\hsindent{1}{}\<[4]%
\>[4]{}\mathbf{do}\;{}\<[8]%
\>[8]{}((\anonymous ,\Varid{s}_{1}),\Varid{t}_{1}){}\<[23]%
\>[23]{}\leftarrow \Varid{\eta}\;(([\mskip1.5mu ()\mskip1.5mu],\Varid{s}),\Conid{Stack}\;(\Conid{Right}\;()\mathbin{:}\Varid{xs})){}\<[E]%
\\
\>[8]{}((\Varid{x},\Varid{s}_{2}),\Varid{t}_{2}){}\<[23]%
\>[23]{}\leftarrow \Varid{h_{State1}}\;((\Varid{h_{Modify1}}\hsdot{\circ }{.}\Varid{h_{ND+f}}\hsdot{\circ }{.}\Varid{(\Leftrightarrow)})\;(\Varid{local2trail}\;\Varid{p})\;\Varid{s}_{1})\;\Varid{t}_{1}{}\<[E]%
\\
\>[8]{}((\anonymous ,\Varid{s3}),\Varid{t3}){}\<[23]%
\>[23]{}\leftarrow \Varid{h_{State1}}\;((\Varid{h_{Modify1}}\hsdot{\circ }{.}\Varid{h_{ND+f}}\hsdot{\circ }{.}\Varid{(\Leftrightarrow)})\;\Varid{untrail}\;\Varid{s}_{2})\;\Varid{t}_{2}{}\<[E]%
\\
\>[8]{}\Varid{\eta}\;((\Varid{x},\Varid{s3}),\Varid{t3}){}\<[E]%
\\
\>[3]{}\mathrel{=}\mbox{\commentbegin ~ monad law   \commentend}{}\<[E]%
\\
\>[3]{}\hsindent{1}{}\<[4]%
\>[4]{}\mathbf{do}\;{}\<[8]%
\>[8]{}((\Varid{x},\Varid{s}_{2}),\Varid{t}_{2}){}\<[23]%
\>[23]{}\leftarrow \Varid{h_{State1}}\;((\Varid{h_{Modify1}}\hsdot{\circ }{.}\Varid{h_{ND+f}}\hsdot{\circ }{.}\Varid{(\Leftrightarrow)})\;(\Varid{local2trail}\;\Varid{p})\;\Varid{s})\;{}\<[E]%
\\
\>[23]{}\hsindent{5}{}\<[28]%
\>[28]{}(\Conid{Stack}\;(\Conid{Right}\;()\mathbin{:}\Varid{xs})){}\<[E]%
\\
\>[8]{}((\anonymous ,\Varid{s3}),\Varid{t3}){}\<[23]%
\>[23]{}\leftarrow \Varid{h_{State1}}\;((\Varid{h_{Modify1}}\hsdot{\circ }{.}\Varid{h_{ND+f}}\hsdot{\circ }{.}\Varid{(\Leftrightarrow)})\;\Varid{untrail}\;\Varid{s}_{2})\;\Varid{t}_{2}{}\<[E]%
\\
\>[8]{}\Varid{\eta}\;((\Varid{x},\Varid{s3}),\Varid{t3}){}\<[E]%
\\
\>[3]{}\mathrel{=}\mbox{\commentbegin ~ by \Cref{lemma:trail-stack-tracks-state}, for some \ensuremath{\Varid{ys}\mathrel{=}[\mskip1.5mu \Conid{Left}\;\Varid{r}_{1},\mathbin{...},\Conid{Left}\;\Varid{r}_{\Varid{n}}\mskip1.5mu]} the equation holds   \commentend}{}\<[E]%
\\
\>[3]{}\hsindent{1}{}\<[4]%
\>[4]{}\mathbf{do}\;{}\<[8]%
\>[8]{}((\Varid{x},\Varid{s}_{2}),\Varid{t}_{2}){}\<[23]%
\>[23]{}\leftarrow \mathbf{do}\;{}\<[30]%
\>[30]{}((\Varid{x},\anonymous ),\anonymous )\leftarrow \Varid{h_{State1}}\;((\Varid{h_{Modify1}}\hsdot{\circ }{.}\Varid{h_{ND+f}}\hsdot{\circ }{.}\Varid{(\Leftrightarrow)}){}\<[E]%
\\
\>[30]{}\hsindent{15}{}\<[45]%
\>[45]{}(\Varid{local2trail}\;\Varid{p})\;\Varid{s})\;(\Conid{Stack}\;[\mskip1.5mu \mskip1.5mu]){}\<[E]%
\\
\>[30]{}\Varid{\eta}\;((\Varid{x},\Varid{fplus}\;\Varid{s}\;\Varid{ys}),\Varid{extend}\;(\Conid{Stack}\;(\Conid{Right}\;()\mathbin{:}\Varid{xs}))\;\Varid{ys}){}\<[E]%
\\
\>[8]{}((\anonymous ,\Varid{s3}),\Varid{t3}){}\<[23]%
\>[23]{}\leftarrow \Varid{h_{State1}}\;((\Varid{h_{Modify1}}\hsdot{\circ }{.}\Varid{h_{ND+f}}\hsdot{\circ }{.}\Varid{(\Leftrightarrow)})\;\Varid{untrail}\;\Varid{s}_{2})\;\Varid{t}_{2}{}\<[E]%
\\
\>[8]{}\Varid{\eta}\;((\Varid{x},\Varid{s3}),\Varid{t3}){}\<[E]%
\\
\>[3]{}\mathrel{=}\mbox{\commentbegin ~ monad law and definition of \ensuremath{\Varid{extend}}  \commentend}{}\<[E]%
\\
\>[3]{}\hsindent{1}{}\<[4]%
\>[4]{}\mathbf{do}\;{}\<[8]%
\>[8]{}((\Varid{x},\anonymous ),\anonymous )\leftarrow \Varid{h_{State1}}\;((\Varid{h_{Modify1}}\hsdot{\circ }{.}\Varid{h_{ND+f}}\hsdot{\circ }{.}\Varid{(\Leftrightarrow)})\;(\Varid{local2trail}\;\Varid{p})\;\Varid{s})\;(\Conid{Stack}\;[\mskip1.5mu \mskip1.5mu]){}\<[E]%
\\
\>[8]{}((\anonymous ,\Varid{s3}),\Varid{t3}){}\<[23]%
\>[23]{}\leftarrow \Varid{h_{State1}}\;((\Varid{h_{Modify1}}\hsdot{\circ }{.}\Varid{h_{ND+f}}\hsdot{\circ }{.}\Varid{(\Leftrightarrow)}){}\<[E]%
\\
\>[23]{}\hsindent{5}{}\<[28]%
\>[28]{}\Varid{untrail}\;(\Varid{fplus}\;\Varid{s}\;\Varid{ys}))\;(\Conid{Stack}\;(\Varid{ys}+\!\!+(\Conid{Right}\;()\mathbin{:}\Varid{xs}))){}\<[E]%
\\
\>[8]{}\Varid{\eta}\;((\Varid{x},\Varid{s3}),\Varid{t3}){}\<[E]%
\\
\>[3]{}\mathrel{=}\mbox{\commentbegin ~ \Cref{lemma:undoTrail-undos} and monad law  \commentend}{}\<[E]%
\\
\>[3]{}\hsindent{1}{}\<[4]%
\>[4]{}\mathbf{do}\;{}\<[8]%
\>[8]{}((\Varid{x},\anonymous ),\anonymous )\leftarrow \Varid{h_{State1}}\;((\Varid{h_{Modify1}}\hsdot{\circ }{.}\Varid{h_{ND+f}}\hsdot{\circ }{.}\Varid{(\Leftrightarrow)})\;(\Varid{local2trail}\;\Varid{p})\;\Varid{s})\;(\Conid{Stack}\;[\mskip1.5mu \mskip1.5mu]){}\<[E]%
\\
\>[8]{}\Varid{\eta}\;((\Varid{x},\Varid{fminus}\;(\Varid{fplus}\;\Varid{s}\;\Varid{ys})\;\Varid{ys}),\Conid{Stack}\;\Varid{xs}){}\<[E]%
\\
\>[3]{}\mathrel{=}\mbox{\commentbegin ~ \Cref{eq:plus-minus} gives \ensuremath{\Varid{fminus}\;(\Varid{fplus}\;\Varid{s}\;\Varid{ys})\;\Varid{ys}\mathrel{=}\Varid{s}}   \commentend}{}\<[E]%
\\
\>[3]{}\hsindent{1}{}\<[4]%
\>[4]{}\mathbf{do}\;{}\<[8]%
\>[8]{}((\Varid{x},\anonymous ),\anonymous )\leftarrow \Varid{h_{State1}}\;((\Varid{h_{Modify1}}\hsdot{\circ }{.}\Varid{h_{ND+f}}\hsdot{\circ }{.}\Varid{(\Leftrightarrow)})\;(\Varid{local2trail}\;\Varid{p})\;\Varid{s})\;(\Conid{Stack}\;[\mskip1.5mu \mskip1.5mu]){}\<[E]%
\\
\>[8]{}\Varid{\eta}\;((\Varid{x},\Varid{s}),\Conid{Stack}\;\Varid{xs}){}\<[E]%
\ColumnHook
\end{hscode}\resethooks
\indentend \end{proof}

\label{lastpage01}
\end{document}

%% file: macros.tex
\usepackage{stmaryrd}
\usepackage{xcolor}
\usepackage{url}
\usepackage{listings}
\usepackage{booktabs}
\usepackage{subcaption}
\usepackage{soul}
\usepackage{scalerel}
\usepackage{mathtools}
\usepackage{quiver}
\newcommand{\ba}{\begin{array}}
\newcommand{\ea}{\end{array}}

\newcommand{\refa}[1]{{\color{red}    {#1}\ifthenelse{\equal{#1}{}}{}{\,}(1)}}
\newcommand{\refb}[1]{{\color{blue}   {#1}\ifthenelse{\equal{#1}{}}{}{\,}(2)}}
\newcommand{\refc}[1]{{\color{violet} {#1}\ifthenelse{\equal{#1}{}}{}{\,}(3)}}
\newcommand{\refd}[1]{{\color{purple} {#1}\ifthenelse{\equal{#1}{}}{}{\,}(4)}}
\newcommand{\refe}[1]{{\color{cyan}   {#1}\ifthenelse{\equal{#1}{}}{}{\,}(5)}}
\newcommand{\reff}[1]{{\color{magenta}{#1}\ifthenelse{\equal{#1}{}}{}{\,}(6)}}
\newcommand{\refg}[1]{{\color{brown}  {#1}\ifthenelse{\equal{#1}{}}{}{\,}(7)}}
\newcommand{\refh}[1]{{\color{orange} {#1}\ifthenelse{\equal{#1}{}}{}{\,}(8)}}
\newcommand{\refs}[1]{{\color{gray} {#1}\ifthenelse{\equal{#1}{}}{}{\,}($\star$)}}